\documentclass[runningheads]{llncs}
\pdfoutput=1

\usepackage{pdfpages} 
\usepackage{url}
\usepackage{float} 
\usepackage{marvosym} 
\newcommand{\GrTTdrule}[4][]{{\displaystyle\frac{\begin{array}{l}#2\end{array}}{#3}\quad\GrTTdrulename{#4}}}

\newcommand{\GrTTpremise}[1]{ #1 \\}
\newenvironment{GrTTdefnblock}[3][]{ \framebox{\mbox{#2}} \quad #3 \\[0pt]}{}

\newcommand{\GrTTmv}[1]{\mathit{#1}}

\newcommand{\GrTTsym}[1]{#1}

\newcommand{\GrTTdrulename}[1]{\textsc{#1}}

\newcommand{\GrTTdruleRhoXXEmptyName}[0]{\GrTTdrulename{Rho\_Empty}}
\newcommand{\GrTTdruleRhoXXEmpty}[1]{\GrTTdrule[#1]{%
}{
  \emptyset   \odot   \emptyset   \models_{  \emptyset  }   \emptyset  }{%
{\GrTTdruleRhoXXEmptyName}{}%
}}

\newcommand{\GrTTdruleRhoXXExtTmName}[0]{\GrTTdrulename{Rho\_ExtTm}}
\newcommand{\GrTTdruleRhoXXExtTm}[1]{\GrTTdrule[#1]{%
\GrTTpremise{t \, \in \,  \interp{ A }_{ \varepsilon } }%
\GrTTpremise{ \Delta  \odot  \Gamma  \models_{ \varepsilon }  \rho   \quad   ( \Delta  \mid  \sigma_{{\mathrm{2}}}  \mid   \textbf{0}  )   \odot  \Gamma  \vdash  A  \GrTTsym{:}   \mathsf{Type}_{ \GrTTsym{0} } }%
}{
 \GrTTsym{(}  \Delta  ,  \sigma_{{\mathrm{2}}}  \GrTTsym{)}  \odot  \GrTTsym{(}  \Gamma  \GrTTsym{,}  x  \GrTTsym{:}  A  \GrTTsym{)}  \models_{ \varepsilon }   \rho [  x  \mapsto  t  ]  }{%
{\GrTTdruleRhoXXExtTmName}{}%
}}

\newcommand{\GrTTdruleRhoXXExtTyName}[0]{\GrTTdrulename{Rho\_ExtTy}}
\newcommand{\GrTTdruleRhoXXExtTy}[1]{\GrTTdrule[#1]{%
\GrTTpremise{\texttt{\textcolor{red}{<<multiple parses>>}}}%
\GrTTpremise{ \Delta  \odot  \Gamma  \models_{ \varepsilon }  \rho   \quad   ( \Delta  \mid  \sigma_{{\mathrm{2}}}  \mid   \textbf{0}  )   \odot  \Gamma  \vdash  A  \GrTTsym{:}   \mathsf{Type}_{ \GrTTsym{1} } }%
}{
 \GrTTsym{(}  \Delta  ,  \sigma_{{\mathrm{2}}}  \GrTTsym{)}  \odot  \GrTTsym{(}  \Gamma  \GrTTsym{,}  x  \GrTTsym{:}  A  \GrTTsym{)}  \models_{ \varepsilon }   \rho [  x  \mapsto  t  ]  }{%
{\GrTTdruleRhoXXExtTyName}{}%
}}

\newcommand{\GrTTdruleEpXXEmptyName}[0]{\GrTTdrulename{Ep\_Empty}}
\newcommand{\GrTTdruleEpXXEmpty}[1]{\GrTTdrule[#1]{%
}{
  \emptyset   \odot   \emptyset   \models   \emptyset  }{%
{\GrTTdruleEpXXEmptyName}{}%
}}

\newcommand{\GrTTdruleEpXXExtTmName}[0]{\GrTTdrulename{Ep\_ExtTm}}
\newcommand{\GrTTdruleEpXXExtTm}[1]{\GrTTdrule[#1]{%
\GrTTpremise{ \Delta  \odot  \Gamma  \models  \varepsilon   \quad   ( \Delta  \mid  \sigma_{{\mathrm{2}}}  \mid   \textbf{0}  )   \odot  \Gamma  \vdash  A  \GrTTsym{:}   \mathsf{Type}_{ \GrTTsym{0} } }%
}{
 \GrTTsym{(}  \Delta  ,  \sigma_{{\mathrm{2}}}  \GrTTsym{)}  \odot  \GrTTsym{(}  \Gamma  \GrTTsym{,}  x  \GrTTsym{:}  A  \GrTTsym{)}  \models  \varepsilon }{%
{\GrTTdruleEpXXExtTmName}{}%
}}

\newcommand{\GrTTdruleEpXXExtTyName}[0]{\GrTTdrulename{Ep\_ExtTy}}
\newcommand{\GrTTdruleEpXXExtTy}[1]{\GrTTdrule[#1]{%
\GrTTpremise{X \, \in \,  \mathcal{K}\interp{ A } }%
\GrTTpremise{ \Delta  \odot  \Gamma  \models  \varepsilon   \quad   ( \Delta  \mid  \sigma_{{\mathrm{2}}}  \mid   \textbf{0}  )   \odot  \Gamma  \vdash  A  \GrTTsym{:}   \mathsf{Type}_{ \GrTTsym{1} } }%
}{
 \GrTTsym{(}  \Delta  ,  \sigma_{{\mathrm{2}}}  \GrTTsym{)}  \odot  \GrTTsym{(}  \Gamma  \GrTTsym{,}  x  \GrTTsym{:}  A  \GrTTsym{)}  \models   \varepsilon [  x  \mapsto  X  ]  }{%
{\GrTTdruleEpXXExtTyName}{}%
}}

\newcommand{\GrTTdruleWfXXEmptyName}[0]{\GrTTdrulename{Wf\_Empty}}
\newcommand{\GrTTdruleWfXXEmpty}[1]{\GrTTdrule[#1]{%
}{
 \emptyset   \odot   \emptyset   \vdash}{%
{\GrTTdruleWfXXEmptyName}{}%
}}

\newcommand{\GrTTdruleWfXXExtName}[0]{\GrTTdrulename{Wf\_Ext}}
\newcommand{\GrTTdruleWfXXExt}[1]{\GrTTdrule[#1]{%
\GrTTpremise{ ( \Delta  \mid  \sigma  \mid   \textbf{0}  )   \odot  \Gamma  \vdash  A  \GrTTsym{:}   \mathsf{Type}_{ l } }%
}{
\Delta  ,  \sigma  \odot  \Gamma  \GrTTsym{,}  x  \GrTTsym{:}  A  \vdash}{%
{\GrTTdruleWfXXExtName}{}%
}}

\newcommand{\GrTTdruleSTXXEqName}[0]{\GrTTdrulename{ST\_Eq}}
\newcommand{\GrTTdruleSTXXEq}[1]{\GrTTdrule[#1]{%
\GrTTpremise{ ( \Delta  \mid  \sigma  \mid   \textbf{0}  )   \odot  \Gamma  \vdash  A  \GrTTsym{=}  B  \GrTTsym{:}   \mathsf{Type}_{ l } }%
}{
 ( \Delta  \mid  \sigma )   \odot  \Gamma  \vdash  A  \leq  B}{%
{\GrTTdruleSTXXEqName}{}%
}}

\newcommand{\GrTTdruleSTXXTransName}[0]{\GrTTdrulename{ST\_Trans}}
\newcommand{\GrTTdruleSTXXTrans}[1]{\GrTTdrule[#1]{%
\GrTTpremise{ ( \Delta  \mid  \sigma )   \odot  \Gamma  \vdash  A  \leq  B  \quad   ( \Delta  \mid  \sigma )   \odot  \Gamma  \vdash  B  \leq  C}%
}{
 ( \Delta  \mid  \sigma )   \odot  \Gamma  \vdash  A  \leq  C}{%
{\GrTTdruleSTXXTransName}{}%
}}

\newcommand{\GrTTdruleSTXXTyName}[0]{\GrTTdrulename{ST\_Ty}}
\newcommand{\GrTTdruleSTXXTy}[1]{\GrTTdrule[#1]{%
\GrTTpremise{\Delta  \odot  \Gamma  \vdash  \quad  l  \leq  l'}%
}{
 ( \Delta  \mid   \textbf{0}  )   \odot  \Gamma  \vdash   \mathsf{Type}_{ l }   \leq   \mathsf{Type}_{ l' } }{%
{\GrTTdruleSTXXTyName}{}%
}}

\newcommand{\GrTTdruleSTXXArrowName}[0]{\GrTTdrulename{ST\_Arrow}}
\newcommand{\GrTTdruleSTXXArrow}[1]{\GrTTdrule[#1]{%
\GrTTpremise{ ( \Delta  ,  \sigma_{{\mathrm{1}}}  \mid  \sigma_{{\mathrm{2}}}  ,  r  \mid   \textbf{0}  )   \odot  \Gamma  \GrTTsym{,}  x  \GrTTsym{:}  A  \vdash  B  \GrTTsym{:}   \mathsf{Type}_{ l } }%
\GrTTpremise{ ( \Delta  \mid  \sigma_{{\mathrm{1}}} )   \odot  \Gamma  \vdash  A'  \leq  A  \quad   ( \Delta  ,  \sigma_{{\mathrm{1}}}  \mid  \sigma_{{\mathrm{2}}}  ,  r )   \odot  \Gamma  \GrTTsym{,}  x  \GrTTsym{:}  A'  \vdash  B  \leq  B'}%
}{
 ( \Delta  \mid  \sigma_{{\mathrm{1}}}  \GrTTsym{+}  \sigma_{{\mathrm{2}}} )   \odot  \Gamma  \vdash   \textstyle (  x  :_{  \textcolor{darkblue}{( s ,  r )}  }  A  )  \to   B   \leq   \textstyle (  x  :_{  \textcolor{darkblue}{( s ,  r )}  }  A'  )  \to   B' }{%
{\GrTTdruleSTXXArrowName}{}%
}}

\newcommand{\GrTTdruleSTXXTenName}[0]{\GrTTdrulename{ST\_Ten}}
\newcommand{\GrTTdruleSTXXTen}[1]{\GrTTdrule[#1]{%
\GrTTpremise{ ( \Delta  ,  \sigma_{{\mathrm{1}}}  \mid  \sigma_{{\mathrm{2}}}  ,  r )   \odot  \Gamma  \GrTTsym{,}  x  \GrTTsym{:}  A  \vdash  B  \leq  B'}%
}{
 ( \Delta  \mid  \sigma_{{\mathrm{1}}}  \GrTTsym{+}  \sigma_{{\mathrm{2}}} )   \odot  \Gamma  \vdash   \textstyle (  x  :_{  \textcolor{darkblue}{ r }  }  A  )  \otimes   B   \leq   \textstyle (  x  :_{  \textcolor{darkblue}{ r }  }  A  )  \otimes   B' }{%
{\GrTTdruleSTXXTenName}{}%
}}

\newcommand{\GrTTdruleSTXXBoxName}[0]{\GrTTdrulename{ST\_Box}}
\newcommand{\GrTTdruleSTXXBox}[1]{\GrTTdrule[#1]{%
\GrTTpremise{ ( \Delta  \mid  \sigma )   \odot  \Gamma  \vdash  A  \leq  A'}%
}{
 ( \Delta  \mid  \sigma )   \odot  \Gamma  \vdash   \square_{  \textcolor{darkblue}{ s }  }  A   \leq   \square_{  \textcolor{darkblue}{ s }  }  A' }{%
{\GrTTdruleSTXXBoxName}{}%
}}

\newcommand{\GrTTdruleTXXTypeName}[0]{\GrTTdrulename{T\_Type}}
\newcommand{\GrTTdruleTXXType}[1]{\GrTTdrule[#1]{%
\GrTTpremise{\Delta  \odot  \Gamma  \vdash}%
}{
 ( \Delta  \mid   \textbf{0}   \mid   \textbf{0}  )   \odot  \Gamma  \vdash   \mathsf{Type}_{ l }   \GrTTsym{:}   \mathsf{Type}_{   \mathsf{suc}\  l   } }{%
{\GrTTdruleTXXTypeName}{}%
}}

\newcommand{\GrTTdruleTXXVarName}[0]{\GrTTdrulename{T\_Var}}
\newcommand{\GrTTdruleTXXVar}[1]{\GrTTdrule[#1]{%
\GrTTpremise{\Delta_{{\mathrm{1}}}  ,  \sigma  ,  \Delta_{{\mathrm{2}}}  \odot  \Gamma_{{\mathrm{1}}}  \GrTTsym{,}  x  \GrTTsym{:}  A  \GrTTsym{,}  \Gamma_{{\mathrm{2}}}  \vdash  \quad   \left|  \Delta_{{\mathrm{1}}}  \right|   \GrTTsym{=}   \left|  \Gamma_{{\mathrm{1}}}  \right| }%
}{
 ( \Delta_{{\mathrm{1}}}  ,  \sigma  ,  \Delta_{{\mathrm{2}}}  \mid   \textbf{0}^{  \left|  \Delta_{{\mathrm{1}}}  \right|  }   ,  \GrTTsym{1}  ,   \textbf{0}   \mid  \sigma  ,  \GrTTsym{0}  ,   \textbf{0}  )   \odot  \Gamma_{{\mathrm{1}}}  \GrTTsym{,}  x  \GrTTsym{:}  A  \GrTTsym{,}  \Gamma_{{\mathrm{2}}}  \vdash  x  \GrTTsym{:}  A}{%
{\GrTTdruleTXXVarName}{}%
}}

\newcommand{\GrTTdruleTXXArrowName}[0]{\GrTTdrulename{T\_Arrow}}
\newcommand{\GrTTdruleTXXArrow}[1]{\GrTTdrule[#1]{%
\GrTTpremise{ ( \Delta  \mid  \sigma_{{\mathrm{1}}}  \mid   \textbf{0}  )   \odot  \Gamma  \vdash  A  \GrTTsym{:}   \mathsf{Type}_{ l_{{\mathrm{1}}} }   \quad   ( \Delta  ,  \sigma_{{\mathrm{1}}}  \mid  \sigma_{{\mathrm{2}}}  ,  r  \mid   \textbf{0}  )   \odot  \Gamma  \GrTTsym{,}  x  \GrTTsym{:}  A  \vdash  B  \GrTTsym{:}   \mathsf{Type}_{ l_{{\mathrm{2}}} } }%
}{
 ( \Delta  \mid  \sigma_{{\mathrm{1}}}  \GrTTsym{+}  \sigma_{{\mathrm{2}}}  \mid   \textbf{0}  )   \odot  \Gamma  \vdash   \textstyle (  x  :_{  \textcolor{darkblue}{( s ,  r )}  }  A  )  \to   B   \GrTTsym{:}   \mathsf{Type}_{  l_{{\mathrm{1}}}  \mathop{\sqcup}  l_{{\mathrm{2}}}  } }{%
{\GrTTdruleTXXArrowName}{}%
}}

\newcommand{\GrTTdruleTXXTenName}[0]{\GrTTdrulename{T\_Ten}}
\newcommand{\GrTTdruleTXXTen}[1]{\GrTTdrule[#1]{%
\GrTTpremise{ ( \Delta  \mid  \sigma_{{\mathrm{1}}}  \mid   \textbf{0}  )   \odot  \Gamma  \vdash  A  \GrTTsym{:}   \mathsf{Type}_{ l_{{\mathrm{1}}} }   \quad   ( \Delta  ,  \sigma_{{\mathrm{1}}}  \mid  \sigma_{{\mathrm{2}}}  ,  r  \mid   \textbf{0}  )   \odot  \Gamma  \GrTTsym{,}  x  \GrTTsym{:}  A  \vdash  B  \GrTTsym{:}   \mathsf{Type}_{ l_{{\mathrm{2}}} } }%
}{
 ( \Delta  \mid  \sigma_{{\mathrm{1}}}  \GrTTsym{+}  \sigma_{{\mathrm{2}}}  \mid   \textbf{0}  )   \odot  \Gamma  \vdash   \textstyle (  x  :_{  \textcolor{darkblue}{ r }  }  A  )  \otimes   B   \GrTTsym{:}   \mathsf{Type}_{  l_{{\mathrm{1}}}  \mathop{\sqcup}  l_{{\mathrm{2}}}  } }{%
{\GrTTdruleTXXTenName}{}%
}}

\newcommand{\GrTTdruleTXXFunName}[0]{\GrTTdrulename{T\_Fun}}
\newcommand{\GrTTdruleTXXFun}[1]{\GrTTdrule[#1]{%
\GrTTpremise{ ( \Delta  ,  \sigma_{{\mathrm{1}}}  \mid  \sigma_{{\mathrm{3}}}  ,  r  \mid   \textbf{0}  )   \odot  \Gamma  \GrTTsym{,}  x  \GrTTsym{:}  A  \vdash  B  \GrTTsym{:}   \mathsf{Type}_{ l }   \quad   ( \Delta  ,  \sigma_{{\mathrm{1}}}  \mid  \sigma_{{\mathrm{2}}}  ,  s  \mid  \sigma_{{\mathrm{3}}}  ,  r )   \odot  \Gamma  \GrTTsym{,}  x  \GrTTsym{:}  A  \vdash  t  \GrTTsym{:}  B}%
}{
 ( \Delta  \mid  \sigma_{{\mathrm{2}}}  \mid  \sigma_{{\mathrm{1}}}  \GrTTsym{+}  \sigma_{{\mathrm{3}}} )   \odot  \Gamma  \vdash   \lambda  x . t   \GrTTsym{:}   \textstyle (  x  :_{  \textcolor{darkblue}{( s ,  r )}  }  A  )  \to   B }{%
{\GrTTdruleTXXFunName}{}%
}}

\newcommand{\GrTTdruleTXXAppName}[0]{\GrTTdrulename{T\_App}}
\newcommand{\GrTTdruleTXXApp}[1]{\GrTTdrule[#1]{%
\GrTTpremise{ ( \Delta  ,  \sigma_{{\mathrm{1}}}  \mid  \sigma_{{\mathrm{3}}}  ,  r  \mid   \textbf{0}  )   \odot  \Gamma  \GrTTsym{,}  x  \GrTTsym{:}  A  \vdash  B  \GrTTsym{:}   \mathsf{Type}_{ l } }%
\GrTTpremise{ ( \Delta  \mid  \sigma_{{\mathrm{2}}}  \mid  \sigma_{{\mathrm{1}}}  \GrTTsym{+}  \sigma_{{\mathrm{3}}} )   \odot  \Gamma  \vdash  t_{{\mathrm{1}}}  \GrTTsym{:}   \textstyle (  x  :_{  \textcolor{darkblue}{( s ,  r )}  }  A  )  \to   B   \quad   ( \Delta  \mid  \sigma_{{\mathrm{4}}}  \mid  \sigma_{{\mathrm{1}}} )   \odot  \Gamma  \vdash  t_{{\mathrm{2}}}  \GrTTsym{:}  A}%
}{
 ( \Delta  \mid  \sigma_{{\mathrm{2}}}  \GrTTsym{+}   s  \ast  \sigma_{{\mathrm{4}}}   \mid  \sigma_{{\mathrm{3}}}  \GrTTsym{+}   r  \ast  \sigma_{{\mathrm{4}}}  )   \odot  \Gamma  \vdash   t_{{\mathrm{1}}} \,{ t_{{\mathrm{2}}} }   \GrTTsym{:}  \GrTTsym{[}  t_{{\mathrm{2}}}  \GrTTsym{/}  x  \GrTTsym{]}  B}{%
{\GrTTdruleTXXAppName}{}%
}}

\newcommand{\GrTTdruleTXXPairName}[0]{\GrTTdrulename{T\_Pair}}
\newcommand{\GrTTdruleTXXPair}[1]{\GrTTdrule[#1]{%
\GrTTpremise{ ( \Delta  ,  \sigma_{{\mathrm{1}}}  \mid  \sigma_{{\mathrm{3}}}  ,  r  \mid   \textbf{0}  )   \odot  \Gamma  \GrTTsym{,}  x  \GrTTsym{:}  A  \vdash  B  \GrTTsym{:}   \mathsf{Type}_{ l } }%
\GrTTpremise{ ( \Delta  \mid  \sigma_{{\mathrm{2}}}  \mid  \sigma_{{\mathrm{1}}} )   \odot  \Gamma  \vdash  t_{{\mathrm{1}}}  \GrTTsym{:}  A  \quad   ( \Delta  \mid  \sigma_{{\mathrm{4}}}  \mid  \sigma_{{\mathrm{3}}}  \GrTTsym{+}   r  \ast  \sigma_{{\mathrm{2}}}  )   \odot  \Gamma  \vdash  t_{{\mathrm{2}}}  \GrTTsym{:}  \GrTTsym{[}  t_{{\mathrm{1}}}  \GrTTsym{/}  x  \GrTTsym{]}  B}%
}{
 ( \Delta  \mid  \sigma_{{\mathrm{2}}}  \GrTTsym{+}  \sigma_{{\mathrm{4}}}  \mid  \sigma_{{\mathrm{1}}}  \GrTTsym{+}  \sigma_{{\mathrm{3}}} )   \odot  \Gamma  \vdash  \GrTTsym{(}  t_{{\mathrm{1}}}  \GrTTsym{,}  t_{{\mathrm{2}}}  \GrTTsym{)}  \GrTTsym{:}   \textstyle (  x  :_{  \textcolor{darkblue}{ r }  }  A  )  \otimes   B }{%
{\GrTTdruleTXXPairName}{}%
}}

\newcommand{\GrTTdruleTXXTenCutName}[0]{\GrTTdrulename{T\_TenCut}}
\newcommand{\GrTTdruleTXXTenCut}[1]{\GrTTdrule[#1]{%
\GrTTpremise{ ( \Delta  \mid  \sigma_{{\mathrm{3}}}  \mid  \sigma_{{\mathrm{1}}}  \GrTTsym{+}  \sigma_{{\mathrm{2}}} )   \odot  \Gamma  \vdash  t_{{\mathrm{1}}}  \GrTTsym{:}   \textstyle (  x  :_{  \textcolor{darkblue}{ r }  }  A  )  \otimes   B }%
\GrTTpremise{ ( \Delta  ,  \GrTTsym{(}  \sigma_{{\mathrm{1}}}  \GrTTsym{+}  \sigma_{{\mathrm{2}}}  \GrTTsym{)}  \mid  \sigma_{{\mathrm{5}}}  ,  r'  \mid   \textbf{0}  )   \odot  \Gamma  \GrTTsym{,}  z  \GrTTsym{:}   \textstyle (  x  :_{  \textcolor{darkblue}{ r }  }  A  )  \otimes   B   \vdash  C  \GrTTsym{:}   \mathsf{Type}_{ l } }%
\GrTTpremise{ ( \Delta  ,  \sigma_{{\mathrm{1}}}  ,  \GrTTsym{(}  \sigma_{{\mathrm{2}}}  ,  r  \GrTTsym{)}  \mid  \sigma_{{\mathrm{4}}}  ,  s  ,  s  \mid  \sigma_{{\mathrm{5}}}  ,  r'  ,  r' )   \odot  \Gamma  \GrTTsym{,}  x  \GrTTsym{:}  A  \GrTTsym{,}  y  \GrTTsym{:}  B  \vdash  t_{{\mathrm{2}}}  \GrTTsym{:}  \GrTTsym{[}  \GrTTsym{(}  x  \GrTTsym{,}  y  \GrTTsym{)}  \GrTTsym{/}  z  \GrTTsym{]}  C}%
}{
 ( \Delta  \mid  \sigma_{{\mathrm{4}}}  \GrTTsym{+}   s  \ast  \sigma_{{\mathrm{3}}}   \mid  \sigma_{{\mathrm{5}}}  \GrTTsym{+}   r'  \ast  \sigma_{{\mathrm{3}}}  )   \odot  \Gamma  \vdash   \mathsf{let}\,  \textstyle (  x  ,  y  )   =  t_{{\mathrm{1}}}  \,\mathsf{in}\, t_{{\mathrm{2}}}   \GrTTsym{:}  \GrTTsym{[}  t_{{\mathrm{1}}}  \GrTTsym{/}  z  \GrTTsym{]}  C}{%
{\GrTTdruleTXXTenCutName}{}%
}}

\newcommand{\GrTTdruleTXXBoxName}[0]{\GrTTdrulename{T\_Box}}
\newcommand{\GrTTdruleTXXBox}[1]{\GrTTdrule[#1]{%
\GrTTpremise{ ( \Delta  \mid  \sigma  \mid   \textbf{0}  )   \odot  \Gamma  \vdash  A  \GrTTsym{:}   \mathsf{Type}_{ l } }%
}{
 ( \Delta  \mid  \sigma  \mid   \textbf{0}  )   \odot  \Gamma  \vdash   \square_{  \textcolor{darkblue}{ s }  }  A   \GrTTsym{:}   \mathsf{Type}_{ l } }{%
{\GrTTdruleTXXBoxName}{}%
}}

\newcommand{\GrTTdruleTXXBoxIName}[0]{\GrTTdrulename{T\_BoxI}}
\newcommand{\GrTTdruleTXXBoxI}[1]{\GrTTdrule[#1]{%
\GrTTpremise{ ( \Delta  \mid  \sigma_{{\mathrm{1}}}  \mid  \sigma_{{\mathrm{2}}} )   \odot  \Gamma  \vdash  t  \GrTTsym{:}  A}%
}{
 ( \Delta  \mid   s  \ast  \sigma_{{\mathrm{1}}}   \mid  \sigma_{{\mathrm{2}}} )   \odot  \Gamma  \vdash   \square  t   \GrTTsym{:}   \square_{  \textcolor{darkblue}{ s }  }  A }{%
{\GrTTdruleTXXBoxIName}{}%
}}

\newcommand{\GrTTdruleTXXBoxEName}[0]{\GrTTdrulename{T\_BoxE}}
\newcommand{\GrTTdruleTXXBoxE}[1]{\GrTTdrule[#1]{%
\GrTTpremise{ ( \Delta  \mid  \sigma_{{\mathrm{1}}}  \mid  \sigma_{{\mathrm{2}}} )   \odot  \Gamma  \vdash  t_{{\mathrm{1}}}  \GrTTsym{:}   \square_{  \textcolor{darkblue}{ s }  }  A }%
\GrTTpremise{ ( \Delta  ,  \sigma_{{\mathrm{2}}}  \mid  \sigma_{{\mathrm{4}}}  ,  r  \mid   \textbf{0}  )   \odot  \Gamma  \GrTTsym{,}  z  \GrTTsym{:}   \square_{  \textcolor{darkblue}{ s }  }  A   \vdash  B  \GrTTsym{:}   \mathsf{Type}_{ l } }%
\GrTTpremise{ ( \Delta  ,  \sigma_{{\mathrm{2}}}  \mid  \sigma_{{\mathrm{3}}}  ,  s  \mid  \sigma_{{\mathrm{4}}}  ,  \GrTTsym{(}   s  \ast  r   \GrTTsym{)} )   \odot  \Gamma  \GrTTsym{,}  x  \GrTTsym{:}  A  \vdash  t_{{\mathrm{2}}}  \GrTTsym{:}  \GrTTsym{[}   \square  x   \GrTTsym{/}  z  \GrTTsym{]}  B}%
}{
 ( \Delta  \mid  \sigma_{{\mathrm{1}}}  \GrTTsym{+}  \sigma_{{\mathrm{3}}}  \mid  \sigma_{{\mathrm{4}}}  \GrTTsym{+}   r  \ast  \sigma_{{\mathrm{1}}}  )   \odot  \Gamma  \vdash   \mathsf{let}\,  \square  x   =  t_{{\mathrm{1}}}  \,\mathsf{in}\, t_{{\mathrm{2}}}   \GrTTsym{:}  \GrTTsym{[}  t_{{\mathrm{1}}}  \GrTTsym{/}  z  \GrTTsym{]}  B}{%
{\GrTTdruleTXXBoxEName}{}%
}}

\newcommand{\GrTTdruleTXXTyConvName}[0]{\GrTTdrulename{T\_TyConv}}
\newcommand{\GrTTdruleTXXTyConv}[1]{\GrTTdrule[#1]{%
\GrTTpremise{ ( \Delta  \mid  \sigma_{{\mathrm{1}}}  \mid  \sigma_{{\mathrm{2}}} )   \odot  \Gamma  \vdash  t  \GrTTsym{:}  A  \quad   ( \Delta  \mid  \sigma_{{\mathrm{2}}} )   \odot  \Gamma  \vdash  A  \leq  B}%
}{
 ( \Delta  \mid  \sigma_{{\mathrm{1}}}  \mid  \sigma_{{\mathrm{2}}} )   \odot  \Gamma  \vdash  t  \GrTTsym{:}  B}{%
{\GrTTdruleTXXTyConvName}{}%
}}

\newcommand{\GrTTdruleSemXXBetaBoxName}[0]{\GrTTdrulename{Sem\_BetaBox}}

\newcommand{\GrTTdruleChkAlgXXFunName}[0]{\GrTTdrulename{ChkAlg\_Fun}}
\newcommand{\GrTTdruleChkAlgXXFun}[1]{\GrTTdrule[#1]{%
\GrTTpremise{ \Delta  ;  \Gamma  \vdash  A  \Rightarrow   \mathsf{Type}_{ l }   ;  \sigma_{{\mathrm{1}}}  ;   \textbf{0}  }%
\GrTTpremise{ \Delta  ,  \sigma_{{\mathrm{1}}}  ;  \Gamma  \GrTTsym{,}  x  \GrTTsym{:}  A  \vdash  t  \Leftarrow  B  ;  \sigma_{{\mathrm{2}}}  ,  s  ;  \sigma_{{\mathrm{3}}}  ,  r }%
}{
 \Delta  ;  \Gamma  \vdash   \lambda  x . t   \Leftarrow   \textstyle (  x  :_{  \textcolor{darkblue}{( s ,  r )}  }  A  )  \to   B   ;  \sigma_{{\mathrm{2}}}  ;  \sigma_{{\mathrm{1}}}  \GrTTsym{+}  \sigma_{{\mathrm{3}}} }{%
{\GrTTdruleChkAlgXXFunName}{}%
}}

\newcommand{\GrTTdruleInfAlgXXAppName}[0]{\GrTTdrulename{InfAlg\_App}}
\newcommand{\GrTTdruleInfAlgXXApp}[1]{\GrTTdrule[#1]{%
\GrTTpremise{ \Delta  ;  \Gamma  \vdash  t_{{\mathrm{1}}}  \Rightarrow   \textstyle (  x  :_{  \textcolor{darkblue}{( s ,  r )}  }  A  )  \to   B   ;  \sigma_{{\mathrm{2}}}  ;  \sigma_{{\mathrm{13}}} }%
\GrTTpremise{ \Delta  ;  \Gamma  \vdash  t_{{\mathrm{2}}}  \Leftarrow  A  ;  \sigma_{{\mathrm{4}}}  ;  \sigma_{{\mathrm{1}}} }%
\GrTTpremise{ \Delta  ,  \sigma_{{\mathrm{1}}}  ;  \Gamma  \GrTTsym{,}  x  \GrTTsym{:}  A  \vdash  B  \Rightarrow   \mathsf{Type}_{ l }   ;   \sigma_{{\mathrm{3}}}  ,  r   ;   \textbf{0}  }%
\GrTTpremise{\sigma_{{\mathrm{13}}}  \GrTTsym{=}  \sigma_{{\mathrm{1}}}  \GrTTsym{+}  \sigma_{{\mathrm{3}}}}%
}{
 \Delta  ;  \Gamma  \vdash   t_{{\mathrm{1}}} \,{ t_{{\mathrm{2}}} }   \Rightarrow  \GrTTsym{[}  t_{{\mathrm{2}}}  \GrTTsym{/}  x  \GrTTsym{]}  B  ;  \sigma_{{\mathrm{2}}}  \GrTTsym{+}   s  \ast  \sigma_{{\mathrm{4}}}   ;  \sigma_{{\mathrm{3}}}  \GrTTsym{+}   r  \ast  \sigma_{{\mathrm{4}}}  }{%
{\GrTTdruleInfAlgXXAppName}{}%
}}

\newcommand{\GrTTdruleTEQXXReflName}[0]{\GrTTdrulename{TEQ\_Refl}}
\newcommand{\GrTTdruleTEQXXRefl}[1]{\GrTTdrule[#1]{%
\GrTTpremise{ ( \Delta  \mid  \sigma_{{\mathrm{1}}}  \mid  \sigma_{{\mathrm{2}}} )   \odot  \Gamma  \vdash  t  \GrTTsym{:}  A}%
}{
 ( \Delta  \mid  \sigma_{{\mathrm{1}}}  \mid  \sigma_{{\mathrm{2}}} )   \odot  \Gamma  \vdash  t  \GrTTsym{=}  t  \GrTTsym{:}  A}{%
{\GrTTdruleTEQXXReflName}{}%
}}

\newcommand{\GrTTdruleTEQXXTransName}[0]{\GrTTdrulename{TEQ\_Trans}}
\newcommand{\GrTTdruleTEQXXTrans}[1]{\GrTTdrule[#1]{%
\GrTTpremise{ ( \Delta  \mid  \sigma_{{\mathrm{1}}}  \mid  \sigma_{{\mathrm{2}}} )   \odot  \Gamma  \vdash  t_{{\mathrm{1}}}  \GrTTsym{=}  t_{{\mathrm{2}}}  \GrTTsym{:}  A  \quad   ( \Delta  \mid  \sigma_{{\mathrm{1}}}  \mid  \sigma_{{\mathrm{2}}} )   \odot  \Gamma  \vdash  t_{{\mathrm{2}}}  \GrTTsym{=}  t_{{\mathrm{3}}}  \GrTTsym{:}  A}%
}{
 ( \Delta  \mid  \sigma_{{\mathrm{1}}}  \mid  \sigma_{{\mathrm{2}}} )   \odot  \Gamma  \vdash  t_{{\mathrm{1}}}  \GrTTsym{=}  t_{{\mathrm{3}}}  \GrTTsym{:}  A}{%
{\GrTTdruleTEQXXTransName}{}%
}}

\newcommand{\GrTTdruleTEQXXSymName}[0]{\GrTTdrulename{TEQ\_Sym}}
\newcommand{\GrTTdruleTEQXXSym}[1]{\GrTTdrule[#1]{%
\GrTTpremise{ ( \Delta  \mid  \sigma_{{\mathrm{1}}}  \mid  \sigma_{{\mathrm{2}}} )   \odot  \Gamma  \vdash  t_{{\mathrm{1}}}  \GrTTsym{=}  t_{{\mathrm{2}}}  \GrTTsym{:}  A}%
}{
 ( \Delta  \mid  \sigma_{{\mathrm{1}}}  \mid  \sigma_{{\mathrm{2}}} )   \odot  \Gamma  \vdash  t_{{\mathrm{2}}}  \GrTTsym{=}  t_{{\mathrm{1}}}  \GrTTsym{:}  A}{%
{\GrTTdruleTEQXXSymName}{}%
}}

\newcommand{\GrTTdruleTEQXXConvTyName}[0]{\GrTTdrulename{TEQ\_ConvTy}}
\newcommand{\GrTTdruleTEQXXConvTy}[1]{\GrTTdrule[#1]{%
\GrTTpremise{ ( \Delta  \mid  \sigma_{{\mathrm{1}}}  \mid  \sigma_{{\mathrm{2}}} )   \odot  \Gamma  \vdash  t_{{\mathrm{1}}}  \GrTTsym{=}  t_{{\mathrm{2}}}  \GrTTsym{:}  A  \quad   ( \Delta  \mid  \sigma_{{\mathrm{2}}} )   \odot  \Gamma  \vdash  A  \leq  B}%
}{
 ( \Delta  \mid  \sigma_{{\mathrm{1}}}  \mid  \sigma_{{\mathrm{2}}} )   \odot  \Gamma  \vdash  t_{{\mathrm{1}}}  \GrTTsym{=}  t_{{\mathrm{2}}}  \GrTTsym{:}  B}{%
{\GrTTdruleTEQXXConvTyName}{}%
}}

\newcommand{\GrTTdruleTEQXXArrowName}[0]{\GrTTdrulename{TEQ\_Arrow}}
\newcommand{\GrTTdruleTEQXXArrow}[1]{\GrTTdrule[#1]{%
\GrTTpremise{ ( \Delta  \mid  \sigma_{{\mathrm{1}}}  \mid   \textbf{0}  )   \odot  \Gamma  \vdash  A  \GrTTsym{=}  C  \GrTTsym{:}   \mathsf{Type}_{ l_{{\mathrm{1}}} }   \quad   ( \Delta  ,  \sigma_{{\mathrm{1}}}  \mid  \sigma_{{\mathrm{2}}}  ,  r  \mid   \textbf{0}  )   \odot  \Gamma  \GrTTsym{,}  x  \GrTTsym{:}  A  \vdash  B  \GrTTsym{=}  D  \GrTTsym{:}   \mathsf{Type}_{ l_{{\mathrm{2}}} } }%
}{
 ( \Delta  \mid  \sigma_{{\mathrm{1}}}  \GrTTsym{+}  \sigma_{{\mathrm{2}}}  \mid   \textbf{0}  )   \odot  \Gamma  \vdash   \textstyle (  x  :_{  \textcolor{darkblue}{( s ,  r )}  }  A  )  \to   B   \GrTTsym{=}   \textstyle (  x  :_{  \textcolor{darkblue}{( s ,  r )}  }  C  )  \to   D   \GrTTsym{:}   \mathsf{Type}_{  l_{{\mathrm{1}}}  \mathop{\sqcup}  l_{{\mathrm{2}}}  } }{%
{\GrTTdruleTEQXXArrowName}{}%
}}

\newcommand{\GrTTdruleTEQXXArrowCompName}[0]{\GrTTdrulename{TEQ\_ArrowComp}}
\newcommand{\GrTTdruleTEQXXArrowComp}[1]{\GrTTdrule[#1]{%
\GrTTpremise{ ( \Delta  ,  \sigma_{{\mathrm{1}}}  \mid  \sigma_{{\mathrm{2}}}  ,  s  \mid  \sigma_{{\mathrm{3}}}  ,  r )   \odot  \Gamma  \GrTTsym{,}  x  \GrTTsym{:}  A  \vdash  t_{{\mathrm{1}}}  \GrTTsym{:}  B  \quad   ( \Delta  \mid  \sigma_{{\mathrm{4}}}  \mid  \sigma_{{\mathrm{1}}} )   \odot  \Gamma  \vdash  t_{{\mathrm{2}}}  \GrTTsym{:}  A}%
}{
 ( \Delta  \mid  \sigma_{{\mathrm{2}}}  \GrTTsym{+}   s  \ast  \sigma_{{\mathrm{4}}}   \mid  \sigma_{{\mathrm{3}}}  \GrTTsym{+}   r  \ast  \sigma_{{\mathrm{4}}}  )   \odot  \Gamma  \vdash   \GrTTsym{(}   \lambda  x . t_{{\mathrm{1}}}   \GrTTsym{)} \,{ t_{{\mathrm{2}}} }   \GrTTsym{=}  \GrTTsym{[}  t_{{\mathrm{2}}}  \GrTTsym{/}  x  \GrTTsym{]}  t_{{\mathrm{1}}}  \GrTTsym{:}  \GrTTsym{[}  t_{{\mathrm{2}}}  \GrTTsym{/}  x  \GrTTsym{]}  B}{%
{\GrTTdruleTEQXXArrowCompName}{}%
}}

\newcommand{\GrTTdruleTEQXXArrowUniqName}[0]{\GrTTdrulename{TEQ\_ArrowUniq}}
\newcommand{\GrTTdruleTEQXXArrowUniq}[1]{\GrTTdrule[#1]{%
\GrTTpremise{ ( \Delta  \mid  \sigma_{{\mathrm{1}}}  \mid  \sigma_{{\mathrm{2}}} )   \odot  \Gamma  \vdash  t  \GrTTsym{:}   \textstyle (  x  :_{  \textcolor{darkblue}{( s ,  r )}  }  A  )  \to   B }%
}{
 ( \Delta  \mid  \sigma_{{\mathrm{1}}}  \mid  \sigma_{{\mathrm{2}}} )   \odot  \Gamma  \vdash  t  \GrTTsym{=}   \lambda  x . \GrTTsym{(}   t \,{ x }   \GrTTsym{)}   \GrTTsym{:}   \textstyle (  x  :_{  \textcolor{darkblue}{( s ,  r )}  }  A  )  \to   B }{%
{\GrTTdruleTEQXXArrowUniqName}{}%
}}

\newcommand{\GrTTdruleTEQXXFunName}[0]{\GrTTdrulename{TEQ\_Fun}}
\newcommand{\GrTTdruleTEQXXFun}[1]{\GrTTdrule[#1]{%
\GrTTpremise{ ( \Delta  ,  \sigma_{{\mathrm{1}}}  \mid  \sigma_{{\mathrm{2}}}  ,  s  \mid  \sigma_{{\mathrm{3}}}  ,  r )   \odot  \Gamma  \GrTTsym{,}  x  \GrTTsym{:}  A  \vdash  t_{{\mathrm{1}}}  \GrTTsym{=}  t_{{\mathrm{2}}}  \GrTTsym{:}  B}%
}{
 ( \Delta  \mid  \sigma_{{\mathrm{2}}}  \mid  \sigma_{{\mathrm{1}}}  \GrTTsym{+}  \sigma_{{\mathrm{3}}} )   \odot  \Gamma  \vdash   \lambda  x . t_{{\mathrm{1}}}   \GrTTsym{=}   \lambda  x . t_{{\mathrm{2}}}   \GrTTsym{:}   \textstyle (  x  :_{  \textcolor{darkblue}{( s ,  r )}  }  A  )  \to   B }{%
{\GrTTdruleTEQXXFunName}{}%
}}

\newcommand{\GrTTdruleTEQXXAppName}[0]{\GrTTdrulename{TEQ\_App}}
\newcommand{\GrTTdruleTEQXXApp}[1]{\GrTTdrule[#1]{%
\GrTTpremise{ ( \Delta  ,  \sigma_{{\mathrm{3}}}  \mid  \sigma_{{\mathrm{4}}}  ,  r  \mid   \textbf{0}  )   \odot  \Gamma  \GrTTsym{,}  x  \GrTTsym{:}  A  \vdash  B  \GrTTsym{:}   \mathsf{Type}_{ l } }%
\GrTTpremise{ ( \Delta  \mid  \sigma_{{\mathrm{1}}}  \mid  \sigma_{{\mathrm{3}}}  \GrTTsym{+}  \sigma_{{\mathrm{4}}} )   \odot  \Gamma  \vdash  t_{{\mathrm{1}}}  \GrTTsym{=}  t_{{\mathrm{2}}}  \GrTTsym{:}   \textstyle (  x  :_{  \textcolor{darkblue}{( s ,  r )}  }  A  )  \to   B   \quad   ( \Delta  \mid  \sigma_{{\mathrm{2}}}  \mid  \sigma_{{\mathrm{3}}} )   \odot  \Gamma  \vdash  t_{{\mathrm{3}}}  \GrTTsym{=}  t_{{\mathrm{4}}}  \GrTTsym{:}  A}%
}{
 ( \Delta  \mid  \sigma_{{\mathrm{1}}}  \GrTTsym{+}   s  \ast  \sigma_{{\mathrm{2}}}   \mid  \sigma_{{\mathrm{4}}}  \GrTTsym{+}   r  \ast  \sigma_{{\mathrm{2}}}  )   \odot  \Gamma  \vdash   t_{{\mathrm{1}}} \,{ t_{{\mathrm{3}}} }   \GrTTsym{=}   t_{{\mathrm{2}}} \,{ t_{{\mathrm{4}}} }   \GrTTsym{:}  \GrTTsym{[}  t_{{\mathrm{3}}}  \GrTTsym{/}  x  \GrTTsym{]}  B}{%
{\GrTTdruleTEQXXAppName}{}%
}}

\newcommand{\GrTTdruleTEQXXTenName}[0]{\GrTTdrulename{TEQ\_Ten}}
\newcommand{\GrTTdruleTEQXXTen}[1]{\GrTTdrule[#1]{%
\GrTTpremise{ ( \Delta  \mid  \sigma_{{\mathrm{1}}}  \mid   \textbf{0}  )   \odot  \Gamma  \vdash  A  \GrTTsym{=}  C  \GrTTsym{:}   \mathsf{Type}_{ l_{{\mathrm{1}}} }   \quad   ( \Delta  ,  \sigma_{{\mathrm{1}}}  \mid  \sigma_{{\mathrm{2}}}  ,  r  \mid   \textbf{0}  )   \odot  \Gamma  \GrTTsym{,}  x  \GrTTsym{:}  A  \vdash  B  \GrTTsym{=}  D  \GrTTsym{:}   \mathsf{Type}_{ l_{{\mathrm{2}}} } }%
}{
 ( \Delta  \mid  \sigma_{{\mathrm{1}}}  \GrTTsym{+}  \sigma_{{\mathrm{2}}}  \mid   \textbf{0}  )   \odot  \Gamma  \vdash   \textstyle (  x  :_{  \textcolor{darkblue}{ s }  }  A  )  \otimes   B   \GrTTsym{=}   \textstyle (  x  :_{  \textcolor{darkblue}{ s }  }  C  )  \otimes   D   \GrTTsym{:}   \mathsf{Type}_{  l_{{\mathrm{1}}}  \mathop{\sqcup}  l_{{\mathrm{2}}}  } }{%
{\GrTTdruleTEQXXTenName}{}%
}}

\newcommand{\GrTTdruleTEQXXTenCompName}[0]{\GrTTdrulename{TEQ\_TenComp}}
\newcommand{\GrTTdruleTEQXXTenComp}[1]{\GrTTdrule[#1]{%
\GrTTpremise{ ( \Delta  ,  \GrTTsym{(}  \sigma_{{\mathrm{1}}}  \GrTTsym{+}  \sigma_{{\mathrm{2}}}  \GrTTsym{)}  \mid  \sigma_{{\mathrm{5}}}  ,  r'  \mid   \textbf{0}  )   \odot  \Gamma  \GrTTsym{,}  z  \GrTTsym{:}   \textstyle (  x  :_{  \textcolor{darkblue}{ r }  }  A  )  \otimes   B   \vdash  C  \GrTTsym{:}   \mathsf{Type}_{ l } }%
\GrTTpremise{ ( \Delta  \mid  \sigma_{{\mathrm{3}}}  \mid  \sigma_{{\mathrm{1}}} )   \odot  \Gamma  \vdash  t_{{\mathrm{1}}}  \GrTTsym{:}  A  \quad   ( \Delta  \mid  \sigma_{{\mathrm{6}}}  \mid  \sigma_{{\mathrm{2}}}  \GrTTsym{+}   r  \ast  \sigma_{{\mathrm{3}}}  )   \odot  \Gamma  \vdash  t_{{\mathrm{2}}}  \GrTTsym{:}  \GrTTsym{[}  t_{{\mathrm{1}}}  \GrTTsym{/}  x  \GrTTsym{]}  B}%
\GrTTpremise{ ( \Delta  ,  \sigma_{{\mathrm{1}}}  ,  \GrTTsym{(}  \sigma_{{\mathrm{2}}}  ,  r  \GrTTsym{)}  \mid  \sigma_{{\mathrm{4}}}  ,  s  ,  s  \mid  \sigma_{{\mathrm{5}}}  ,  r'  ,  r' )   \odot  \Gamma  \GrTTsym{,}  x  \GrTTsym{:}  A  \GrTTsym{,}  y  \GrTTsym{:}  B  \vdash  t_{{\mathrm{3}}}  \GrTTsym{:}  \GrTTsym{[}  \GrTTsym{(}  x  \GrTTsym{,}  y  \GrTTsym{)}  \GrTTsym{/}  z  \GrTTsym{]}  C}%
}{
 ( \Delta  \mid  \sigma_{{\mathrm{4}}}  \GrTTsym{+}   s  \ast  \GrTTsym{(}  \sigma_{{\mathrm{3}}}  \GrTTsym{+}  \sigma_{{\mathrm{6}}}  \GrTTsym{)}   \mid  \sigma_{{\mathrm{5}}}  \GrTTsym{+}   r'  \ast  \GrTTsym{(}  \sigma_{{\mathrm{3}}}  \GrTTsym{+}  \sigma_{{\mathrm{6}}}  \GrTTsym{)}  )   \odot  \Gamma  \vdash   \mathsf{let}\,  \textstyle (  x  ,  y  )   =  \GrTTsym{(}  t_{{\mathrm{1}}}  \GrTTsym{,}  t_{{\mathrm{2}}}  \GrTTsym{)}  \,\mathsf{in}\, t_{{\mathrm{3}}}   \GrTTsym{=}  \GrTTsym{[}  t_{{\mathrm{1}}}  \GrTTsym{,}  t_{{\mathrm{2}}}  \GrTTsym{/}  x  \GrTTsym{,}  y  \GrTTsym{]}  t_{{\mathrm{3}}}  \GrTTsym{:}  \GrTTsym{[}  \GrTTsym{(}  t_{{\mathrm{1}}}  \GrTTsym{,}  t_{{\mathrm{2}}}  \GrTTsym{)}  \GrTTsym{/}  z  \GrTTsym{]}  C}{%
{\GrTTdruleTEQXXTenCompName}{}%
}}

\newcommand{\GrTTdruleTEQXXPairName}[0]{\GrTTdrulename{TEQ\_Pair}}
\newcommand{\GrTTdruleTEQXXPair}[1]{\GrTTdrule[#1]{%
\GrTTpremise{ ( \Delta  ,  \sigma_{{\mathrm{1}}}  \mid  \sigma_{{\mathrm{3}}}  ,  r  \mid   \textbf{0}  )   \odot  \Gamma  \GrTTsym{,}  x  \GrTTsym{:}  A  \vdash  B  \GrTTsym{:}   \mathsf{Type}_{ l } }%
\GrTTpremise{ ( \Delta  \mid  \sigma_{{\mathrm{2}}}  \mid  \sigma_{{\mathrm{1}}} )   \odot  \Gamma  \vdash  t_{{\mathrm{1}}}  \GrTTsym{=}  t'_{{\mathrm{1}}}  \GrTTsym{:}  A  \quad   ( \Delta  \mid  \sigma_{{\mathrm{4}}}  \mid  \sigma_{{\mathrm{3}}}  \GrTTsym{+}   r  \ast  \sigma_{{\mathrm{2}}}  )   \odot  \Gamma  \vdash  t_{{\mathrm{2}}}  \GrTTsym{=}  t'_{{\mathrm{2}}}  \GrTTsym{:}  \GrTTsym{[}  t_{{\mathrm{1}}}  \GrTTsym{/}  x  \GrTTsym{]}  B}%
}{
 ( \Delta  \mid  \sigma_{{\mathrm{2}}}  \GrTTsym{+}  \sigma_{{\mathrm{4}}}  \mid  \sigma_{{\mathrm{1}}}  \GrTTsym{+}  \sigma_{{\mathrm{3}}} )   \odot  \Gamma  \vdash  \GrTTsym{(}  t_{{\mathrm{1}}}  \GrTTsym{,}  t_{{\mathrm{2}}}  \GrTTsym{)}  \GrTTsym{=}  \GrTTsym{(}  t'_{{\mathrm{1}}}  \GrTTsym{,}  t'_{{\mathrm{2}}}  \GrTTsym{)}  \GrTTsym{:}   \textstyle (  x  :_{  \textcolor{darkblue}{ r }  }  A  )  \otimes   B }{%
{\GrTTdruleTEQXXPairName}{}%
}}

\newcommand{\GrTTdruleTEQXXTenCutName}[0]{\GrTTdrulename{TEQ\_TenCut}}
\newcommand{\GrTTdruleTEQXXTenCut}[1]{\GrTTdrule[#1]{%
\GrTTpremise{ ( \Delta  \mid  \sigma_{{\mathrm{3}}}  \mid  \sigma_{{\mathrm{1}}}  \GrTTsym{+}  \sigma_{{\mathrm{2}}} )   \odot  \Gamma  \vdash  t_{{\mathrm{1}}}  \GrTTsym{=}  t_{{\mathrm{2}}}  \GrTTsym{:}   \textstyle (  x  :_{  \textcolor{darkblue}{ r }  }  A  )  \otimes   B }%
\GrTTpremise{ ( \Delta  ,  \GrTTsym{(}  \sigma_{{\mathrm{1}}}  \GrTTsym{+}  \sigma_{{\mathrm{2}}}  \GrTTsym{)}  \mid  \sigma_{{\mathrm{5}}}  ,  r'  \mid   \textbf{0}  )   \odot  \Gamma  \GrTTsym{,}  z  \GrTTsym{:}   \textstyle (  x  :_{  \textcolor{darkblue}{ r }  }  A  )  \otimes   B   \vdash  C  \GrTTsym{:}   \mathsf{Type}_{ l } }%
\GrTTpremise{ ( \Delta  ,  \sigma_{{\mathrm{1}}}  ,  \GrTTsym{(}  \sigma_{{\mathrm{2}}}  ,  r  \GrTTsym{)}  \mid  \sigma_{{\mathrm{4}}}  ,  s  ,  s  \mid  \sigma_{{\mathrm{5}}}  ,  r'  ,  r' )   \odot  \Gamma  \GrTTsym{,}  x  \GrTTsym{:}  A  \GrTTsym{,}  y  \GrTTsym{:}  B  \vdash  t_{{\mathrm{3}}}  \GrTTsym{=}  t_{{\mathrm{4}}}  \GrTTsym{:}  \GrTTsym{[}  \GrTTsym{(}  x  \GrTTsym{,}  y  \GrTTsym{)}  \GrTTsym{/}  z  \GrTTsym{]}  C}%
}{
 ( \Delta  \mid  \sigma_{{\mathrm{4}}}  \GrTTsym{+}   s  \ast  \sigma_{{\mathrm{3}}}   \mid  \sigma_{{\mathrm{5}}}  \GrTTsym{+}   r'  \ast  \sigma_{{\mathrm{3}}}  )   \odot  \Gamma  \vdash  \GrTTsym{(}   \mathsf{let}\,  \textstyle (  x  ,  y  )   =  t_{{\mathrm{1}}}  \,\mathsf{in}\, t_{{\mathrm{3}}}   \GrTTsym{)}  \GrTTsym{=}  \GrTTsym{(}   \mathsf{let}\,  \textstyle (  x  ,  y  )   =  t_{{\mathrm{2}}}  \,\mathsf{in}\, t_{{\mathrm{4}}}   \GrTTsym{)}  \GrTTsym{:}  \GrTTsym{[}  t_{{\mathrm{1}}}  \GrTTsym{/}  z  \GrTTsym{]}  C}{%
{\GrTTdruleTEQXXTenCutName}{}%
}}

\newcommand{\GrTTdruleTEQXXTenUName}[0]{\GrTTdrulename{TEQ\_TenU}}
\newcommand{\GrTTdruleTEQXXTenU}[1]{\GrTTdrule[#1]{%
\GrTTpremise{ ( \Delta  \mid  \sigma_{{\mathrm{1}}}  \mid  \sigma_{{\mathrm{2}}} )   \odot  \Gamma  \vdash  t  \GrTTsym{:}   \textstyle (  x  :_{  \textcolor{darkblue}{ r }  }  A  )  \otimes   B }%
}{
 ( \Delta  \mid  \sigma_{{\mathrm{1}}}  \mid  \sigma_{{\mathrm{2}}} )   \odot  \Gamma  \vdash  t  \GrTTsym{=}  \GrTTsym{(}   \mathsf{let}\,  \textstyle (  x  ,  y  )   =  t  \,\mathsf{in}\, \GrTTsym{(}  x  \GrTTsym{,}  y  \GrTTsym{)}   \GrTTsym{)}  \GrTTsym{:}   \textstyle (  x  :_{  \textcolor{darkblue}{ r }  }  A  )  \otimes   B }{%
{\GrTTdruleTEQXXTenUName}{}%
}}

\newcommand{\GrTTdruleTEQXXBoxName}[0]{\GrTTdrulename{TEQ\_Box}}
\newcommand{\GrTTdruleTEQXXBox}[1]{\GrTTdrule[#1]{%
\GrTTpremise{ ( \Delta  \mid  \sigma  \mid   \textbf{0}  )   \odot  \Gamma  \vdash  A  \GrTTsym{=}  B  \GrTTsym{:}   \mathsf{Type}_{ l } }%
}{
 ( \Delta  \mid  \sigma  \mid   \textbf{0}  )   \odot  \Gamma  \vdash   \square_{  \textcolor{darkblue}{ s }  }  A   \GrTTsym{=}   \square_{  \textcolor{darkblue}{ s }  }  B   \GrTTsym{:}   \mathsf{Type}_{ l } }{%
{\GrTTdruleTEQXXBoxName}{}%
}}

\newcommand{\GrTTdruleTEQXXBoxIName}[0]{\GrTTdrulename{TEQ\_BoxI}}
\newcommand{\GrTTdruleTEQXXBoxI}[1]{\GrTTdrule[#1]{%
\GrTTpremise{ ( \Delta  \mid  \sigma_{{\mathrm{1}}}  \mid  \sigma_{{\mathrm{2}}} )   \odot  \Gamma  \vdash  t_{{\mathrm{1}}}  \GrTTsym{=}  t_{{\mathrm{2}}}  \GrTTsym{:}  A}%
}{
 ( \Delta  \mid   s  \ast  \sigma_{{\mathrm{1}}}   \mid  \sigma_{{\mathrm{2}}} )   \odot  \Gamma  \vdash   \square  t_{{\mathrm{1}}}   \GrTTsym{=}   \square  t_{{\mathrm{2}}}   \GrTTsym{:}   \square_{  \textcolor{darkblue}{ s }  }  A }{%
{\GrTTdruleTEQXXBoxIName}{}%
}}

\newcommand{\GrTTdruleTEQXXBoxBName}[0]{\GrTTdrulename{TEQ\_BoxB}}
\newcommand{\GrTTdruleTEQXXBoxB}[1]{\GrTTdrule[#1]{%
\GrTTpremise{ ( \Delta  \mid  \sigma_{{\mathrm{1}}}  \mid  \sigma_{{\mathrm{2}}} )   \odot  \Gamma  \vdash  t_{{\mathrm{1}}}  \GrTTsym{:}  A}%
\GrTTpremise{ ( \Delta  ,  \sigma_{{\mathrm{2}}}  \mid  \sigma_{{\mathrm{4}}}  ,  r  \mid   \textbf{0}  )   \odot  \Gamma  \GrTTsym{,}  z  \GrTTsym{:}   \square_{  \textcolor{darkblue}{ s }  }  A   \vdash  B  \GrTTsym{:}   \mathsf{Type}_{ l } }%
\GrTTpremise{ ( \Delta  ,  \sigma_{{\mathrm{2}}}  \mid  \sigma_{{\mathrm{3}}}  ,  s  \mid  \sigma_{{\mathrm{4}}}  ,  \GrTTsym{(}   s  \ast  r   \GrTTsym{)} )   \odot  \Gamma  \GrTTsym{,}  x  \GrTTsym{:}  A  \vdash  t_{{\mathrm{2}}}  \GrTTsym{:}  \GrTTsym{[}   \square  x   \GrTTsym{/}  z  \GrTTsym{]}  B}%
}{
 ( \Delta  \mid  \sigma_{{\mathrm{3}}}  \GrTTsym{+}   s  \ast  \sigma_{{\mathrm{1}}}   \mid  \sigma_{{\mathrm{4}}}  \GrTTsym{+}     s  \ast  r    \ast  \sigma_{{\mathrm{1}}}  )   \odot  \Gamma  \vdash  \GrTTsym{(}   \mathsf{let}\,  \square  x   =   \square  t_{{\mathrm{1}}}   \,\mathsf{in}\, t_{{\mathrm{2}}}   \GrTTsym{)}  \GrTTsym{=}  \GrTTsym{[}  t_{{\mathrm{1}}}  \GrTTsym{/}  x  \GrTTsym{]}  t_{{\mathrm{2}}}  \GrTTsym{:}  \GrTTsym{[}   \square  t_{{\mathrm{1}}}   \GrTTsym{/}  z  \GrTTsym{]}  B}{%
{\GrTTdruleTEQXXBoxBName}{}%
}}

\newcommand{\GrTTdruleTEQXXBoxEName}[0]{\GrTTdrulename{TEQ\_BoxE}}
\newcommand{\GrTTdruleTEQXXBoxE}[1]{\GrTTdrule[#1]{%
\GrTTpremise{ ( \Delta  ,  \sigma_{{\mathrm{2}}}  \mid  \sigma_{{\mathrm{4}}}  ,  r  \mid   \textbf{0}  )   \odot  \Gamma  \GrTTsym{,}  z  \GrTTsym{:}   \square_{  \textcolor{darkblue}{ s }  }  A   \vdash  B  \GrTTsym{:}   \mathsf{Type}_{ l } }%
\GrTTpremise{ ( \Delta  \mid  \sigma_{{\mathrm{1}}}  \mid  \sigma_{{\mathrm{2}}} )   \odot  \Gamma  \vdash  t_{{\mathrm{1}}}  \GrTTsym{=}  t_{{\mathrm{2}}}  \GrTTsym{:}   \square_{  \textcolor{darkblue}{ s }  }  A   \quad   ( \Delta  ,  \sigma_{{\mathrm{2}}}  \mid  \sigma_{{\mathrm{3}}}  ,  s  \mid  \sigma_{{\mathrm{4}}}  ,  \GrTTsym{(}   s  \ast  r   \GrTTsym{)} )   \odot  \Gamma  \GrTTsym{,}  x  \GrTTsym{:}  A  \vdash  t_{{\mathrm{3}}}  \GrTTsym{=}  t_{{\mathrm{4}}}  \GrTTsym{:}  \GrTTsym{[}   \square  x   \GrTTsym{/}  z  \GrTTsym{]}  B}%
}{
 ( \Delta  \mid  \sigma_{{\mathrm{1}}}  \GrTTsym{+}  \sigma_{{\mathrm{3}}}  \mid  \sigma_{{\mathrm{4}}}  \GrTTsym{+}   r  \ast  \sigma_{{\mathrm{1}}}  )   \odot  \Gamma  \vdash  \GrTTsym{(}   \mathsf{let}\,  \square  x   =  t_{{\mathrm{1}}}  \,\mathsf{in}\, t_{{\mathrm{3}}}   \GrTTsym{)}  \GrTTsym{=}  \GrTTsym{(}   \mathsf{let}\,  \square  x   =  t_{{\mathrm{2}}}  \,\mathsf{in}\, t_{{\mathrm{4}}}   \GrTTsym{)}  \GrTTsym{:}  \GrTTsym{[}  t_{{\mathrm{1}}}  \GrTTsym{/}  z  \GrTTsym{]}  B}{%
{\GrTTdruleTEQXXBoxEName}{}%
}}

\newcommand{\GrTTdruleTEQXXBoxUName}[0]{\GrTTdrulename{TEQ\_BoxU}}
\newcommand{\GrTTdruleTEQXXBoxU}[1]{\GrTTdrule[#1]{%
\GrTTpremise{ ( \Delta  \mid  \sigma_{{\mathrm{1}}}  \mid  \sigma_{{\mathrm{2}}} )   \odot  \Gamma  \vdash  t  \GrTTsym{:}   \square_{  \textcolor{darkblue}{ s }  }  A }%
}{
 ( \Delta  \mid  \sigma_{{\mathrm{1}}}  \mid  \sigma_{{\mathrm{2}}} )   \odot  \Gamma  \vdash  t  \GrTTsym{=}  \GrTTsym{(}   \mathsf{let}\,  \square  x   =  t  \,\mathsf{in}\,  \square  x    \GrTTsym{)}  \GrTTsym{:}   \square_{  \textcolor{darkblue}{ s }  }  A }{%
{\GrTTdruleTEQXXBoxUName}{}%
}}

\renewcommand{\GrTTdrule}[4][]{{\displaystyle\frac{\begin{array}{l}#2\end{array}}{#3}\,{#4}}}
\renewcommand{\GrTTpremise}[1]{ #1 \\}
\renewcommand{\GrTTsym}[1]{#1}

\providecommand{\GrTTdruleTXXTypeName}{}

\providecommand{\GrTTdruleTXXArrowName}{}
\providecommand{\GrTTdruleTXXTenName}{}
\providecommand{\GrTTdruleTXXVarName}{}
\providecommand{\GrTTdruleTXXFunName}{}
\providecommand{\GrTTdruleTXXAppName}{}
\providecommand{\GrTTdruleTXXPairName}{}
\providecommand{\GrTTdruleTXXTenCutName}{}

\providecommand{\GrTTdruleTXXBoxName}{}
\providecommand{\GrTTdruleTXXBoxIName}{}
\providecommand{\GrTTdruleTXXBoxEName}{}
\providecommand{\GrTTdruleWfXXEmptyName}{}
\providecommand{\GrTTdruleWfXXExtName}{}
\providecommand{\GrTTdruleChkAlgXXFunName}{}
\providecommand{\GrTTdruleInfAlgXXAppName}{}

\providecommand{\GrTTdruleSemXXBetaBoxName}{}

\renewcommand{\GrTTdruleSemXXBetaBoxName}{\beta\Box}

\providecommand{\GrTTdruleRhoXXEmptyName}{}
\providecommand{\GrTTdruleRhoXXExtTmName}{}
\providecommand{\GrTTdruleRhoXXExtTyName}{}

\renewcommand{\GrTTdruleRhoXXEmptyName}[0]{\textsc{E}}
\renewcommand{\GrTTdruleRhoXXExtTmName}[0]{\textsc{Tm}}
\renewcommand{\GrTTdruleRhoXXExtTyName}[0]{\textsc{Ty}}

\providecommand{\GrTTdruleEpXXEmptyName}{}
\providecommand{\GrTTdruleEpXXExtTmName}{}
\providecommand{\GrTTdruleEpXXExtTyName}{}

\renewcommand{\GrTTdruleEpXXEmptyName}[0]{\textsc{E}}
\renewcommand{\GrTTdruleEpXXExtTmName}[0]{\textsc{Tm}}
\renewcommand{\GrTTdruleEpXXExtTyName}[0]{\textsc{Ty}}

\renewcommand{\GrTTdruleTXXTypeName}[0]{{\footnotesize{\textsc{Type}}}}

\renewcommand{\GrTTdruleTXXArrowName}[0]{\multimap}
\renewcommand{\GrTTdruleTXXTenName}[0]{\otimes}
\renewcommand{\GrTTdruleTXXVarName}{{\footnotesize{\textsc{Var}}}}
\renewcommand{\GrTTdruleTXXFunName}{\lambda_i}
\renewcommand{\GrTTdruleTXXAppName}{\lambda_e}
\renewcommand{\GrTTdruleTXXPairName}{\otimes_i}
\renewcommand{\GrTTdruleTXXTenCutName}{\otimes_e}

\renewcommand{\GrTTdruleTXXBoxName}{\square}
\renewcommand{\GrTTdruleTXXBoxIName}{\square_i}
\renewcommand{\GrTTdruleTXXBoxEName}{\square_e}
\renewcommand{\GrTTdruleTXXTyConvName}{\footnotesize{\textsc{Conv}}}
\renewcommand{\GrTTdruleWfXXEmptyName}{\textsc{WfEmp}}
\renewcommand{\GrTTdruleWfXXExtName}{\textsc{WfExt}}
\renewcommand{\GrTTdruleChkAlgXXFunName}{\Leftarrow\lambda_i}
\renewcommand{\GrTTdruleInfAlgXXAppName}{\Rightarrow\lambda_e}

\providecommand{\GrTTdruleTEQXXReflName}{}
\providecommand{\GrTTdruleTEQXXTransName}{}
\providecommand{\GrTTdruleTEQXXSymName}{}
\providecommand{\GrTTdruleTEQXXConvTyName}{}
\providecommand{\GrTTdruleTEQXXArrowName}{}
\providecommand{\GrTTdruleTEQXXArrowCompName}{}
\providecommand{\GrTTdruleTEQXXArrowUniqName}{}
\providecommand{\GrTTdruleTEQXXFunName}{}
\providecommand{\GrTTdruleTEQXXAppName}{}
\providecommand{\GrTTdruleTEQXXTenName}{}
\providecommand{\GrTTdruleTEQXXTenCompName}{}
\providecommand{\GrTTdruleTEQXXPairName}{}
\providecommand{\GrTTdruleTEQXXTenCutName}{}
\providecommand{\GrTTdruleTEQXXTenUName}{}
\providecommand{\GrTTdruleTEQXXBoxName}{}
\providecommand{\GrTTdruleTEQXXBoxIName}{}
\providecommand{\GrTTdruleTEQXXBoxBName}{}
\providecommand{\GrTTdruleTEQXXBoxEName}{}
\providecommand{\GrTTdruleTEQXXBoxUName}{}

\renewcommand{\GrTTdruleTEQXXReflName}{\footnotesize{\textsc{Refl}}}
\renewcommand{\GrTTdruleTEQXXTransName}{\footnotesize{\textsc{Trans}}}
\renewcommand{\GrTTdruleTEQXXSymName}{\footnotesize{\textsc{Sym}}}
\renewcommand{\GrTTdruleTEQXXConvTyName}{\footnotesize{\textsc{ConvTy}}}
\renewcommand{\GrTTdruleTEQXXArrowName}{\footnotesize{\textsc{Eq}}_{\rightarrow}}
\renewcommand{\GrTTdruleTEQXXArrowCompName}{\footnotesize{\textsc{Eq}}_{\rightarrow_c}}
\renewcommand{\GrTTdruleTEQXXArrowUniqName}{\footnotesize{\textsc{Eq}}_{\rightarrow_u}}
\renewcommand{\GrTTdruleTEQXXFunName}{\footnotesize{\textsc{Eq}}_{\rightarrow_i}}
\renewcommand{\GrTTdruleTEQXXAppName}{\footnotesize{\textsc{Eq}}_{\rightarrow_e}}
\renewcommand{\GrTTdruleTEQXXTenName}{\footnotesize{\textsc{Eq}}_\otimes}
\renewcommand{\GrTTdruleTEQXXTenCompName}{\footnotesize{\textsc{Eq}}_{\otimes_c}}
\renewcommand{\GrTTdruleTEQXXPairName}{\footnotesize{\textsc{Eq}}_{\otimes_i}}
\renewcommand{\GrTTdruleTEQXXTenCutName}{\footnotesize{\textsc{Eq}}_{\otimes_e}}
\renewcommand{\GrTTdruleTEQXXTenUName}{\footnotesize{\textsc{Eq}}_{\otimes_u}}
\renewcommand{\GrTTdruleTEQXXBoxName}{\footnotesize{\textsc{Eq}}_{\square}}
\renewcommand{\GrTTdruleTEQXXBoxIName}{\footnotesize{\textsc{Eq}}_{\square_i}}
\renewcommand{\GrTTdruleTEQXXBoxBName}{\footnotesize{\textsc{Eq}}_{\square_c}}
\renewcommand{\GrTTdruleTEQXXBoxEName}{\footnotesize{\textsc{Eq}}_{\square_e}}
\renewcommand{\GrTTdruleTEQXXBoxUName}{\footnotesize{\textsc{Eq}}_{\square_u}}

\providecommand{\GrTTdruleSTXXEqName}{}
\providecommand{\GrTTdruleSTXXTransName}{}
\providecommand{\GrTTdruleSTXXTyName}{}
\providecommand{\GrTTdruleSTXXArrowName}{}
\providecommand{\GrTTdruleSTXXTenName}{}
\providecommand{\GrTTdruleSTXXBoxName}{}

\renewcommand{\GrTTdruleSTXXEqName}{\leq_{\footnotesize{\textsc{Eq}}}}
\renewcommand{\GrTTdruleSTXXTransName}{\leq_{\footnotesize{\textsc{Trans}}}}
\renewcommand{\GrTTdruleSTXXTyName}{\leq_{\footnotesize{\textsc{Type}}}}
\renewcommand{\GrTTdruleSTXXArrowName}{\leq_\to}
\renewcommand{\GrTTdruleSTXXTenName}{\leq_\otimes}
\renewcommand{\GrTTdruleSTXXBoxName}{\leq_\square}

\renewcommand{\GrTTdruleWfXXEmpty}[1]{\GrTTdrule[#1]{%
}{ \emptyset   \odot   \emptyset   \vdash}{
  {\GrTTdruleWfXXEmptyName{}}{}%
}}

\renewcommand{\GrTTdruleWfXXExt}[1]{\GrTTdrule[#1]{%
  { ( \Delta  \mid  \sigma  \mid   \textbf{0}  )   \odot  \Gamma  \vdash  A  \GrTTsym{:}   \mathsf{Type}_{ l } }
}{\Delta  ,  \sigma  \odot  \Gamma  \GrTTsym{,}  x  \GrTTsym{:}  A  \vdash}{
  {\GrTTdruleWfXXExtName{}}{}%
}}

\renewcommand{\GrTTdruleTXXType}[1]{\GrTTdrule[#1]{%
  \Delta  \odot  \Gamma  \vdash
}{ ( \Delta  \mid   \textbf{0}   \mid   \textbf{0}  )   \odot  \Gamma  \vdash   \mathsf{Type}_{ l }   \GrTTsym{:}   \mathsf{Type}_{   \mathsf{suc}\  l   } }{
  {\GrTTdruleTXXTypeName{}}{}%
}}

\renewcommand{\GrTTdruleTXXArrow}[1]{\GrTTdrule[#1]{%
    \GrTTpremise{
       ( \Delta  \mid  \sigma_{{\mathrm{1}}}  \mid   \textbf{0}  )   \odot  \Gamma  \vdash  A  \GrTTsym{:}   \mathsf{Type}_{ l_{{\mathrm{1}}} } 
      \quad
       ( \Delta  ,  \sigma_{{\mathrm{1}}}  \mid  \sigma_{{\mathrm{2}}}  ,  r  \mid   \textbf{0}  )   \odot  \Gamma  \GrTTsym{,}  x  \GrTTsym{:}  A  \vdash  B  \GrTTsym{:}   \mathsf{Type}_{ l_{{\mathrm{2}}} } }
}{ ( \Delta  \mid  \sigma_{{\mathrm{1}}}  \GrTTsym{+}  \sigma_{{\mathrm{2}}}  \mid   \textbf{0}  )   \odot  \Gamma  \vdash   \textstyle (  x  :_{  \textcolor{darkblue}{( s ,  r )}  }  A  )  \to   B   \GrTTsym{:}   \mathsf{Type}_{  l_{{\mathrm{1}}}  \mathop{\sqcup}  l_{{\mathrm{2}}}  } }{%
{\GrTTdruleTXXArrowName{}}{}%
}}

\renewcommand{\GrTTdruleTXXTen}[1]{\GrTTdrule[#1]{%
    \GrTTpremise{
       ( \Delta  \mid  \sigma_{{\mathrm{1}}}  \mid   \textbf{0}  )   \odot  \Gamma  \vdash  A  \GrTTsym{:}   \mathsf{Type}_{ l } 
      \quad
       ( \Delta  ,  \sigma_{{\mathrm{1}}}  \mid  \sigma_{{\mathrm{2}}}  ,  r  \mid   \textbf{0}  )   \odot  \Gamma  \GrTTsym{,}  x  \GrTTsym{:}  A  \vdash  B  \GrTTsym{:}   \mathsf{Type}_{ l } }
}{ ( \Delta  \mid  \sigma_{{\mathrm{1}}}  \GrTTsym{+}  \sigma_{{\mathrm{2}}}  \mid   \textbf{0}  )   \odot  \Gamma  \vdash   \textstyle (  x  :_{  \textcolor{darkblue}{ r }  }  A  )  \otimes   B   \GrTTsym{:}   \mathsf{Type}_{ l } }{%
{\GrTTdruleTXXTenName{}}{}%
}}

\renewcommand{\GrTTdruleTXXVar}[1]{\GrTTdrule[#1]{%
    \GrTTpremise{
      {\GrTTsym{(}  \Delta_{{\mathrm{1}}}  ,  \sigma  ,  \Delta_{{\mathrm{2}}}  \GrTTsym{)}  \odot  \Gamma_{{\mathrm{1}}}  \GrTTsym{,}  x  \GrTTsym{:}  A  \GrTTsym{,}  \Gamma_{{\mathrm{2}}}  \vdash}
      \quad
      { \left|  \Delta_{{\mathrm{1}}}  \right|   \GrTTsym{=}   \left|  \Gamma_{{\mathrm{1}}}  \right| }
    }%
}{ ( \Delta_{{\mathrm{1}}}  ,  \sigma  ,  \Delta_{{\mathrm{2}}}  \mid   \textbf{0}^{  \left|  \Delta_{{\mathrm{1}}}  \right|  }   ,  \GrTTsym{1}  ,   \textbf{0}   \mid  \sigma  ,  \GrTTsym{0}  ,   \textbf{0}  )   \odot  \Gamma_{{\mathrm{1}}}  \GrTTsym{,}  x  \GrTTsym{:}  A  \GrTTsym{,}  \Gamma_{{\mathrm{2}}}  \vdash  x  \GrTTsym{:}  A}{%
{\GrTTdruleTXXVarName{}}{}%
}}

\renewcommand{\GrTTdruleTXXFun}[1]{\GrTTdrule[#1]{%
    \GrTTpremise{
       ( \Delta  ,  \sigma_{{\mathrm{1}}}  \mid  \sigma_{{\mathrm{3}}}  ,  r  \mid   \textbf{0}  )   \odot  \Gamma  \GrTTsym{,}  x  \GrTTsym{:}  A  \vdash  B  \GrTTsym{:}   \mathsf{Type}_{ l } %
      \quad
       ( \Delta  ,  \sigma_{{\mathrm{1}}}  \mid  \sigma_{{\mathrm{2}}}  ,  s  \mid  \sigma_{{\mathrm{3}}}  ,  r )   \odot  \Gamma  \GrTTsym{,}  x  \GrTTsym{:}  A  \vdash  t  \GrTTsym{:}  B}%
}{ ( \Delta  \mid  \sigma_{{\mathrm{2}}}  \mid  \sigma_{{\mathrm{1}}}  \GrTTsym{+}  \sigma_{{\mathrm{3}}} )   \odot  \Gamma  \vdash   \lambda  x . t   \GrTTsym{:}   \textstyle (  x  :_{  \textcolor{darkblue}{( s ,  r )}  }  A  )  \to   B }{%
{\GrTTdruleTXXFunName{}}{}%
}}

\renewcommand{\GrTTdruleTXXApp}[1]{\GrTTdrule[#1]{%
    \GrTTpremise{
       ( \Delta  ,  \sigma_{{\mathrm{1}}}  \mid  \sigma_{{\mathrm{3}}}  ,  r  \mid   \textbf{0}  )   \odot  \Gamma  \GrTTsym{,}  x  \GrTTsym{:}  A  \vdash  B  \GrTTsym{:}   \mathsf{Type}_{ l } 
      \\
       ( \Delta  \mid  \sigma_{{\mathrm{2}}}  \mid  \sigma_{{\mathrm{1}}}  \GrTTsym{+}  \sigma_{{\mathrm{3}}} )   \odot  \Gamma  \vdash  t_{{\mathrm{1}}}  \GrTTsym{:}   \textstyle (  x  :_{  \textcolor{darkblue}{( s ,  r )}  }  A  )  \to   B 
      \quad
       ( \Delta  \mid  \sigma_{{\mathrm{4}}}  \mid  \sigma_{{\mathrm{1}}} )   \odot  \Gamma  \vdash  t_{{\mathrm{2}}}  \GrTTsym{:}  A}
}{ ( \Delta  \mid  \sigma_{{\mathrm{2}}}  \GrTTsym{+}   s  \ast  \sigma_{{\mathrm{4}}}   \mid  \sigma_{{\mathrm{3}}}  \GrTTsym{+}   r  \ast  \sigma_{{\mathrm{4}}}  )   \odot  \Gamma  \vdash   t_{{\mathrm{1}}} \,{ t_{{\mathrm{2}}} }   \GrTTsym{:}  \GrTTsym{[}  t_{{\mathrm{2}}}  \GrTTsym{/}  x  \GrTTsym{]}  B}{%
{\GrTTdruleTXXAppName{}}{}%
}}

\renewcommand{\GrTTdruleTXXPair}[1]{\GrTTdrule[#1]{%
    \GrTTpremise{
       ( \Delta  ,  \sigma_{{\mathrm{1}}}  \mid  \sigma_{{\mathrm{3}}}  ,  r  \mid   \textbf{0}  )   \odot  \Gamma  \GrTTsym{,}  x  \GrTTsym{:}  A  \vdash  B  \GrTTsym{:}   \mathsf{Type}_{ l } 
      \\
       ( \Delta  \mid  \sigma_{{\mathrm{2}}}  \mid  \sigma_{{\mathrm{1}}} )   \odot  \Gamma  \vdash  t_{{\mathrm{1}}}  \GrTTsym{:}  A
      \qquad
       ( \Delta  \mid  \sigma_{{\mathrm{4}}}  \mid  \sigma_{{\mathrm{3}}}  \GrTTsym{+}   r  \ast  \sigma_{{\mathrm{2}}}  )   \odot  \Gamma  \vdash  t_{{\mathrm{2}}}  \GrTTsym{:}  \GrTTsym{[}  t_{{\mathrm{1}}}  \GrTTsym{/}  x  \GrTTsym{]}  B}
}{ ( \Delta  \mid  \sigma_{{\mathrm{2}}}  \GrTTsym{+}  \sigma_{{\mathrm{4}}}  \mid  \sigma_{{\mathrm{1}}}  \GrTTsym{+}  \sigma_{{\mathrm{3}}} )   \odot  \Gamma  \vdash  \GrTTsym{(}  t_{{\mathrm{1}}}  \GrTTsym{,}  t_{{\mathrm{2}}}  \GrTTsym{)}  \GrTTsym{:}   \textstyle (  x  :_{  \textcolor{darkblue}{ r }  }  A  )  \otimes   B }{%
{\GrTTdruleTXXPairName{}}{}%
}}

\renewcommand{\GrTTdruleTXXTenCut}[1]{\GrTTdrule[#1]{%
\GrTTpremise{ ( \Delta  \mid  \sigma_{{\mathrm{3}}}  \mid  \sigma_{{\mathrm{1}}}  \GrTTsym{+}  \sigma_{{\mathrm{2}}} )   \odot  \Gamma  \vdash  t_{{\mathrm{1}}}  \GrTTsym{:}   \textstyle (  x  :_{  \textcolor{darkblue}{ r }  }  A  )  \otimes   B }
\GrTTpremise{ ( \Delta  ,  \GrTTsym{(}  \sigma_{{\mathrm{1}}}  \GrTTsym{+}  \sigma_{{\mathrm{2}}}  \GrTTsym{)}  \mid  \sigma_{{\mathrm{5}}}  ,  r'  \mid   \textbf{0}  )   \odot  \Gamma  \GrTTsym{,}  z  \GrTTsym{:}   \textstyle (  x  :_{  \textcolor{darkblue}{ r }  }  A  )  \otimes   B   \vdash  C  \GrTTsym{:}   \mathsf{Type}_{ l } }
\GrTTpremise{ ( \Delta  ,  \sigma_{{\mathrm{1}}}  ,  \GrTTsym{(}  \sigma_{{\mathrm{2}}}  ,  r  \GrTTsym{)}  \mid  \sigma_{{\mathrm{4}}}  ,  s  ,  s  \mid  \sigma_{{\mathrm{5}}}  ,  r'  ,  r' )   \odot  \Gamma  \GrTTsym{,}  x  \GrTTsym{:}  A  \GrTTsym{,}  y  \GrTTsym{:}  B  \vdash  t_{{\mathrm{2}}}  \GrTTsym{:}  \GrTTsym{[}  \GrTTsym{(}  x  \GrTTsym{,}  y  \GrTTsym{)}  \GrTTsym{/}  z  \GrTTsym{]}  C}
}{ ( \Delta  \mid  \sigma_{{\mathrm{4}}}  \GrTTsym{+}   s  \ast  \sigma_{{\mathrm{3}}}   \mid  \sigma_{{\mathrm{5}}}  \GrTTsym{+}   r'  \ast  \sigma_{{\mathrm{3}}}  )   \odot  \Gamma  \vdash   \mathsf{let}\,  \textstyle (  x  ,  y  )   =  t_{{\mathrm{1}}}  \,\mathsf{in}\, t_{{\mathrm{2}}}   \GrTTsym{:}  \GrTTsym{[}  t_{{\mathrm{1}}}  \GrTTsym{/}  z  \GrTTsym{]}  C}{%
{\GrTTdruleTXXTenCutName{}}{}%
}}

\renewcommand{\GrTTdruleTXXBox}[1]{\GrTTdrule[#1]{%
    \GrTTpremise{ ( \Delta  \mid  \sigma  \mid   \textbf{0}  )   \odot  \Gamma  \vdash  A  \GrTTsym{:}   \mathsf{Type}_{ l } }%
  }{ ( \Delta  \mid  \sigma  \mid   \textbf{0}  )   \odot  \Gamma  \vdash   \square_{  \textcolor{darkblue}{ s }  }  A   \GrTTsym{:}   \mathsf{Type}_{ l } }{%
    {\GrTTdruleTXXBoxName{}}{}%
}}

\renewcommand{\GrTTdruleTXXBoxI}[1]{\GrTTdrule[#1]{%
\GrTTpremise{ ( \Delta  \mid  \sigma_{{\mathrm{1}}}  \mid  \sigma_{{\mathrm{2}}} )   \odot  \Gamma  \vdash  t  \GrTTsym{:}  A}%
}{ ( \Delta  \mid   s  \ast  \sigma_{{\mathrm{1}}}   \mid  \sigma_{{\mathrm{2}}} )   \odot  \Gamma  \vdash   \square  t   \GrTTsym{:}   \square_{  \textcolor{darkblue}{ s }  }  A }{
{\GrTTdruleTXXBoxIName{}}{}%
}}

\renewcommand{\GrTTdruleTXXBoxE}[1]{\GrTTdrule[#1]{%
    \GrTTpremise{
       ( \Delta  ,  \sigma_{{\mathrm{2}}}  \mid  \sigma_{{\mathrm{4}}}  ,  r  \mid   \textbf{0}  )   \odot  \Gamma  \GrTTsym{,}  z  \GrTTsym{:}   \square_{  \textcolor{darkblue}{ s }  }  A   \vdash  B  \GrTTsym{:}   \mathsf{Type}_{ l } %
      \\
       ( \Delta  \mid  \sigma_{{\mathrm{1}}}  \mid  \sigma_{{\mathrm{2}}} )   \odot  \Gamma  \vdash  t_{{\mathrm{1}}}  \GrTTsym{:}   \square_{  \textcolor{darkblue}{ s }  }  A %
      \quad
       ( \Delta  ,  \sigma_{{\mathrm{2}}}  \mid  \sigma_{{\mathrm{3}}}  ,  s  \mid  \sigma_{{\mathrm{4}}}  ,  \GrTTsym{(}   s  \ast  r   \GrTTsym{)} )   \odot  \Gamma  \GrTTsym{,}  x  \GrTTsym{:}  A  \vdash  t_{{\mathrm{2}}}  \GrTTsym{:}  \GrTTsym{[}   \square  x   \GrTTsym{/}  z  \GrTTsym{]}  B}%
}{ ( \Delta  \mid  \sigma_{{\mathrm{1}}}  \GrTTsym{+}  \sigma_{{\mathrm{3}}}  \mid  \sigma_{{\mathrm{4}}}  \GrTTsym{+}   r  \ast  \sigma_{{\mathrm{1}}}  )   \odot  \Gamma  \vdash   \mathsf{let}\,  \square  x   =  t_{{\mathrm{1}}}  \,\mathsf{in}\, t_{{\mathrm{2}}}   \GrTTsym{:}  \GrTTsym{[}  t_{{\mathrm{1}}}  \GrTTsym{/}  z  \GrTTsym{]}  B}{
{\GrTTdruleTXXBoxEName{}}{}%
}}

\renewcommand{\GrTTdruleTXXTyConv}[1]{\GrTTdrule[#1]{%
    \GrTTpremise{
       ( \Delta  \mid  \sigma_{{\mathrm{1}}}  \mid  \sigma_{{\mathrm{2}}} )   \odot  \Gamma  \vdash  t  \GrTTsym{:}  A%
      \quad
       ( \Delta  \mid  \sigma_{{\mathrm{2}}} )   \odot  \Gamma  \vdash  A  \leq  B}%
}{ ( \Delta  \mid  \sigma_{{\mathrm{1}}}  \mid  \sigma_{{\mathrm{2}}} )   \odot  \Gamma  \vdash  t  \GrTTsym{:}  B}{
{\GrTTdruleTXXTyConvName{}}{}%
}}

\renewcommand{\GrTTdruleChkAlgXXFun}[1]{\GrTTdrule[#1]{%
    \GrTTpremise{
       \Delta  ;  \Gamma  \vdash  A  \Rightarrow   \mathsf{Type}_{ l }   ;  \sigma_{{\mathrm{1}}}  ;   \textbf{0}  %
      \qquad
       \Delta  ,  \sigma_{{\mathrm{1}}}  ;  \Gamma  \GrTTsym{,}  x  \GrTTsym{:}  A  \vdash  t  \Leftarrow  B  ;  \sigma_{{\mathrm{2}}}  ,  s  ;  \sigma_{{\mathrm{3}}}  ,  r }%
}{ \Delta  ;  \Gamma  \vdash   \lambda  x . t   \Leftarrow   \textstyle (  x  :_{  \textcolor{darkblue}{( s ,  r )}  }  A  )  \to   B   ;  \sigma_{{\mathrm{2}}}  ;  \sigma_{{\mathrm{1}}}  \GrTTsym{+}  \sigma_{{\mathrm{3}}} }{
{\GrTTdruleChkAlgXXFunName{}}{}%
}}

\providecommand{\GrTTdruleInfAlgXXApp}{}
\renewcommand{\GrTTdruleInfAlgXXApp}[1]{\GrTTdrule[#1]{%
    \GrTTpremise{
       \Delta  ;  \Gamma  \vdash  t_{{\mathrm{1}}}  \Rightarrow   \textstyle (  x  :_{  \textcolor{darkblue}{( s ,  r )}  }  A  )  \to   B   ;  \sigma_{{\mathrm{2}}}  ;  \sigma_{{\mathrm{13}}} %
      \\
       \Delta  ;  \Gamma  \vdash  t_{{\mathrm{2}}}  \Leftarrow  A  ;  \sigma_{{\mathrm{4}}}  ;  \sigma_{{\mathrm{1}}} %
      \\
       \Delta  ,  \sigma_{{\mathrm{1}}}  ;  \Gamma  \GrTTsym{,}  x  \GrTTsym{:}  A  \vdash  B  \Rightarrow   \mathsf{Type}_{ l }   ;   \sigma_{{\mathrm{3}}}  ,  r   ;   \textbf{0}  %
      \qquad
      \sigma_{{\mathrm{13}}}  \GrTTsym{=}  \sigma_{{\mathrm{1}}}  \GrTTsym{+}  \sigma_{{\mathrm{3}}}}
}{ \Delta  ;  \Gamma  \vdash   t_{{\mathrm{1}}} \,{ t_{{\mathrm{2}}} }   \Rightarrow  \GrTTsym{[}  t_{{\mathrm{2}}}  \GrTTsym{/}  x  \GrTTsym{]}  B  ;  \sigma_{{\mathrm{2}}}  \GrTTsym{+}   s  \ast  \sigma_{{\mathrm{4}}}   ;  \sigma_{{\mathrm{3}}}  \GrTTsym{+}   r  \ast  \sigma_{{\mathrm{4}}}  }{
{\GrTTdruleInfAlgXXAppName{}}{}%
}}

\renewcommand{\GrTTdruleRhoXXEmpty}[1]{\GrTTdrule[#1]{%
}{  \emptyset   \odot   \emptyset   \models_{  \emptyset  }   \emptyset  }{
  {\!\GrTTdruleRhoXXEmptyName{}}{}%
}}

\renewcommand{\GrTTdruleRhoXXExtTm}[1]{\GrTTdrule[#1]{%
    \GrTTpremise{
      t \, \in \,  \interp{ A }_{ \varepsilon } %
      \\
       \Delta  \odot  \Gamma  \models_{ \varepsilon }  \rho %
      \\
       ( \Delta  \mid  \sigma  \mid   \textbf{0}  )   \odot  \Gamma  \vdash  A  \GrTTsym{:}   \mathsf{Type}_{ \GrTTsym{0} } }%
}{ \GrTTsym{(}  \Delta  ,  \sigma  \GrTTsym{)}  \odot  \Gamma  \GrTTsym{,}  x  \GrTTsym{:}  A  \models_{ \varepsilon }   \rho [  x  \mapsto  t  ]  }{
{\!\GrTTdruleRhoXXExtTmName{}}{}%
}}

\renewcommand{\GrTTdruleRhoXXExtTy}[1]{\GrTTdrule[#1]{%
    \GrTTpremise{
      t \, \in \, \GrTTsym{(}   \interp{ A }_{ \varepsilon }   \GrTTsym{)} \, \GrTTsym{(}   \varepsilon \, x   \GrTTsym{)}%
      \\
       \Delta  \odot  \Gamma  \models_{ \varepsilon }  \rho %
      \\
       ( \Delta  \mid  \sigma  \mid   \textbf{0}  )   \odot  \Gamma  \vdash  A  \GrTTsym{:}   \mathsf{Type}_{ \GrTTsym{1} } }%
}{ \GrTTsym{(}  \Delta  ,  \sigma  \GrTTsym{)}  \odot  \Gamma  \GrTTsym{,}  x  \GrTTsym{:}  A  \models_{ \varepsilon }   \rho [  x  \mapsto  t  ]  }{
{\GrTTdruleRhoXXExtTyName{}}{}%
}}

\renewcommand{\GrTTdruleEpXXEmpty}[1]{\GrTTdrule[#1]{%
}{  \emptyset   \odot   \emptyset   \models   \emptyset  }{
  {\GrTTdruleEpXXEmptyName{}}{}%
}}

\renewcommand{\GrTTdruleEpXXExtTm}[1]{\GrTTdrule[#1]{%
    \GrTTpremise{
       \Delta  \odot  \Gamma  \models  \varepsilon %
      \\
       ( \Delta  \mid  \sigma  \mid   \textbf{0}  )   \odot  \Gamma  \vdash  A  \GrTTsym{:}   \mathsf{Type}_{ \GrTTsym{0} } }%
}{ \GrTTsym{(}  \Delta  ,  \sigma  \GrTTsym{)}  \odot  \GrTTsym{(}  \Gamma  \GrTTsym{,}  x  \GrTTsym{:}  A  \GrTTsym{)}  \models  \varepsilon }{
{\GrTTdruleEpXXExtTmName{}}{}%
}}

\renewcommand{\GrTTdruleEpXXExtTy}[1]{\GrTTdrule[#1]{%
    \GrTTpremise{
      X \, \in \,  \mathcal{K}\interp{ A }  \qquad
       \Delta  \odot  \Gamma  \models  \varepsilon %
      \\
       ( \Delta  \mid  \sigma  \mid   \textbf{0}  )   \odot  \Gamma  \vdash  A  \GrTTsym{:}   \mathsf{Type}_{ \GrTTsym{1} } }%
}{ \GrTTsym{(}  \Delta  ,  \sigma  \GrTTsym{)}  \odot  \GrTTsym{(}  \Gamma  \GrTTsym{,}  x  \GrTTsym{:}  A  \GrTTsym{)}  \models   \varepsilon [  x  \mapsto  X  ]  }{
{\GrTTdruleEpXXExtTyName{}}{}%
}}

\renewcommand{\GrTTdruleTEQXXRefl}[1]{\GrTTdrule[#1]{%
    \GrTTpremise{
       ( \Delta  \mid  \sigma_{{\mathrm{1}}}  \mid  \sigma_{{\mathrm{2}}} )   \odot  \Gamma  \vdash  t  \GrTTsym{:}  A}
}{ ( \Delta  \mid  \sigma_{{\mathrm{1}}}  \mid  \sigma_{{\mathrm{2}}} )   \odot  \Gamma  \vdash  t  \GrTTsym{=}  t  \GrTTsym{:}  A}{%
{\GrTTdruleTEQXXReflName{}}{}%
}}

\renewcommand{\GrTTdruleTEQXXTrans}[1]{\GrTTdrule[#1]{%
    \GrTTpremise{
       ( \Delta  \mid  \sigma_{{\mathrm{1}}}  \mid  \sigma_{{\mathrm{2}}} )   \odot  \Gamma  \vdash  t_{{\mathrm{1}}}  \GrTTsym{=}  t_{{\mathrm{2}}}  \GrTTsym{:}  A
      \quad
       ( \Delta  \mid  \sigma_{{\mathrm{1}}}  \mid  \sigma_{{\mathrm{2}}} )   \odot  \Gamma  \vdash  t_{{\mathrm{2}}}  \GrTTsym{=}  t_{{\mathrm{3}}}  \GrTTsym{:}  A}
}{ ( \Delta  \mid  \sigma_{{\mathrm{1}}}  \mid  \sigma_{{\mathrm{2}}} )   \odot  \Gamma  \vdash  t_{{\mathrm{1}}}  \GrTTsym{=}  t_{{\mathrm{3}}}  \GrTTsym{:}  A}{%
{\GrTTdruleTEQXXTransName{}}{}%
}}

\renewcommand{\GrTTdruleTEQXXSym}[1]{\GrTTdrule[#1]{%
    \GrTTpremise{
       ( \Delta  \mid  \sigma_{{\mathrm{1}}}  \mid  \sigma_{{\mathrm{2}}} )   \odot  \Gamma  \vdash  t_{{\mathrm{1}}}  \GrTTsym{=}  t_{{\mathrm{2}}}  \GrTTsym{:}  A}
}{ ( \Delta  \mid  \sigma_{{\mathrm{1}}}  \mid  \sigma_{{\mathrm{2}}} )   \odot  \Gamma  \vdash  t_{{\mathrm{2}}}  \GrTTsym{=}  t_{{\mathrm{1}}}  \GrTTsym{:}  A}{%
{\GrTTdruleTEQXXSymName{}}{}%
}}

\renewcommand{\GrTTdruleTEQXXConvTy}[1]{\GrTTdrule[#1]{%
    \GrTTpremise{
       ( \Delta  \mid  \sigma_{{\mathrm{1}}}  \mid  \sigma_{{\mathrm{2}}} )   \odot  \Gamma  \vdash  t_{{\mathrm{1}}}  \GrTTsym{=}  t_{{\mathrm{2}}}  \GrTTsym{:}  A
      \quad
       ( \Delta  \mid  \sigma_{{\mathrm{2}}} )   \odot  \Gamma  \vdash  A  \leq  B}
}{ ( \Delta  \mid  \sigma_{{\mathrm{1}}}  \mid  \sigma_{{\mathrm{2}}} )   \odot  \Gamma  \vdash  t_{{\mathrm{1}}}  \GrTTsym{=}  t_{{\mathrm{2}}}  \GrTTsym{:}  B}{%
{\GrTTdruleTEQXXConvTyName{}}{}%
}}

\renewcommand{\GrTTdruleTEQXXArrow}[1]{\GrTTdrule[#1]{%
    \GrTTpremise{
       ( \Delta  \mid  \sigma_{{\mathrm{1}}}  \mid   \textbf{0}  )   \odot  \Gamma  \vdash  A  \GrTTsym{=}  C  \GrTTsym{:}   \mathsf{Type}_{ l_{{\mathrm{1}}} } 
      \quad
       ( \Delta  ,  \sigma_{{\mathrm{1}}}  \mid  \sigma_{{\mathrm{2}}}  ,  r  \mid   \textbf{0}  )   \odot  \Gamma  \GrTTsym{,}  x  \GrTTsym{:}  A  \vdash  B  \GrTTsym{=}  D  \GrTTsym{:}   \mathsf{Type}_{ l_{{\mathrm{2}}} } }
}{ ( \Delta  \mid  \sigma_{{\mathrm{1}}}  \GrTTsym{+}  \sigma_{{\mathrm{2}}}  \mid   \textbf{0}  )   \odot  \Gamma  \vdash   \textstyle (  x  :_{  \textcolor{darkblue}{( s ,  r )}  }  A  )  \to   B   \GrTTsym{=}   \textstyle (  x  :_{  \textcolor{darkblue}{( s ,  r )}  }  C  )  \to   D   \GrTTsym{:}   \mathsf{Type}_{  l_{{\mathrm{1}}}  \mathop{\sqcup}  l_{{\mathrm{2}}}  } }{%
{\GrTTdruleTEQXXArrowName{}}{}%
}}

\renewcommand{\GrTTdruleTEQXXArrowComp}[1]{\GrTTdrule[#1]{%
    \GrTTpremise{
       ( \Delta  ,  \sigma_{{\mathrm{1}}}  \mid  \sigma_{{\mathrm{2}}}  ,  s  \mid  \sigma_{{\mathrm{3}}}  ,  r )   \odot  \Gamma  \GrTTsym{,}  x  \GrTTsym{:}  A  \vdash  t_{{\mathrm{1}}}  \GrTTsym{:}  B%
      \quad
       ( \Delta  \mid  \sigma_{{\mathrm{4}}}  \mid  \sigma_{{\mathrm{1}}} )   \odot  \Gamma  \vdash  t_{{\mathrm{2}}}  \GrTTsym{:}  A}
}{ ( \Delta  \mid  \sigma_{{\mathrm{2}}}  \GrTTsym{+}   s  \ast  \sigma_{{\mathrm{4}}}   \mid  \sigma_{{\mathrm{3}}}  \GrTTsym{+}   r  \ast  \sigma_{{\mathrm{4}}}  )   \odot  \Gamma  \vdash   \GrTTsym{(}   \lambda  x . t_{{\mathrm{1}}}   \GrTTsym{)} \,{ t_{{\mathrm{2}}} }   \GrTTsym{=}  \GrTTsym{[}  t_{{\mathrm{2}}}  \GrTTsym{/}  x  \GrTTsym{]}  t_{{\mathrm{1}}}  \GrTTsym{:}  \GrTTsym{[}  t_{{\mathrm{2}}}  \GrTTsym{/}  x  \GrTTsym{]}  B}{%
{\GrTTdruleTEQXXArrowCompName{}}{}%
}}

\renewcommand{\GrTTdruleTEQXXArrowUniq}[1]{\GrTTdrule[#1]{%
    \GrTTpremise{
       ( \Delta  \mid  \sigma_{{\mathrm{1}}}  \mid  \sigma_{{\mathrm{2}}} )   \odot  \Gamma  \vdash  t  \GrTTsym{:}   \textstyle (  x  :_{  \textcolor{darkblue}{( s ,  r )}  }  A  )  \to   B }
}{ ( \Delta  \mid  \sigma_{{\mathrm{1}}}  \mid  \sigma_{{\mathrm{2}}} )   \odot  \Gamma  \vdash  t  \GrTTsym{=}   \lambda  x . \GrTTsym{(}   t \,{ x }   \GrTTsym{)}   \GrTTsym{:}   \textstyle (  x  :_{  \textcolor{darkblue}{( s ,  r )}  }  A  )  \to   B }{%
{\GrTTdruleTEQXXArrowUniqName{}}{}%
}}

\renewcommand{\GrTTdruleTEQXXFun}[1]{\GrTTdrule[#1]{%
    \GrTTpremise{
       ( \Delta  ,  \sigma_{{\mathrm{1}}}  \mid  \sigma_{{\mathrm{2}}}  ,  s  \mid  \sigma_{{\mathrm{3}}}  ,  r )   \odot  \Gamma  \GrTTsym{,}  x  \GrTTsym{:}  A  \vdash  t_{{\mathrm{1}}}  \GrTTsym{=}  t_{{\mathrm{2}}}  \GrTTsym{:}  B}%
}{ ( \Delta  \mid  \sigma_{{\mathrm{2}}}  \mid  \sigma_{{\mathrm{1}}}  \GrTTsym{+}  \sigma_{{\mathrm{3}}} )   \odot  \Gamma  \vdash   \lambda  x . t_{{\mathrm{1}}}   \GrTTsym{=}   \lambda  x . t_{{\mathrm{2}}}   \GrTTsym{:}   \textstyle (  x  :_{  \textcolor{darkblue}{( s ,  r )}  }  A  )  \to   B }{%
{\GrTTdruleTEQXXFunName{}}{}%
}}

\renewcommand{\GrTTdruleTEQXXApp}[1]{\GrTTdrule[#1]{%
    \GrTTpremise{
       ( \Delta  ,  \sigma_{{\mathrm{1}}}  \mid  \sigma_{{\mathrm{3}}}  ,  r  \mid   \textbf{0}  )   \odot  \Gamma  \GrTTsym{,}  x  \GrTTsym{:}  A  \vdash  B  \GrTTsym{:}   \mathsf{Type}_{ l } 
      \\
       ( \Delta  \mid  \sigma_{{\mathrm{2}}}  \mid  \sigma_{{\mathrm{1}}}  \GrTTsym{+}  \sigma_{{\mathrm{3}}} )   \odot  \Gamma  \vdash  t_{{\mathrm{1}}}  \GrTTsym{=}  t_{{\mathrm{2}}}  \GrTTsym{:}   \textstyle (  x  :_{  \textcolor{darkblue}{( s ,  r )}  }  A  )  \to   B 
      \quad
       ( \Delta  \mid  \sigma_{{\mathrm{4}}}  \mid  \sigma_{{\mathrm{1}}} )   \odot  \Gamma  \vdash  t_{{\mathrm{3}}}  \GrTTsym{=}  t_{{\mathrm{4}}}  \GrTTsym{:}  A}
}{ ( \Delta  \mid  \sigma_{{\mathrm{2}}}  \GrTTsym{+}   s  \ast  \sigma_{{\mathrm{4}}}   \mid  \sigma_{{\mathrm{3}}}  \GrTTsym{+}   r  \ast  \sigma_{{\mathrm{4}}}  )   \odot  \Gamma  \vdash   t_{{\mathrm{1}}} \,{ t_{{\mathrm{3}}} }   \GrTTsym{=}   t_{{\mathrm{2}}} \,{ t_{{\mathrm{4}}} }   \GrTTsym{:}  \GrTTsym{[}  t_{{\mathrm{3}}}  \GrTTsym{/}  x  \GrTTsym{]}  B}{%
{\GrTTdruleTEQXXAppName{}}{}%
}}

\renewcommand{\GrTTdruleTEQXXTen}[1]{\GrTTdrule[#1]{%
    \GrTTpremise{
       ( \Delta  \mid  \sigma_{{\mathrm{1}}}  \mid   \textbf{0}  )   \odot  \Gamma  \vdash  A  \GrTTsym{=}  C  \GrTTsym{:}   \mathsf{Type}_{ l } 
      \quad
       ( \Delta  ,  \sigma_{{\mathrm{1}}}  \mid  \sigma_{{\mathrm{2}}}  ,  r  \mid   \textbf{0}  )   \odot  \Gamma  \GrTTsym{,}  x  \GrTTsym{:}  A  \vdash  B  \GrTTsym{=}  D  \GrTTsym{:}   \mathsf{Type}_{ l } }
}{ ( \Delta  \mid  \sigma_{{\mathrm{1}}}  \GrTTsym{+}  \sigma_{{\mathrm{2}}}  \mid   \textbf{0}  )   \odot  \Gamma  \vdash   \textstyle (  x  :_{  \textcolor{darkblue}{ r }  }  A  )  \otimes   B   \GrTTsym{=}   \textstyle (  x  :_{  \textcolor{darkblue}{ s }  }  C  )  \otimes   D   \GrTTsym{:}   \mathsf{Type}_{ l } }{%
{\GrTTdruleTEQXXTenName{}}{}%
}}

\renewcommand{\GrTTdruleTEQXXTenComp}[1]{\GrTTdrule[#1]{%
    \GrTTpremise{
       ( \Delta  ,  \GrTTsym{(}  \sigma_{{\mathrm{1}}}  \GrTTsym{+}  \sigma_{{\mathrm{2}}}  \GrTTsym{)}  \mid  \sigma_{{\mathrm{5}}}  ,  r'  \mid   \textbf{0}  )   \odot  \Gamma  \GrTTsym{,}  z  \GrTTsym{:}   \textstyle (  x  :_{  \textcolor{darkblue}{ r }  }  A  )  \otimes   B   \vdash  C  \GrTTsym{:}   \mathsf{Type}_{ l } %
      \\
       ( \Delta  \mid  \sigma_{{\mathrm{3}}}  \mid  \sigma_{{\mathrm{1}}} )   \odot  \Gamma  \vdash  t_{{\mathrm{1}}}  \GrTTsym{:}  A%
      \quad
       ( \Delta  \mid  \sigma_{{\mathrm{6}}}  \mid  \sigma_{{\mathrm{2}}}  \GrTTsym{+}   r  \ast  \sigma_{{\mathrm{3}}}  )   \odot  \Gamma  \vdash  t_{{\mathrm{2}}}  \GrTTsym{:}  \GrTTsym{[}  t_{{\mathrm{1}}}  \GrTTsym{/}  x  \GrTTsym{]}  B%
      \\
       ( \Delta  ,  \sigma_{{\mathrm{1}}}  ,  \GrTTsym{(}  \sigma_{{\mathrm{2}}}  ,  r  \GrTTsym{)}  \mid  \sigma_{{\mathrm{4}}}  ,  s  ,  s  \mid  \sigma_{{\mathrm{5}}}  ,  r'  ,  r' )   \odot  \Gamma  \GrTTsym{,}  x  \GrTTsym{:}  A  \GrTTsym{,}  y  \GrTTsym{:}  B  \vdash  t_{{\mathrm{3}}}  \GrTTsym{:}  \GrTTsym{[}  \GrTTsym{(}  x  \GrTTsym{,}  y  \GrTTsym{)}  \GrTTsym{/}  z  \GrTTsym{]}  C}
}{ ( \Delta  \mid  \sigma_{{\mathrm{4}}}  \GrTTsym{+}   s  \ast  \GrTTsym{(}  \sigma_{{\mathrm{3}}}  \GrTTsym{+}  \sigma_{{\mathrm{6}}}  \GrTTsym{)}   \mid  \sigma_{{\mathrm{5}}}  \GrTTsym{+}   r'  \ast  \GrTTsym{(}  \sigma_{{\mathrm{3}}}  \GrTTsym{+}  \sigma_{{\mathrm{6}}}  \GrTTsym{)}  )   \odot  \Gamma  \vdash   \mathsf{let}\,  \textstyle (  x  ,  y  )   =  \GrTTsym{(}  t_{{\mathrm{1}}}  \GrTTsym{,}  t_{{\mathrm{2}}}  \GrTTsym{)}  \,\mathsf{in}\, t_{{\mathrm{3}}}   \GrTTsym{=}  \GrTTsym{[}  t_{{\mathrm{1}}}  \GrTTsym{,}  t_{{\mathrm{2}}}  \GrTTsym{/}  x  \GrTTsym{,}  y  \GrTTsym{]}  t_{{\mathrm{3}}}  \GrTTsym{:}  \GrTTsym{[}  \GrTTsym{(}  t_{{\mathrm{1}}}  \GrTTsym{,}  t_{{\mathrm{2}}}  \GrTTsym{)}  \GrTTsym{/}  z  \GrTTsym{]}  C}{%
{\GrTTdruleTEQXXTenCompName{}}{}%
}}

\renewcommand{\GrTTdruleTEQXXPair}[1]{\GrTTdrule[#1]{%
    \GrTTpremise{
       ( \Delta  ,  \sigma_{{\mathrm{1}}}  \mid  \sigma_{{\mathrm{3}}}  ,  r  \mid   \textbf{0}  )   \odot  \Gamma  \GrTTsym{,}  x  \GrTTsym{:}  A  \vdash  B  \GrTTsym{:}   \mathsf{Type}_{ l } 
      \\
       ( \Delta  \mid  \sigma_{{\mathrm{2}}}  \mid  \sigma_{{\mathrm{1}}} )   \odot  \Gamma  \vdash  t_{{\mathrm{1}}}  \GrTTsym{=}  t_{{\mathrm{2}}}  \GrTTsym{:}  A
      \quad
       ( \Delta  \mid  \sigma_{{\mathrm{4}}}  \mid  \sigma_{{\mathrm{3}}}  \GrTTsym{+}   r  \ast  \sigma_{{\mathrm{2}}}  )   \odot  \Gamma  \vdash  t_{{\mathrm{3}}}  \GrTTsym{=}  t_{{\mathrm{4}}}  \GrTTsym{:}  \GrTTsym{[}  t_{{\mathrm{1}}}  \GrTTsym{/}  x  \GrTTsym{]}  B}
}{ ( \Delta  \mid  \sigma_{{\mathrm{2}}}  \GrTTsym{+}  \sigma_{{\mathrm{4}}}  \mid  \sigma_{{\mathrm{1}}}  \GrTTsym{+}  \sigma_{{\mathrm{3}}} )   \odot  \Gamma  \vdash  \GrTTsym{(}  t_{{\mathrm{1}}}  \GrTTsym{,}  t_{{\mathrm{2}}}  \GrTTsym{)}  \GrTTsym{=}  \GrTTsym{(}  t_{{\mathrm{3}}}  \GrTTsym{,}  t_{{\mathrm{4}}}  \GrTTsym{)}  \GrTTsym{:}   \textstyle (  x  :_{  \textcolor{darkblue}{ r }  }  A  )  \otimes   B }{%
{\GrTTdruleTEQXXPairName{}}{}%
}}

\renewcommand{\GrTTdruleTEQXXTenCut}[1]{\GrTTdrule[#1]{%
\GrTTpremise{ ( \Delta  \mid  \sigma_{{\mathrm{3}}}  \mid  \sigma_{{\mathrm{1}}}  \GrTTsym{+}  \sigma_{{\mathrm{2}}} )   \odot  \Gamma  \vdash  t_{{\mathrm{1}}}  \GrTTsym{=}  t_{{\mathrm{2}}}  \GrTTsym{:}   \textstyle (  x  :_{  \textcolor{darkblue}{ r }  }  A  )  \otimes   B }
\GrTTpremise{ ( \Delta  ,  \GrTTsym{(}  \sigma_{{\mathrm{1}}}  \GrTTsym{+}  \sigma_{{\mathrm{2}}}  \GrTTsym{)}  \mid  \sigma_{{\mathrm{5}}}  ,  r'  \mid   \textbf{0}  )   \odot  \Gamma  \GrTTsym{,}  z  \GrTTsym{:}   \textstyle (  x  :_{  \textcolor{darkblue}{ r }  }  A  )  \otimes   B   \vdash  C  \GrTTsym{:}   \mathsf{Type}_{ l } }
\GrTTpremise{ ( \Delta  ,  \sigma_{{\mathrm{1}}}  ,  \GrTTsym{(}  \sigma_{{\mathrm{2}}}  ,  r  \GrTTsym{)}  \mid  \sigma_{{\mathrm{4}}}  ,  s  ,  s  \mid  \sigma_{{\mathrm{5}}}  ,  r'  ,  r' )   \odot  \Gamma  \GrTTsym{,}  x  \GrTTsym{:}  A  \GrTTsym{,}  y  \GrTTsym{:}  B  \vdash  t_{{\mathrm{3}}}  \GrTTsym{=}  t_{{\mathrm{4}}}  \GrTTsym{:}  \GrTTsym{[}  \GrTTsym{(}  x  \GrTTsym{,}  y  \GrTTsym{)}  \GrTTsym{/}  z  \GrTTsym{]}  C}
}{ ( \Delta  \mid  \sigma_{{\mathrm{4}}}  \GrTTsym{+}   s  \ast  \sigma_{{\mathrm{3}}}   \mid  \sigma_{{\mathrm{5}}}  \GrTTsym{+}   r'  \ast  \sigma_{{\mathrm{3}}}  )   \odot  \Gamma  \vdash  \GrTTsym{(}   \mathsf{let}\,  \textstyle (  x  ,  y  )   =  t_{{\mathrm{1}}}  \,\mathsf{in}\, t_{{\mathrm{3}}}   \GrTTsym{)}  \GrTTsym{=}  \GrTTsym{(}   \mathsf{let}\,  \textstyle (  x  ,  y  )   =  t_{{\mathrm{2}}}  \,\mathsf{in}\, t_{{\mathrm{4}}}   \GrTTsym{)}  \GrTTsym{:}  \GrTTsym{[}  t_{{\mathrm{1}}}  \GrTTsym{/}  z  \GrTTsym{]}  C}{%
{\GrTTdruleTEQXXTenCutName{}}{}%
}}

\renewcommand{\GrTTdruleTEQXXTenU}[1]{\GrTTdrule[#1]{%
\GrTTpremise{ ( \Delta  \mid  \sigma_{{\mathrm{1}}}  \mid  \sigma_{{\mathrm{2}}} )   \odot  \Gamma  \vdash  t  \GrTTsym{:}   \textstyle (  x  :_{  \textcolor{darkblue}{ r }  }  A  )  \otimes   B }%
}{ ( \Delta  \mid  \sigma_{{\mathrm{1}}}  \mid  \sigma_{{\mathrm{2}}} )   \odot  \Gamma  \vdash  t  \GrTTsym{=}  \GrTTsym{(}   \mathsf{let}\,  \textstyle (  x  ,  y  )   =  t  \,\mathsf{in}\, \GrTTsym{(}  x  \GrTTsym{,}  y  \GrTTsym{)}   \GrTTsym{)}  \GrTTsym{:}   \textstyle (  x  :_{  \textcolor{darkblue}{ r }  }  A  )  \otimes   B }{
{\GrTTdruleTEQXXTenUName{}}{}%
}}

\renewcommand{\GrTTdruleTEQXXBox}[1]{\GrTTdrule[#1]{%
    \GrTTpremise{ ( \Delta  \mid  \sigma  \mid   \textbf{0}  )   \odot  \Gamma  \vdash  A  \GrTTsym{=}  B  \GrTTsym{:}   \mathsf{Type}_{ l } }%
  }{ ( \Delta  \mid  \sigma  \mid   \textbf{0}  )   \odot  \Gamma  \vdash   \square_{  \textcolor{darkblue}{ s }  }  A   \GrTTsym{=}   \square_{  \textcolor{darkblue}{ s }  }  B   \GrTTsym{:}   \mathsf{Type}_{ l } }{%
    {\GrTTdruleTEQXXBoxName{}}{}%
}}

\renewcommand{\GrTTdruleTEQXXBoxI}[1]{\GrTTdrule[#1]{%
\GrTTpremise{ ( \Delta  \mid  \sigma_{{\mathrm{1}}}  \mid  \sigma_{{\mathrm{2}}} )   \odot  \Gamma  \vdash  t_{{\mathrm{1}}}  \GrTTsym{=}  t_{{\mathrm{2}}}  \GrTTsym{:}  A}%
}{ ( \Delta  \mid   s  \ast  \sigma_{{\mathrm{1}}}   \mid  \sigma_{{\mathrm{2}}} )   \odot  \Gamma  \vdash   \square  t_{{\mathrm{1}}}   \GrTTsym{=}   \square  t_{{\mathrm{2}}}   \GrTTsym{:}   \square_{  \textcolor{darkblue}{ s }  }  A }{
{\GrTTdruleTEQXXBoxIName{}}{}%
}}

\renewcommand{\GrTTdruleTEQXXBoxB}[1]{\GrTTdrule[#1]{%
    \GrTTpremise{
       ( \Delta  ,  \sigma_{{\mathrm{2}}}  \mid  \sigma_{{\mathrm{4}}}  ,  r  \mid   \textbf{0}  )   \odot  \Gamma  \GrTTsym{,}  z  \GrTTsym{:}   \square_{  \textcolor{darkblue}{ s }  }  A   \vdash  B  \GrTTsym{:}   \mathsf{Type}_{ l } %
      \\
       ( \Delta  \mid  \sigma_{{\mathrm{1}}}  \mid  \sigma_{{\mathrm{2}}} )   \odot  \Gamma  \vdash  t_{{\mathrm{1}}}  \GrTTsym{:}  A \quad\;\;
       ( \Delta  ,  \sigma_{{\mathrm{2}}}  \mid  \sigma_{{\mathrm{3}}}  ,  s  \mid  \sigma_{{\mathrm{4}}}  ,  \GrTTsym{(}   s  \ast  r   \GrTTsym{)} )   \odot  \Gamma  \GrTTsym{,}  x  \GrTTsym{:}  A  \vdash  t_{{\mathrm{2}}}  \GrTTsym{:}  \GrTTsym{[}   \square  x   \GrTTsym{/}  z  \GrTTsym{]}  B}
}{ ( \Delta  \mid  \sigma_{{\mathrm{3}}}  \GrTTsym{+}   s  \ast  \sigma_{{\mathrm{1}}}   \mid  \sigma_{{\mathrm{4}}}  \GrTTsym{+}     s  \ast  r    \ast  \sigma_{{\mathrm{1}}}  )   \odot  \Gamma  \vdash  \GrTTsym{(}   \mathsf{let}\,  \square  x   =   \square  t_{{\mathrm{1}}}   \,\mathsf{in}\, t_{{\mathrm{2}}}   \GrTTsym{)}  \GrTTsym{=}  \GrTTsym{[}  t_{{\mathrm{1}}}  \GrTTsym{/}  x  \GrTTsym{]}  t_{{\mathrm{2}}}  \GrTTsym{:}  \GrTTsym{[}   \square  t_{{\mathrm{1}}}   \GrTTsym{/}  z  \GrTTsym{]}  B}{%
{\GrTTdruleTEQXXBoxBName{}}{}%
}}

\renewcommand{\GrTTdruleTEQXXBoxE}[1]{\GrTTdrule[#1]{%
    \GrTTpremise{
       ( \Delta  ,  \sigma_{{\mathrm{2}}}  \mid  \sigma_{{\mathrm{4}}}  ,  r  \mid   \textbf{0}  )   \odot  \Gamma  \GrTTsym{,}  z  \GrTTsym{:}   \square_{  \textcolor{darkblue}{ s }  }  A   \vdash  B  \GrTTsym{:}   \mathsf{Type}_{ l } %
      \\
       ( \Delta  \mid  \sigma_{{\mathrm{1}}}  \mid  \sigma_{{\mathrm{2}}} )   \odot  \Gamma  \vdash  t_{{\mathrm{1}}}  \GrTTsym{=}  t_{{\mathrm{2}}}  \GrTTsym{:}   \square_{  \textcolor{darkblue}{ s }  }  A %
      \\
       ( \Delta  ,  \sigma_{{\mathrm{2}}}  \mid  \sigma_{{\mathrm{3}}}  ,  s  \mid  \sigma_{{\mathrm{4}}}  ,  \GrTTsym{(}   s  \ast  r   \GrTTsym{)} )   \odot  \Gamma  \GrTTsym{,}  x  \GrTTsym{:}  A  \vdash  t_{{\mathrm{3}}}  \GrTTsym{=}  t_{{\mathrm{4}}}  \GrTTsym{:}  \GrTTsym{[}   \square  x   \GrTTsym{/}  z  \GrTTsym{]}  B}%
}{ ( \Delta  \mid  \sigma_{{\mathrm{1}}}  \GrTTsym{+}  \sigma_{{\mathrm{3}}}  \mid  \sigma_{{\mathrm{4}}}  \GrTTsym{+}   r  \ast  \sigma_{{\mathrm{1}}}  )   \odot  \Gamma  \vdash  \GrTTsym{(}   \mathsf{let}\,  \square  x   =  t_{{\mathrm{1}}}  \,\mathsf{in}\, t_{{\mathrm{3}}}   \GrTTsym{)}  \GrTTsym{=}  \GrTTsym{(}   \mathsf{let}\,  \square  x   =  t_{{\mathrm{2}}}  \,\mathsf{in}\, t_{{\mathrm{4}}}   \GrTTsym{)}  \GrTTsym{:}  \GrTTsym{[}  t_{{\mathrm{1}}}  \GrTTsym{/}  z  \GrTTsym{]}  B}{
{\GrTTdruleTEQXXBoxEName{}}{}%
}}

\renewcommand{\GrTTdruleTEQXXBoxU}[1]{\GrTTdrule[#1]{%
    \GrTTpremise{
       ( \Delta  \mid  \sigma_{{\mathrm{1}}}  \mid  \sigma_{{\mathrm{2}}} )   \odot  \Gamma  \vdash  t  \GrTTsym{:}   \square_{  \textcolor{darkblue}{ s }  }  A }%
}{ ( \Delta  \mid  \sigma_{{\mathrm{1}}}  \mid  \sigma_{{\mathrm{2}}} )   \odot  \Gamma  \vdash  t  \GrTTsym{=}  \GrTTsym{(}   \mathsf{let}\,  \square  x   =  t  \,\mathsf{in}\,  \square  x    \GrTTsym{)}  \GrTTsym{:}   \square_{  \textcolor{darkblue}{ s }  }  A }{
{\GrTTdruleTEQXXBoxUName{}}{}%
}}

\renewcommand{\GrTTdruleSTXXEq}[1]{\GrTTdrule[#1]{%
    \GrTTpremise{
       ( \Delta  \mid  \sigma  \mid   \textbf{0}  )   \odot  \Gamma  \vdash  A  \GrTTsym{=}  B  \GrTTsym{:}   \mathsf{Type}_{ l } }%
}{ ( \Delta  \mid  \sigma )   \odot  \Gamma  \vdash  A  \leq  B}{
{\GrTTdruleSTXXEqName{}}{}%
}}

\renewcommand{\GrTTdruleSTXXTrans}[1]{\GrTTdrule[#1]{%
    \GrTTpremise{
       ( \Delta  \mid  \sigma )   \odot  \Gamma  \vdash  A  \leq  B%
      \quad
       ( \Delta  \mid  \sigma )   \odot  \Gamma  \vdash  B  \leq  C}
}{ ( \Delta  \mid  \sigma )   \odot  \Gamma  \vdash  A  \leq  C}{
{\GrTTdruleSTXXTransName{}}{}%
}}

\renewcommand{\GrTTdruleSTXXTy}[1]{\GrTTdrule[#1]{%
    \GrTTpremise{
      \Delta  \odot  \Gamma  \vdash%
      \quad
      l  \leq  l'}
}{ ( \Delta  \mid   \textbf{0}  )   \odot  \Gamma  \vdash   \mathsf{Type}_{ l }   \leq   \mathsf{Type}_{ l' } }{
{\GrTTdruleSTXXTyName{}}{}%
}}

\renewcommand{\GrTTdruleSTXXArrow}[1]{\GrTTdrule[#1]{%
    \GrTTpremise{
       ( \Delta  ,  \sigma_{{\mathrm{1}}}  \mid  \sigma_{{\mathrm{2}}}  ,  r  \mid   \textbf{0}  )   \odot  \Gamma  \GrTTsym{,}  x  \GrTTsym{:}  A  \vdash  B  \GrTTsym{:}   \mathsf{Type}_{ l } %
      \\
       ( \Delta  \mid  \sigma_{{\mathrm{1}}} )   \odot  \Gamma  \vdash  A'  \leq  A%
      \quad
       ( \Delta  ,  \sigma_{{\mathrm{1}}}  \mid  \sigma_{{\mathrm{2}}}  ,  r )   \odot  \Gamma  \GrTTsym{,}  x  \GrTTsym{:}  A'  \vdash  B  \leq  B'}
}{ ( \Delta  \mid  \sigma_{{\mathrm{1}}}  \GrTTsym{+}  \sigma_{{\mathrm{2}}} )   \odot  \Gamma  \vdash   \textstyle (  x  :_{  \textcolor{darkblue}{( s ,  r )}  }  A  )  \to   B   \leq   \textstyle (  x  :_{  \textcolor{darkblue}{( s ,  r )}  }  A'  )  \to   B' }{
{\GrTTdruleSTXXArrowName{}}{}%
}}

\renewcommand{\GrTTdruleSTXXTen}[1]{\GrTTdrule[#1]{%
    \GrTTpremise{
       ( \Delta  ,  \sigma_{{\mathrm{1}}}  \mid  \sigma_{{\mathrm{2}}}  ,  r )   \odot  \Gamma  \GrTTsym{,}  x  \GrTTsym{:}  A  \vdash  B  \leq  B'}
}{ ( \Delta  \mid  \sigma_{{\mathrm{1}}}  \GrTTsym{+}  \sigma_{{\mathrm{2}}} )   \odot  \Gamma  \vdash   \textstyle (  x  :_{  \textcolor{darkblue}{ r }  }  A  )  \otimes   B   \leq   \textstyle (  x  :_{  \textcolor{darkblue}{ r }  }  A  )  \otimes   B' }{
{\GrTTdruleSTXXTenName{}}{}%
}}

\renewcommand{\GrTTdruleSTXXBox}[1]{\GrTTdrule[#1]{%
    \GrTTpremise{
       ( \Delta  \mid  \sigma )   \odot  \Gamma  \vdash  A  \leq  A'}
}{ ( \Delta  \mid  \sigma )   \odot  \Gamma  \vdash   \square_{  \textcolor{darkblue}{ s }  }  A   \leq   \square_{  \textcolor{darkblue}{ s }  }  A' }{
{\GrTTdruleSTXXBoxName{}}{}%
}}

\usepackage{mathpartir}
\usepackage{thm-restate}
\usepackage{mathpartir}
\usepackage{stmaryrd}
\usepackage{amssymb}
\usepackage{enumitem}
\usepackage{amsmath}

\newcommand{\interp}[1]{\llbracket{#1}\rrbracket}

\newenvironment{glemma}[2]{%
  \restatable[#2]{lemma}{#1}\label{lemma:#1}%
}{\endrestatable}

\definecolor{darkblue}{rgb}{0,0.1,0.65}

\newcommand{\grtt}[0]{\textsc{Grtt}}
\newcommand{\grttfull}[0]{Graded Modal Dependent Type Theory}
\newcommand{\qtt}[0]{\textsc{Qtt}}
\newcommand{\mlof}[0]{Martin-L\"{o}f}

\providecommand{\GrTTdruleTXXTypeName}{}
\providecommand{\GrTTdruleTXXArrowName}{}
\providecommand{\GrTTdruleTXXTenName}{}
\providecommand{\GrTTdruleTXXVarName}{}
\providecommand{\GrTTdruleTXXFunName}{}
\providecommand{\GrTTdruleTXXAppName}{}
\providecommand{\GrTTdruleTXXPairName}{}
\providecommand{\GrTTdruleTXXTenCutName}{}

\providecommand{\GrTTdruleTXXBoxName}{}
\providecommand{\GrTTdruleTXXBoxIName}{}
\providecommand{\GrTTdruleTXXBoxEName}{}

\providecommand{\GrTTdruleWfXXEmptyName}{}
\providecommand{\GrTTdruleWfXXExtName}{}

\providecommand{\GrTTdruleChkAlgXXFunName}{}

\renewcommand{\GrTTdruleTXXTypeName}{\mathsf{Type}}
\renewcommand{\GrTTdruleTXXArrowName}{\to}
\renewcommand{\GrTTdruleTXXTenName}{\otimes}
\renewcommand{\GrTTdruleTXXVarName}{\textsc{Var}}
\renewcommand{\GrTTdruleTXXFunName}{\lambda_i}
\renewcommand{\GrTTdruleTXXAppName}{\lambda_e}
\renewcommand{\GrTTdruleTXXPairName}{\otimes_i}
\renewcommand{\GrTTdruleTXXTenCutName}{\otimes_e}

\renewcommand{\GrTTdruleTXXBoxName}{\square}
\renewcommand{\GrTTdruleTXXBoxIName}{\square_i}
\renewcommand{\GrTTdruleTXXBoxEName}{\square_e}

\renewcommand{\GrTTdruleWfXXEmptyName}{\textsc{wf$\emptyset$}}
\renewcommand{\GrTTdruleWfXXExtName}{\textsc{wfExt}}

\renewcommand{\GrTTdruleChkAlgXXFunName}{\Leftarrow\lambda_i}

\usepackage{listings}
\definecolor{DMidnightBlue}{RGB}{10,83,154}
\definecolor{MidnightBlue}{RGB}{19,113,214}

\lstdefinelanguage{Gerty}{%
  mathescape=true,
  morecomment=[l]{--},
  moredelim=[s][\itshape]{`}{`},
  showspaces=false,
  aboveskip=1em,
  belowskip=1em,
  commentstyle=\itshape\color{black!60},
  basicstyle=\footnotesize\ttfamily,
  flexiblecolumns=true,
  columns=[l]flexible,
  keepspaces=true,
  xleftmargin=.75em,
  literate=%
  {[}{\textcolor{MidnightBlue}{[}}1
  {]}{\textcolor{MidnightBlue}{]}}1
  {[([}{\textcolor{MidnightBlue}{[[}}2
  {])]}{\textcolor{MidnightBlue}{]]}}2
  {Inf}{$\infty$}1
  {.0}{{.\textcolor{DMidnightBlue}{0}}}1
  {.1}{{.\textcolor{DMidnightBlue}{1}}}1
  {.2}{{.\textcolor{DMidnightBlue}{2}}}1
  {.3}{{.\textcolor{DMidnightBlue}{3}}}1
  {.4}{{.\textcolor{DMidnightBlue}{4}}}1
  {.5}{{.\textcolor{DMidnightBlue}{5}}}1
  {.6}{{.\textcolor{DMidnightBlue}{6}}}1
  {.7}{{.\textcolor{DMidnightBlue}{7}}}1
  {.8}{{.\textcolor{DMidnightBlue}{8}}}1,
  keywordstyle = \color{blue!60!black}\bfseries,
  keywords = {level, data, Type, let, in, case, as, of, if, then, else, where}
}

\lstnewenvironment{gerty}
{\lstset{language=Gerty}}{}

\lstnewenvironment{gertyError}
{\lstset{basicstyle=\ttfamily\bfseries\footnotesize, xleftmargin=2em}}{}

\lstnewenvironment{gertySmall}
{\lstset{language=Gerty,basicstyle=\ttfamily\footnotesize}}{}

\lstnewenvironment{gertyReallySmall}
{\lstset{language=Gerty,basicstyle=\ttfamily\footnotesize}}{}

\newcommand{\gertyin}[1]{\text{\lstinline[language=Gerty]{#1}}}

\usepackage{graphicx}
\makeatletter
\RequirePackage[bookmarks,unicode,colorlinks=true]{hyperref}%
   \def\@citecolor{blue}%
   \def\@urlcolor{blue}%
   \def\@linkcolor{blue}%

\def\orcidID#1{\smash{\href{http://orcid.org/#1}{\protect\raisebox{-1.25pt}{\protect\includegraphics{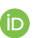}}}}}
\makeatother

\let\olddefinition\definition
\renewcommand{\definition}{\olddefinition\normalfont}

\newcommand{\implName}{\textbf{Gerty}}

\title{\grttfull{}}

\author{Benjamin Moon\inst{1}\textsuperscript{(\Letter)}\orcidID{0000-0003-2367-6321} \and
Harley Eades III\inst{2}\orcidID{0000-0001-8474-5971} \and
Dominic Orchard\inst{1}\orcidID{0000-0002-7058-7842}
}
\institute{University of Kent, Canterbury, UK\\\email{\{bgm4,d.a.orchard\}@kent.ac.uk} \and Augusta University, Augusta, USA\\\email{harley.eades@gmail.com}}

\authorrunning{Benjamin Moon, Harley Eades III, and Dominic Orchard}

\begin{document}

\maketitle

\begin{abstract}
  Graded type theories are an emerging paradigm for augmenting
  the reasoning power of types with parameterizable,
  fine-grained analyses of program properties. There have been
  many such theories in recent years which equip a type theory
  with quantitative dataflow tracking, usually via a
  semiring-like structure which provides analysis on variables
  (often called `quantitative' or `coeffect' theories). We
  present \grttfull{} (\grtt{} for short), which equips a
  dependent type theory with a general, parameterizable
  analysis of the flow of data, both in and between
  computational terms and types. In this theory, it is
  possible to study, restrict, and reason about data use in
  programs and types, enabling, for example, parametric
  quantifiers and linearity to be captured in a dependent
  setting. We propose \grtt{}, study its metatheory, and
  explore various case studies of its use in reasoning about
  programs and studying other type theories. We have
  implemented the theory and highlight the interesting
  details, including showing an application of grading to
  optimising the type checking procedure itself.
\end{abstract}

\section{Introduction}\label{sec:introduction}

The difference between simply-typed, polymorphically-typed, and
dependently-typed languages can be characterised by the
\emph{dataflow} permitted by each type theory. In each, dataflow can be
enacted by \emph{substituting} a term for occurrences of a
variable in another term, the scope of which is delineated by a
binder. In the simply-typed $\lambda$-calculus, data can only flow in
`computational' terms; computations and types are separate syntactic
categories, with variables, bindings ($\lambda$), and
substitution---and thus dataflow---only at the computational level. In
contrast, polymorphic calculi like System
F~\cite{girard1971extension,reynolds1974towards} permit dataflow
within types, via type quantification ($\forall$), and a limited form
of dataflow from computations to types, via type abstraction
($\Lambda$) and type application. Dependently-typed calculi
 (e.g.,~\cite{coquand1986calculus,martin-lofIntuitionisticTheoryTypes1975,martin-lofIntuitionisticTypeTheory1980,martin-lofConstructiveMathematicsComputer1982})
break
down the barrier between computations and types further: variables are
bound simultaneously in types and computations, such that data can
flow both to computations and types via dependent
functions ($\Pi$) and application.
This pervasive dataflow enables the Curry-Howard correspondence
to be leveraged for program reasoning and theorem
proving~\cite{wadler2015propositions}.
However, unrestricted dataflow between computations and types can
impede reasoning and can interact poorly with
other type theoretic ideas.

Firstly, System F allows \emph{parametric reasoning} and
notions of representation
independence~\cite{reynolds1983types,wadler1989theorems}, but this
is lost in general in dependently-typed
languages when quantifying over higher-kinded
types~\cite{nuytsParametricQuantifiersDependent2017} (rather than
just `small' types~\cite{atkey2014relationally,krishnaswami2013internalizing}).
Furthermore, unrestricted dataflow impedes
efficient compilation as compilers do not know, from the types alone,
where a term is actually needed.  Additional static analyses are
needed to recover dataflow information for
optimisation and reasoning. For example, a
term shown to be used only for type checking (not flowing to the computational
`run time' level) can be erased~\cite{brady2003inductive}.
Thus, dependent theories do not expose the
distinction between proof relevant and irrelevant terms,
requiring extensions to capture
irrelevance~\cite{DBLP:journals/corr/abs-1203-4716,pfenning2001intensionality,reed2003extending}.
Whilst unrestricted dataflow between computations and terms has
its benefits, the permissive nature of dependent types can hide useful
information. This permissiveness also interacts poorly with other type
theories which seek to deliberately restrict dataflow, notably
\emph{linear types}.

Linear types allow data to be treated as a `resource' which
must be consumed exactly
once:
linearly-typed values are restricted to linear
dataflow~\cite{Girard87,wadler1990linear,walker2005substructural}.
Reasoning about resourceful data has been
exploited by several languages, e.g.,
ATS~\cite{shi2013linear},
Alms~\cite{DBLP:conf/popl/TovP11}, Clean~\cite{de2007uniqueness},
Granule~\cite{orchard2019quantitative}, and Linear
Haskell~\cite{bernardy2017linear}. However, linear dataflow is
rare in a dependently-typed setting.
Consider typing the body of the
polymorphic identity function in \mlof{} type theory:
$$a : \mathsf{Type}, x : a \vdash x : a$$
This judgment uses $a$ twice (typing $x$ in
the context and the subject of the judgment) and $x$ once
in the term but not at all in the type.
 There have been various attempts to meaningfully reconcile linear and dependent
types~\cite{cervesato2002linear,dal2011linear,krishnaswami2015integrating,luo2016linear}
usually by keeping them separate, allowing types to
depend only on non-linear variables. All such theories cannot
distinguish variables used for computation from those used
purely for type formation, which could be erased at runtime.

Recent work by McBride~\cite{McBride2016}, refined by
Atkey~\cite{quantitative-type-theory}, generalises ideas
from `coeffect analyses' (variable usage analyses, like that of
Petricek et al.~\cite{Petricek:2014})
to a dependently-typed setting to reconcile the ubiquitous flow of data
in dependent types with the restricted dataflow of linearity.
This approach, called
Quantitative Type Theory (\qtt{}), types the above example as:
$$a \stackrel{0}{:} \mathsf{Type}, x \stackrel{1}{:} a \vdash x
\stackrel{1}{:} a$$
The annotation $0$ on $a$ explains that we
can use $a$ to form a type, but we cannot, or do not, use it
at the term level, thus it can be erased at runtime. The
cornerstone of \qtt{}'s approach is that dataflow of a term
to the type level counts as $0$ use, so arbitrary type-level
use is allowed whilst still permitting quantitative analysis of
computation-level dataflow. Whilst this gives a useful way to
relate linear and dependent types, it cannot however
reason about dataflow at the type-level (all type-level
usage counts as $0$). Thus, for example, \qtt{} cannot
express that a variable is used just computationally but not at
all in types.

In an extended abstract, Abel proposes a generalisation
of \qtt{} to track variable use in both types
and computations~\cite{abel2018}, suggesting
that tracking in types enables type checking optimisations
and increased expressivity. We develop a core dependent type theory along
the same lines, using the paradigm of \emph{grading}: graded
systems augment types with additional information, capturing the structure
of programs~\cite{gaboardi2016combining,orchard2019quantitative}.
We therefore name our approach \emph{\grttfull{}}
(\grtt{} for short). Our type theory is parameterised by a semiring which, like
other coeffect and quantitative
approaches~\cite{DBLP:journals/pacmpl/AbelB20,quantitative-type-theory,gaboardi2014,ghica2014,McBride2016,Petricek:2014,DBLP:journals/corr/abs-2005-02247},
describes dataflow through a program, but in \emph{both types and
computations equally}, remedying \qtt{}'s inability to track
type-level use. We extend Abel's initial idea by presenting a
rich language, including dependent tensors, a complete metatheory,
and a \emph{graded modality} which
aids the practical use of this approach (e.g., enabling functions
to use components of data non-uniformly).
The result is a calculus which extends the power of existing non-dependent graded languages, like
Granule~\cite{orchard2019quantitative}, to a dependent setting.

We begin with the definition of \grtt{} in Section~\ref{sec:main},
before demonstrating the power of \grtt{} through case studies in
Section~\ref{sec:case-studies}, where we show how to use grading to
restrict \grtt{} terms to simply-typed reasoning, parametric reasoning
(regaining universal quantification smoothly within a dependent
theory), existential types, and linear types. The calculus can be
instantiated to different kinds of dataflow reasoning: we show an
example application to information-flow security. We then show the
metatheory of \grtt{} in Section~\ref{sec:metatheory}: admissibility
of graded structural rules, substitution, type preservation, and
strong normalisation.

We implemented a prototype language based on
\grtt{} called \implName{}.\footnote{\url{https://github.com/granule-project/gerty/releases/tag/esop2021}}
We briefly mention its syntax in Section~\ref{grtt--implementation}
for use in examples. Later, Section~\ref{sec:implementation}
describes how the formal definition of \grtt{} is implemented as a
bidirectional type checking algorithm, interfacing with an
SMT solver to solve constraints over grades. Furthermore,
Abel conjectured that a quantitative dependent theory could
enable usage-based optimisation of type-checking itself~\cite{abel2018},
which would assist dependently-typed programming at scale. We validate this
claim in Section~\ref{sec:implementation} showing a grade-directed
optimisation to \implName{}'s type checker.

Section~\ref{sec:discussion} discusses next steps for
increasing the expressive power of \grtt{}.
Full proofs and details
are provided in Appendix~\ref{appendix--supplement}.

\implName{} has some similarity to
Granule~\cite{orchard2019quantitative}: both are functional
languages with graded types. However, Granule has a linearly typed
core and no dependent types (only indexed types), thus has
no need for resource tracking at the type level (type indices
are not subject to tracking and their syntax is restricted).

\section{GrTT: \grttfull{}}\label{sec:main}

\noindent
\grtt{} augments a standard
presentation of dependent type theory with `grades' (elements
of a semiring) which account for how variables are
used, i.e., their \emph{dataflow}. Whilst existing work uses grades to describe
usage only in computational terms
(e.g.~\cite{gaboardi2014}), \grtt{} incorporates additional grades
to account for how variables are used in
types.
We introduce here the syntax and typing, and
briefly show the syntax of the implementation.
Section~\ref{sec:metatheory} describes its metatheory.

\subsection{Syntax}\label{grtt--syntax}
\noindent
The syntax of \grtt{} is that of a standard \mlof{} type
theory, with the addition of a \emph{graded modality} and
grade annotations on function and tensor binders. Throughout, $s$ and $r$
range over grades, which are elements of a semiring
$(\mathcal{R}, \ast, 1, +, 0)$. It is instructive to
instantiate this semiring to the natural number semiring
$(\mathbb{N}, \times, 1, +, 0)$, which captures the exact
number of times variables are used. We appeal to
this example in descriptions here.

\grtt{} has a single syntactic sort for computations and types:
\begin{align*}
\begin{array}{lrllll}
 (\textit{terms}) \;\;\; &t, A, B, C &::= &x\ &|\  \mathsf{Type}_{ l } 
  \\&&&|\  \textstyle (  x  :_{  \textcolor{darkblue}{( s ,  r )}  }  A  )  \to   B \ &|\  \lambda  x . t \   &|\  t_{{\mathrm{1}}} \,{ t_{{\mathrm{2}}} } 
  \\&&&|\  \textstyle (  x  :_{  \textcolor{darkblue}{ r }  }  A  )  \otimes   B \       &|\ \GrTTsym{(}  t_{{\mathrm{1}}}  \GrTTsym{,}  t_{{\mathrm{2}}}  \GrTTsym{)}\ &|\  \mathsf{let}\,  \textstyle (  x  ,  y  )   =  t_{{\mathrm{1}}}  \,\mathsf{in}\, t_{{\mathrm{2}}} 
  \\&&&|\  \square_{  \textcolor{darkblue}{ s }  }  A \        &|\  \square  t \  &|\  \mathsf{let}\,  \square  x   =  t_{{\mathrm{1}}}  \,\mathsf{in}\, t_{{\mathrm{2}}}  \\
 (\textit{levels}) & l &::= & \multicolumn{3}{l}{\GrTTsym{0} \mid  \mathsf{suc}\  l  \mid  l_{{\mathrm{1}}}  \mathop{\sqcup}  l_{{\mathrm{2}}} }
\end{array}
\end{align*}
Terms include variables and a constructor for an inductive
hierarchy of universes, annotated by a level
$l$. Dependent function types are annotated
with a pair of grades $s$ and $r$, with $s$
capturing how $x$ is used in the body of the inhabiting function
 and $r$ capturing how $x$ is used in the
codomain $B$.  Dependent tensors have a single grade
$r$, which describes how the first element is used in the
typing of the second. The graded modal type operator
$ \square_{  \textcolor{darkblue}{ s }  }  A $ `packages' a term and its dependencies so that
values of type $A$ can be used with grade $s$ in
the future.  Graded modal types are introduced via \emph{promotion}
$ \square  t $ and eliminated via
$ \mathsf{let}\,  \square  x   =  t_{{\mathrm{1}}}  \,\mathsf{in}\, t_{{\mathrm{2}}} $. The following sections
explain the semantics of each piece of syntax with respect to
its typing. We typically use $A$ and $B$ to connote terms used as types.

\subsection{Typing Judgments, Contexts, and Grading}\label{grtt--typing-judgments}
\noindent
Typing judgments are written in either of the following two
equivalent forms:
\begin{align*}
   ( \Delta  \mid  \sigma_{{\mathrm{1}}}  \mid  \sigma_{{\mathrm{2}}} )   \odot  \Gamma  \vdash  t  \GrTTsym{:}  A
\qquad
 &
\qquad
   \left(\begin{smallmatrix} \Delta  \\  \sigma_{{\mathrm{1}}}  \\  \sigma_{{\mathrm{2}}}  \end{smallmatrix}\right)   \odot  \Gamma  \vdash  t  \GrTTsym{:}  A
\end{align*}
The `horizontal' syntax (left) is used most often,
with the equivalent `vertical' form (right) used for clarity in some
places. Ignoring the part to the left of $\odot$,
typing judgments and their rules are essentially those of
\mlof{} type theory (with the addition of the modality)
where $\Gamma$ ranges over usual dependently-typed typing \emph{contexts}.
The left of $\odot$ provides the grading information,
where $\sigma$ and $\Delta$ range
over \emph{grade vectors} and \emph{context grade vectors}
respectively, of the form:
\begin{align*}
\begin{array}{ccc}
(\textit{contexts}) &
(\textit{grade vectors}) &
(\textit{context grade vectors})
\\
\quad \Gamma ::=  \emptyset \ |\ \Gamma  \GrTTsym{,}  x  \GrTTsym{:}  A \quad
&
\qquad \sigma ::=  \emptyset \ |\ \sigma,s \qquad
&
\qquad \Delta ::=  \emptyset \ |\ \Delta  ,  \sigma \qquad
\end{array}
\end{align*}
A grade vector $\sigma$ is a vector of semiring elements,
and a context vector $\Delta$ is a vector of
grade vectors. We write
$(s_1, \ldots, s_n)$ to denote an $n$-vector
and likewise for context grade
vectors. We omit parentheses when this would not cause
ambiguity. Throughout, a comma is used to concatenate vectors and disjoint
contexts, and to extend vectors with a single grade, grade vector,
or typing assumption.

For a judgment $ ( \Delta  \mid  \sigma_s  \mid  \sigma_r )   \odot  \Gamma  \vdash  t  \GrTTsym{:}  A$ the vectors
$\Gamma$, $\Delta$, $\sigma_s$, and $\sigma_r$ are all of equal
size. Given a typing assumption $y  \GrTTsym{:}  B$ at index $i$ in
$\Gamma$, the grade $\sigma_s[i] \in \mathcal{R}$ denotes the use
of $y$ in $t$ (the \emph{subject} of the judgment), the grade
$\sigma_r[i] \in \mathcal{R}$ denotes the use of $y$ in $A$ (the
\emph{subject's type}), and $\Delta[i] \in \mathcal{R}^i$ (of
size $i$) describes how assumptions prior to $y$ are used to
form $y$'s type, $B$.

Consider the following
example, which types the body of a function that takes two
arguments of type $a$, and returns only the first:
\begin{align*}
   \left(\begin{smallmatrix}   \GrTTsym{(}  \GrTTsym{)}  ,  \GrTTsym{(}  \GrTTsym{1}  \GrTTsym{)}  ,  \GrTTsym{(}  \GrTTsym{1}  ,  \GrTTsym{0}  \GrTTsym{)}    \\    \GrTTsym{0}  ,  \GrTTsym{1}  ,  \GrTTsym{0}    \\    \GrTTsym{1}  ,  \GrTTsym{0}  ,  \GrTTsym{0}    \end{smallmatrix}\right)   \odot  a  \GrTTsym{:}   \mathsf{Type}_{ l }   \GrTTsym{,}  x  \GrTTsym{:}  a  \GrTTsym{,}  y  \GrTTsym{:}  a  \vdash  x  \GrTTsym{:}  a
\end{align*}
Let the context grade vector be called $\Delta$.  Then,
$\Delta[0] = \GrTTsym{(}  \GrTTsym{)}$ (empty vector) explains that there are no assumptions
that are used to type $a$ in the context, as $ \mathsf{Type}_{ l } $ is a
closed term and the first assumption. $\Delta[1] = (1)$
explains that the first assumption $a$ is used (grade $1$)
in the typing of $x$ in the context, and
$\Delta[2] = (1, 0)$, explains that $a$ is used once in the
typing of $y$ in the context, and $x$ is unused in the typing
of $y$. The subject grade vector $\sigma_s = (0, 1, 0)$
explains that $a$ is unused in the subject, $x$ is used once,
and $y$ is unused. Finally, the subject type vector
$\sigma_r = (1, 0, 0)$ explains that $a$ appears once in the
subject's type (which is just $a$), and $x$ and $y$ are unused
in the formation of the subject's type.

To aid reading, recall that
standard typing rules typically have the form $\textit{context}
\vdash \textit{subject} : \textit{subject-type}$, the order of which
is reflected by $ ( \Delta  \mid  \sigma_s  \mid  \sigma_r )  \odot \ldots$ giving
the context, subject, and subject-type grading respectively.

\paragraph{Well-formed Contexts}\label{grtt--well-formed-contexts}

The relation $\Delta  \odot  \Gamma  \vdash$ identifies a context $\Gamma$ as
well-formed with respect to context grade vector $\Delta$,
defined by the following rules:
\begin{align*}
  \GrTTdruleWfXXEmpty{} \qquad \GrTTdruleWfXXExt{}
\end{align*}
Unlike typing, well-formedness does not need
to include subject and subject-type grade vectors, as it
considers only the well-formedness of the assumptions in a
context with respect to prior assumptions in the context. The
$\GrTTdruleWfXXEmptyName{}$ rule states that the empty context
is well-formed with an empty context grade vector as there
are no assumptions to account for.  The
$\GrTTdruleWfXXExtName{}$ rule states that given $A$ is a type
under the assumptions in $\Gamma$, with $\sigma$ accounting for
the usage of $\Gamma$ variables in $A$, and $\Delta$ accounting for usage within
$\Gamma$, then we can form the well-formed context
$\Gamma  \GrTTsym{,}  x  \GrTTsym{:}  A$ by extending $\Delta$ with $\sigma$ to account
for the usage of $A$ in forming the context. The notation $ \textbf{0} $
denotes a vector for which each element is the semiring
$\GrTTsym{0}$. Note that the well-formedness $\Delta  \odot  \Gamma  \vdash$ is
inherent from the premise of $\GrTTdruleWfXXExtName{}$ due to the following lemma:
\begin{glemma}{wfFromTyping}{Typing contexts are well-formed}
  If $ ( \Delta  \mid  \sigma_{{\mathrm{1}}}  \mid  \sigma_{{\mathrm{2}}} )   \odot  \Gamma  \vdash  t  \GrTTsym{:}  A$ then $\Delta  \odot  \Gamma  \vdash$.
\end{glemma}

\subsection{Typing Rules}\label{grtt--typing-rules}

We examine the typing rules of \grtt{} one at a time.
The rules are collected in Appendix~\ref{grtt--full-rules}.

\noindent
Variables are introduced as
follows:
\begin{align*}
  \GrTTdruleTXXVar{}
\end{align*}
The premise identifies $\Gamma_{{\mathrm{1}}}  \GrTTsym{,}  x  \GrTTsym{:}  A  \GrTTsym{,}  \Gamma_{{\mathrm{2}}}$ as well-formed
under the context grade vector $\Delta_{{\mathrm{1}}}  ,  \sigma  ,  \Delta_{{\mathrm{2}}}$. By the size
condition $ \left|  \Delta_{{\mathrm{1}}}  \right|   \GrTTsym{=}   \left|  \Gamma_{{\mathrm{1}}}  \right| $, we are able to identify $\sigma$ as
capturing the usage of the variables $\Gamma_{{\mathrm{1}}}$ in forming $A$. This
information is used in the conclusion, capturing type-level variable
usage as $\sigma  ,  \GrTTsym{0}  ,   \textbf{0} $, which describes that $\Gamma_{{\mathrm{1}}}$ is used
according to $\sigma$ in the subject's type ($A$), and that
the $x$ and the variables of $\Gamma_{{\mathrm{2}}}$ are
used with grade $\GrTTsym{0}$.  For subject usage, we annotate the
first zero vector with a size $|\Delta_{{\mathrm{1}}}|$, allowing us to single out
$x$ as being the only assumption used with grade $\GrTTsym{1}$
in the subject; all other assumptions are used with
grade $\GrTTsym{0}$.

For example, typing the body of the polymorphic identity ends with $\GrTTdruleTXXVarName{}$:
\begin{align*}
  \inferrule*[right=\GrTTdruleTXXVarName{}]{
    \inferrule*[right=\GrTTdruleWfXXExtName{}]{
      \cdots
    }{ {\GrTTsym{(}  \GrTTsym{(}  \GrTTsym{)}  ,  \GrTTsym{(}  \GrTTsym{1}  \GrTTsym{)}  \GrTTsym{)}  \odot  a  \GrTTsym{:}   \mathsf{Type}   \GrTTsym{,}  x  \GrTTsym{:}  a  \vdash} }
    \quad
    { \left|  \GrTTsym{(}  \GrTTsym{(}  \GrTTsym{)}  \GrTTsym{)}  \right|   \GrTTsym{=}   \left|  a  \GrTTsym{:}   \mathsf{Type}   \right| }
  }{ { ( \GrTTsym{(}  \GrTTsym{(}  \GrTTsym{)}  ,  \GrTTsym{(}  \GrTTsym{1}  \GrTTsym{)}  \GrTTsym{)}  \mid    \GrTTsym{0}  ,  \GrTTsym{1}    \mid    \GrTTsym{1}  ,  \GrTTsym{0}   )   \odot  a  \GrTTsym{:}   \mathsf{Type}   \GrTTsym{,}  x  \GrTTsym{:}  a  \vdash  x  \GrTTsym{:}  a} }
\end{align*}
The premise implies that
$\GrTTsym{(}  \GrTTsym{(}  \GrTTsym{)}  \GrTTsym{,}  \GrTTsym{1}  \GrTTsym{,}  \GrTTsym{0}  \GrTTsym{)}  \odot  a  \GrTTsym{:}   \mathsf{Type}   \vdash  a  \GrTTsym{:}   \mathsf{Type} $ by the following
lemma:
\begin{lemma}[Typing an assumption in a well-formed context]
\label{lemma:wfTypingAss}
  If $\Delta_{{\mathrm{1}}}  ,  \sigma  ,  \Delta_{{\mathrm{2}}}$ $\odot\,\Gamma_{{\mathrm{1}}}  \GrTTsym{,}  x  \GrTTsym{:}  A  \GrTTsym{,}  \Gamma_{{\mathrm{2}}} \vdash$ with
  $ \left|  \Delta_{{\mathrm{1}}}  \right|   \GrTTsym{=}   \left|  \Gamma_{{\mathrm{1}}}  \right| $, then
  $ ( \Delta_{{\mathrm{1}}}  \mid  \sigma  \mid   \textbf{0}  )   \odot  \Gamma_{{\mathrm{1}}}  \vdash  A  \GrTTsym{:}   \mathsf{Type}_{ l } $ for some
  $l$.
\end{lemma}
In the conclusion of $\GrTTdruleTXXVarName{}$,
the typing $\GrTTsym{(}  \GrTTsym{(}  \GrTTsym{)}  \GrTTsym{,}  \GrTTsym{1}  \GrTTsym{,}  \GrTTsym{0}  \GrTTsym{)}  \odot  a  \GrTTsym{:}   \mathsf{Type}   \vdash  a  \GrTTsym{:}   \mathsf{Type} $ is `distributed' to the typing of $x$ in the context
and to the formation the subject's
type. Thus subject grade $\GrTTsym{(}  \GrTTsym{0}  ,  \GrTTsym{1}  \GrTTsym{)}$ corresponds to the
absence of $a$ from the subject and the presence of
$x$, and subject-type grade $\GrTTsym{(}  \GrTTsym{1}  ,  \GrTTsym{0}  \GrTTsym{)}$ corresponds
to the presence of $a$ in the subject's type ($a$),
and the absence of $x$.

Typing universes are formed as follows:
\begin{align*}
\GrTTdruleTXXType{}
\end{align*}
We use an inductive hierarchy of
universes~\cite{palmgren1998universes} with ordering
$<$ such that $l <  \mathsf{suc}\  l $. Universes can be formed under any well-formed context,
with every assumption graded with $0$ subject and subject-type
use, capturing the absence of any assumptions from the
universes, which are closed forms.

\paragraph{Functions}

Function types $ \textstyle (  x  :_{  \textcolor{darkblue}{( s ,  r )}  }  A  )  \to   B $ are annotated with two grades:
explaining that $x$ is used with
grade $s$ in the body of the inhabiting function
and with grade $r$ in $B$. Function types have the following formation rule:
\begin{align*}
  \GrTTdruleTXXArrow{}
\end{align*}
The usage of the dependencies of $A$ and $B$ (excepting $x$)
are given by $\sigma_{{\mathrm{1}}}$ and $\sigma_{{\mathrm{2}}}$ in the premises (in the
`subject' position)
which are combined as $\sigma_{{\mathrm{1}}}  \GrTTsym{+}  \sigma_{{\mathrm{2}}}$ (via pointwise vector addition using
the $+$ of the semiring), which serves to \emph{contract} the
dependencies of the two types. The usage of $x$ in $B$ is
captured by $r$, and then internalised to the binder in the
conclusion of the rule. An arbitrary grade for $s$ is allowed
here as there is no information on how $x$ is used in
an inhabiting function body.
Function terms are then typed by the following rule:
\begin{align*}
  \footnotesize
  \GrTTdruleTXXFun{}
\end{align*}
The second premise types the body of the $\lambda$-term,
showing that $s$ captures the usage of $x$ in $t$
and $r$ captures the usage of $x$ in $B$; the
subject and subject-type grades of $x$ are then internalised as
annotations on the function type's binder.

Dependent functions are eliminated through application:
\begin{align*}
  \GrTTdruleTXXApp{}
\end{align*}
where $*$ is the scalar
multiplication of a vector, using the
semiring multiplication. Given a function $t_{{\mathrm{1}}}$ which
uses its parameter with grade $s$ to compute and
with grade $r$ in the typing of the result, we can apply it to
a term $t_{{\mathrm{2}}}$, provided that we have the resources required
to form $t_{{\mathrm{2}}}$ scaled by $s$ at the subject level and by $r$ at the
subject-type level, since $t_{{\mathrm{2}}}$ is substituted into the return
type $B$. This scaling behaviour is akin to that used in
coeffect calculi~\cite{ghica2014,Petricek:2014},
\qtt{}~\cite{quantitative-type-theory,McBride2016} and Linear
Haskell~\cite{bernardy2017linear}, but scalar multiplication
happens here at both the subject and subject-type level.  The
use of variables in $A$ is accounted for by $\sigma_{{\mathrm{1}}}$
as explained in the third premise, but these usages are not
present in the resulting application since $A$ no longer
appears in the types or the terms.

Consider the constant function
$ \lambda  x .  \lambda  y . x   :  \textstyle (  x  :_{  \textcolor{darkblue}{( \GrTTsym{1} ,  \GrTTsym{0} )}  }  A  )  \to    \textstyle (  y  :_{  \textcolor{darkblue}{( \GrTTsym{0} ,  \GrTTsym{0} )}  }  B  )  \to   A  $ (for
some $A$ and $B$). Here the resources required for the second
parameter will always be scaled by $0$, which is absorbing,
meaning that anything passed as the second argument has $0$
subject and subject-type use. This example begins to show some
of the power of grading---the grades capture the program
structure at all levels.

\paragraph{Tensors}

The rule for forming dependent tensor types is as follows:
\begin{align*}
  \GrTTdruleTXXTen{}
\end{align*}
This rule is almost identical to function type formation
$\GrTTdruleTXXArrowName{}$ but with only a single grade
$r$ on the binder, since $x$ is only bound in $B$ (the
type of the second component), and not computationally. For
`quantitative' semirings, where $0$ really means unused (see
Section~\ref{sec:case-studies}), $ \textstyle (  x  :_{  \textcolor{darkblue}{ \GrTTsym{0} }  }  A  )  \otimes   B $ is then a
product $A \times B$.

Dependent tensors are introduced as follows:
\begin{align*}
  \GrTTdruleTXXPair{}
\end{align*}
In the typing premise for $t_{{\mathrm{2}}}$, occurrences of $x$ are
replaced with $t_{{\mathrm{1}}}$ in the type, ensuring that the type of
the second component ($t_{{\mathrm{2}}}$) is calculated using the first
component ($t_{{\mathrm{1}}}$). The resources for $t_{{\mathrm{1}}}$ in this
substitution are scaled by $r$, accounting for the
existing usage of $x$ in $B$. In the conclusion, we
see the resources for the two components (and their types)
combined via the semiring addition.

Finally, tensors are eliminated with the following rule:
\begin{align*}
  \GrTTdruleTXXTenCut{}
\end{align*}
As this is a dependent eliminator, we allow the result type $C$
to depend upon the value of the tensor as a whole, bound
as $z$ in the second premise with grade $r'$, into which is substituted
our actual tensor term $t_{{\mathrm{1}}}$ in the conclusion.

Eliminating a tensor ($t_{{\mathrm{1}}}$) requires that we consider each
component ($x$ and $y$) is used with the same grade $s$ in the
resulting expression $t_{{\mathrm{2}}}$, and that we scale the resources
of $t_{{\mathrm{1}}}$ by $s$. This is because we cannot inspect $t_{{\mathrm{1}}}$
itself, and semiring addition is not injective (preventing us
from splitting the grades required to form $t_{{\mathrm{1}}}$). This
prevents forming certain functions (e.g., projections) under
some semirings, but this can be overcome by the introduction
of \emph{graded modalities}.

\paragraph{Graded Modality}

Graded binders alone do not allow different parts of a value
to be used differently, e.g., computing the length of a list
ignores the elements, projecting from a pair discards one
component. We therefore introduce a \emph{graded modality}
(\`{a} la~\cite{gaboardi2014,orchard2019quantitative}) allowing us to
capture the notion of local inspection on data and
internalising usage information into types. A type
$ \square_{  \textcolor{darkblue}{ s }  }  A $ denotes terms of type $A$ that are used with
grade $s$.
Type formation and introduction rules are:
\begin{align*}
  \GrTTdruleTXXBox{}
\quad
  \GrTTdruleTXXBoxI{}
\end{align*}
To form a term of type $ \square_{  \textcolor{darkblue}{ s }  }  A $, we `promote' a term $t$
of type $A$ by requiring that we can use the resources used to
form $t$ ($\sigma_{{\mathrm{1}}}$) according to grade $s$. This
`promotion' resembles that of other graded modal systems
(e.g.,~\cite{DBLP:journals/pacmpl/AbelB20,gaboardi2014,gaboardi2016combining,orchard2019quantitative}), but the
elimination needs to also account for type usage due to
dependent elimination.

We can see promotion $\Box_i$ as capturing $t$ for later
use according to grade $s$. Thus, when eliminating a term of
type $ \square_{  \textcolor{darkblue}{ s }  }  A $, we must consider how the `unboxed' term is
used with grade $s$, as per the following dependent
eliminator:
\begin{align*}
  \small
  \GrTTdruleTXXBoxE{}
\end{align*}
This rule can be understood as a kind of `cut', connecting a
`capability' to use a term of type $A$ according to grade
$s$ with the requirement that $x  \GrTTsym{:}  A$ is used according to
grade $s$ as a dependency of $t_{{\mathrm{2}}}$. Since we are in a
dependently-typed setting, we also substitute $t_{{\mathrm{1}}}$ into
the type level such that $B$ can depend on $t_{{\mathrm{1}}}$
according to grade $r$ which then causes the dependencies
of $t_{{\mathrm{1}}}$ ($\sigma_{{\mathrm{1}}}$) to be scaled-up by $r$ and added to
the subject-type grading.

\paragraph{Equality, Conversion, and Subtyping}

A key part of dependent type theories is a notion of
term equality and type conversion~\cite{hofmann1997syntax}.
\grtt{} term equality is via judgments $ ( \Delta  \mid  \sigma_{{\mathrm{1}}}  \mid  \sigma_{{\mathrm{2}}} )   \odot  \Gamma  \vdash  t_{{\mathrm{1}}}  \GrTTsym{=}  t_{{\mathrm{2}}}  \GrTTsym{:}  A$ equating terms $t_{{\mathrm{1}}}$ and $t_{{\mathrm{2}}}$ of type $A$.
Equality includes full
congruences as well as $\beta\eta$-equality for functions,
tensors, and graded modalities, of which the latter are:
\begin{align*}
  \small
\begin{array}{c}
\GrTTdruleTEQXXBoxB{}
\\[1.5em]
\GrTTdruleTEQXXBoxU{}
\end{array}
\end{align*}
A subtyping relation ($ ( \Delta  \mid  \sigma )   \odot  \Gamma  \vdash  A  \leq  B$) subsumes
equality, adding ordering of universe levels. \emph{Type
  conversion} allows re-typing terms based on the judgment:
\begin{align*}
  \GrTTdruleTXXTyConv{}
\end{align*}
The full rules for equality and subtyping are in
Appendix~\ref{grtt--full-rules}.

\subsection{Operational Semantics}
\label{sec:operational-semantics}

As with other graded modal calculi (e.g.,~\cite{DBLP:journals/pacmpl/AbelB20,gaboardi2014,gaboardi2016combining}),
the core calculus of \grtt{} has a Call-by-Name small-step operational semantics
with reductions $ t  \leadsto  t' $. The rules are standard, with the
addition of the $\beta$-rule for the graded modality:
\begin{align*}
  \mathsf{let}\,  \square  x   =   \square  t_{{\mathrm{1}}}   \,\mathsf{in}\, t_{{\mathrm{2}}}   \leadsto  \GrTTsym{[}  t_{{\mathrm{1}}}  \GrTTsym{/}  x  \GrTTsym{]}  t_{{\mathrm{2}}} 
\qquad
(\GrTTdruleSemXXBetaBoxName{})
\end{align*}
Type preservation and normalisation are considered in Section~\ref{sec:metatheory}.

\subsection{Implementation and Examples}\label{grtt--implementation}
To explore our theory, we provide an implementation,
\implName{}. Section~\ref{sec:implementation}
describes how the declarative definition of the type theory is implemented as a
bidirectional type checking algorithm. We briefly mention the syntax
here for use in later examples. The following is the polymorphic identity function
in \implName{}:
\begin{gerty}
id : (a : (.0, .2) Type 0) -> (x : (.1, .0) a) -> a
id = \a -> \x -> x
\end{gerty}
The syntax resembles the theory, where grading terms
\gertyin{.n} are syntactic sugar for a unary encoding of
grades in terms of $0$ and repeated addition of $1$, e.g.,
\gertyin{.2} $=$ \gertyin{(.0 + .1) + .1}.  This syntax can be
used for grade terms of any semiring, which can be resolved to
particular built-in semirings at other points of type
checking.

The following shows first projection on (non-dependent) pairs, using
the graded modality (at grade 0 here) to give fine-grained usage on compound data:
\begin{gerty}
fst : (a : (.0, .2) Type 0) (b : (.0, .1) Type 0) -> <a * [.0] b> -> a
fst = \a b p -> case p of <x, y> -> let [z] = y in x
\end{gerty}
The implementation adds various built-in semirings, some syntactic sugar,
and extras such as: a singleton \emph{unit} type, extensions
of the theory to semirings with a pre-ordering (discussed
further in Section~\ref{sec:discussion}), and some implicit
resolution. Anywhere a grade is expected, an underscore can
be supplied to indicate that \implName{} should try to resolve the
grade implicitly. Grades may also be omitted from binders (see above
in \gertyin{fst}), in
which case they are treated as implicits. Currently, implicits
are handled by generating existentially quantified grade
variables, and using SMT to solve the necessary constraints
(see Section~\ref{sec:implementation}).

So far we have considered the natural numbers semiring providing an
analysis of usage. We come back to this and similar examples in
Section~\ref{sec:case-studies}.  To show another kind of example, we
consider a lattice semiring of privacy levels (appearing
elsewhere~\cite{DBLP:journals/pacmpl/AbelB20,gaboardi2016combining,orchard2019quantitative})
which enforces information-flow control, akin to
DCC~\cite{DCC}. Differently to DCC, dataflow is tracked through
variable dependencies, rather than through the results of computations
in the monadic style of DCC.

\begin{definition}[Security levels]
Let $\mathcal{R} = \gertyin{Lo} \leq \gertyin{Hi}$ be a set
of labels with $0 = \gertyin{Hi}$ and
$1 = \gertyin{Lo}$, semiring addition as the meet and
multiplication as join.
Here, $1 = \gertyin{Lo}$ treats the base notion of dataflow as
being in the low security (public) domain.  Variables graded with
\gertyin{Hi} must then be unused, or guarded by a graded
modality.
This semiring is primitive in \implName{};
we can express the following example:
\begin{gerty}
idLo : (a : (.0, .2) Type 0) -> (x : (Lo, Hi) a) -> a
idLo = \a -> \x -> x
-- The following is rejected as ill-typed
leak : (a : (.0, .2) Type 0) -> (x : (Hi, Hi) a) -> a
leak = \a -> \x -> idLo a x
\end{gerty}
The first definition is well-typed, but the second yields
a typing error originating from the application in its body:
\begin{gertyError}
At subject stage got the following mismatched grades:
 For 'x' expected Hi but got .1
\end{gertyError}
where grade \gertyin{1} is \gertyin{Lo}
here. Thus we can use this abstract
label semiring as a way of restricting flow of data between regions
(\emph{cf.} region typing
systems~\cite{henglein2005effect,tofte1997region}).
Note that the ordering is not leveraged here other than in the
lattice operations.
\end{definition}

\section{Case Studies}\label{sec:case-studies}

We now demonstrate \grtt{} via
several cases studies that focus the reasoning power of
dependent types via grading. Since grading in \grtt{} serves
to explain dataflow, we can characterise subsets of \grtt{}
that correspond to various type theories. We demonstrate the
approach with simple types, parametric polymorphism, and
linearity. In each case study, we restrict \grtt{} to a subset
by \emph{a characterisation of the grades}, rather than
by, say, placing detailed syntactic restrictions or employing
meta-level operations or predicates that restrict syntax
(as one might do for example to map a subset of
Martin-L\"{o}f type theory into the simply-typed
$\lambda$-calculus by restriction to closed
types, requiring deep inspection of type terms).  Since this
restriction is only on grades, we can harness the specific
reasoning power of particular calculi from within the language
itself, simply by specifications on grades. In the context of
an implementation like \implName{}, this amounts to using type
signatures to restrict dataflow.

This section shows the power of tracking dataflow in types via
grades, going beyond
\qtt{}~\cite{quantitative-type-theory} and
\textsc{GraD}~\cite{choudhury2021}. For `quantitative'
semirings, a $0$ type-grade means that we can recover simply-typed
reasoning (Section~\ref{subsec:embedding-stlc}) and distinguish
computational functions from type-parameter functions for parametric
reasoning (Section~\ref{sec:parametricity}), embedding a
grade-restricted subset of \grtt{} into System F.

Section~\ref{sec:more-case-studies} returns to a case study that
builds on the implementation.

\subsection{Recovering \mlof{} Type Theory}

When the semiring parameterising \grtt{} is the singleton
semiring (i.e., any semiring where $1 = 0$), we have an
isomorphism $ \square_{  \textcolor{darkblue}{ r }  }  A  \cong A$, and grade annotations
become redundant, as all grades are equal. All vectors and
grades on binders may then be omitted, and we can write typing
judgments as $\Gamma \vdash t : A$, giving rise to a standard
\mlof{} type theory as a special case of \grtt{}.

\subsection{Determining Usage via Quantitative Semirings}

Unlike existing systems, we can use the fine-grained grading
to \emph{guarantee} the relevance or irrelevance of
assumptions in types. To do this we must consider a subset of
semirings $(\mathcal{R}, *, 1, +, 0)$ called
\emph{quantitative} semirings, satisfying:
\begin{itemize}[leftmargin=10em]
\item[(zero-unique)] $1 \neq 0$;
\item[(positivity)] $\forall r, s .\ r  \GrTTsym{+}  s = 0 \implies r = 0 \wedge s = 0$;
\item[(zero-product)] $\forall r, s .\  r  \ast  s  = 0 \implies r = 0 \vee s = 0$.
\end{itemize}
These axioms\footnote{Atkey requires
  \emph{positivity} and \emph{zero-product} for all
  semirings parameterising \qtt{}~\cite{quantitative-type-theory} (as
  does Abel~\cite{abel2018}). Atkey
  imposes this for admissibility of substitution.  We need
  not place this restriction on \grtt{} to have
  substitution in general (Sec.~\ref{grtt--substitution}).}  ensure
that a $0$-grade in a quantitative semiring represents irrelevant
variable use. This notion has recently been proved
for computational use
by Choudhury et al.~\cite{choudhury2021} via a heap-based semantics for
grading (on computations) and the same result applies here.
Conversely, in a quantitative semiring any grade other than $0$
denotes relevance. From this, we can \emph{directly} encode
non-dependent tensors and arrows: in $ \textstyle (  x  :_{  \textcolor{darkblue}{ \GrTTsym{0} }  }  A  )  \otimes   B $ the
grade $0$ captures that $x$ cannot have any computational
content in $B$, and likewise for $ \textstyle (  x  :_{  \textcolor{darkblue}{( s ,  \GrTTsym{0} )}  }  A  )  \to   B $
the grade $0$ explains that $x$ cannot have any computational
content in $B$, but may have computational use
according to $s$ in the inhabiting function.
Thus, the grade $0$ here describes that
elimination forms \emph{cannot} ever inspect the variable during
normalisation.
Additionally, quantitative semirings can be used for encoding simply-typed and
polymorphic reasoning.

\begin{example}Some quantitative semirings are:
\begin{itemize}[topsep=0em,itemsep=0.25em,leftmargin=1em]
\item (\textit{Exact usage}) $(\mathbb{N}, \times, 1, +, 0)$;
\item (\textit{0-1}) The semiring over $\mathcal{R} = \{0, 1\}$ with $1 + 1 = 1$
  which describes relevant \textit{vs.} irrelevant dependencies, but
  no further information.
\item (\textit{None-One-Tons}~\cite{McBride2016})
The semiring on $\mathcal{R} = \{0, 1, \infty\}$ is
more fine-grained than 0-1, where $\infty$ represents more
than 1 usage, with $1 + 1 = \infty = 1 + \infty$.
\end{itemize}
\end{example}
\subsection{Simply-typed Reasoning}
\label{subsec:embedding-stlc}

\newcommand{\stlc}[1]{\textsc{Stlc}(#1)} As discussed in
Section~\ref{sec:introduction}, the simply-typed
$\lambda$-calculus (STLC) can be distinguished from
dependently-typed calculi via the restriction of dataflow: in
simple types, data can only flow at the computational level,
with no dataflow within, into, or from types. We can thus
view a \grtt{} function as simply typed when its variable
is irrelevant in the type, e.g., $ \textstyle (  x  :_{  \textcolor{darkblue}{( s ,  \GrTTsym{0} )}  }  A  )  \to   B $ for
quantitative semirings.  We define a subset of \grtt{}
restricted to simply-typed reasoning:
\begin{definition}[Simply-typed \grtt{}]
\label{sec:def-stlc-grtt-subset}
For a quantitative semiring,
the following predicate $\stlc{-}$
determines a subset of simply-typed \grtt{} programs:
\begin{align*}
& \stlc{ (  \emptyset   \mid   \emptyset   \mid   \emptyset  )   \odot   \emptyset   \vdash  t  \GrTTsym{:}  A}
\\
& \stlc{ ( \Delta  \mid  \sigma_{{\mathrm{1}}}  \mid  \sigma_{{\mathrm{2}}} )   \odot  \Gamma  \vdash  t  \GrTTsym{:}  A}
\!\implies\!
\stlc{ ( \Delta  ,   \textbf{0}   \mid  \sigma_{{\mathrm{1}}}  ,  s  \mid  \sigma_{{\mathrm{2}}}  ,  \GrTTsym{0} )   \odot  \Gamma  \GrTTsym{,}  x  \GrTTsym{:}  B  \vdash  t  \GrTTsym{:}  A}
\end{align*}
That is, all subject-type grades are $0$
(thus function types are of the form $ \textstyle (  x  :_{  \textcolor{darkblue}{( s ,  \GrTTsym{0} )}  }  A  )  \to   B $).
A similar predicate is
defined on well-formed contexts (elided),
restricting context grades of well-formed contexts to
only zero grading vectors.
\end{definition}
\noindent

Under the restriction of
Definition~\ref{sec:def-stlc-grtt-subset}, a subset of \grtt{}
terms embeds into the simply-typed $\lambda$-calculus in a
sound and complete way. Since STLC does not have a notion of tensor or modality, this is
omitted from the encoding:
\begin{align*}
\interp{x} = x \quad
\interp{ \lambda  x . t } = \lambda x . \interp{t} \quad
\interp{ t_{{\mathrm{1}}} \,{ t_{{\mathrm{2}}} } } = \interp{t_{{\mathrm{1}}}}\interp{t_{{\mathrm{2}}}} \quad
\interp{ \textstyle (  x  :_{  \textcolor{darkblue}{( s ,  \GrTTsym{0} )}  }  A  )  \to   B }_\tau \!=\! \interp{A}_\tau \rightarrow \interp{B}_\tau
\end{align*}
Variable contexts of \grtt{} are interpreted by point-wise
applying $\interp{-}_\tau$ to typing assumptions. We then get
the following preservation of typing into the simply-typed
$\lambda$-calculus, and soundness and completeness of this
encoding:
\begin{lemma}[Soundness of typing]
Given a derivation of
$ ( \Delta  \mid  \sigma_{{\mathrm{1}}}  \mid  \sigma_{{\mathrm{2}}} )   \odot  \Gamma  \vdash  t  \GrTTsym{:}  A$
such that
$\stlc{ ( \Delta  \mid  \sigma_{{\mathrm{1}}}  \mid  \sigma_{{\mathrm{2}}} )   \odot  \Gamma  \vdash  t  \GrTTsym{:}  A}$
then
$\interp{\Gamma}_\tau \vdash \interp{t} : \interp{A}_\tau$
in STLC.
\end{lemma}

\begin{theorem}[Soundness and completeness of the embedding]
Given $\stlc{ ( \Delta  \mid  \sigma_{{\mathrm{1}}}  \mid  \sigma_{{\mathrm{2}}} )   \odot  \Gamma  \vdash  t  \GrTTsym{:}  A}$
and
$\interp{ ( \Delta  \mid  \sigma_{{\mathrm{1}}}  \mid  \sigma_{{\mathrm{2}}} )   \odot  \Gamma  \vdash  t  \GrTTsym{:}  A}$
then for CBN reduction $\leadsto^{\textsc{stlc}}$ in
simply-typed $\lambda$-calculus:
\begin{align*}
\begin{array}{rrll}
\textit{(soundness)} & \forall t' . & \textit{if}\  t \leadsto t' & \textit{then}\ \interp{t}
  \leadsto^{\textsc{stlc}} \interp{t'} \\
\textit{(completeness)} & \forall t_a. & \textit{if}\ \interp{t} \leadsto^{\textsc{stlc}} t_a
& \textit{then}\ \exists t' .\ t \leadsto t'\ \wedge\ \interp{t'} \equiv_{\beta\eta}
  t_a
\end{array}
\end{align*}
\end{theorem}
\noindent
Thus, we capture simply-typed
reasoning just by restricting type grades to $0$
for quantitative semirings. We consider quantitative
semirings again for parametric reasoning,
but first recall issues with parametricity and
dependent types.

\subsection{Recovering Parametricity via Grading}\label{sec:parametricity}

\newcommand{\ri}[3][]{\mathsf{RI}#1 \ #2\ #3}

One powerful feature of grading in a dependent type setting is
the ability to recover parametricity from dependent function
types. Consider the following type of functions in System F
(we borrow this example from Nuyts et al.~\cite{nuytsParametricQuantifiersDependent2017}):
\begin{align*}
\ri{A}{B}\ \triangleq\ \forall \gamma . (\gamma \rightarrow A) \rightarrow (\gamma \rightarrow B)
\end{align*}
Due to parametricity, we get the following notion of
\emph{representation independence} in System F: for a function
$f : \ri{A}{B}$, some type $\gamma'$,
and terms $h : \gamma' \rightarrow A$ and $c : \gamma'$, then we know
that $f$ can only use $c$ by applying $h\ c$. Subsequently, $\ri{A}{B} \cong A \rightarrow B$ by
 parametricity~\cite{reynolds1974towards}, defined uniquely as:
\begin{align*}
\begin{array}{ll}
\textit{iso} : \ri{A}{B} \rightarrow (A \rightarrow B) \qquad & \qquad
\textit{iso}^{-1} : (A \rightarrow B) \rightarrow \ri{A}{B} \\[0em]
\textit{iso}\ f = f\ A\ (\textit{id}\ A)
\qquad & \qquad
\textit{iso}^{-1}\ g = \Lambda \gamma .\ \lambda h .\ \lambda (c :
         \gamma) .\ g (h \, c)
\end{array}
\end{align*}
In a dependently-typed language, one might seek to replace System F's
universal quantifier with $\Pi$-types, i.e.
\begin{align*}
\ri[']{A}{B}\ \triangleq\ (\gamma : \mathsf{Type}) \rightarrow (\gamma \rightarrow A) \rightarrow (\gamma \rightarrow B)
\end{align*}
However, we can no longer reason parametrically about the inhabitants
of such types (we cannot prove that $\ri[']{A}{B} \cong A \rightarrow B$) as the
free interaction of types and computational terms allows
us to give the following non-parametric element of
$\ri[']{A}{B}$ over `large' type instances:
\begin{align*}
\textit{leak} = \lambda \gamma .\ \lambda h .\ \lambda c .\ \gamma : \ri[']{A}{\mathsf{Type}}
\end{align*}
Instead of applying $h\ c$, the above
``leaks'' the type parameter $\gamma$. \grtt{} can
recover universal quantification, and hence parametric
reasoning, by using grading to restrict the data-flow
capabilities of a $\Pi$-type. We can refine
representation independence to the following:
\begin{align*}
\ri['']{A}{B}\ \triangleq\  \textstyle (  \gamma  :_{  \textcolor{darkblue}{( \GrTTsym{0} ,  \GrTTsym{2} )}  }   \mathsf{Type}   )  \to    \textstyle (  h  :_{  \textcolor{darkblue}{( s_{{\mathrm{1}}} ,  \GrTTsym{0} )}  }   \textstyle (  x  :_{  \textcolor{darkblue}{( s_{{\mathrm{2}}} ,  \GrTTsym{0} )}  }  \gamma  )  \to   A   )  \to    \textstyle (  \GrTTmv{c}  :_{  \textcolor{darkblue}{( s_{{\mathrm{3}}} ,  \GrTTsym{0} )}  }  \gamma  )  \to   B   
\end{align*}
for some grades $s_{{\mathrm{1}}}$, $s_{{\mathrm{2}}}$, and $s_{{\mathrm{3}}}$,
and with shorthand $2 = 1 + 1$.

If we look at the definition of \textit{leak} above, we see
that $\gamma$ is used in the body of the function and thus
requires usage $1$, so \textit{leak} cannot
inhabit $\ri['']{A}{ \mathsf{Type} }$. Instead, \textit{leak} would
be typed differently as:
\begin{equation*}
\textit{leak} :  \textstyle (  \gamma  :_{  \textcolor{darkblue}{( \GrTTsym{1} ,  \GrTTsym{2} )}  }   \mathsf{Type}   )  \to    \textstyle (  h  :_{  \textcolor{darkblue}{( \GrTTsym{0} ,  \GrTTsym{0} )}  }   \textstyle (  x  :_{  \textcolor{darkblue}{( s ,  \GrTTsym{0} )}  }  \gamma  )  \to   A   )  \to    \textstyle (  \GrTTmv{c}  :_{  \textcolor{darkblue}{( \GrTTsym{0} ,  \GrTTsym{0} )}  }  \gamma  )  \to    \mathsf{Type}    
\end{equation*}
The problematic behaviour (that the type parameter $\gamma$ is
returned by the inner function) is exposed by the subject
grade $1$ on the binder of $\gamma$. We can thus define a graded universal quantification
from a graded $\Pi$-typed:
\begin{equation}
\label{eqn:universal-quantifier}
 \forall_{  \textcolor{darkblue}{ r }  }( \gamma  :  A ). B  \triangleq  \textstyle (  \gamma  :_{  \textcolor{darkblue}{( \GrTTsym{0} ,  r )}  }  A  )  \to   B 
\end{equation}
This denotes that the type parameter $\gamma$ can appear freely in
$B$ described by grade $r$, but is irrelevant in the body of any corresponding
$\lambda$-abstraction. This is akin to the work of Nuyts et al. who
develop a system with several modalities for regaining parametricity
within a dependent type
theory~\cite{nuytsParametricQuantifiersDependent2017}.  Note however
that parametricity is recovered for us here as one of many possible
options coming from systematically specialising the grading.

\paragraph{Capturing Existential Types}

With the ability to capture universal
quantifier, we can similarly define existentials
(allowing, e.g., abstraction~\cite{cardelliUnderstandingTypesData1985}). We
define the existential type via a Church-encoding as follows:
\begin{equation*}
 \exists_{  \textcolor{darkblue}{ r }  }( x  :  A ). B  \triangleq  \forall_{  \textcolor{darkblue}{ \GrTTsym{2} }  }( C  :   \mathsf{Type}_{ l }  ).  \textstyle (  f  :_{  \textcolor{darkblue}{( \GrTTsym{1} ,  \GrTTsym{0} )}  }   \forall_{  \textcolor{darkblue}{ r }  }( x  :  A ).  \textstyle (  b  :_{  \textcolor{darkblue}{( s ,  \GrTTsym{0} )}  }  B  )  \to   C    )  \to   C  
\end{equation*}

\paragraph{Embedding into Stratified System F}

\newcommand{\ssf}[1]{\textsc{Ssf}(#1)}
\newcommand{\kind}[1]{\star_{#1}}
\newcommand{\ctxtOk}[1]{#1\ \textit{Ok}}

We show that parametricity is regained here
(and thus eqn.~\eqref{eqn:universal-quantifier} really behaves as a
universal quantifier and not a general $\Pi$-type) by showing
that we can embed a subset of \grtt{} into System F, based solely
on a classification of the grades. We follow a similar approach
to Section~\ref{subsec:embedding-stlc} for simply-typed reasoning
but rather than defining
a purely syntactic encoding (and then proving it type sound)
 our encoding is type directed since
we embed \grtt{} functions of type $ \textstyle (  x  :_{  \textcolor{darkblue}{( \GrTTsym{0} ,  r )}  }   \mathsf{Type}_{ l }   )  \to   B $
as universal types in System F with corresponding type abstractions
($\Lambda$) as their inhabitants. Since \grtt{} employs a predicative
hierarchy of universes, we target Stratified System F (hereafter SSF)
since it includes the analogous inductive hierarchy of
kinds~\cite{leivant1991finitely}.
We use the formulation of Eades and Stump~\cite{eades2010hereditary} with
terms $t_s$ and types $T$:
\begin{align*}
t_s \!::=\! x \mid \lambda (x : T) . t_s \mid t_s\, t_s' \mid \Lambda (X : K) . t_s \mid
  t_s \, [ T ]
\;\;\; & \;\;\,
T \!::=\! X \mid T \rightarrow T' \mid \forall (X : K) . T
\end{align*}
with kinds $K ::= \kind{l}$
where $l \in \mathbb{N}$ providing the stratified kind hierarchy.
Capitalised variables $X$ are System F type variables
and $t_s \, [ T ]$ is type application. Contexts may contain
both type and computational variables, and so free-variable
type assumptions may have dependencies, akin to dependent type systems.
Kinding is via judgments
$\Gamma \vdash T : \kind{l}$
and typing via $\Gamma \vdash t : T$.

We define a type directed encoding on a subset
of \grtt{} typing derivations characterised by
the following predicate:
{\small
\begin{align*}
& \ssf{ (  \emptyset   \mid   \emptyset   \mid   \emptyset  )   \odot   \emptyset   \vdash  t  \GrTTsym{:}  A}
\\
& \ssf{ ( \Delta  \mid  \sigma_{{\mathrm{1}}}  \mid  \sigma_{{\mathrm{2}}} )   \odot  \Gamma  \vdash  t  \GrTTsym{:}  A}
\implies
\ssf{ ( \Delta  ,   \textbf{0}   \mid  \sigma_{{\mathrm{1}}}  ,  \GrTTsym{0}  \mid  \sigma_{{\mathrm{2}}}  ,  r )   \odot  \Gamma  \GrTTsym{,}  x  \GrTTsym{:}   \mathsf{Type}_{ l }   \vdash  t  \GrTTsym{:}  A}
\\
& \ssf{ ( \Delta  \mid  \sigma_{{\mathrm{1}}}  \mid  \sigma_{{\mathrm{2}}} )   \odot  \Gamma  \vdash  t  \GrTTsym{:}  A} \wedge  \mathsf{Type}_{ l }  \not\in^{+ve}
  B \\
& \hspace{12.7em}
\implies
\ssf{ ( \Delta  ,  \sigma_{{\mathrm{3}}}  \mid  \sigma_{{\mathrm{1}}}  ,  s  \mid  \sigma_{{\mathrm{2}}}  ,  \GrTTsym{0} )   \odot  \Gamma  \GrTTsym{,}  x  \GrTTsym{:}  B  \vdash  t  \GrTTsym{:}  A}
\end{align*}}
\hspace{-0.35em}By  $ \mathsf{Type}_{ l }  \not\in^{+ve} B$ we mean
$ \mathsf{Type}_{ l } $ is not a positive subterm of $B$,
avoiding higher-order typing terms (e.g., type constructors)
which do not exist in SSF.

Under this restriction,
we give a type-directed encoding mapping
derivations of \grtt{} to SSF: given
a \grtt{} derivation of judgment
${ ( \Delta  \mid  \sigma_{{\mathrm{1}}}  \mid  \sigma_{{\mathrm{2}}} )   \odot  \Gamma  \vdash  t  \GrTTsym{:}  A}$
we have that $\exists t_s $ (an SSF term)
such that there is a derivation of judgment
$\interp{\Gamma} \vdash t_s : \interp{A}_\tau$ in SSF
where we interpret a subset of \grtt{} terms $A$ as types:
\begin{align*}
\begin{array}{rll}
\interp{x}_\tau & = x \\
\interp{ \mathsf{Type}_{ l } }_\tau & = \kind{l} \\
\interp{ \textstyle (  x  :_{  \textcolor{darkblue}{( \GrTTsym{0} ,  r )}  }   \mathsf{Type}_{ l }   )  \to   B }_\tau
& = \forall x : \kind{l} . \interp{B}_\tau
& \;\;\textit{where}\ \mathsf{Type}_{ l }  \not\in^{+ve} B \\
\interp{ \textstyle (  x  :_{  \textcolor{darkblue}{( s ,  \GrTTsym{0} )}  }  A  )  \to   B }_\tau
& = \interp{A}_\tau \rightarrow \interp{B}_\tau &
\;\;\textit{where}\  \mathsf{Type}_{ l }  \not\in^{+ve} A, B
\end{array}
\end{align*}
Thus, dependent functions with $ \mathsf{Type} $ parameters that
are computationally irrelevant (subject grade $0$)
map to $\forall$ types, and dependent functions with
parameters irrelevant in types (subject-type grade
$0$) map to regular function types. We elide the full details but sketch
key parts where functions and applications are translated
inductively (where $ \mathsf{Ty}_{ l } $ is shorthand for $ \mathsf{Type}_{ l } $):
{\small{
    \begin{align*}
      \begin{array}{rl}
\interp{\dfrac{{ ( \Delta  ,  \sigma_{{\mathrm{1}}}  \mid  \sigma_{{\mathrm{2}}}  ,  \GrTTsym{0}  \mid  \sigma_{{\mathrm{3}}}  ,  r )   \odot  \Gamma  \GrTTsym{,}  x  \GrTTsym{:}   \mathsf{Ty}_{ l }   \vdash  t  \GrTTsym{:}  B}}
{{ ( \Delta  \mid  \sigma_{{\mathrm{2}}}  \mid  \sigma_{{\mathrm{1}}}  \GrTTsym{+}  \sigma_{{\mathrm{3}}} )   \odot  \Gamma  \vdash   \lambda  x . t   \GrTTsym{:}   \textstyle (  x  :_{  \textcolor{darkblue}{( \GrTTsym{0} ,  r )}  }   \mathsf{Ty}_{ l }   )  \to   B }}}
& =
\dfrac{\interp{\Gamma}, x : \kind{l} \vdash t_s : \interp{B}_\tau}
{\interp{\Gamma} \vdash \Lambda (x : \kind{l}) . t_s : \forall x : \kind{l}
  . \interp{B}_\tau} \\[1.5em]
\multicolumn{2}{l}{\interp{
\!
\dfrac
{ ( \Delta  ,  \sigma_{{\mathrm{1}}}  \mid  \sigma_{{\mathrm{2}}}  ,  s  \mid  \sigma_{{\mathrm{3}}}  ,  \GrTTsym{0} )   \odot  \Gamma  \GrTTsym{,}  x  \GrTTsym{:}  A  \vdash  t  \GrTTsym{:}  B}
{ ( \Delta  \mid  \sigma_{{\mathrm{2}}}  \mid  \sigma_{{\mathrm{1}}}  \GrTTsym{+}  \sigma_{{\mathrm{3}}} )   \odot  \Gamma  \vdash   \lambda  x . t   \GrTTsym{:}   \textstyle (  x  :_{  \textcolor{darkblue}{( s ,  \GrTTsym{0} )}  }  A  )  \to   B }
}
\!=\!
\dfrac
 {\interp{\Gamma}, x : \interp{A}_\tau \vdash t_s : \interp{B}_\tau}
 {\interp{\Gamma} \vdash \lambda (x : \interp{A}_\tau) . t_s : \interp{A}_\tau
\rightarrow \interp{B}_\tau}} \\[1.5em]
\interp{
\dfrac{
{\begin{array}{l}
  ( \Delta  \mid  \sigma_{{\mathrm{2}}}  \mid  \sigma_{{\mathrm{1}}}  \GrTTsym{+}  \sigma_{{\mathrm{3}}} )   \odot  \Gamma  \vdash  t_{{\mathrm{1}}}  \GrTTsym{:}   \textstyle (  x  :_{  \textcolor{darkblue}{( \GrTTsym{0} ,  r )}  }   \mathsf{Ty}_{ l }   )  \to   B  \\
   ( \Delta  \mid  \sigma_{{\mathrm{4}}}  \mid  \sigma_{{\mathrm{1}}} )   \odot  \Gamma  \vdash  t_{{\mathrm{2}}}  \GrTTsym{:}   \mathsf{Ty}_{ l } 
\end{array}}}
{  ( \Delta  \mid  \sigma_{{\mathrm{2}}}  \mid  \sigma_{{\mathrm{3}}}  \GrTTsym{+}   r  \ast  \sigma_{{\mathrm{4}}}  )   \odot  \Gamma  \vdash   t_{{\mathrm{1}}} \,{ t_{{\mathrm{2}}} }   \GrTTsym{:}  \GrTTsym{[}  t_{{\mathrm{2}}}  \GrTTsym{/}  x  \GrTTsym{]}  B }
}
& =
\dfrac{
\begin{array}{l}
\interp{\Gamma} \vdash t_s : \forall (x : \kind{l}) . \interp{B}_\tau \\
\interp{\Gamma} \vdash T : \kind{l}
\end{array}
}{
{\interp{\Gamma} \vdash t_s[T] : [T/x]\interp{B}_\tau}
}
\\[1.5em]
\interp{
\dfrac{
{\begin{array}{l}
  ( \Delta  \mid  \sigma_{{\mathrm{2}}}  \mid  \sigma_{{\mathrm{1}}}  \GrTTsym{+}  \sigma_{{\mathrm{3}}} )   \odot  \Gamma  \vdash  t_{{\mathrm{1}}}  \GrTTsym{:}   \textstyle (  x  :_{  \textcolor{darkblue}{( s ,  \GrTTsym{0} )}  }  A  )  \to   B  \\
   ( \Delta  \mid  \sigma_{{\mathrm{4}}}  \mid  \sigma_{{\mathrm{1}}} )   \odot  \Gamma  \vdash  t_{{\mathrm{2}}}  \GrTTsym{:}  A
\end{array}}}
{  ( \Delta  \mid  \sigma_{{\mathrm{2}}}  \GrTTsym{+}   s  \ast  \sigma_{{\mathrm{4}}}   \mid  \sigma_{{\mathrm{3}}} )   \odot  \Gamma  \vdash   t_{{\mathrm{1}}} \,{ t_{{\mathrm{2}}} }   \GrTTsym{:}  \GrTTsym{[}  t_{{\mathrm{2}}}  \GrTTsym{/}  x  \GrTTsym{]}  B }
}
& =
\dfrac{
\begin{array}{l}
\interp{\Gamma} \vdash t_s : \interp{A}_\tau \rightarrow \interp{B}_\tau \\
\interp{\Gamma} \vdash t_s' : \interp{A}_\tau
\end{array}
}
{\interp{\Gamma} \vdash t_s\ t_s' : [t_s'/x]\interp{B}_\tau}
      \end{array}
    \end{align*}
}}
\hspace{-0.35em}In the last case, note the presence of $[t_s'/x]\interp{B}_\tau$.
Reasoning under the context of the encoding, this is proven
equivalent to $\interp{B}_\tau$ since the subject type grade is $0$
and therefore use of $x$ in $B$ is irrelevant.

\begin{theorem}[Soundness and completeness of \textsc{SSF} embedding]
Given $\ssf{ ( \Delta  \mid  \sigma_{{\mathrm{1}}}  \mid  \sigma_{{\mathrm{2}}} )   \odot  \Gamma  \vdash  t  \GrTTsym{:}  A}$
and
$t_a$ in SSF where
$\interp{ ( \Delta  \mid  \sigma_{{\mathrm{1}}}  \mid  \sigma_{{\mathrm{2}}} )   \odot  \Gamma  \vdash  t  \GrTTsym{:}  A} =
\interp{\Gamma} \vdash t_s : \interp{A}_\tau$
then for CBN reduction $\leadsto^{\textsc{Ssf}}$ in
Stratified System F:
\begin{align*}
\begin{array}{rrll}
\textit{(soundness)} & \forall t' . & t \leadsto t' \implies
\exists t_s' . t_s \leadsto^{\textsc{SSF}} t_s' \\
& & \hspace{5em} \wedge \ \interp{ ( \Delta  \mid  \sigma_{{\mathrm{1}}}  \mid  \sigma_{{\mathrm{2}}} )   \odot  \Gamma  \vdash  t'  \GrTTsym{:}  A} = \interp{\Gamma} \vdash t_s' : \interp{A}_\tau   \\
\textit{(completeness)} & \forall t_s'. & t_s \leadsto^{\textsc{Ssf}} t_s'
\implies \exists t' . t \leadsto t'\ \\
& & \hspace{5em} \wedge\ \interp{ ( \Delta  \mid  \sigma_{{\mathrm{1}}}  \mid  \sigma_{{\mathrm{2}}} )   \odot  \Gamma  \vdash  t'  \GrTTsym{:}  A} = \interp{\Gamma} \vdash t_s' : \interp{A}_\tau
\end{array}
\end{align*}
\end{theorem}
\noindent
Thus, we can capture parametricity in \grtt{} via
the judicious use of $0$ grading (at either the type or
computational level) for quantitative
semirings.
This embedding is not possible from \qtt{}
since \qtt{} variables graded with $0$ may be used arbitrarily
in the types; the embedding here relies on
\grtt{}'s $0$ type-grade capturing abscence in types
for quantitative semirings.

\subsection{Graded Modal Types and Non-dependent Linear Types}
\grtt{} can embed the reasoning present in
other graded modal type theories (which often have a linear base), for example
the explicit semiring-graded necessity modality found in
coeffect calculi~\cite{gaboardi2014,gaboardi2016combining} and
Granule~\cite{orchard2019quantitative}.
We can recover the axioms of a graded necessity modality (usually
modelled by an exponential graded
comonad~\cite{gaboardi2016combining}).
For example, in \implName{} the following are well typed:
\begin{gerty}
counit : (a : (.0, .2) Type) -> (z : (.1 , .0) [.1] a) -> a
counit = \a z -> case z of [y] -> y
comult : (a : (.0, .2) Type) -> (z : (.1 , .0) [.6] a) -> [.2] ([.3] a)
comult = \a z -> case z of [y] -> [([y])]
\end{gerty}
corresponding to $\varepsilon : \Box_1 A \rightarrow A$ and
$\delta_{r,s} : \Box_{r * s} A \rightarrow \Box_r (\Box_s A)$:
operations of graded necessity / graded comonads.  Since we cannot use
arbitrary terms for grades in the implementation, we have picked some
particular grades here for \gertyin{comult}.  First-class grading
is future work, discussed in Section~\ref{sec:discussion}.

Linear functions can be captured as $A \multimap B \triangleq
 \textstyle (  x  :_{  \textcolor{darkblue}{( \GrTTsym{1} ,  r )}  }  A  )  \to   B $ for an exact usage
semiring. It is straightforward to characterise a subset of
\grtt{} programs that maps to the linear $\lambda$-calculus
akin to the encodings above.
Thus, \grtt{} provides a suitable basis
for studying both linear and non-linear theories alike.

\section{Metatheory}\label{sec:metatheory}

We now study \grtt{}'s metatheory. We first explain how
substitution presents itself in the theory, and how type preservation
follows from a relationship between equality and reduction. We then
show admissibility of graded structural rules for contraction,
exchange, and weakening, and strong normalization.

\subsection{Substitution}\label{grtt--substitution}
\noindent
We introducing substitution for well-formed contexts and then typing.
\begin{glemma}{substitutionWf}{Substitution for well-formed contexts}
If the following hold:
  \begin{enumerate}
    \item $ ( \Delta  \mid  \sigma_{{\mathrm{2}}}  \mid  \sigma_{{\mathrm{1}}} )   \odot  \Gamma_{{\mathrm{1}}}  \vdash  t  \GrTTsym{:}  A$
    \textit{\hspace{1em}  and \hspace{1em} 2.}\;
    $\GrTTsym{(}  \Delta  ,  \sigma_{{\mathrm{1}}}  ,  \Delta'  \GrTTsym{)}  \odot  \Gamma_{{\mathrm{1}}}  \GrTTsym{,}  x  \GrTTsym{:}  A  \GrTTsym{,}  \Gamma_{{\mathrm{2}}}  \vdash$
  \end{enumerate}
Then: $\Delta  ,  \GrTTsym{(}    \Delta' \backslash  \left|  \Delta  \right|    +    \GrTTsym{(}   \Delta' /  \left|  \Delta  \right|    \GrTTsym{)}  \ast  \sigma_{{\mathrm{2}}}     \GrTTsym{)}  \odot  \Gamma_{{\mathrm{1}}}  \GrTTsym{,}  \GrTTsym{[}  t  \GrTTsym{/}  x  \GrTTsym{]}  \Gamma_{{\mathrm{2}}}  \vdash$
\end{glemma}
\noindent
That is, given $\Gamma_{{\mathrm{1}}}  \GrTTsym{,}  x  \GrTTsym{:}  A  \GrTTsym{,}  \Gamma_{{\mathrm{2}}}$ is well-formed, we can
cut out $x$ by substituting $t$ for $x$ in $\Gamma_{{\mathrm{2}}}$,
accounting for the new usage in the context grade vectors.
The usage of $\Gamma_{{\mathrm{1}}}$ in $t$ is given by
$\sigma_{{\mathrm{2}}}$, and the usage in $A$ by $\sigma_{{\mathrm{1}}}$. When substituting,
 $\Delta$ remains the same, as $\Gamma_{{\mathrm{1}}}$ is
unchanged. However, to account for the usage in
$\GrTTsym{[}  t  \GrTTsym{/}  x  \GrTTsym{]}  \Gamma_{{\mathrm{2}}}$, we have to form a new context grade vector
$  \Delta' \backslash  \left|  \Delta  \right|    +    \GrTTsym{(}   \Delta' /  \left|  \Delta  \right|    \GrTTsym{)}  \ast  \sigma_{{\mathrm{2}}}   $.

The operation $ \Delta' \backslash  \left|  \Delta  \right|  $ (pronounced `discard')
removes grades corresponding to $x$, by removing the grade
at index $ \left|  \Delta  \right| $ from each grade vector in
$\Delta'$. Everything previously used in the typing of $x$
in the context must now be distributed across $\GrTTsym{[}  t  \GrTTsym{/}  x  \GrTTsym{]}  \Gamma_{{\mathrm{2}}}$,
which is done by adding on $ \GrTTsym{(}   \Delta' /  \left|  \Delta  \right|    \GrTTsym{)}  \ast  \sigma_{{\mathrm{2}}} $, which
uses $ \Delta' /  \left|  \Delta  \right|  $ (pronounced `choose') to produce a
vector of grades, which correspond to the grades cut out in
$ \Delta' \backslash  \left|  \Delta  \right|  $. The multiplication of $ \GrTTsym{(}   \Delta' /  \left|  \Delta  \right|    \GrTTsym{)}  \ast  \sigma_{{\mathrm{2}}} $
produces a context grade vector by scaling
$\sigma_{{\mathrm{2}}}$ by each element of $\GrTTsym{(}   \Delta' /  \left|  \Delta  \right|    \GrTTsym{)}$. When adding
vectors, if the sizes of the vectors are different, then the
shorter vector is right-padded with zeroes. Thus
$  \Delta' \backslash  \left|  \Delta  \right|    +    \GrTTsym{(}   \Delta' /  \left|  \Delta  \right|    \GrTTsym{)}  \ast  \sigma_{{\mathrm{2}}}   $ can be read
as \emph{`$\Delta'$ without the grades corresponding to
  $x$, plus the usage of $t$ scaled by the prior usage
  of $x$'}.

For example, given typing
$ (   \GrTTsym{(}  \GrTTsym{)}  ,  \GrTTsym{(}  \GrTTsym{1}  \GrTTsym{)}    \mid    \GrTTsym{0}  ,  \GrTTsym{1}    \mid    \GrTTsym{1}  ,  \GrTTsym{0}   )   \odot  a  \GrTTsym{:}   \mathsf{Type}   \GrTTsym{,}  y  \GrTTsym{:}  a  \vdash  y  \GrTTsym{:}  a$ and well-formed
context $\GrTTsym{(}  \GrTTsym{(}  \GrTTsym{)}  ,  \GrTTsym{(}  \GrTTsym{1}  \GrTTsym{)}  ,  \GrTTsym{(}  \GrTTsym{1}  ,  \GrTTsym{0}  \GrTTsym{)}  ,  \GrTTsym{(}  \GrTTsym{0}  ,  \GrTTsym{0}  ,  \GrTTsym{2}  \GrTTsym{)}  \GrTTsym{)}  \odot  a  \GrTTsym{:}   \mathsf{Type}   \GrTTsym{,}  y  \GrTTsym{:}  a  \GrTTsym{,}  x  \GrTTsym{:}  a  \GrTTsym{,}  z  \GrTTsym{:}  t'  \vdash$, where $t'$ uses $x$ twice, we can substitute $y$ for
$x$. Therefore, let $\Gamma_{{\mathrm{1}}} = a  \GrTTsym{:}   \mathsf{Type}   \GrTTsym{,}  y  \GrTTsym{:}  a$
thus $ \left|  \Gamma_{{\mathrm{1}}}  \right|  = 2$ and $\Gamma_{{\mathrm{2}}} = z  \GrTTsym{:}  x$ and $\Delta' = \GrTTsym{(}  \GrTTsym{(}  \GrTTsym{0}  ,  \GrTTsym{0}  ,  \GrTTsym{2}  \GrTTsym{)}  \GrTTsym{)}$ and $\sigma_{{\mathrm{1}}} =   \GrTTsym{1}  ,  \GrTTsym{0}  $ and $\sigma_{{\mathrm{2}}} =   \GrTTsym{0}  ,  \GrTTsym{1}  $.
Then the context grade of the substitution $\GrTTsym{[}  y  \GrTTsym{/}  x  \GrTTsym{]}  \Gamma_{{\mathrm{2}}}$ is calculated as:
\begin{align*}
 \GrTTsym{(}  \GrTTsym{(}  \GrTTsym{0}  ,  \GrTTsym{0}  ,  \GrTTsym{2}  \GrTTsym{)}  \GrTTsym{)} \backslash  \left|  \Gamma_{{\mathrm{1}}}  \right|   = \GrTTsym{(}  \GrTTsym{(}  \GrTTsym{0}  ,  \GrTTsym{0}  \GrTTsym{)}  \GrTTsym{)} \qquad\quad
 \GrTTsym{(}   \GrTTsym{(}  \GrTTsym{(}  \GrTTsym{0}  ,  \GrTTsym{1}  ,  \GrTTsym{2}  \GrTTsym{)}  \GrTTsym{)} /  \left|  \Gamma_{{\mathrm{1}}}  \right|    \GrTTsym{)}  \ast  \sigma_{{\mathrm{2}}}  =  \GrTTsym{(}  \GrTTsym{2}  \GrTTsym{)}  \ast  \GrTTsym{(}  \GrTTsym{0}  ,  \GrTTsym{1}  \GrTTsym{)}  = \GrTTsym{(}  \GrTTsym{(}  \GrTTsym{0}  ,  \GrTTsym{2}  \GrTTsym{)}  \GrTTsym{)}
\end{align*}
Thus the resulting judgment is $\GrTTsym{(}  \GrTTsym{(}  \GrTTsym{)}  ,  \GrTTsym{(}  \GrTTsym{1}  \GrTTsym{)}  ,  \GrTTsym{(}  \GrTTsym{0}  ,  \GrTTsym{2}  \GrTTsym{)}  \GrTTsym{)}  \odot  a  \GrTTsym{:}   \mathsf{Type}   \GrTTsym{,}  y  \GrTTsym{:}  a  \GrTTsym{,}  z  \GrTTsym{:}  \GrTTsym{[}  y  \GrTTsym{/}  x  \GrTTsym{]}  t'  \vdash$.
\begin{glemma}{substitutionForTyping}{Substitution for typing}
  If the following premises hold:
  \begin{enumerate}
    \item $ ( \Delta  \mid  \sigma_{{\mathrm{2}}}  \mid  \sigma_{{\mathrm{1}}} )   \odot  \Gamma_{{\mathrm{1}}}  \vdash  t  \GrTTsym{:}  A$
    \item $ ( \Delta  ,  \sigma_{{\mathrm{1}}}  ,  \Delta'  \mid  \sigma_{{\mathrm{3}}}  ,  s  ,  \sigma_{{\mathrm{4}}}  \mid  \sigma_{{\mathrm{5}}}  ,  r  ,  \sigma_{{\mathrm{6}}} )   \odot  \Gamma_{{\mathrm{1}}}  \GrTTsym{,}  x  \GrTTsym{:}  A  \GrTTsym{,}  \Gamma_{{\mathrm{2}}}  \vdash  t'  \GrTTsym{:}  B$
    \item $ \left|  \sigma_{{\mathrm{3}}}  \right|   \GrTTsym{=}   \left|  \sigma_{{\mathrm{5}}}  \right|   \GrTTsym{=}   \left|  \Gamma_{{\mathrm{1}}}  \right| $
  \end{enumerate}
  Then
  $ \left(\begin{smallmatrix} \Delta  ,  \GrTTsym{(}    \Delta' \backslash  \left|  \Delta  \right|    +    \GrTTsym{(}   \Delta' /  \left|  \Delta  \right|    \GrTTsym{)}  \ast  \sigma_{{\mathrm{2}}}     \GrTTsym{)}  \\  \GrTTsym{(}  \sigma_{{\mathrm{3}}}  \GrTTsym{+}   s  \ast  \sigma_{{\mathrm{2}}}   \GrTTsym{)}  ,  \sigma_{{\mathrm{4}}}  \\  \GrTTsym{(}  \sigma_{{\mathrm{5}}}  \GrTTsym{+}   r  \ast  \sigma_{{\mathrm{2}}}   \GrTTsym{)}  ,  \sigma_{{\mathrm{6}}}  \end{smallmatrix}\right)   \odot  \Gamma_{{\mathrm{1}}}  \GrTTsym{,}  \GrTTsym{[}  t  \GrTTsym{/}  x  \GrTTsym{]}  \Gamma_{{\mathrm{2}}}  \vdash  \GrTTsym{[}  t  \GrTTsym{/}  x  \GrTTsym{]}  t'  \GrTTsym{:}  \GrTTsym{[}  t  \GrTTsym{/}  x  \GrTTsym{]}  B$.
\end{glemma}
\noindent
As with substitution for well-formed contexts, we account for
the replacement of $x$ with $t$ in $\Gamma_{{\mathrm{2}}}$ by `cutting
out' $x$ from the context grade vectors, and adding on the
grades required to form $t$, scaled by the grades that
described $x$'s usage.  We additionally must account for the
altered subject and subject-type usage. We do this in a
similar manner, by taking, for example, the usage of $\Gamma_{{\mathrm{1}}}$
in the subject ($\sigma_{{\mathrm{3}}}$), and adding on the grades required
to form $t$, scaled by the grade with which $x$ was
previously used ($s$). Subject-type grades are calculated
similarly.

\subsection{Type Preservation}\label{grtt--type-preservation}

\begin{lemma}{Reduction implies equality}
\label{lemma:redImplEq}
  If $ ( \Delta  \mid  \sigma_{{\mathrm{1}}}  \mid  \sigma_{{\mathrm{2}}} )   \odot  \Gamma  \vdash  t_{{\mathrm{1}}}  \GrTTsym{:}  A$ and $ t_{{\mathrm{1}}}  \leadsto  t_{{\mathrm{2}}} $, then
  $ ( \Delta  \mid  \sigma_{{\mathrm{1}}}  \mid  \sigma_{{\mathrm{2}}} )   \odot  \Gamma  \vdash  t_{{\mathrm{1}}}  \GrTTsym{=}  t_{{\mathrm{2}}}  \GrTTsym{:}  A$.
\end{lemma}

\begin{lemma}{Equality inversion}
\label{lemma:eqInv}
  If $ ( \Delta  \mid  \sigma_{{\mathrm{1}}}  \mid  \sigma_{{\mathrm{2}}} )   \odot  \Gamma  \vdash  t_{{\mathrm{1}}}  \GrTTsym{=}  t_{{\mathrm{2}}}  \GrTTsym{:}  A$, then
  $ ( \Delta  \mid  \sigma_{{\mathrm{1}}}  \mid  \sigma_{{\mathrm{2}}} )   \odot  \Gamma  \vdash  t_{{\mathrm{1}}}  \GrTTsym{:}  A$ and
  $ ( \Delta  \mid  \sigma_{{\mathrm{1}}}  \mid  \sigma_{{\mathrm{2}}} )   \odot  \Gamma  \vdash  t_{{\mathrm{2}}}  \GrTTsym{:}  A$.
\end{lemma}

\begin{lemma}{Type preservation}
\label{lemma:typePres}
  If $ ( \Delta  \mid  \sigma_{{\mathrm{1}}}  \mid  \sigma_{{\mathrm{2}}} )   \odot  \Gamma  \vdash  t  \GrTTsym{:}  A$ and $ t  \leadsto  t' $, then
  $ ( \Delta  \mid  \sigma_{{\mathrm{1}}}  \mid  \sigma_{{\mathrm{2}}} )   \odot  \Gamma  \vdash  t'  \GrTTsym{:}  A$.
\end{lemma}

\begin{proof}
  By Lemma~\ref{lemma:redImplEq} we have
  $ ( \Delta  \mid  \sigma_{{\mathrm{1}}}  \mid  \sigma_{{\mathrm{2}}} )   \odot  \Gamma  \vdash  t  \GrTTsym{=}  t'  \GrTTsym{:}  A$, and therefore by
  Lemma~\ref{lemma:eqInv} we have $ ( \Delta  \mid  \sigma_{{\mathrm{1}}}  \mid  \sigma_{{\mathrm{2}}} )   \odot  \Gamma  \vdash  t'  \GrTTsym{:}  A$, as
  required.
\end{proof}

\subsection{Structural Rules}\label{grtt--structural-rules}

We now consider the structural rules of \emph{contraction},
\emph{exchange}, and \emph{weakening}.

\begin{glemma}{contractionAdmissible}{Contraction}
  The following rule is admissible:
  \begin{align*}
    \inferrule*[right=\textsc{Contr}]{
       { \left(\begin{smallmatrix} \Delta_{{\mathrm{1}}}  ,  \sigma_{{\mathrm{1}}}  ,  \GrTTsym{(}  \sigma_{{\mathrm{1}}}  ,  \GrTTsym{0}  \GrTTsym{)}  ,  \Delta_{{\mathrm{2}}}  \\  \sigma_{{\mathrm{2}}}  ,  s_{{\mathrm{1}}}  ,  s_{{\mathrm{2}}}  ,  \sigma_{{\mathrm{3}}}  \\  \sigma_{{\mathrm{4}}}  ,  r_{{\mathrm{1}}}  ,  r_{{\mathrm{2}}}  ,  \sigma_{{\mathrm{5}}}  \end{smallmatrix}\right)   \odot  \Gamma_{{\mathrm{1}}}  \GrTTsym{,}  x  \GrTTsym{:}  A  \GrTTsym{,}  y  \GrTTsym{:}  A  \GrTTsym{,}  \Gamma_{{\mathrm{2}}}  \vdash  t  \GrTTsym{:}  B}
       \quad
       { \left|  \Delta_{{\mathrm{1}}}  \right|   \GrTTsym{=}   \left|  \sigma_{{\mathrm{2}}}  \right|   \GrTTsym{=}   \left|  \sigma_{{\mathrm{4}}}  \right|   \GrTTsym{=}   \left|  \Gamma_{{\mathrm{1}}}  \right| }
    }{ { \left(\begin{smallmatrix} \Delta_{{\mathrm{1}}}  ,  \sigma_{{\mathrm{1}}}  ,   \mathsf{contr}(  \left|  \Delta_{{\mathrm{1}}}  \right|  ;  \Delta_{{\mathrm{2}}} )   \\  \sigma_{{\mathrm{2}}}  ,  \GrTTsym{(}  s_{{\mathrm{1}}}  \GrTTsym{+}  s_{{\mathrm{2}}}  \GrTTsym{)}  ,  \sigma_{{\mathrm{3}}}  \\  \sigma_{{\mathrm{4}}}  ,  \GrTTsym{(}  r_{{\mathrm{1}}}  \GrTTsym{+}  r_{{\mathrm{2}}}  \GrTTsym{)}  ,  \sigma_{{\mathrm{5}}}  \end{smallmatrix}\right)   \odot  \Gamma_{{\mathrm{1}}}  \GrTTsym{,}  z  \GrTTsym{:}  A  \GrTTsym{,}  \GrTTsym{[}  z  \GrTTsym{,}  z  \GrTTsym{/}  x  \GrTTsym{,}  y  \GrTTsym{]}  \Gamma_{{\mathrm{2}}}  \vdash  \GrTTsym{[}  z  \GrTTsym{,}  z  \GrTTsym{/}  x  \GrTTsym{,}  y  \GrTTsym{]}  t  \GrTTsym{:}  \GrTTsym{[}  z  \GrTTsym{,}  z  \GrTTsym{/}  x  \GrTTsym{,}  y  \GrTTsym{]}  B} }
  \end{align*}
\end{glemma}

\noindent
The operation $ \mathsf{contr}( \pi ;  \Delta ) $ contracts the elements at
index $\pi$ and $\pi  \GrTTsym{+}  \GrTTsym{1}$ for each vector in $\Delta$
by combining them with the semiring addition, defined
$ \mathsf{contr}( \pi ;  \Delta )   \GrTTsym{=}     \Delta \backslash \GrTTsym{(}  \pi  \GrTTsym{+}  \GrTTsym{1}  \GrTTsym{)}    +     \Delta / \GrTTsym{(}  \pi  \GrTTsym{+}  \GrTTsym{1}  \GrTTsym{)}    \ast  \GrTTsym{(}   \textbf{0}^{ \pi }   ,  \GrTTsym{1}  \GrTTsym{)}  $. Admissibility follows from the semiring
addition, which serves to contract dependencies, being
threaded throughout the rules.

\begin{glemma}{exchangeAdmissible}{Exchange}
  The following rule is admissible:\vspace{-5pt}
  \begin{align*}
    \inferrule*[flushleft,right=\textsc{Exc}] {
      {x \, \not\in \, \mathsf{FV} \, \GrTTsym{(}  B  \GrTTsym{)}}\\\\
      { \left|  \Delta_{{\mathrm{1}}}  \right|   \GrTTsym{=}   \left|  \sigma_{{\mathrm{3}}}  \right|   \GrTTsym{=}   \left|  \sigma_{{\mathrm{5}}}  \right|   \GrTTsym{=}   \left|  \Gamma_{{\mathrm{1}}}  \right| }\\
      { \left(\begin{smallmatrix} \Delta_{{\mathrm{1}}}  ,  \sigma_{{\mathrm{1}}}  ,  \GrTTsym{(}  \sigma_{{\mathrm{2}}}  ,  \GrTTsym{0}  \GrTTsym{)}  ,  \Delta_{{\mathrm{2}}}  \\  \sigma_{{\mathrm{3}}}  ,  s_{{\mathrm{1}}}  ,  s_{{\mathrm{2}}}  ,  \sigma_{{\mathrm{4}}}  \\  \sigma_{{\mathrm{5}}}  ,  r_{{\mathrm{1}}}  ,  r_{{\mathrm{2}}}  ,  \sigma_{{\mathrm{6}}}  \end{smallmatrix}\right)   \odot  \Gamma_{{\mathrm{1}}}  \GrTTsym{,}  x  \GrTTsym{:}  A  \GrTTsym{,}  y  \GrTTsym{:}  B  \GrTTsym{,}  \Gamma_{{\mathrm{2}}}  \vdash  t  \GrTTsym{:}  C}
    }{{ \left(\begin{smallmatrix} \Delta_{{\mathrm{1}}}  ,  \sigma_{{\mathrm{2}}}  ,  \GrTTsym{(}  \sigma_{{\mathrm{1}}}  ,  \GrTTsym{0}  \GrTTsym{)}  ,   \mathsf{exch}(  \left|  \Delta_{{\mathrm{1}}}  \right|  ;  \Delta_{{\mathrm{2}}} )   \\  \sigma_{{\mathrm{3}}}  ,  s_{{\mathrm{2}}}  ,  s_{{\mathrm{1}}}  ,  \sigma_{{\mathrm{4}}}  \\  \sigma_{{\mathrm{5}}}  ,  r_{{\mathrm{2}}}  ,  r_{{\mathrm{1}}}  ,  \sigma_{{\mathrm{6}}}  \end{smallmatrix}\right)   \odot  \Gamma_{{\mathrm{1}}}  \GrTTsym{,}  y  \GrTTsym{:}  B  \GrTTsym{,}  x  \GrTTsym{:}  A  \GrTTsym{,}  \Gamma_{{\mathrm{2}}}  \vdash  t  \GrTTsym{:}  C} }
  \end{align*}
\end{glemma}

\noindent
Notice that if you strip away the vector fragment and sizing
premise, this is exactly the form of exchange we would expect
in a dependent type theory: if $x$ and $y$ are assumptions in
a context typing $t : C$, and the type of $y$ does not depend
upon $x$, then we can type $t : C$ when we swap the order of
$x$ and $y$.

The action on grade vectors is simple: we swap the grades
associated with each of the variables. For the context grade
vector however, we must do two things: first, we capture the
formation of $A$ with $\sigma_{{\mathrm{1}}}$, and the formation of $B$ with
$\sigma_{{\mathrm{1}}}  ,  \GrTTsym{0}$ (indicating $x$ being used with grade 0 in $B$),
then swap these around, cutting the final grade from
$\sigma_{{\mathrm{2}}}  ,  \GrTTsym{0}$, and adding 0 to the end of $\sigma_{{\mathrm{1}}}$ to ensure
correct sizing. Next, the operation
$ \mathsf{exch}(  \left|  \Delta_{{\mathrm{1}}}  \right|  ;  \Delta_{{\mathrm{2}}} ) $ swaps the element at index
$ \left|  \Delta_{{\mathrm{1}}}  \right| $ (i.e., that corresponding to usage of $x$) with the
element at index $ \left|  \Delta_{{\mathrm{1}}}  \right|  + 1$ (corresponding to $y$) for
every vector in $\Delta_{{\mathrm{2}}}$; this exchange operation ensures that
usage in the trailing context is reordered appropriately.

\begin{glemma}{weakAdmissible}{Weakening}
  The following rule is admissible:
  \begin{align*}
    \inferrule*[right=\textsc{Weak}]{
    \begin{array}{ll}
      { ( \Delta_{{\mathrm{1}}}  ,  \Delta_{{\mathrm{2}}}  \mid  \sigma_{{\mathrm{1}}}  ,  \sigma'_{{\mathrm{1}}}  \mid  \sigma_{{\mathrm{2}}}  ,  \sigma'_{{\mathrm{2}}} )   \odot  \Gamma_{{\mathrm{1}}}  \GrTTsym{,}  \Gamma_{{\mathrm{2}}}  \vdash  t  \GrTTsym{:}  B} \cr
       ( \Delta_{{\mathrm{1}}}  \mid  \sigma_{{\mathrm{3}}}  \mid   \textbf{0}  )   \odot  \Gamma_{{\mathrm{1}}}  \vdash  A  \GrTTsym{:}   \mathsf{Type}_{ l }  &
        \qquad  \left|  \sigma_{{\mathrm{1}}}  \right|   \GrTTsym{=}   \left|  \sigma_{{\mathrm{2}}}  \right|   \GrTTsym{=}   \left|  \Gamma_{{\mathrm{1}}}  \right| 
    \end{array}
    }{ { (  \Delta_{{\mathrm{1}}}  ,  \sigma_{{\mathrm{3}}}  ,   \mathsf{ins}(  \left|  \Delta_{{\mathrm{1}}}  \right|  ;  \GrTTsym{0} ;  \Delta_{{\mathrm{2}}} )    \mid  \sigma_{{\mathrm{1}}}  ,  \GrTTsym{0}  ,  \sigma'_{{\mathrm{1}}}  \mid  \sigma_{{\mathrm{2}}}  ,  \GrTTsym{0}  ,  \sigma'_{{\mathrm{2}}} )   \odot  \Gamma_{{\mathrm{1}}}  \GrTTsym{,}  x  \GrTTsym{:}  A  \GrTTsym{,}  \Gamma_{{\mathrm{2}}}  \vdash  t  \GrTTsym{:}  B} }
  \end{align*}
\end{glemma}
\noindent
Weakening introduces irrelevant assumptions to a context. We
do this by capturing the usage in the formation of the
assumption's type with $\sigma_{{\mathrm{3}}}$ to preserve the
well-formedness of the context. We then indicate irrelevance
of the assumption by grading with $0$ in appropriate
places. The operation $ \mathsf{ins}( \pi ;  s ;  \Delta ) $ inserts the element
$s$ at index $\pi$ for each $\sigma$ in $\Delta$, such
that all elements preceding index $\pi$ (in $\sigma$) keep
their positions, and every element at index $\pi$ or
greater (in $\sigma$) will be shifted one index later in the
new vector. The $0$ grades in the subject and subject-type
grade vector positions correspond to the absence of the
irrelevant assumption from the subject and subject's type.

\newcommand{\grttZO}[0]{\grtt{}^{\{0,1\}}}
\subsection{Strong Normalization}
\label{subsec:strong_normalization}

We adapt Geuvers' strong normalization proof for the
Calculus of Constructions (CC)~\cite{Geuvers:1995} to a fragment
of
\grtt{} (called $\grtt{}^{\{0,1\}}$) restricted to two
universe levels and without variables of type $ \mathsf{Type}_{ \GrTTsym{1} } $.
This results in a less expressive system
than full \grtt{} when it comes to higher kinds, but
this is orthogonal to the main idea here of grading.
We briefly overview the strong normalization proof; details can be
found in Appendix~\ref{appendix--supplement}.
Note this strong normalization result is with respect to $\beta$-reduction only (our semantics does not include $\eta$-reduction).

We use the proof technique of saturated sets, based on the
reducibility candidates of Girard~\cite{girard1989proofs}.  While $\grttZO{}$
has a collapsed syntax we use judgments to break typing up into
stages.  We use these sets to match on whether a
term is a kind, type, constructor, or a function (we will refer to
these as terms).
\begin{definition}
  \label{def:type-cats}
  Typing can be broken up into the following stages:
  \[ \small
  \setlength\arraycolsep{2pt}
  \begin{array}{lll}
    \mathsf{Kind}  & := & \{A \mid \exists \Delta,\sigma_{{\mathrm{1}}},\Gamma. ( \Delta  \mid  \sigma_{{\mathrm{1}}}  \mid   \textbf{0}  )   \odot  \Gamma  \vdash  A  \GrTTsym{:}   \mathsf{Type}_{ \GrTTsym{1} } \}\\[1pt]
    \mathsf{Type}  & := & \{A \mid \exists \Delta,\sigma_{{\mathrm{1}}},\Gamma. ( \Delta  \mid  \sigma_{{\mathrm{1}}}  \mid   \textbf{0}  )   \odot  \Gamma  \vdash  A  \GrTTsym{:}   \mathsf{Type}_{ \GrTTsym{0} } \}\\[1pt]
     \mathsf{Con}  & := & \{t \mid \exists \Delta,\sigma_{{\mathrm{1}}},\sigma_{{\mathrm{2}}},\Gamma,A. ( \Delta  \mid  \sigma_{{\mathrm{1}}}  \mid  \sigma_{{\mathrm{2}}} )   \odot  \Gamma  \vdash  t  \GrTTsym{:}  A \,\land  ( \Delta  \mid  \sigma_{{\mathrm{2}}}  \mid   \textbf{0}  )   \odot  \Gamma  \vdash  A  \GrTTsym{:}   \mathsf{Type}_{ \GrTTsym{1} } \}\\[1pt]
     \mathsf{Term}   & := & \{t \mid \exists \Delta,\sigma_{{\mathrm{1}}},\sigma_{{\mathrm{2}}},\Gamma,A. ( \Delta  \mid  \sigma_{{\mathrm{1}}}  \mid  \sigma_{{\mathrm{2}}} )   \odot  \Gamma  \vdash  t  \GrTTsym{:}  A \,\land  ( \Delta  \mid  \sigma_{{\mathrm{2}}}  \mid   \textbf{0}  )   \odot  \Gamma  \vdash  A  \GrTTsym{:}   \mathsf{Type}_{ \GrTTsym{0} } \}
  \end{array}
  \]
\end{definition}

\begin{lemma}[Classification]\label{lemma:classification}
  We have $\mathsf{Kind} \cap \mathsf{Type} = \emptyset$ and $ \mathsf{Con}  \cap  \mathsf{Term}  = \emptyset$.
\end{lemma}
\noindent
The classification lemma states that we can safely case split over
kinds and types, or constructors and terms without fear of an overlap
occurring.

Saturated sets are essentially collections of strongly normalizing
terms that are closed under $\beta$-reduction.  The intuition
behind this proof is that every typable program ends up in some
saturated set, and hence, is strongly normalizing.
\begin{definition}[Base terms and saturated terms]
  \label{def:sat-set}
Informally, the set of base terms $ \mathcal{B} $ is inductively defined from variables
and $ \mathsf{Type}_{ \GrTTsym{0} } $ and $ \mathsf{Type}_{ \GrTTsym{1} } $, and compound terms over base $ \mathcal{B} $ and strongly normalising terms $\mathsf{SN}$.

  A set of terms $X$ is \emph{saturated} if
  $X \subset \mathsf{SN}$,
  $ \mathcal{B}  \subset X$, and
  if $ \mathsf{red_k}\, t  \in X$ and $t \, \in \, \mathsf{SN}$, then $t \in X$.
Thus saturated sets
are closed under strongly normalizing terms with a \emph{key redex},
denoted $ \mathsf{red_k}\, t $, which are redexes or a redex at the head of an
elimination form.
$\mathsf{SAT}$ denotes the collection of saturated sets.
\end{definition}

\begin{lemma}[$\mathsf{SN}$ saturated]\label{lemma:sn_is_saturated}
  All saturated sets are non-empty; $\mathsf{SN}$ is saturated.
\end{lemma}

\noindent
Since $\grttZO{}$ allows computation in types as well as in
types, we separate the interpretations for kinds and types,
where the former is a set of the latter.
\begin{definition}
  \label{def:kind-interp}
  For $A \, \in \, \mathsf{Kind}$, the kind interpretation, $ \mathcal{K}\interp{ A } $, is defined:
  \[\small
  \begin{array}{llll}
     \mathcal{K}\interp{  \mathsf{Type}_{ \GrTTsym{0} }  }  & = \mathsf{SAT} \;\;\; &
     \mathcal{K}\interp{  \textstyle (  x  :_{  \textcolor{darkblue}{( s ,  r )}  }  A  )  \to   B  }  & = \{f \mid f :  \mathcal{K}\interp{ A }  \to  \mathcal{K}\interp{ B }  \}\text{, if }A  \GrTTsym{,}  B \, \in \, \mathsf{Kind}\\
     \mathcal{K}\interp{  \square_{  \textcolor{darkblue}{ s }  }  A  }  & =  \mathcal{K}\interp{ A }  \;\;\; &
     \mathcal{K}\interp{  \textstyle (  x  :_{  \textcolor{darkblue}{( s ,  r )}  }  A  )  \to   B  }  & =  \mathcal{K}\interp{ A } \text{, if }A \, \in \, \mathsf{Kind}, B \, \in \, \mathsf{Type}\\
    & &  \mathcal{K}\interp{  \textstyle (  x  :_{  \textcolor{darkblue}{( s ,  r )}  }  A  )  \to   B  }  & =  \mathcal{K}\interp{ B } \text{, if }A \, \in \, \mathsf{Type},B \, \in \, \mathsf{Kind}\\
    & &  \mathcal{K}\interp{  \textstyle (  x  :_{  \textcolor{darkblue}{ r }  }  A  )  \otimes   B  }  & =  \mathcal{K}\interp{ A }  \times  \mathcal{K}\interp{ B } \text{, if }A  \GrTTsym{,}  B \, \in \, \mathsf{Kind}\\
    & &  \mathcal{K}\interp{  \textstyle (  x  :_{  \textcolor{darkblue}{ r }  }  A  )  \otimes   B  }  & =  \mathcal{K}\interp{ A } \text{, if }A \, \in \, \mathsf{Kind}, B \, \in \, \mathsf{Type}\\
    & &  \mathcal{K}\interp{  \textstyle (  x  :_{  \textcolor{darkblue}{ r }  }  A  )  \otimes   B  }  & =  \mathcal{K}\interp{ B } \text{, if }A \, \in \, \mathsf{Type}, B \, \in \, \mathsf{Kind}\\
  \end{array}
  \]
\end{definition}

\noindent
Next we define the interpretation of types, which requires the
interpretation to be parametric on an interpretation of type variables
called a type evaluation. This is necessary to make the
interpretation well-founded (first realized by
Girard~\cite{girard1989proofs}).
\begin{definition}
  \label{def:valid-type-env}
  \emph{Type valuations}, $ \Delta  \odot  \Gamma  \models  \varepsilon $, are defined as follows:
  \footnotesize
  \begin{align*}
    \begin{array}{c}
    \GrTTdruleEpXXEmpty{} \quad
    \GrTTdruleEpXXExtTy{} \quad
    \GrTTdruleEpXXExtTm{}
    \end{array}
  \end{align*}
\end{definition}
\noindent
Type valuations ignore term variables
(rule~\GrTTdruleEpXXExtTmName{}), in
fact, the interpretations of both types and kinds
ignores them because we are defining sets of terms over types, and
thus terms in types do not contribute to the definition of these
sets. However as these interpretations define sets of open
terms we must carry a graded context around where
necessary. Thus, type valuations are with respect to a
well-formed graded context $\Delta  \odot  \Gamma$. We now outline the
type interpretation.
\begin{definition}
  \label{def:interpretation-of-types}
  For type valuation $ \Delta  \odot  \Gamma  \models  \varepsilon $ and a type $A
  \in ( \mathsf{Kind}  \cup  \mathsf{Type}  \cup \mathsf{Con})$ with $A$ typable in
  $\Delta  \odot  \Gamma$, the interpretation of types $ \interp{ A }_{ \varepsilon } $ is defined inductively. For brevity, we list just a few
  illustrative cases, including modalities and some function
  cases; the complete
  definition is given in Appendix~\ref{appendix--supplement}.
  \[\footnotesize
    \setlength\arraycolsep{1pt}
  \begin{array}{rllr}
     \interp{  \mathsf{Type}_{ \GrTTsym{1} }  }_{ \varepsilon }  & = & \mathsf{SN}\\

     \interp{  \mathsf{Type}_{ \GrTTsym{0} }  }_{ \varepsilon }  & = &  \lambda  X  \in  \mathsf{SAT} . \mathsf{SN} \\


     \interp{ x }_{ \varepsilon }           & = &  \varepsilon \, x  & \text{if }x \, \in \, \mathsf{Con}\\


     \interp{  \square_{  \textcolor{darkblue}{ s }  }  A  }_{ \varepsilon }  & = &  \interp{ A }_{ \varepsilon } \\


     \interp{  \lambda  x  \GrTTsym{:}  A . B  }_{ \varepsilon }      & = &  \lambda  X  \in   \mathcal{K}\interp{ A }  .  \interp{ B }_{  \varepsilon [  x  \mapsto  X  ]  }   \qquad & \text{if }A \, \in \, \mathsf{Kind}, B \, \in \, \mathsf{Con}\\

     \interp{  A \,{ B }  }_{ \varepsilon }   & = &  \interp{ A }_{ \varepsilon } ( \interp{ B }_{ \varepsilon } ) & \text{if }B \, \in \, \mathsf{Con}\\

 \interp{  \textstyle (  x  :_{  \textcolor{darkblue}{( s ,  r )}  }  A  )  \to   B  }_{ \varepsilon }  & = & \multicolumn{2}{l}{\lambda X \in \mathcal{K}\interp{A} \to \mathcal{K}\interp{B} .
        \bigcap_{ Y  \in   \mathcal{K}\interp{ A }  }  \GrTTsym{(}   \interp{ A }_{ \varepsilon }  \, Y  \to   \interp{ B }_{  \varepsilon [  x  \mapsto  Y  ]  }  \, \GrTTsym{(}  X \, \GrTTsym{(}  Y  \GrTTsym{)}  \GrTTsym{)}  \GrTTsym{)} } \\
                         &&& \text{if }A  \GrTTsym{,}  B \, \in \, \mathsf{Kind}
  \end{array}
  \]
\end{definition}
\noindent
Grades play no role in the reduction relation for \grtt, and hence,
our interpretation erases graded modalities and their
introductory and elimination forms (translated into substitutions).  In fact, the
above interpretation can be seen as a translation of $\grttZO{}$ into
non-substructural set theory; there is no data-usage
tracking in the image of the interpretation.  Tensors are
translated into Cartesian products whose eliminators are translated
into substitutions similarly to graded modalities.  All terms however
remain well-typed through the interpretation.

The interpretation of terms corresponds to term valuations that are used to close the term before interpreting it into the interpretation of its type.
\begin{definition}
  \label{def:valid-term-env}
  \emph{Valid term valuations}, $ \Delta  \odot  \Gamma  \models_{ \varepsilon }  \rho $, are defined as follows:
  \begin{mathpar}\small
    \GrTTdruleRhoXXEmpty{} \;
    \GrTTdruleRhoXXExtTy{} \;
    \GrTTdruleRhoXXExtTm{}
  \end{mathpar}
\end{definition}
\noindent
We interpret terms as substitutions, but graded modalities must be
erased and their elimination forms converted into
substitutions (and similarly for the eliminator for tensor products).
\begin{definition}
  \label{def:valid-substitutions}
  Suppose $ \Delta  \odot  \Gamma  \models_{ \varepsilon }  \rho $.  Then the \emph{interpretation of a
    term} $t$ typable in $\Delta  \odot  \Gamma$ is $ \llparenthesis  t  \rrparenthesis_{ \rho }  =  \rho \, t $, but where all let-expressions are translated into
  substitutions, and all graded modalities are erased.
\end{definition}

\noindent
Finally, we prove our main result using semantic typing which
will imply strong normalization. Suppose $ ( \Delta  \mid  \sigma_{{\mathrm{1}}}  \mid  \sigma_{{\mathrm{2}}} )   \odot  \Gamma  \vdash  t  \GrTTsym{:}  A$, then:
\begin{definition}\label{def:int-type-judgments}
  \emph{Semantic typing}, $ ( \Delta  \mid  \sigma_{{\mathrm{1}}}  \mid  \sigma_{{\mathrm{2}}} )   \odot  \Gamma  \models  t  \GrTTsym{:}  A$, is defined as follows:
  \begin{enumerate}\small
  \item If $ ( \Delta  \mid  \sigma  \mid   \textbf{0}  )   \odot  \Gamma  \vdash  A  \GrTTsym{:}   \mathsf{Type}_{ \GrTTsym{1} } $, then for every $ \Delta  \odot  \Gamma  \models_{ \varepsilon }  \rho $, $ \llparenthesis  t  \rrparenthesis_{ \rho }  \, \in \,  \interp{ A }_{ \varepsilon }  \, \GrTTsym{(}   \interp{ t }_{ \varepsilon }   \GrTTsym{)}$.
  \item If $ ( \Delta  \mid  \sigma  \mid   \textbf{0}  )   \odot  \Gamma  \vdash  A  \GrTTsym{:}   \mathsf{Type}_{ \GrTTsym{0} } $, then for every $ \Delta  \odot  \Gamma  \models_{ \varepsilon }  \rho $, $ \llparenthesis  t  \rrparenthesis_{ \rho }  \, \in \,  \interp{ A }_{ \varepsilon } $.
  \end{enumerate}
\end{definition}

\begin{theorem}[Soundness for Semantic Typing]\label{theorem:soundness_for_semantic_typing}
  $ ( \Delta  \mid  \sigma_{{\mathrm{1}}}  \mid  \sigma_{{\mathrm{2}}} )   \odot  \Gamma  \models  t  \GrTTsym{:}  A$.
\end{theorem}

\begin{corollary}[Strong Normalization]\label{corollary:strong_normalization}
  We have $t \, \in \, \mathsf{SN}$.
\end{corollary}

\section{Implementation}\label{sec:implementation}
Our implementation
\implName{}
is based on
a bidirectionalised version of the typing rules here, somewhat
following traditional schemes of bidirectional
typing~\cite{DunfieldK19,dunfield2004tridirectional} but with
grading (similar to Granule~\cite{orchard2019quantitative} but adapted considerably for
the dependent setting). We briefly outline the implementation
scheme and highlight a few key points, rules, and
examples. We use this implementation to explore further
applications of \grtt{}, namely optimising type checking
algorithms.

Bidirectional typing splits declarative typing rules into
\emph{check} and \emph{infer} modes. Furthermore,
bidirectional \grtt{} rules split the grading context (left of
$\odot$) into \emph{input} and \emph{output} contexts where
$ ( \Delta  \mid  \sigma_{{\mathrm{1}}}  \mid  \sigma_{{\mathrm{2}}} )   \odot  \Gamma  \vdash  t  \GrTTsym{:}  A$ is implemented via:
\begin{align*}
(\emph{check})\;\;
 \Delta  ;  \Gamma  \vdash  t  \Leftarrow  A  ;  \sigma_{{\mathrm{1}}}  ;  \sigma_{{\mathrm{2}}} 
\quad \textit{or}\quad
(\emph{infer})\;\;
 \Delta  ;  \Gamma  \vdash  t  \Rightarrow  A  ;  \sigma_{{\mathrm{1}}}  ;  \sigma_{{\mathrm{2}}} 
\end{align*}
where $\Leftarrow$ rules \emph{check} that $t$ has type $A$ and
$\Rightarrow$ rules \emph{infer} (calculate) that $t$ has type $A$.  In both
judgments, the context grading $\Delta$ and context $\Gamma$
left of $\vdash$ are inputs
whereas the grade vectors $\sigma_{{\mathrm{1}}}$ and $\sigma_{{\mathrm{2}}}$
 to the right of $A$ are outputs. This
input-output context approach resembles that employed in linear type
checking~\cite{allais2018typing,hodas1994logic,zalakain2020pi}.
Rather than following a ``left over'' scheme as in these works
(where the output context explains what resources are left), the
output grades here explain what has been used according to the
analysis of grading (`adding up' rather than `taking away').

For example, the following is the \emph{infer} rule for function
elimination:
\begin{align*}
\GrTTdruleInfAlgXXApp{}
\end{align*}
The rule can be read by starting at the input of the
conclusion (left of $\vdash$), then reading top down through each
premise, to calculate the output grades in the rule's conclusion. Any
concrete value or already-bound variable appearing in the output
grades of a premise can be read as causing an equality check in the
type checker.  The last premise checks that the output
subject-type grade $\sigma_{{\mathrm{13}}}$ from the first premise matches
$\sigma_{{\mathrm{1}}}  \GrTTsym{+}  \sigma_{{\mathrm{3}}}$ (which were calculated by later premises).

In contrast, function introduction is a \emph{check} rule:
\begin{align*}
\GrTTdruleChkAlgXXFun{}
\end{align*}
Thus, dependent functions can be checked against type
$ \textstyle (  x  :_{  \textcolor{darkblue}{( s ,  r )}  }  A  )  \to   B $ given input $\Delta; \Gamma$ by first
inferring the type of $A$ and checking that its output
subject-type grade comprises all zeros $ \textbf{0} $. Then the body of the function
$t$ is checked against $B$ under the context
$\Delta  ,  \sigma_{{\mathrm{1}}}; \Gamma  \GrTTsym{,}  x  \GrTTsym{:}  A$ producing grade vectors $\sigma_{{\mathrm{2}}}  ,  s'$
and $\sigma_{{\mathrm{1}}}  ,  r'$ where it is checked that $s = s'$ and $r = r'$
(described implicitly in the rule), i.e., the calculated grades match
those of the binder.

The implementation anticipates some further work for \grtt{}:
the potential for grades which are first-class terms, for which we anticipate complex equations on
grades. For grade equality, \implName{}
has two modes: one which normalises terms
and then compares for syntactic equality, and the other which
discharges constraints via an off-the-shelf SMT solver (we use
Z3~\cite{de2008z3}). We discuss briefly some performance
implications in the next section.

\paragraph{Using Grades to Optimise Type Checking}
\label{sec:more-case-studies}
Abel posited that a dependent theory with quantitative
resource tracking at the type level could
leverage linearity-like optimisations in type checking~\cite{abel2018}.
Our implementation provides a research vehicle for exploring
this idea; we consider one possible optimisation here.

Key to dependent type checking is the substitution of terms into types
in elimination forms (i.e., application, tensor elimination). However,
in a quantitative semiring setting, if a variable has $0$ subject-type
grade, then we know it is irrelevant to type formation
(it is not semantically depended upon, i.e., during normalisation).
Subsequently, substitutions into a $0$-graded variable can be elided
(or allocations to a closure environment can be avoided).
We implemented this optimisation in \implName{} when inferring the type of an
application for $ t_{{\mathrm{1}}} \,{ t_{{\mathrm{2}}} } $ (rule $\GrTTdruleInfAlgXXAppName{}$
above), where the type of $t_{{\mathrm{1}}}$ is inferred as $ \textstyle (  x  :_{  \textcolor{darkblue}{( s ,  \GrTTsym{0} )}  }  A  )  \to   B $. For a
quantitative semiring we know that $x$ irrelevant in $B$,
thus we need not perform the substitution
$\GrTTsym{[}  t_{{\mathrm{2}}}  \GrTTsym{/}  x  \GrTTsym{]}  B$ when type checking the application.

We evaluate this on simple \implName{} programs of an $n$-ary ``fanout'' combinator
implemented via an $n$-ary application combinator, e.g., for arity 3:
\begin{gertyReallySmall}\small
app3 : (a : (0, 6) Type 0) -> (b : (0, 2) Type 0)
-> (x0 : (1, 0) a) -> (x1 : (1, 0) a) -> (x2 : (1, 0) a)
-> (f:(1, 0) ((y0:(1,0) a) -> (y1:(1,0) a) -> (y2:(1,0) a) -> b)) -> b
app3 = \a -> \b -> \x0 -> \x1 -> \x2 -> \f -> f x0 x1 x2

fan3 : (a : (0, 4) Type 0) -> (b : (0, 2) Type 0)
-> (f : (1,0) ((z0 : (1,0) a) -> (z1 : (1,0) a) -> (z2 : (1,0) a) -> b))
-> (x : (3, 0) a) -> b
fan3 = \a -> \b -> \f -> \x -> app3 a b x x x f
\end{gertyReallySmall}
\noindent
Note that \gertyin{fan3} uses its parameter \gertyin{x} three times (hence the
grade $3$) which then incurs substitutions into the type of
\gertyin{app3} during type checking, but each such substitution is
redundant since the type does not depend on these parameters, as
reflected by the $0$ subject-type grades.

\newcommand{\stderr}[1]{&\hspace{-0.5em} (\textcolor{gray}{${#1}$})}
\newcommand{\spdup}[1]{#1}
\definecolor{DarkRed}{rgb}{0.5,0,0}
\newcommand{\esmal}[1]{\textcolor{DarkRed}{#1}}

To evaluate the optimisation and SMT solving vs. normalisation-based equality, we ran
\implName{} on the fan out program for
arities from 3 to 8, with and without the optimisation and
under the two equality approaches.

\noindent
Table~\ref{tab:results} gives the results. For grade equality by
normalisation, the optimisation has a
positive effect on speedup, getting increasingly significant (up to 38\%) as the
overall cost increases. For SMT-based grade equality,
the optimisation causes some slow down
for arity $4$ and $5$ (and just breaking even for arity $3$).  This is
because working out whether the optimisation can be applied
requires checking whether grades are equal to $0$, which incurs extra
SMT solver calls. Eventually, this cost is outweighed by the
time saved by reducing substitutions.
Since the grades here are all relatively simple, it is usually more
efficient for the type checker to normalise and compare terms rather
than compiling to SMT and starting up the external solver, as seen
by longer times for the SMT approach.

The baseline performance here is poor (the implementation
is not highly optimised) partly due to the overhead of computing type formation
judgments often to accurately account for grading. However,
such checks are often recomputed and could be optimised away by
 memoisation. Nevertheless this experiment
gives the evidence that grades can indeed be used to optimise type
checking. A thorough investigation
of grade-directed optimisations is future work.

\begin{table}[t]
\setlength{\tabcolsep}{0.16em}
\begin{center}
\begin{tabular}{c||rlrlc|rlrlc}
& \multicolumn{5}{c|}{\textbf{Normalisation}} & \multicolumn{5}{c}{\qquad\quad\textbf{SMT}} \\[-0.1em] \hline
$n$ & \multicolumn{2}{c}{Base $ms$} &
\multicolumn{2}{c}{Optimised $ms$}
& {\footnotesize{Speedup}} & \multicolumn{2}{c}{Base $ms$} &
\multicolumn{2}{c}{Optimised $ms$}
& {\footnotesize{Speedup}} \\[-0.05em] \hline \hline
3 & 45.71 \stderr{1.72} & 44.08 \stderr{1.28} & \spdup{1.04}
& 77.12 \stderr{2.65} & 76.91 \stderr{2.36} & \spdup{1.00} \\[-0.05em]
4 & 108.75 \stderr{4.09} &  89.73 \stderr{4.73} & \spdup{1.21}
& 136.18 \stderr{5.23} &  162.95 \stderr{3.62} & \spdup{\esmal{0.84}} \\[-0.05em]
5 & 190.57 \stderr{8.31} & 191.25 \stderr{8.13} & \spdup{1.00}
& 279.49 \stderr{15.73} & 289.73 \stderr{23.30} & \spdup{\esmal{0.96}} \\[-0.05em]
6 & 552.11 \stderr{29.00} & 445.26 \stderr{23.50} & \spdup{1.24}
& 680.11 \stderr{16.28} & 557.08 \stderr{13.87} & \spdup{1.22} \\[-0.05em]
7 & 1821.49 \stderr{49.44} & 1348.85 \stderr{26.37} & \spdup{1.35}
& 1797.09 \stderr{43.53} & 1368.45 \stderr{20.16} & \spdup{1.31} \\[-0.05em]
8 & 6059.30 \stderr{132.01} & 4403.10 \stderr{86.57} & \spdup{1.38}
& 5913.06 \stderr{118.83} & 4396.90 \stderr{59.82} & \spdup{1.34}
\end{tabular}
\end{center}
\vspace{-0.2em}
\caption{Performance analysis of grade-based optimisations to type
  checking. Times in milliseconds to 2 d.p. with
  the standard error given in brackets. Measurements are the mean of 10
  trials (run on a 2.7 Ghz Intel Core, 8Gb of RAM,
  Z3 4.8.8).}
\label{tab:results}
\vspace{-1.4em}
\end{table}

\section{Discussion}\label{sec:discussion}
\paragraph{Grading, Coeffects, and Quantitative Types}
The notion of \emph{coeffects}, describing how a program
depends on its context, arose in the literature from two
directions: as a dualisation of effect
types~\cite{petricek2013coeffects,Petricek:2014} and a generalisation of Bounded
Linear Logic to general resource semirings~\cite{ghica2014,gaboardi2014}.
Coeffect systems can capture reuse bounds,
information flow security~\cite{gaboardi2016combining},
hardware scheduling constraints~\cite{ghica2014}, and sensitivity for
differential privacy~\cite{de2014really,GaboardiHHNP13}.
A coeffect-style approach also enables linear types to be retrofitted to
Haskell~\cite{bernardy2017linear}.
A common thread is the annotation of variables in the context with
usage information, drawn from a semiring.
Our approach generalises this idea to capture type, context, and
computational usage.

\newcommand{\qtype}[3]{#1 \stackrel{#2}{:} #3}

McBride~\cite{McBride2016} reconciles linear and dependent
types, allowing types to depend on linear values, refined by
Atkey~\cite{quantitative-type-theory} as Quantitative Type
Theory.  \qtt{} employs coeffect-style annotation of each
assumption in a context with an element of a resource
accounting algebra, with judgments of the form:
\begin{equation*}
x_1 \stackrel{\rho_1}{:} A_1, \ldots, x_n \stackrel{\rho_n}{:} A_n \vdash \qtype{M}{\rho}{B}
\end{equation*}
where $\rho_i, \rho$ are elements of a semiring, and
$\rho = 0$ or $\rho = 1$, respectively
denoting a term which can be used in type formation
(erased at runtime) or at runtime.
Dependent function arrows are of the
form $(x \stackrel{\rho}{:} A) \rightarrow B$, where $\rho$ is
a semiring element that denotes the computational usage of the
parameter.

Variables used for type formation but not computation are
annotated by $0$. Subsequently, type formation
rules are all of the form $0\Gamma \vdash T$, meaning every
variable assumption has a $0$ annotation.
\grtt{} is similar to \qtt{}, but differs in its more extensive grading
to track usage in types, rather than blanketing all type usage with
$0$.  In Atkey's formulation, a term can be promoted to a type if its
result and dependency quantities are all $0$.
A set of rules provide formation of computational
type terms, but these are also graded at $0$.
Subsequently, it is not possible to construct an inhabitant of
$\mathsf{Type}$ that can be used at runtime. We avoid this
shortcoming allowing matching on types.
For example, a computation $t$ that
inspects a type variable $a$ would be typed as:
$ ( \Delta  ,   \textbf{0}   ,  \Delta'  \mid  \sigma_{{\mathrm{1}}}  ,  \GrTTsym{1}  ,  \sigma'_{{\mathrm{1}}}  \mid  \sigma_{{\mathrm{2}}}  ,  r  ,  \sigma'_{{\mathrm{2}}} )   \odot  \Gamma  \GrTTsym{,}  a  \GrTTsym{:}   \mathsf{Type}   \GrTTsym{,}  \Gamma'  \vdash  t  \GrTTsym{:}  B$ denoting $1$ computational use and $r$ type uses in
$B$.

\newcommand{\extra}{\hat{0}}

At first glance, it seems \qtt{}
could be encoded into \grtt{} taking the semiring
 $\mathcal{R}$ of \qtt{} and parameterising \grtt{} by the semiring
$\mathcal{R} \cup \{\extra\}$ where $\extra$ denotes arbitrary
usage in type formation.
However, there is impedance between the two systems as \qtt{}
always annotates type use with $0$. It is not clear how
to make this happen in \grtt{} whilst still having non-$0$
tracking at the computational level, since we use one semiring
for both. Exploring an encoding is future work.

Choudhury et al.~\cite{choudhury2021} give a system closely related
(but arguably simpler) to \qtt{} called $\GrTTdrulename{GraD}$.
One key difference is that rather than annotating type usage with $0$,
grades are simply ignored in types. This makes for a surprisingly
flexible system.  In addition, they show that irrelevance is captured
by the $0$ grade using a heap-based semantics (a result leveraged in
Section~\ref{sec:case-studies}). $\GrTTdrulename{GraD}$ however
does not have the power of type-grades presented here.

\paragraph{Dependent Types and Modalities}
Dal Lago and Gaboardi extend PCF with linear and lightweight dependent
types~\cite{dal2011linear} (then adapted for differential
privacy analysis~\cite{GaboardiHHNP13}). They add
a natural number type indexed by upper and lower bound terms which
index a modality. Combined with linear arrows of the form
$[a < I] . \sigma \multimap \tau$ these describe functions using
the parameter at most $I$ times (where the modality acts
as a binder for index variable $a$ which denotes instantiations).
Their system is leveraged to give fine-grained cost analyses in the context
of Implicit Computational Complexity. Whilst a powerful system, their
approach is restricted in terms of
dependency, where only a specialised type can depend on
specialised natural-number indexed terms (which are non-linear).

Gratzer et al. define a dependently-typed language with a Fitch-style
modality~\cite{DBLP:journals/pacmpl/GratzerSB19}. It seems that such
an approach could also be generalised to a graded modality, although
we have used the natural-deduction style for
our graded modality rather than the Fitch-style.

As discussed in Section~\ref{sec:introduction}, our approach closely
resembles Abel's \emph{resourceful dependent types}~\cite{abel2018}.
Our work expands on the idea, including tensors and the graded
modalities. We considerably developed the associated metatheory,
provide an implementation, and study applications.

\paragraph{Further Work}
One expressive extension is to capture analyses which have an
ordering, e.g., grading by a \emph{pre-ordered} semiring, allowing a
notion of \emph{approximation}. This would enable analyses such as
bounded reuse from Bounded Linear Logic~\cite{girard1992bounded},
intervals with least- and upper-bounds on
use~\cite{orchard2019quantitative}, and top-completed
semirings, with an $\infty$-element denoting arbitrary usage as a
fall-back. We have made progress into exploring the interaction
between approximation and dependent types, and the remainder of this
is left as future work.

A powerful extension of \grtt{} for future work is to
allow grades to be first-class terms.
Typing rules in \grtt{} involving grades could be adapted to
internalise the elements as first-class terms. We could then, e.g., define the map function over
sized vectors, which requires that the parameter function is
used exactly the same number of times as the length of the
vector:
{\small{
\begin{align*}
&  \mathit{map} :  \textstyle (  n  :_{  \textcolor{darkblue}{( \GrTTsym{0} ,  \GrTTsym{5} )}  }   \mathsf{nat}   )  \to    \textstyle (  a  :_{  \textcolor{darkblue}{( \GrTTsym{0} ,  n  \GrTTsym{+}  \GrTTsym{1} )}  }   \mathsf{Type}   )  \to    \textstyle (  b  :_{  \textcolor{darkblue}{( \GrTTsym{0} ,  n  \GrTTsym{+}  \GrTTsym{1} )}  }   \mathsf{Type}   )    \to\\[-0.1em]
& \qquad\,\,\,\,\;\;\;\;  \textstyle (  f  :_{  \textcolor{darkblue}{( n ,  \GrTTsym{0} )}  }   \textstyle (  x  :_{  \textcolor{darkblue}{( \GrTTsym{1} ,  \GrTTsym{0} )}  }  a  )  \to   b   )  \to    \textstyle (  xs  :_{  \textcolor{darkblue}{( \GrTTsym{1} ,  \GrTTsym{0} )}  }   \mathsf{Vec}\, n \, a   )  \to    \mathsf{Vec}\, n \, b   
\end{align*}}}
This type provides strong guarantees: the only
well-typed implementations do the correct
thing, up to permutations of the result vector. Without the
grading, an implementation could apply $f$ fewer than $n$
times, replicating some of the transformed elements;
here we know that $f$ must be applied exactly $n$-times.

A further appealing possibility for \grtt{} is to allow
the semiring to be defined internally, rather than as
a meta-level parameter, leveraging dependent types for
proofs of key properties. An implementation could specify what is
required for a semiring instance, e.g., a record type capturing the
operations and properties of a semiring. The rules of
\grtt{} could then be extended, similarly to the extension to
first-class grades, with the provision of the semiring(s) coming from
\grtt{} terms. Thus, anywhere with a grading premise
$ ( \Delta  \mid  \sigma_{{\mathrm{1}}}  \mid  \sigma_{{\mathrm{2}}} )   \odot  \Gamma  \vdash   r   \GrTTsym{:}   \mathcal{R} $ would also require a premise
$ ( \Delta  \mid  \sigma_{{\mathrm{2}}}  \mid   \textbf{0}  )   \odot  \Gamma  \vdash   \mathcal{R}   \GrTTsym{:}   \mathsf{Semiring} $.
This opens up the ability for programmers and library
developers to provide custom modes of resource tracking with their
libraries, allowing domain-specific program
verification.

\paragraph{Conclusions}

The paradigm of `grading' exposes the inherent structure of a
type theory, proof theory, or semantics by matching the underlying
structure with some algebraic structure augmenting the types. This
idea has been employed for reasoning about side effects via
graded monads~\cite{katsumata2014parametric}, and
reasoning about data flow as discussed here by semiring grading.
 Richer algebras could be employed to capture other aspects,
such as \emph{ordered logics} in which the exchange rule can be
controlled via grading (existing work has done this via
modalities~\cite{DBLP:journals/corr/abs-1904-06847}).

We developed the core of grading in the context of
dependent-types, treating
types and terms equally (as one comes to expect in
dependent-type theories). The tracking of data flow in types
appears complex since we must account for how variables are
used to form types in both the context and in the subject type, making
sure not to repeat context formation use. The result however is a
powerful system for studying dependencies in type theories, as shown
by our ability to study different theories just be specialising grades.
Whilst not yet a fully fledged implementation, \implName{} is a useful
test bed for further exploration.

\medskip
\noindent
\textit{Acknowledgments} Orchard is supported by EPSRC grant EP/T013516/1.

\bibliographystyle{splncs04}
\bibliography{ref}

\vfill

{\small\medskip\noindent{\bf Open Access} This chapter is licensed under the terms of the Creative Commons\break Attribution 4.0 International License (\url{http://creativecommons.org/licenses/by/4.0/}), which permits use, sharing, adaptation, distribution and reproduction in any medium or format, as long as you give appropriate credit to the original author(s) and the source, provide a link to the Creative Commons license and indicate if changes were made.}

{\small \spaceskip .28em plus .1em minus .1em The images or other third party material in this chapter are included in the chapter's Creative Commons license, unless indicated otherwise in a credit line to the material.~If material is not included in the chapter's Creative Commons license and your intended\break use is not permitted by statutory regulation or exceeds the permitted use, you will need to obtain permission directly from the copyright holder.}

\medskip\noindent\includegraphics{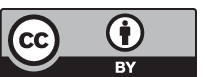}

\newpage

\appendix

\section{Appendix}\label{appendix}

\subsection{Full Typing Rules for \grtt{}}\label{grtt--full-rules}

\begin{figure}[h]
  {\small{
 \begin{minipage}{1.0\linewidth}
  \vspace{-0.5em}
\begin{align*}
\hspace{-0.75em}
  \begin{array}{c}
    \GrTTdruleWfXXEmpty{}
    \quad
    \GrTTdruleWfXXExt{}
  \end{array}
\end{align*}
\end{minipage}}}
\vspace{-0.5em}
\caption{Well-formed contexts for \grtt{}}\label{fig:wf-contexts-GrTT}
\end{figure}

\begin{figure}[h]
\begin{align*}
  \begin{array}{c}
    \GrTTdruleSTXXEq{}
    \\[2em]
    \GrTTdruleSTXXTrans{}
    \\[2em]
    \GrTTdruleSTXXTy{}
    \\[2em]
    \GrTTdruleSTXXArrow{}
    \\[2em]
    \GrTTdruleSTXXTen{}
    \\[2em]
    \GrTTdruleSTXXBox{}
  \end{array}
\end{align*}
\vspace{-0.5em}
\caption{Subtyping for \grtt{}}\label{fig:subtyping-GrTT}
\end{figure}

\begin{figure}[H]
  {\small{
 \begin{minipage}{1.0\linewidth}
  \vspace{-0.5em}
\begin{align*}
\hspace{-0.75em}
  \begin{array}{c}
    \GrTTdruleTXXType{}
    \quad
    \GrTTdruleTXXVar{}
    \\[1.45em]
    \GrTTdruleTXXArrow{}
    \\[1.45em]
    \GrTTdruleTXXFun{}
    \\[1.45em]
    \GrTTdruleTXXApp{}
    \\[1.45em]
    \GrTTdruleTXXTen{}
    \\[1.45em]
    \GrTTdruleTXXPair{}
    \\[1.45em]
    \GrTTdruleTXXTenCut{}
    \\[1.45em]
    \GrTTdruleTXXBox{}
    \quad
    \GrTTdruleTXXBoxI{}
    \\[1.45em]
    \GrTTdruleTXXBoxE{}
    \\[1.45em]
    \GrTTdruleTXXTyConv{}
  \end{array}
\end{align*}
\end{minipage}}}
\vspace{-0.5em}
\caption{Typing for \grtt{}}\label{fig:typing-GrTT}
\end{figure}

\begin{figure}[H]
  {\small{
 \begin{minipage}{1.0\linewidth}
  \vspace{-0.5em}
\begin{align*}
\hspace{-0.75em}
  \begin{array}{c}
    \GrTTdruleTEQXXRefl{}
    \quad
    \GrTTdruleTEQXXTrans{}
    \\[1.45em]
    \GrTTdruleTEQXXSym{}
    \quad
    \GrTTdruleTEQXXConvTy{}
    \\[1.45em]
    \GrTTdruleTEQXXArrow{}
    \\[1.45em]
    \GrTTdruleTEQXXArrowComp{}
    \\[1.45em]
    \GrTTdruleTEQXXArrowUniq{}
    \\[1.45em]
    \GrTTdruleTEQXXFun{}
    \\[1.45em]
    \GrTTdruleTEQXXApp{}
    \\[1.45em]
    \GrTTdruleTEQXXTen{}
    \\[1.45em]
    \GrTTdruleTEQXXTenComp{}
    \\[1.45em]
    \GrTTdruleTEQXXPair{}
    \\[1.45em]
    \GrTTdruleTEQXXTenCut{}
    \\[1.45em]
    \GrTTdruleTEQXXTenU{}
    \\[1.45em]
    \GrTTdruleTEQXXBox{}
    \quad
    \GrTTdruleTEQXXBoxI{}
    \\[1.45em]
    \GrTTdruleTEQXXBoxB{}
    \\[1.45em]
    \GrTTdruleTEQXXBoxE{}
    \\[1.45em]
    \GrTTdruleTEQXXBoxU{}
  \end{array}
\end{align*}
\end{minipage}}}
\vspace{-0.5em}
\caption{Term equality for \grtt{}}\label{fig:term-equality-GrTT}
\end{figure}

\newpage

\subsection{Supplement}\label{appendix--supplement}


\section*{Supplementary material}

# A  Appendix

## A.1  Full Typing Rules for Grtt

$$\frac{}{\emptyset \odot \emptyset \vdash} \text{WF}\emptyset \qquad \frac{(\Delta \mid \sigma \mid \mathbf{0}) \odot \Gamma \vdash A : \text{Type}_l}{\Delta, \sigma \odot \Gamma, x : A \vdash} \text{WFExt}$$

**Fig. 1.** Well-formed contexts for GRTT

$$\frac{(\Delta \mid \sigma \mid \mathbf{0}) \odot \Gamma \vdash A = B : \text{Type}_l}{(\Delta \mid \sigma) \odot \Gamma \vdash A \leq B} \leq_{\text{Eq}}$$

$$\frac{(\Delta \mid \sigma) \odot \Gamma \vdash A \leq B \quad (\Delta \mid \sigma) \odot \Gamma \vdash B \leq C}{(\Delta \mid \sigma) \odot \Gamma \vdash A \leq C} \leq_{\text{Trans}}$$

$$\frac{\Delta \odot \Gamma \vdash \quad l \leq l'}{(\Delta \mid \mathbf{0}) \odot \Gamma \vdash \text{Type}_l \leq \text{Type}_{l'}} \leq_{\text{Type}}$$

$$\frac{\begin{array}{c}(\Delta, \sigma_1 \mid \sigma_2, r \mid \mathbf{0}) \odot \Gamma, x : A \vdash B : \text{Type}_l \\ (\Delta \mid \sigma_1) \odot \Gamma \vdash A' \leq A \quad (\Delta, \sigma_1 \mid \sigma_2, r) \odot \Gamma, x : A' \vdash B \leq B'\end{array}}{(\Delta \mid \sigma_1 + \sigma_2) \odot \Gamma \vdash (x :_{(s,r)} A) \to B \leq (x :_{(s,r)} A') \to B'} \leq_{\to}$$

$$\frac{(\Delta, \sigma_1 \mid \sigma_2, r) \odot \Gamma, x : A \vdash B \leq B'}{(\Delta \mid \sigma_1 + \sigma_2) \odot \Gamma \vdash (x :_r A) \otimes B \leq (x :_r A) \otimes B'} \leq_{\otimes}$$

$$\frac{(\Delta \mid \sigma) \odot \Gamma \vdash A \leq A'}{(\Delta \mid \sigma) \odot \Gamma \vdash \Box_s A \leq \Box_s A'} \leq_{\Box}$$

**Fig. 2.** Subtyping for GRTT

$$\frac{\Delta \odot \Gamma \vdash}{(\Delta \mid \mathbf{0} \mid \mathbf{0}) \odot \Gamma \vdash \mathsf{Type}_l : \mathsf{Type}_{\mathsf{suc}\ l}}\ \mathsf{Type} \qquad \frac{(\Delta_1, \sigma, \Delta_2) \odot \Gamma_1, x : A, \Gamma_2 \vdash \quad |\Delta_1| = |\Gamma_1|}{(\Delta_1, \sigma, \Delta_2 \mid \mathbf{0}^{|\Delta_1|}, 1, \mathbf{0} \mid \sigma, 0, \mathbf{0}) \odot \Gamma_1, x : A, \Gamma_2 \vdash x : A}\ \mathrm{Var}$$

$$\frac{(\Delta \mid \sigma_1 \mid \mathbf{0}) \odot \Gamma \vdash A : \mathsf{Type}_{l_1} \quad (\Delta, \sigma_1 \mid \sigma_2, r \mid \mathbf{0}) \odot \Gamma, x : A \vdash B : \mathsf{Type}_{l_2}}{(\Delta \mid \sigma_1 + \sigma_2 \mid \mathbf{0}) \odot \Gamma \vdash (x :_{(s,r)} A) \to B : \mathsf{Type}_{l_1 \sqcup l_2}}\ \to$$

$$\frac{(\Delta, \sigma_1 \mid \sigma_3, r \mid \mathbf{0}) \odot \Gamma, x : A \vdash B : \mathsf{Type}_l \quad (\Delta, \sigma_1 \mid \sigma_2, s \mid \sigma_3, r) \odot \Gamma, x : A \vdash t : B}{(\Delta \mid \sigma_2 \mid \sigma_1 + \sigma_3) \odot \Gamma \vdash \lambda x.t : (x :_{(s,r)} A) \to B}\ \lambda_i$$

$$\frac{\begin{array}{c}(\Delta, \sigma_1 \mid \sigma_3, r \mid \mathbf{0}) \odot \Gamma, x : A \vdash B : \mathsf{Type}_l \\ (\Delta \mid \sigma_2 \mid \sigma_1 + \sigma_3) \odot \Gamma \vdash t_1 : (x :_{(s,r)} A) \to B \quad (\Delta \mid \sigma_4 \mid \sigma_1) \odot \Gamma \vdash t_2 : A\end{array}}{(\Delta \mid \sigma_2 + s * \sigma_4 \mid \sigma_3 + r * \sigma_4) \odot \Gamma \vdash t_1\, t_2 : [t_2/x]B}\ \lambda_e$$

$$\frac{(\Delta \mid \sigma_1 \mid \mathbf{0}) \odot \Gamma \vdash A : \mathsf{Type}_l \quad (\Delta, \sigma_1 \mid \sigma_2, r \mid \mathbf{0}) \odot \Gamma, x : A \vdash B : \mathsf{Type}_l}{(\Delta \mid \sigma_1 + \sigma_2 \mid \mathbf{0}) \odot \Gamma \vdash (x :_r A) \otimes B : \mathsf{Type}_l}\ \otimes$$

$$\frac{\begin{array}{c}(\Delta, \sigma_1 \mid \sigma_3, r \mid \mathbf{0}) \odot \Gamma, x : A \vdash B : \mathsf{Type}_l \\ (\Delta \mid \sigma_2 \mid \sigma_1) \odot \Gamma \vdash t_1 : A \quad (\Delta \mid \sigma_4 \mid \sigma_3 + r * \sigma_2) \odot \Gamma \vdash t_2 : [t_1/x]B\end{array}}{(\Delta \mid \sigma_2 + \sigma_4 \mid \sigma_1 + \sigma_3) \odot \Gamma \vdash (t_1, t_2) : (x :_r A) \otimes B}\ \otimes_i$$

$$\frac{\begin{array}{c}(\Delta \mid \sigma_3 \mid \sigma_1 + \sigma_2) \odot \Gamma \vdash t_1 : (x :_r A) \otimes B \\ (\Delta, (\sigma_1 + \sigma_2) \mid \sigma_5, r' \mid \mathbf{0}) \odot \Gamma, z : (x :_r A) \otimes B \vdash C : \mathsf{Type}_l \\ (\Delta, \sigma_1, (\sigma_2, r) \mid \sigma_4, s, s \mid \sigma_5, r', r') \odot \Gamma, x : A, y : B \vdash t_2 : [(x,y)/z]C\end{array}}{(\Delta \mid \sigma_4 + s * \sigma_3 \mid \sigma_5 + r' * \sigma_3) \odot \Gamma \vdash \mathsf{let}\ (x,y) = t_1\ \mathsf{in}\ t_2 : [t_1/z]C}\ \otimes_e$$

$$\frac{(\Delta \mid \sigma \mid \mathbf{0}) \odot \Gamma \vdash A : \mathsf{Type}_l}{(\Delta \mid \sigma \mid \mathbf{0}) \odot \Gamma \vdash \Box_s A : \mathsf{Type}_l}\ \Box \qquad \frac{(\Delta \mid \sigma_1 \mid \sigma_2) \odot \Gamma \vdash t : A}{(\Delta \mid s * \sigma_1 \mid \sigma_2) \odot \Gamma \vdash \Box t : \Box_s A}\ \Box_i$$

$$\frac{\begin{array}{c}(\Delta, \sigma_2 \mid \sigma_4, r \mid \mathbf{0}) \odot \Gamma, z : \Box_s A \vdash B : \mathsf{Type}_l \\ (\Delta \mid \sigma_1 \mid \sigma_2) \odot \Gamma \vdash t_1 : \Box_s A \quad (\Delta, \sigma_2 \mid \sigma_3, s \mid \sigma_4, (s * r)) \odot \Gamma, x : A \vdash t_2 : [\Box x/z]B\end{array}}{(\Delta \mid \sigma_1 + \sigma_3 \mid \sigma_4 + r * \sigma_1) \odot \Gamma \vdash \mathsf{let}\ \Box x = t_1\ \mathsf{in}\ t_2 : [t_1/z]B}\ \Box_e$$

$$\frac{(\Delta \mid \sigma_1 \mid \sigma_2) \odot \Gamma \vdash t : A \quad (\Delta \mid \sigma_2) \odot \Gamma \vdash A \leq B}{(\Delta \mid \sigma_1 \mid \sigma_2) \odot \Gamma \vdash t : B}\ \mathrm{Conv}$$

**Fig. 3.** Typing for GRTT

$$\frac{(\Delta \mid \sigma_1 \mid \sigma_2) \odot \Gamma \vdash t : A}{(\Delta \mid \sigma_1 \mid \sigma_2) \odot \Gamma \vdash t = t : A} \text{Refl} \quad \frac{(\Delta \mid \sigma_1 \mid \sigma_2) \odot \Gamma \vdash t_1 = t_2 : A \quad (\Delta \mid \sigma_1 \mid \sigma_2) \odot \Gamma \vdash t_2 = t_3 : A}{(\Delta \mid \sigma_1 \mid \sigma_2) \odot \Gamma \vdash t_1 = t_3 : A} \text{Trans}$$

$$\frac{(\Delta \mid \sigma_1 \mid \sigma_2) \odot \Gamma \vdash t_1 = t_2 : A}{(\Delta \mid \sigma_1 \mid \sigma_2) \odot \Gamma \vdash t_2 = t_1 : A} \text{Sym} \quad \frac{(\Delta \mid \sigma_1 \mid \sigma_2) \odot \Gamma \vdash t_1 = t_2 : A \quad (\Delta \mid \sigma_2) \odot \Gamma \vdash A \leq B}{(\Delta \mid \sigma_1 \mid \sigma_2) \odot \Gamma \vdash t_1 = t_2 : B} \text{ConvTy}$$

$$\frac{(\Delta \mid \sigma_1 \mid \mathbf{0}) \odot \Gamma \vdash A = C : \mathsf{Type}_{l_1} \quad (\Delta, \sigma_1 \mid \sigma_2, r \mid \mathbf{0}) \odot \Gamma, x : A \vdash B = D : \mathsf{Type}_{l_2}}{(\Delta \mid \sigma_1 + \sigma_2 \mid \mathbf{0}) \odot \Gamma \vdash (x :_{(s,r)} A) \to B = (x :_{(s,r)} C) \to D : \mathsf{Type}_{l_1 \sqcup l_2}} \text{Eq}_\to$$

$$\frac{(\Delta, \sigma_1 \mid \sigma_2, s \mid \sigma_3, r) \odot \Gamma, x : A \vdash t_1 : B \quad (\Delta \mid \sigma_4 \mid \sigma_1) \odot \Gamma \vdash t_2 : A}{(\Delta \mid \sigma_2 + s * \sigma_4 \mid \sigma_3 + r * \sigma_4) \odot \Gamma \vdash (\lambda x.t_1)\, t_2 = [t_2/x]t_1 : [t_2/x]B} \text{Eq}_{\to_c}$$

$$\frac{(\Delta \mid \sigma_1 \mid \sigma_2) \odot \Gamma \vdash t : (x :_{(s,r)} A) \to B}{(\Delta \mid \sigma_1 \mid \sigma_2) \odot \Gamma \vdash t = \lambda x.(t\, x) : (x :_{(s,r)} A) \to B} \text{Eq}_{\to_u}$$

$$\frac{(\Delta, \sigma_1 \mid \sigma_2, s \mid \sigma_3, r) \odot \Gamma, x : A \vdash t_1 = t_2 : B}{(\Delta \mid \sigma_2 \mid \sigma_1 + \sigma_3) \odot \Gamma \vdash \lambda x.t_1 = \lambda x.t_2 : (x :_{(s,r)} A) \to B} \text{Eq}_{\to_i}$$

$$\frac{\begin{array}{c}(\Delta, \sigma_1 \mid \sigma_3, r \mid \mathbf{0}) \odot \Gamma, x : A \vdash B : \mathsf{Type}_l \\ (\Delta \mid \sigma_2 \mid \sigma_1 + \sigma_3) \odot \Gamma \vdash t_1 = t_2 : (x :_{(s,r)} A) \to B \quad (\Delta \mid \sigma_4 \mid \sigma_1) \odot \Gamma \vdash t_3 = t_4 : A\end{array}}{(\Delta \mid \sigma_2 + s * \sigma_4 \mid \sigma_3 + r * \sigma_4) \odot \Gamma \vdash t_1\, t_3 = t_2\, t_4 : [t_3/x]B} \text{Eq}_{\to_e}$$

$$\frac{(\Delta \mid \sigma_1 \mid \mathbf{0}) \odot \Gamma \vdash A = C : \mathsf{Type}_l \quad (\Delta, \sigma_1 \mid \sigma_2, r \mid \mathbf{0}) \odot \Gamma, x : A \vdash B = D : \mathsf{Type}_l}{(\Delta \mid \sigma_1 + \sigma_2 \mid \mathbf{0}) \odot \Gamma \vdash (x :_r A) \otimes B = (x :_s C) \otimes D : \mathsf{Type}_l} \text{Eq}_\otimes$$

$$\frac{\begin{array}{c}(\Delta, (\sigma_1 + \sigma_2) \mid \sigma_5, r' \mid \mathbf{0}) \odot \Gamma, z : (x :_r A) \otimes B \vdash C : \mathsf{Type}_l \\ (\Delta \mid \sigma_3 \mid \sigma_1) \odot \Gamma \vdash t_1 : A \quad (\Delta \mid \sigma_6 \mid \sigma_2 + r * \sigma_3) \odot \Gamma \vdash t_2 : [t_1/x]B \\ (\Delta, \sigma_1, (\sigma_2, r) \mid \sigma_4, s, s \mid \sigma_5, r', r') \odot \Gamma, x : A, y : B \vdash t_3 : [(x,y)/z]C\end{array}}{(\Delta \mid \sigma_4 + s * (\sigma_3 + \sigma_6) \mid \sigma_5 + r' * (\sigma_3 + \sigma_6)) \odot \Gamma \vdash \mathsf{let}\,(x,y) = (t_1, t_2)\,\mathsf{in}\,t_3 = [t_1, t_2/x, y]t_3 : [(t_1, t_2)/z]C} \text{Eq}_{\otimes_c}$$

$$\frac{\begin{array}{c}(\Delta, \sigma_1 \mid \sigma_3, r \mid \mathbf{0}) \odot \Gamma, x : A \vdash B : \mathsf{Type}_l \\ (\Delta \mid \sigma_2 \mid \sigma_1) \odot \Gamma \vdash t_1 = t_2 : A \quad (\Delta \mid \sigma_4 \mid \sigma_3 + r * \sigma_2) \odot \Gamma \vdash t_3 = t_4 : [t_1/x]B\end{array}}{(\Delta \mid \sigma_2 + \sigma_4 \mid \sigma_1 + \sigma_3) \odot \Gamma \vdash (t_1, t_2) = (t_3, t_4) : (x :_r A) \otimes B} \text{Eq}_{\otimes_i}$$

$$\frac{\begin{array}{c}(\Delta \mid \sigma_3 \mid \sigma_1 + \sigma_2) \odot \Gamma \vdash t_1 = t_2 : (x :_r A) \otimes B \\ (\Delta, (\sigma_1 + \sigma_2) \mid \sigma_5, r' \mid \mathbf{0}) \odot \Gamma, z : (x :_r A) \otimes B \vdash C : \mathsf{Type}_l \\ (\Delta, \sigma_1, (\sigma_2, r) \mid \sigma_4, s, s \mid \sigma_5, r', r') \odot \Gamma, x : A, y : B \vdash t_3 = t_4 : [(x,y)/z]C\end{array}}{(\Delta \mid \sigma_4 + s * \sigma_3 \mid \sigma_5 + r' * \sigma_3) \odot \Gamma \vdash (\mathsf{let}\,(x,y) = t_1\,\mathsf{in}\,t_3) = (\mathsf{let}\,(x,y) = t_2\,\mathsf{in}\,t_4) : [t_1/z]C} \text{Eq}_{\otimes_e}$$

$$\frac{(\Delta \mid \sigma_1 \mid \sigma_2) \odot \Gamma \vdash t : (x :_r A) \otimes B}{(\Delta \mid \sigma_1 \mid \sigma_2) \odot \Gamma \vdash t = (\mathsf{let}\,(x,y) = t\,\mathsf{in}\,(x,y)) : (x :_r A) \otimes B} \text{Eq}_{\otimes_u}$$

$$\frac{(\Delta \mid \sigma \mid \mathbf{0}) \odot \Gamma \vdash A = B : \mathsf{Type}_l}{(\Delta \mid \sigma \mid \mathbf{0}) \odot \Gamma \vdash \Box_s A = \Box_s B : \mathsf{Type}_l} \text{Eq}_\Box \quad \frac{(\Delta \mid \sigma_1 \mid \sigma_2) \odot \Gamma \vdash t_1 = t_2 : A}{(\Delta \mid s * \sigma_1 \mid \sigma_2) \odot \Gamma \vdash \Box t_1 = \Box t_2 : \Box_s A} \text{Eq}_{\Box_i}$$

$$\frac{\begin{array}{c}(\Delta, \sigma_2 \mid \sigma_4, r \mid \mathbf{0}) \odot \Gamma, z : \Box_s A \vdash B : \mathsf{Type}_l \\ (\Delta \mid \sigma_1 \mid \sigma_2) \odot \Gamma \vdash t_1 : A \quad (\Delta, \sigma_2 \mid \sigma_3, s \mid \sigma_4, (s * r)) \odot \Gamma, x : A \vdash t_2 : [\Box x/z]B\end{array}}{(\Delta \mid \sigma_3 + s * \sigma_1 \mid \sigma_4 + s * r * \sigma_1) \odot \Gamma \vdash (\mathsf{let}\,\Box x = \Box t_1\,\mathsf{in}\,t_2) = [t_1/x]t_2 : [\Box t_1/z]B} \text{Eq}_{\Box_c}$$

$$\frac{\begin{array}{c}(\Delta, \sigma_2 \mid \sigma_4, r \mid \mathbf{0}) \odot \Gamma, z : \Box_s A \vdash B : \mathsf{Type}_l \\ (\Delta \mid \sigma_1 \mid \sigma_2) \odot \Gamma \vdash t_1 = t_2 : \Box_s A \\ (\Delta, \sigma_2 \mid \sigma_3, s \mid \sigma_4, (s * r)) \odot \Gamma, x : A \vdash t_3 = t_4 : [\Box x/z]B\end{array}}{(\Delta \mid \sigma_1 + \sigma_3 \mid \sigma_4 + r * \sigma_1) \odot \Gamma \vdash (\mathsf{let}\,\Box x = t_1\,\mathsf{in}\,t_3) = (\mathsf{let}\,\Box x = t_2\,\mathsf{in}\,t_4) : [t_1/z]B} \text{Eq}_{\Box_e}$$

$$\frac{(\Delta \mid \sigma_1 \mid \sigma_2) \odot \Gamma \vdash t : \Box_s A}{(\Delta \mid \sigma_1 \mid \sigma_2) \odot \Gamma \vdash t = (\mathsf{let}\,\Box x = t\,\mathsf{in}\,\Box x) : \Box_s A} \text{Eq}_{\Box_u}$$

**Fig. 4.** Term equality for GRTT

## A.2 Supplement

# Specification and Analysis of Graded Modal Dependent Type Theory (GRTT)

# Contents



# 1 Typing judgment overview

Typing judgments in GRTT have the form $(\Delta \mid \sigma_1 \mid \sigma_2) \odot \Gamma \vdash t : A$, which may be written as $\left(\begin{smallmatrix}\Delta\\\sigma_1\\\sigma_2\end{smallmatrix}\right) \odot \Gamma \vdash t : A$, should it aid readability.

$\sigma_1$ and $\sigma_2$ are *grade vectors*, which are vectors of grades which describe subject ($t$) and subject type ($A$) use, respectively. $\Delta$ is a *context grade vector*, which is a vector of grade vectors, accounting for usage in the typing of assumptions in the context ($\Gamma$).

Throughout the theory, we implicitly assume that for any judgment $(\Delta \mid \sigma_1 \mid \sigma_2) \odot \Gamma \vdash \mathcal{J}$, we have $|\Delta| = |\sigma_1| = |\sigma_2| = |\Gamma|$ (i.e., that sizes align), and that for every vector $\sigma$ in $\Delta$, the size of $\sigma$ is the same as its index in $\Delta$. We assume these sizing requirements for all judgments (even those without grade vectors), not just typing ones.

The notation $\mathbf{0}$ is used to denote a vector consisting entirely of 0 grades, to whichever size would be necessary to satisfy the above sizing conditions. We may write $\mathbf{0}^\pi$ to specify a size ($\pi$), if we feel that it aids understanding.

**Proofs over multiple judgments** When lemmas should hold for multiple forms of judgment, we may use the syntax $(\Delta \mid \sigma_1 \mid \sigma_2) \odot \Gamma \vdash \mathcal{J}$ (where $\mathcal{J}$ ranges over judgments), optionally restricting $\mathcal{J}$ to a subset of judgments, with $\mathcal{J}$ by default ranging over the following forms of judgment:

- $\Delta \odot \Gamma \vdash$;
- $(\Delta \mid \sigma) \odot \Gamma \vdash A \leq B$;
- $(\Delta \mid \sigma_1 \mid \sigma_2) \odot \Gamma \vdash t_1 : A$;
- $(\Delta \mid \sigma_1 \mid \sigma_2) \odot \Gamma \vdash t_1 = t_2 : A$

Thus, for example, a lemma stating "if $(\Delta \mid \sigma_1 \mid \sigma_2) \odot \Gamma \vdash \mathcal{J}$ then $\Delta \odot \Gamma \vdash$", this would mean all of the following hold:

- if $\Delta \odot \Gamma \vdash$ then $\Delta \odot \Gamma \vdash$;
- if $(\Delta \mid \sigma) \odot \Gamma \vdash A \leq B$ then $\Delta \odot \Gamma \vdash$;
- if $(\Delta \mid \sigma_1 \mid \sigma_2) \odot \Gamma \vdash t_1 : A$ then $\Delta \odot \Gamma \vdash$;
- if $(\Delta \mid \sigma_1 \mid \sigma_2) \odot \Gamma \vdash t_1 = t_2 : A$ then $\Delta \odot \Gamma \vdash$

Note that repeated uses of $\mathcal{J}$ restrict the form of resulting judgments accordingly, and term operations map over the judgment, thus a lemma "if $(\Delta \mid \sigma_1 \mid \sigma_2) \odot \Gamma \vdash \mathcal{J}$ then $(\Delta' \mid \sigma_1' \mid \sigma_2') \odot \Gamma' \vdash [t/x]\mathcal{J}$" would mean all of the following hold:

- if $\Delta \odot \Gamma \vdash$ then $\Delta' \odot \Gamma' \vdash$;
- if $(\Delta \mid \sigma) \odot \Gamma \vdash A \leq B$ then $(\Delta' \mid \sigma') \odot \Gamma' \vdash [t/x]A \leq [t/x]B$;
- if $(\Delta \mid \sigma_1 \mid \sigma_2) \odot \Gamma \vdash t_1 : A$ then $(\Delta' \mid \sigma_1' \mid \sigma_2') \odot \Gamma' \vdash [t/x]t_1 : [t/x]A$;
- if $(\Delta \mid \sigma_1 \mid \sigma_2) \odot \Gamma \vdash t_1 = t_2 : A$ then $(\Delta' \mid \sigma_1' \mid \sigma_2') \odot \Gamma' \vdash [t/x]t_1 = [t/x]t_2 : [t/x]A$

Operations on terms lift naturally to operations on judgments. For example:

**Definition 1.1** (Subsntitution for multiple judgments). *We write $[t/x]\mathcal{J}$ to mean the substitution of t for x in each component of $\mathcal{J}$. For example, if $\mathcal{J}$ is $t_1 = t_2 : A$, then $[t/x]\mathcal{J}$ is $[t/x]t_1 = [t/x]t_2 : [t/x]A$.*

**Definition 1.2** (Free-variable non-membership over multiple judgments). *We write $x \notin \mathsf{FV}(\mathcal{J})$ to mean x is not in* any *of the components of $\mathcal{J}$. E.g., if $\mathcal{J}$ is $t_1 = t_2 : A$, then $x \notin \mathsf{FV}(\mathcal{J})$ is $x \notin (\mathsf{FV}(t_1) \cup \mathsf{FV}(t_2) \cup \mathsf{FV}(A))$.*

# 2 Definitions

**Definition 2.1** (Well-formed contexts). *We say that a context $\Gamma$ is* well-formed *under a context grade vector $\Delta$ $(\Delta \mid \sigma_1 \mid \sigma_2) \odot \Gamma$ $\Delta \odot \Gamma \vdash$), if and only if for every $\Delta_1, \sigma, \Delta_2, \Gamma_1, x, B,$ and $\Gamma_2$ satisfying the equations $\Delta = (\Delta_1, \sigma, \Delta_2); \Gamma = (\Gamma_1, x : B, \Gamma_2);$ and $|\Delta_1| = |\Gamma_1|$; we have $(\Delta_1 \mid \sigma \mid \mathbf{0}) \odot \Gamma_1 \vdash B : \mathsf{Type}_l$ for some level $l$.*

**Definition 2.2** (Move on grade vectors). *$\mathsf{mv}(\pi_1; \pi_2; \sigma)$ moves the element at index $\pi_1$ to index $\pi_2$ in $\sigma$ (pushing back elements as necessary). As a special case of this, we define $\mathsf{exch}(\pi; \sigma) = \mathsf{mv}((\pi+1); \pi; \sigma)$.*

**Definition 2.3** (Move on context grade vectors). *The operation $\mathsf{mv}(\pi_1; \pi_2; \Delta)$ is $\mathsf{mv}(\pi_1; \pi_2; \sigma)$ for each $\sigma$ in $\Delta$. As a special case of this, we define $\mathsf{exch}(\pi; \Delta) = \mathsf{mv}((\pi+1); \pi; \Delta)$.*

**Definition 2.4** (Contraction on context grade vectors). *$\mathsf{contr}(\pi; \Delta)$ combines the elements at index $\pi$ and $\pi+1$ for each grade vector in $\Delta$, via addition. This is defined as $\mathsf{contr}(\pi; \Delta) = \Delta\backslash(\pi+1) + (\Delta/(\pi+1))*(\mathbf{0}^\pi, 1)$.*

**Definition 2.5** (Insert on grade vectors). *The operation $\mathsf{ins}(\pi; s; \sigma)$ inserts the element $s$ at index $\pi$ in $\sigma$, such that all elements preceeding index $\pi$ keep their positions, and every element at index $\pi$ or greater in $\sigma$ will be shifted one index later in the new vector.*

**Definition 2.6** (Insert on context grade vectors). *The operation $\mathsf{ins}(\pi; s; \Delta)$ is $\mathsf{ins}(\pi; s; \sigma)$ for each $\sigma$ in $\Delta$.*

**Definition 2.7** (Choose on grade vectors). *The operation $\sigma/\pi$ selects the element at index $\pi$ of $\sigma$.*

**Definition 2.8** (Choose on context grade vectors). *The operation $\Delta/\pi$ is $\sigma/\pi$ on each $\sigma$ in $\Delta$, producing a new grade vector of size $|\Delta|$.*

**Definition 2.9** (Discard on grade vectors). *The operation $\sigma\backslash\pi$ removes the element at index $\pi$ from $\sigma$.*

**Definition 2.10** (Discard on context grade vectors). *The operation $\Delta\backslash\pi$ is $\sigma\backslash\pi$ for each $\sigma$ in $\Delta$.*

**Definition 2.11** (Splash multiplication of grade vectors). *The operation $\sigma_1 * \sigma_2$ scales $\sigma_2$ by each element of $\sigma_1$ to produce a context grade vector.*

**Definition 2.12** (Addition on grade vectors). *The operation $\sigma_1 + \sigma_2$ combines the two vectors element-wise using the semiring addition, right-padding the shorter vector with $0$, to ensure correct sizing.*

**Definition 2.13** (Addition on context grade vectors). *The operation $\Delta_1 + \Delta_2$ is pointwise addition of the elements of $\Delta_1$ and $\Delta_2$ in 'as much as possible.' I.e., if the corresponding vectors at a given index are of different sizes, then the shorter vector is treated as if it were right-padded with zeroes in the addition.*

# 3 List of results

**Lemma 3.1** (Strengthening). *If $(\Delta, \sigma_1, \Delta' \mid \sigma_2, s, \sigma_2' \mid \sigma_3, r, \sigma_3') \odot \Gamma, x : A, \Gamma' \vdash \mathcal{J}$, $x \notin \mathsf{FV}(\mathcal{J})$, and $x \notin \mathsf{FV}(\Gamma')$, with $|\Delta| = |\sigma_2| = |\sigma_3| = |\Gamma|$, then $(\Delta, (\Delta'\backslash\pi) \mid \sigma_2, \sigma_2' \mid \sigma_3, \sigma_3') \odot \Gamma, \Gamma' \vdash \mathcal{J}$ where $\pi = |\Gamma|$.*

**Lemma 3.2** (Judgments determine vector sizing). *If $(\Delta \mid \sigma_1 \mid \sigma_2) \odot \Gamma \vdash \mathcal{J}$, then $|\sigma_1| = |\sigma_2| = |\Delta| = |\Gamma|$ and for each element $\sigma$ of $\Delta$, the size of $\sigma$ is the same as its index.*

**Lemma 3.3** (Judgmental contexts are well-formed). *If $(\Delta \mid \sigma_1 \mid \sigma_2) \odot \Gamma \vdash \mathcal{J}$, then $\Delta \odot \Gamma \vdash$.*

**Lemma 3.4** (Subcontext well-formedness). *If $|\Delta_1| = |\Gamma_1|$ and $\Delta_1, \Delta_2 \odot \Gamma_1, \Gamma_2 \vdash$, then $\Delta_1 \odot \Gamma_1 \vdash$.*

**Lemma 3.5** (Typing an assumption in a judgmental context). *If $(\Delta_1, \sigma, \Delta_2 \mid \sigma_1 \mid \sigma_2) \odot \Gamma_1, x : A, \Gamma_2 \vdash \mathcal{J}$ with $|\Delta_1| = |\Gamma_1|$, then $(\Delta_1 \mid \sigma \mid \mathbf{0}) \odot \Gamma_1 \vdash A : \mathsf{Type}_l$, for some level $l$.*

**Lemma 3.6** (Typing the type of a term). *Given $(\Delta \mid \sigma_1 \mid \sigma_2) \odot \Gamma \vdash t : A$, then there exists a level $l$ such that $(\Delta \mid \sigma_2 \mid \mathbf{0}) \odot \Gamma \vdash A : \mathsf{Type}_l$.*

**Lemma 3.7** (Subtyping in context). *If $(\Delta_1, \sigma_1, \Delta_2 \mid \sigma_2, s, \sigma_3 \mid \sigma_4, r, \sigma_5) \odot \Gamma_1, x : A, \Gamma_2 \vdash \mathcal{J}$ and $(\Delta_1 \mid \sigma_1) \odot \Gamma \vdash A' \leq A$, then $(\Delta_1, \sigma_1, \Delta_2 \mid \sigma_2, s, \sigma_3 \mid \sigma_4, r, \sigma_5) \odot \Gamma_1, x : A', \Gamma_2 \vdash \mathcal{J}$.*

## 3.1 Inversions

**Lemma 3.8** (Inversion on arrow typing). *If $(\Delta \mid \sigma_1 \mid \sigma_2) \odot \Gamma \vdash (x :_{(s,r)} A) \to B : C$ then there exist grade vectors $\hat{\sigma_1}$, $\hat{\sigma_1}'$, and levels $l$ and $l'$, such that $(\Delta \mid \hat{\sigma_1} \mid \mathbf{0}) \odot \Gamma \vdash A : \mathsf{Type}_{l'}$ and $(\Delta, \hat{\sigma_1} \mid \hat{\sigma_1}', r \mid \mathbf{0}) \odot \Gamma, x : A \vdash B : \mathsf{Type}_{l''}$, and $\hat{\sigma_1} + \hat{\sigma_1}' = \sigma_1$.*

**Lemma 3.9** (Inversion on tensor typing). *If $(\Delta \mid \sigma_1 \mid \sigma_2) \odot \Gamma \vdash (x :_r A) \otimes B : C$ then there exist grade vectors $\hat{\sigma_1}$, $\hat{\sigma_1}'$, and levels $l$ and $l'$, such that $(\Delta \mid \hat{\sigma_1} \mid \mathbf{0}) \odot \Gamma \vdash A : \mathsf{Type}_{l'}$ and $(\Delta, \hat{\sigma_1} \mid \hat{\sigma_1}', r \mid \mathbf{0}) \odot \Gamma, x : A \vdash B : \mathsf{Type}_{l''}$, and $\hat{\sigma_1} + \hat{\sigma_1}' = \sigma_1$.*

**Lemma 3.10** (Inversion on box typing). *If $(\Delta \mid \sigma_1 \mid \sigma_2) \odot \Gamma \vdash \Box_s A : B$, then $(\Delta \mid \sigma_1 \mid \sigma_2) \odot \Gamma \vdash A : \mathsf{Type}_l$, for some level $l$.*

**Lemma 3.11** (Function inversion). *If $(\Delta \mid \sigma_2 \mid \sigma_1 + \sigma_3) \odot \Gamma \vdash \lambda x.t : C$, $(\Delta \mid \sigma_1 + \sigma_3) \odot \Gamma \vdash C \leq (x :_{(s,r)} A) \to B$, and $(\Delta, \sigma_1 \mid \sigma_3, r \mid \mathbf{0}) \odot \Gamma, x : A \vdash B : \mathsf{Type}_l$, then $(\Delta, \sigma_1 \mid \sigma_2, s \mid \sigma_3, r) \odot \Gamma, x : A \vdash t : B$.*

**Lemma 3.12** (Pair inversion). *If $(\Delta \mid \sigma_1 \mid \sigma_2 + \sigma_3) \odot \Gamma \vdash (t_1, t_2) : (x :_r A) \otimes B$ and $(\Delta, \sigma_2 \mid \sigma_3, r \mid \mathbf{0}) \odot \Gamma, x : A \vdash B : \mathsf{Type}_l$, then there exist grade vectors $\hat{\sigma_1}$ and $\hat{\sigma_1}'$ such that $(\Delta \mid \hat{\sigma_1} \mid \sigma_2) \odot \Gamma \vdash t_1 : A$, and $(\Delta \mid \hat{\sigma_1}' \mid \sigma_3 + r * \hat{\sigma_1}) \odot \Gamma \vdash t_2 : [t_1/x]B$, with $\hat{\sigma_1} + \hat{\sigma_1}' = \sigma_1$.*

**Lemma 3.13** (Box inversion). *If $(\Delta \mid \sigma_1 \mid \sigma_2) \odot \Gamma \vdash \Box t : B$ and $(\Delta \mid \sigma_2) \odot \Gamma \vdash B \leq \Box_s A$, then $(\Delta \mid \hat{\sigma_1} \mid \sigma_2) \odot \Gamma \vdash t : A$ for some $\hat{\sigma_1}$ such that $s * \hat{\sigma_1} = \sigma_1$.*

**Lemma 3.14** (Arrow subtyping inversion). *If $(\Delta \mid \sigma_1 + \sigma_2 \mid \mathbf{0}) \odot \Gamma \vdash \mathcal{J}_1$ where $\mathcal{J}_1$ is $(x :_{(s,r)} A) \to B \leq C$ or $(x :_{(s,r)} A) \to B = C : D$ and $(\Delta \mid \sigma_1 + \sigma_2 \mid \mathbf{0}) \odot \Gamma \vdash \mathcal{J}_2$, where $\mathcal{J}_2$ is respectively $C \leq (y :_{(s',r')} A') \to B'$ or $C = (y :_{(s',r')} A') \to B' : E$ (with $(\Delta \mid \mathbf{0}) \odot \Gamma \vdash D \leq \mathsf{Type}_l$ and $(\Delta \mid \mathbf{0}) \odot \Gamma \vdash E \leq \mathsf{Type}_l$), with $(\Delta, \sigma_1 \mid \sigma_3, r \mid \mathbf{0}) \odot \Gamma, y : A' \vdash B' : \mathsf{Type}_{l'}$, then $x = y$, $s = s'$, $r = r'$, and respectively, based on $\mathcal{J}_1$, $(\Delta \mid \sigma_1) \odot \Gamma \vdash A' \leq A$, and $(\Delta, \sigma_1 \mid \sigma_3, r) \odot \Gamma, x : A' \vdash B \leq B'$, or $(\Delta \mid \sigma_1 \mid \mathbf{0}) \odot \Gamma \vdash A' = A : \mathsf{Type}_{l''}$, and $(\Delta, \sigma_1 \mid \sigma_3, r \mid \mathbf{0}) \odot \Gamma, x : A' \vdash B = B' : \mathsf{Type}_{l'''}$.*

**Lemma 3.15** (Tensor subtyping inversion). *If $(\Delta \mid \sigma_1 + \sigma_2 \mid \mathbf{0}) \odot \Gamma \vdash \mathcal{J}_1$ where $\mathcal{J}_1$ is $(x :_r A) \otimes B \leq C$ or $(x :_r A) \otimes B = C : D$ and $(\Delta \mid \sigma_1 + \sigma_2 \mid \mathbf{0}) \odot \Gamma \vdash \mathcal{J}_2$, where $\mathcal{J}_2$ is respectively $C \leq (y :_{r'} A') \otimes B'$ or $C = (y :_{r'} A') \otimes B' : E$ (with $(\Delta \mid \mathbf{0}) \odot \Gamma \vdash D \leq \mathsf{Type}_l$ and $(\Delta \mid \mathbf{0}) \odot \Gamma \vdash E \leq \mathsf{Type}_l$), with $(\Delta, \sigma_1 \mid \sigma_3, r \mid \mathbf{0}) \odot \Gamma, y : A' \vdash B' : \mathsf{Type}_{l'}$, then $x = y$, $s = s'$, $r = r'$, $(\Delta \mid \sigma_1 \mid \mathbf{0}) \odot \Gamma \vdash A = A' : \mathsf{Type}_{l''''}$, and respectively, based on $\mathcal{J}_1$, $(\Delta, \sigma_1 \mid \sigma_3, r) \odot \Gamma, x : A' \vdash B \leq B'$, or $(\Delta, \sigma_1 \mid \sigma_3, r \mid \mathbf{0}) \odot \Gamma, x : A' \vdash B = B' : \mathsf{Type}_{l'''}$.*

**Lemma 3.16** (Box subtyping inversion). *If $(\Delta \mid \sigma \mid \mathbf{0}) \odot \Gamma \vdash \mathcal{J}_1$ where $\mathcal{J}_1$ is $\Box_s A \leq B$ or $\Box_s A = B : D$, and $(\Delta \mid \sigma \mid \mathbf{0}) \odot \Gamma \vdash \mathcal{J}_2$, where $\mathcal{J}_2$ is respectively $B \leq \Box_{s'} A'$ or $\Box_{s'} A' = B : E$ (where $(\Delta \mid \mathbf{0}) \odot \Gamma \vdash D \leq \mathsf{Type}_l$ and $(\Delta \mid \mathbf{0}) \odot \Gamma \vdash E \leq \mathsf{Type}_l$), then $s = s'$, and respectively, based on $\mathcal{J}_1$, $(\Delta \mid \sigma) \odot \Gamma \vdash A \leq A'$, or $(\Delta \mid \sigma \mid \mathbf{0}) \odot \Gamma \vdash A = A' : \mathsf{Type}_{l'}$.*

5

## 3.2 Vector manipulation

**Lemma 3.17** (Factoring choose and discard). *If $|\sigma_1| = |\sigma_3| = \pi$, then $(\Delta_1, (\sigma_1, r, \sigma_2), \Delta_2) \backslash \pi + ((\Delta_1, (\sigma_1, r, \sigma_2), \Delta_2)/\pi) * \sigma_3 = (\Delta_1 \backslash \pi + \Delta_1/\pi * \sigma_3), ((\sigma_1 + r * \sigma_3), \sigma_2), (\Delta_2 \backslash \pi + \Delta_2/\pi * \sigma_3)$.*

**Lemma 3.18** (Factoring vector addition). *If $\hat{\sigma_1} + \hat{\sigma_1}' = \sigma_1$; $\hat{s} + \hat{s}' = s$; and $\hat{\sigma_3} + \hat{\sigma_3}' = \sigma_3$, then $((\hat{\sigma_1} + \hat{s} * \sigma_2), \hat{\sigma_3}) + ((\hat{\sigma_1}' + \hat{s}' * \sigma_2), \hat{\sigma_3}') = (\sigma_1 + s * \sigma_2), \sigma_3$ for all $\sigma_2$.*

**Lemma 3.19** (Addition across same-sized components). *If $|\sigma_1| = |\sigma_3|$ then $\sigma_1, \sigma_2 + \sigma_3 = (\sigma_1 + \sigma_3), \sigma_2$.*

**Lemma 3.20** (Vector addition across components). *If $|\sigma_1| = |\sigma_4|$ and $|\sigma_2| = |\sigma_5|$, then $\sigma_1, \sigma_2, \sigma_3 + \sigma_4, \sigma_5, \sigma_6 = (\sigma_1 + \sigma_4), (\sigma_2 + \sigma_5), (\sigma_3 + \sigma_6)$.*

**Lemma 3.21** (Moving an exchange). *We have $\mathsf{mv}((\pi_1 + \pi_2); \pi_3; (\mathsf{exch}((\pi_1 + \pi_2); \Delta))) = \mathsf{mv}((\pi_1 + \pi_2 + 1); \pi_3; \Delta)$.*

## 3.3 Meta properties

### 3.3.1 Contraction

**Lemma 3.22** (Contraction). *The following rule is admissible:*

$$\frac{\begin{pmatrix}\Delta_1,\sigma_1,(\sigma_1,0),\Delta_2\\ \sigma_2,s_1,s_2,\sigma_3\\ \sigma_4,r_1,r_2,\sigma_5\end{pmatrix} \odot \Gamma_1, x : A, y : A, \Gamma_2 \vdash \mathcal{J} \quad |\Delta_1| = |\sigma_2| = |\sigma_4| = |\Gamma_1|}{\begin{pmatrix}\Delta_1,\sigma_1,\mathsf{contr}(|\Delta_1|;\Delta_2)\\ \sigma_2,(s_1+s_2),\sigma_3\\ \sigma_4,(r_1+r_2),\sigma_5\end{pmatrix} \odot \Gamma_1, z : A, [z,z/x,y]\Gamma_2 \vdash [z,z/x,y]\mathcal{J}} \text{Cont}$$

### 3.3.2 Exchange

**Lemma 3.23** (Exchange). *The following rule is admissible:*

$$\frac{\begin{pmatrix}\Delta_1,\sigma_1,(\sigma_2,0),\Delta_2\\ \sigma_3,s_1,s_2,\sigma_4\\ \sigma_5,r_1,r_2,\sigma_6\end{pmatrix} \odot \Gamma_1, x : A, y : B, \Gamma_2 \vdash \mathcal{J} \quad x \notin \mathsf{FV}(B) \quad |\Delta_1| = |\sigma_3| = |\sigma_5| = |\Gamma_1|}{\begin{pmatrix}\Delta_1,\sigma_2,(\sigma_1,0),\mathsf{exch}(|\Delta_1|;\Delta_2)\\ \sigma_3,s_2,s_1,\sigma_4\\ \sigma_5,r_2,r_1,\sigma_6\end{pmatrix} \odot \Gamma_1, y : B, x : A, \Gamma_2 \vdash \mathcal{J}} \text{Exchange}$$

**Corollary 3.23.1** (Exchange (general)). *As a corollary to Lemma 3.23, the following rule is admissible:*

$$\frac{\begin{pmatrix}\Delta_1,\Delta_2,(\sigma_1,0^{|\Delta_2|}),\Delta_3\\ \sigma_2,\sigma_3,s,\sigma_4\\ \sigma_5,\sigma_6,r,\sigma_7\end{pmatrix} \odot \Gamma_1, \Gamma_2, x : A, \Gamma_3 \vdash \mathcal{J} \quad \mathsf{Dom}(\Gamma_2) \cap \mathsf{FV}(A) = \emptyset \quad |\sigma_1| = |\Delta_1| = |\sigma_2| = |\sigma_5| = |\Gamma_1|}{\begin{pmatrix}\Delta_1,\sigma_1,\mathsf{ins}(|\Delta_1|;0;\Delta_2),\mathsf{mv}(|\Delta_1,\Delta_2|;|\Delta_1|;\Delta_3)\\ \sigma_2,s,\sigma_3,\sigma_4\\ \sigma_5,r,\sigma_6,\sigma_7\end{pmatrix} \odot \Gamma_1, x : A, \Gamma_2, \Gamma_3 \vdash \mathcal{J}} \text{ExchangeGen}$$

*Where $\mathsf{Dom}(\Gamma)$ is the domain of $\Gamma$.*

**Corollary 3.23.2** (Exchange from end). *As a corollary to Corollary 3.23.1, the following rule is admissible:*

$$\frac{\begin{pmatrix}\Delta_1,\Delta_2,(\sigma_1,0^{|\Delta_2|})\\ \sigma_2,\sigma_3,s\\ \sigma_4,\sigma_5,r\end{pmatrix} \odot \Gamma_1, \Gamma_2, x : A \vdash \mathcal{J} \quad \mathsf{Dom}(\Gamma_2) \cap \mathsf{FV}(A) = \emptyset \quad |\Delta_1| = |\sigma_2| = |\sigma_4| = |\Gamma_1|}{\begin{pmatrix}\Delta_1,\sigma_1,\mathsf{ins}(|\Delta_1|;0;\Delta_2)\\ \sigma_2,s,\sigma_3\\ \sigma_4,r,\sigma_5\end{pmatrix} \odot \Gamma_1, x : A, \Gamma_2 \vdash \mathcal{J}} \text{ExchangeEnd}$$

### 3.3.3 Weakening

**Lemma 3.24** (Weakening). *The following rule is admissible:*

$$\frac{(\Delta \mid \sigma_1 \mid \sigma_2) \odot \Gamma \vdash \mathcal{J} \quad (\Delta \mid \sigma_3 \mid \boldsymbol{0}) \odot \Gamma \vdash A : \mathsf{Type}_l}{(\Delta, \sigma_3 \mid \sigma_1, 0 \mid \sigma_2, 0) \odot \Gamma, x : A \vdash \mathcal{J}} \text{ W{\scriptsize EAK}}$$

*Where $\mathcal{J}$ is typing, equality, or subtyping.*

**Lemma 3.25** (Weakening for well-formed contexts). *The following rule is admissible:*

$$\frac{\Delta_1, \Delta_2 \odot \Gamma_1, \Gamma_2 \vdash \quad (\Delta_1 \mid \sigma_1 \mid \sigma_2) \odot \Gamma_1 \vdash t : A}{(\Delta_1, \Delta_2 \mid \sigma_1, \boldsymbol{0} \mid \sigma_2, \boldsymbol{0}) \odot \Gamma_1, \Gamma_2 \vdash t : A} \text{ W{\scriptsize EAK}WF}$$

**Lemma 3.26** (Weakening (general)). *The following rule is admissible (for $\mathcal{J}$ is typing, equality, or subtyping):*

$$\frac{(\Delta_1, \Delta_2 \mid \sigma_1, \sigma_1' \mid \sigma_2, \sigma_2') \odot \Gamma_1, \Gamma_2 \vdash \mathcal{J} \quad (\Delta_1 \mid \sigma_3 \mid \boldsymbol{0}) \odot \Gamma_1 \vdash A : \mathsf{Type}_l \quad |\sigma_1| = |\sigma_2| = |\Gamma_1|}{(\Delta_1, \sigma_3, \mathsf{ins}(|\Delta_1|; 0; \Delta_2) \mid \sigma_1, 0, \sigma_1' \mid \sigma_2, 0, \sigma_2') \odot \Gamma_1, x : A, \Gamma_2 \vdash \mathcal{J}} \text{ W{\scriptsize EAK}G{\scriptsize EN}}$$

### 3.3.4 Substitution

**Lemma 3.27** (Substitution for judgments). *If the following premises hold:*

1. $(\Delta \mid \sigma_2 \mid \sigma_1) \odot \Gamma_1 \vdash t : A$
2. $(\Delta, \sigma_1, \Delta' \mid \sigma_3, s, \sigma_4 \mid \sigma_5, r, \sigma_6) \odot \Gamma_1, x : A, \Gamma_2 \vdash \mathcal{J}$
3. $|\sigma_3| = |\sigma_5| = |\Gamma_1|$

*Then* $\begin{pmatrix} \Delta, (\Delta' \backslash \pi + (\Delta'/\pi) * \sigma_2) \\ (\sigma_3 + s*\sigma_2), \sigma_4 \\ (\sigma_5 + r*\sigma_2), \sigma_6 \end{pmatrix} \odot \Gamma_1, [t/x]\Gamma_2 \vdash [t/x]\mathcal{J}$, *where* $\pi = |\Gamma_1|$.

**Lemma 3.28** (Equality through substitution). *If the following premises hold:*

1. $(\Delta_1 \mid \sigma_1 \mid \sigma_2) \odot \Gamma_1 \vdash t_1 = t_2 : A$
2. $(\Delta_1, \sigma_2, \Delta_2 \mid \sigma_3, s, \sigma_4 \mid \sigma_5, r, \sigma_6) \odot \Gamma_1, x : A, \Gamma_2 \vdash t_3 : B$
3. $|\sigma_3| = |\sigma_5| = |\Gamma_1|$

*Then* $\begin{pmatrix} \Delta_1, (\Delta_2 \backslash \pi + (\Delta_2/\pi) * \sigma_1) \\ (\sigma_3 + s*\sigma_1), \sigma_4 \\ (\sigma_5 + r*\sigma_1), \sigma_6 \end{pmatrix} \odot \Gamma_1, [t_1/x]\Gamma_2 \vdash [t_1/x]t_3 = [t_2/x]t_3 : [t_1/x]B$, *where* $\pi = |\Gamma_1|$.

## 3.4 Properties of operations

**Lemma 3.29** (Properties of insertion). *The following properties hold for any valid insertion (i.e., where for all $i < |\Delta|$, $\pi \leq |\Delta[i]|$, for context grade vectors; and $\pi \leq |\sigma|$, for grade vectors):*

1. *(insPreservesSize)* $|\mathsf{ins}(\pi; R; \Delta)| = |\Delta|$;
2. *(insIncSizes)* if $\pi \leq |\Delta[i]|$, then $|(\mathsf{ins}(\pi; R; \Delta))[i]| = |\Delta[i]| + 1$;
3. *(insIncSizeGV)* $|\mathsf{ins}(\pi; R; \sigma)| = |\sigma| + 1$;
4. *(insCVthenGV)* $\mathsf{ins}(\pi; R; \Delta), \mathsf{ins}(\pi; R; \sigma) = \mathsf{ins}(\pi; R; (\Delta, \sigma))$

## 3.5 Properties of equality

**Lemma 3.30** (Equality is an equivalence relation). *For all, we have:*

- *(reflexivity) if* $(\Delta \mid \sigma_1 \mid \sigma_2) \odot \Gamma \vdash t : A$ *then* $(\Delta \mid \sigma_1 \mid \sigma_2) \odot \Gamma \vdash t = t : A$;

- *(transitivity) if* $(\Delta \mid \sigma_1 \mid \sigma_2) \odot \Gamma \vdash t_1 = t_2 : A$ *and* $(\Delta \mid \sigma_1 \mid \sigma_2) \odot \Gamma \vdash t_2 = t_3 : A$, *then* $(\Delta \mid \sigma_1 \mid \sigma_2) \odot \Gamma \vdash t_1 = t_3 : A$;

- *(symmetry) if* $(\Delta \mid \sigma_1 \mid \sigma_2) \odot \Gamma \vdash t_1 = t_2 : A$, *then* $(\Delta \mid \sigma_1 \mid \sigma_2) \odot \Gamma \vdash t_2 = t_1 : A$

**Lemma 3.31** (Reduction implies equality). *If* $(\Delta \mid \sigma_1 \mid \sigma_2) \odot \Gamma \vdash t_1 : A$ *and* $t_1 \rightsquigarrow t_2$, *then* $(\Delta \mid \sigma_1 \mid \sigma_2) \odot \Gamma \vdash t_1 = t_2 : A$.

**Lemma 3.32** (Equality inversion). *If* $(\Delta \mid \sigma_1 \mid \sigma_2) \odot \Gamma \vdash t_1 = t_2 : A$, *then* $(\Delta \mid \sigma_1 \mid \sigma_2) \odot \Gamma \vdash t_1 : A$ *and* $(\Delta \mid \sigma_1 \mid \sigma_2) \odot \Gamma \vdash t_2 : A$.

**Lemma 3.33** (Deriving judgments under equal contexts). *If* $(\Delta_1, \sigma_1, \Delta_2 \mid \sigma_2, s, \sigma_3 \mid \sigma_4, r, \sigma_5) \odot \Gamma_1, x : A, \Gamma_2 \vdash \mathcal{J}$ *and* $(\Delta_1 \mid \sigma_1 \mid \mathbf{0}) \odot \Gamma \vdash A = A' : \mathsf{Type}_l$, *then* $(\Delta_1, \sigma_1, \Delta_2 \mid \sigma_2, s, \sigma_3 \mid \sigma_4, r, \sigma_5) \odot \Gamma_1, x : A', \Gamma_2 \vdash \mathcal{J}$.

## 3.6 Properties of subtyping

**Lemma 3.34** (Subtyping inversion to typing). *If* $(\Delta \mid \sigma) \odot \Gamma \vdash A \leq B$, *then* $(\Delta \mid \sigma \mid \mathbf{0}) \odot \Gamma \vdash A : \mathsf{Type}_l$ *and* $(\Delta \mid \sigma \mid \mathbf{0}) \odot \Gamma \vdash B : \mathsf{Type}_l$ *for some level $l$.*

## 3.7 Type preservation

**Lemma 3.35** (Type preservation). *If* $(\Delta \mid \sigma_1 \mid \sigma_2) \odot \Gamma \vdash t : A$ *and* $t \rightsquigarrow t'$, *then* $(\Delta \mid \sigma_1 \mid \sigma_2) \odot \Gamma \vdash t' : A$.

# 4 Strong Normalization

**Definition 4.1.** *Typing can be broken up into the following stages:*

$$
\begin{aligned}
\mathsf{Kind} &:= \{A \mid \exists \Delta, \sigma_1, \Gamma. (\Delta \mid \sigma_1 \mid \mathbf{0}) \odot \Gamma \vdash A : \mathsf{Type}_1\} \\
\mathsf{Type} &:= \{A \mid \exists \Delta, \sigma_1, \Gamma. (\Delta \mid \sigma_1 \mid \mathbf{0}) \odot \Gamma \vdash A : \mathsf{Type}_0\} \\
\mathsf{Const} &:= \{t \mid \exists \Delta, \sigma_1, \sigma_2, \Gamma, A. (\Delta \mid \sigma_1 \mid \sigma_2) \odot \Gamma \vdash t : A \wedge \\
&\qquad (\Delta \mid \sigma_2 \mid \mathbf{0}) \odot \Gamma \vdash A : \mathsf{Type}_1\} \\
\mathsf{Term} &:= \{t \mid \exists \Delta, \sigma_1, \sigma_2, \Gamma, A. (\Delta \mid \sigma_1 \mid \sigma_2) \odot \Gamma \vdash t : A \wedge \\
&\qquad (\Delta \mid \sigma_2 \mid \mathbf{0}) \odot \Gamma \vdash A : \mathsf{Type}_0\}
\end{aligned}
$$

**Lemma 4.2** (Classification). *It is the case that* $\mathsf{Kind} \cap \mathsf{Type} = \emptyset$ *and* $\mathsf{Const} \cap \mathsf{Term} = \emptyset$.

**Definition 4.3.** *The set of base terms* $\mathcal{B}$ *is defined by:*

1. *For $x$ a variable, $x \in \mathcal{B}$,*

2. $\mathsf{Type}_0, \mathsf{Type}_1 \in \mathcal{B}$,

3. *If $t_1 \in \mathcal{B}$ and $t_2 \in \mathsf{SN}$, then $(t_1 \; t_2) \in \mathcal{B}$,*

4. *If $t_2 \in \mathcal{B}$ and $t_1 \in \mathsf{SN}$, then $(\mathsf{let}\ (x, y) = t_1\ \mathsf{in}\ t_2) \in \mathcal{B}$,*

5. *If $t_2 \in \mathcal{B}$ and $t_1 \in \mathsf{SN}$, then $(\mathsf{let}\ \square x = t_1\ \mathsf{in}\ t_2) \in \mathcal{B}$,*

6. If $A, B \in \mathsf{SN}$, then $((x :_{(r,s)} A) \to B) \in \mathcal{B}$ for any $r, s \in R$,

7. If $A, B \in \mathsf{SN}$, then $((x :_r A) \otimes B) \in \mathcal{B}$ for any $r \in R$,

8. If $A \in \mathsf{SN}$, then $(\Box_r A) \in \mathcal{B}$ for any $r \in R$.

**Definition 4.4.** *The* key redex *of a term is defined by:*

1. *If $t$ is a redex, then $t$ is its own key redex,*

2. *If $t_1$ has key redex $t$, then $(t_1\, t_2)$ has key redex $t$,*

3. *If $t_1$ has key redex $t$, then $(\mathsf{let}\,(x, y) = t_1\,\mathsf{in}\,t_2)$ has key redex $t$,*

4. *If $t_1$ has key redex $t$, then $(\mathsf{let}\,\Box x = t_1\,\mathsf{in}\,t_2)$ has key redex $t$.*

*The term obtained from $t$ by contracting its key redex is denoted by $\mathsf{red}_\mathsf{k}\, t$.*

**Lemma 4.5.** *The following are both true:*

1. $\mathcal{B} \subseteq \mathsf{SN}$

2. *The key redex of a term is unique and a head redex.*

**Definition 4.6.** *A set of terms $X$ is* saturated *if:*

1. $X \subset \mathsf{SN}$,

2. $\mathcal{B} \subset X$,

3. *If $\mathsf{red}_\mathsf{k}\, t \in X$ and $t \in \mathsf{SN}$, then $t \in X$.*

*The collection of saturated sets is denoted by $\mathsf{SAT}$.*

**Lemma 4.7** (SN is saturated)**.** *Every saturated set is non-empty and $\mathsf{SN}$ is saturated.*

*Proof.* By definition. $\square$

**Definition 4.8.** *For $T \in \mathsf{Kind}$, the kind interpretation, $\mathcal{K}[\![T]\!]$, is defined inductively as follows:*

$$\mathcal{K}[\![\mathsf{Type}_0]\!] = \mathsf{SAT}$$
$$\mathcal{K}[\![(x :_{(r,s)} A) \to B]\!] = \{f \mid f : \mathcal{K}[\![A]\!] \to \mathcal{K}[\![B]\!]\}, \text{ if } A, B \in \mathsf{Kind}$$
$$\mathcal{K}[\![(x :_{(r,s)} A) \to B]\!] = \mathcal{K}[\![A]\!], \text{ if } A \in \mathsf{Kind}, B \in \mathsf{Type}$$
$$\mathcal{K}[\![(x :_{(r,s)} A) \to B]\!] = \mathcal{K}[\![B]\!], \text{ if } A \in \mathsf{Type}, B \in \mathsf{Kind}$$
$$\mathcal{K}[\![(x :_s A) \otimes B]\!] = \mathcal{K}[\![A]\!] \times \mathcal{K}[\![B]\!], \text{ if } A, B \in \mathsf{Kind}$$
$$\mathcal{K}[\![(x :_s A) \otimes B]\!] = \mathcal{K}[\![A]\!], \text{ if } A \in \mathsf{Kind}, B \in \mathsf{Type}$$
$$\mathcal{K}[\![(x :_s A) \otimes B]\!] = \mathcal{K}[\![B]\!], \text{ if } A \in \mathsf{Type}, B \in \mathsf{Kind}$$
$$\mathcal{K}[\![\Box_s A]\!] = \mathcal{K}[\![A]\!]$$

**Definition 4.9.** Type valuations, $\Delta \odot \Gamma \models \varepsilon$, are defined as follows:

$$\frac{}{\emptyset \odot \emptyset \models \emptyset}\,\text{Ep\_Empty} \qquad \frac{\Delta \odot \Gamma \models \varepsilon \quad X \in \mathcal{K}[\![A]\!] \quad (\Delta \mid \sigma_2 \mid \mathbf{0}) \odot \Gamma \vdash A : \mathsf{Type}_1}{(\Delta, \sigma_2) \odot (\Gamma, x : A) \models \varepsilon[x \mapsto X]}\,\text{Ep\_ExtTy}$$

$$\frac{\Delta \odot \Gamma \models \varepsilon \quad (\Delta \mid \sigma_2 \mid \mathbf{0}) \odot \Gamma \vdash A : \mathsf{Type}_0}{(\Delta, \sigma_2) \odot (\Gamma, x : A) \models \varepsilon}\,\text{Ep\_ExtTm}$$

**Definition 4.10.** *Given a type valuation $\Delta \odot \Gamma \models \varepsilon$ and a type $T \in (\mathsf{Kind} \cup \mathsf{Type} \cup \mathsf{Con})$ with $T$ typable in $\Delta \odot \Gamma$, we define the interpretation of types, $[\![T]\!]_\varepsilon$ inductively as follows:*

$$
\begin{aligned}
[\![\mathsf{Type}_1]\!]_\varepsilon &= \mathsf{SN} \\
[\![\mathsf{Type}_0]\!]_\varepsilon &= \lambda X \in \mathsf{SAT}.\mathsf{SN} \\
[\![x]\!]_\varepsilon &= \varepsilon\, x, \text{ if } x \in \mathsf{Con} \\
[\![(x :_{(r,s)} A) \to B]\!]_\varepsilon &= \lambda X \in \mathcal{K}[\![A]\!] \to \mathcal{K}[\![B]\!]. \bigcap_{Y \in \mathcal{K}[\![A]\!]} (([\![A]\!]_\varepsilon)(Y) \to ([\![B]\!]_{\varepsilon[x \mapsto Y]})(X(Y))), \\
&\quad \text{if } A, B \in \mathsf{Kind} \\
[\![(x :_{(r,s)} A) \to B]\!]_\varepsilon &= \bigcap_{Y \in \mathcal{K}[\![A]\!]} (([\![A]\!]_\varepsilon)(Y) \to ([\![B]\!]_{\varepsilon[x \mapsto Y]})), \\
&\quad \text{if } A \in \mathsf{Kind}, B \in \mathsf{Type} \\
[\![(x :_{(r,s)} A) \to B]\!]_\varepsilon &= \lambda X \in \mathcal{K}[\![B]\!]. [\![A]\!]_\varepsilon \to ([\![B]\!]_\varepsilon)(X), \text{ if } A \in \mathsf{Type}, B \in \mathsf{Kind} \\
[\![(x :_{(r,s)} A) \to B]\!]_\varepsilon &= [\![A]\!]_\varepsilon \to [\![B]\!]_\varepsilon, \text{ if } A, B \in \mathsf{Type} \\
[\![\lambda x : A.B]\!]_\varepsilon &= \lambda X \in \mathcal{K}[\![A]\!]. [\![B]\!]_{\varepsilon[x \mapsto X]}, \text{ if } A \in \mathsf{Kind}, B \in \mathsf{Con} \\
[\![\lambda x : A.B]\!]_\varepsilon &= [\![A]\!]_\varepsilon, \text{ for any } A \in \mathsf{Kind}, B \in \mathsf{Term} \\
[\![\lambda x : A.B]\!]_\varepsilon &= [\![B]\!]_\varepsilon, \text{ for any } A \in \mathsf{Type}, B \in \mathsf{Con} \\
[\![A\,B]\!]_\varepsilon &= [\![A]\!]_\varepsilon([\![B]\!]_\varepsilon), \text{ if } B \in \mathsf{Con} \\
[\![A\,t]\!]_\varepsilon &= [\![A]\!]_\varepsilon, \text{ if } t \in \mathsf{Term} \\
[\![(x :_s A) \otimes B]\!]_\varepsilon &= \lambda X \in \mathcal{K}[\![A]\!] \times \mathcal{K}[\![B]\!]. ([\![A]\!]_\varepsilon)(\pi_1 X) \times ([\![B]\!]_{\varepsilon[x \mapsto \pi_1 X]})(\pi_2 X), \\
&\quad \text{if } A, B \in \mathsf{Kind} \\
[\![(x :_s A) \otimes B]\!]_\varepsilon &= \lambda X \in \mathcal{K}[\![B]\!]. [\![A]\!]_\varepsilon \times ([\![B]\!]_\varepsilon)(X), \text{ if } A \in \mathsf{Type}, B \in \mathsf{Kind} \\
[\![(x :_s A) \otimes B]\!]_\varepsilon &= \lambda X \in \mathcal{K}[\![A]\!]. ([\![A]\!]_\varepsilon)(X) \times [\![B]\!]_{\varepsilon[x \mapsto X]}, \text{ if } A \in \mathsf{Kind}, B \in \mathsf{Type} \\
[\![(x :_s A) \otimes B]\!]_\varepsilon &= [\![A]\!]_\varepsilon \times [\![B]\!]_\varepsilon, \text{ if } A, B \in \mathsf{Type} \\
[\![(A, B)]\!]_\varepsilon &= ([\![A]\!]_\varepsilon, [\![B]\!]_\varepsilon) \\
[\![(A, t)]\!]_\varepsilon &= [\![A]\!]_\varepsilon, \text{ if } t \in \mathsf{Term} \\
[\![(t, B)]\!]_\varepsilon &= [\![B]\!]_\varepsilon, \text{ if } t \in \mathsf{Term} \\
[\![\mathsf{let}\,(x:A, y:B) = t\,\mathsf{in}\,C]\!]_\varepsilon &= [\![C]\!]_{\varepsilon[x \mapsto \pi_1 [\![t]\!]_\varepsilon, y \mapsto \pi_2 [\![t]\!]_\varepsilon]}, \text{ if } A, B \in \mathsf{Kind} \\
[\![\mathsf{let}\,(x:A, y:B) = t\,\mathsf{in}\,C]\!]_\varepsilon &= [\![C]\!]_{\varepsilon[x \mapsto [\![t]\!]_\varepsilon]}, \text{ if } A \in \mathsf{Kind}, B \in \mathsf{Type} \\
[\![\mathsf{let}\,(x:A, y:B) = t\,\mathsf{in}\,C]\!]_\varepsilon &= [\![C]\!]_{\varepsilon[y \mapsto [\![t]\!]_\varepsilon]}, \text{ if } A \in \mathsf{Type}, B \in \mathsf{Kind} \\
[\![\mathsf{let}\,(x:A, y:B) = t\,\mathsf{in}\,C]\!]_\varepsilon &= [\![C]\!]_\varepsilon, \text{ if } A, B \in \mathsf{Type} \\
[\![\Box_{(r,s)} A]\!]_\varepsilon &= [\![A]\!]_\varepsilon \\
[\![\Box A]\!]_\varepsilon &= [\![A]\!]_\varepsilon \\
[\![\mathsf{let}\,\Box(x:A) = t\,\mathsf{in}\,B]\!]_\varepsilon &= [\![B]\!]_{\varepsilon[x \mapsto [\![t]\!]_\varepsilon]}, \text{ if } A \in \mathsf{Kind} \\
[\![\mathsf{let}\,\Box(x:A) = t\,\mathsf{in}\,B]\!]_\varepsilon &= [\![B]\!]_\varepsilon, \text{ if } A \in \mathsf{Type}
\end{aligned}
$$

**Definition 4.11.** *Suppose $\Delta \odot \Gamma \models \varepsilon$. Then* Valid term valuations, *$\Delta \odot \Gamma \models_\varepsilon \rho$, are defined as follows:*

$$
\frac{}{\emptyset \odot \emptyset \models_\emptyset \emptyset}\text{ RHO\_EMPTY} \qquad \frac{t \in ([\![A]\!]_\varepsilon)(\varepsilon\,x) \quad \Delta \odot \Gamma \models_\varepsilon \rho \quad (\Delta \mid \sigma_2 \mid \mathbf{0}) \odot \Gamma \vdash A : \mathsf{Type}_1}{(\Delta, \sigma_2) \odot (\Gamma, x : A) \models_\varepsilon \rho[x \mapsto t]}\text{ RHO\_EXTTY}
$$

$$
\frac{t \in [\![A]\!]_\varepsilon \quad \Delta \odot \Gamma \models_\varepsilon \rho \quad (\Delta \mid \sigma_2 \mid \mathbf{0}) \odot \Gamma \vdash A : \mathsf{Type}_0}{(\Delta, \sigma_2) \odot (\Gamma, x : A) \models_\varepsilon \rho[x \mapsto t]}\text{ RHO\_EXTTM}
$$

**Definition 4.12.** *Suppose $\Delta \odot \Gamma \models_\varepsilon \rho$. Then the* interpretation of a term *$t$ typable in $\Delta \odot \Gamma$ is $(\!|t|\!)_\rho = \rho\,t$, but where all let-expressions are translated into substitutions, and all graded modalities are erased.*

**Definition 4.13.** *Suppose $(\Delta \mid \sigma_1 \mid \sigma_2) \odot \Gamma \vdash t : A$. Then* semantic typing, *$(\Delta \mid \sigma_1 \mid \sigma_2) \odot \Gamma \models t : A$, is defined as follows:*

1. *If $(\Delta \mid \sigma \mid \mathbf{0}) \odot \Gamma \vdash A : \mathsf{Type}_1$, then for every $\Delta \odot \Gamma \models_\varepsilon \rho$, $(\!|t|\!)_\rho \in [\![A]\!]_\varepsilon([\![t]\!]_\varepsilon)$.*

2. *If $(\Delta \mid \sigma \mid \mathbf{0}) \odot \Gamma \vdash A : \mathsf{Type}_0$, then for every $\Delta \odot \Gamma \models_\varepsilon \rho$, $(\!|t|\!)_\rho \in [\![A]\!]_\varepsilon$.*

**Lemma 4.14** (Substitution for Typing Interpretation)**.** *Suppose $\Delta \odot \Gamma \models \varepsilon$ and we have types $T_2 \in (\mathsf{Kind} \cup \mathsf{Con})$ and $T_1 \in \mathsf{Con}$ with $T_1$ and $T_2$ typable in $\Delta \odot \Gamma$, and a term $t \in \mathsf{Term}$ typable in $\Delta \odot \Gamma$. Then:*

1. $[\![T_2]\!]_{\varepsilon[x \mapsto [\![T_1]\!]_\varepsilon]} = [\![[T_1/x]T_2]\!]_\varepsilon$.

2. $[\![T_2]\!]_\varepsilon = [\![[t/x]T_2]\!]_\varepsilon$.

*Proof.* By straightforward induction on $T_2$ with the fact that substitutions disappear in the kind interpretation. □

**Lemma 4.15** (Equality of Interpretations). *Suppose $\Delta \odot \Gamma \models \varepsilon$. Then:*

1. *if $(\Delta \mid \sigma_1 \mid \sigma_2) \odot \Gamma \vdash T_1 = T_2 : \mathsf{Type}_1$, then $\mathcal{K}[\![T_1]\!] = \mathcal{K}[\![T_2]\!]$.*

2. *if $(\Delta \mid \sigma_1 \mid \sigma_2) \odot \Gamma \vdash T_1 = T_2 : \mathsf{Type}_l$, then $[\![T_1]\!]_\varepsilon = [\![T_2]\!]_\varepsilon$.*

*Proof.* Part one follows by induction on $(\Delta \mid \sigma_1 \mid \sigma_2) \odot \Gamma \vdash T_1 = T_2 : \mathsf{Type}_l$, and so does part two, but it also depends on part one and the previous lemma. □

**Lemma 4.16** (Interpretation Soundness). *Suppose $\Delta \odot \Gamma \models \varepsilon$ and $(\Delta \mid \sigma \mid \mathbf{0}) \odot \Gamma \vdash A : \mathsf{Type}_1$. Then:*

1. *If $(\Delta \mid \sigma_1 \mid \sigma_2) \odot \Gamma \vdash t : A$, then $[\![t]\!]_\varepsilon \in \mathcal{K}[\![A]\!]$*

2. *$[\![A]\!]_\varepsilon \in \mathcal{K}[\![A]\!] \to \mathsf{SAT}$*

*Proof.* This is a proof by simultaneous induction over $(\Delta \mid \sigma_1 \mid \sigma_2) \odot \Gamma \vdash t : A$ and $(\Delta \mid \sigma \mid \mathbf{0}) \odot \Gamma \vdash A : \mathsf{Type}_1$. We consider part 1 assuming 2, and vice versa.

Proof of part 1:

Case 1:
$$\frac{\Delta \odot \Gamma \vdash}{(\Delta \mid \mathbf{0} \mid \mathbf{0}) \odot \Gamma \vdash \mathsf{Type}_0 : \mathsf{Type}_1} \text{ T\_Type}$$

This case holds trivially, because $\mathsf{Type}_1$ cannot be of type $\mathsf{Type}_1$.

Case 2:
$$\frac{\Delta_1, \sigma, \Delta_2 \odot \Gamma_1, x : A, \Gamma_2 \vdash \quad |\Delta_1| = |\Gamma_1|}{(\Delta_1, \sigma, \Delta_2 \mid \mathbf{0}^{|\Delta_1|}, 1, \mathbf{0} \mid \sigma, 0, \mathbf{0}) \odot \Gamma_1, x : A, \Gamma_2 \vdash x : A} \text{ T\_Var}$$

In this case we have:
$$\Delta = (\Delta_1, \sigma, \Delta_2)$$
$$\sigma_1 = (\mathbf{0}^{|\Delta_1|}, 1, \mathbf{0})$$
$$\sigma_2 = (\sigma, 0, \mathbf{0})$$
$$t = x$$

Thus, we must show that:
$$[\![x]\!]_\varepsilon \in \mathcal{K}[\![A]\!]$$

We know by Definition C.10 and Definition C.9 that $[\![x]\!]_\varepsilon = \varepsilon\, x \in \mathcal{K}[\![A]\!]$; thus, we obtain our result.

Case 3:
$$\frac{(\Delta \mid \sigma_3 \mid \mathbf{0}) \odot \Gamma \vdash B_1 : \mathsf{Type}_{l_1} \quad (\Delta, \sigma_3 \mid \sigma_4, r \mid \mathbf{0}) \odot \Gamma, x : B_1 \vdash B_2 : \mathsf{Type}_{l_2}}{(\Delta \mid \sigma_3 + \sigma_4 \mid \mathbf{0}) \odot \Gamma \vdash (x :_{(s,r)} B_1) \to B_2 : \mathsf{Type}_{l_1 \sqcup l_2}} \text{ T\_Arrow}$$

In this case we know:
$$\sigma_1 = (\sigma_3 + \sigma_4)$$
$$\sigma_2 = \mathbf{0}$$
$$t = (x :_{(s,r)} B_1) \to B_2$$
$$A = \mathsf{Type}_{l_1 \sqcup l_2}$$

However, by assumption we know that $(\Delta \mid \mathbf{0} \mid \mathbf{0}) \odot \Gamma \vdash \mathsf{Type}_{l_1 \sqcup l_2} : \mathsf{Type}_1$, and hence, $\mathsf{Type}_{l_1 \sqcup l_2} = \mathsf{Type}_0$ which implies that $l_1 = l_2 = 0$. This all implies that $B_1, B_2 \in \mathsf{Type}$. Furthermore, we know that $(\Delta, \sigma_3) \odot (\Gamma, x : B_1) \models \varepsilon$, because $B_1 \in \mathsf{Type}$.

11

Thus, by Definition C.8 we must show that:
$$[\![(x :_{(s,r)} B_1) \to B_2]\!]_\varepsilon \in \mathcal{K}[\![\mathsf{Type}_0]\!] = \mathsf{SAT}$$

By Definition C.10 we know that:
$$[\![(x :_{(s,r)} B_1) \to B_2]\!]_\varepsilon = [\![B_1]\!]_\varepsilon \to [\![B_2]\!]_\varepsilon$$

By the IH:

IH(1): $[\![B_1]\!]_\varepsilon \in \mathcal{K}[\![\mathsf{Type}_0]\!] = \mathsf{SAT}$

IH(2): $[\![B_2]\!]_\varepsilon \in \mathcal{K}[\![\mathsf{Type}_0]\!] = \mathsf{SAT}$

Therefore, $[\![(x :_{(s,r)} B_1) \to B_2]\!]_\varepsilon = ([\![B_1]\!]_\varepsilon \to [\![B_2]\!]_\varepsilon) = (\mathsf{SAT} \to \mathsf{SAT}) \in \mathsf{SAT}$ which holds by the closer of $\mathsf{SAT}$ under function spaces.

Case 4:
$$\frac{(\Delta \mid \sigma_3 \mid \mathbf{0}) \odot \Gamma \vdash B_1 : \mathsf{Type}_{l_1} \quad (\Delta, \sigma_3 \mid \sigma_4, r \mid \mathbf{0}) \odot \Gamma, x : B_1 \vdash B_2 : \mathsf{Type}_{l_2}}{(\Delta \mid \sigma_3 + \sigma_4 \mid \mathbf{0}) \odot \Gamma \vdash (x :_r B_1) \otimes B_2 : \mathsf{Type}_{l_1 \sqcup l_2}} \text{ T\_Ten}$$

This case follows nearly exactly as the previous case, but ending with a cartesian product, $\mathsf{SAT} \times \mathsf{SAT}$, rather than the function space.

Case 5:
$$\frac{(\Delta, \sigma_3 \mid \sigma_4, s \mid \sigma_5, r) \odot \Gamma, x : B_1 \vdash t' : B_2}{(\Delta \mid \sigma_4 \mid \sigma_3 + \sigma_5) \odot \Gamma \vdash \lambda x : B_1.t' : (x :_{(s,r)} B_1) \to B_2} \text{ T\_Fun}$$

In this case we know:
$$\sigma_1 = \sigma_4$$
$$\sigma_2 = \sigma_3 + \sigma_5$$
$$t = \lambda x : B_1.t'$$
$$A = (x :_{(s,r)} B_1) \to B_2$$

We also know that $(\Delta \mid \sigma_3 + \sigma_5 \mid \mathbf{0}) \odot \Gamma \vdash (x :_{(s,r)} B_1) \to B_2 : \mathsf{Type}_1$ by assumption. This implies by inversion that $B_1, B_2 \in \mathsf{Kind}$, $B_1 \in \mathsf{Kind}$ and $B_2 \in \mathsf{Type}$, or $B_1 \in \mathsf{Type}$ and $B_2 \in \mathsf{Kind}$. We consider each case in turn:

Subcase 1: Suppose $B_1, B_2 \in \mathsf{Kind}$. Then we must show that:
$$[\![\lambda x : B_1.t']\!]_\varepsilon \in \mathcal{K}[\![(x :_{(s,r)} B_1) \to B_2]\!] = \{f \mid f : \mathcal{K}[\![B_1]\!] \to \mathcal{K}[\![B_2]\!]\}$$

By Definition C.10 we know that:
$$[\![\lambda x : B_1.t']\!]_\varepsilon = \lambda X \in \mathcal{K}[\![B_1]\!].[\![t']\!]_{\varepsilon[x \mapsto X]}$$

Suppose $X \in \mathcal{K}[\![B_1]\!]$. Then we will show that $[\![t']\!]_{\varepsilon[x \mapsto X]} \in \mathcal{K}[\![B_2]\!]$. We know by assumption that $\Delta \odot \Gamma \models \varepsilon$ and $(\Delta, \sigma_3 \mid \sigma_4, s \mid \sigma_5, r) \odot \Gamma, x : B_1 \vdash t' : B_2$. Hence, by Lemma 3.5, we know that $(\Delta \mid \sigma_3 \mid \mathbf{0}) \odot \Gamma \vdash B_1 : \mathsf{Type}_1$. Thus, we know that $(\Delta, \sigma_3) \odot (\Gamma, x : B_1) \models \varepsilon[x \mapsto X]$.

Therefore, by the IH:

IH(1): $[\![t']\!]_{\varepsilon[x \mapsto X]} \in \mathcal{K}[\![B_2]\!]$

which is what was to be shown.

12

Subcase 2: Suppose $B_1 \in$ Kind and $B_2 \in$ Type. Then we must show that:
$$[\![\lambda x : B_1.t']\!]_\varepsilon \in \mathcal{K}[\![(x :_{(s,r)} B_1) \to B_2]\!] = \mathcal{K}[\![B_1]\!]$$

By Definition C.10 we know that:
$$[\![\lambda x : B_1.t']\!]_\varepsilon = \mathcal{K}[\![B_1]\!]$$

which is what was to be shown.

Subcase 3: Suppose $B_1 \in$ Type and $B_2 \in$ Kind. Then we must show that:
$$[\![\lambda x : B_1.t']\!]_\varepsilon \in \mathcal{K}[\![(x :_{(s,r)} B_1) \to B_2]\!] = \mathcal{K}[\![B_2]\!]$$

By Definition C.10 we know that:
$$[\![\lambda x : B_1.t']\!]_\varepsilon = \mathcal{K}[\![B_2]\!]$$

which is what was to be shown.

Case 6:
$$\frac{(\Delta \mid \sigma_4 \mid \sigma_3 + \sigma_5) \odot \Gamma \vdash t_1 : (x :_{(s,r)} B_1) \to B_2 \quad (\Delta \mid \sigma_6 \mid \sigma_3) \odot \Gamma \vdash t_2 : B_1}{(\Delta \mid \sigma_4 + s * \sigma_6 \mid \sigma_5 + r * \sigma_6) \odot \Gamma \vdash t_1\, t_2 : [t_2/x]B_2} \text{ T\_App}$$

In this case we know that:
$$\sigma_1 = \sigma_4 + s * \sigma_6$$
$$\sigma_2 = \sigma_5 + r * \sigma_6$$
$$t = t_1\, t_2$$
$$A = [t_2/x]B_2$$

It suffices to show that:
$$[\![t_1\, t_2]\!]_\varepsilon \in \mathcal{K}[\![[t_2/x]B_2]\!] = \mathcal{K}[\![B_2]\!]$$

In this case we know by assumption that $[t_2/x]B_2 \in$ Kind which implies that $B_2 \in$ Kind, and this further implies that $((x :_{(s,r)} B_1) \to B_2) \in$ Kind. However, there are two cases for $B_1$, either $B_1 \in$ Kind or $B_1 \in$ Type. We consider each case in turn.

Subcase 1: Suppose $B_1 \in$ Kind. Then we know by assumption that $\Delta \odot \Gamma \models \varepsilon$. Thus, we apply the IH to obtain:

IH(1): $[\![t_1]\!]_\varepsilon \in \mathcal{K}[\![(x :_{(s,r)} B_1) \to B_2]\!] = \{f \mid f : \mathcal{K}[\![B_1]\!] \to \mathcal{K}[\![B_2]\!]\}$

IH(2): $[\![t_2]\!]_\varepsilon \in \mathcal{K}[\![B_1]\!]$

Using the IH's and Definition C.10 we know:
$$[\![t_1\, t_2]\!]_\varepsilon = [\![t_1]\!]_\varepsilon\, ([\![t_2]\!]_\varepsilon) \in \mathcal{K}[\![B_2]\!]$$

which was what was to be shown.

Subcase 2: Suppose $B_1 \in$ Type. Then we know by assumption that $\Delta \odot \Gamma \models \varepsilon$. Thus, we apply the IH to obtain:

IH(1): $[\![t_1]\!]_\varepsilon \in \mathcal{K}[\![(x :_{(s,r)} B_1) \to B_2]\!] = \mathcal{K}[\![B_2]\!]$

Using the IH's and Definition C.10 we know:
$$[\![t_1\, t_2]\!]_\varepsilon = [\![t_1]\!]_\varepsilon \in \mathcal{K}[\![B_2]\!]$$

which was what was to be shown.

Case 7:
$$\frac{\begin{array}{c}(\Delta, \sigma_3 \mid \sigma_5, r \mid \mathbf{0}) \odot \Gamma, x : B_1 \vdash B_2 : \mathsf{Type}_l \\ (\Delta \mid \sigma_4 \mid \sigma_3) \odot \Gamma \vdash t_1 : B_1 \\ (\Delta \mid \sigma_6 \mid \sigma_5 + r * \sigma_4) \odot \Gamma \vdash t_2 : [t_1/x]B_2\end{array}}{(\Delta \mid \sigma_4 + \sigma_6 \mid \sigma_3 + \sigma_5) \odot \Gamma \vdash (t_1, t_2) : (x :_r B_1) \otimes B_2} \text{ T\_PAIR}$$

In this case we know that:
$$\begin{aligned}\sigma_1 &= \sigma_4 + \sigma_6 \\ \sigma_2 &= \sigma_3 + \sigma_5 \\ t &= (t_1, t_2) \\ A &= (x :_r B_1) \otimes B_2\end{aligned}$$

It suffices to show that:
$$[\![(t_1, t_2)]\!]_\varepsilon \in \mathcal{K}[\![(x :_r B_1) \otimes B_2]\!]$$

By assumption we know that $((x :_r B_1) \otimes B_2) \in \mathsf{Kind}$. Thus, it must be the case that either $B_1, B_2 \in \mathsf{Kind}$, $B_1 \in \mathsf{Kind}$ and $B_2 \in \mathsf{Type}$, or $B_1 \in \mathsf{Type}$ and $B_2 \in \mathsf{Kind}$. We consider each of these cases in turn.

Subcase 1: Suppose $B_1, B_2 \in \mathsf{Kind}$ so $l = 1$. Then by Definition C.8 and Definition C.10 it suffices to show that:
$$\begin{aligned}[\![(t_1, t_2)]\!]_\varepsilon &= ([\![t_1]\!]_\varepsilon, [\![t_2]\!]_\varepsilon) \\ &\in \mathcal{K}[\![(x :_r B_1) \otimes B_2]\!] \\ &= \mathcal{K}[\![B_1]\!] \times \mathcal{K}[\![B_2]\!]\end{aligned}$$

We know by assumption that $\Delta \odot \Gamma \models \varepsilon$ and $(\Delta \mid \sigma_4 \mid \sigma_3) \odot \Gamma \vdash t_1 : B_1$. By kinding for typing (Lemma 3.6) we know that $(\Delta \mid \sigma_3 \mid \mathbf{0}) \odot \Gamma \vdash B_1 : \mathsf{Type}_1$, and by assumption we know that $(\Delta, \sigma_3 \mid \sigma_5, r \mid \mathbf{0}) \odot (\Gamma, x : B_1) \vdash B_2 : \mathsf{Type}_1$. By substitution for typing (Lemma 3.27), we know that $(\Delta \mid \sigma_5 + r * \sigma_4 \mid \mathbf{0}) \odot \Gamma \vdash [t_1/x]B_2 : \mathsf{Type}_1$. We can now apply the IH to the premises for $t_1$ and $t_2$ using the fact that we know $(\Delta \mid \sigma_3 \mid \mathbf{0}) \odot \Gamma \vdash B_1 : \mathsf{Type}_1$ and $(\Delta \mid \sigma_5 + r * \sigma_4 \mid \mathbf{0}) \odot \Gamma \vdash [t_1/x]B_2 : \mathsf{Type}_1$.

Thus, by the IH:

IH(1): $[\![t_1]\!]_\varepsilon \in \mathcal{K}[\![B_1]\!]$

IH(2): $[\![t_2]\!]_\varepsilon \in \mathcal{K}[\![[t_1/x]B_2]\!] = \mathcal{K}[\![B_2]\!]$

Therefore, $([\![t_1]\!]_\varepsilon, [\![t_2]\!]_\varepsilon) \in \mathcal{K}[\![B_1]\!] \times \mathcal{K}[\![B_2]\!]$.

Subcase 2: Suppose $B_1 \in \mathsf{Kind}$ and $B_2 \in \mathsf{Type}$. Then by Definition C.8 and Definition C.10 it suffices to show that:
$$\begin{aligned}[\![(t_1, t_2)]\!]_\varepsilon &= [\![t_1]\!]_\varepsilon \\ &\in \mathcal{K}[\![(x :_r B_1) \otimes B_2]\!] \\ &= \mathcal{K}[\![B_1]\!]\end{aligned}$$

We now by assumption that $\Delta \odot \Gamma \models \varepsilon$ and $(\Delta \mid \sigma_4 \mid \sigma_3) \odot \Gamma \vdash t_1 : B_1$. By kinding for typing (Lemma 3.6) we know that $(\Delta \mid \sigma_3 \mid \mathbf{0}) \odot \Gamma \vdash B_1 : \mathsf{Type}_1$. Thus, by the IH:

IH(1): $[\![t_1]\!]_\varepsilon \in \mathcal{K}[\![B_1]\!]$

which is what was to be shown.

Subcase 3: Suppose $B_1 \in \mathsf{Type}$ and $B_2 \in \mathsf{Kind}$. Then by Definition C.8 and Definition C.10 it suffices to show that:
$$\begin{aligned}[\![(t_1, t_2)]\!]_\varepsilon &= [\![t_2]\!]_\varepsilon \\ &\in \mathcal{K}[\![(x :_r B_1) \otimes B_2]\!] \\ &= \mathcal{K}[\![B_2]\!]\end{aligned}$$

We know by assumption that $\Delta \odot \Gamma \models \varepsilon$ $(\Delta \mid \sigma_4 \mid \sigma_3) \odot \Gamma \vdash t_1 : B_1$, and $(\Delta, \sigma_3 \mid \sigma_5, r \mid \mathbf{0}) \odot (\Gamma, x : B_1) \vdash B_2 : \mathsf{Type}_1$. By substitution for typing (Lemma 3.27), we know that

14

$(\Delta \mid \sigma_5 + r * \sigma_4 \mid \mathbf{0}) \odot \Gamma \vdash [t_1/x]B_2 : \mathsf{Type}_1$. We can now apply the IH to the premises for $t_2$ using the fact that we know $(\Delta \mid \sigma_5 + r * \sigma_4 \mid \mathbf{0}) \odot \Gamma \vdash [t_1/x]B_2 : \mathsf{Type}_1$.

Thus, by the IH:

IH(1): $[\![t_2]\!]_\varepsilon \in \mathcal{K}[\![[t_1/x]B_2]\!] = \mathcal{K}[\![B_2]\!]$

which is what was to be shown.

Case 8:

$$\frac{\begin{array}{c}(\Delta \mid \sigma_5 \mid \sigma_3 + \sigma_4) \odot \Gamma \vdash t_1 : (x :_r B_1) \otimes B_2 \\ (\Delta, (\sigma_3 + \sigma_4) \mid \sigma_7, r' \mid \mathbf{0}) \odot \Gamma, z : (x :_r B_1) \otimes B_2 \vdash C : \mathsf{Type}_1 \\ (\Delta, \sigma_3, (\sigma_4, r) \mid \sigma_6, s, s \mid \sigma_7, r', r') \odot \Gamma, x : B_1, y : B_2 \vdash t_2 : [(x,y)/z]C\end{array}}{(\Delta \mid \sigma_6 + s * \sigma_5 \mid \sigma_7 + r' * \sigma_5) \odot \Gamma \vdash \mathsf{let}\,(x : B_1, y : B_2) = t_1 \,\mathsf{in}\, t_2 : [t_1/z]C} \; \text{T\_TenCut}$$

In this case we know that:
$$\begin{aligned}\sigma_1 &= \sigma_6 + s * \sigma_5 \\ \sigma_2 &= \sigma_7 + r' * \sigma_5 \\ t &= (\mathsf{let}\,(x : B_1, y : B_2) = t_1 \,\mathsf{in}\, t_2) \\ A &= [t_1/z]C\end{aligned}$$

It suffices to show that:

$$[\![\mathsf{let}\,(x : B_1, y : B_2) = t_1 \,\mathsf{in}\, t_2]\!]_\varepsilon \in \mathcal{K}[\![[t_1/z]C]\!] = \mathcal{K}[\![C]\!]$$

Based on these assumptions, it must be the case that $C \in \mathsf{Kind}$ and $t_2 \in \mathsf{Type}$. First, we can conclude by kinding for typing that:

$$(\Delta, \sigma_3, (\sigma_4, r) \mid \sigma_{11} \mid \mathbf{0}) \odot (\Gamma, x : B_1, y : B_2) \vdash [(x,y)/z]C : \mathsf{Type}_1$$

for some vector $\sigma_1$. Then by well-formed contexts for typing we know that:

$$(\Delta, \sigma_3, (\sigma_4, r)) \odot (\Gamma, x : B_1, y : B_2) \vdash$$

This then implies that:

$$(\Delta \mid \sigma_3 \mid \mathbf{0}) \odot \Gamma \vdash B_1 : \mathsf{Type}_{l_1}, (\Delta, \sigma_3 \mid \sigma_4, r \mid \mathbf{0}) \odot (\Gamma, x : B_1) \vdash B_2 : \mathsf{Type}_{l_2}$$

We will use these to apply the IH and define the typing evaluation below.

We must now consider cases for when $B_1, B_2 \in \mathsf{Kind}$, $B_1 \in \mathsf{Kind}$ and $B_2 \in \mathsf{Type}$, $B_1 \in \mathsf{Type}$ and $B_2 \in \mathsf{Kind}$, and $B_1, B_2 \in \mathsf{Type}$. We consider each case in turn.

Subcase 1: Suppose $B_1, B_2 \in \mathsf{Kind}$. Then we know:

$$\begin{aligned}(\Delta \mid \sigma_3 \mid \mathbf{0}) \odot \Gamma &\vdash B_1 : \mathsf{Type}_1 \\ (\Delta, \sigma_3 \mid \sigma_4, r \mid \mathbf{0}) \odot (\Gamma, x : B_1) &\vdash B_2 : \mathsf{Type}_1 \\ (\Delta \mid \sigma_3 + \sigma_4 \mid \mathbf{0}) \odot \Gamma &\vdash (x :_r B_1) \otimes B_2 : \mathsf{Type}_1\end{aligned}$$

By Definition C.8 and Definition C.10 it suffices to show:

$$\begin{aligned}[\![\mathsf{let}\,(x : B_1, y : B_2) = t_1 \,\mathsf{in}\, t_2]\!]_\varepsilon &= [\![t_2]\!]_{\varepsilon[x \mapsto \pi_1\,[\![t_1]\!]_\varepsilon, y \mapsto \pi_2\,[\![t_1]\!]_\varepsilon]} \\ &\in \mathcal{K}[\![[t_1/z]C]\!] \\ &= \mathcal{K}[\![C]\!]\end{aligned}$$

By the applying the IH to the premise for $t_1$ and Definition C.8:

IH(1): $[\![t_1]\!]_\varepsilon \in \mathcal{K}[\![(x :_r B_1) \otimes B_2]\!] = \mathcal{K}[\![B_1]\!] \times \mathcal{K}[\![B_2]\!]$

15

Hence, $\pi_1 \llbracket t_1 \rrbracket_\varepsilon \in \mathcal{K}\llbracket B_1 \rrbracket$ and $\pi_2 \llbracket t_1 \rrbracket_\varepsilon \in \mathcal{K}\llbracket B_2 \rrbracket$. This along with the kinding judgments given above imply that

$$(\Delta, \sigma_3, (\sigma_4, r)) \odot (\Gamma, x : B_1, y : B_2) \models \varepsilon[x \mapsto \pi_1 \llbracket t_1 \rrbracket_\varepsilon, y \mapsto \pi_2 \llbracket t_1 \rrbracket_\varepsilon]$$

given the assumption $\Delta \odot \Gamma \models \varepsilon$. We also know:

$$(\Delta, \sigma_3, (\sigma_4, r) \mid \sigma_{12} \mid \mathbf{0}) \odot (\Gamma, x : B_1, y : B_2) \vdash [(x, y)/z]C : \mathsf{Type}_1$$

by applying kinding for typing to the premise for $t_2$. We now have everything we need to apply the IH to the premise for $t_2$:

IH(2): $\llbracket t_2 \rrbracket_{\varepsilon[x \mapsto \pi_1 \llbracket t_1 \rrbracket_\varepsilon, y \mapsto \pi_2 \llbracket t_1 \rrbracket_\varepsilon]} \in \mathcal{K}\llbracket [(x, y)/z]C \rrbracket = \mathcal{K}\llbracket C \rrbracket$

Now by Definition C.10 and the previous results:

$$\llbracket \mathsf{let}\ (x : B_1, y : B_2) = t_1\ \mathsf{in}\ t_2 \rrbracket_\varepsilon$$
$$= \llbracket t_2 \rrbracket_{\varepsilon[x \mapsto \pi_1 \llbracket t_1 \rrbracket_\varepsilon, y \mapsto \pi_2 \llbracket t_1 \rrbracket_\varepsilon]}$$
$$\in K[|C|]$$

which was what was to be shown.

Subcase 2: Suppose $B_1 \in \mathsf{Kind}$ and $B_2 \in \mathsf{Type}$. Then we know:

$$(\Delta \mid \sigma_3 \mid \mathbf{0}) \odot \Gamma \vdash B_1 : \mathsf{Type}_1$$
$$(\Delta, \sigma_3 \mid \sigma_4, r \mid \mathbf{0}) \odot (\Gamma, x : B_1) \vdash B_2 : \mathsf{Type}_0$$
$$(\Delta \mid \sigma_3 + \sigma_4 \mid \mathbf{0}) \odot \Gamma \vdash (x :_r B_1) \otimes B_2 : \mathsf{Type}_1$$

By Definition C.8 and Definition C.10 it suffices to show:

$$\llbracket \mathsf{let}\ (x : B_1, y : B_2) = t_1\ \mathsf{in}\ t_2 \rrbracket_\varepsilon$$
$$= \llbracket t_2 \rrbracket_{\varepsilon[x \mapsto \llbracket t_1 \rrbracket_\varepsilon]}$$
$$\in \mathcal{K}\llbracket [t_1/z]C \rrbracket$$
$$= \mathcal{K}\llbracket C \rrbracket$$

By the applying the IH to the premise for $t_1$ and Definition C.8:

IH(1): $\llbracket t_1 \rrbracket_\varepsilon \in \mathcal{K}\llbracket (x :_r B_1) \otimes B_2 \rrbracket = \mathcal{K}\llbracket B_1 \rrbracket$

This along with the kinding judgments given above imply that

$$(\Delta, \sigma_3, (\sigma_4, r)) \odot (\Gamma, x : B_1, y : B_2) \models \varepsilon[x \mapsto \llbracket t_1 \rrbracket_\varepsilon]$$

given the assumption $\Delta \odot \Gamma \models \varepsilon$. We also know:

$$(\Delta, \sigma_3, (\sigma_4, r) \mid \sigma_{12} \mid \mathbf{0}) \odot (\Gamma, x : B_1, y : B_2) \vdash [(x, y)/z]C : \mathsf{Type}_1$$

by applying kinding for typing to the premise for $t_2$. We now have everything we need to apply the IH to the premise for $t_2$:

IH(2): $\llbracket t_2 \rrbracket_{\varepsilon[x \mapsto \llbracket t_1 \rrbracket_\varepsilon]} \in \mathcal{K}\llbracket [(x, y)/z]C \rrbracket = \mathcal{K}\llbracket C \rrbracket$

Now by Definition C.10 and the previous results:

$$\llbracket \mathsf{let}\ (x : B_1, y : B_2) = t_1\ \mathsf{in}\ t_2 \rrbracket_\varepsilon$$
$$= \llbracket t_2 \rrbracket_{\varepsilon[x \mapsto \llbracket t_1 \rrbracket_\varepsilon]}$$
$$\in K[|C|]$$

which was what was to be shown.

Subcase 3: Suppose $B_1 \in \text{Type}$ and $B_2 \in \text{Kind}$. This case is similar to the previous case using:
$$(\Delta, \sigma_3, (\sigma_4, r)) \odot (\Gamma, x : B_1, y : B_2) \models \varepsilon[x \mapsto [\![t_1]\!]_\varepsilon]$$

Subcase 4: Suppose $B_1, B_2 \in \text{Type}$. This case is similar to the previous case using:
$$(\Delta, \sigma_3, (\sigma_4, r)) \odot (\Gamma, x : B_1, y : B_2) \models \varepsilon$$

Case 9:
$$\frac{(\Delta \mid \sigma \mid \mathbf{0}) \odot \Gamma \vdash B : \text{Type}_0}{(\Delta \mid \sigma \mid \mathbf{0}) \odot \Gamma \vdash \Box_s B : \text{Type}_0} \text{ T\_Box}$$

This case follows directly from the induction hypothesis.

Case 10:
$$\frac{(\Delta \mid \sigma_3 \mid \sigma_4) \odot \Gamma \vdash t' : B}{(\Delta \mid s * \sigma_3 \mid \sigma_4) \odot \Gamma \vdash \Box t' : \Box_s B} \text{ T\_BoxI}$$

This case follows directly from the induction hypothesis.

Case 11:
$$\frac{\begin{array}{c}(\Delta \mid \sigma_3 \mid \sigma_7) \odot \Gamma \vdash t_1 : \Box_s B_1 \\ (\Delta, \sigma_7 \mid \sigma_5, s \mid \sigma_6, (s * r)) \odot \Gamma, x : B_1 \vdash t_2 : [\Box x/z] B_2 \\ (\Delta, \sigma_7 \mid \sigma_6, r \mid \mathbf{0}) \odot \Gamma, z : \Box_s B_1 \vdash B_2 : \text{Type}_l\end{array}}{(\Delta \mid \sigma_3 + \sigma_5 \mid \sigma_6 + r * \sigma_3) \odot \Gamma \vdash \text{let} \, \Box(x : B_1) = t_1 \, \text{in} \, t_2 : [t_1/z] B_2} \text{ T\_BoxE}$$

In this case we know that:
$$\sigma_1 = \sigma_3 + \sigma_5$$
$$\sigma_2 = \sigma_6 + r * \sigma_3$$
$$t = (\text{let} \, \Box(x : B_1) = t_1 \, \text{in} \, t_2)$$
$$A = [t_1/z] B_2$$

It suffices to show that:
$$[\![\text{let} \, \Box(x : B_1) = t_1 \, \text{in} \, t_2]\!]_\varepsilon \in \mathcal{K}[\![[t_1/z] B_2]\!] = \mathcal{K}[\![B_2]\!]$$

In this case we know that $([t_1/z] B_2) \in \text{Kind}$, and thus, $B_2 \in \text{Kind}$ and either $B_1 \in \text{Kind}$ or $B_1 \in \text{Type}$. We cover both of these cases in turn.

Subcase 1: Suppose $B_1 \in \text{Kind}$. It suffices to show that:
$$\begin{aligned}[\![\text{let} \, \Box(x : B_1) = t_1 \, \text{in} \, t_2]\!]_\varepsilon &= [\![t_2]\!]_{\varepsilon[x \mapsto [\![t_1]\!]_\varepsilon]} \\ &\in \mathcal{K}[\![[t_1/z] B_2]\!] \\ &= \mathcal{K}[\![B_2]\!]\end{aligned}$$

As we have seen in the previous cases we can apply well-formed contexts for typing to obtain that:
$$(\Delta \mid \sigma_7 \mid \mathbf{0}) \odot \Gamma \vdash \Box_s B_1 : \text{Type}_1$$

We can now apply the IH to the premise for $t_1$ to obtain:

IH(1): $[\![t_1]\!]_\varepsilon \in \mathcal{K}[\![\Box_s B_1]\!] = \mathcal{K}[\![B_1]\!]$

Using the previous two facts along with the assumption that $\Delta \odot \Gamma \models \varepsilon$ we may obtain
$$(\Delta, \sigma_7) \odot (\Gamma, x : B_1) \models \varepsilon[x \mapsto [\![t_1]\!]_\varepsilon]$$

In addition, we know that $(\Delta, \sigma_7 \mid \sigma_6, r \mid \mathbf{0}) \odot \Gamma, z : \Box_s B_1 \vdash B_2 : \text{Type}_1$. Thus, we can now apply the induction hypothesis a second time.

17

IH(2): $[\![t_2]\!]_{\varepsilon[x \mapsto [\![t_1]\!]_\varepsilon]} \in \mathcal{K}[\![B_2]\!]$

which was what was to be shown.

Subcase 2: Suppose $B_1 \in \mathsf{Type}$. It suffices to show that:

$$[\![\mathsf{let}\,\Box(x : B_1) = t_1 \text{ in } t_2]\!]_\varepsilon = [\![t_2]\!]_\varepsilon$$
$$\in \mathcal{K}[\![[t_1/z]B_2]\!]$$
$$= \mathcal{K}[\![B_2]\!]$$

As we have seen in the previous cases we can apply well-formed contexts for typing to obtain that:

$$(\Delta \mid \sigma_7 \mid \mathbf{0}) \odot \Gamma \vdash \Box_s B_1 : \mathsf{Type}_0$$

Using this along with the assumption that $\Delta \odot \Gamma \models \varepsilon$ we may obtain

$$(\Delta, \sigma_7) \odot (\Gamma, x : B_1) \models \varepsilon$$

In addition, we know that $(\Delta, \sigma_7 \mid \sigma_6, r \mid \mathbf{0}) \odot \Gamma, z : \Box_s B_1 \vdash B_2 : \mathsf{Type}_1$. Thus, we can now apply the induction hypothesis a second time.

IH(2): $[\![t_2]\!]_\varepsilon \in \mathcal{K}[\![B_2]\!]$

which was what was to be shown.

Case 12:

$$\frac{(\Delta \mid \sigma_1 \mid \sigma_2) \odot \Gamma \vdash t : A' \quad (\Delta \mid \sigma_2 \mid \mathbf{0}) \odot \Gamma \vdash A' = A : \mathsf{Type}_1}{(\Delta \mid \sigma_1 \mid \sigma_2) \odot \Gamma \vdash t : A} \text{ T\_TYCONV}$$

This case follows by first applying the induction hypothesis to the typing premise, and then applying Lemma C.15 to obtain that $\mathcal{K}[\![A]\!] = \mathcal{K}[\![A']\!]$ obtaining our result.

We now move onto the second part of this result assuming the first. In this part we will show:

If $\Delta \odot \Gamma \models \varepsilon$ and $(\Delta \mid \sigma \mid \mathbf{0}) \odot \Gamma \vdash A : \mathsf{Type}_1$, then $[\![A]\!]_\varepsilon \in \mathcal{K}[\![A]\!] \to \mathsf{SAT}$.

Recall that this is a proof by mutual induction on $(\Delta \mid \sigma_1 \mid \sigma_2) \odot \Gamma \vdash t : A$ and $(\Delta \mid \sigma \mid \mathbf{0}) \odot \Gamma \vdash A : \mathsf{Type}_1$.

Case 1:

$$\frac{\Delta \odot \Gamma \vdash}{(\Delta \mid \mathbf{0} \mid \mathbf{0}) \odot \Gamma \vdash \mathsf{Type}_0 : \mathsf{Type}_1} \text{ T\_TYPE}$$

In this case we know that:
$$\sigma = \mathbf{0}$$
$$A = \mathsf{Type}_0$$

It suffices to show that:

$$[\![\mathsf{Type}_0]\!]_\varepsilon \in (\mathcal{K}[\![\mathsf{Type}_0]\!] \to \mathsf{SAT}) = (\mathsf{SAT} \to \mathsf{SAT})$$

But, $[\![\mathsf{Type}_0]\!]_\varepsilon = \lambda x \in \mathsf{SAT}.\mathsf{SN}$, and by Lemma C.7 $\mathsf{SN} \in \mathsf{SAT}$; hence, $(\lambda x \in \mathsf{SAT}.\mathsf{SN}) \in (\mathsf{SAT} \to \mathsf{SAT})$.

Case 2:

$$\frac{\Delta_1, \mathbf{0}, \Delta_2 \odot \Gamma_1, x : \mathsf{Type}_1, \Gamma_2 \vdash \quad |\Delta_1| = |\Gamma_1|}{(\Delta_1, \mathbf{0}, \Delta_2 \mid \mathbf{0}^{|\Delta_1|}, 1, \mathbf{0} \mid \mathbf{0}) \odot (\Gamma_1, x : \mathsf{Type}_1, \Gamma_2) \vdash x : \mathsf{Type}_1} \text{ T\_VAR}$$

This case is impossible, because the well-formed context premise fails, because $\mathsf{Type}_1$ has no type.

Case 3:
$$\frac{(\Delta \mid \sigma_3 \mid \mathbf{0}) \odot \Gamma \vdash B_1 : \mathsf{Type}_{l_1} \quad (\Delta, \sigma_3 \mid \sigma_4, r \mid \mathbf{0}) \odot \Gamma, x : B_1 \vdash B_2 : \mathsf{Type}_{l_2}}{(\Delta \mid \sigma_3 + \sigma_4 \mid \mathbf{0}) \odot \Gamma \vdash (x :_{(s,r)} B_1) \to B_2 : \mathsf{Type}_1} \text{ T\_Arrow}$$

In this case we know that:
$$\sigma = \sigma_3 + \sigma_4$$
$$A = (x :_{(s,r)} B_1) \to B_2$$

It suffices to show that:
$$[\![(x :_{(s,r)} B_1) \to B_2]\!]_\varepsilon \in (\mathcal{K}[\![(x :_{(s,r)} B_1) \to B_2]\!] \to \mathsf{SAT})$$

In this case either $B_1, B_2 \in \mathsf{Kind}$, $B_1 \in \mathsf{Kind}$ and $B_2 \in \mathsf{Type}$, or $B_1 \in \mathsf{Type}$ and $B_2 \in \mathsf{Kind}$. We consider each case in turn.

Subcase 1: Suppose $B_1, B_2 \in \mathsf{Kind}$. It suffices to show:

$$[\![(x :_{(s,r)} B_1) \to B_2]\!]_\varepsilon$$
$$= \lambda X \in \mathcal{K}[\![B_1]\!] \to \mathcal{K}[\![B_2]\!]. \bigcap_{Y \in \mathcal{K}[\![B_1]\!]} (([\![B_1]\!]_\varepsilon)(Y) \to ([\![B_2]\!]_{\varepsilon[x \mapsto Y]})(X(Y)))$$
$$\in (\mathcal{K}[\![(x :_{(s,r)} B_1) \to B_2]\!] \to \mathsf{SAT})$$
$$= \{f \mid f : \mathcal{K}[\![B_1]\!] \to \mathcal{K}[\![B_2]\!]\} \to \mathsf{SAT}$$

Now suppose $X \in (\mathcal{K}[\![B_1]\!] \to \mathcal{K}[\![B_2]\!])$ and $Y \in \mathcal{K}[\![B_1]\!]$. We know by assumption that $(\Delta \mid \sigma_3 \mid \mathbf{0}) \odot \Gamma \vdash B_1 : \mathsf{Type}_{l_1}$ and $\Delta \odot \Gamma \models \varepsilon$, and so we can apply the induction hypothesis to the premise for $B_1$ to obtain:

IH(1): $[\![B_1]\!]_\varepsilon \in (\mathcal{K}[\![B_1]\!] \to \mathsf{SAT})$

The previous facts now allow us, by Definition C.9, to obtain:

$$(\Delta, \sigma_3) \odot (\Gamma, x : B_1) \models \varepsilon[x \mapsto Y]$$

Thus, we can now apply the induction hypothesis to the premise for $B_2$ to obtain:

IH(2): $[\![B_2]\!]_{\varepsilon[x \mapsto Y]} \in (\mathcal{K}[\![B_2]\!] \to \mathsf{SAT})$

Then we know by IH(1) that $[\![B_1]\!]_\varepsilon(Y) \in \mathsf{SAT}$ and by IH(2) $[\![B_2]\!]_{\varepsilon[x \mapsto Y]}(X(Y)) \in \mathsf{SAT}$, thus:
$$([\![B_1]\!]_\varepsilon(Y) \to [\![B_2]\!]_{\varepsilon[x \mapsto Y]}(X(Y))) \in (\mathsf{SAT} \to \mathsf{SAT}) \in \mathsf{SAT}$$

Then by Lemma C.8:
$$\bigcap_{Y \in \mathcal{K}[\![B_1]\!]} ([\![B_1]\!]_\varepsilon(Y) \to [\![B_2]\!]_{\varepsilon[x \mapsto Y]}(X(Y))) \in \mathsf{SAT}$$

Therefore, we obtain our result.

Subcase 2: Suppose $B_1 \in \mathsf{Kind}$ and $B_2 \in \mathsf{Type}$. Similar to the previous case.

Subcase 3: Suppose $B_1 \in \mathsf{Type}$ and $B_2 \in \mathsf{Kind}$. It suffices to show:

$$[\![(x :_{(s,r)} B_1) \to B_2]\!]_\varepsilon$$
$$= \lambda X \in \mathcal{K}[\![B_2]\!].([\![B_1]\!]_\varepsilon \to ([\![B_2]\!]_\varepsilon)(X))$$
$$\in (\mathcal{K}[\![(x :_{(s,r)} B_1) \to B_2]\!] \to \mathsf{SAT})$$
$$= (\mathcal{K}[\![B_2]\!] \to \mathsf{SAT})$$

Now suppose $X \in \mathcal{K}[\![B_2]\!]$. We know by assumption that $(\Delta \mid \sigma_3 \mid \mathbf{0}) \odot \Gamma \vdash B_1 : \mathsf{Type}_0$ and $\Delta \odot \Gamma \models \varepsilon$, and so we can apply the first part of the induction hypothesis to the premise for $B_1$ to obtain:

IH(1): $[\![B_1]\!]_\varepsilon \in \mathcal{K}[\![\mathsf{Type}_0]\!] = \mathsf{SAT}$

Now we know by assumption that $(\Delta, \sigma_3 \mid \sigma_4, r \mid \mathbf{0}) \odot \Gamma, x : B_1 \vdash B_2 : \mathsf{Type}_0$, and we can now show by Definition C.9 that $(\Delta, \sigma_3) \odot (\Gamma, x : B_1) \models \varepsilon$ holds. So we can apply the IH to the former judgment to obtain:

IH(2): $[\![B_2]\!]_\varepsilon \in (\mathcal{K}[\![B_2]\!] \to \mathsf{SAT})$

At this point we can see that $(\lambda X \in \mathcal{K}[\![B_2]\!].([\![B_1]\!]_\varepsilon \to [\![B_2]\!]_\varepsilon(X))) \in \mathsf{SAT}$ by the previous facts, and the fact that $\mathsf{SAT}$ is closed under function spaces.

Case 4:
$$\frac{(\Delta \mid \sigma_3 \mid \mathbf{0}) \odot \Gamma \vdash B_1 : \mathsf{Type}_{l_1} \quad (\Delta, \sigma_3 \mid \sigma_4, r \mid \mathbf{0}) \odot \Gamma, x : B_1 \vdash B_2 : \mathsf{Type}_{l_2}}{(\Delta \mid \sigma_3 + \sigma_4 \mid \mathbf{0}) \odot \Gamma \vdash (x :_r B_1) \otimes B_2 : \mathsf{Type}_{l_1 \sqcup l_2}} \text{ T\_Ten}$$

This case is similar to the previous case.

Case 5:
$$\frac{(\Delta \mid \sigma \mid \mathbf{0}) \odot \Gamma \vdash B : \mathsf{Type}_1}{(\Delta \mid \sigma \mid \mathbf{0}) \odot \Gamma \vdash \Box_s B : \mathsf{Type}_1} \text{ T\_Box}$$

This case follows from the IH.

Case 6:
$$\frac{(\Delta \mid \sigma_1 \mid \sigma_2) \odot \Gamma \vdash A : A' \quad (\Delta \mid \sigma_2 \mid \mathbf{0}) \odot \Gamma \vdash A' = \mathsf{Type}_1 : B}{(\Delta \mid \sigma_1 \mid \sigma_2) \odot \Gamma \vdash A : \mathsf{Type}_1} \text{ T\_TyConv}$$

This case is impossible, because $\mathsf{Type}_1$ has no type $B$.

$\square$

**Theorem 4.17** (Soundness for Semantic Typing). *If $(\Delta \mid \sigma_1 \mid \sigma_2) \odot \Gamma \vdash t : A$, then $(\Delta \mid \sigma_1 \mid \sigma_2) \odot \Gamma \models t : A$.*

*Proof.* This is a proof by induction on the assumed typing derivation.

Case 1:
$$\frac{\Delta \odot \Gamma \vdash}{(\Delta \mid \mathbf{0} \mid \mathbf{0}) \odot \Gamma \vdash \mathsf{Type}_0 : \mathsf{Type}_1} \text{ T\_Type}$$

In this case we have:
$$\sigma_1 = \mathbf{0}$$
$$\sigma_2 = \mathbf{0}$$
$$t = \mathsf{Type}_0$$
$$A = \mathsf{Type}_1$$

We can now see that this case holds trivially, because $\mathsf{Type}_1$ has no type.

Case 2:
$$\frac{\Delta_1, \sigma, \Delta_2 \odot \Gamma_1, x : A, \Gamma_2 \vdash \quad |\Delta_1| = |\Gamma_1|}{(\Delta_1, \sigma, \Delta_2 \mid \mathbf{0}^{|\Delta_1|}, 1, \mathbf{0} \mid \sigma, 0, \mathbf{0}) \odot \Gamma_1, x : A, \Gamma_2 \vdash x : A} \text{ T\_Var}$$

In this case we have:
$$\Delta = (\Delta_1, \sigma, \Delta_2)$$
$$\sigma_1 = \mathbf{0}^{|\Delta_1|}, 1, \mathbf{0}$$
$$\sigma_2 = \sigma, 0, \mathbf{0}$$
$$t = x$$

Now either $A \in \mathsf{Kind}$ or $A \in \mathsf{Type}$. We consider both cases in turn.

20

Subcase 1: Suppose $A \in \mathsf{Kind}$ and $\Delta \odot \Gamma \models_\varepsilon \rho$. It suffices to show:
$$(\!|x|\!)_\rho = \rho\, x \in [\![A]\!]_\varepsilon ([\![x]\!]_\varepsilon) = [\![A]\!]_\varepsilon (\varepsilon\, x)$$

But, this holds by Definition C.12, because the well-formed context premise above implies the proper kinding of $A$.

Subcase 2: Suppose $A \in \mathsf{Type}$ and $\Delta \odot \Gamma \models_\varepsilon \rho$. It suffices to show:
$$(\!|x|\!)_\rho = \rho\, x \in [\![A]\!]_\varepsilon$$

But, this holds by Definition C.12, because the well-formed context premise above implies the proper kinding of $A$.

Case 3:
$$\frac{(\Delta \mid \sigma_3 \mid \mathbf{0}) \odot \Gamma \vdash B_1 : \mathsf{Type}_0 \qquad (\Delta, \sigma_3 \mid \sigma_4, r \mid \mathbf{0}) \odot \Gamma, x : B_1 \vdash B_2 : \mathsf{Type}_0}{(\Delta \mid \sigma_3 + \sigma_4 \mid \mathbf{0}) \odot \Gamma \vdash (x :_{(s,r)} B_1) \to B_2 : \mathsf{Type}_0} \text{ T\_Arrow}$$

In this case we have:
$$\sigma_1 = (\sigma_3 + \sigma_4)$$
$$\sigma_2 = \mathbf{0}$$
$$t = (x :_{(s,r)} B_1) \to B_2$$
$$A = \mathsf{Type}_0$$

We only need to consider the first case of this theorem, because $\mathsf{Type}_1$ has no type. So suppose $\Delta \odot \Gamma \models_\varepsilon \rho$. It suffices to show:
$$(\!|(x :_{(s,r)} B_1) \to B_2|\!)_\rho$$
$$= ((x :_{(s,r)} (\!|B_1|\!)_\rho) \to (\!|B_2|\!)_\rho)$$
$$\in [\![\mathsf{Type}_0]\!]_\varepsilon ([\![(x :_{(s,r)} B_1) \to B_2]\!]_\varepsilon)$$
$$= \mathsf{SN}$$

We know by assumption that $\Delta \odot \Gamma \models_\varepsilon \rho$ so we can apply the IH to conclude:

IH(1): $(\!|B_1|\!)_\rho \in [\![\mathsf{Type}_0]\!]_\varepsilon ([\![B_1]\!]_\varepsilon) = \mathsf{SN}$

Now suppose $t \in [\![B_1]\!]_\varepsilon$. Then we know by Definition C.12 that $(\Delta, \sigma_3) \odot (\Gamma, x : B_1) \models_\varepsilon \rho[x \mapsto t]$. Thus, by applying the IH to the premise for $B_2$ we may conclude that

IH(2): $(\!|B_2|\!)_{\rho[x \mapsto t]} \in [\![\mathsf{Type}_0]\!]_\varepsilon ([\![B_2]\!]_\varepsilon) = \mathsf{SN}$

holds for every $t$. Therefore, we may conclude out result.

Case 4:
$$\frac{(\Delta \mid \sigma_3 \mid \mathbf{0}) \odot \Gamma \vdash B_1 : \mathsf{Type}_0 \qquad (\Delta, \sigma_3 \mid \sigma_4, r \mid \mathbf{0}) \odot \Gamma, x : B_1 \vdash B_2 : \mathsf{Type}_0}{(\Delta \mid \sigma_3 + \sigma_4 \mid \mathbf{0}) \odot \Gamma \vdash (x :_r B_1) \otimes B_2 : \mathsf{Type}_0} \text{ T\_Ten}$$

Similar to the previous case.

Case 5:
$$\frac{(\Delta, \sigma_3 \mid \sigma_4, s \mid \sigma_5, r) \odot \Gamma, x : B_1 \vdash t' : B_2}{(\Delta \mid \sigma_4 \mid \sigma_3 + \sigma_5) \odot \Gamma \vdash \lambda x : B_1.t' : (x :_{(s,r)} B_1) \to B_2} \text{ T\_Fun}$$

In this case we have:
$$\sigma_1 = \sigma_4$$
$$\sigma_2 = (\sigma_3 + \sigma_5)$$
$$t = \lambda x : B_1.t'$$
$$A = (x :_{(s,r)} B_1) \to B_2$$

We have two cases to consider, either $((x :_{(s,r)} B_1) \to B_2) \in \mathsf{Kind}$ or $((x :_{(s,r)} B_1) \to B_2) \in \mathsf{Type}$. We cover both cases in turn.

Subcase 1: Suppose $((x :_{(s,r)} B_1) \to B_2) \in \mathsf{Kind}$. We now have three subcases depending on the typing for both $B_1$ and $B_2$.

Subcase 1: Suppose $B_1, B_2 \in \mathsf{Kind}$. It suffices to show:

$$
\begin{aligned}
&(\!|\lambda x : B_1.t'|\!)_\rho \\
&= \lambda x : (\!|B_1|\!)_\rho.(\!|t'|\!)_\rho \\
&\in [\![(x :_{(s,r)} B_1) \to B_2]\!]_\varepsilon ([\![\lambda x : B_1.t']\!]_\varepsilon) \\
&= \bigcap_{Y \in \mathcal{K}[\![B_1]\!]} (([\![B_1]\!]_\varepsilon)(Y) \to ([\![B_2]\!]_{\varepsilon[x \mapsto Y]}) ([\![\lambda x : B_1.t']\!]_\varepsilon (Y))) \\
&= \bigcap_{Y \in \mathcal{K}[\![B_1]\!]} (([\![B_1]\!]_\varepsilon)(Y) \to ([\![B_2]\!]_{\varepsilon[x \mapsto Y]}) ([\![t']\!]_{\varepsilon[x \mapsto Y]}))
\end{aligned}
$$

So suppose we have a $Y \in \mathcal{K}[\![B_1]\!]$ and a $t \in [\![B_1]\!]_\varepsilon (Y) = [\![B_1]\!]_{(\varepsilon[x \mapsto Y])} (\varepsilon\, x)$ (since $x \notin \mathsf{FV}(B_2)$). Then we know that $(\Delta, \sigma_3) \odot (\Gamma, x : B_1) \models_{\varepsilon[x \mapsto Y]} \rho[x \mapsto t]$. Then by the IH:

IH(1): $(\!|t'|\!)_{\rho[x \mapsto t]} \in [\![B_2]\!]_{\varepsilon[x \mapsto Y]} ([\![t']\!]_{\varepsilon[x \mapsto Y]})$

Thus, we obtain our result.

Subcase 2: Suppose $B_1 \in \mathsf{Kind}$ and $B_2 \in \mathsf{Type}$. This case is similar to the previous case, except we will use part two of the IH.

Subcase 3: Suppose $B_1 \in \mathsf{Type}$ and $B_2 \in \mathsf{Kind}$. Similar to the previous case.

Subcase 2: Suppose $((x :_{(s,r)} B_1) \to B_2) \in \mathsf{Type}$. This case is similar to the above, but we will use the second part of the IH.

Case 6:
$$
\frac{(\Delta \mid \sigma_4 \mid \sigma_3 + \sigma_5) \odot \Gamma \vdash t_1 : (x :_{(s,r)} B_1) \to B_2 \qquad (\Delta \mid \sigma_6 \mid \sigma_3) \odot \Gamma \vdash t_2 : B_1}{(\Delta \mid \sigma_4 + s * \sigma_6 \mid \sigma_5 + r * \sigma_6) \odot \Gamma \vdash t_1\, t_2 : [t_2/x]B_2} \text{ T\_App}
$$

In this case we have:
$$
\begin{aligned}
\sigma_1 &= (\sigma_4 + s * \sigma_6) \\
\sigma_2 &= (\sigma_5 + r * \sigma_6) \\
t &= (t_1\, t_2) \\
A &= [t_2/x]B_2
\end{aligned}
$$

We have several cases to consider.

Subcase 1: Suppose $B_1, [t_2/x]B_2 \in \mathsf{Kind}$. It suffices to show:

$$
\begin{aligned}
&(\!|t_1\, t_2|\!)_\rho \\
&= (\!|t_1|\!)_\rho\, (\!|t_2|\!)_\rho \\
&\in [\![[t_2/x]B_2]\!]_\varepsilon ([\![t_1\, t_2]\!]_\varepsilon) \\
&= [\![B_2]\!]_{\varepsilon[x \mapsto [\![t_2]\!]_\varepsilon]} ([\![t_1]\!]_\varepsilon ([\![t_2]\!]_\varepsilon))
\end{aligned}
$$

By the IH we have:

IH(1): $(\!|t_1|\!)_\rho \in [\![(x :_{(s,r)} B_1) \to B_2]\!]_\varepsilon ([\![t_1]\!]_\varepsilon)$
IH(2): $(\!|t_2|\!)_\rho \in [\![B_1]\!]_\varepsilon ([\![t_2]\!]_\varepsilon)$

Notice that by Lemma C.16 we know that $[\![t_2]\!]_\varepsilon \in \mathcal{K}[\![B_1]\!]$. So using C.10 we know:

$$
\begin{aligned}
(\!|t_1|\!)_\rho &\in [\![(x :_{(s,r)} B_1) \to B_2]\!]_\varepsilon ([\![t_1]\!]_\varepsilon) \\
&= \bigcap_{Y \in \mathcal{K}[\![A]\!]} (([\![A]\!]_\varepsilon)(Y) \to ([\![B]\!]_{\varepsilon[x \mapsto Y]}) ([\![t_1]\!]_\varepsilon (Y)))
\end{aligned}
$$

Thus, $(\!|t_1|\!)_\rho\, (\!|t_2|\!)_\rho \in ([\![B]\!]_{\varepsilon[x \mapsto [\![t_2]\!]_\varepsilon]}) ([\![t_1]\!]_\varepsilon ([\![t_2]\!]_\varepsilon))$.

**Subcase 2:** Suppose $B_1 \in \mathsf{Kind}$ and $[t_2/x]B_2 \in \mathsf{Type}$. It suffices to show:

$$\begin{aligned}
&(\!|t_1\, t_2|\!)_\rho \\
&= (\!|t_1|\!)_\rho\, (\!|t_2|\!)_\rho \\
&\in [\![[t_2/x]B_2]\!]_\varepsilon \\
&= [\![B_2]\!]_{\varepsilon[x \mapsto [\![t_2]\!]_\varepsilon]}
\end{aligned}$$

By the IH we have:

IH(1): $(\!|t_1|\!)_\rho \in [\![(x :_{(s,r)} B_1) \to B_2]\!]_\varepsilon ([\![t_1]\!]_\varepsilon)$
IH(2): $(\!|t_2|\!)_\rho \in [\![B_1]\!]_\varepsilon ([\![t_2]\!]_\varepsilon)$

Notice that by Lemma C.16 we know that $[\![t_2]\!]_\varepsilon \in \mathcal{K}[\![B_1]\!]$. So using C.10 we know:

$$\begin{aligned}
&(\!|t_1|\!)_\rho \in [\![(x :_{(s,r)} B_1) \to B_2]\!]_\varepsilon ([\![t_1]\!]_\varepsilon) \\
&= \bigcap_{Y \in \mathcal{K}[\![B_1]\!]} (([\![B_1]\!]_\varepsilon)(Y) \to ([\![B_2]\!]_{\varepsilon[x \mapsto Y]}))
\end{aligned}$$

Thus, $(\!|t_1|\!)_\rho\, (\!|t_2|\!)_\rho \in ([\![B_2]\!]_{\varepsilon[x \mapsto [\![t_2]\!]_\varepsilon]})$.

**Subcase 3:** Suppose $B_1 \in \mathsf{Type}$ and $[t_2/x]B_2 \in \mathsf{Type}$. It suffices to show:

$$\begin{aligned}
&(\!|t_1\, t_2|\!)_\rho \\
&= (\!|t_1|\!)_\rho\, (\!|t_2|\!)_\rho \\
&\in [\![[t_2/x]B_2]\!]_\varepsilon \\
&= [\![B_2]\!]_{\varepsilon[x \mapsto [\![t_2]\!]_\varepsilon]}
\end{aligned}$$

By the IH we have:

IH(1): $(\!|t_1|\!)_\rho \in [\![(x :_{(s,r)} B_1) \to B_2]\!]_\varepsilon$
IH(2): $(\!|t_2|\!)_\rho \in [\![B_1]\!]_\varepsilon$

Using C.10 we know:

$$\begin{aligned}
&(\!|t_1|\!)_\rho \in [\![(x :_{(s,r)} B_1) \to B_2]\!]_\varepsilon ([\![t_1]\!]_\varepsilon) \\
&= [\![B_1]\!]_\varepsilon \to [\![B_2]\!]_\varepsilon
\end{aligned}$$

Thus, $(\!|t_1|\!)_\rho\, (\!|t_2|\!)_\rho \in [\![B_2]\!]_\varepsilon$.

**Case 7:**

$$\frac{\begin{array}{c}(\Delta, \sigma_3 \mid \sigma_5, r \mid \mathbf{0}) \odot \Gamma, x : B_1 \vdash B_2 : \mathsf{Type}_l \\ (\Delta \mid \sigma_4 \mid \sigma_3) \odot \Gamma \vdash t_1 : B_1 \\ (\Delta \mid \sigma_6 \mid \sigma_5 + r * \sigma_4) \odot \Gamma \vdash t_2 : [t_1/x]B_2\end{array}}{(\Delta \mid \sigma_4 + \sigma_6 \mid \sigma_3 + \sigma_5) \odot \Gamma \vdash (t_1, t_2) : (x :_r B_1) \otimes B_2}\ \text{T\_Pair}$$

This case is similar to the case for $\lambda$-abstraction above.

**Case 8:**

$$\frac{\begin{array}{c}(\Delta \mid \sigma_5 \mid \sigma_3 + \sigma_4) \odot \Gamma \vdash t_1 : (x :_r B_1) \otimes B_2 \\ (\Delta, (\sigma_3 + \sigma_4) \mid \sigma_7, r' \mid \mathbf{0}) \odot \Gamma, z : (x :_r B_1) \otimes B_2 \vdash C : \mathsf{Type}_1 \\ (\Delta, \sigma_3, (\sigma_4, r) \mid \sigma_6, s, s \mid \sigma_7, r', r') \odot \Gamma, x : B_1, y : B_2 \vdash t_2 : [(x,y)/z]C\end{array}}{(\Delta \mid \sigma_6 + s * \sigma_5 \mid \sigma_7 + r' * \sigma_5) \odot \Gamma \vdash \mathsf{let}\,(x : B_1, y : B_2) = t_1\,\mathsf{in}\,t_2 : [t_1/z]C}\ \text{T\_TenCut}$$

Similar to the application case above.

**Case 9:**

$$\frac{(\Delta \mid \sigma \mid \mathbf{0}) \odot \Gamma \vdash B : \mathsf{Type}_0}{(\Delta \mid \sigma \mid \mathbf{0}) \odot \Gamma \vdash \Box_s B : \mathsf{Type}_0}\ \text{T\_Box}$$

This case follows from the IH.

Case 10:
$$\dfrac{(\Delta \mid \sigma_3 \mid \sigma_2) \odot \Gamma \vdash t' : B}{(\Delta \mid s * \sigma_3 \mid \sigma_2) \odot \Gamma \vdash \Box t' : \Box_s B} \ \text{T\_BoxI}$$

In this case we have:
$$\sigma_1 = (s * \sigma_3)$$
$$t = (\Box t')$$
$$A = \Box_s B$$

We have two cases to consider.

Subcase 1: Suppose $B \in \mathsf{Kind}$. It suffices to show:

$$\begin{aligned}
&(\!|\Box t'|\!)_\rho \\
={}& (\!|t'|\!)_\rho \\
\in{}& [\![\Box_s B]\!]_\varepsilon \,([\![\Box t']\!]_\varepsilon) \\
={}& [\![B]\!]_\varepsilon \,([\![t']\!]_\varepsilon)
\end{aligned}$$

At this point, this case holds by the IH.

Subcase 2: Suppose $B \in \mathsf{Type}$. Similar to the previous case.

Case 11:
$$\dfrac{\begin{array}{c}(\Delta \mid \sigma_3 \mid \sigma_7) \odot \Gamma \vdash t_1 : \Box_s B_1 \\ (\Delta, \sigma_7 \mid \sigma_6, r \mid \mathbf{0}) \odot \Gamma, z : \Box_s B_1 \vdash B_2 : \mathsf{Type}_0 \\ (\Delta, \sigma_7 \mid \sigma_5, s \mid \sigma_6, (s * r)) \odot \Gamma, x : B_1 \vdash t_2 : [\Box x/z]B_2\end{array}}{(\Delta \mid \sigma_3 + \sigma_5 \mid \sigma_6 + r * \sigma_3) \odot \Gamma \vdash \mathsf{let}\,\Box(x : B_1) = t_1\,\mathsf{in}\,t_2 : [t_1/z]B_2} \ \text{T\_BoxE}$$

In this case we have:
$$\sigma_1 = (\sigma_3 + \sigma_5)$$
$$\sigma_2 = (\sigma_6 + r * \sigma_3)$$
$$t = (\mathsf{let}\,\Box(x : B_1) = t_1\,\mathsf{in}\,t_2)$$
$$A = [t_1/z]B_2$$

We have several cases to consider:

Subcase 1: Suppose $B_1, B_2 \in \mathsf{Kind}$. It suffices to show:

$$\begin{aligned}
&(\!|\mathsf{let}\,\Box(x : B_1) = t_1\,\mathsf{in}\,t_2|\!)_\rho \\
={}& (\!|t_2|\!)_{\rho[x \mapsto (\!|t_1|\!)_\rho]} \\
\in{}& [\![[t_1/z]B_2]\!]_\varepsilon \,([\![\mathsf{let}\,\Box(x : B_1) = t_1\,\mathsf{in}\,t_2]\!]_\varepsilon) \\
={}& [\![B_2]\!]_{\varepsilon[x \mapsto [\![t_1]\!]_\varepsilon]} \,([\![t_2]\!]_{\varepsilon[x \mapsto [\![t_1]\!]_\varepsilon]})
\end{aligned}$$

By the IH:

IH(1): $(\!|t_1|\!)_\rho \in [\![\Box_s B_1]\!]_\varepsilon \,([\![t_1]\!]_\varepsilon) = [\![B_1]\!]_\varepsilon \,([\![t_1]\!]_\varepsilon)$

At this point we need:

$$(\Delta, \sigma_7) \odot (\Gamma, x : B_1) \models_{\varepsilon[x \mapsto [\![t_1]\!]_\varepsilon]} \rho[x \mapsto (\!|t_1|\!)_\rho]$$

But, this follows by definition and Lemma C.16.

By the IH:

IH(2): $(\!|t_2|\!)_{\rho[x \mapsto (\!|t_1|\!)_\rho]} \in [\![B_2]\!]_{\varepsilon[x \mapsto [\![t_1]\!]_\varepsilon]} \,([\![t_2]\!]_{\varepsilon[x \mapsto [\![t_1]\!]_\varepsilon]})$

Subcase 2: Suppose $B_1 \in \mathsf{Kind}$ and $B_2 \in \mathsf{Type}$. Similar to the previous case.

Subcase 3: Suppose $B_1 \in \mathsf{Type}$ and $B_2 \in \mathsf{Kind}$. It suffices to show:

$$\begin{aligned}
&(\!|\mathsf{let}\,\Box(x:B_1) = t_1\,\mathsf{in}\,t_2|\!)_\rho \\
&= (\!|t_2|\!)_{\rho[x \mapsto (\!|t_1|\!)_\rho]} \\
&\in [\![[t_1/z]B_2]\!]_\varepsilon ([\![\mathsf{let}\,\Box(x:B_1) = t_1\,\mathsf{in}\,t_2]\!]_\varepsilon) \\
&= [\![B_2]\!]_\varepsilon ([\![t_2]\!]_\varepsilon)
\end{aligned}$$

By the IH:

IH(1): $(\!|t_1|\!)_\rho \in [\![\Box_s B_1]\!]_\varepsilon = [\![B_1]\!]_\varepsilon$

At this point we need:

$$(\Delta, \sigma_7) \odot (\Gamma, x : B_1) \models_\varepsilon \rho[x \mapsto (\!|t_1|\!)_\rho]$$

But, this follows by definition.

By the IH:

IH(2): $(\!|t_2|\!)_{\rho[x \mapsto (\!|t_1|\!)_\rho]} \in [\![B_2]\!]_\varepsilon ([\![t_2]\!]_\varepsilon)$

Subcase 4: Suppose $B_1, B_2 \in \mathsf{Type}$. Similar to the previous case.

Case 12:

$$\frac{(\Delta \mid \sigma_1 \mid \sigma_2) \odot \Gamma \vdash t : A' \quad (\Delta \mid \sigma_2 \mid \mathbf{0}) \odot \Gamma \vdash A' = A : \mathsf{Type}_l}{(\Delta \mid \sigma_1 \mid \sigma_2) \odot \Gamma \vdash t : A} \;\text{T\_TyConv}$$

By the induction hypothesis we know $(\Delta \mid \sigma_1 \mid \sigma_2) \odot \Gamma \models t : A'$ holds, but this implies $(\Delta \mid \sigma_1 \mid \sigma_2) \odot \Gamma \models t : A'$ holds by Lemma C.15.

□

**Corollary 4.17.1** (Strong Normalization). *For every* $(\Delta \mid \sigma_1 \mid \sigma_2) \odot \Gamma \vdash t : A$, $t \in \mathsf{SN}$.

*Proof.* Similarly to CC, we can define a notion of canonical element in $\mathcal{K}[\![A]\!]$, and define a term valuation $\Delta \odot \Gamma \models_\varepsilon \rho$, and then conclude SN by the previous theorem. □

## A  Proofs for Graded Modal Dependent Type Theory

**Lemma 3.1** (Strengthening). *If* $(\Delta, \sigma_1, \Delta' \mid \sigma_2, s, \sigma'_2 \mid \sigma_3, r, \sigma'_3) \odot \Gamma, x : A, \Gamma' \vdash \mathcal{J}$, $x \notin \mathsf{FV}(\mathcal{J})$, *and* $x \notin \mathsf{FV}(\Gamma')$, *with* $|\Delta| = |\sigma_2| = |\sigma_3| = |\Gamma|$, *then* $(\Delta, (\Delta' \backslash \pi) \mid \sigma_2, \sigma'_2 \mid \sigma_3, \sigma'_3) \odot \Gamma, \Gamma' \vdash \mathcal{J}$ *where* $\pi = |\Gamma|$.

*Proof.* For well-formedness:

By induction on the definition of well-formed contexts.

- Case WF_EMPTY. Trivial since it does not match the form of the lemma.

- Case WF_EXT We consider two cases depending on the syntactic structure of $\Gamma_2$ and $\Delta_2$ (simultaneously, since they must have the same size by the lemma statement).

    - $\Gamma_2 = \emptyset$ and $\Delta_2 = \emptyset$ then we have:

    $$\frac{(\Delta_1 \mid \sigma \mid \mathbf{0}) \odot \Gamma_1 \vdash A : \mathsf{Type}_l}{\Delta_1, \sigma \odot \Gamma_1, x : A \vdash} \;\text{WF\_EXT}$$

    By Lemma 3.3 (on the first premise) then we have $\Delta_1 \odot \Gamma_1 \vdash$ which is the goal here.

- $\Gamma_2 = \Gamma'_2, y : A'$ and $\Delta_2 = \Delta'_2, \sigma'$ then we have:

$$\frac{(\Delta_1, \sigma, \Delta'_2 \mid \sigma' \mid \mathbf{0}) \odot \Gamma_1, x : A, \Gamma'_2 \vdash A' : \mathsf{Type}_l}{\Delta_1, \sigma, \Delta'_2, \sigma' \odot \Gamma_1, x : A, \Gamma'_2, y : A' \vdash} \text{ W}_{\text{F\_EXT}}$$

By Lemma 3.1 on the premise we get $(\Delta_1, (\Delta'_2 \backslash \pi) \mid \sigma' \backslash \pi \mid \mathbf{0}) \odot \Gamma_1, \Gamma'_2 \vdash A' : \mathsf{Type}_l$
We then construct the derivation:

$$\frac{(\Delta_1, (\Delta'_2 \backslash \pi) \mid \sigma' \backslash \pi \mid \mathbf{0}) \odot \Gamma_1, \Gamma'_2 \vdash A' : \mathsf{Type}_l}{\Delta_1, (\Delta'_2 \backslash \pi), \sigma' \backslash \pi \odot \Gamma_1, \Gamma'_2, y : A' \vdash} \text{ W}_{\text{F\_EXT}}$$

Since $(\Delta'_2 \backslash \pi), \sigma' \backslash \pi = (\Delta'_2, \sigma') \backslash \pi$ then the above derivation satisfies the goal.

For typing: By induction on $(\Delta, \sigma_1, \Delta' \mid \sigma_2, s, \sigma'_2 \mid \sigma_3, r, \sigma'_3) \odot \Gamma, x : A, \Gamma' \vdash t : B$

---

Case T_TYPE

$$\frac{\Delta, \sigma_1, \Delta' \odot \Gamma, x : A, \Gamma' \vdash}{(\Delta, \sigma_1, \Delta' \mid \mathbf{0} \mid \mathbf{0}) \odot \Gamma, x : A, \Gamma' \vdash \mathsf{Type}_l : \mathsf{Type}_{\mathsf{suc}\ l}} \text{ T\_TYPE}$$

By induction on $\Delta, \sigma_1, \Delta' \odot \Gamma, x : A, \Gamma' \vdash$ we have that $\Delta, (\Delta' \backslash \pi) \odot \Gamma, \Gamma' \vdash$ we can thus form the derivation:

$$\frac{\Delta, (\Delta' \backslash \pi) \odot \Gamma, \Gamma' \vdash}{(\Delta, (\Delta' \backslash \pi) \mid \mathbf{0} \mid \mathbf{0}) \odot \Gamma \vdash \mathsf{Type}_l : \mathsf{Type}_{\mathsf{suc}\ l}} \text{ T\_TYPE}$$

thus satisfying the goal here.

---

Case T_VAR

$$\frac{\Delta_1, \sigma_1, \Delta_2, \sigma_1, \Delta'_2 \odot \Gamma_1, y : A', \Gamma_2, x : A, \Gamma'_2 \vdash \quad |\Delta_1| = |\Gamma_1|}{(\Delta_1, \sigma_1, \Delta_2, \sigma_1, \Delta'_2 \mid \mathbf{0}^{|\Delta_1|}, 1, \mathbf{0} \mid \sigma_1, 0, \mathbf{0}) \odot \Gamma_1, y : A', \Gamma_2, x : A, \Gamma'_2 \vdash y : A'} \text{ T\_VAR}$$

By induction on $\Delta_1, \sigma'_1, \Delta_2, \sigma_1, \Delta'_2 \odot \Gamma_1, y : A', \Gamma_2, x : A, \Gamma'_2 \vdash$ we have that $\Delta_1, \sigma'_1, \Delta_2, (\Delta'_2 \backslash \pi) \odot \Gamma_1, y : A', \Gamma_2, \Gamma'_2$.

Thus, we can build the var rule:

$$\frac{\Delta_1, \sigma_1, \Delta_2, (\Delta'_2 \backslash \pi) \odot \Gamma_1, y : A', \Gamma_2, \Gamma'_2 \vdash \quad |\Delta_1| = |\Gamma_1|}{(\Delta_1, \sigma_1, \Delta_2, (\Delta'_2 \backslash \pi) \mid \mathbf{0}^{|\Delta_1|}, 1, \mathbf{0} \mid \sigma_1, \mathbf{0}) \odot \Gamma_1, y : A', \Gamma_2, \Gamma'_2 \vdash y : A'} \text{ T\_VAR}$$

which satisfies the goal here.

---

Case T_ARROW

$$\frac{\begin{array}{c}(\Delta, \sigma_1, \Delta' \mid \sigma_4 \mid \mathbf{0}) \odot \Gamma, x : A, \Gamma' \vdash A' : \mathsf{Type}_{l_1} \\ (\Delta, \sigma_1, \Delta', \sigma_4 \mid \sigma_5, r' \mid \mathbf{0}) \odot \Gamma, x : A, \Gamma', y : A' \vdash B : \mathsf{Type}_{l_2}\end{array}}{(\Delta, \sigma_1, \Delta' \mid \sigma_4 + \sigma_5 \mid \mathbf{0}) \odot \Gamma, x : A, \Gamma' \vdash (y :_{(s', r')} A') \to B : \mathsf{Type}_{l_1 \sqcup l_2}} \text{ T\_ARROW}$$

Thus, in the case we have that we have that $\sigma_3 = \mathbf{0}$, $r = 0$, and $\sigma'_3 = \mathbf{0}$ and that $\sigma_4 + \sigma_5 = \sigma_2, s, \sigma'_2$ thus we must have $\sigma'_4, s_1, \sigma''_4 = \sigma_4$ and $\sigma'_5, s_2, \sigma''_5 = \sigma_5$ and $s_1 + s_2 = s$ such that $(\sigma''_1, s_1, \sigma'''_1) + (\sigma''_2, s_2, \sigma'''_2) = \sigma_2, s, \sigma'_2$.

26

By induction on the premises we get:

$$(\Delta, (\Delta'\backslash\pi) \mid \sigma'_4, \sigma''_4 \mid \mathbf{0}) \odot \Gamma, \Gamma' \vdash A' : \mathsf{Type}_{l_1} \quad \text{(ih1)}$$

$$(\Delta, (\Delta'\backslash\pi), \sigma'_4, \sigma''_4 \mid \sigma'_5, \sigma''_5, r' \mid \mathbf{0}) \odot \Gamma, \Gamma', y : A' \vdash B : \mathsf{Type}_{l_2} \quad \text{(ih2)}$$

Then we can derive:

$$\frac{\text{(ih1)} \quad \text{(ih2)}}{(\Delta, (\Delta'\backslash\pi) \mid \sigma'_4, \sigma''_4 + \sigma'_5, \sigma''_5 \mid \mathbf{0}) \odot \Gamma, \Gamma' \vdash (y :_{(s',r')} A') \to B : \mathsf{Type}_{l_1 \sqcup l_2}} \ \text{T\_Arrow}$$

which satisfies this case.

---

Case T_Ten Same reasoning as above since the structure is exactly the same

---

Case T_Fun

$$\frac{\begin{array}{c}(\Delta, \sigma, \Delta', \sigma_1 \mid \sigma_4, r' \mid \mathbf{0}) \odot \Gamma, x : A, \Gamma', y : A' \vdash B : \mathsf{Type}_l \\ (\Delta, \sigma, \Delta', \sigma_1 \mid \sigma_2, s, \sigma'_2, s' \mid \sigma_4, r') \odot \Gamma, x : A, \Gamma', y : A' \vdash t : B\end{array}}{(\Delta, \sigma, \Delta' \mid \sigma_2, s, \sigma'_2 \mid \sigma_1 + \sigma_4) \odot \Gamma, x : A, \Gamma' \vdash \lambda y.t : (x :_{(s',r')} A') \to B} \ \text{T\_Fun}$$

with $\sigma_1 + \sigma_4 = \sigma_3, r, \sigma'_3$ thus let $\sigma_1 = \sigma'_1, r_1, \sigma''_1$ $\sigma_4 = \sigma'_4, r_2, \sigma''_4$.

By induction we have:

$$(\Delta, (\Delta', \sigma_1)\backslash\pi \mid \sigma'_4, \sigma''_4, r' \mid \mathbf{0}) \odot \Gamma, x : A, \Gamma', y : A' \vdash B : \mathsf{Type}_l$$
$$(\Delta, (\Delta', \sigma_1)\backslash\pi \mid \sigma_2, \sigma'_2, s' \mid \sigma'_4, \sigma''_4, r') \odot \Gamma, \Gamma', y : A' \vdash t : B$$

and thus $\sigma_1\backslash\pi = \sigma'_1, \sigma''_1$.

From which we form the derivation:

$$\frac{\begin{array}{c}(\Delta, (\Delta', \sigma_1)\backslash\pi \mid \sigma'_4, \sigma''_4, r' \mid \mathbf{0}) \odot \Gamma, x : A, \Gamma', y : A' \vdash B : \mathsf{Type}_l \\ (\Delta, (\Delta', \sigma_1)\backslash\pi \mid \sigma_2, \sigma'_2, s' \mid \sigma'_4, \sigma''_4, r') \odot \Gamma, \Gamma', y : A' \vdash t : B\end{array}}{(\Delta, (\Delta'\backslash\pi) \mid \sigma_2, \sigma'_2 \mid \sigma'_1, \sigma''_1 + \sigma'_4, \sigma''_4) \odot \Gamma, \Gamma' \vdash \lambda y.t : (x :_{(s',r')} A') \to B} \ \text{T\_Fun}$$

satisfying the goal here.

---

- Case T_App

$$\frac{\begin{array}{c}(\Delta, \sigma, \Delta', \sigma_1 \mid \sigma_6, r' \mid \mathbf{0}) \odot \Gamma, x : A, \Gamma', y : A' \vdash B : \mathsf{Type}_l \\ (\Delta, \sigma, \Delta' \mid \sigma_5 \mid \sigma_1 + \sigma_6) \odot \Gamma, x : A, \Gamma' \vdash t_1 : (y :_{(s',r')} A') \to B \\ (\Delta, \sigma, \Delta' \mid \sigma_4 \mid \sigma_1) \odot \Gamma, x : A, \Gamma' \vdash t_2 : A'\end{array}}{(\Delta, \sigma, \Delta' \mid \sigma_5 + s' * \sigma_4 \mid \sigma_6 + r' * \sigma_4) \odot \Gamma, x : A, \Gamma' \vdash t_1 \, t_2 : [t_2/x]B} \ \text{T\_App}$$

with $\sigma_5 + s' * \sigma_4 = \sigma_2, s, \sigma'_2$ and $\sigma_6 + r' * \sigma_4 = \sigma_3, r, \sigma'_3$ thus let $\sigma_5 = \sigma'_5, s_1, \sigma''_5$ and $\sigma_4 = \sigma'_4, s_2, \sigma''_4$ and $\sigma_6 = \sigma'_6, r_1, \sigma''_6$ and $\sigma_1 = \sigma'_1, r_3, \sigma''_1$.

therefore we have:
$$\sigma_2 = \sigma_5' + s' * \sigma_4''$$
$$s = s_1 + s' * s_2$$
$$\sigma_2' = \sigma_5'' + s' * \sigma_4''$$
$$\sigma_3 = \sigma_6' + r' * \sigma_4''$$
$$r = r_1 + r' * s_2$$
$$\sigma_3' = \sigma_6'' + r' * \sigma_4''$$

By induction we then have:

$$(\Delta, (\Delta'\backslash\pi), (\sigma_1', \sigma_1'') \mid \sigma_6', \sigma_6'', r' \mid \mathbf{0}) \odot \Gamma, \Gamma', y : A' \vdash B : \mathsf{Type}_l \quad \text{(ih0)}$$
$$(\Delta, (\Delta'\backslash\pi) \mid \sigma_5', \sigma_5'' \mid (\sigma_1' + \sigma_6'), (\sigma_1'' + \sigma_6'')) \odot \Gamma, \Gamma' \vdash t_1 : (y :_{(s',r')} A') \to B \quad \text{(ih1)}$$
$$(\Delta, (\Delta'\backslash\pi) \mid \sigma_4', \sigma_4'' \mid \sigma_1', \sigma_1'') \odot \Gamma, \Gamma' \vdash t_2 : A' \quad \text{(ih2)}$$

Then we can form the derivation:

$$\frac{(\text{ih0}) \quad (\text{ih1}) \quad (\text{ih2})}{(\Delta, (\Delta'\backslash\pi) \mid \sigma_5', \sigma_5'' + s' * \sigma_4', \sigma_4'' \mid \sigma_6', \sigma_6'' + r' * \sigma_4', \sigma_4'') \odot \Gamma, \Gamma' \vdash t_1\, t_2 : [t_2/x]B} \text{ T\_App}$$

whose concluding judgment is equal to :

$$(\Delta, (\Delta'\backslash\pi) \mid \sigma_5' + s' * \sigma_4', \sigma_5'' + s' * \sigma_4'' \mid \sigma_6' + r' * \sigma_4', \sigma_6'' + r' * \sigma_4'') \odot \Gamma, \Gamma' \vdash t_1\, t_2 : [t_2/x]B$$

thus matching the goal

- Case T_Pair

$$\frac{\begin{array}{c}(\Delta, \sigma, \Delta', \sigma_1 \mid \sigma_6, r' \mid \mathbf{0}) \odot \Gamma, x : A, \Gamma', y : A' \vdash B : \mathsf{Type}_l \\ (\Delta, \sigma, \Delta' \mid \sigma_5 \mid \sigma_1) \odot \Gamma, x : A, \Gamma' \vdash t_1 : A' \\ (\Delta, \sigma, \Delta' \mid \sigma_4 \mid \sigma_6 + r' * \sigma_5) \odot \Gamma, x : A, \Gamma' \vdash t_2 : [t_1/x]B \end{array}}{(\Delta, \sigma, \Delta' \mid \sigma_5 + \sigma_4 \mid \sigma_1 + \sigma_6) \odot \Gamma, x : A, \Gamma' \vdash (t_1, t_2) : (y :_{r'} A') \otimes B} \text{ T\_Pair}$$

with $\sigma_2, s, \sigma_2' = \sigma_5 + \sigma_4$ thus we have $\sigma_5 = \sigma_5', s_1, \sigma_5''$ and $\sigma_4 = \sigma_4', s_2, \sigma_4''$ and with $\sigma_3, r, \sigma_3' = \sigma_1 + \sigma_6$ thus we have $\sigma_1 = \sigma_1', r_1, \sigma_1''$ and $\sigma_6 = \sigma_6', r_2, \sigma_6''$ such that:

$$\sigma_2 = \sigma_5' + \sigma_4'$$
$$s = s_1 + s_2$$
$$\sigma_2' = \sigma_5'' + \sigma_4''$$
$$\sigma_3 = \sigma_1' + \sigma_6'$$
$$r = r_1 + r_2$$
$$\sigma_3' = \sigma_1'' + \sigma_6''$$

By induction we have that:

$$(\Delta, (\Delta'\backslash\pi), (\sigma_1', \sigma_1'') \mid \sigma_6', \sigma_6'', r' \mid \mathbf{0}) \odot \Gamma, \Gamma', y : A' \vdash B : \mathsf{Type}_l \quad \text{(ih1)}$$
$$(\Delta, (\Delta'\backslash\pi) \mid \sigma_5', \sigma_5'' \mid \sigma_1', \sigma_1'') \odot \Gamma, \Gamma' \vdash t_1 : A' \quad \text{(ih2)}$$
$$(\Delta, (\Delta'\backslash\pi) \mid \sigma_4', \sigma_4'' \mid \sigma_6' + r' * \sigma_5', \sigma_6'' + r' * \sigma_5'') \odot \Gamma, \Gamma' \vdash t_2 : [t_1/x]B \quad \text{(ih3)}$$

28

From which we construct the application:

$$\frac{(\text{ih1}) \quad (\text{ih2}) \quad (\text{ih3})}{(\Delta, (\Delta'\backslash\pi) \mid \sigma_5', \sigma_5'' + \sigma_4', \sigma_4'' \mid \sigma_1', \sigma_1'' + \sigma_6', \sigma_6'') \odot \Gamma, \Gamma' \vdash (t_1, t_2) : (y :_{r'} A') \otimes B} \text{ T\_PAIR}$$

which is equal to the judgment:

$$(\Delta, (\Delta'\backslash\pi) \mid \sigma_5' + \sigma_4', \sigma_4'' + \sigma_5'' \mid \sigma_1' + \sigma_6', \sigma_1'' + \sigma_6'') \odot \Gamma, \Gamma' \vdash (t_1, t_2) : (y :_{r'} A') \otimes B$$

satisfying the goal of this case.

- Case T_TenCut

$$(\Delta, \sigma, \Delta' \mid \sigma_7 \mid \sigma_1 + \sigma_6) \odot \Gamma, x : A, \Gamma' \vdash t_1 : (w :_{r'} A') \otimes B$$
$$(\Delta, \sigma, \Delta', (\sigma_1 + \sigma_6) \mid \sigma_5, r'' \mid \mathbf{0}) \odot \Gamma, x : A, \Gamma', z : (w :_{r'} A') \otimes B \vdash C : \text{Type}_l$$
$$\frac{(\Delta, \sigma, \Delta', \sigma_1, (\sigma_6, r') \mid \sigma_4, s', s' \mid \sigma_5, r'', r'') \odot \Gamma, x : A, \Gamma', w : A', y : B \vdash t_2 : [(w, y)/z]C}{(\Delta, \sigma, \Delta' \mid \sigma_4 + s' * \sigma_7 \mid \sigma_5 + r'' * \sigma_7) \odot \Gamma, x : A, \Gamma' \vdash \text{let}\,(w, y) = t_1 \,\text{in}\, t_2 : [t_1/z]C} \text{ T\_TenCut}$$

where $\sigma_2, s, \sigma_2' = \sigma_4 + s' * \sigma_7$ and $\sigma_3, r, \sigma_3' = \sigma_5 + r'' * \sigma_7$. Let $\sigma_4 = \sigma_4', s_1, \sigma_4''$ and $\sigma_7 = \sigma_7', s_2, \sigma_7''$ and $\sigma_5 = \sigma_5', r_1, \sigma_5''$ and $\sigma_1 = \sigma_1', r_2, \sigma_1''$ and $\sigma_6 = \sigma_6', r_3, \sigma_6''$ then

$$\sigma_2 = \sigma_4' + s' * \sigma_7'$$
$$s = s_1 + s' * s_2$$
$$\sigma_2' = \sigma_4'' + s' * \sigma_7''$$
$$\sigma_3 = \sigma_5' + r'' * \sigma_5''$$
$$r = r_1 + r'' * s_2$$
$$\sigma_3' = \sigma_5'' + r'' * \sigma_7''$$

Then by induction we have:

$$(\Delta, (\Delta'\backslash\pi) \mid \sigma_7', \sigma_7'' \mid \sigma_1' + \sigma_6', \sigma_1'' + \sigma_6'') \odot \Gamma, \Gamma' \vdash t_1 : (w :_{r'} A') \otimes B \tag{ih1}$$

$$(\Delta, (\Delta'\backslash\pi), (\sigma_1' + \sigma_6', \sigma_1'' + \sigma_6'') \mid \sigma_5', \sigma_5'', r'' \mid \mathbf{0}) \odot \Gamma, \Gamma', z : (w :_{r'} A') \otimes B \vdash C : \text{Type}_l \tag{ih2}$$

$$(\Delta, (\Delta'\backslash\pi), (\sigma_1', \sigma_1''), (\sigma_6', \sigma_6'', r') \mid \sigma_4', \sigma_4'', s', s' \mid \sigma_5', \sigma_5'', r'', r'') \odot \Gamma, \Gamma', w : A', y : B \vdash t_2 : [(w, y)/z]C \tag{ih3}$$

Then we can form the derivation:

$$\frac{(\text{ih1}) \quad (\text{ih2}) \quad (\text{ih3})}{(\Delta, (\Delta'\backslash\pi) \mid \sigma_4', \sigma_4'' + s' * (\sigma_7', \sigma_7'') \mid \sigma_5', \sigma_5'' + r'' * \sigma_7', \sigma_7'') \odot \Gamma, \Gamma' \vdash \text{let}\,(w, y) = t_1 \,\text{in}\, t_2 : [t_1/z]C} \text{ T\_TenCut}$$

giving the equivalent judgment:

$$(\Delta, (\Delta'\backslash\pi) \mid \sigma_4' + s' * \sigma_7', \sigma_4'' + s' * \sigma_7'' \mid \sigma_5' + r'' * \sigma_7', \sigma_5'' + r'' * \sigma_7'') \odot \Gamma, \Gamma' \vdash \text{let}\,(w, y) = t_1 \,\text{in}\, t_2 : [t_1/z]C$$

satisfying the goal.

- Case T_Box

$$\dfrac{(\Delta, \sigma_1, \Delta' \mid \sigma \mid \mathbf{0}) \odot \Gamma, x : A, \Gamma' \vdash A' : \mathsf{Type}_l}{(\Delta, \sigma_1, \Delta' \mid \sigma \mid \mathbf{0}) \odot \Gamma, x : A, \Gamma' \vdash \Box_{s'} A' : \mathsf{Type}_l} \ \text{T\_Box}$$

where $\sigma = \sigma_2, s, \sigma_2'$ and $\mathbf{0} = \sigma_3, r, \sigma_3'$ thus let $\sigma = \sigma', s_1, \sigma''$.

By induction on the first premise we have:

$$(\Delta, (\Delta'\backslash \pi) \mid \sigma', \sigma'' \mid \mathbf{0}) \odot \Gamma, \Gamma' \vdash A' : \mathsf{Type}_l \tag{ih}$$

Thus we can form the derivation:

$$\dfrac{\text{(ih)}}{(\Delta, (\Delta'\backslash \pi) \mid \sigma', \sigma'' \mid \mathbf{0}) \odot \Gamma, \Gamma' \vdash \Box_{s'} A' : \mathsf{Type}_l} \ \text{T\_Box}$$

Satisfying the goal.

- Case T_BoxI

$$\dfrac{(\Delta, \sigma, \Delta' \mid \sigma_1 \mid \sigma_3, r, \sigma_3') \odot \Gamma, x : A, \Gamma' \vdash t : A'}{(\Delta, \sigma, \Delta' \mid s' * \sigma_1 \mid \sigma_3, r, \sigma_3') \odot \Gamma, x : A, \Gamma' \vdash \Box t : \Box_{s'} A'} \ \text{T\_BoxI}$$

where $s' * \sigma_1 = \sigma_2, s, \sigma_2'$ thus let: $\sigma_1 = \sigma_1', s_1, \sigma_1''$ then:

$$\sigma_2 = s' * \sigma_1'$$
$$s = s' * s_1$$
$$\sigma_2' = s' * \sigma_1''$$

By induction on the first premise have that:

$$(\Delta, (\Delta'\backslash \pi) \mid \sigma_1', \sigma_1'' \mid \sigma_3, \sigma_3') \odot \Gamma, \Gamma' \vdash t : A' \tag{ih}$$

$$\dfrac{\text{(ih)}}{(\Delta, (\Delta'\backslash \pi) \mid s' * \sigma_1', \sigma_1'' \mid \sigma_3, \sigma_3') \odot \Gamma, \Gamma' \vdash \Box t : \Box_{s'} A'} \ \text{T\_BoxI}$$

which yields the equivalent judgment:

$$(\Delta, (\Delta'\backslash \pi) \mid s' * \sigma_1', s' * \sigma_1'' \mid \sigma_3, \sigma_3') \odot \Gamma, \Gamma' \vdash \Box t : \Box_{s'} A'$$

satisfying the goal

- Case T_BoxE

$$\dfrac{\begin{array}{c}(\Delta, \sigma, \Delta', \sigma_5 \mid \sigma_4, r' \mid \mathbf{0}) \odot \Gamma, x : A, \Gamma', z : \Box_{s'} A' \vdash B : \mathsf{Type}_l \\ (\Delta, \sigma, \Delta' \mid \sigma_1 \mid \sigma_5) \odot \Gamma, x : A, \Gamma' \vdash t_1 : \Box_{s'} A' \\ (\Delta, \sigma, \Delta', \sigma_5 \mid \sigma_7, s' \mid \sigma_4, (s' * r')) \odot \Gamma, x : A, \Gamma', y : A' \vdash t_2 : [\Box y/z]B\end{array}}{(\Delta, \sigma, \Delta' \mid \sigma_1 + \sigma_7 \mid \sigma_4 + r' * \sigma_1) \odot \Gamma, x : A, \Gamma' \vdash \mathsf{let}\ \Box x = t_1\ \mathsf{in}\ t_2 : [t_1/z]B} \ \text{T\_BoxE}$$

where $\sigma_2, s, \sigma_2' = \sigma_1 + \sigma_7$ and $\sigma_3, r, \sigma_3' = \sigma_4 + r' * \sigma_1$. Let $\sigma_1 = \sigma_1', s_1, \sigma_1''$ and $\sigma_7 = \sigma_7', s_2, \sigma_7''$ and $\sigma_6 = \sigma_6', r_1, \sigma_6''$ and $\sigma_4 = \sigma_4', r_2, \sigma_4''$ and $\sigma_5 = \sigma_5', r_3, \sigma_5''$ then we have that:

$$\sigma_2 = \sigma_1' + \sigma_7'$$
$$s = s_1 + s_2$$
$$\sigma_2' = \sigma_1'' + \sigma_7''$$
$$\sigma_3 = \sigma_4' + r' * \sigma_6'$$
$$r = r_1 + r_2$$
$$\sigma_3' = \sigma_4'' + r' * \sigma_6''$$

Thus by induction we get that:

$$(\Delta, (\Delta'\backslash\pi), (\sigma_5', \sigma_5'') \mid \sigma_4', \sigma_4'', r' \mid \mathbf{0}) \odot \Gamma, \Gamma', z : \Box_{s'} A' \vdash B : \mathsf{Type}_l \qquad \text{(ih0)}$$
$$(\Delta, (\Delta'\backslash\pi) \mid \sigma_1', \sigma_1'' \mid \sigma_5', \sigma_5'') \odot \Gamma, \Gamma' \vdash t_1 : \Box_{s'} A' \qquad \text{(ih1)}$$
$$(\Delta, (\Delta'\backslash\pi), (\sigma_5', \sigma_5'') \mid \sigma_7', \sigma_7'', s' \mid \sigma_4', \sigma_4'', (s' * r')) \odot \Gamma, \Gamma', y : A' \vdash t_2 : [\Box y/z]B \qquad \text{(ih2)}$$

then we can form the derivation:

$$\frac{\text{(ih0)} \quad \text{(ih1)} \quad \text{(ih2)}}{(\Delta, (\Delta'\backslash\pi) \mid \sigma_1', \sigma_1'' + \sigma_7', \sigma_7'' \mid \sigma_4', \sigma_4'' + r' * \sigma_6', \sigma_6'') \odot \Gamma, \Gamma' \vdash \mathsf{let}\,\Box x = t_1\,\mathsf{in}\,t_2 : [t_1/z]B} \text{ T\_BoxE}$$

which is equal to the judgment:

$$(\Delta, (\Delta'\backslash\pi) \mid \sigma_1' + \sigma_7', \sigma_1'' + \sigma_7'' \mid (\sigma_4' + r' * \sigma_6'), (\sigma_4'' + r' * \sigma_6'')) \odot \Gamma, \Gamma' \vdash \mathsf{let}\,\Box x = t_1\,\mathsf{in}\,t_2 : [t_1/z]B$$

satisfying the goal.

- Case T_TyConv

$$\frac{(\Delta, \sigma_1, \Delta' \mid \sigma_2, s, \sigma_2' \mid \sigma_3, r, \sigma_3') \odot \Gamma, x : A, \Gamma' \vdash t : C \quad (\Delta, \sigma_1, \Delta' \mid \sigma_3, r, \sigma_3') \odot \Gamma, x : A, \Gamma' \vdash C \leq B}{(\Delta, \sigma_1, \Delta' \mid \sigma_2, s, \sigma_2' \mid \sigma_3, r, \sigma_3') \odot \Gamma, x : A, \Gamma' \vdash t : B} \text{ T\_TyConv}$$

By induction we have that:

$$(\Delta, (\Delta'\backslash\pi) \mid \sigma_2, \sigma_2' \mid \sigma_3, \sigma_3') \odot \Gamma, \Gamma' \vdash t : C \qquad \text{(ih1)}$$
$$(\Delta, (\Delta'\backslash\pi) \mid \sigma_3, \sigma_3') \odot \Gamma, \Gamma' \vdash C \leq B \qquad \text{(ih2)}$$

Then we can build the judgment:

$$\frac{\text{(ih1)} \quad \text{(ih2)}}{(\Delta, (\Delta'\backslash\pi) \mid \sigma_2, \sigma_2' \mid \sigma_3, \sigma_3') \odot \Gamma, \Gamma' \vdash t : B} \text{ T\_TyConv}$$

As required.

For subtyping, by standard induction and re-application (see Section A.8).
For equality, by standard induction and re-application (see Section A.8). $\square$

**Lemma 3.3** (Judgmental contexts are well-formed)**.** *If* $(\Delta \mid \sigma_1 \mid \sigma_2) \odot \Gamma \vdash \mathcal{J}$, *then* $\Delta \odot \Gamma \vdash$.

*Proof.* For well-formedness: our goal is the premise.
For typing: By induction on the form of $(\Delta \mid \sigma_1 \mid \sigma_2) \odot \Gamma \vdash t : A$, as follows:

---

Case.

$$\frac{\Delta \odot \Gamma \vdash}{(\Delta \mid \mathbf{0} \mid \mathbf{0}) \odot \Gamma \vdash \mathsf{Type}_l : \mathsf{Type}_{\mathsf{suc}\ l}} \text{ T\_Type}$$

We need to show $\Delta \odot \Gamma \vdash$, which holds by premise.

---

Case.

$$\frac{(\Delta \mid \sigma_1 \mid \mathbf{0}) \odot \Gamma \vdash A : \mathsf{Type}_{l_1} \quad (\Delta, \sigma_1 \mid \sigma_2, r \mid \mathbf{0}) \odot \Gamma, x : A \vdash B : \mathsf{Type}_{l_2}}{(\Delta \mid \sigma_1 + \sigma_2 \mid \mathbf{0}) \odot \Gamma \vdash (x :_{(s,r)} A) \to B : \mathsf{Type}_{l_1 \sqcup l_2}} \text{ T\_Arrow}$$

This holds by induction on the typing premise for $A$. Each of the cases for T_Ten, T_App, T_Pair, T_Box, T_BoxI, T_BoxE, and T_TenCut proceed similarly.

---

Case.

$$\frac{\Delta_1, \sigma, \Delta_2 \odot \Gamma_1, x : A, \Gamma_2 \vdash \quad |\Delta_1| = |\Gamma_1|}{(\Delta_1, \sigma, \Delta_2 \mid \mathbf{0}^{|\Delta_1|}, 1, \mathbf{0} \mid \sigma, 0, \mathbf{0}) \odot \Gamma_1, x : A, \Gamma_2 \vdash x : A} \text{ T\_Var}$$

We need to show $\Delta_1, \sigma_1, \Delta_2 \odot \Gamma_1, x : A, \Gamma_2 \vdash$, which holds by premise.

---

Case.

$$\frac{(\Delta, \sigma_1 \mid \sigma_3, r \mid \mathbf{0}) \odot \Gamma, x : A \vdash B : \mathsf{Type}_l \quad (\Delta, \sigma_1 \mid \sigma_2, s \mid \sigma_3, r) \odot \Gamma, x : A \vdash t : B}{(\Delta \mid \sigma_2 \mid \sigma_1 + \sigma_3) \odot \Gamma \vdash \lambda x.t : (x :_{(s,r)} A) \to B} \text{ T\_Fun}$$

Our goal is $\Delta \odot \Gamma \vdash$. By induction on the typing premise for $t$, we know that $\Delta, \sigma_1 \odot \Gamma, x : A \vdash$, therefore, we have $\Delta \odot \Gamma \vdash$ by Lemma 3.4, as required.

---

Case.

$$\frac{(\Delta \mid \sigma_1 \mid \sigma_2) \odot \Gamma \vdash t : A \quad (\Delta \mid \sigma_2) \odot \Gamma \vdash A \leq B}{(\Delta \mid \sigma_1 \mid \sigma_2) \odot \Gamma \vdash t : B} \text{ T\_TyConv}$$

Our goal is $\Delta \odot \Gamma \vdash$. This holds by induction on the typing premise for $(\Delta \mid \sigma_1 \mid \sigma_2) \odot \Gamma \vdash t : A$.

---

For subtyping: all cases proceed trivially by induction.

For equality: all cases proceed trivially by induction, with the exception of TEQ_FUN, which proceeds similarly to T_FUN. □

**Lemma 3.2** (Judgments determine vector sizing). *If $(\Delta \mid \sigma_1 \mid \sigma_2) \odot \Gamma \vdash \mathcal{J}$, then $|\sigma_1| = |\sigma_2| = |\Delta| = |\Gamma|$ and for each element $\sigma$ of $\Delta$, the size of $\sigma$ is the same as its index.*

*Proof.* This holds by assumption over judgments. This may also be proven inductively, by adding size annotations to occurrences of $\mathbf{0}$. □

**Lemma 3.4** (Subcontext well-formedness). *If $|\Delta_1| = |\Gamma_1|$ and $\Delta_1, \Delta_2 \odot \Gamma_1, \Gamma_2 \vdash$, then $\Delta_1 \odot \Gamma_1 \vdash$.*

*Proof.* By induction on the definition of well-formed contexts.

- Case WF_EMPTY. Trivial since it does not match the form of the lemma.
- Case WF_EXT We consider two cases depending on the syntactic structure of $\Gamma_2$ and $\Delta_2$ (simultaneously, since they must have the same size by the lemma statement).
    - $\Gamma_2 = \emptyset$ and $\Delta_2 = \emptyset$ then this lemma is trivial since we have $\Delta_1 \odot \Gamma_1 \vdash$ already.
    - $\Gamma_2 = \Gamma_2', x : A$ and $\Delta_2 = \Delta_2', \sigma$ then we have:
    
    $$\frac{(\Delta_1, \Delta_2' \mid \sigma \mid \mathbf{0}) \odot \Gamma_1, \Gamma_2 \vdash A : \mathsf{Type}_l}{\Delta_1, \Delta_2', \sigma \odot \Gamma_1, \Gamma_2, x : A \vdash} \text{ WF\_EXT}$$
    
    By Lemma 3.3 then we have $\Delta_1, \Delta_2' \odot \Gamma_1, \Gamma_2 \vdash$
    By induction on this then we get that $\Delta_1' \odot \Gamma_1 \vdash$ which is the goal of this case.

□

**Lemma 3.5** (Typing an assumption in a judgmental context). *If $(\Delta_1, \sigma, \Delta_2 \mid \sigma_1 \mid \sigma_2) \odot \Gamma_1, x : A, \Gamma_2 \vdash \mathcal{J}$ with $|\Delta_1| = |\Gamma_1|$, then $(\Delta_1 \mid \sigma \mid \mathbf{0}) \odot \Gamma_1 \vdash A : \mathsf{Type}_l$, for some level $l$.*

*Proof.* By Lemma 3.3 we have $\Delta_1, \sigma, \Delta_2 \odot \Gamma_1, x : A, \Gamma_2 \vdash$, therefore by Lemma 3.4 we have $\Delta_1, \sigma \odot \Gamma_1, x : A \vdash$, which must come from the following derivation:

$$\frac{(\Delta_1 \mid \sigma \mid \mathbf{0}) \odot \Gamma_1 \vdash A : \mathsf{Type}_l}{\Delta_1, \sigma \odot \Gamma_1, x : A \vdash} \text{ WF\_EXT}$$

Therefore we have $(\Delta_1 \mid \sigma \mid \mathbf{0}) \odot \Gamma_1 \vdash A : \mathsf{Type}_l$ by the premise, as required. □

**Lemma 3.6** (Typing the type of a term). *Given $(\Delta \mid \sigma_1 \mid \sigma_2) \odot \Gamma \vdash t : A$, then there exists a level $l$ such that $(\Delta \mid \sigma_2 \mid \mathbf{0}) \odot \Gamma \vdash A : \mathsf{Type}_l$.*

*Proof.* By induction on the form of $(\Delta \mid \sigma_1 \mid \sigma_2) \odot \Gamma \vdash t : A$, as follows:

> Case.
>
> $$\frac{\Delta \odot \Gamma \vdash}{(\Delta \mid \mathbf{0} \mid \mathbf{0}) \odot \Gamma \vdash \mathsf{Type}_l : \mathsf{Type}_{\mathsf{suc}\ l}} \text{ T\_TYPE}$$
>
> Then we will show $(\Delta \mid \mathbf{0} \mid \mathbf{0}) \odot \Gamma \vdash \mathsf{Type}_{\mathsf{suc}\ l} : \mathsf{Type}_{\mathsf{suc}\ \mathsf{suc}\ l}$ which holds by T_TYPE and the well-formedness premise for $\Gamma$.

Case.

$$\frac{(\Delta \mid \sigma_1 \mid \mathbf{0}) \odot \Gamma \vdash A : \mathsf{Type}_{l_1} \quad (\Delta, \sigma_1 \mid \sigma_2, r \mid \mathbf{0}) \odot \Gamma, x : A \vdash B : \mathsf{Type}_{l_2}}{(\Delta \mid \sigma_1 + \sigma_2 \mid \mathbf{0}) \odot \Gamma \vdash (x :_{(s,r)} A) \to B : \mathsf{Type}_{l_1 \sqcup l_2}} \text{ T\_Arrow}$$

Then we will show $(\Delta \mid \mathbf{0} \mid \mathbf{0}) \odot \Gamma \vdash \mathsf{Type}_{l_1 \sqcup l_2} : \mathsf{Type}_{l_3}$ for some level $l_3$. Consider $l_2 \le l_1$, then we have $l_1 \sqcup l_2 = l_1$, and thus need to show $(\Delta \mid \mathbf{0} \mid \mathbf{0}) \odot \Gamma \vdash \mathsf{Type}_{l_1} : \mathsf{Type}_{l_3}$, which holds by induction on the typing premise for $A$. Consider $l_1 < l_2$, then we have $l_1 \sqcup l_2 = l_2$, and thus need to show $(\Delta \mid \mathbf{0} \mid \mathbf{0}) \odot \Gamma \vdash \mathsf{Type}_{l_2} : \mathsf{Type}_{l_3}$. By induction on the typing premise for $B$, we have $(\Delta, \sigma_1 \mid \mathbf{0} \mid \mathbf{0}) \odot \Gamma, x : A \vdash \mathsf{Type}_{l_2} : \mathsf{Type}_{l_3}$, therefore, by Lemma 3.1, we have $(\Delta \mid \mathbf{0} \mid \mathbf{0}) \odot \Gamma \vdash \mathsf{Type}_{l_2} : \mathsf{Type}_{l_3}$, as required.
The proof for T_Ten follows the same process.

Case.

$$\frac{\Delta_1, \sigma, \Delta_2 \odot \Gamma_1, x : A, \Gamma_2 \vdash \quad |\Delta_1| = |\Gamma_1|}{(\Delta_1, \sigma, \Delta_2 \mid \mathbf{0}^{|\Delta_1|}, 1, \mathbf{0} \mid \sigma, 0, \mathbf{0}) \odot \Gamma_1, x : A, \Gamma_2 \vdash x : A} \text{ T\_Var}$$

Then we will show $(\Delta_1, \sigma, \Delta_2 \mid \sigma, 0, \mathbf{0} \mid \mathbf{0}) \odot \Gamma_1, x : A, \Gamma_2 \vdash A : \mathsf{Type}_l$ (for some level $l$). This holds by the following derivation:

$$\frac{\Delta_1, \sigma, \Delta_2 \odot \Gamma_1, x : A, \Gamma_2 \vdash \quad \dfrac{\Delta_1, \sigma, \Delta_2 \odot \Gamma_1, x : A, \Gamma_2 \vdash \quad |\Delta_1| = |\Gamma_1|}{(\Delta_1 \mid \sigma \mid \mathbf{0}) \odot \Gamma_1 \vdash A : \mathsf{Type}_l} \text{ Lemma 3.5}}{(\Delta_1, \sigma, \Delta_2 \mid \sigma, 0, \mathbf{0} \mid \mathbf{0}) \odot \Gamma_1, x : A, \Gamma_2 \vdash A : \mathsf{Type}_l} \text{ Lemma 3.25}$$

Case.

$$\frac{(\Delta, \sigma_1 \mid \sigma_3, r \mid \mathbf{0}) \odot \Gamma, x : A \vdash B : \mathsf{Type}_l \quad (\Delta, \sigma_1 \mid \sigma_2, s \mid \sigma_3, r) \odot \Gamma, x : A \vdash t : B}{(\Delta \mid \sigma_2 \mid \sigma_1 + \sigma_3) \odot \Gamma \vdash \lambda x.t : (x :_{(s,r)} A) \to B} \text{ T\_Fun}$$

Then we will show $(\Delta \mid \sigma_1 + \sigma_3 \mid \mathbf{0}) \odot \Gamma \vdash (x :_{(s,r)} A) \to B : \mathsf{Type}_{l'}$ for some level $l'$. Using Lemma 3.5 on the typing premise for $t$ gives $(\Delta \mid \sigma_1 \mid \mathbf{0}) \odot \Gamma \vdash A : \mathsf{Type}_{l_1}$ (for some level $l_1$). We can then achieve our goal via an application to T_Arrow, as follows:

$$\frac{(\Delta \mid \sigma_1 \mid \mathbf{0}) \odot \Gamma \vdash A : \mathsf{Type}_{l_1} \quad (\Delta, \sigma_1 \mid \sigma_3, r \mid \mathbf{0}) \odot \Gamma, x : A \vdash B : \mathsf{Type}_l}{(\Delta \mid \sigma_1 + \sigma_3 \mid \mathbf{0}) \odot \Gamma \vdash (x :_{(s,r)} A) \to B : \mathsf{Type}_{l_1 \sqcup l}} \text{ T\_Arrow}$$

Case.

**Case.**

$$\dfrac{(\Delta, \sigma_1 \mid \sigma_3, r \mid \mathbf{0}) \odot \Gamma, x : A \vdash B : \mathsf{Type}_l \quad (\Delta \mid \sigma_2 \mid \sigma_1 + \sigma_3) \odot \Gamma \vdash t_1 : (x :_{(s,r)} A) \to B \quad (\Delta \mid \sigma_4 \mid \sigma_1) \odot \Gamma \vdash t_2 : A}{(\Delta \mid \sigma_2 + s * \sigma_4 \mid \sigma_3 + r * \sigma_4) \odot \Gamma \vdash t_1\, t_2 : [t_2/x]B}\ \text{T\_App}$$

Then we will show $(\Delta \mid \sigma_3 + r * \sigma_4 \mid \mathbf{0}) \odot \Gamma \vdash [t_2/x]B : \mathsf{Type}_l$. We have the following:

1. $(\Delta \mid \sigma_4 \mid \sigma_1) \odot \Gamma \vdash t_2 : A$
2. $(\Delta, \sigma_1 \mid \sigma_3, r \mid \mathbf{0}) \odot \Gamma, x : A \vdash B : \mathsf{Type}_l$
3. $|\sigma_3| = |\mathbf{0}| = |\Gamma|$ (trivially, and by Lemma 3.2)

Which we can pass to Lemma 3.27, to obtain $(\Delta \mid \sigma_3 + r * \sigma_4 \mid \mathbf{0}) \odot \Gamma \vdash [t_2/x]B : \mathsf{Type}_l$, as required.

**Case.**

$$\dfrac{(\Delta, \sigma_1 \mid \sigma_3, r \mid \mathbf{0}) \odot \Gamma, x : A \vdash B : \mathsf{Type}_l \quad (\Delta \mid \sigma_2 \mid \sigma_1) \odot \Gamma \vdash t_1 : A \quad (\Delta \mid \sigma_4 \mid \sigma_3 + r * \sigma_2) \odot \Gamma \vdash t_2 : [t_1/x]B}{(\Delta \mid \sigma_2 + \sigma_4 \mid \sigma_1 + \sigma_3) \odot \Gamma \vdash (t_1, t_2) : (x :_r A) \otimes B}\ \text{T\_Pair}$$

Then we will show $(\Delta \mid \sigma_1 + \sigma_3 \mid \mathbf{0}) \odot \Gamma \vdash (x :_r A) \otimes B : \mathsf{Type}_{l'}$ for some level $l'$. Applying Lemma 3.5 to the typing premise for $B$, we have $(\Delta \mid \sigma_1 \mid \mathbf{0}) \odot \Gamma \vdash A : \mathsf{Type}_{l''}$ (for some level $l''$). We can then apply this result, along with the typing premise for $B$ to T\_Ten, to obtain:

$$\dfrac{(\Delta \mid \sigma_1 \mid \mathbf{0}) \odot \Gamma \vdash A : \mathsf{Type}_{l''} \quad (\Delta, \sigma_1 \mid \sigma_3, r \mid \mathbf{0}) \odot \Gamma, x : A \vdash B : \mathsf{Type}_l}{(\Delta \mid \sigma_1 + \sigma_3 \mid \mathbf{0}) \odot \Gamma \vdash (x :_r A) \otimes B : \mathsf{Type}_{l'' \sqcup l}}\ \text{T\_Ten}$$

As required.

**Case.**

$$\dfrac{\begin{array}{c}(\Delta \mid \sigma_3 \mid \sigma_1 + \sigma_2) \odot \Gamma \vdash t_1 : (x :_r A) \otimes B \\ (\Delta, (\sigma_1 + \sigma_2) \mid \sigma_5, r' \mid \mathbf{0}) \odot \Gamma, z : (x :_r A) \otimes B \vdash C : \mathsf{Type}_l \\ (\Delta, \sigma_1, (\sigma_2, r) \mid \sigma_4, s, s \mid \sigma_5, r', r') \odot \Gamma, x : A, y : B \vdash t_2 : [(x, y)/z]C\end{array}}{(\Delta \mid \sigma_4 + s * \sigma_3 \mid \sigma_5 + r' * \sigma_3) \odot \Gamma \vdash \mathsf{let}\,(x, y) = t_1\ \mathsf{in}\ t_2 : [t_1/z]C}\ \text{T\_TenCut}$$

Then we will show $(\Delta \mid \sigma_5 + r' * \sigma_3 \mid \mathbf{0}) \odot \Gamma \vdash [t_1/z]C : \mathsf{Type}_l$. We can form the following premises:

1. $(\Delta \mid \sigma_3 \mid \sigma_1 + \sigma_2) \odot \Gamma \vdash t_1 : (x :_r A) \otimes B$
2. $(\Delta, (\sigma_1 + \sigma_2) \mid \sigma_5, r' \mid \mathbf{0}) \odot \Gamma, z : (x :_r A) \otimes B \vdash C : \mathsf{Type}_l$
3. $|\sigma_5| = |\mathbf{0}| = |\Gamma|$ (trivially, and by Lemma 3.2)

Which we can pass to Lemma 3.27, to obtain $(\Delta \mid \sigma_5 + r' * \sigma_3 \mid \mathbf{0}) \odot \Gamma \vdash [t_1/z]C : \mathsf{Type}_l$, as required.

**Case.**

$$\frac{(\Delta \mid \sigma \mid \mathbf{0}) \odot \Gamma \vdash A : \mathsf{Type}_l}{(\Delta \mid \sigma \mid \mathbf{0}) \odot \Gamma \vdash \Box_s A : \mathsf{Type}_l} \; \text{T\_Box}$$

Then we will show $(\Delta \mid \mathbf{0} \mid \mathbf{0}) \odot \Gamma \vdash \mathsf{Type}_l : \mathsf{Type}_{\mathsf{suc}\; l}$ which holds by the following derivation:

$$\frac{\dfrac{(\Delta \mid \sigma \mid \mathbf{0}) \odot \Gamma \vdash A : \mathsf{Type}_l}{\Delta \odot \Gamma \vdash} \; \text{L. 3.3}}{(\Delta \mid \mathbf{0} \mid \mathbf{0}) \odot \Gamma \vdash \mathsf{Type}_l : \mathsf{Type}_{\mathsf{suc}\; l}} \; \text{T\_Type}$$

Case.

$$\frac{(\Delta \mid \sigma_1 \mid \sigma_2) \odot \Gamma \vdash t : A}{(\Delta \mid s * \sigma_1 \mid \sigma_2) \odot \Gamma \vdash \Box t : \Box_s A} \; \text{T\_BoxI}$$

Then we will show $(\Delta \mid \sigma_2 \mid \mathbf{0}) \odot \Gamma \vdash \Box_s A : \mathsf{Type}_l$, for some level $l$. By induction on the typing premise for $t$, we have $(\Delta \mid \sigma_2 \mid \mathbf{0}) \odot \Gamma \vdash A : \mathsf{Type}_l$, therefore, we can perform the following application to achieve our goal:

$$\frac{(\Delta \mid \sigma_2 \mid \mathbf{0}) \odot \Gamma \vdash A : \mathsf{Type}_l}{(\Delta \mid \sigma_2 \mid \mathbf{0}) \odot \Gamma \vdash \Box_s A : \mathsf{Type}_l} \; \text{T\_Box}$$

Case.

$$\frac{\begin{array}{c}(\Delta \mid \sigma_1 \mid \sigma_2) \odot \Gamma \vdash t_1 : \Box_s A \\ (\Delta, \sigma_2 \mid \sigma_4, r \mid \mathbf{0}) \odot \Gamma, z : \Box_s A \vdash B : \mathsf{Type}_l \\ (\Delta, \sigma_2 \mid \sigma_3, s \mid \sigma_4, (s*r)) \odot \Gamma, x : A \vdash t_2 : [\Box x/z]B\end{array}}{(\Delta \mid \sigma_1 + \sigma_3 \mid \sigma_4 + r * \sigma_1) \odot \Gamma \vdash \mathsf{let}\,\Box x = t_1\,\mathsf{in}\,t_2 : [t_1/z]B} \; \text{T\_BoxE}$$

Then we will show $(\Delta \mid \sigma_4 + r * \sigma_1 \mid \mathbf{0}) \odot \Gamma \vdash [t_1/z]B : \mathsf{Type}_l$. We can form the following premises:

1. $(\Delta \mid \sigma_1 \mid \sigma_2) \odot \Gamma \vdash t_1 : \Box_s A$
2. $(\Delta, \sigma_2 \mid \sigma_4, r \mid \mathbf{0}) \odot \Gamma, z : \Box_s A \vdash B : \mathsf{Type}_l$
3. $|\sigma_4| = |\mathbf{0}| = |\Gamma|$ (trivially, and by Lemma 3.2)

Which we can pass to Lemma 3.27, to obtain $(\Delta \mid \sigma_4 + r * \sigma_1 \mid \mathbf{0}) \odot \Gamma \vdash [t_1/z]B : \mathsf{Type}_l$, as required.

Case.

$$\frac{(\Delta \mid \sigma_1 \mid \sigma_2) \odot \Gamma \vdash t : A \quad (\Delta \mid \sigma_2) \odot \Gamma \vdash A \leq B}{(\Delta \mid \sigma_1 \mid \sigma_2) \odot \Gamma \vdash t : B} \; \text{T\_TyConv}$$

Then we will show $(\Delta \mid \sigma_2 \mid \mathbf{0}) \odot \Gamma \vdash B : \mathsf{Type}_l$, which holds by Lemma 3.34.

**Lemma 3.7** (Subtyping in context). *If $(\Delta_1, \sigma_1, \Delta_2 \mid \sigma_2, s, \sigma_3 \mid \sigma_4, r, \sigma_5) \odot \Gamma_1, x : A, \Gamma_2 \vdash \mathcal{J}$ and $(\Delta_1 \mid \sigma_1) \odot \Gamma \vdash A' \leq A$, then $(\Delta_1, \sigma_1, \Delta_2 \mid \sigma_2, s, \sigma_3 \mid \sigma_4, r, \sigma_5) \odot \Gamma_1, x : A', \Gamma_2 \vdash \mathcal{J}$.*

*Proof.* For well-formed contexts:

---

Case.

$$\frac{(\Delta \mid \sigma \mid \mathbf{0}) \odot \Gamma \vdash A : \mathsf{Type}_l}{\Delta, \sigma \odot \Gamma, x : A \vdash} \text{W}_\text{F}\text{\_E}\text{XT}$$

With $(\Delta \mid \mathbf{0}) \odot \Gamma \vdash A' \leq A$. Then our goal holds by the following derivation:

$$\frac{\dfrac{(\Delta \mid \sigma) \odot \Gamma \vdash A' \leq A}{(\Delta \mid \sigma \mid \mathbf{0}) \odot \Gamma \vdash A' : \mathsf{Type}_l} \text{ L. 3.34}}{\Delta, \sigma \odot \Gamma, x : A' \vdash} \text{W}_\text{F}\text{\_E}\text{XT}$$

---

Case.

$$\frac{(\Delta_1, \sigma_1, \Delta_2 \mid \sigma_2 \mid \mathbf{0}) \odot \Gamma_1, x : A, \Gamma_2 \vdash B : \mathsf{Type}_{l'}}{\Delta_1, \sigma_1, \Delta_2, \sigma_2 \odot \Gamma_1, x : A, \Gamma_2, y : B \vdash} \text{W}_\text{F}\text{\_E}\text{XT}$$

Then our goal holds by the following derivation:

$$\frac{\dfrac{(\Delta_1, \sigma_1, \Delta_2 \mid \sigma_2 \mid \mathbf{0}) \odot \Gamma_1, x : A, \Gamma_2 \vdash B : \mathsf{Type}_{l'} \quad (\Delta_1 \mid \sigma_1) \odot \Gamma_1 \vdash A' \leq A}{(\Delta_1, \sigma_1, \Delta_2 \mid \sigma_2 \mid \mathbf{0}) \odot \Gamma_1, x : A', \Gamma_2 \vdash B : \mathsf{Type}_{l'}} \text{ I.H.}}{\Delta_1, \sigma_1, \Delta_2, \sigma_2 \odot \Gamma_1, x : A', \Gamma_2, y : B \vdash} \text{W}_\text{F}\text{\_E}\text{XT}$$

---

For typing: most cases hold by induction and reapplication. Use $\dfrac{(\Delta \mid \sigma_1 \mid \sigma_2) \odot \Gamma \vdash t : A \quad (\Delta \mid \sigma_2) \odot \Gamma \vdash A \leq B}{(\Delta \mid \sigma_1 \mid \sigma_2) \odot \Gamma \vdash t : B}$ T_TyCon when necessary. For equality, all cases proceed by induction then re-application to respective rules. For subtyping, all cases proceed by induction then re-application to respective rules. □

## A.1 Proofs for inversions

**Lemma 3.8** (Inversion on arrow typing). *If $(\Delta \mid \sigma_1 \mid \sigma_2) \odot \Gamma \vdash (x :_{(s,r)} A) \to B : C$ then there exist grade vectors $\hat{\sigma_1}$, $\hat{\sigma_1}'$, and levels $l$ and $l'$, such that $(\Delta \mid \hat{\sigma_1} \mid \mathbf{0}) \odot \Gamma \vdash A : \mathsf{Type}_{l'}$ and $(\Delta, \hat{\sigma_1} \mid \hat{\sigma_1}', r \mid \mathbf{0}) \odot \Gamma, x : A \vdash B : \mathsf{Type}_{l''}$, and $\hat{\sigma_1} + \hat{\sigma_1}' = \sigma_1$.*

*Proof.* By induction on $(\Delta \mid \sigma_1 \mid \sigma_2) \odot \Gamma \vdash (x :_{(s,r)} A) \to B : C$, as follows:

---

Case.

$$\frac{(\Delta \mid \sigma_1 \mid \mathbf{0}) \odot \Gamma \vdash A : \mathsf{Type}_{l_1} \quad (\Delta, \sigma_1 \mid \sigma_2, r \mid \mathbf{0}) \odot \Gamma, x : A \vdash B : \mathsf{Type}_{l_2}}{(\Delta \mid \sigma_1 + \sigma_2 \mid \mathbf{0}) \odot \Gamma \vdash (x :_{(s,r)} A) \to B : \mathsf{Type}_{l_1 \sqcup l_2}} \text{T\_A}\text{RROW}$$

---

Then we have $(\Delta \mid \sigma_1 \mid \mathbf{0}) \odot \Gamma \vdash A : \mathsf{Type}_{l_1}$, $(\Delta, \sigma_1 \mid \sigma_2, r \mid \mathbf{0}) \odot \Gamma, x : A \vdash B : \mathsf{Type}_{l_2}$, with $\sigma_1 + \sigma_2 = \sigma_1 + \sigma_2$, as required.

Case.

$$\dfrac{(\Delta \mid \sigma_1 \mid \sigma_2) \odot \Gamma \vdash (x :_{(s,r)} A) \to B : D \quad (\Delta \mid \sigma_2) \odot \Gamma \vdash D \leq C}{(\Delta \mid \sigma_1 \mid \sigma_2) \odot \Gamma \vdash (x :_{(s,r)} A) \to B : C} \text{ T\_TyConv}$$

Then our goal holds by induction.

$\square$

**Lemma 3.9** (Inversion on tensor typing). *If* $(\Delta \mid \sigma_1 \mid \sigma_2) \odot \Gamma \vdash (x :_r A) \otimes B : C$ *then there exist grade vectors* $\hat{\sigma_1}$, $\hat{\sigma_1}'$, *and levels* $l$ *and* $l'$, *such that* $(\Delta \mid \hat{\sigma_1} \mid \mathbf{0}) \odot \Gamma \vdash A : \mathsf{Type}_{l'}$ *and* $(\Delta, \hat{\sigma_1} \mid \hat{\sigma_1}', r \mid \mathbf{0}) \odot \Gamma, x : A \vdash B : \mathsf{Type}_{l''}$, *and* $\hat{\sigma_1} + \hat{\sigma_1}' = \sigma_1$.

*Proof.* By induction on $(\Delta \mid \sigma_1 \mid \sigma_2) \odot \Gamma \vdash (x :_r A) \otimes B : C$, as follows:

Case.

$$\dfrac{(\Delta \mid \sigma_1 \mid \mathbf{0}) \odot \Gamma \vdash A : \mathsf{Type}_{l_1} \quad (\Delta, \sigma_1 \mid \sigma_2, r \mid \mathbf{0}) \odot \Gamma, x : A \vdash B : \mathsf{Type}_{l_2}}{(\Delta \mid \sigma_1 + \sigma_2 \mid \mathbf{0}) \odot \Gamma \vdash (x :_r A) \otimes B : \mathsf{Type}_{l_1 \sqcup l_2}} \text{ T\_Ten}$$

Then we have $(\Delta \mid \sigma_1 \mid \mathbf{0}) \odot \Gamma \vdash A : \mathsf{Type}_{l_1}$, $(\Delta, \sigma_1 \mid \sigma_2, r \mid \mathbf{0}) \odot \Gamma, x : A \vdash B : \mathsf{Type}_{l_2}$, with $\sigma_1 + \sigma_2 = \sigma_1 + \sigma_2$, as required.

Case.

$$\dfrac{(\Delta \mid \sigma_1 \mid \sigma_2) \odot \Gamma \vdash (x :_r A) \otimes B : D \quad (\Delta \mid \sigma_2) \odot \Gamma \vdash D \leq C}{(\Delta \mid \sigma_1 \mid \sigma_2) \odot \Gamma \vdash (x :_r A) \otimes B : C} \text{ T\_TyConv}$$

Then our goal holds by induction.

$\square$

**Lemma 3.10** (Inversion on box typing). *If* $(\Delta \mid \sigma_1 \mid \sigma_2) \odot \Gamma \vdash \Box_s A : B$, *then* $(\Delta \mid \sigma_1 \mid \sigma_2) \odot \Gamma \vdash A : \mathsf{Type}_l$, *for some level* $l$.

*Proof.* By induction on $(\Delta \mid \sigma_1 \mid \sigma_2) \odot \Gamma \vdash \Box_s A : B$, as follows:

Case.

$$\dfrac{(\Delta \mid \sigma \mid \mathbf{0}) \odot \Gamma \vdash A : \mathsf{Type}_l}{(\Delta \mid \sigma \mid \mathbf{0}) \odot \Gamma \vdash \Box_s A : \mathsf{Type}_l} \text{ T\_Box}$$

Then we have $(\Delta \mid \sigma \mid \mathbf{0}) \odot \Gamma \vdash A : \mathsf{Type}_l$, as required.

Case.

$$\dfrac{(\Delta \mid \sigma_1 \mid \sigma_2) \odot \Gamma \vdash \Box_s A : C \quad (\Delta \mid \sigma_2) \odot \Gamma \vdash C \leq B}{(\Delta \mid \sigma_1 \mid \sigma_2) \odot \Gamma \vdash \Box_s A : B} \text{ T\_TyConv}$$

Then our goal holds by induction.

□

**Lemma 3.11** (Function inversion). *If $(\Delta \mid \sigma_2 \mid \sigma_1 + \sigma_3) \odot \Gamma \vdash \lambda x.t : C$, $(\Delta \mid \sigma_1 + \sigma_3) \odot \Gamma \vdash C \leq (x :_{(s,r)} A) \to B$, and $(\Delta, \sigma_1 \mid \sigma_3, r \mid \mathbf{0}) \odot \Gamma, x : A \vdash B : \mathsf{Type}_l$, then $(\Delta, \sigma_1 \mid \sigma_2, s \mid \sigma_3, r) \odot \Gamma, x : A \vdash t : B$.*

*Proof.* By induction on the form of $(\Delta \mid \sigma_2 \mid \sigma_1 + \sigma_3) \odot \Gamma \vdash \lambda x.t : C$, as follows:

Case.

$$\dfrac{(\Delta, \sigma_1 \mid \sigma_3, r' \mid \mathbf{0}) \odot \Gamma, x' : A' \vdash B' : \mathsf{Type}_l \quad (\Delta, \sigma_1 \mid \sigma_2, s' \mid \sigma_3, r') \odot \Gamma, x' : A' \vdash t : B'}{(\Delta \mid \sigma_2 \mid \sigma_1 + \sigma_3) \odot \Gamma \vdash \lambda x.t : (x' :_{(s',r')} A') \to B'} \text{ T\_Fun}$$

By Lemma 3.14 we have $x' = x$, $s' = s$, $r' = r$, $(\Delta \mid \sigma_1) \odot \Gamma \vdash A \leq A'$, and $(\Delta, \sigma_1 \mid \sigma_3, r) \odot \Gamma, x : A \vdash B' \leq B$. Therefore we have by Lemma 3.7 that $(\Delta, \sigma_1 \mid \sigma_2, s \mid \sigma_3, r) \odot \Gamma, x : A \vdash t : B'$, and applying this to T\_TyConv, we have $(\Delta, \sigma_1 \mid \sigma_2, s \mid \sigma_3, r) \odot \Gamma, x : A \vdash t : B$, as required.

Case.

$$\dfrac{(\Delta \mid \sigma_1 \mid \sigma_2) \odot \Gamma \vdash \lambda x.t : C \quad (\Delta \mid \sigma_2) \odot \Gamma \vdash C \leq D}{(\Delta \mid \sigma_1 \mid \sigma_2) \odot \Gamma \vdash \lambda x.t : D} \text{ T\_TyConv}$$

Then our goal holds by the following derivation:

$$\dfrac{(\Delta \mid \sigma_1 \mid \sigma_2) \odot \Gamma \vdash \lambda x.t : C \quad \dfrac{(\Delta \mid \sigma_2) \odot \Gamma \vdash C \leq D \quad (\Delta \mid \sigma_2) \odot \Gamma \vdash D \leq (x :_{(s,r)} A) \to B}{(\Delta \mid \sigma_2) \odot \Gamma \vdash C \leq (x :_{(s,r)} A) \to B} \text{ ST\_Trans}}{(\Delta, \sigma_1 \mid \sigma_2, s \mid \sigma_3, r) \odot \Gamma, x : A \vdash t : B} \text{ I.H.}$$

□

**Lemma 3.12** (Pair inversion). *If $(\Delta \mid \sigma_1 \mid \sigma_2 + \sigma_3) \odot \Gamma \vdash (t_1, t_2) : (x :_r A) \otimes B$ and $(\Delta, \sigma_2 \mid \sigma_3, r \mid \mathbf{0}) \odot \Gamma, x : A \vdash B : \mathsf{Type}_l$, then there exist grade vectors $\hat{\sigma_1}$ and $\hat{\sigma_1}'$ such that $(\Delta \mid \hat{\sigma_1} \mid \sigma_2) \odot \Gamma \vdash t_1 : A$, and $(\Delta \mid \hat{\sigma_1}' \mid \sigma_3 + r * \hat{\sigma_1}) \odot \Gamma \vdash t_2 : [t_1/x]B$, with $\hat{\sigma_1} + \hat{\sigma_1}' = \sigma_1$.*

*Proof.* Proof proceeds similarly to the proof for Lemma 3.11, but using Lemma 3.15. □

**Lemma 3.13** (Box inversion). *If $(\Delta \mid \sigma_1 \mid \sigma_2) \odot \Gamma \vdash \Box t : B$ and $(\Delta \mid \sigma_2) \odot \Gamma \vdash B \leq \Box_s A$, then $(\Delta \mid \hat{\sigma_1} \mid \sigma_2) \odot \Gamma \vdash t : A$ for some $\hat{\sigma_1}$ such that $s * \hat{\sigma_1} = \sigma_1$.*

*Proof.* By induction on the form of $(\Delta \mid \sigma_1 \mid \sigma_2) \odot \Gamma \vdash \Box t : B$, as follows:

---

Case.

$$\frac{(\Delta \mid \hat{\sigma_1} \mid \sigma_2) \odot \Gamma \vdash t : A'}{(\Delta \mid s' * \hat{\sigma_1} \mid \sigma_2) \odot \Gamma \vdash \Box t : \Box_{s'} A'} \text{ T\_BoxI}$$

Then by Lemma 3.16 we have:

- $s' = s$ (therefore $s' * \hat{\sigma_1} = s * \hat{\sigma_1}$);
- $(\Delta \mid \sigma_2) \odot \Gamma \vdash A' \leq A$;

Then our goal holds by the following derivation:

$$\frac{(\Delta \mid \hat{\sigma_1} \mid \sigma_2) \odot \Gamma \vdash t : A' \quad (\Delta \mid \sigma_2) \odot \Gamma \vdash A' \leq A}{(\Delta \mid \hat{\sigma_1} \mid \sigma_2) \odot \Gamma \vdash t : A} \text{ T\_TyConv}$$

---

Case.

$$\frac{(\Delta \mid \sigma_1 \mid \sigma_2) \odot \Gamma \vdash \Box t : C \quad (\Delta \mid \sigma_2) \odot \Gamma \vdash C \leq B}{(\Delta \mid \sigma_1 \mid \sigma_2) \odot \Gamma \vdash \Box t : B} \text{ T\_TyConv}$$

Then our goal holds by the following derivation:

$$\frac{(\Delta \mid \sigma_1 \mid \sigma_2) \odot \Gamma \vdash \Box t : C \quad \dfrac{(\Delta \mid \sigma_2) \odot \Gamma \vdash C \leq B \quad (\Delta \mid \sigma_2) \odot \Gamma \vdash B \leq \Box_s A}{(\Delta \mid \sigma_2) \odot \Gamma \vdash C \leq \Box_s A} \text{ ST\_Trans}}{(\Delta \mid \hat{\sigma_1} \mid \sigma_2) \odot \Gamma \vdash t : A} \text{ I.H.}$$

---

□

**Lemma 3.14** (Arrow subtyping inversion). *If $(\Delta \mid \sigma_1 + \sigma_2 \mid \mathbf{0}) \odot \Gamma \vdash \mathcal{J}_1$ where $\mathcal{J}_1$ is $(x :_{(s,r)} A) \to B \leq C$ or $(x :_{(s,r)} A) \to B = C : D$ and $(\Delta \mid \sigma_1 + \sigma_2 \mid \mathbf{0}) \odot \Gamma \vdash \mathcal{J}_2$, where $\mathcal{J}_2$ is respectively $C \leq (y :_{(s',r')} A') \to B'$ or $C = (y :_{(s',r')} A') \to B' : E$ (with $(\Delta \mid \mathbf{0}) \odot \Gamma \vdash D \leq \mathsf{Type}_l$ and $(\Delta \mid \mathbf{0}) \odot \Gamma \vdash E \leq \mathsf{Type}_l$), with $(\Delta, \sigma_1 \mid \sigma_3, r \mid \mathbf{0}) \odot \Gamma, y : A' \vdash B' : \mathsf{Type}_{l'}$, then $x = y$, $s = s'$, $r = r'$, and respectively, based on $\mathcal{J}_1$, $(\Delta \mid \sigma_1) \odot \Gamma \vdash A' \leq A$, and $(\Delta, \sigma_1 \mid \sigma_3, r) \odot \Gamma, x : A' \vdash B \leq B'$, or $(\Delta \mid \sigma_1 \mid \mathbf{0}) \odot \Gamma \vdash A' = A : \mathsf{Type}_{l''}$, and $(\Delta, \sigma_1 \mid \sigma_3, r \mid \mathbf{0}) \odot \Gamma, x : A' \vdash B = B' : \mathsf{Type}_{l'''}$.*

*Proof.* Proof proceeds similarly to the proof for Lemma 3.16. □

**Lemma 3.15** (Tensor subtyping inversion). *If $(\Delta \mid \sigma_1 + \sigma_2 \mid \mathbf{0}) \odot \Gamma \vdash \mathcal{J}_1$ where $\mathcal{J}_1$ is $(x :_r A) \otimes B \leq C$ or $(x :_r A) \otimes B = C : D$ and $(\Delta \mid \sigma_1 + \sigma_2 \mid \mathbf{0}) \odot \Gamma \vdash \mathcal{J}_2$, where $\mathcal{J}_2$ is respectively $C \leq (y :_{r'} A') \otimes B'$ or $C = (y :_{r'} A') \otimes B' : E$ (with $(\Delta \mid \mathbf{0}) \odot \Gamma \vdash D \leq \mathsf{Type}_l$ and $(\Delta \mid \mathbf{0}) \odot \Gamma \vdash E \leq \mathsf{Type}_l$), with $(\Delta, \sigma_1 \mid \sigma_3, r \mid \mathbf{0}) \odot \Gamma, y : A' \vdash B' : \mathsf{Type}_{l'}$, then $x = y$, $s = s'$, $r = r'$, $(\Delta \mid \sigma_1 \mid \mathbf{0}) \odot \Gamma \vdash A = A' : \mathsf{Type}_{l''''}$, and respectively, based on $\mathcal{J}_1$, $(\Delta, \sigma_1 \mid \sigma_3, r) \odot \Gamma, x : A' \vdash B \leq B'$, or $(\Delta, \sigma_1 \mid \sigma_3, r \mid \mathbf{0}) \odot \Gamma, x : A' \vdash B = B' : \mathsf{Type}_{l'''}$.*

*Proof.* Proof proceeds similarly to the proof for Lemma 3.14. □

**Lemma 3.16** (Box subtyping inversion). *If $(\Delta \mid \sigma \mid \mathbf{0}) \odot \Gamma \vdash \mathcal{J}_1$ where $\mathcal{J}_1$ is $\Box_s A \leq B$ or $\Box_s A = B : D$, and $(\Delta \mid \sigma \mid \mathbf{0}) \odot \Gamma \vdash \mathcal{J}_2$, where $\mathcal{J}_2$ is respectively $B \leq \Box_{s'} A'$ or $\Box_{s'} A' = B : E$ (where $(\Delta \mid \mathbf{0}) \odot \Gamma \vdash D \leq \mathsf{Type}_l$ and $(\Delta \mid \mathbf{0}) \odot \Gamma \vdash E \leq \mathsf{Type}_l$), then $s = s'$, and respectively, based on $\mathcal{J}_1$, $(\Delta \mid \sigma) \odot \Gamma \vdash A \leq A'$, or $(\Delta \mid \sigma \mid \mathbf{0}) \odot \Gamma \vdash A = A' : \mathsf{Type}_{l'}$.*

*Proof.* For subtyping, by induction, as follows:

> Case.
>
> $$\frac{(\Delta \mid \sigma \mid \mathbf{0}) \odot \Gamma \vdash \Box_s A = B : \mathsf{Type}_l}{(\Delta \mid \sigma) \odot \Gamma \vdash \Box_s A \leq B} \ \text{ST\_EQ}$$
>
> With:
>
> > Case.
> >
> > $$\frac{(\Delta \mid \sigma \mid \mathbf{0}) \odot \Gamma \vdash B = \Box_{s'} A' : \mathsf{Type}_l}{(\Delta \mid \sigma) \odot \Gamma \vdash B \leq \Box_{s'} A'} \ \text{ST\_EQ}$$
> >
> > Then by induction we have $s = s'$ and $(\Delta \mid \sigma \mid \mathbf{0}) \odot \Gamma \vdash A = A' : \mathsf{Type}_l$, and therefore have $(\Delta \mid \sigma) \odot \Gamma \vdash A \leq A'$ by ST\_EQ.
>
> > Case.
> >
> > $$\frac{(\Delta \mid \sigma) \odot \Gamma \vdash B \leq C \quad (\Delta \mid \sigma) \odot \Gamma \vdash C \leq \Box_{s'} A'}{(\Delta \mid \sigma) \odot \Gamma \vdash B \leq \Box_{s'} A'} \ \text{ST\_TRANS}$$
> >
> > Then by ST\_TRANS we have $(\Delta \mid \sigma) \odot \Gamma \vdash \Box_s A \leq C$, and therefore our goal holds by induction.
>
> > Case.
> >
> > $$\frac{(\Delta \mid \sigma) \odot \Gamma \vdash A'' \leq A'}{(\Delta \mid \sigma) \odot \Gamma \vdash \Box_{s'} A'' \leq \Box_{s'} A'} \ \text{ST\_BOX}$$
> >
> > Then by induction with $(\Delta \mid \sigma \mid \mathbf{0}) \odot \Gamma \vdash \Box_s A = \Box_{s'} A'' : \mathsf{Type}_l$ and $(\Delta \mid \sigma \mid \mathbf{0}) \odot \Gamma \vdash \Box_{s'} A'' = \Box_{s'} A'' : \mathsf{Type}_l$, then we have $s = s'$, and $(\Delta \mid \sigma \mid \mathbf{0}) \odot \Gamma \vdash A = A'' : \mathsf{Type}_l$, and therefore by ST\_EQ we have $(\Delta \mid \sigma) \odot \Gamma \vdash A \leq A''$, and therefore by ST\_TRANS, we have $(\Delta \mid \sigma) \odot \Gamma \vdash A \leq A'$, as required.

> Case.

$$\frac{(\Delta \mid \sigma) \odot \Gamma \vdash \Box_s A \leq C \quad (\Delta \mid \sigma) \odot \Gamma \vdash C \leq B}{(\Delta \mid \sigma) \odot \Gamma \vdash \Box_s A \leq B} \text{ ST\_Trans}$$

Then by ST_Trans we have $(\Delta \mid \sigma) \odot \Gamma \vdash C \leq \Box_{s'} A'$, and our goal holds by induction.

Case.

$$\frac{(\Delta \mid \sigma) \odot \Gamma \vdash A \leq A''}{(\Delta \mid \sigma) \odot \Gamma \vdash \Box_s A \leq \Box_s A''} \text{ ST\_Box}$$

With:

Case.

$$\frac{(\Delta \mid \sigma \mid \mathbf{0}) \odot \Gamma \vdash \Box_s A'' = \Box_{s'} A' : \mathsf{Type}_l}{(\Delta \mid \sigma) \odot \Gamma \vdash \Box_s A'' \leq \Box_{s'} A'} \text{ ST\_Eq}$$

Then by induction with $(\Delta \mid \sigma \mid \mathbf{0}) \odot \Gamma \vdash \Box_s A'' = \Box_{s'} A' : \mathsf{Type}_l$ and $(\Delta \mid \sigma \mid \mathbf{0}) \odot \Gamma \vdash \Box_{s'} A' = \Box_{s'} A' : \mathsf{Type}_l$, then we have $s = s'$, and $(\Delta \mid \sigma \mid \mathbf{0}) \odot \Gamma \vdash A'' = A' : \mathsf{Type}_l$, and therefore by ST_Eq we have $(\Delta \mid \sigma) \odot \Gamma \vdash A'' \leq A'$, and therefore by ST_Trans, we have $(\Delta \mid \sigma) \odot \Gamma \vdash A \leq A'$, as required.

Case.

$$\frac{(\Delta \mid \sigma) \odot \Gamma \vdash \Box_s A'' \leq C \quad (\Delta \mid \sigma) \odot \Gamma \vdash C \leq \Box_{s'} A'}{(\Delta \mid \sigma) \odot \Gamma \vdash \Box_s A'' \leq \Box_{s'} A'} \text{ ST\_Trans}$$

Then by induction we have $(\Delta \mid \sigma) \odot \Gamma \vdash A'' \leq A'$, and therefore our goal holds by ST_Trans.

Case.

$$\frac{(\Delta \mid \sigma) \odot \Gamma \vdash A'' \leq A'}{(\Delta \mid \sigma) \odot \Gamma \vdash \Box_s A'' \leq \Box_s A'} \text{ ST\_Box}$$

Then our goal holds by ST_Trans.

For equality, by induction, as follows:

Case.

$$\frac{(\Delta \mid \sigma \mid \mathbf{0}) \odot \Gamma \vdash \Box_s A = C : D \quad (\Delta \mid \sigma \mid \mathbf{0}) \odot \Gamma \vdash C = B : D}{(\Delta \mid \sigma \mid \mathbf{0}) \odot \Gamma \vdash \Box_s A = B : D} \text{ TEQ\_Trans}$$

Then by TEQ_Trans we have $(\Delta \mid \sigma \mid \mathbf{0}) \odot \Gamma \vdash C = \Box_{s'} A' : D$, and therefore our goal holds by induction.

Case.

$$\frac{(\Delta \mid \sigma \mid \mathbf{0}) \odot \Gamma \vdash \Box_s A = B : C \quad (\Delta \mid \mathbf{0}) \odot \Gamma \vdash C \leq D}{(\Delta \mid \sigma \mid \mathbf{0}) \odot \Gamma \vdash \Box_s A = B : D} \text{ TEQ\_ConvTy}$$

By ST_Trans, we have $(\Delta \mid \mathbf{0}) \odot \Gamma \vdash C \leq \mathsf{Type}_l$. With:

Case.

$$\frac{(\Delta \mid \sigma \mid \mathbf{0}) \odot \Gamma \vdash \Box_{s'} A' : E}{(\Delta \mid \sigma \mid \mathbf{0}) \odot \Gamma \vdash \Box_{s'} A' = \Box_{s'} A' : E} \text{ TEQ\_Refl}$$

Then our goal holds by induction.

Case.

$$\frac{(\Delta \mid \sigma \mid \mathbf{0}) \odot \Gamma \vdash \Box_{s'} A' = B : E}{(\Delta \mid \sigma \mid \mathbf{0}) \odot \Gamma \vdash B = \Box_{s'} A' : E} \text{ TEQ\_Sym}$$

Then we have by induction that $(\Delta \mid \sigma \mid \mathbf{0}) \odot \Gamma \vdash A' = A : \mathsf{Type}_l$ and thus by TEQ_Sym than $(\Delta \mid \sigma \mid \mathbf{0}) \odot \Gamma \vdash A = A' : \mathsf{Type}_l$, as required.

Case.

$$\frac{(\Delta \mid \sigma \mid \mathbf{0}) \odot \Gamma \vdash B = C : E \quad (\Delta \mid \sigma \mid \mathbf{0}) \odot \Gamma \vdash C = \Box_{s'} A' : E}{(\Delta \mid \sigma \mid \mathbf{0}) \odot \Gamma \vdash B = \Box_{s'} A' : E} \text{ TEQ\_Trans}$$

Use TEQ_ConvTy to obtain $(\Delta \mid \sigma \mid \mathbf{0}) \odot \Gamma \vdash \Box_s A = B : \mathsf{Type}_l$ and $(\Delta \mid \sigma \mid \mathbf{0}) \odot \Gamma \vdash C = \Box_{s'} A' : \mathsf{Type}_l$, then use TEQ_Trans to obtain $(\Delta \mid \sigma \mid \mathbf{0}) \odot \Gamma \vdash \Box_s A = C : \mathsf{Type}_l$, and therefore our goal holds by induction.

Case.

$$\dfrac{(\Delta \mid \sigma \mid \mathbf{0}) \odot \Gamma \vdash B = \Box_{s'} A' : C' \quad (\Delta \mid \sigma) \odot \Gamma \vdash C' \leq E}{(\Delta \mid \sigma \mid \mathbf{0}) \odot \Gamma \vdash B = \Box_{s'} A' : E} \text{ TEQ\_ConvTy}$$

Then by ST_Trans we have $(\Delta \mid \sigma) \odot \Gamma \vdash C' \leq \mathsf{Type}_l$, and our goal holds by induction.

Case.

$$\dfrac{(\Delta \mid \sigma \mid \mathbf{0}) \odot \Gamma \vdash A'' = A' : \mathsf{Type}_l}{(\Delta \mid \sigma \mid \mathbf{0}) \odot \Gamma \vdash \Box_{s'} A'' = \Box_{s'} A' : \mathsf{Type}_l} \text{ TEQ\_Box}$$

Then by induction with $(\Delta \mid \sigma \mid \mathbf{0}) \odot \Gamma \vdash \Box_s A = \Box_{s'} A'' : D$ and $(\Delta \mid \sigma \mid \mathbf{0}) \odot \Gamma \vdash \Box_{s'} A'' = \Box_{s'} A'' : \mathsf{Type}_l$, then we have $s = s'$, and $(\Delta \mid \sigma \mid \mathbf{0}) \odot \Gamma \vdash A = A'' : \mathsf{Type}_l$, and therefore by TEQ_Trans, we have $(\Delta \mid \sigma \mid \mathbf{0}) \odot \Gamma \vdash A = A' : \mathsf{Type}_{l'}$, as required.

Remaining cases proceed similarly.

$\Box$

## A.2 Proofs of meta properties

### A.2.1 Contraction

**Lemma 3.22** (Contraction). *The following rule is admissible:*

$$\dfrac{\left(\begin{smallmatrix}\Delta_1, \sigma_1, (\sigma_1, 0), \Delta_2 \\ \sigma_2, s_1, s_2, \sigma_3 \\ \sigma_4, r_1, r_2, \sigma_5\end{smallmatrix}\right) \odot \Gamma_1, x : A, y : A, \Gamma_2 \vdash \mathcal{J} \quad |\Delta_1| = |\sigma_2| = |\sigma_4| = |\Gamma_1|}{\left(\begin{smallmatrix}\Delta_1, \sigma_1, \mathsf{contr}(|\Delta_1|; \Delta_2) \\ \sigma_2, (s_1 + s_2), \sigma_3 \\ \sigma_4, (r_1 + r_2), \sigma_5\end{smallmatrix}\right) \odot \Gamma_1, z : A, [z, z/x, y]\Gamma_2 \vdash [z, z/x, y]\mathcal{J}} \text{ Cont}$$

*Proof.* For this lemma, it is sufficient to consider that $z$ and $x$ are the same, in which case our goal is:

$$\left(\begin{smallmatrix}\Delta_1, \sigma_1, \mathsf{contr}(|\Delta_1|; \Delta_2) \\ \sigma_2, (s_1 + s_2), \sigma_3 \\ \sigma_4, (r_1 + r_2), \sigma_5\end{smallmatrix}\right) \odot \Gamma_1, x : A, [x/y]\Gamma_2 \vdash [x/y]\mathcal{J}$$

We can form the following premises for Lemma 3.27:

$$\mathcal{C}$$
$$\dfrac{\dfrac{\dfrac{\left(\begin{smallmatrix}\Delta_1, \sigma_1, (\sigma_1, 0), \Delta_2 \\ \sigma_2, s_1, s_2, \sigma_3 \\ \sigma_4, r_1, r_2, \sigma_5\end{smallmatrix}\right) \odot \Gamma_1, x : A, y : A, \Gamma_2 \vdash \mathcal{J}}{\Delta_1, \sigma_1, (\sigma_1, 0), \Delta_2 \odot \Gamma_1, x : A, y : A, \Gamma_2 \vdash} \text{ Lemma 3.3}}{\Delta_1, \sigma_1 \odot \Gamma_1, x : A \vdash} \text{ Lemma 3.4}}{(\Delta_1, \sigma_1 \mid \mathbf{0}^{|\Delta_1|}, 1 \mid \sigma_1, 0) \odot \Gamma_1, x : A \vdash x : A} \text{ T\_Var}$$

1. $(\Delta_1, \sigma_1 \mid \mathbf{0}^{|\Delta_1|}, 1 \mid \sigma_1, 0) \odot \Gamma_1, x : A \vdash x : A$ ($\mathcal{C}$)

2. $(\Delta_1, \sigma_1, (\sigma_1, 0), \Delta_2 \mid \sigma_2, s_1, s_2, \sigma_3 \mid \sigma_4, r_1, r_2, \sigma_5) \odot \Gamma_1, x : A, y : A, \Gamma_2 \vdash \mathcal{J}$ (premise)

44

3. $|\sigma_2, s_1| = |\sigma_4, r_1| = |\Gamma_1, x : A|$ (trivially)

Giving:

$$\begin{pmatrix} \Delta_1, \sigma_1, (\Delta_2 \backslash \pi + (\Delta_2/\pi) * (\mathbf{0}^{|\Delta_1|}, 1)) \\ (\sigma_2, s_1 + s_2 * (\mathbf{0}^{|\Delta_1|}, 1)), \sigma_3 \\ (\sigma_4, r_1 + r_2 * (\mathbf{0}^{|\Delta_1|}, 1)), \sigma_5 \end{pmatrix} \odot \Gamma_1, x : A, [x/y]\Gamma_2 \vdash [x/y]\mathcal{J}$$

where $\pi = |\Gamma, x : A| = |\Delta_1| + 1$. Recall that $\mathsf{contr}(|\Delta_1|; \Delta_2) = \Delta_2\backslash(|\Delta_1|+1) + (\Delta_2/(|\Delta_1|+1)) * (\mathbf{0}^{|\Delta_1|}, 1) = \Delta_2\backslash\pi + (\Delta_2/\pi) * (\mathbf{0}^{|\Delta_1|}, 1)$. From sizing, we can see that $(\sigma_2, s_1 + s_2 * (\mathbf{0}^{|\Delta_1|}, 1)), \sigma_3 = (\sigma_2, s_1 + \mathbf{0}^{|\Delta_1|}, s_2), \sigma_3 = \sigma_2, (s_1 + s_2), \sigma_3$, and similarly for $(\sigma_4, r_1 + r_2 * (\mathbf{0}^{|\Delta_1|}, 1)), \sigma_5$. Therefore, we have:

$$\begin{pmatrix} \Delta_1, \sigma_1, \mathsf{contr}(|\Delta_1|;\Delta_2) \\ \sigma_2, (s_1+s_2), \sigma_3 \\ \sigma_4, (r_1+r_2), \sigma_5 \end{pmatrix} \odot \Gamma_1, x : A, [x/y]\Gamma_2 \vdash [x/y]\mathcal{J}$$

As required. □

### A.2.2 Exchange

**Lemma 3.23** (Exchange). *The following rule is admissible:*

$$\frac{\begin{pmatrix} \Delta_1, \sigma_1, (\sigma_2, 0), \Delta_2 \\ \sigma_3, s_1, s_2, \sigma_4 \\ \sigma_5, r_1, r_2, \sigma_6 \end{pmatrix} \odot \Gamma_1, x : A, y : B, \Gamma_2 \vdash \mathcal{J} \quad x \notin \mathsf{FV}(B) \quad |\Delta_1| = |\sigma_3| = |\sigma_5| = |\Gamma_1|}{\begin{pmatrix} \Delta_1, \sigma_2, (\sigma_1, 0), \mathsf{exch}(|\Delta_1|;\Delta_2) \\ \sigma_3, s_2, s_1, \sigma_4 \\ \sigma_5, r_2, r_1, \sigma_6 \end{pmatrix} \odot \Gamma_1, y : B, x : A, \Gamma_2 \vdash \mathcal{J}} \text{ EXCHANGE}$$

*Proof.* For well-formed contexts: by induction on the form of $\Delta_1, \sigma_1, (\sigma_2, 0), \Delta_2 \odot \Gamma_1, x : A, y : B, \Gamma_2 \vdash$, for which it suffices to consider whether $\Gamma_2$ is empty or not. In the case that $\Gamma_2$ is empty, we have:

$$\frac{(\Delta_1, \sigma_1 \mid \sigma_2, 0 \mid \mathbf{0}) \odot \Gamma_1, x : A \vdash B : \mathsf{Type}_l}{\Delta_1, \sigma_1, (\sigma_2, 0) \odot \Gamma_1, x : A, y : B \vdash} \text{ WF\_EXT}$$

And we obtain our goal via the following derivation:

$$\frac{\dfrac{(\Delta_1, \sigma_1 \mid \sigma_2, 0 \mid \mathbf{0}) \odot \Gamma_1, x : A \vdash B : \mathsf{Type}_l}{(\Delta_1 \mid \sigma_1 \mid \mathbf{0}) \odot \Gamma_1 \vdash A : \mathsf{Type}_{l'}} \text{ L. 3.5} \quad \dfrac{(\Delta_1, \sigma_1 \mid \sigma_2, 0 \mid \mathbf{0}) \odot \Gamma_1, x : A \vdash B : \mathsf{Type}_l}{(\Delta_1 \mid \sigma_2 \mid \mathbf{0}) \odot \Gamma_1 \vdash B : \mathsf{Type}_l} \text{ L. 3.1}}{\dfrac{(\Delta_1, \sigma_2 \mid \sigma_1, 0 \mid \mathbf{0}) \odot \Gamma_1, y : B \vdash A : \mathsf{Type}_{l'}}{\Delta_1, \sigma_2, (\sigma_1, 0) \odot \Gamma_1, y : B, x : A \vdash} \text{ WF\_EXT}} \text{ L. 3.24}$$

Now consider that $\Gamma_2 = \Gamma_2', z : C$, then we have:

$$\frac{(\Delta_1, \sigma_1, (\sigma_2, 0), \Delta_2' \mid \sigma_3, s, r, \sigma_3' \mid \mathbf{0}) \odot \Gamma_1, x : A, y : B, \Gamma_2' \vdash C : \mathsf{Type}_l}{\Delta_1, \sigma_1, (\sigma_2, 0), \Delta_2', (\sigma_3, s, r, \sigma_3') \odot \Gamma_1, x : A, y : B, \Gamma_2', z : C \vdash} \text{ WF\_EXT}$$

Where $|\sigma_3| = |\Gamma_1|$. We can then obtain our goal via the following derivation:

$$\dfrac{\dfrac{\dfrac{(\Delta_1, \sigma_1, (\sigma_2, 0), \Delta_2' \mid \sigma_3, s, r, \sigma_3' \mid \mathbf{0}) \odot \Gamma_1, x : A, y : B, \Gamma_2' \vdash C : \mathsf{Type}_l}{(\Delta_1, \sigma_2, (\sigma_1, 0), \mathsf{exch}(|\Delta_1|;\Delta_2') \mid \sigma_3, r, s, \sigma_3' \mid \mathbf{0}) \odot \Gamma_1, y : B, x : A, \Gamma_2' \vdash C : \mathsf{Type}_l} \text{ L. 3.23}}{\Delta_1, \sigma_2, (\sigma_1, 0), \mathsf{exch}(|\Delta_1|;\Delta_2'), (\sigma_3, r, s, \sigma_3') \odot \Gamma_1, y : B, x : A, \Gamma_2', z : C \vdash} \text{ WF\_EXT}}{\Delta_1, \sigma_2, (\sigma_1, 0), \mathsf{exch}(|\Delta_1|; (\Delta_2', (\sigma_3, s, r, \sigma_3'))) \odot \Gamma_1, y : B, x : A, \Gamma_2', z : C \vdash} \equiv$$

For typing $(t : C)$: By induction on the form of $(\Delta_1, \sigma_1, (\sigma_2, 0), \Delta_2 \mid \sigma_3, s_1, s_2, \sigma_4 \mid \sigma_5, r_1, r_2, \sigma_6) \odot \Gamma_1, x : A, y : B, \Gamma_2 \vdash t : C$ as follows:

45

Case.

$$\frac{\Delta_1, \sigma_1, (\sigma_2, 0), \Delta_2 \odot \Gamma_1, x:A, y:B, \Gamma_2 \vdash}{(\Delta_1, \sigma_1, (\sigma_2, 0), \Delta_2 \mid \mathbf{0} \mid \mathbf{0}) \odot \Gamma_1, x:A, y:B, \Gamma_2 \vdash \mathsf{Type}_l : \mathsf{Type}_{\mathsf{suc}\ l}} \text{ T\_Type}$$

Then our goal holds by the following derivation:

$$\frac{\dfrac{\Delta_1, \sigma_1, (\sigma_2, 0), \Delta_2 \odot \Gamma_1, x:A, y:B, \Gamma_2 \vdash}{\Delta_1, \sigma_2, (\sigma_1, 0), \mathsf{exch}(|\Delta_1|; \Delta_2) \odot \Gamma_1, y:B, x:A, \Gamma_2 \vdash}\text{ I.H.}}{(\Delta_1, \sigma_2, (\sigma_1, 0), \mathsf{exch}(|\Delta_1|; \Delta_2) \mid \mathbf{0} \mid \mathbf{0}) \odot \Gamma_1, y:B, x:A, \Gamma_2 \vdash \mathsf{Type}_l : \mathsf{Type}_{\mathsf{suc}\ l}} \text{ T\_Type}$$

---

Case.

$$\frac{\begin{pmatrix}\Delta_1, \sigma_1, (\sigma_2, 0), \Delta_2 \\ \sigma_3, s_1, s_2, \sigma_4 \\ \mathbf{0}\end{pmatrix} \odot \Gamma_1, x:A, y:B, \Gamma_2 \vdash C : \mathsf{Type}_{l_1} \quad \begin{pmatrix}\Delta_1, \sigma_1, (\sigma_2, 0), \Delta_2, (\sigma_3, s_1, s_2, \sigma_4) \\ \sigma'_3, s'_1, s'_2, \sigma'_4, r \\ \mathbf{0}\end{pmatrix} \odot \Gamma_1, x:A, y:B, \Gamma_2, z:C \vdash D : \mathsf{Type}_{l_2}}{\begin{pmatrix}\Delta_1, \sigma_1, (\sigma_2, 0), \Delta_2 \\ \sigma_3, s_1, s_2, \sigma_4 + \sigma'_3, s'_1, s'_2, \sigma'_4 \\ \mathbf{0}\end{pmatrix} \odot \Gamma_1, x:A, y:B, \Gamma_2 \vdash (z :_{(s,r)} C) \to D : \mathsf{Type}_{l_1 \sqcup l_2}} \text{ T\_Arrow}$$

Performing induction on both premises, and putting them back together via T_Arrow, we have:

$$\frac{\begin{pmatrix}\Delta_1, \sigma_2, (\sigma_1, 0), \mathsf{exch}(|\Delta_1|; \Delta_2) \\ \sigma_3, s_2, s_1, \sigma_4 \\ \mathbf{0}\end{pmatrix} \odot \Gamma_1, y:B, x:A, \Gamma_2 \vdash C : \mathsf{Type}_{l_1} \quad \begin{pmatrix}\Delta_1, \sigma_2, (\sigma_1, 0), \mathsf{exch}(|\Delta_1|; \Delta_2), (\sigma_3, s_2, s_1, \sigma_4) \\ \sigma'_3, s'_2, s'_1, \sigma'_4, r \\ \mathbf{0}\end{pmatrix} \odot \Gamma_1, x:A, y:B, \Gamma_2, z:C \vdash D : \mathsf{Type}_{l_2}}{\begin{pmatrix}\Delta_1, \sigma_2, (\sigma_1, 0), \mathsf{exch}(|\Delta_1|; \Delta_2) \\ \sigma_3, s_2, s_1, \sigma_4 + \sigma'_3, s'_2, s'_1, \sigma'_4 \\ \mathbf{0}\end{pmatrix} \odot \Gamma_1, y:B, x:A, \Gamma_2 \vdash (z :_{(s,r)} C) \to D : \mathsf{Type}_{l_1 \sqcup l_2}} \text{ T\_Arrow}$$

As required. The proofs for T_Ten, T_Fun, T_App, T_Pair, T_TenCut, T_Box, T_BoxI, T_BoxE, and T_TyConv proceed similarly.

---

Case.

$$\frac{\Delta_1, \sigma_1, (\sigma_2, 0), \Delta_2 \odot \Gamma_1, x:A, y:B, \Gamma_2 \vdash \quad |\Delta_1| = |\Gamma_1|}{(\Delta_1, \sigma_1, (\sigma_2, 0), \Delta_2 \mid \mathbf{0}^{|\Delta_1|}, 1, 0, \mathbf{0} \mid \sigma_1, 0, 0, \mathbf{0}) \odot \Gamma_1, x:A, y:B, \Gamma_2 \vdash x:A} \text{ T\_Var}$$

Our goal holds by the following derivation:

$$\frac{\dfrac{\Delta_1, \sigma_1, (\sigma_2, 0), \Delta_2 \odot \Gamma_1, x:A, y:B, \Gamma_2 \vdash \quad |\Delta_1| = |\Gamma_1| \quad x \notin \mathsf{FV}(B)}{\Delta_1, \sigma_2, (\sigma_1, 0), \mathsf{exch}(|\Delta_1|; \Delta_2) \odot \Gamma_1, y:B, x:A, \Gamma_2 \vdash}\text{ I.H.} \quad |\Delta_1, \sigma_2| = |\Gamma_1, y:B|}{(\Delta_1, \sigma_2, (\sigma_1, 0), \mathsf{exch}(|\Delta_1|; \Delta_2) \mid \mathbf{0}^{|\Delta_1|}, 0, 1, \mathbf{0} \mid \sigma_1, 0, 0, \mathbf{0}) \odot \Gamma_1, y:B, x:A, \Gamma_2 \vdash x:A} \text{ T\_Var}$$

Case.

$$\dfrac{\Delta_1, \sigma_1, (\sigma_2, 0), \Delta_2 \odot \Gamma_1, x:A, y:B, \Gamma_2 \vdash \quad |\Delta_1, \sigma_1| = |\Gamma_1, x:A|}{(\Delta_1, \sigma_1, (\sigma_2, 0), \Delta_2 \mid \mathbf{0}^{|\Delta_1|}, 0, 1, \mathbf{0} \mid \sigma_2, 0, 0, \mathbf{0}) \odot \Gamma_1, x:A, y:B, \Gamma_2 \vdash y:B} \text{T\_VAR}$$

Our goal holds by the following derivation:

$$\dfrac{\dfrac{\Delta_1, \sigma_1, (\sigma_2, 0), \Delta_2 \odot \Gamma_1, x:A, y:B, \Gamma_2 \vdash \quad |\Delta_1| = |\Gamma_1| \quad x \notin \mathsf{FV}(B)}{\Delta_1, \sigma_2, (\sigma_1, 0), \mathsf{exch}(|\Delta_1|; \Delta_2) \odot \Gamma_1, y:B, x:A, \Gamma_2 \vdash} \text{I.H.} \quad |\Delta_1| = |\Gamma_1|}{(\Delta_1, \sigma_2, (\sigma_1, 0), \mathsf{exch}(|\Delta_1|; \Delta_2) \mid \mathbf{0}^{|\Delta_1|}, 1, 0, \mathbf{0} \mid \sigma_2, 0, 0, \mathbf{0}) \odot \Gamma_1, y:B, x:A, \Gamma_2 \vdash y:B} \text{T\_VAR}$$

Case.

$$\dfrac{\hat{\Delta}_1, \sigma, \hat{\Delta}_1{}', \sigma_1, (\sigma_2, 0), \Delta_2 \odot \hat{\Gamma}_1, z:C, \hat{\Gamma}_1{}', x:A, y:B, \Gamma_2 \vdash \quad \left|\hat{\Delta}_1\right| = \left|\hat{\Gamma}_1\right|}{(\hat{\Delta}_1, \sigma, \hat{\Delta}_1{}', \sigma_1, (\sigma_2, 0), \Delta_2 \mid \mathbf{0}^{|\hat{\Delta}_1|}, 1, \mathbf{0} \mid \sigma, \mathbf{0}) \odot \hat{\Gamma}_1, z:C, \hat{\Gamma}_1{}', x:A, y:B, \Gamma_2 \vdash z:C} \text{T\_VAR}$$

Our goal holds by the following derivation:

$$\dfrac{\dfrac{\left|\hat{\Delta}_1, \sigma, \hat{\Delta}_1{}'\right| = \left|\hat{\Gamma}_1, z:C, \hat{\Gamma}_1{}'\right| \quad x \notin \mathsf{FV}(B)}{\hat{\Delta}_1, \sigma, \hat{\Delta}_1{}', \sigma_1, (\sigma_2, 0), \Delta_2 \odot \hat{\Gamma}_1, z:C, \hat{\Gamma}_1{}', x:A, y:B, \Gamma_2 \vdash}{\hat{\Delta}_1, \sigma, \hat{\Delta}_1{}', \sigma_2, (\sigma_1, 0), \mathsf{exch}(\left|\hat{\Delta}_1, \sigma, \hat{\Delta}_1{}'\right|; \Delta_2) \odot \hat{\Gamma}_1, z:C, \hat{\Gamma}_1{}', y:B, x:A, \Gamma_2 \vdash} \text{I.H.} \quad \left|\hat{\Delta}_1\right| = \left|\hat{\Gamma}_1\right|}{(\hat{\Delta}_1, \sigma, \hat{\Delta}_1{}', \sigma_2, (\sigma_1, 0), \mathsf{exch}(\left|\hat{\Delta}_1, \sigma, \hat{\Delta}_1{}'\right|; \Delta_2) \mid \mathbf{0}^{|\hat{\Delta}_1|}, 1, \mathbf{0} \mid \sigma, \mathbf{0}) \odot \hat{\Gamma}_1, z:C, \hat{\Gamma}_1{}', y:B, x:A, \Gamma_2 \vdash z:C} \text{T\_VAR}$$

Case.

$$\dfrac{\Delta_1, \sigma_1, (\sigma_2, 0), \hat{\Delta}_2, \sigma, \hat{\Delta}_2{}' \odot \Gamma_1, x:A, y:B, \hat{\Gamma}_2, z:C, \hat{\Gamma}_2{}' \vdash \quad \left|\Delta_1, \sigma_1, (\sigma_2, 0), \hat{\Delta}_2\right| = \left|\Gamma_1, x:A, y:B, \hat{\Gamma}_2\right|}{(\Delta_1, \sigma_1, (\sigma_2, 0), \hat{\Delta}_2, \sigma, \hat{\Delta}_2{}' \mid \mathbf{0}^{|\hat{\Delta}_1|}, 1, \mathbf{0} \mid \sigma, \mathbf{0}) \odot \Gamma_1, x:A, y:B, \hat{\Gamma}_2, z:C, \hat{\Gamma}_2{}' \vdash z:C} \text{T\_VAR}$$

Where $\sigma = \hat{\sigma}, s, r, \hat{\sigma}'$ with $|\hat{\sigma}| = |\Delta_1|$. Our goal holds by the following derivation:

$$\dfrac{\dfrac{\left|\Delta_1, \sigma_1, (\sigma_2, 0), \mathsf{exch}(|\Delta_1|; \hat{\Delta}_2)\right| = \left|\Gamma_1, y:B, x:A, \hat{\Gamma}_2\right|}{|\Delta_1| = |\Gamma_1| \quad x \notin \mathsf{FV}(B)} \\ \dfrac{\Delta_1, \sigma_1, (\sigma_2, 0), \hat{\Delta}_2, \sigma, \hat{\Delta}_2{}' \odot \Gamma_1, x:A, y:B, \hat{\Gamma}_2, z:C, \hat{\Gamma}_2{}' \vdash}{\Delta_1, \sigma_2, (\sigma_1, 0), \mathsf{exch}(|\Delta_1|; (\hat{\Delta}_2, \sigma, \hat{\Delta}_2{}')) \odot \Gamma_1, y:B, x:A, \hat{\Gamma}_2, z:C, \hat{\Gamma}_2{}' \vdash} \text{I.H.}}{\dfrac{\Delta_1, \sigma_2, (\sigma_1, 0), \mathsf{exch}(|\Delta_1|; \hat{\Delta}_2), (\hat{\sigma}, r, s, \hat{\sigma}'), \mathsf{exch}(|\Delta_1|; \hat{\Delta}_2{}') \odot \Gamma_1, y:B, x:A, \hat{\Gamma}_2, z:C, \hat{\Gamma}_2{}' \vdash}{\dfrac{(\Delta_1, \sigma_2, (\sigma_1, 0), \mathsf{exch}(|\Delta_1|; \hat{\Delta}_2), (\hat{\sigma}, r, s, \hat{\sigma}'), \mathsf{exch}(|\Delta_1|; \hat{\Delta}_2{}') \mid \mathbf{0}^{|\hat{\Delta}_1|}, 1, \mathbf{0} \mid \hat{\sigma}, r, s, \hat{\sigma}', \mathbf{0}) \odot \Gamma_1, y:B, x:A, \hat{\Gamma}_2, z:C, \hat{\Gamma}_2{}' \vdash z:C}{(\Delta_1, \sigma_2, (\sigma_1, 0), \mathsf{exch}(|\Delta_1|; (\hat{\Delta}_2, \sigma, \hat{\Delta}_2{}')) \mid \mathbf{0}^{|\hat{\Delta}_1|}, 1, \mathbf{0} \mid \hat{\sigma}, r, s, \hat{\sigma}', \mathbf{0}) \odot \Gamma_1, y:B, x:A, \hat{\Gamma}_2, z:C, \hat{\Gamma}_2{}' \vdash z:C} \equiv} \text{T\_VAR}} \equiv$$

For equality, by standard induction and re-application (see Section A.8).
For subtyping, by standard induction and re-application (see Section A.8). □

**Corollary 3.23.1** (Exchange (general))**.** *As a corollary to Lemma 3.23, the following rule is admissible:*

$$\dfrac{\begin{pmatrix}\Delta_1,\Delta_2,(\sigma_1,\mathbf{0}^{|\Delta_2|}),\Delta_3\\ \sigma_2,\sigma_3,s,\sigma_4\\ \sigma_5,\sigma_6,r,\sigma_7\end{pmatrix}\odot\Gamma_1,\Gamma_2,x:A,\Gamma_3\vdash\mathcal{J}\\ \mathsf{Dom}(\Gamma_2)\cap\mathsf{FV}(A)=\emptyset\quad |\sigma_1|=|\Delta_1|=|\sigma_2|=|\sigma_5|=|\Gamma_1|}{\begin{pmatrix}\Delta_1,\sigma_1,\mathsf{ins}(|\Delta_1|;0;\Delta_2),\mathsf{mv}(|\Delta_1,\Delta_2|;|\Delta_1|;\Delta_3)\\ \sigma_2,s,\sigma_3,\sigma_4\\ \sigma_5,r,\sigma_6,\sigma_7\end{pmatrix}\odot\Gamma_1,x:A,\Gamma_2,\Gamma_3\vdash\mathcal{J}}\;\textsc{ExchangeGen}$$

*Where* $\mathsf{Dom}(\Gamma)$ *is the domain of* $\Gamma$.

*Proof.* Consider $\Gamma_2$ is empty, then by premise we have $(\Delta_1,\sigma_1,\Delta_3\mid\sigma_2,s,\sigma_4\mid\sigma_5,r,\sigma_7)\odot\Gamma_1,x:A,\Gamma_3\vdash\mathcal{J}$, and need to show $(\Delta_1,\sigma_1,\Delta_3\mid\sigma_2,s,\sigma_4\mid\sigma_5,r,\sigma_7)\odot\Gamma_1,x:A,\Gamma_3\vdash\mathcal{J}$, which holds by premise. Now consider $\Gamma_2=\Gamma_2',y:C$, then by premise we have:

$$\begin{pmatrix}\Delta_1,(\Delta_2',\sigma_8),(\sigma_1,\mathbf{0}^{|\Delta_2',\sigma_8|}),\Delta_3\\ \sigma_2,\sigma_3',s',s,\sigma_4\\ \sigma_5,\sigma_6',r',r,\sigma_7\end{pmatrix}\odot\Gamma_1,\Gamma_2',y:C,x:A,\Gamma_3\vdash\mathcal{J}$$

And our goal is:

$$\begin{pmatrix}\Delta_1,\sigma_1,\mathsf{ins}(|\Delta_1|;0;(\Delta_2',\sigma_8)),\mathsf{mv}(|\Delta_1,(\Delta_2',\sigma_8)|;|\Delta_1|;\Delta_3)\\ \sigma_2,s,\sigma_3',s',\sigma_4\\ \sigma_5,r,\sigma_6',r',\sigma_7\end{pmatrix}\odot\Gamma_1,x:A,\Gamma_2',y:C,\Gamma_3\vdash\mathcal{J}$$

Doing some trivial vector rewriting, and then applying our premise in the following derivation, we have:

$$\dfrac{\begin{pmatrix}\Delta_1,\Delta_2',\sigma_8,((\sigma_1,\mathbf{0}^{|\Delta_2'|}),0),\Delta_3\\ \sigma_2,\sigma_3',s',s,\sigma_4\\ \sigma_5,\sigma_6',r',r,\sigma_7\end{pmatrix}\odot\Gamma_1,\Gamma_2',y:C,x:A,\Gamma_3\vdash\mathcal{J}\\ y\notin\mathsf{FV}(A)\quad |\Delta_1,\Delta_2'|=|\sigma_2,\sigma_3'|=|\sigma_5,\sigma_6'|=|\Gamma_1,\Gamma_2'|}{\begin{pmatrix}\Delta_1,\Delta_2',(\sigma_1,\mathbf{0}^{|\Delta_2'|}),(\sigma_8,0),\mathsf{exch}(|\Delta_1,\Delta_2'|;\Delta_3)\\ \sigma_2,\sigma_3',s,s',\sigma_4\\ \sigma_5,\sigma_6',r,r',\sigma_7\end{pmatrix}\odot\Gamma_1,\Gamma_2',x:A,y:C,\Gamma_3\vdash\mathcal{J}}\;\textsc{Lemma 3.23}$$

Applying this inductively, we have:

$$\dfrac{\begin{pmatrix}\Delta_1,\Delta_2',(\sigma_1,\mathbf{0}^{|\Delta_2'|}),(\sigma_8,0),\mathsf{exch}(|\Delta_1,\Delta_2'|;\Delta_3)\\ \sigma_2,\sigma_3',s,s',\sigma_4\\ \sigma_5,\sigma_6',r,r',\sigma_7\end{pmatrix}\odot\Gamma_1,\Gamma_2',x:A,y:C,\Gamma_3\vdash\mathcal{J}\\ \mathsf{Dom}(\Gamma_2')\cap\mathsf{FV}(A)=\emptyset\quad |\sigma_1|=|\Delta_1|=|\sigma_2|=|\sigma_5|=|\Gamma_1|}{\begin{pmatrix}\Delta_1,\sigma_1,\mathsf{ins}(|\Delta_1|;0;\Delta_2'),\mathsf{mv}(|\Delta_1,\Delta_2'|;|\Delta_1|;((\sigma_8,0),\mathsf{exch}(|\Delta_1,\Delta_2'|;\Delta_3)))\\ \sigma_2,s,\sigma_3',s',\sigma_4\\ \sigma_5,r,\sigma_6',r',\sigma_7\end{pmatrix}\odot\Gamma_1,x:A,\Gamma_2',y:C,\Gamma_3\vdash\mathcal{J}}\;\text{I.H.}$$

By Lemma 3.21, we have $\mathsf{mv}(|\Delta_1,\Delta_2'|;|\Delta_1|;(\mathsf{exch}(|\Delta_1,\Delta_2'|;\Delta_3)))=\mathsf{mv}(|\Delta_1,(\Delta_2',\sigma_8)|;|\Delta_1|;\Delta_3)$, and as by Lemma 3.2 we know $|\sigma_8|=|\Delta_1,\Delta_2'|$, we have:

$$\mathsf{mv}(|\Delta_1,\Delta_2'|;|\Delta_1|;((\sigma_8,0),\mathsf{exch}(|\Delta_1,\Delta_2'|;\Delta_3)))=(\sigma_8',0,\sigma_8''),\mathsf{mv}(|\Delta_1,(\Delta_2',\sigma_8)|;|\Delta_1|;\Delta_3)$$

for some $\sigma_8'$ and $\sigma_8''$, with $|\sigma_8'|=|\Delta_1|$. Therefore, we have $\mathsf{ins}(|\Delta_1|;0;\Delta_2'),(\sigma_8',0,\sigma_8'')=\mathsf{ins}(|\Delta_1|;0;(\Delta_2',\sigma_8))$. Rewriting the conclusion by this information, we have:

$$\begin{pmatrix}\Delta_1,\sigma_1,\mathsf{ins}(|\Delta_1|;0;(\Delta_2',\sigma_8)),\mathsf{mv}(|\Delta_1,(\Delta_2',\sigma_8)|;|\Delta_1|;\Delta_3)\\ \sigma_2,s,\sigma_3',s',\sigma_4\\ \sigma_5,r,\sigma_6',r',\sigma_7\end{pmatrix}\odot\Gamma_1,x:A,\Gamma_2',y:C,\Gamma_3\vdash\mathcal{J}$$

Which matches our goal. □

**Corollary 3.23.2** (Exchange from end). *As a corollary to Corollary 3.23.1, the following rule is admissible:*

$$\frac{\left(\begin{smallmatrix}\Delta_1,\Delta_2,(\sigma_1,\mathbf{0}^{|\Delta_2|})\\ \sigma_2,\sigma_3,s\\ \sigma_4,\sigma_5,r\end{smallmatrix}\right) \odot \Gamma_1, \Gamma_2, x : A \vdash \mathcal{J} \quad \mathsf{Dom}(\Gamma_2) \cap \mathsf{FV}(A) = \emptyset \quad |\Delta_1| = |\sigma_2| = |\sigma_4| = |\Gamma_1|}{\left(\begin{smallmatrix}\Delta_1,\sigma_1,\mathsf{ins}(|\Delta_1|;0;\Delta_2)\\ \sigma_2,s,\sigma_3\\ \sigma_4,r,\sigma_5\end{smallmatrix}\right) \odot \Gamma_1, x : A, \Gamma_2 \vdash \mathcal{J}} \text{ ExchangeEnd}$$

*Proof.* This is Corollary 3.23.1 with $\Gamma_3$ being empty. □

### A.2.3 Weakening

**Lemma 3.24** (Weakening). *The following rule is admissible:*

$$\frac{(\Delta \mid \sigma_1 \mid \sigma_2) \odot \Gamma \vdash \mathcal{J} \quad (\Delta \mid \sigma_3 \mid \mathbf{0}) \odot \Gamma \vdash A : \mathsf{Type}_l}{(\Delta, \sigma_3 \mid \sigma_1, 0 \mid \sigma_2, 0) \odot \Gamma, x : A \vdash \mathcal{J}} \text{ Weak}$$

*Where $\mathcal{J}$ is typing, equality, or subtyping.*

*Proof.* For typing: By induction on the form of $(\Delta \mid \sigma_1 \mid \sigma_2) \odot \Gamma \vdash t : B$ as follows:

---

Case.

$$\frac{\Delta \odot \Gamma \vdash}{(\Delta \mid \mathbf{0} \mid \mathbf{0}) \odot \Gamma \vdash \mathsf{Type}_l : \mathsf{Type}_{\mathsf{suc}\ l}} \text{ T\_Type}$$

Then our goal is to show:

$$(\Delta, \sigma_3 \mid \mathbf{0} \mid \mathbf{0}) \odot \Gamma, x : A \vdash \mathsf{Type}_l : \mathsf{Type}_{\mathsf{suc}\ l}$$

We we obtain through the following derivation:

$$\frac{\dfrac{(\Delta \mid \sigma_3 \mid \mathbf{0}) \odot \Gamma \vdash A : \mathsf{Type}_{l'}}{\Delta, \sigma_3 \odot \Gamma, x : A \vdash} \text{ Wf\_Ext}}{(\Delta, \sigma_3 \mid \mathbf{0} \mid \mathbf{0}) \odot \Gamma, x : A \vdash \mathsf{Type}_l : \mathsf{Type}_{\mathsf{suc}\ l}} \text{ T\_Type}$$

---

Case.

$$\frac{(\Delta \mid \sigma_1 \mid \mathbf{0}) \odot \Gamma \vdash C : \mathsf{Type}_{l_1} \quad (\Delta, \sigma_1 \mid \sigma_2, r \mid \mathbf{0}) \odot \Gamma, y : C \vdash D : \mathsf{Type}_{l_2}}{(\Delta \mid \sigma_1 + \sigma_2 \mid \mathbf{0}) \odot \Gamma \vdash (y :_{(s,r)} C) \to D : \mathsf{Type}_{l_1 \sqcup l_2}} \text{ T\_Arrow}$$

Then our goal is to show:

$$(\Delta, \sigma_3 \mid (\sigma_1 + \sigma_2), 0 \mid \mathbf{0}) \odot \Gamma, x : A \vdash (y :_{(s,r)} C) \to D : \mathsf{Type}_{l_1 \sqcup l_2}$$

We we obtain through the following derivation:

$$\frac{\begin{array}{c}\mathcal{C}\\ (\Delta \mid \sigma_1 \mid \mathbf{0}) \odot \Gamma \vdash C : \mathsf{Type}_{l_1} \quad (\Delta \mid \sigma_3 \mid \mathbf{0}) \odot \Gamma \vdash A : \mathsf{Type}_l\end{array}}{(\Delta, \sigma_3 \mid \sigma_1, 0 \mid \mathbf{0}) \odot \Gamma, x : A \vdash C : \mathsf{Type}_{l_1}} \text{ I.H.}$$

$$
\begin{array}{c}
\mathcal{D} \\
\end{array}
$$

$$
\cfrac{
  \cfrac{
    (\Delta, \sigma_1 \mid \sigma_2, r \mid \mathbf{0}) \odot \Gamma, y : C \vdash D : \mathsf{Type}_{l_2}
    \qquad
    \cfrac{
      (\Delta \mid \sigma_3 \mid \mathbf{0}) \odot \Gamma \vdash A : \mathsf{Type}_l
      \qquad
      (\Delta \mid \sigma_1 \mid \mathbf{0}) \odot \Gamma \vdash C : \mathsf{Type}_{l_1}
    }{
      (\Delta, \sigma_1 \mid \sigma_3, 0 \mid \mathbf{0}) \odot \Gamma, y : C \vdash A : \mathsf{Type}_l
    } \text{ I.H.}
  }{
    (\Delta, \sigma_1, (\sigma_3, 0) \mid \sigma_2, r, 0 \mid \mathbf{0}) \odot \Gamma, y : C, x : A \vdash D : \mathsf{Type}_{l_2}
  } \text{ I.H.}
}{
  (\Delta, \sigma_3, (\sigma_1, 0) \mid \sigma_2, 0, r \mid \mathbf{0}) \odot \Gamma, x : A, y : C \vdash D : \mathsf{Type}_{l_2}
} \text{ L. 3.23}
$$

$$
\cfrac{
  \cfrac{\mathcal{C} \qquad \mathcal{D}}{
    (\Delta, \sigma_3 \mid \sigma_1, 0 + \sigma_2, 0 \mid \mathbf{0}) \odot \Gamma, x : A \vdash (y :_{(s,r)} C) \to D : \mathsf{Type}_{l_1 \sqcup l_2}
  } \text{ T\_Arrow}
}{
  (\Delta, \sigma_3 \mid (\sigma_1 + \sigma_2), 0 \mid \mathbf{0}) \odot \Gamma, x : A \vdash (y :_{(s,r)} C) \to D : \mathsf{Type}_{l_1 \sqcup l_2}
} \equiv
$$

The proofs for T\_Ten, T\_Fun, T\_App, T\_Pair, T\_TenCut, T\_Box, T\_BoxI, T\_BoxE, and T\_TyConv proceed similarly, using induction and Lemma 3.23.

Case.

$$
\cfrac{
  \Delta_1, \sigma, \Delta_2 \odot \Gamma_1, y : C, \Gamma_2 \vdash \qquad |\Delta_1| = |\Gamma_1|
}{
  (\Delta_1, \sigma, \Delta_2 \mid \mathbf{0}^{|\Delta_1|}, 1, \mathbf{0} \mid \sigma, 0, \mathbf{0}) \odot \Gamma_1, y : C, \Gamma_2 \vdash y : C
} \text{ T\_Var}
$$

Then our goal is to show:

$$
(\Delta_1, \sigma, \Delta_2, \sigma_3 \mid \mathbf{0}^{|\Delta_1|}, 1, \mathbf{0} \mid \sigma, \mathbf{0}) \odot \Gamma_1, y : C, \Gamma_2, x : A \vdash y : C
$$

Which holds by the following derivation:

$$
\cfrac{
  \cfrac{
    (\Delta_1, \sigma, \Delta_2 \mid \sigma_3 \mid \mathbf{0}) \odot \Gamma_1, y : C, \Gamma_2 \vdash A : \mathsf{Type}_l
  }{
    \Delta_1, \sigma, \Delta_2, \sigma_3 \odot \Gamma_1, y : C, \Gamma_2, x : A \vdash
  } \text{ W}_{\text{F\_Ext}}
  \qquad |\Delta_1| = |\Gamma_1|
}{
  (\Delta_1, \sigma, \Delta_2, \sigma_3 \mid \mathbf{0}^{|\Delta_1|}, 1, \mathbf{0} \mid \sigma, \mathbf{0}) \odot \Gamma_1, y : C, \Gamma_2, x : A \vdash y : C
} \text{ T\_Var}
$$

For equality, by standard induction and re-application (see Section A.8).
For subtyping, by standard induction and re-application (see Section A.8). □

**Lemma 3.25** (Weakening for well-formed contexts). *The following rule is admissible:*

$$
\cfrac{
  \Delta_1, \Delta_2 \odot \Gamma_1, \Gamma_2 \vdash \qquad (\Delta_1 \mid \sigma_1 \mid \sigma_2) \odot \Gamma_1 \vdash t : A
}{
  (\Delta_1, \Delta_2 \mid \sigma_1, \mathbf{0} \mid \sigma_2, \mathbf{0}) \odot \Gamma_1, \Gamma_2 \vdash t : A
} \text{ W}_{\text{EAK}}\text{WF}
$$

*Proof.* By induction on the form of $(\Delta_1 \mid \sigma_1 \mid \sigma_2) \odot \Gamma_1 \vdash t : A$, and $\Delta_1, \Delta_2 \odot \Gamma_1, \Gamma_2 \vdash$, as follows: First, consider that $\Gamma_2$ is empty, then the typing premise for $t$ matches the goal, and we're done. Now consider that $\Delta_1, \Delta_2 \odot \Gamma_1, \Gamma_2 \vdash$ is:

$$
\cfrac{
  \mathcal{C} = (\Delta_1, \Delta_2' \mid \sigma_3 \mid \mathbf{0}) \odot \Gamma_1, \Gamma_2' \vdash B : \mathsf{Type}_l
}{
  \Delta_1, \Delta_2', \sigma_3 \odot \Gamma_1, \Gamma_2', x : B \vdash
} \text{ W}_{\text{F\_Ext}}
$$

Then our goal is to show:

$$
(\Delta_1, \Delta_2', \sigma_3 \mid \sigma_1, \mathbf{0} \mid \sigma_2, \mathbf{0}) \odot \Gamma_1, \Gamma_2', x : B \vdash t : A
$$

Which we obtain by the following derivation:

$$\cfrac{\cfrac{\mathcal{C}}{\Delta_1, \Delta_2' \odot \Gamma_1, \Gamma_2' \vdash} \text{L. 3.3} \quad (\Delta_1 \mid \sigma_1 \mid \sigma_2) \odot \Gamma_1 \vdash t : A}{\cfrac{(\Delta_1, \Delta_2' \mid \sigma_1, \mathbf{0} \mid \sigma_2, \mathbf{0}) \odot \Gamma_1, \Gamma_2' \vdash t : A}{(\Delta_1, \Delta_2', \sigma_3 \mid \sigma_1, \mathbf{0} \mid \sigma_2, \mathbf{0}) \odot \Gamma_1, \Gamma_2', x : B \vdash t : A} \text{ I.H.} \quad \mathcal{C}} \text{L. 3.24}}$$

$\square$

**Lemma 3.26** (Weakening (general)). *The following rule is admissible (for $\mathcal{J}$ is typing, equality, or subtyping):*

$$\cfrac{(\Delta_1, \Delta_2 \mid \sigma_1, \sigma_1' \mid \sigma_2, \sigma_2') \odot \Gamma_1, \Gamma_2 \vdash \mathcal{J} \quad (\Delta_1 \mid \sigma_3 \mid \mathbf{0}) \odot \Gamma_1 \vdash A : \mathsf{Type}_l \quad |\sigma_1| = |\sigma_2| = |\Gamma_1|}{(\Delta_1, \sigma_3, \mathsf{ins}(|\Delta_1|; 0; \Delta_2) \mid \sigma_1, 0, \sigma_1' \mid \sigma_2, 0, \sigma_2') \odot \Gamma_1, x : A, \Gamma_2 \vdash \mathcal{J}} \text{ WeakGen}$$

*Proof.* By Lemma 3.3 with the premise for $\mathcal{J}$, we have $(\Delta_1, \Delta_2) \odot \Gamma_1, \Gamma_2 \vdash$. Therefore, using this and the typing premise for $A$ with Lemma 3.25, we have $(\Delta_1, \Delta_2 \mid \sigma_3, \mathbf{0}^{|\Gamma_2|} \mid \mathbf{0}) \odot \Gamma_1, \Gamma_2 \vdash A : \mathsf{Type}_l$. Applying this, and the premise for $\mathcal{J}$, in the following derivation, we have:

$$\cfrac{(\Delta_1, \Delta_2 \mid \sigma_1, \sigma_1' \mid \sigma_2, \sigma_2') \odot \Gamma_1, \Gamma_2 \vdash \mathcal{J} \quad (\Delta_1, \Delta_2 \mid \sigma_3, \mathbf{0}^{|\Gamma_2|} \mid \mathbf{0}) \odot \Gamma_1, \Gamma_2 \vdash A : \mathsf{Type}_l}{(\Delta_1, \Delta_2, (\sigma_3, \mathbf{0}^{|\Gamma_2|}) \mid \sigma_1, \sigma_1', 0 \mid \sigma_2, \sigma_2', 0) \odot \Gamma_1, \Gamma_2, x : A \vdash \mathcal{J}} \text{ Lemma 3.24}$$

As $(\Delta_1 \mid \sigma_3 \mid \mathbf{0}) \odot \Gamma_1 \vdash A : \mathsf{Type}_l$, by disjointness of context extension we therefore have $\mathsf{Dom}(\Gamma_2) \cap \mathsf{FV}(A) = \emptyset$. By Lemma 3.2 we have $|\Gamma_2| = |\Delta_2|$ and $|\Delta_1| = |\Gamma_1|$. We can put this information together via the following derivation, to obtain the goal:

$$\cfrac{\begin{pmatrix}\Delta_1, \Delta_2, (\sigma_3, \mathbf{0}^{|\Delta_2|}) \\ \sigma_1, \sigma_1', 0 \\ \sigma_2, \sigma_2', 0\end{pmatrix} \odot \Gamma_1, \Gamma_2, x : A \vdash \mathcal{J} \quad \mathsf{Dom}(\Gamma_2) \cap \mathsf{FV}(A) = \emptyset \quad |\Delta_1| = |\sigma_1| = |\sigma_2| = |\Gamma_1|}{\begin{pmatrix}\Delta_1, \sigma_3, \mathsf{ins}(|\Delta_1|; 0; \Delta_2) \\ \sigma_1, 0, \sigma_1' \\ \sigma_2, 0, \sigma_2'\end{pmatrix} \odot \Gamma_1, x : A, \Gamma_2 \vdash \mathcal{J}} \text{ Corollary 3.23.2}$$

$\square$

### A.2.4 Substitution

**Lemma 3.27** (Substitution for judgments). *If the following premises hold:*

1. $(\Delta \mid \sigma_2 \mid \sigma_1) \odot \Gamma_1 \vdash t : A$

2. $(\Delta, \sigma_1, \Delta' \mid \sigma_3, s, \sigma_4 \mid \sigma_5, r, \sigma_6) \odot \Gamma_1, x : A, \Gamma_2 \vdash \mathcal{J}$

3. $|\sigma_3| = |\sigma_5| = |\Gamma_1|$

*Then* $\begin{pmatrix}\Delta, (\Delta' \backslash \pi + (\Delta'/\pi) * \sigma_2) \\ (\sigma_3 + s * \sigma_2), \sigma_4 \\ (\sigma_5 + r * \sigma_2), \sigma_6\end{pmatrix} \odot \Gamma_1, [t/x]\Gamma_2 \vdash [t/x]\mathcal{J}$, *where $\pi = |\Gamma_1|$.*

*Proof.* Throughout the proof, we make implicit use of the following size information derived from the premises (and largely Lemma 3.2), further size calculations are trivial, and we typically do not bring attention to them:

- $|\Delta| = |\sigma_1| = |\sigma_2| = |\sigma_3| = |\sigma_5| = |\Gamma_1|$
- $|\Delta'| = |\sigma_4| = |\sigma_6| = |\Gamma_2|$

For well-formed contexts, we proceed by induction on the structure of $(\Delta, \sigma_1, \Delta') \odot \Gamma_1, x : A, \Gamma_2 \vdash$, as follows:

51

Case.

$$\overline{\emptyset \odot \emptyset \vdash} \;\; \text{Wf\_Empty}$$

Trivial as it does not match the form of the typing premise.

---

Case.

$$\frac{(\Delta \mid \sigma_1 \mid \mathbf{0}) \odot \Gamma_1 \vdash A : \mathsf{Type}_l}{\Delta, \sigma_1 \odot \Gamma_1, x : A \vdash} \;\; \text{Wf\_Ext}$$

Then our goal is:

$$\Delta \odot \Gamma_1 \vdash$$

Which holds by Lemma 3.3 on the typing premise for $A$.

---

Case.

$$\frac{(\Delta, \sigma_1, \Delta' \mid \sigma_3, s, \sigma_4 \mid \mathbf{0}) \odot \Gamma_1, x : A, \Gamma_2' \vdash B : \mathsf{Type}_l}{\Delta, \sigma_1, \Delta', (\sigma_3, s, \sigma_4) \odot \Gamma_1, x : A, \Gamma_2', y : B \vdash} \;\; \text{Wf\_Ext}$$

With $|\sigma_3| = |\Gamma_1|$. Then our goal is:

$$\Delta, ((\Delta', (\sigma_3, s, \sigma_4))\backslash \pi + ((\Delta', (\sigma_3, s, \sigma_4))/\pi) * \sigma_2) \odot \Gamma_1, [t/x]\Gamma_2', y : [t/x]B \vdash$$

As $|\sigma_3| = |\Gamma_1|$, we can see that by Lemma 3.17, our goal becomes:

$$\Delta, (\Delta'\backslash \pi + (\Delta'/\pi) * \sigma_2), ((\sigma_3 + s * \sigma_2), \sigma_4) \odot \Gamma_1, [t/x]\Gamma_2', y : [t/x]B \vdash$$

We can form the following premises, which we then apply to Lemma 3.27:

1. $(\Delta \mid \sigma_2 \mid \sigma_1) \odot \Gamma_1 \vdash t : A$ (premise (1))
2. $(\Delta, \sigma_1, \Delta' \mid \sigma_3, s, \sigma_4 \mid \mathbf{0}) \odot \Gamma_1, x : A, \Gamma_2' \vdash B : \mathsf{Type}_l$ (Wf\_Ext typing premise for $B$)
3. $|\sigma_3| = |\mathbf{0}| = |\Gamma_1|$ (trivially, by assumption of sizing for judgments)

Giving:

$$\begin{pmatrix} \Delta, (\Delta'\backslash \pi + (\Delta'/\pi) * \sigma_2) \\ (\sigma_3 + s*\sigma_2), \sigma_4 \\ \mathbf{0} \end{pmatrix} \odot \Gamma_1, [t/x]\Gamma_2' \vdash [t/x]B : \mathsf{Type}_l$$

We can apply the result of this substitution to Wf\_Ext, giving:

$$\frac{\begin{pmatrix} \Delta, (\Delta'\backslash \pi + (\Delta'/\pi) * \sigma_2) \\ (\sigma_3 + s*\sigma_2), \sigma_4 \\ \mathbf{0} \end{pmatrix} \odot \Gamma_1, [t/x]\Gamma_2' \vdash [t/x]B : \mathsf{Type}_l}{\Delta, (\Delta'\backslash \pi + (\Delta'/\pi) * \sigma_2), ((\sigma_3 + s * \sigma_2), \sigma_4) \odot \Gamma_1, [t/x]\Gamma_2', y : [t/x]B \vdash} \;\; \text{Wf\_Ext}$$

Which matches our goal.

For typing, we proceed by induction on the structure of $\left(\begin{smallmatrix}\Delta,\sigma_1,\Delta'\\\sigma_3,s,\sigma_4\\\sigma_5,r,\sigma_6\end{smallmatrix}\right) \odot \Gamma_1, x : A, \Gamma_2 \vdash t' : B$, as follows:

---

Case.

$$\frac{\Delta, \sigma_1, \Delta' \odot \Gamma_1, x : A, \Gamma_2 \vdash}{(\Delta, \sigma_1, \Delta' \mid \mathbf{0} \mid \mathbf{0}) \odot \Gamma_1, x : A, \Gamma_2 \vdash \mathsf{Type}_l : \mathsf{Type}_{\mathsf{suc}\ l}} \text{ T\_Type}$$

Therefore we have $t' = \mathsf{Type}_l$ and $B = \mathsf{Type}_{\mathsf{suc}\ l}$, and our goal is:

$$\left(\begin{smallmatrix}\Delta,(\Delta'\backslash\pi+(\Delta'/\pi)*\sigma_2)\\\mathbf{0}\\\mathbf{0}\end{smallmatrix}\right) \odot \Gamma_1, [t/x]\Gamma_2 \vdash \mathsf{Type}_l : \mathsf{Type}_{\mathsf{suc}\ l}$$

We can form the following inductive premises:

1. $(\Delta \mid \sigma_2 \mid \sigma_1) \odot \Gamma_1 \vdash t : A$ (premise (1))
2. $\Delta, \sigma_1, \Delta' \odot \Gamma_1, x : A, \Gamma_2 \vdash$ (T\_Type well-formedness premise)

Giving:

$$\Delta, (\Delta'\backslash\pi + (\Delta'/\pi) * \sigma_2) \odot \Gamma_1, [t/x]\Gamma_2 \vdash$$

We can then apply this information to T\_Type, to achieve our goal:

$$\frac{\Delta, (\Delta'\backslash\pi + (\Delta'/\pi) * \sigma_2) \odot \Gamma_1, [t/x]\Gamma_2 \vdash}{(\Delta, (\Delta'\backslash\pi + (\Delta'/\pi) * \sigma_2) \mid \mathbf{0} \mid \mathbf{0}) \odot \Gamma_1, [t/x]\Gamma_2 \vdash \mathsf{Type}_l : \mathsf{Type}_{\mathsf{suc}\ l}} \text{ T\_Type}$$

---

Case.

$$\frac{\left(\begin{smallmatrix}\Delta,\sigma_1,\Delta'\\\sigma_7\\\mathbf{0}\end{smallmatrix}\right) \odot \Gamma_1, x : A, \Gamma_2 \vdash C : \mathsf{Type}_{l_1} \quad \left(\begin{smallmatrix}\Delta,\sigma_7\\\sigma_8,r'\\\mathbf{0}\end{smallmatrix}\right) \odot \Gamma_1, x : A, \Gamma_2, y : C \vdash D : \mathsf{Type}_{l_2}}{(\Delta, \sigma_1, \Delta' \mid \sigma_7 + \sigma_8 \mid \mathbf{0}) \odot \Gamma_1, x : A, \Gamma_2 \vdash (y :_{(s',r')} C) \to D : \mathsf{Type}_{l_1 \sqcup l_2}} \text{ T\_Arrow}$$

Therefore we have $\sigma_3, s, \sigma_4 = \sigma_7 + \sigma_8$; $\sigma_5, r, \sigma_6 = \mathbf{0}$; $t' = (y :_{(s',r')} C) \to D$; and $B = \mathsf{Type}_{l_1 \sqcup l_2}$. Our full goal is therefore:

$$\left(\begin{smallmatrix}\Delta,(\Delta'\backslash\pi+(\Delta'/\pi)*\sigma_2)\\(\sigma_3+s*\sigma_2),\sigma_4\\\mathbf{0}\end{smallmatrix}\right) \odot \Gamma_1, [t/x]\Gamma_2 \vdash [t/x]((y :_{(s',r')} C) \to D) : [t/x]\mathsf{Type}_{l_1 \sqcup l_2}$$

As $\sigma_3, s, \sigma_4 = \sigma_7 + \sigma_8$, we therefore have $\sigma_7 = \hat{\sigma_3}, \hat{s}, \hat{\sigma_4}$ and $\sigma_8 = \hat{\sigma_3}', \hat{s}', \hat{\sigma_4}'$ where $\hat{\sigma_3} + \hat{\sigma_3}' = \sigma_3$; $\hat{s} + \hat{s}' = s$; and $\hat{\sigma_4} + \hat{\sigma_4}' = \sigma_4$. Rewriting the first and second premises using the above information, we have the following inductive premises:

1. $(\Delta \mid \sigma_2 \mid \sigma_1) \odot \Gamma_1 \vdash t : A$ (premise (1))
2. $(\Delta, \sigma_1, \Delta' \mid \hat{\sigma_3}, \hat{s}, \hat{\sigma_4} \mid \mathbf{0}) \odot \Gamma_1, x : A, \Gamma_2 \vdash C : \mathsf{Type}_{l_1}$ (first premise of T\_Arrow)
3. $|\hat{\sigma_3}| = |\mathbf{0}| = |\Gamma_1|$ (trivially, and by premise (3))

And:

1. $(\Delta \mid \sigma_2 \mid \sigma_1) \odot \Gamma_1 \vdash t : A$ (premise (1))

2. $\begin{pmatrix} \Delta,\sigma_1,\Delta',(\hat{\sigma_3},\hat{s},\hat{\sigma_4}) \\ \hat{\sigma_3}',\hat{s}',\hat{\sigma_4}',r' \\ \mathbf{0} \end{pmatrix} \odot \Gamma_1, x : A, \Gamma_2, y : C \vdash D : \mathsf{Type}_{l_2}$ (second premise of T_Arrow)

3. $|\hat{\sigma_3}'| = |\mathbf{0}| = |\Gamma_1|$ (trivially, and by premise (3))

Giving, respectively:

$$\begin{pmatrix} \Delta,((\Delta'\backslash\pi)+(\Delta'/\pi)*\sigma_2) \\ (\hat{\sigma_3}+\hat{s}*\sigma_2),\hat{\sigma_4} \\ \mathbf{0} \end{pmatrix} \odot \Gamma_1, [t/x]\Gamma_2 \vdash [t/x]C : \mathsf{Type}_{l_1}$$

And:

$$\begin{pmatrix} \Delta,(((\Delta',(\hat{\sigma_3},\hat{s},\hat{\sigma_4}))\backslash\pi)+((\Delta',(\hat{\sigma_3},\hat{s},\hat{\sigma_4}))/\pi)*\sigma_2) \\ (\hat{\sigma_3}'+\hat{s}'*\sigma_2),\hat{\sigma_4}',r' \\ \mathbf{0} \end{pmatrix} \odot \Gamma_1, [t/x]\Gamma_2, y : [t/x]C \vdash [t/x]D : \mathsf{Type}_{l_2}$$

Rewriting $((\Delta',(\hat{\sigma_3},\hat{s},\hat{\sigma_4}))\backslash\pi)+((\Delta',(\hat{\sigma_3},\hat{s},\hat{\sigma_4}))/\pi)*\sigma_2$ by Lemma 3.17, we have $(\Delta'\backslash\pi+(\Delta'/\pi)*\sigma_2),((\hat{\sigma_3}+\hat{s}*\sigma_2),\hat{\sigma_4})$, and can thus apply our two inductive hypotheses via T_Arrow, giving:

$$\frac{\begin{pmatrix} \Delta,(\Delta'\backslash\pi+(\Delta'/\pi)*\sigma_2) \\ (\hat{\sigma_3}+\hat{s}*\sigma_2),\hat{\sigma_4} \\ \mathbf{0} \end{pmatrix} \odot \Gamma_1, [t/x]\Gamma_2 \vdash [t/x]C : \mathsf{Type}_{l_1} \quad \begin{pmatrix} \Delta,(\Delta'\backslash\pi+(\Delta'/\pi)*\sigma_2),((\hat{\sigma_3}+\hat{s}*\sigma_2),\hat{\sigma_4}) \\ ((\hat{\sigma_3}'+\hat{s}'*\sigma_2),\hat{\sigma_4}'),r' \\ \mathbf{0} \end{pmatrix} \odot \Gamma_1, [t/x]\Gamma_2, y : [t/x]C \vdash [t/x]D : \mathsf{Type}_{l_2}}{\begin{pmatrix} \Delta,(\Delta'\backslash\pi+(\Delta'/\pi)*\sigma_2) \\ ((\hat{\sigma_3}+\hat{s}*\sigma_2),\hat{\sigma_4})+((\hat{\sigma_3}'+\hat{s}'*\sigma_2),\hat{\sigma_4}') \\ \mathbf{0} \end{pmatrix} \odot \Gamma_1, [t/x]\Gamma_2 \vdash (y :_{(s',r')} [t/x]C) \to [t/x]D : \mathsf{Type}_{l_1 \sqcup l_2}} \text{ T\_Arrow}$$

Rewriting the subject grades of this conclusion by Lemma 3.18, we have $((\hat{\sigma_3}+\hat{s}*\sigma_2),\hat{\sigma_4}) + ((\hat{\sigma_3}'+\hat{s}'*\sigma_2),\hat{\sigma_4}') = (\sigma_3+s*\sigma_2),\sigma_4$. As $x \neq y$ (by disjointness of concatenation), we have $(y :_{(s',r')} [t/x]C) \to [t/x]D = [t/x]((y :_{(s',r')} C) \to D)$. We have $[t/x]\mathsf{Type}_{l_1 \sqcup l_2} = \mathsf{Type}_{l_1 \sqcup l_2}$ trivially, as $\mathsf{Type}_{l_1 \sqcup l_2}$ is a constant. Using this information to rewrite the conclusion of T_Arrow derived above, we have:

$$\begin{pmatrix} \Delta,(\Delta'\backslash\pi+(\Delta'/\pi)*\sigma_2) \\ (\sigma_3+s*\sigma_2),\sigma_4 \\ \mathbf{0} \end{pmatrix} \odot \Gamma_1, [t/x]\Gamma_2 \vdash [t/x]((y :_{(s',r')} C) \to D) : [t/x]\mathsf{Type}_{l_1 \sqcup l_2}$$

Which matches the goal.

---

Case.

$$\frac{(\Delta \mid \sigma_1 \mid \mathbf{0}) \odot \Gamma \vdash A : \mathsf{Type}_{l_1} \quad (\Delta,\sigma_1 \mid \sigma_2,r \mid \mathbf{0}) \odot \Gamma, x : A \vdash B : \mathsf{Type}_{l_2}}{(\Delta \mid \sigma_1 + \sigma_2 \mid \mathbf{0}) \odot \Gamma \vdash (x :_r A) \otimes B : \mathsf{Type}_{l_1 \sqcup l_2}} \text{ T\_Ten}$$

This proceeds similarly to the case for T_Arrow.

---

Case.

$$\frac{\Delta,\sigma_1,\Delta' \odot \Gamma_1, x : A, \Gamma_2 \vdash \quad |\Delta| = |\Gamma_1|}{(\Delta,\sigma_1,\Delta' \mid \mathbf{0}^{|\Delta|}, 1, \mathbf{0} \mid \sigma_1, 0, \mathbf{0}) \odot \Gamma_1, x : A, \Gamma_2 \vdash x : A} \text{ T\_Var}$$

Therefore we have $t' = x$; and $B = A$. Our full goal is therefore:

$$\begin{pmatrix} \Delta,(\Delta'\backslash\pi+(\Delta'/\pi)*\sigma_2) \\ \sigma_2,\mathbf{0} \\ \sigma_1,\mathbf{0} \end{pmatrix} \odot \Gamma_1, [t/x]\Gamma_2 \vdash t : A$$

We can form the following inductive premises:

1. $(\Delta \mid \sigma_2 \mid \sigma_1) \odot \Gamma_1 \vdash t : A$ (premise (1))
2. $\Delta, \sigma_1, \Delta' \odot \Gamma_1, x : A, \Gamma_2 \vdash$ (T_VAR well-formedness premise)

Giving:

$$\Delta, (\Delta'\backslash\pi + (\Delta'/\pi) * \sigma_2) \odot \Gamma_1, [t/x]\Gamma_2 \vdash$$

We can then form the following derivation, to achieve our goal:

$$\frac{\Delta,(\Delta'\backslash\pi+(\Delta'/\pi)*\sigma_2) \odot \Gamma_1, [t/x]\Gamma_2 \vdash \quad (\Delta \mid \sigma_2 \mid \sigma_1) \odot \Gamma_1 \vdash t : A}{\begin{pmatrix} \Delta,(\Delta'\backslash\pi+(\Delta'/\pi)*\sigma_2) \\ \sigma_2,\mathbf{0} \\ \sigma_1,\mathbf{0} \end{pmatrix} \odot \Gamma_1, [t/x]\Gamma_2 \vdash t : A} \text{ Lemma 3.25}$$

---

Case.

$$\frac{\Delta, \sigma_1, \hat{\Delta}', (\sigma_5, r, \hat{\sigma_6}), \hat{\Delta}'' \odot \Gamma_1, x : A, \hat{\Gamma_2}, y : C, \hat{\Gamma_2}' \vdash \quad \left|\Delta, \sigma_1, \hat{\Delta}'\right| = \left|\Gamma_1, x : A, \hat{\Gamma_2}\right|}{(\Delta, \sigma_1, \hat{\Delta}', (\sigma_5, r, \hat{\sigma_6}), \hat{\Delta}'' \mid \mathbf{0}^{|\Delta,(\sigma_5,r,\hat{\sigma_6}),\hat{\Delta}'|}, 1, \mathbf{0} \mid \sigma_5, r, \hat{\sigma_6}, 0, \mathbf{0}) \odot \Gamma_1, x : A, \hat{\Gamma_2}, y : C, \hat{\Gamma_2}' \vdash y : C} \text{ T_VAR}$$

Therefore we have $t' = y$; and $B = C$. Our full goal is therefore:

$$\begin{pmatrix} \Delta,((\hat{\Delta}',(\sigma_5,r,\hat{\sigma_6}),\hat{\Delta}'')\backslash\pi+((\hat{\Delta}',(\sigma_5,r,\hat{\sigma_6}),\hat{\Delta}'')/\pi)*\sigma_2) \\ \mathbf{0}^{|\Delta|},1,\mathbf{0} \\ (\sigma_5+r*\sigma_2),\hat{\sigma_6},\mathbf{0} \end{pmatrix} \odot \Gamma_1, [t/x]\hat{\Gamma_2}, y : [t/x]C, [t/x]\hat{\Gamma_2}' \vdash y : [t/x]C$$

We can form the following inductive premises:

1. $(\Delta \mid \sigma_2 \mid \sigma_1) \odot \Gamma_1 \vdash t : A$ (premise (1))
2. $\Delta, \sigma_1, \hat{\Delta}', (\sigma_5, r, \hat{\sigma_6}), \hat{\Delta}'' \odot \Gamma_1, x : A, \hat{\Gamma_2}, y : C, \hat{\Gamma_2}' \vdash$ (T_VAR well-formedness premise)

Giving:

$$\Delta, ((\hat{\Delta}', (\sigma_5, r, \hat{\sigma_6}), \hat{\Delta}'')\backslash\pi + ((\hat{\Delta}', (\sigma_5, r, \hat{\sigma_6}), \hat{\Delta}'')/\pi * \sigma_2)) \odot \Gamma_1, [t/x]\hat{\Gamma_2}, y : [t/x]C, [t/x]\hat{\Gamma_2}' \vdash$$

By Lemma 3.17, we have:

$$(\hat{\Delta}', (\sigma_5, r, \hat{\sigma_6}), \hat{\Delta}'')\backslash\pi + (\hat{\Delta}', (\sigma_5, r, \hat{\sigma_6}))/\pi * \sigma_2 =$$
$$(\hat{\Delta}'\backslash\pi + \hat{\Delta}'/\pi * \sigma_2), ((\sigma_5 + r * \sigma_2), \hat{\sigma_6}), (\hat{\Delta}''\backslash\pi + \hat{\Delta}''/\pi * \sigma_2)$$

We can apply this information to T_VAR, and use the context grade vector equality derived above to rewrite the conclusion, to attain our goal:

$$\frac{|\Delta| = |\Gamma_1|}{\Delta, (\hat{\Delta}'\backslash\pi + \hat{\Delta}'/\pi * \sigma_2), ((\sigma_5 + r * \sigma_2), \hat{\sigma_6}), (\hat{\Delta}''\backslash\pi + \hat{\Delta}''/\pi * \sigma_2) \odot \Gamma_1, [t/x]\hat{\Gamma_2}, y : [t/x]C, [t/x]\hat{\Gamma_2}' \vdash}{\begin{pmatrix} \Delta,((\hat{\Delta}',(\sigma_5,r,\hat{\sigma_6}),\hat{\Delta}'')\backslash\pi+((\hat{\Delta}',(\sigma_5,r,\hat{\sigma_6}),\hat{\Delta}'')/\pi*\sigma_2)) \\ \mathbf{0}^{|\Delta|},1,\mathbf{0} \\ (\sigma_5+r*\sigma_2),\hat{\sigma_6},\mathbf{0} \end{pmatrix} \odot \Gamma_1, [t/x]\hat{\Gamma_2}, y : [t/x]C, [t/x]\hat{\Gamma_2}' \vdash y : [t/x]C} \text{ T_VAR}$$

Case.

$$\dfrac{\hat{\Delta},\hat{\sigma_5},\hat{\Delta}',\sigma_1,\Delta' \odot \hat{\Gamma_1}, y : C, \hat{\Gamma_1}', x : A, \Gamma_2 \vdash \quad \left|\hat{\Delta}\right| = \left|\hat{\Gamma_1}\right|}{(\hat{\Delta},\hat{\sigma_5},\hat{\Delta}',\sigma_1,\Delta' \mid \mathbf{0}^{|\hat{\Delta}|}, 1, \mathbf{0} \mid \hat{\sigma_5}, \mathbf{0}) \odot \hat{\Gamma_1}, y : C, \hat{\Gamma_1}', x : A, \Gamma_2 \vdash y : C} \text{ T\_Var}$$

Therefore we have $t' = y$; and $B = C$. Our full goal is therefore:

$$\begin{pmatrix} \hat{\Delta},\hat{\sigma_5},\hat{\Delta}',(\Delta'\backslash\pi + \Delta'/\pi * \sigma_2) \\ \mathbf{0}^{|\hat{\Delta}|}, 1, \mathbf{0} \\ \hat{\sigma_5}, \mathbf{0} \end{pmatrix} \odot \hat{\Gamma_1}, y : C, [t/x]\hat{\Gamma_1}', [t/x]\Gamma_2 \vdash y : C$$

We can form the following inductive premises:

1. $(\hat{\Delta}, \hat{\sigma_5}, \hat{\Delta}' \mid \sigma_2 \mid \sigma_1) \odot \Gamma_1 \vdash t : A$ (premise (1))
2. $\hat{\Delta}, \hat{\sigma_5}, \hat{\Delta}', \sigma_1, \Delta' \odot \hat{\Gamma_1}, y : C, \hat{\Gamma_1}', x : A, \Gamma_2 \vdash$ (T\_Var well-formedness premise)

Giving:

$$\hat{\Delta}, \hat{\sigma_5}, \hat{\Delta}', (\Delta'\backslash\pi + \Delta'/\pi * \sigma_2) \odot \hat{\Gamma_1}, y : C, \hat{\Gamma_1}', [t/x]\Gamma_2 \vdash$$

We can apply this information in the following derivation, to attain our goal:

$$\dfrac{\hat{\Delta}, \hat{\sigma_5}, \hat{\Delta}', (\Delta'\backslash\pi + \Delta'/\pi * \sigma_2) \odot \hat{\Gamma_1}, y : C, \hat{\Gamma_1}', [t/x]\Gamma_2 \vdash \quad \left|\hat{\Delta}\right| = \left|\hat{\Gamma_1}\right|}{\begin{pmatrix} (\hat{\Delta}, \hat{\sigma_5}, \hat{\Delta}', (\Delta'\backslash\pi + \Delta'/\pi * \sigma_2)) \\ \mathbf{0}^{|\hat{\Delta}|}, 1, \mathbf{0} \\ \hat{\sigma_5}, \mathbf{0} \end{pmatrix} \odot \hat{\Gamma_1}, y : C, \hat{\Gamma_1}', [t/x]\Gamma_2 \vdash y : C} \text{ T\_Var}$$

Case.

$$\dfrac{\begin{pmatrix} \Delta,\sigma_1,\Delta',\sigma_7 \\ \sigma_8,r' \\ \mathbf{0} \end{pmatrix} \odot \Gamma_1, x : A, \Gamma_2, y : C \vdash D : \mathsf{Type}_l \quad \begin{pmatrix} \Delta,\sigma_1,\Delta',\sigma_7 \\ \sigma_3,s,\sigma_4,s' \\ \sigma_8,r' \end{pmatrix} \odot \Gamma_1, x : A, \Gamma_2, y : C \vdash t'' : D}{\begin{pmatrix} (\Delta,\sigma_1),\Delta' \\ (\sigma_3,s),\sigma_4 \\ \sigma_7 + \sigma_8 \end{pmatrix} \odot \Gamma_1, x : A, \Gamma_2 \vdash \lambda y.t'' : (y :_{(s',r')} C) \to D} \text{ T\_Fun}$$

Therefore we have $\sigma_5, r, \sigma_6 = \sigma_7 + \sigma_8$; $t' = \lambda y.t''$; and $B = (y :_{(s',r')} C) \to D$. Our goal is therefore:

$$\begin{pmatrix} \Delta,(\Delta'\backslash\pi+(\Delta'/\pi)*\sigma_2) \\ (\sigma_3+s*\sigma_2),\sigma_4 \\ (\sigma_5+r*\sigma_2),\sigma_6 \end{pmatrix} \odot \Gamma_1, [t/x]\Gamma_2 \vdash [t/x](\lambda y.t'') : [t/x]((y :_{(s',r')} C) \to D)$$

As $\sigma_5, r, \sigma_6 = \sigma_7 + \sigma_8$, we therefore have $\sigma_7 = \hat{\sigma_5}, \hat{r}, \hat{\sigma_6}$ and $\sigma_8 = \hat{\sigma_5}', \hat{r}', \hat{\sigma_6}'$ where $\hat{\sigma_5} + \hat{\sigma_5}' = \sigma_5$; $\hat{r} + \hat{r}' = r$; and $\hat{\sigma_6} + \hat{\sigma_6}' = \sigma_6$. Rewriting the premise of T\_Fun using the above information, we have the following inductive premises:

1. $(\Delta \mid \sigma_2 \mid \sigma_1) \odot \Gamma_1 \vdash t : A$ (premise (1))
2. $\begin{pmatrix} \Delta,\sigma_1,\Delta',(\hat{\sigma_5},\hat{r},\hat{\sigma_6}) \\ \hat{\sigma_5}',\hat{r}',\hat{\sigma_6}',r' \\ \mathbf{0} \end{pmatrix} \odot \Gamma_1, x : A, \Gamma_2, y : C \vdash D : \mathsf{Type}_l$ (T\_Fun premise for $D$)
3. $\left|\hat{\sigma_5}'\right| = \left|\hat{\sigma_5}'\right| = |\Gamma_1|$ (trivially, and by premise (3))

56

And:

1. $(\Delta \mid \sigma_2 \mid \sigma_1) \odot \Gamma_1 \vdash t : A$ (premise (1))
2. $\begin{pmatrix} \Delta, \sigma_1, \Delta', (\hat{\sigma_5}, \hat{r}, \hat{\sigma_6}) \\ \sigma_3, s, \sigma_4, s' \\ \hat{\sigma_5}', \hat{r}', \hat{\sigma_6}', r' \end{pmatrix} \odot \Gamma_1, x : A, \Gamma_2, y : C \vdash t'' : D$ (T_Fun premise for $t''$)
3. $|\hat{\sigma_3}| = |\hat{\sigma_5}'| = |\Gamma_1|$ (trivially, and by premise (3))

Giving, respectively:

$$\begin{pmatrix} \Delta, (((\Delta', (\hat{\sigma_5}, \hat{r}, \hat{\sigma_6}))\backslash \pi) + ((\Delta', (\hat{\sigma_5}, \hat{r}, \hat{\sigma_6}))/\pi) * \sigma_2) \\ (\hat{\sigma_5}' + \hat{r}' * \sigma_2), \hat{\sigma_6}', r' \\ \mathbf{0} \end{pmatrix} \odot \Gamma_1, [t/x]\Gamma_2, y : [t/x]C \vdash [t/x]D : \mathsf{Type}_l$$

And:

$$\begin{pmatrix} \Delta, (((\Delta', (\hat{\sigma_5}, \hat{r}, \hat{\sigma_6}))\backslash \pi) + ((\Delta', (\hat{\sigma_5}, \hat{r}, \hat{\sigma_6}))/\pi) * \sigma_2) \\ (\sigma_3 + s * \sigma_2), \sigma_4, s' \\ (\hat{\sigma_5}' + \hat{r}' * \sigma_2), \hat{\sigma_6}', r' \end{pmatrix} \odot \Gamma_1, [t/x]\Gamma_2, y : [t/x]C \vdash [t/x]t'' : [t/x]D$$

Rewriting $((\Delta', (\hat{\sigma_5}, \hat{r}, \hat{\sigma_6}))\backslash \pi) + ((\Delta', (\hat{\sigma_5}, \hat{r}, \hat{\sigma_6}))/\pi) * \sigma_2$ by Lemma 3.17, we have $(\Delta'\backslash \pi + (\Delta'/\pi) * \sigma_2), ((\hat{\sigma_5} + \hat{r} * \sigma_2), \hat{\sigma_6})$, and can thus apply our inductive hypotheses to T_Fun, giving:

$$\dfrac{\begin{pmatrix} \Delta, ((\Delta'\backslash \pi + (\Delta'/\pi) * \sigma_2), ((\hat{\sigma_5} + \hat{r} * \sigma_2), \hat{\sigma_6})) \\ (\hat{\sigma_5}' + \hat{r}' * \sigma_2), \hat{\sigma_6}', r' \\ \mathbf{0} \end{pmatrix} \odot \Gamma_1, [t/x]\Gamma_2, y : [t/x]C \vdash [t/x]D : \mathsf{Type}_l \quad \begin{pmatrix} \Delta, ((\Delta'\backslash \pi + (\Delta'/\pi) * \sigma_2), ((\hat{\sigma_5} + \hat{r} * \sigma_2), \hat{\sigma_6})) \\ (\sigma_3 + s * \sigma_2), \sigma_4, s' \\ (\hat{\sigma_5}' + \hat{r}' * \sigma_2), \hat{\sigma_6}', r' \end{pmatrix} \odot \Gamma_1, [t/x]\Gamma_2, y : [t/x]C \vdash [t/x]t'' : [t/x]D}{\begin{pmatrix} \Delta, (\Delta'\backslash \pi + (\Delta'/\pi) * \sigma_2) \\ (\sigma_3 + s * \sigma_2), \sigma_4 \\ ((\hat{\sigma_5} + \hat{r} * \sigma_2), \hat{\sigma_6}) + ((\hat{\sigma_5}' + \hat{r}' * \sigma_2), \hat{\sigma_6}') \end{pmatrix} \odot \Gamma_1, [t/x]\Gamma_2 \vdash \lambda y.[t/x]t'' : (y :_{(s', r')} [t/x]C) \to [t/x]D} \text{ T\_Fun}$$

Rewriting the subject grades of this conclusion by Lemma 3.18, we have $((\hat{\sigma_5} + \hat{r} * \sigma_2), \hat{\sigma_6}) + ((\hat{\sigma_5}' + \hat{r}' * \sigma_2), \hat{\sigma_6}') = (\sigma_5 + r * \sigma_2), \sigma_6$. As $x \neq y$ (by disjointness of concatenation), we have $\lambda y.[t/x]t'' = [t/x](\lambda y.t'')$ and $(y :_{(s', r')} [t/x]C) \to [t/x]D = [t/x]((y :_{(s', r')} C) \to D)$. Using this information to rewrite the conclusion of T_Fun derived above, we have:

$$\begin{pmatrix} \Delta, (\Delta'\backslash \pi + (\Delta'/\pi) * \sigma_2) \\ (\sigma_3 + s * \sigma_2), \sigma_4 \\ (\sigma_5 + r * \sigma_2), \sigma_6 \end{pmatrix} \odot \Gamma_1, [t/x]\Gamma_2 \vdash [t/x](\lambda y.t'') : [t/x]((y :_{(s', r')} C) \to D)$$

Which matches the goal.

---

Case.

$$\dfrac{\begin{array}{c}(\Delta, \sigma_1, \Delta', \sigma_7 \mid \sigma_9, r' \mid \mathbf{0}) \odot \Gamma_1, x : A, \Gamma_2, y : C \vdash D : \mathsf{Type}_l \\ (\Delta, \sigma_1, \Delta' \mid \sigma_8 \mid \sigma_7 + \sigma_9) \odot \Gamma_1, x : A, \Gamma_2 \vdash t_1 : (y :_{(s', r')} C) \to D \\ (\Delta, \sigma_1, \Delta' \mid \sigma_{10} \mid \sigma_7) \odot \Gamma_1, x : A, \Gamma_2 \vdash t_2 : C \end{array}}{(\Delta, \sigma_1, \Delta' \mid \sigma_8 + s' * \sigma_{10} \mid \sigma_9 + r' * \sigma_{10}) \odot \Gamma_1, x : A, \Gamma_2 \vdash t_1\,t_2 : [t_2/y]D} \text{ T\_App}$$

Therefore we have $\sigma_3, s, \sigma_4 = \sigma_8 + s' * \sigma_{10}$; $\sigma_8 = \hat{\sigma_3}, \hat{s}, \hat{\sigma_4}$ (where $|\hat{\sigma_3}| = |\Gamma_1|$); $\sigma_{10} = \hat{\sigma_3}', \hat{s}', \hat{\sigma_4}'$ (where $|\hat{\sigma_3}'| = |\Gamma_1|$); $\sigma_5, r, \sigma_6 = \sigma_7 + r' * \sigma_{10}$; $\sigma_7 = \hat{\sigma_5}, \hat{r}, \hat{\sigma_6}$ (where $|\hat{\sigma_5}| = |\Gamma_1|$); $\sigma_9 = \hat{\sigma_5}', \hat{r}', \hat{\sigma_6}'$ (where $|\hat{\sigma_5}'| = |\Gamma_1|$); $t' = t_1\,t_2$; and $B = [t_2/y]D$. We have $\hat{\sigma_3} + s' * \hat{\sigma_3}' = \sigma_3$, $\hat{\sigma_5} + r' * \hat{\sigma_5}' = \sigma_5$,

and likewise for other components. Our goal is therefore:

$$\begin{pmatrix} \Delta, (\Delta'\backslash\pi + (\Delta'/\pi)*\sigma_2) \\ (\sigma_3 + s*\sigma_2), \sigma_4 \\ (\sigma_5 + r*\sigma_2), \sigma_6 \end{pmatrix} \odot \Gamma_1, [t/x]\Gamma_2 \vdash [t/x](t_1\ t_2) : [t/x][t_2/y]D$$

Rewriting by Lemma 3.20, we have $\sigma_7 + \sigma_9 = (\hat{\sigma_5} + \hat{\sigma_5}'), (\hat{r} + \hat{r}'), (\hat{\sigma_6} + \hat{\sigma_6}')$. Rewriting the premises of T_App using the above information, we have the following inductive premises:

1. $(\Delta \mid \sigma_2 \mid \sigma_1) \odot \Gamma_1 \vdash t : A$ (premise (1))
2. $\begin{pmatrix} \Delta, \sigma_1, \Delta', (\hat{\sigma_5}, \hat{r}, \hat{\sigma_6}) \\ \hat{\sigma_5}', \hat{r}', \hat{\sigma_6}', r' \\ \mathbf{0} \end{pmatrix} \odot \Gamma_1, x : A, \Gamma_2, y : C \vdash D : \mathsf{Type}_l$ (T_App premise for $D$)
3. $|\hat{\sigma_5}'| = |\hat{\sigma_5} + \hat{\sigma_5}'| = |\Gamma_1|$ (trivially, and by premise (3))

And:

1. $(\Delta \mid \sigma_2 \mid \sigma_1) \odot \Gamma_1 \vdash t : A$ (premise (1))
2. $\begin{pmatrix} \Delta, \sigma_1, \Delta' \\ \hat{\sigma_3}, \hat{s}, \hat{\sigma_4} \\ (\hat{\sigma_5}+\hat{\sigma_5}'),(\hat{r}+\hat{r}'),(\hat{\sigma_6}+\hat{\sigma_6}') \end{pmatrix} \odot \Gamma_1, x : A, \Gamma_2 \vdash t_1 : (y :_{(s',r')} C) \to D$ (T_App premise for $t_1$)
3. $|\hat{\sigma_3}| = |\hat{\sigma_5} + \hat{\sigma_5}'| = |\Gamma_1|$ (trivially, and by premise (3))

And:

1. $(\Delta \mid \sigma_2 \mid \sigma_1) \odot \Gamma_1 \vdash t : A$ (premise (1))
2. $\begin{pmatrix} \Delta, \sigma_1, \Delta' \\ \hat{\sigma_3}', \hat{s}', \hat{\sigma_4}' \\ \hat{\sigma_5}, \hat{r}, \hat{\sigma_6} \end{pmatrix} \odot \Gamma_1, x : A, \Gamma_2 \vdash t_2 : C$ (T_App premise for $t_2$)
3. $|\hat{\sigma_3}'| = |\hat{\sigma_5}| = |\Gamma_1|$ (trivially, and by premise (3))

Giving, respectively:

$$\begin{pmatrix} \Delta, (((\Delta',(\hat{\sigma_5},\hat{r},\hat{\sigma_6}))\backslash\pi) + ((\Delta',(\hat{\sigma_5},\hat{r},\hat{\sigma_6}))/\pi)*\sigma_2) \\ (\hat{\sigma_5}'+\hat{r}'*\sigma_2),\hat{\sigma_6}',r' \\ \mathbf{0} \end{pmatrix} \odot \Gamma_1, [t/x]\Gamma_2, y : [t/x]C \vdash [t/x]D : \mathsf{Type}_l$$

And:

$$\begin{pmatrix} \Delta, (\Delta'\backslash\pi + (\Delta'/\pi)*\sigma_2) \\ (\hat{\sigma_3}+\hat{s}*\sigma_2),\hat{\sigma_4} \\ (\hat{\sigma_5}+\hat{\sigma_5}'+(\hat{r}+\hat{r}')*\sigma_2),(\hat{\sigma_6}+\hat{\sigma_6}') \end{pmatrix} \odot \Gamma_1, [t/x]\Gamma_2 \vdash [t/x]t_1 : (y :_{(s',r')} [t/x]C) \to [t/x]D$$

And:

$$\begin{pmatrix} \Delta, (\Delta'\backslash\pi + (\Delta'/\pi)*\sigma_2) \\ (\hat{\sigma_3}'+\hat{s}'*\sigma_2),\hat{\sigma_4}' \\ (\hat{\sigma_5}+\hat{r}*\sigma_2),\hat{\sigma_6} \end{pmatrix} \odot \Gamma_1, [t/x]\Gamma_2 \vdash [t/x]t_2 : [t/x]C$$

By vector rewriting $(\hat{\sigma_5} + \hat{\sigma_5}' + (\hat{r} + \hat{r}')*\sigma_2), (\hat{\sigma_6} + \hat{\sigma_6}')$ we have $(\hat{\sigma_5}+\hat{r}*\sigma_2), \hat{\sigma_6} + (\hat{\sigma_5}'+\hat{r}'*\sigma_2), \hat{\sigma_6}'$. Rewriting $((\Delta', (\hat{\sigma_5}, \hat{r}, \hat{\sigma_6}))\backslash\pi) + ((\Delta', (\hat{\sigma_5}, \hat{r}, \hat{\sigma_6}))/\pi)*\sigma_2$ by Lemma 3.17, we have $(\Delta'\backslash\pi + (\Delta'/\pi)*\sigma_2), ((\hat{\sigma_5} + \hat{r}*\sigma_2), \hat{\sigma_6})$, and can thus apply our inductive hypotheses to T_App, giving:

$$\cfrac{\begin{pmatrix} \Delta,(\Delta'\backslash\pi+(\Delta'/\pi)*\sigma_2),((\hat{\sigma_5}+\hat{r}*\sigma_2),\hat{\sigma_6}) \\ (\hat{\sigma_5}'+\hat{r}'*\sigma_2),\hat{\sigma_6}',r' \\ \mathbf{0} \end{pmatrix} \odot \Gamma_1, [t/x]\Gamma_2, y : [t/x]C \vdash [t/x]D : \mathsf{Type}_l \\ \begin{pmatrix} \Delta,(\Delta'\backslash\pi+(\Delta'/\pi)*\sigma_2) \\ (\hat{\sigma_3}+\hat{s}*\sigma_2),\hat{\sigma_4} \\ (\hat{\sigma_5}+\hat{r}*\sigma_2),\hat{\sigma_6}+(\hat{\sigma_5}'+\hat{r}'*\sigma_2),\hat{\sigma_6}' \end{pmatrix} \odot \Gamma_1, [t/x]\Gamma_2 \vdash [t/x]t_1 : (y :_{(s',r')} [t/x]C) \to [t/x]D \\ \begin{pmatrix} \Delta,(\Delta'\backslash\pi+(\Delta'/\pi)*\sigma_2) \\ (\hat{\sigma_3}'+\hat{s}'*\sigma_2),\hat{\sigma_4}' \\ (\hat{\sigma_5}+\hat{r}*\sigma_2),\hat{\sigma_6} \end{pmatrix} \odot \Gamma_1, [t/x]\Gamma_2 \vdash [t/x]t_2 : [t/x]C}{\begin{pmatrix} \Delta,(\Delta'\backslash\pi+(\Delta'/\pi)*\sigma_2) \\ (\hat{\sigma_3}+\hat{s}*\sigma_2),\hat{\sigma_4}+s'*(\hat{\sigma_3}'+\hat{s}'*\sigma_2),\hat{\sigma_4}' \\ (\hat{\sigma_5}'+\hat{r}'*\sigma_2),\hat{\sigma_6}'+r'*(\hat{\sigma_3}'+\hat{s}'*\sigma_2),\hat{\sigma_4}' \end{pmatrix} \odot \Gamma_1, x : A, \Gamma_2 \vdash ([t/x]t_1)\,([t/x]t_2) : [t_2/y][t/x]D}\ \text{T\_App}$$

Rewriting the subject and subject type grades of this conclusion by Lemma 3.18 (and expanding the scaling), we have $(\hat{\sigma_3}+\hat{s}*\sigma_2), \hat{\sigma_4}+(s'*\hat{\sigma_3}'+s'*\hat{s}'*\sigma_2), (s'*\hat{\sigma_4}') = (\sigma_3+s*\sigma_2), \sigma_4$ and $(\hat{\sigma_5}+\hat{r}*\sigma_2), \hat{\sigma_6}+(r'*\hat{\sigma_5}'+r'*\hat{r}'*\sigma_2), (r'*\hat{\sigma_6}') = (\sigma_5+r*\sigma_2), \sigma_6$. By definition we have $([t/x]t_1)([t/x]t_2) = [t/x](t_1\ t_2)$. As $x \neq y$ (by disjointness of concatenation), we have $[t_2/y][t/x]D = [t/x][t_2/y]D$. Using this information to rewrite the conclusion of T_APP derived above, we have:

$$\begin{pmatrix} \Delta, (\Delta' \backslash \pi + (\Delta'/\pi)*\sigma_2) \\ (\sigma_3+s*\sigma_2), \sigma_4 \\ (\sigma_5+r*\sigma_2), \sigma_6 \end{pmatrix} \odot \Gamma_1, [t/x]\Gamma_2 \vdash [t/x](t_1\ t_2) : [t/x][t_2/y]D$$

Which matches the goal.

Case.

$$\frac{(\Delta, \sigma_1, \Delta', \sigma_7 \mid \sigma_9, r' \mid \mathbf{0}) \odot \Gamma_1, x:A, \Gamma_2, y:C \vdash D : \mathsf{Type}_l \quad (\Delta, \sigma_1, \Delta' \mid \sigma_8 \mid \sigma_7) \odot \Gamma_1, x:A, \Gamma_2 \vdash t_1 : C \quad (\Delta, \sigma_1, \Delta' \mid \sigma_{10} \mid \sigma_9 + r'*\sigma_8) \odot \Gamma_1, x:A, \Gamma_2 \vdash t_2 : [t_1/y]D}{(\Delta, \sigma_1, \Delta' \mid \sigma_8 + \sigma_{10} \mid \sigma_7 + \sigma_9) \odot \Gamma_1, x:A, \Gamma_2 \vdash (t_1, t_2) : (y :_{r'} C) \otimes D} \text{ T\_PAIR}$$

Therefore we have $\sigma_3, s, \sigma_4 = \sigma_8 + \sigma_{10}$; $\sigma_8 = \hat{\sigma_3}, \hat{s}, \hat{\sigma_4}$ (where $|\hat{\sigma_3}| = |\Gamma_1|$); $\sigma_{10} = \hat{\sigma_3}', \hat{s}', \hat{\sigma_4}'$ (where $|\hat{\sigma_3}'| = |\Gamma_1|$); $\sigma_5, r, \sigma_6 = \sigma_7 + \sigma_9$; $\sigma_7 = \hat{\sigma_5}, \hat{r}, \hat{\sigma_6}$ (where $|\hat{\sigma_5}| = |\Gamma_1|$); $\sigma_9 = \hat{\sigma_5}', \hat{r}', \hat{\sigma_6}'$ (where $|\hat{\sigma_5}'| = |\Gamma_1|$); $t' = (t_1, t_2)$; and $B = (y :_{r'} C) \otimes D$. We have $\hat{\sigma_3} + \hat{\sigma_3}' = \sigma_3$, and likewise for other components. Our goal is therefore:

$$\begin{pmatrix} \Delta, (\Delta' \backslash \pi + (\Delta'/\pi)*\sigma_2) \\ (\sigma_3+s*\sigma_2), \sigma_4 \\ (\sigma_5+r*\sigma_2), \sigma_6 \end{pmatrix} \odot \Gamma_1, [t/x]\Gamma_2 \vdash [t/x](t_1, t_2) : [t/x]((y :_{r'} C) \otimes D)$$

As $\sigma_9 = \hat{\sigma_5}', \hat{r}', \hat{\sigma_6}'$, $\sigma_8 = \hat{\sigma_3}, \hat{s}, \hat{\sigma_4}$, and $|\hat{\sigma_5}'| = |\hat{\sigma_3}|$, we can rewrite $\sigma_9 + r'*\sigma_8$ by Lemma 3.20 to obtain $(\hat{\sigma_5}' + r'*\hat{\sigma_3}), (\hat{r}' + r'*\hat{s}), (\hat{\sigma_6}' + r'*\hat{\sigma_4})$. Rewriting the premises of T_PAIR using the above information, we have the following inductive premises:

1. $(\Delta \mid \sigma_2 \mid \sigma_1) \odot \Gamma_1 \vdash t : A$ (premise (1))
2. $\begin{pmatrix} \Delta, \sigma_1, \Delta', (\hat{\sigma_5}, \hat{r}, \hat{\sigma_6}) \\ \hat{\sigma_5}', \hat{r}', \hat{\sigma_6}', r' \\ \mathbf{0} \end{pmatrix} \odot \Gamma_1, x:A, \Gamma_2, y:C \vdash D : \mathsf{Type}_l$ (first premise of T_PAIR)
3. $|\hat{\sigma_5}'| = |\mathbf{0}| = |\Gamma_1|$ (trivially, and by premise (3))

And:

1. $(\Delta \mid \sigma_2 \mid \sigma_1) \odot \Gamma_1 \vdash t : A$ (premise (1))
2. $\begin{pmatrix} \Delta, \sigma_1, \Delta' \\ \hat{\sigma_3}, \hat{s}, \hat{\sigma_4} \\ \hat{\sigma_5}, \hat{r}, \hat{\sigma_6} \end{pmatrix} \odot \Gamma_1, x:A, \Gamma_2 \vdash t_1 : C$ (second premise of T_PAIR)
3. $|\hat{\sigma_3}| = |\hat{\sigma_5}| = |\Gamma_1|$ (trivially, and by premise (3))

And:

1. $(\Delta \mid \sigma_2 \mid \sigma_1) \odot \Gamma_1 \vdash t : A$ (premise (1))
2. $\begin{pmatrix} \Delta, \sigma_1, \Delta' \\ \hat{\sigma_3}', \hat{s}', \hat{\sigma_4}' \\ (\hat{\sigma_5}'+r'*\hat{\sigma_3}), (\hat{r}'+r'*\hat{s}), (\hat{\sigma_6}'+r'*\hat{\sigma_4}) \end{pmatrix} \odot \Gamma_1, x:A, \Gamma_2 \vdash t_2 : [t_1/y]D$ (third premise of T_PAIR)
3. $|\hat{\sigma_3}'| = |\hat{\sigma_5}'+r'*\hat{\sigma_3}| = |\Gamma_1|$ (trivially, and by premise (3))

Giving, respectively:

$$\begin{pmatrix} \Delta,(((\Delta',(\hat{\sigma_5},\hat{r},\hat{\sigma_6}))\backslash\pi)+((\Delta',(\hat{\sigma_5},\hat{r},\hat{\sigma_6}))/\pi)*\sigma_2) \\ (\hat{\sigma_5}'+\hat{r}'*\sigma_2),\hat{\sigma_6}',r' \\ \mathbf{0} \end{pmatrix} \odot \Gamma_1, [t/x]\Gamma_2, y : [t/x]C \vdash [t/x]D : \mathsf{Type}_l$$

And:

$$\begin{pmatrix} \Delta,\sigma_1,(\Delta'\backslash\pi+(\Delta'/\pi)*\sigma_2) \\ (\hat{\sigma_3}+\hat{s}*\sigma_2),\hat{\sigma_4} \\ (\hat{\sigma_5}+\hat{r}*\sigma_2),\hat{\sigma_6} \end{pmatrix} \odot \Gamma_1, [t/x]\Gamma_2 \vdash [t/x]t_1 : [t/x]C$$

And:

$$\begin{pmatrix} \Delta,\sigma_1,(\Delta'\backslash\pi+(\Delta'/\pi)*\sigma_2) \\ (\hat{\sigma_3}'+\hat{s}'*\sigma_2),\hat{\sigma_4}' \\ ((\hat{\sigma_5}'+r'*\hat{\sigma_3})+(\hat{r}'+r'*\hat{s})*\sigma_2),(\hat{\sigma_6}'+r'*\hat{\sigma_4}) \end{pmatrix} \odot \Gamma_1, [t/x]\Gamma_2 \vdash [t/x]t_2 : [t/x][t_1/y]D$$

By vector rewriting $((\hat{\sigma_5}'+r'*\hat{\sigma_3})+(\hat{r}'+r'*\hat{s})*\sigma_2),(\hat{\sigma_6}'+r'*\hat{\sigma_4})$ we have $(\hat{\sigma_5}'+\hat{r}'*\sigma_2),\hat{\sigma_6}'+r'*(\hat{\sigma_3}+\hat{s}*\sigma_2),\hat{\sigma_4}$. As $y \notin \mathsf{FV}(t)$ and $x \neq y$, we have $[t/x][t_1/y]D = [[t/x]t_1/y][t/x]D$.
Rewriting $((\Delta',(\hat{\sigma_5},\hat{r},\hat{\sigma_6}))\backslash\pi)+((\Delta',(\hat{\sigma_5},\hat{r},\hat{\sigma_6}))/\pi)*\sigma_2$ by Lemma 3.17, we have $(\Delta'\backslash\pi+(\Delta'/\pi)*\sigma_2),((\hat{\sigma_5}+\hat{r}*\sigma_2),\hat{\sigma_6})$. Thus we can apply our inductive hypotheses via T_PAIR, giving:

$$\cfrac{\begin{pmatrix} \Delta,((\Delta'\backslash\pi+(\Delta'/\pi)*\sigma_2),((\hat{\sigma_5}+\hat{r}*\sigma_2),\hat{\sigma_6})) \\ (\hat{\sigma_5}'+\hat{r}'*\sigma_2),\hat{\sigma_6}',r' \\ \mathbf{0} \end{pmatrix} \odot \Gamma_1, [t/x]\Gamma_2, y : [t/x]C \vdash [t/x]D : \mathsf{Type}_l \\ \begin{pmatrix} \Delta,\sigma_1,(\Delta'\backslash\pi+(\Delta'/\pi)*\sigma_2) \\ (\hat{\sigma_3}+\hat{s}*\sigma_2),\hat{\sigma_4} \\ (\hat{\sigma_5}+\hat{r}*\sigma_2),\hat{\sigma_6} \end{pmatrix} \odot \Gamma_1, [t/x]\Gamma_2 \vdash [t/x]t_1 : [t/x]C \\ \begin{pmatrix} \Delta,\sigma_1,(\Delta'\backslash\pi+(\Delta'/\pi)*\sigma_2) \\ (\hat{\sigma_3}'+\hat{s}'*\sigma_2),\hat{\sigma_4}' \\ (\hat{\sigma_5}'+\hat{r}'*\sigma_2),\hat{\sigma_6}'+r'*((\hat{\sigma_3}+\hat{s}*\sigma_2),\hat{\sigma_4}) \end{pmatrix} \odot \Gamma_1, [t/x]\Gamma_2 \vdash [t/x]t_2 : [[t/x]t_1/y][t/x]D}{\begin{pmatrix} \Delta,(\Delta'\backslash\pi+(\Delta'/\pi)*\sigma_2) \\ (\hat{\sigma_3}+\hat{s}*\sigma_2),\hat{\sigma_4}+(\hat{\sigma_3}'+\hat{s}'*\sigma_2),\hat{\sigma_4}' \\ (\hat{\sigma_5}+\hat{r}*\sigma_2),\hat{\sigma_6}+(\hat{\sigma_5}'+\hat{r}'*\sigma_2),\hat{\sigma_6}' \end{pmatrix} \odot \Gamma_1, [t/x]\Gamma_2 \vdash ([t/x]t_1,[t/x]t_2) : (y :_{r'} [t/x]C) \otimes [t/x]D} \text{ T\_PAIR}$$

Rewriting the subject and subject type grades of this conclusion by Lemma 3.18, we have $(\hat{\sigma_3}+\hat{s}*\sigma_2),\hat{\sigma_4}+(\hat{\sigma_3}'+\hat{s}'*\sigma_2),\hat{\sigma_4}' = (\sigma_3+s*\sigma_2),\sigma_4$ and $(\hat{\sigma_5}+\hat{r}*\sigma_2),\hat{\sigma_6}+(\hat{\sigma_5}'+\hat{r}'*\sigma_2),\hat{\sigma_6}' = (\sigma_5+r*\sigma_2),\sigma_6$. By definition we have $([t/x]t_1,[t/x]t_2) = [t/x](t_1,t_2)$. As $x \neq y$ (by disjointness of concatenation), we have $(y :_{r'} [t/x]C) \otimes [t/x]D = [t/x]((y :_{r'} C) \otimes D)$. Using this information to rewrite the conclusion of T_PAIR derived above, we have:

$$\begin{pmatrix} \Delta,(\Delta'\backslash\pi+(\Delta'/\pi)*\sigma_2) \\ (\sigma_3+s*\sigma_2),\sigma_4 \\ (\sigma_5+r*\sigma_2),\sigma_6 \end{pmatrix} \odot \Gamma_1, [t/x]\Gamma_2 \vdash [t/x](t_1,t_2) : [t/x]((y :_{r'} C) \otimes D)$$

Which matches the goal.

Case.

$$\cfrac{\begin{array}{c}(\Delta,\sigma_1,\Delta' \mid \sigma_9 \mid \sigma_7+\sigma_8) \odot \Gamma_1, x : A, \Gamma_2 \vdash t_1 : (y :_{r'} C) \otimes D \\ \begin{pmatrix} \Delta,\sigma_1,\Delta',\sigma_7,(\sigma_8,r') \\ \sigma_{10},s',s' \\ \sigma_{11},r'',r'' \end{pmatrix} \odot \Gamma_1, x : A, \Gamma_2, y : C, z : D \vdash t_2 : [(y,z)/w]E \\ (\Delta,\sigma_1,\Delta',(\sigma_7+\sigma_8) \mid \sigma_{11}, r'' \mid \mathbf{0}) \odot \Gamma_1, x : A, \Gamma_2, w : (y :_{r'} C) \otimes D \vdash E : \mathsf{Type}_l \end{array}}{(\Delta,\sigma_1,\Delta' \mid \sigma_{10}+s'*\sigma_9 \mid \sigma_{11}+r''*\sigma_9) \odot \Gamma_1, x : A, \Gamma_2 \vdash \mathsf{let}\ (y,z) = t_1\ \mathsf{in}\ t_2 : [t_1/w]E} \text{ T\_TENCUT}$$

Therefore we have $\sigma_3,s,\sigma_4 = \sigma_{10}+s'*\sigma_9$; $\sigma_5,r,\sigma_6 = \sigma_{11}+r''*\sigma_9$; $\sigma_7 = \hat{\sigma_7},s_7,\hat{\sigma_7}'$ (where $|\hat{\sigma_7}| = |\Gamma_1|$); $\sigma_8 = \hat{\sigma_8},s_8,\hat{\sigma_8}'$ (where $|\hat{\sigma_8}| = |\Gamma_1|$); $t' = \mathsf{let}\ (y,z) = t_1\ \mathsf{in}\ t_2$; and $B = [t_1/w]E$. Our

goal is therefore:

$$\begin{pmatrix} \Delta,(\Delta'\backslash\pi+(\Delta'/\pi)*\sigma_2) \\ (\sigma_3+s*\sigma_2),\sigma_4 \\ (\sigma_5+r*\sigma_2),\sigma_6 \end{pmatrix} \odot \Gamma_1, [t/x]\Gamma_2 \vdash [t/x](\text{let } (y,z) = t_1 \text{ in } t_2) : [t/x][t_1/w]E$$

As $\sigma_3, s, \sigma_4 = \sigma_{10} + s'*\sigma_9$, we therefore have $\sigma_{10} = \hat{\sigma_3}, \hat{s}, \hat{\sigma_4}$ and $\sigma_9 = \hat{\sigma_3}', \hat{s}', \hat{\sigma_4}'$ where $\hat{\sigma_3} + s'*\hat{\sigma_3}' = \sigma_3$; $\hat{s}+s'*\hat{s}' = s$; and $\hat{\sigma_4}+s'*\hat{\sigma_4}' = \sigma_4$. As $\sigma_5, r, \sigma_6 = \sigma_{11}+r''*\sigma_9$, we therefore have $\sigma_{11} = \hat{\sigma_5}, \hat{r}, \hat{\sigma_6}$ where $\hat{\sigma_5} + r'' * \hat{\sigma_3}' = \sigma_5$; $\hat{r} + r'' * \hat{s}' = r$; and $\hat{\sigma_6} + r'' * \hat{\sigma_4}' = \sigma_6$. As $\sigma_7 = \hat{\sigma_7}, s_7, \hat{\sigma_7}'$ and $\sigma_8 = \hat{\sigma_8}, s_8, \hat{\sigma_8}'$, we have $\sigma_7 + \sigma_8 = (\hat{\sigma_7}+\hat{\sigma_8}), (s_7+s_8), (\hat{\sigma_7}'+\hat{\sigma_8}')$. Rewriting the premises using the above information, we have the following inductive premises:

1. $(\Delta \mid \sigma_2 \mid \sigma_1) \odot \Gamma_1 \vdash t : A$ (premise (1))
2. $\begin{pmatrix} \Delta,\sigma_1,\Delta' \\ \hat{\sigma_3}',\hat{s}',\hat{\sigma_4}' \\ (\hat{\sigma_7}+\hat{\sigma_8}),(s_7+s_8),(\hat{\sigma_7}'+\hat{\sigma_8}') \end{pmatrix} \odot \Gamma_1, x:A, \Gamma_2 \vdash t_1 : (y :_{r'} C) \otimes D$ ($t_1$ premise of T_TenCut)
3. $|\hat{\sigma_3}'| = |\hat{\sigma_7} + \hat{\sigma_8}| = |\Gamma_1|$ (trivially, and by premise (3))

And:

1. $(\Delta \mid \sigma_2 \mid \sigma_1) \odot \Gamma_1 \vdash t : A$ (premise (1))
2. $\begin{pmatrix} \Delta,\sigma_1,\Delta',(\hat{\sigma_7},s_7,\hat{\sigma_7}'),(\hat{\sigma_8},s_8,\hat{\sigma_8}',r') \\ \hat{\sigma_3},\hat{s},\hat{\sigma_4},s',s' \\ \hat{\sigma_5},\hat{r},\hat{\sigma_6},r'',r'' \end{pmatrix} \odot \Gamma_1, x:A, \Gamma_2, y:C, z:D \vdash t_2 : [(y,z)/w]E$ ($t_2$ premise of T_TenCut)
3. $|\hat{\sigma_3}| = |\hat{\sigma_5}| = |\Gamma_1|$ (trivially, and by premise (3))

And:

1. $(\Delta \mid \sigma_2 \mid \sigma_1) \odot \Gamma_1 \vdash t : A$ (premise (1))
2. $\begin{pmatrix} \Delta,\sigma_1,\Delta',((\hat{\sigma_7}+\hat{\sigma_8}),(s_7+s_8),(\hat{\sigma_7}'+\hat{\sigma_8}')) \\ \hat{\sigma_5},\hat{r},\hat{\sigma_6},r'' \\ \mathbf{0} \end{pmatrix} \odot \Gamma_1, x:A, \Gamma_2, w:(y :_{r'} C) \otimes D \vdash E : \text{Type}_l$ ($E$ premise of T_TenCut)
3. $|\hat{\sigma_3}| = |\sigma_5| = |\Gamma_1|$ (trivially, and by premise (3))

Giving, respectively:

$$\begin{pmatrix} \Delta,(\Delta'\backslash\pi+(\Delta'/\pi)*\sigma_2) \\ (\hat{\sigma_3}'+\hat{s}'*\sigma_2),\hat{\sigma_4}' \\ (\hat{\sigma_7}+\hat{\sigma_8}+(s_7+s_8)*\sigma_2),(\hat{\sigma_7}'+\hat{\sigma_8}') \end{pmatrix} \odot \Gamma_1, [t/x]\Gamma_2 \vdash [t/x]t_1 : (y :_{r'} [t/x]C) \otimes [t/x]D$$

And:

$$\begin{pmatrix} \Delta,(((\Delta',(\hat{\sigma_7},s_7,\hat{\sigma_7}'),(\hat{\sigma_8},s_8,\hat{\sigma_8}',r'))\backslash\pi)+((\Delta',(\hat{\sigma_7},s_7,\hat{\sigma_7}'),(\hat{\sigma_8},s_8,\hat{\sigma_8}',r'))/\pi)*\sigma_2) \\ (\hat{\sigma_3}+\hat{s}*\sigma_2),\hat{\sigma_4},s',s' \\ (\hat{\sigma_5}+\hat{r}*\sigma_2),\hat{\sigma_6},r'',r'' \end{pmatrix} \odot \Gamma_1, [t/x]\Gamma_2, y:[t/x]C, z:[t/x]D$$
$$\vdash [t/x]t_2 : [t/x][(y,z)/w]E$$

And:

$$\begin{pmatrix} \Delta,(((\Delta',(\hat{\sigma_7}+\hat{\sigma_8}),(s_7+s_8),(\hat{\sigma_7}'+\hat{\sigma_8}'))\backslash\pi)+((\Delta',(\hat{\sigma_7}+\hat{\sigma_8}),(s_7+s_8),(\hat{\sigma_7}'+\hat{\sigma_8}'))/\pi)*\sigma_2) \\ (\hat{\sigma_5}+\hat{r}*\sigma_2),\hat{\sigma_6},r'' \\ \mathbf{0} \end{pmatrix} \odot \Gamma_1, [t/x]\Gamma_2, w:(y :_{r'} [t/x]C) \otimes [t/x]D$$
$$\vdash [t/x]E : \text{Type}_l$$

As $w \neq x$ (by disjointness of concatenation), we have $[t/x][(y,z)/w]E = [(y,z)/w][t/x]E$. By vector rewriting $(\hat{\sigma_7}+\hat{\sigma_8}+(s_7+s_8)*\sigma_2), (\hat{\sigma_7}'+\hat{\sigma_8}')$ we have $(\hat{\sigma_7}+s_7*\sigma_2), \hat{\sigma_7}' + (\hat{\sigma_8}+s_8*\sigma_2), \hat{\sigma_8}'$.

Rewriting $((\Delta', (\hat{\sigma_7}, s_7, \hat{\sigma_7}'), (\hat{\sigma_8}, s_8, \hat{\sigma_8}'))\backslash\pi) + ((\Delta', (\hat{\sigma_7}, s_7, \hat{\sigma_7}'), (\hat{\sigma_8}, s_8, \hat{\sigma_8}'))/\pi) * \sigma_2$ by Lemma 3.17, we have $(\Delta'\backslash\pi + (\Delta'/\pi) * \sigma_2), ((\hat{\sigma_7} + s_7 * \sigma_2), \hat{\sigma_7}'), ((\hat{\sigma_8} + s_8 * \sigma_2), \hat{\sigma_8}', r')$. By vector rewriting $(((\Delta', (\hat{\sigma_7} + \hat{\sigma_8}), (s_7 + s_8), (\hat{\sigma_7}' + \hat{\sigma_8}'))\backslash\pi) + ((\Delta', (\hat{\sigma_7} + \hat{\sigma_8}), (s_7 + s_8), (\hat{\sigma_7}' + \hat{\sigma_8}'))/\pi) * \sigma_2)$ we have $(\Delta'\backslash\pi + (\Delta'/\pi) * \sigma_2), (((\hat{\sigma_7} + s_7 * \sigma_2), \hat{\sigma_7}') + ((\hat{\sigma_8} + s_8 * \sigma_2), \hat{\sigma_8}'))$. Therefore we can apply our inductive hypotheses via T_TENCUT, giving:

$$\frac{\begin{pmatrix}\Delta,(\Delta'\backslash\pi+(\Delta'/\pi)*\sigma_2)\\(\hat{\sigma_3}'+\hat{s}'*\sigma_2),\hat{\sigma_4}'\\((\hat{\sigma_7}+s_7*\sigma_2),\hat{\sigma_7}')+((\hat{\sigma_8}+s_8*\sigma_2),\hat{\sigma_8}')\end{pmatrix} \odot \Gamma_1, [t/x]\Gamma_2 \vdash [t/x]t_1 : (y :_{r'} [t/x]C) \otimes [t/x]D \quad \begin{pmatrix}\Delta,(\Delta'\backslash\pi+(\Delta'/\pi)*\sigma_2),((\hat{\sigma_7}+s_7*\sigma_2),\hat{\sigma_7}'),((\hat{\sigma_8}+s_8*\sigma_2),\hat{\sigma_8}',r')\\(\hat{\sigma_3}+\hat{s}*\sigma_2),\hat{\sigma_4},s',s'\\(\hat{\sigma_5}+\hat{r}*\sigma_2),\hat{\sigma_6},r'',r''\end{pmatrix} \odot \Gamma_1, [t/x]\Gamma_2, y:[t/x]C, z:[t/x]D \vdash [t/x]t_2 : [(y,z)/w][t/x]E \quad \begin{pmatrix}\Delta,(\Delta'\backslash\pi+(\Delta'/\pi)*\sigma_2),(((\hat{\sigma_7}+s_7*\sigma_2),\hat{\sigma_7}')+((\hat{\sigma_8}+s_8*\sigma_2),\hat{\sigma_8}'))\\(\hat{\sigma_5}+\hat{r}*\sigma_2),\hat{\sigma_6},r''\\\mathbf{0}\end{pmatrix} \odot \Gamma_1, [t/x]\Gamma_2, w : (y :_{r'} [t/x]C) \otimes [t/x]D \vdash [t/x]E : \mathsf{Type}_l}{\begin{pmatrix}\Delta,(\Delta'\backslash\pi+(\Delta'/\pi)*\sigma_2)\\((\hat{\sigma_3}+\hat{s}*\sigma_2),\hat{\sigma_4})+s'*((\hat{\sigma_3}'+\hat{s}'*\sigma_2),\hat{\sigma_4}')\\((\hat{\sigma_5}+\hat{r}*\sigma_2),\hat{\sigma_6})+r''*((\hat{\sigma_3}'+\hat{s}'*\sigma_2),\hat{\sigma_4}')\end{pmatrix} \odot \Gamma_1, [t/x]\Gamma_2 \vdash \mathsf{let}\,(y,z) = [t/x]t_1\,\mathsf{in}\,[t/x]t_2 : [t_1/w][t/x]E} \text{ T\_TENCUT}$$

Rewriting the subject and subject-type grades of this conclusion by Lemma 3.18 (after expanding the scaling), we have $((\hat{\sigma_3} + \hat{s} * \sigma_2), \hat{\sigma_4}) + ((s' * \hat{\sigma_3}' + s' * \hat{s}' * \sigma_2), (s' * \hat{\sigma_4}')) = (\sigma_3 + s * \sigma_2), \sigma_4$, and $((\hat{\sigma_5} + \hat{r} * \sigma_2), \hat{\sigma_6}) + ((r'' * \hat{\sigma_3}' + r'' * \hat{s}' * \sigma_2), (r'' * \hat{\sigma_4}')) = (\sigma_5 + r * \sigma_2), \sigma_6$. As $x \neq y$ and $x \neq z$ (by disjointness of concatenation), we have $\mathsf{let}\,(y,z) = [t/x]t_1\,\mathsf{in}\,[t/x]t_2 = [t/x](\mathsf{let}\,(y,z) = t_1\,\mathsf{in}\,t_2)$. As $w \neq x$, and $w \notin \mathsf{FV}(t)$, we have $[t_1/w][t/x]E = [t/x][t_1/w]E$. Using this information to rewrite the conclusion of T_TENCUT derived above, we have:

$$\begin{pmatrix}\Delta,(\Delta'\backslash\pi+(\Delta'/\pi)*\sigma_2)\\(\sigma_3+s*\sigma_2),\sigma_4\\(\sigma_5+r*\sigma_2),\sigma_6\end{pmatrix} \odot \Gamma_1, [t/x]\Gamma_2 \vdash [t/x](\mathsf{let}\,(y,z) = t_1\,\mathsf{in}\,t_2) : [t/x][t_1/w]E$$

Which matches the goal.

Case.

$$\frac{(\Delta, \sigma_1, \Delta' \mid \sigma_3, s, \sigma_4 \mid \mathbf{0}) \odot \Gamma_1, x : A, \Gamma_2 \vdash C : \mathsf{Type}_{l_1}}{(\Delta, \sigma_1, \Delta' \mid \sigma_3, s, \sigma_4 \mid \mathbf{0}) \odot \Gamma_1, x : A, \Gamma_2 \vdash \Box_{s'}C : \mathsf{Type}_l} \text{ T\_BOX}$$

Therefore we have $\sigma_5 = \mathbf{0}$, $r = 0$, $\sigma_6 = \mathbf{0}$, $t' = \Box_{s'}C$, and $B = \mathsf{Type}_l$. Our full goal is therefore:

$$\begin{pmatrix}\Delta,(\Delta'\backslash\pi+(\Delta'/\pi)*\sigma_2)\\(\sigma_3+s*\sigma_2),\sigma_4\\\mathbf{0}\end{pmatrix} \odot \Gamma_1, [t/x]\Gamma_2 \vdash [t/x](\Box_{s'}C) : [t/x]\mathsf{Type}_l$$

Rewriting the premise using the above information, we have the following inductive premises:

1. $(\Delta \mid \sigma_2 \mid \sigma_1) \odot \Gamma_1 \vdash t : A$ (premise (1))
2. $(\Delta, \sigma_1, \Delta' \mid \sigma_3, s, \sigma_4 \mid \mathbf{0}) \odot \Gamma_1, x : A, \Gamma_2 \vdash C : \mathsf{Type}_l$ (premise of T_BOX)
3. $|\sigma_3| = |\mathbf{0}| = |\Gamma_1|$ (trivially, and by premise (3))

Giving:

$$\begin{pmatrix}\Delta,((\Delta'\backslash\pi)+(\Delta'/\pi)*\sigma_2)\\(\sigma_3+s*\sigma_2),\sigma_4\\\mathbf{0}\end{pmatrix} \odot \Gamma_1, [t/x]\Gamma_2 \vdash [t/x]C : [t/x]\mathsf{Type}_l$$

Rewriting $[t/x]\mathsf{Type}_l = \mathsf{Type}_l$ and applying the inductive hypothesis to T_Box, we have:

$$\frac{(\Delta,((\Delta'\backslash\pi)+(\Delta'/\pi)*\sigma_2) \mid (\sigma_3 + s*\sigma_2), \sigma_4 \mid \mathbf{0}) \odot \Gamma_1, [t/x]\Gamma_2 \vdash [t/x]C : \mathsf{Type}_{l_1}}{(\Delta,((\Delta'\backslash\pi)+(\Delta'/\pi)*\sigma_2) \mid (\sigma_3 + s*\sigma_2), \sigma_4 \mid \mathbf{0}) \odot \Gamma_1, [t/x]\Gamma_2 \vdash \Box_{s'}[t/x]C : \mathsf{Type}_l} \; \text{T\_Box}$$

As $[t/x]\mathsf{Type}_l = \mathsf{Type}_l$ and $[t/x](\Box_{s'}C) = \Box_{s'}[t/x]C$, we have:

$$\begin{pmatrix} \Delta, (\Delta'\backslash\pi + (\Delta'/\pi)*\sigma_2) \\ (\sigma_3 + s*\sigma_2), \sigma_4 \\ \mathbf{0} \end{pmatrix} \odot \Gamma_1, [t/x]\Gamma_2 \vdash [t/x](\Box_{s'}C) : [t/x]\mathsf{Type}_l$$

Which matches the goal.

Case.

$$\frac{(\Delta, \sigma_1, \Delta' \mid \sigma_7 \mid \sigma_5, r, \sigma_6) \odot \Gamma_1, x : A, \Gamma_2 \vdash t'' : C}{(\Delta, \sigma_1, \Delta' \mid s' * \sigma_7 \mid \sigma_5, r, \sigma_6) \odot \Gamma_1, x : A, \Gamma_2 \vdash \Box t'' : \Box_{s'} C} \; \text{T\_BoxI}$$

Therefore we have $\sigma_3, s, \sigma_4 = s' * \sigma_7$; $\sigma_7 = \hat{\sigma}_3, \hat{s}, \hat{\sigma}_4$ (where $|\hat{\sigma}_3| = |\Gamma_1|$); $t' = \Box t''$; and $B = \Box_{s'}C$. We have $s' * \hat{\sigma}_3 = \sigma_3$, and likewise for other components. Our goal is therefore:

$$\begin{pmatrix} \Delta, (\Delta'\backslash\pi + (\Delta'/\pi)*\sigma_2) \\ (\sigma_3 + s*\sigma_2), \sigma_4 \\ (\sigma_5 + r*\sigma_2), \sigma_6 \end{pmatrix} \odot \Gamma_1, [t/x]\Gamma_2 \vdash [t/x](\Box t'') : [t/x](\Box_{s'}C)$$

Rewriting the premise of T_BoxI using the above information, we have the following inductive premises:

1. $(\Delta \mid \sigma_2 \mid \sigma_1) \odot \Gamma_1 \vdash t : A$ (premise (1))
2. $\begin{pmatrix} \Delta, \sigma_1, \Delta' \\ \hat{\sigma}_3, \hat{s}, \hat{\sigma}_4 \\ \sigma_5, r, \sigma_6 \end{pmatrix} \odot \Gamma_1, x : A, \Gamma_2 \vdash t'' : C$ (premise of T_BoxI)
3. $|\hat{\sigma}_3| = |\sigma_5| = |\Gamma_1|$ (trivially, and by premise (3))

Giving:

$$\begin{pmatrix} \Delta, (\Delta'\backslash\pi + (\Delta'/\pi)*\sigma_2) \\ (\hat{\sigma}_3 + \hat{s}*\sigma_2), \hat{\sigma}_4 \\ (\sigma_5 + r*\sigma_2), \sigma_6 \end{pmatrix} \odot \Gamma_1, [t/x]\Gamma_2 \vdash [t/x]t'' : [t/x]C$$

We can apply our inductive hypothesis to T_BoxI, giving:

$$\frac{\begin{pmatrix} \Delta, (\Delta'\backslash\pi + (\Delta'/\pi)*\sigma_2) \\ (\hat{\sigma}_3 + \hat{s}*\sigma_2), \hat{\sigma}_4 \\ (\sigma_5 + r*\sigma_2), \sigma_6 \end{pmatrix} \odot \Gamma_1, [t/x]\Gamma_2 \vdash [t/x]t'' : [t/x]C}{\begin{pmatrix} \Delta, (\Delta'\backslash\pi + (\Delta'/\pi)*\sigma_2) \\ s'*(\hat{\sigma}_3 + \hat{s}*\sigma_2), \hat{\sigma}_4 \\ (\sigma_5 + r*\sigma_2), \sigma_6 \end{pmatrix} \odot \Gamma_1, x : A, \Gamma_2 \vdash \Box[t/x]t'' : \Box_{s'}[t/x]C} \; \text{T\_BoxI}$$

Rewriting the subject and subject type grades of this conclusion by Lemma 3.18 (and expanding the scaling), we have $(s' * \hat{\sigma}_3) + (s' * \hat{s} * \sigma_2), (s' * \hat{\sigma}_4) = (\sigma_3 + s * \sigma_2), \sigma_4$. By definition we have $\Box[t/x]t'' = [t/x](\Box t'')$ and $\Box_{s'}[t/x]C = [t/x](\Box_{s'}C)$. Using this information to rewrite the conclusion of T_BoxI derived above, we have:

$$\begin{pmatrix} \Delta, (\Delta'\backslash\pi + (\Delta'/\pi)*\sigma_2) \\ (\sigma_3 + s*\sigma_2), \sigma_4 \\ (\sigma_5 + r*\sigma_2), \sigma_6 \end{pmatrix} \odot \Gamma_1, [t/x]\Gamma_2 \vdash [t/x](\Box t'') : [t/x](\Box_{s'}C)$$

Which matches the goal.

Case.

$$\frac{\begin{array}{c}(\Delta, \sigma_1, \Delta', \sigma_8 \mid \sigma_{10}, r' \mid \mathbf{0}) \odot \Gamma_1, x : A, \Gamma_2, z : \Box_{s'}C \vdash D : \mathsf{Type}_l \\ (\Delta, \sigma_1, \Delta' \mid \sigma_7 \mid \sigma_8) \odot \Gamma_1, x : A, \Gamma_2 \vdash t_1 : \Box_{s'}C \\ (\Delta, \sigma_1, \Delta', \sigma_8 \mid \sigma_9, s' \mid \sigma_{10}, (s' * r')) \odot \Gamma_1, x : A, \Gamma_2, y : C \vdash t_2 : [\Box y/z]D\end{array}}{(\Delta, \sigma_1, \Delta' \mid \sigma_7 + \sigma_9 \mid \sigma_{10} + r' * \sigma_7) \odot \Gamma_1, x : A, \Gamma_2 \vdash \mathsf{let}\,\Box y = t_1 \,\mathsf{in}\, t_2 : [t_1/z]D} \; \text{T\_BoxE}$$

Therefore we have $\sigma_3, s, \sigma_4 = \sigma_7 + \sigma_9$; $\sigma_7 = \hat{\sigma_3}, \hat{s}, \hat{\sigma_4}$ (where $|\hat{\sigma_3}| = |\Gamma_1|$); $\sigma_8 = \hat{\sigma_8}, s'', \hat{\sigma_8}'$ (where $|\hat{\sigma_8}| = |\Gamma_1|$); $\sigma_9 = \hat{\sigma_3}', \hat{s}', \hat{\sigma_4}'$ (where $|\hat{\sigma_3}'| = |\Gamma_1|$); $\sigma_5, r, \sigma_6 = \sigma_{10} + r' * \sigma_7$; $\sigma_{10} = \hat{\sigma_5}, \hat{r}, \hat{\sigma_6}$ (where $|\hat{\sigma_5}| = |\Gamma_1|$); $t' = \mathsf{let}\,\Box y = t_1 \,\mathsf{in}\, t_2$; and $B = [t_1/z]D$. We have $\hat{\sigma_3} + \hat{\sigma_3}' = \sigma_3$, $\hat{\sigma_5} + r' * \hat{\sigma_3} = \sigma_5$, and likewise for other components. Our goal is therefore:

$$\begin{pmatrix}\Delta, (\Delta' \backslash \pi + (\Delta'/\pi)*\sigma_2) \\ (\sigma_3 + s*\sigma_2), \sigma_4 \\ (\sigma_5 + r*\sigma_2), \sigma_6\end{pmatrix} \odot \Gamma_1, [t/x]\Gamma_2 \vdash [t/x](\mathsf{let}\,\Box y = t_1 \,\mathsf{in}\, t_2) : [t/x]([t_2/z]D)$$

We have the following inductive premises:

1. $(\Delta \mid \sigma_2 \mid \sigma_1) \odot \Gamma_1 \vdash t : A$ (premise (1))
2. $\begin{pmatrix}\Delta, \sigma_1, \Delta', (\hat{\sigma_8}, s'', \hat{\sigma_8}') \\ \hat{\sigma_5}, \hat{r}, \hat{\sigma_6}, r' \\ \mathbf{0}\end{pmatrix} \odot \Gamma_1, x : A, \Gamma_2, z : \Box_{s'}C \vdash D : \mathsf{Type}_l$ (T\_BoxE premise for $D$)
3. $|\hat{\sigma_5}| = |\mathbf{0}| = |\Gamma_1|$ (trivially, and by premise (3))

And:

1. $(\Delta \mid \sigma_2 \mid \sigma_1) \odot \Gamma_1 \vdash t : A$ (premise (1))
2. $\begin{pmatrix}\Delta, \sigma_1, \Delta' \\ \hat{\sigma_3}, \hat{s}, \hat{\sigma_4} \\ \hat{\sigma_8}, s'', \hat{\sigma_8}'\end{pmatrix} \odot \Gamma_1, x : A, \Gamma_2 \vdash t_1 : \Box_{s'}C$ (T\_BoxE premise for $t_1$)
3. $|\hat{\sigma_3}| = |\hat{\sigma_8}| = |\Gamma_1|$ (trivially, and by premise (3))

And:

1. $(\Delta \mid \sigma_2 \mid \sigma_1) \odot \Gamma_1 \vdash t : A$ (premise (1))
2. $\begin{pmatrix}\Delta, \sigma_1, \Delta', (\hat{\sigma_8}, s'', \hat{\sigma_8}') \\ \hat{\sigma_3}', \hat{s}', \hat{\sigma_4}', s' \\ \hat{\sigma_5}, \hat{r}, \hat{\sigma_6}, (s' * r')\end{pmatrix} \odot \Gamma_1, x : A, \Gamma_2, y : C \vdash t_2 : [\Box y/z]D$ (T\_BoxE premise for $t_2$)
3. $|\hat{\sigma_3}'| = |\hat{\sigma_5}| = |\Gamma_1|$ (trivially, and by premise (3))

Giving, respectively:

$$\begin{pmatrix}\Delta, (((\Delta', (\hat{\sigma_8}, s'', \hat{\sigma_8}'))\backslash \pi) + ((\Delta', (\hat{\sigma_8}, s'', \hat{\sigma_8}'))/\pi)*\sigma_2) \\ (\hat{\sigma_5} + \hat{r}*\sigma_2), \hat{\sigma_6}, r' \\ \mathbf{0}\end{pmatrix} \odot \Gamma_1, [t/x]\Gamma_2, z : \Box_{s'}[t/x]C \vdash [t/x]D : \mathsf{Type}_l$$

And:

$$\begin{pmatrix}\Delta, (\Delta' \backslash \pi + (\Delta'/\pi)*\sigma_2) \\ (\hat{\sigma_3} + \hat{s}*\sigma_2), \hat{\sigma_4} \\ (\hat{\sigma_8} + s''*\sigma_2), \hat{\sigma_8}'\end{pmatrix} \odot \Gamma_1, [t/x]\Gamma_2 \vdash [t/x]t_1 : \Box_{s'}[t/x]C$$

And:

$$\begin{pmatrix}\Delta, (((\Delta', (\hat{\sigma_8}, s'', \hat{\sigma_8}'))\backslash \pi) + ((\Delta', (\hat{\sigma_8}, s'', \hat{\sigma_8}'))/\pi)*\sigma_2) \\ (\hat{\sigma_3}' + \hat{s}'*\sigma_2), \hat{\sigma_4}', s' \\ (\hat{\sigma_5} + \hat{r}*\sigma_2), \hat{\sigma_6}, (s' * r')\end{pmatrix} \odot \Gamma_1, [t/x]\Gamma_2, y : [t/x]C \vdash [t/x]t_2 : [t/x]D$$

Rewriting $((\Delta', (\hat{\sigma_8}, s'', \hat{\sigma_8}'))\backslash\pi) + ((\Delta', (\hat{\sigma_8}, s'', \hat{\sigma_8}'))/\pi) * \sigma_2$ by Lemma 3.17, we have $(\Delta'\backslash\pi + (\Delta'/\pi) * \sigma_2), ((\hat{\sigma_8} + s'' * \sigma_2), \hat{\sigma_8}')$.

We can thus apply our inductive hypotheses via T_BoxE, giving:

$$\begin{pmatrix} \Delta, ((\Delta'\backslash\pi + (\Delta'/\pi)*\sigma_2), ((\hat{\sigma_8}+s''*\sigma_2),\hat{\sigma_8}')) \\ (\hat{\sigma_5}+\hat{r}*\sigma_2), \hat{\sigma_6}, r' \\ \mathbf{0} \end{pmatrix} \odot \Gamma_1, [t/x]\Gamma_2, z : \Box_{s'}[t/x]C \vdash [t/x]D : \mathsf{Type}_l$$

$$\begin{pmatrix} \Delta, (\Delta'\backslash\pi + (\Delta'/\pi)*\sigma_2) \\ (\hat{\sigma_3}+\hat{s}*\sigma_2), \hat{\sigma_4} \\ (\hat{\sigma_8}+s''*\sigma_2), \hat{\sigma_8}' \end{pmatrix} \odot \Gamma_1, [t/x]\Gamma_2 \vdash [t/x]t_1 : \Box_{s'}[t/x]C$$

$$\begin{pmatrix} \Delta, ((\Delta'\backslash\pi + (\Delta'/\pi)*\sigma_2), ((\hat{\sigma_8}+s''*\sigma_2),\hat{\sigma_8}')) \\ (\hat{\sigma_3}'+\hat{s}'*\sigma_2), \hat{\sigma_4}', s' \\ (\hat{\sigma_5}+\hat{r}*\sigma_2), \hat{\sigma_6}, (s'*r') \end{pmatrix} \odot \Gamma_1, [t/x]\Gamma_2, y : [t/x]C \vdash [t/x]t_2 : [t/x]D$$

$$\overline{\begin{pmatrix} \Delta, (\Delta'\backslash\pi + (\Delta'/\pi)*\sigma_2) \\ (\hat{\sigma_3}+\hat{s}*\sigma_2), \hat{\sigma_4}+(\hat{\sigma_3}'+\hat{s}'*\sigma_2), \hat{\sigma_4}' \\ (\hat{\sigma_5}+\hat{r}*\sigma_2), \hat{\sigma_6}+r'*(\hat{\sigma_3}+\hat{s}*\sigma_2), \hat{\sigma_4} \end{pmatrix} \odot \Gamma_1, x : A, \Gamma_2 \vdash \mathsf{let}\,\Box y = [t/x]t_1 \mathsf{\ in\ } [t/x]t_2 : [t_2/z]([t/x]D)}$$ T_BoxE

Rewriting the subject and subject type grades of this conclusion by Lemma 3.18, we have $(\hat{\sigma_3}+\hat{s}*\sigma_2), \hat{\sigma_4}+(\hat{\sigma_3}'+\hat{s}'*\sigma_2), \hat{\sigma_4}' = (\sigma_3+s*\sigma_2), \sigma_4$ and $(\hat{\sigma_5}+\hat{r}*\sigma_2), \hat{\sigma_6}+r'*(\hat{\sigma_3}+\hat{s}*\sigma_2), \hat{\sigma_4} = (\sigma_5+r*\sigma_2), \sigma_6$. By definition we have $(\mathsf{let}\,\Box y = [t/x]t_1 \mathsf{\ in\ } [t/x]t_2) = [t/x](\mathsf{let}\,\Box y = t_1 \mathsf{\ in\ } t_2)$ (as $x \neq y$ by disjointness of concatenation), and similarly, we have $[t_2/z][t/x]D = [t/x]([t_2/z]D)$. Using this information to rewrite the conclusion of T_BoxE derived above, we have:

$$\begin{pmatrix} \Delta, (\Delta'\backslash\pi + (\Delta'/\pi)*\sigma_2) \\ (\sigma_3+s*\sigma_2), \sigma_4 \\ (\sigma_5+r*\sigma_2), \sigma_6 \end{pmatrix} \odot \Gamma_1, [t/x]\Gamma_2 \vdash [t/x](\mathsf{let}\,\Box y = t_1 \mathsf{\ in\ } t_2) : [t/x]([t_2/z]D)$$

Which matches the goal.

Case.

$$\frac{(\Delta, \sigma_1, \Delta' \mid \sigma_3, s, \sigma_4 \mid \sigma_5, r, \sigma_6) \odot \Gamma_1, x : A, \Gamma_2 \vdash t' : C \quad (\Delta, \sigma_1, \Delta' \mid \sigma_5, r, \sigma_6 \mid \mathbf{0}) \odot \Gamma_1, x : A, \Gamma_2 \vdash C = B : \mathsf{Type}_l}{(\Delta, \sigma_1, \Delta' \mid \sigma_3, s, \sigma_4 \mid \sigma_5, r, \sigma_6) \odot \Gamma_1, x : A, \Gamma_2 \vdash t' : B}$$ T_TyConv

Our full goal is therefore:

$$\begin{pmatrix} \Delta, (\Delta'\backslash\pi + (\Delta'/\pi)*\sigma_2) \\ (\sigma_3+s*\sigma_2), \sigma_4 \\ (\sigma_5+r*\sigma_2), \sigma_6 \end{pmatrix} \odot \Gamma_1, [t/x]\Gamma_2 \vdash [t/x]t' : [t/x]B$$

Rewriting the first and second premises using the above information, we have the following inductive premises:

1. $(\Delta \mid \sigma_2 \mid \sigma_1) \odot \Gamma_1 \vdash t : A$ (premise (1))
2. $(\Delta, \sigma_1, \Delta' \mid \sigma_3, s, \sigma_4 \mid \sigma_5, r, \sigma_6) \odot \Gamma_1, x : A, \Gamma_2 \vdash t' : C$ (T_TyConv premise for $t'$)
3. $|\sigma_3| = |\sigma_5| = |\Gamma_1|$ (premise (3))

And:

1. $(\Delta \mid \sigma_2 \mid \sigma_1) \odot \Gamma_1 \vdash t : A$ (premise (1))
2. $\begin{pmatrix} \Delta, \sigma_1, \Delta' \\ \sigma_5, r, \sigma_6 \\ \mathbf{0} \end{pmatrix} \odot \Gamma_1, x : A, \Gamma_2 \vdash C = B : \mathsf{Type}_l$ (T_TyConv premise for $C = B$)
3. $|\sigma_5| = |\mathbf{0}| = |\Gamma_1|$ (trivially, and by premise (3))

> Giving, respectively:
> $$\begin{pmatrix} \Delta,(\Delta'\backslash\pi+(\Delta'/\pi)*\sigma_2) \\ (\sigma_3+s*\sigma_2),\sigma_4 \\ (\sigma_5+r*\sigma_2),\sigma_6 \end{pmatrix} \odot \Gamma_1, [t/x]\Gamma_2 \vdash [t/x]t' : [t/x]C$$
>
> And:
> $$\begin{pmatrix} \Delta,(\Delta'\backslash\pi+(\Delta'/\pi)*\sigma_2) \\ (\sigma_5+r*\sigma_2),\sigma_6 \\ \mathbf{0} \end{pmatrix} \odot \Gamma_1, [t/x]\Gamma_2 \vdash [t/x]C = [t/x]B : \mathsf{Type}_l$$
>
> Which we can apply to T_TYCONV, giving:
>
> $$\cfrac{\begin{pmatrix} \Delta,(\Delta'\backslash\pi+(\Delta'/\pi)*\sigma_2) \\ (\sigma_3+s*\sigma_2),\sigma_4 \\ (\sigma_5+r*\sigma_2),\sigma_6 \end{pmatrix} \odot \Gamma_1, [t/x]\Gamma_2 \vdash [t/x]t' : [t/x]C \quad \begin{pmatrix} \Delta,(\Delta'\backslash\pi+(\Delta'/\pi)*\sigma_2) \\ (\sigma_5+r*\sigma_2),\sigma_6 \\ \mathbf{0} \end{pmatrix} \odot \Gamma_1, [t/x]\Gamma_2 \vdash [t/x]C = [t/x]B : \mathsf{Type}_l}{\begin{pmatrix} \Delta,(\Delta'\backslash\pi+(\Delta'/\pi)*\sigma_2) \\ (\sigma_3+s*\sigma_2),\sigma_4 \\ (\sigma_5+r*\sigma_2),\sigma_6 \end{pmatrix} \odot \Gamma_1, [t/x]\Gamma_2 \vdash [t/x]t' : [t/x]B} \text{ T\_TyConv}$$
>
> Which matches the goal.

For equality, we proceed by induction and re-application on the structure of $\left(\begin{smallmatrix} \Delta,\sigma_1,\Delta' \\ \sigma_3,s,\sigma_4 \\ \sigma_5,r,\sigma_6 \end{smallmatrix}\right) \odot \Gamma_1, x : A, \Gamma_2 \vdash t' = t'' : B$ (in some cases substitutions need to be rewritten, to see how to rewrite substitutions, refer to the typing cases). See Section A.8. For subtyping, by standard induction and re-application (see Section A.8). □

**Lemma 3.28** (Equality through substitution). *If the following premises hold:*

1. $(\Delta_1 \mid \sigma_1 \mid \sigma_2) \odot \Gamma_1 \vdash t_1 = t_2 : A$
2. $(\Delta_1, \sigma_2, \Delta_2 \mid \sigma_3, s, \sigma_4 \mid \sigma_5, r, \sigma_6) \odot \Gamma_1, x : A, \Gamma_2 \vdash t_3 : B$
3. $|\sigma_3| = |\sigma_5| = |\Gamma_1|$

*Then* $\begin{pmatrix} \Delta_1,(\Delta_2\backslash\pi+(\Delta_2/\pi)*\sigma_1) \\ (\sigma_3+s*\sigma_1),\sigma_4 \\ (\sigma_5+r*\sigma_1),\sigma_6 \end{pmatrix} \odot \Gamma_1, [t_1/x]\Gamma_2 \vdash [t_1/x]t_3 = [t_2/x]t_3 : [t_1/x]B$, *where* $\pi = |\Gamma_1|$.

*Proof.* Cases proceed very similarly to the proof for Lemma 3.27, inducting on the form of $(\Delta_1, \sigma_2, \Delta_2 \mid \sigma_3, s, \sigma_4 \mid \sigma_5, r, \sigma_6) \odot \Gamma_1, x : A, \Gamma_2 \vdash t_3 : B$, and building results using appropriate equality rules. As such, we omit most cases, but provide a couple for demonstration, as follows:

> Case.
>
> $$\cfrac{(\Delta_1, \sigma_2, \Delta_2 \mid \sigma_3, s, \sigma_4 \mid \mathbf{0}) \odot \Gamma, x : A \vdash C : \mathsf{Type}_l}{(\Delta_1, \sigma_2, \Delta_2 \mid \sigma_3, s, \sigma_4 \mid \mathbf{0}) \odot \Gamma, x : A \vdash \Box_{s'}C : \mathsf{Type}_l} \text{ T\_Box}$$
>
> Then our goal holds by the following derivation:
>
> $$\cfrac{\cfrac{\cfrac{(\Delta, \sigma_2, \Delta_2 \mid \sigma_3, s, \sigma_4 \mid \mathbf{0}) \odot \Gamma, x : A \vdash C : \mathsf{Type}_l}{(\Delta_1,(\Delta_2\backslash\pi+(\Delta_2/\pi)*\sigma_1) \mid (\sigma_3+s*\sigma_1),\sigma_4 \mid \mathbf{0}) \odot \Gamma_1, [t_1/x]\Gamma_2 \vdash [t_1/x]C = [t_2/x]C : \mathsf{Type}_l} \text{ I.H.}}{(\Delta_1,(\Delta_2\backslash\pi+(\Delta_2/\pi)*\sigma_1) \mid (\sigma_3+s*\sigma_1),\sigma_4 \mid \mathbf{0}) \odot \Gamma_1, [t_1/x]\Gamma_2 \vdash \Box_{s'}[t_1/x]C = \Box_{s'}[t_2/x]C : \mathsf{Type}_l} \text{ TEQ\_Box}}{(\Delta_1,(\Delta_2\backslash\pi+(\Delta_2/\pi)*\sigma_1) \mid (\sigma_3+s*\sigma_1),\sigma_4 \mid \mathbf{0}) \odot \Gamma_1, [t_1/x]\Gamma_2 \vdash [t_1/x](\Box_{s'}C) = [t_2/x](\Box_{s'}C) : \mathsf{Type}_l} \equiv$$

Case.

$$\dfrac{\begin{array}{c}(\Delta_1,\sigma_2,\Delta_2\mid\hat{\sigma_3},\hat{s},\hat{\sigma_4}\mid\hat{\sigma_7},r_1,\hat{\sigma_7}')\odot\Gamma_1,x:A,\Gamma_2\vdash t_4:\Box_{s'}C\\(\Delta_1,\sigma_2,\Delta_2,(\hat{\sigma_7},r_1,\hat{\sigma_7}')\mid\hat{\sigma_5},\hat{r},\hat{\sigma_6},r\mid\mathbf{0})\odot\Gamma_1,x:A,\Gamma_2,z:\Box_{s'}C\vdash D:\mathsf{Type}_l\\(\Delta_1,\sigma_2,\Delta_2,(\hat{\sigma_7},r_1,\hat{\sigma_7}')\mid\hat{\sigma_3}',\hat{s}',\hat{\sigma_4}',s'\mid\hat{\sigma_5},\hat{r},\hat{\sigma_6},(s'*r'))\odot\Gamma_1,x:A,\Gamma_2,y:C\vdash t_5:[\Box y/z]D\end{array}}{(\Delta_1,\sigma_2,\Delta_2\mid\hat{\sigma_3},\hat{s},\hat{\sigma_4}+\hat{\sigma_3}',\hat{s}',\hat{\sigma_4}'\mid\hat{\sigma_5},\hat{r},\hat{\sigma_6}+r'*\hat{\sigma_3},\hat{s},\hat{\sigma_4})\odot\Gamma_1,x:A,\Gamma_2\vdash\mathsf{let}\,\Box y=t_4\,\mathsf{in}\,t_5:[t_4/z]D}\;\text{T\_BoxE}$$

By Lemma 3.32 we have $(\Delta_1\mid\sigma_1\mid\sigma_2)\odot\Gamma_1\vdash t_1:A$, and therefore by induction on the premises of T_BoxE for $t_4$ and $t_5$, and by Lemma 3.27 on the premise for $D$, we have the following derivation:

$$\dfrac{\begin{array}{c}\left(\begin{smallmatrix}\Delta_1,(\Delta_2\backslash\pi+(\Delta_2/\pi)*\sigma_1),((\hat{\sigma_7}+r_1*\sigma_1),\hat{\sigma_7}')\\(\hat{\sigma_5}+\hat{r}*\sigma_1),\hat{\sigma_6},r\\\mathbf{0}\end{smallmatrix}\right)\odot\Gamma_1,[t_1/x]\Gamma_2,z:\Box_{s'}[t_1/x]C\vdash[t_1/x]D:\mathsf{Type}_l\\\left(\begin{smallmatrix}\Delta_1\\\hat{\sigma_3}+\hat{s}*\sigma_1\\\hat{\sigma_5}+r_1*\sigma_1\end{smallmatrix}\right)\odot\Gamma\vdash[t_1/x]t_3=[t_2/x]t_3:\Box_{s'}[t_1/x]C\\\left(\begin{smallmatrix}\Delta_1,(\Delta_2\backslash\pi+(\Delta_2/\pi)*\sigma_1),(\hat{\sigma_7}+r_1*\hat{\sigma_7}')\\(\hat{\sigma_3}'+\hat{s}'*\hat{\sigma_4}'),s'\\(\hat{\sigma_4}+\hat{r}*\sigma_1),(s'*r')\end{smallmatrix}\right)\odot\Gamma_1,[t_1/x]\Gamma_2,y:C\vdash[t_1/x]t_4=[t_2/x]t_4:[\Box y/z]([t_1/x]D)\end{array}}{\left(\begin{smallmatrix}\Delta_1,(\Delta_2\backslash\pi+(\Delta_2/\pi)*\sigma_1)\\(\sigma_3+s*\sigma_1),\sigma_4\\(\sigma_5+r*\sigma_1),\sigma_6\end{smallmatrix}\right)\odot\Gamma_1,[t_1/x]\Gamma_2\vdash(\mathsf{let}\,\Box y=[t_1/x]t_4\,\mathsf{in}\,[t_1/x]t_5)=(\mathsf{let}\,\Box y=[t_2/x]t_4\,\mathsf{in}\,[t_2/x]t_5):[t_4/z]([t_1/x]D)}\;\text{TEQ\_BoxE}$$

$$\overline{\left(\begin{smallmatrix}\Delta_1,(\Delta_2\backslash\pi+(\Delta_2/\pi)*\sigma_1)\\(\sigma_3+s*\sigma_1),\sigma_4\\(\sigma_5+r*\sigma_1),\sigma_6\end{smallmatrix}\right)\odot\Gamma_1,[t_1/x]\Gamma_2\vdash[t_1/x](\mathsf{let}\,\Box y=t_4\,\mathsf{in}\,t_5)=[t_2/x](\mathsf{let}\,\Box y=t_4\,\mathsf{in}\,t_5):[t_1/x]([t_4/z]D)}\;\equiv$$

Giving us our goal.

$\square$

### A.3 Proofs for vector manipulation

**Lemma 3.17** (Factoring choose and discard). *If $|\sigma_1|=|\sigma_3|=\pi$, then $(\Delta_1,(\sigma_1,r,\sigma_2),\Delta_2)\backslash\pi+((\Delta_1,(\sigma_1,r,\sigma_2),\Delta_2)/\pi)*\sigma_3=(\Delta_1\backslash\pi+\Delta_1/\pi*\sigma_3),((\sigma_1+r*\sigma_3),\sigma_2),(\Delta_2\backslash\pi+\Delta_2/\pi*\sigma_3)$.*

*Proof.* We show this by rewriting equationally, as follows:

$$\begin{aligned}&(\Delta_1,(\sigma_1,s,\sigma_2),\Delta_2)\backslash\pi+(\Delta_1,(\sigma_1,s,\sigma_2),\Delta_2)/\pi*\sigma_3\\\{\text{defn. 2.10 and defn. 2.8}\}\quad&=(\Delta_1\backslash\pi),((\sigma_1,r,\sigma_2)\backslash\pi),(\Delta_2\backslash\pi)+((\Delta_1/\pi),((\sigma_1,r,\sigma_2)/\pi),(\Delta_2/\pi))*\sigma_3\\\{\text{defn. 2.9 and defn. 2.7, using }|\sigma_1|=\pi\}\quad&=(\Delta_1\backslash\pi),(\sigma_1,\sigma_2),(\Delta_2\backslash\pi)+((\Delta_1/\pi),r,(\Delta_2/\pi))*\sigma_3\\\{\text{defn. 2.11}\}\quad&=(\Delta_1\backslash\pi),(\sigma_1,\sigma_2),(\Delta_2\backslash\pi)+(\Delta_1/\pi*\sigma_3),(r*\sigma_3),(\Delta_2/\pi*\sigma_3)\\\{\text{defn. 2.13}\}\quad&=(\Delta_1\backslash\pi+\Delta_1/\pi*\sigma_3),(\sigma_1,\sigma_2+r*\sigma_3),(\Delta_2\backslash\pi+\Delta_2/\pi*\sigma_3)\\\{\text{L. 3.19 with }|\sigma_1|=|\sigma_3|\}\quad&=(\Delta_1\backslash\pi+\Delta_1/\pi*\sigma_3),((\sigma_1+r*\sigma_3),\sigma_2),(\Delta_2\backslash\pi+\Delta_2/\pi*\sigma_3)\end{aligned}$$

As required. $\square$

**Lemma 3.18** (Factoring vector addition). *If $\hat{\sigma_1}+\hat{\sigma_1}'=\sigma_1$; $\hat{s}+\hat{s}'=s$; and $\hat{\sigma_3}+\hat{\sigma_3}'=\sigma_3$, then $((\hat{\sigma_1}+\hat{s}*\sigma_2),\hat{\sigma_3})+((\hat{\sigma_1}'+\hat{s}'*\sigma_2),\hat{\sigma_3}')=(\sigma_1+s*\sigma_2),\sigma_3$ for all $\sigma_2$.*

*Proof.* This holds equationally, by laws of the semiring, and vector operations, as follows:

$$\begin{aligned}&((\hat{\sigma_1}+\hat{s}*\sigma_2),\hat{\sigma_3})+((\hat{\sigma_1}'+\hat{s}'*\sigma_2),\hat{\sigma_3}')\\&=((\hat{\sigma_1}+\hat{s}*\sigma_2)+(\hat{\sigma_1}'+\hat{s}'*\sigma_2)),(\hat{\sigma_3}+\hat{\sigma_3}')\\&=((\hat{\sigma_1}+\hat{\sigma_1}')+(\hat{s}+\hat{s}')*\sigma_2),\sigma_3\\&=(\sigma_1+s*\sigma_2),\sigma_3\end{aligned}$$

**Lemma 3.19** (Addition across same-sized components). *If $|\sigma_1| = |\sigma_3|$ then $\sigma_1, \sigma_2 + \sigma_3 = (\sigma_1 + \sigma_3), \sigma_2$.*

*Proof.* Trivially, by the implicit padding of the right-vector by zeros. □

**Lemma 3.20** (Vector addition across components). *If $|\sigma_1| = |\sigma_4|$ and $|\sigma_2| = |\sigma_5|$, then $\sigma_1, \sigma_2, \sigma_3 + \sigma_4, \sigma_5, \sigma_6 = (\sigma_1 + \sigma_4), (\sigma_2 + \sigma_5), (\sigma_3 + \sigma_6)$.*

*Proof.* This holds trivially, by definition of addition on vectors. □

**Lemma 3.21** (Moving an exchange). *We have $\mathsf{mv}((\pi_1 + \pi_2); \pi_3; (\mathsf{exch}((\pi_1 + \pi_2); \Delta))) = \mathsf{mv}((\pi_1 + \pi_2 + 1); \pi_3; \Delta)$.*

*Proof.* By definition, we have $\mathsf{exch}((\pi_1 + \pi_2); \Delta) = \mathsf{mv}((\pi_1 + \pi_2 + 1); (\pi_1 + \pi_2); \Delta)$. Therefore, the index $\pi_1 + \pi_2 + 1$ now corresponds to index $\pi_1 + \pi_2$ in $\mathsf{exch}((\pi_1 + \pi_2); \Delta)$, and thus $\mathsf{mv}((\pi_1 + \pi_2 + 1); \pi_3; \Delta)$ and $\mathsf{mv}((\pi_1 + \pi_2); \pi_3; (\mathsf{exch}((\pi_1 + \pi_2); \Delta)))$ are the same. □

## A.4 Proofs for properties of operations

**Lemma 3.29** (Properties of insertion). *The following properties hold for any valid insertion (i.e., where for all $i < |\Delta|$, $\pi \leq |\Delta[i]|$, for context grade vectors; and $\pi \leq |\sigma|$, for grade vectors):*

1. *(insPreservesSize)* $|\mathsf{ins}(\pi; R; \Delta)| = |\Delta|$;

2. *(insIncSizes)* if $\pi \leq |\Delta[i]|$, then $|(\mathsf{ins}(\pi; R; \Delta))[i]| = |\Delta[i]| + 1$;

3. *(insIncSizeGV)* $|\mathsf{ins}(\pi; R; \sigma)| = |\sigma| + 1$;

4. *(insCVthenGV)* $\mathsf{ins}(\pi; R; \Delta), \mathsf{ins}(\pi; R; \sigma) = \mathsf{ins}(\pi; R; (\Delta, \sigma))$

*Proof.* Trivial by definition. □

## A.5 Proofs for equality and conversion

**Lemma 3.30** (Equality is an equivalence relation). *For all, we have:*

- *(reflexivity) if $(\Delta \mid \sigma_1 \mid \sigma_2) \odot \Gamma \vdash t : A$ then $(\Delta \mid \sigma_1 \mid \sigma_2) \odot \Gamma \vdash t = t : A$;*

- *(transitivity) if $(\Delta \mid \sigma_1 \mid \sigma_2) \odot \Gamma \vdash t_1 = t_2 : A$ and $(\Delta \mid \sigma_1 \mid \sigma_2) \odot \Gamma \vdash t_2 = t_3 : A$, then $(\Delta \mid \sigma_1 \mid \sigma_2) \odot \Gamma \vdash t_1 = t_3 : A$;*

- *(symmetry) if $(\Delta \mid \sigma_1 \mid \sigma_2) \odot \Gamma \vdash t_1 = t_2 : A$, then $(\Delta \mid \sigma_1 \mid \sigma_2) \odot \Gamma \vdash t_2 = t_1 : A$*

*Proof.* Reflexivity holds by TEQ_REFL, transitivity holds by TEQ_TRANS, and symmetry holds by TEQ_SYM. □

**Lemma 3.31** (Reduction implies equality). *If $(\Delta \mid \sigma_1 \mid \sigma_2) \odot \Gamma \vdash t_1 : A$ and $t_1 \leadsto t_2$, then $(\Delta \mid \sigma_1 \mid \sigma_2) \odot \Gamma \vdash t_1 = t_2 : A$.*

*Proof.* By induction on the form of $(\Delta \mid \sigma_1 \mid \sigma_2) \odot \Gamma \vdash t_1 : A$, as follows:

> Case.
>
> $$\frac{\Delta \odot \Gamma \vdash}{(\Delta \mid \mathbf{0} \mid \mathbf{0}) \odot \Gamma \vdash \mathsf{Type}_l : \mathsf{Type}_{\mathsf{suc}\ l}} \text{T\_TYPE}$$
>
> No reductions possible, we are done.

Case.

$$\frac{(\Delta \mid \sigma_1 \mid \mathbf{0}) \odot \Gamma \vdash A : \mathsf{Type}_{l_1} \quad (\Delta, \sigma_1 \mid \sigma_2, r \mid \mathbf{0}) \odot \Gamma, x : A \vdash B : \mathsf{Type}_{l_2}}{(\Delta \mid \sigma_1 + \sigma_2 \mid \mathbf{0}) \odot \Gamma \vdash (x :_{(s,r)} A) \to B : \mathsf{Type}_{l_1 \sqcup l_2}} \text{ T\_Arrow}$$

Case.

$$\frac{A \rightsquigarrow A'}{(x :_{(s,r)} A) \to B \rightsquigarrow (x :_{(s,r)} A') \to B} \text{ Sem\_CongArrow1}$$

By induction we have $(\Delta \mid \sigma_1 \mid \mathbf{0}) \odot \Gamma \vdash A = A' : \mathsf{Type}_{l_1}$. Therefore, we can form the following derivation:

$$\frac{(\Delta \mid \sigma_1 \mid \mathbf{0}) \odot \Gamma \vdash A = A' : \mathsf{Type}_{l_1} \quad \frac{(\Delta, \sigma_1 \mid \sigma_2, r \mid \mathbf{0}) \odot \Gamma, x : A \vdash B : \mathsf{Type}_{l_2}}{(\Delta, \sigma_1 \mid \sigma_2, r \mid \mathbf{0}) \odot \Gamma, x : A \vdash B = B : \mathsf{Type}_{l_2}} \text{ TEQ\_Refl}}{(\Delta \mid \sigma_1 + \sigma_2 \mid \mathbf{0}) \odot \Gamma \vdash (x :_{(s,r)} A) \to B = (x :_{(s,r)} A') \to B : \mathsf{Type}_{l_1 \sqcup l_2}} \text{ TEQ\_Arrow}$$

Thus obtaining our goal.

Case.

$$\frac{B \rightsquigarrow B'}{(x :_{(s,r)} A) \to B \rightsquigarrow (x :_{(s,r)} A) \to B'} \text{ Sem\_CongArrow2}$$

By induction we have $(\Delta, \sigma_1 \mid \sigma_2, r \mid \mathbf{0}) \odot \Gamma, x : A \vdash B = B' : \mathsf{Type}_{l_2}$. Therefore, we can form the following derivation:

$$\frac{\frac{(\Delta \mid \sigma_1 \mid \mathbf{0}) \odot \Gamma \vdash A : \mathsf{Type}_{l_1}}{(\Delta \mid \sigma_1 \mid \mathbf{0}) \odot \Gamma \vdash A = A : \mathsf{Type}_{l_1}} \text{ TEQ\_Refl} \quad (\Delta, \sigma_1 \mid \sigma_2, r \mid \mathbf{0}) \odot \Gamma \vdash B = B' : \mathsf{Type}_{l_2}}{(\Delta \mid \sigma_1 + \sigma_2 \mid \mathbf{0}) \odot \Gamma \vdash (x :_{(s,r)} A) \to B = (x :_{(s,r)} A) \to B' : \mathsf{Type}_{l_1 \sqcup l_2}} \text{ TEQ\_Arrow}$$

Thus obtaining our goal.

The case for T_Ten proceeds similarly, using Sem_CongTen1, Sem_CongTen2, and TEQ_Ten.

Case.

$$\frac{\begin{array}{c}(\Delta, \sigma_1 \mid \sigma_3, r \mid \mathbf{0}) \odot \Gamma, x : A \vdash B : \mathsf{Type}_l \\ (\Delta \mid \sigma_2 \mid \sigma_1) \odot \Gamma \vdash t_1 : A \quad (\Delta \mid \sigma_4 \mid \sigma_3 + r * \sigma_2) \odot \Gamma \vdash t_2 : [t_1/x]B\end{array}}{(\Delta \mid \sigma_2 + \sigma_4 \mid \sigma_1 + \sigma_3) \odot \Gamma \vdash (t_1, t_2) : (x :_r A) \otimes B} \text{ T\_Pair}$$

Proceeds similarly to the case for T_Arrow, but using Sem_CongPair, Sem_CongPairTwo, and TEQ_Pair.

Case.

$$\dfrac{\Delta_1, \sigma, \Delta_2 \odot \Gamma_1, x : A, \Gamma_2 \vdash \quad |\Delta_1| = |\Gamma_1|}{(\Delta_1, \sigma, \Delta_2 \mid \mathbf{0}^{|\Delta_1|}, 1, \mathbf{0} \mid \sigma, 0, \mathbf{0}) \odot \Gamma_1, x : A, \Gamma_2 \vdash x : A} \text{ T\_Var}$$

No reductions possible, we are done.

Case.

$$\dfrac{\begin{array}{c}(\Delta, \sigma_1 \mid \sigma_3, r \mid \mathbf{0}) \odot \Gamma, x : A \vdash B : \mathsf{Type}_l \\ (\Delta \mid \sigma_2 \mid \sigma_1 + \sigma_3) \odot \Gamma \vdash \lambda x.t_1 : (x :_{(s,r)} A) \to B \quad (\Delta \mid \sigma_4 \mid \sigma_1) \odot \Gamma \vdash t_2 : A \end{array}}{(\Delta \mid \sigma_2 + s * \sigma_4 \mid \sigma_3 + r * \sigma_4) \odot \Gamma \vdash (\lambda x.t_1)\, t_2 : [t_2/x]B} \text{ T\_App}$$

With:

$$\dfrac{}{(\lambda x.t_1)\, t_2 \leadsto [t_2/x]t_1} \text{ Sem\_BetaFun}$$

By Lemma 3.11 we have $(\Delta, \sigma_1 \mid \sigma_2, s \mid \sigma_3, r) \odot \Gamma, x : A \vdash t_1 : B$. Then we have the following derivation:

$$\dfrac{(\Delta, \sigma_1 \mid \sigma_2, s \mid \sigma_3, r) \odot \Gamma, x : A \vdash t_1 : B \quad (\Delta \mid \sigma_4 \mid \sigma_1) \odot \Gamma \vdash t_2 : A}{(\Delta \mid \sigma_2 + s * \sigma_4 \mid \sigma_3 + r * \sigma_4) \odot \Gamma \vdash (\lambda x.t_1)\, t_2 = [t_2/x]t_1 : [t_2/x]B} \text{ TEQ\_ArrowComp}$$

Giving us our goal.

Case.

$$\dfrac{\begin{array}{c}(\Delta, \sigma_1 \mid \sigma_3, r \mid \mathbf{0}) \odot \Gamma, x : A \vdash B : \mathsf{Type}_l \\ (\Delta \mid \sigma_2 \mid \sigma_1 + \sigma_3) \odot \Gamma \vdash t_1 : (x :_{(s,r)} A) \to B \quad (\Delta \mid \sigma_4 \mid \sigma_1) \odot \Gamma \vdash t_2 : A \end{array}}{(\Delta \mid \sigma_2 + s * \sigma_4 \mid \sigma_3 + r * \sigma_4) \odot \Gamma \vdash t_1\, t_2 : [t_2/x]B} \text{ T\_App}$$

With:

$$\dfrac{t_1 \leadsto t_1'}{t_1\, t_2 \leadsto t_1'\, t_2} \text{ Sem\_CongFunOne}$$

By induction we have $(\Delta \mid \sigma_2 \mid \sigma_1 + \sigma_3) \odot \Gamma \vdash t_1 = t_1' : (x :_{(s,r)} A) \to B$. Therefore, we have the following derivation:

$$\dfrac{\begin{array}{c}(\Delta, \sigma_1 \mid \sigma_3, r \mid \mathbf{0}) \odot \Gamma, x : A \vdash B : \mathsf{Type}_l \\ (\Delta \mid \sigma_2 \mid \sigma_1 + \sigma_3) \odot \Gamma \vdash t_1 = t_1' : (x :_{(s,r)} A) \to B \quad \dfrac{(\Delta \mid \sigma_4 \mid \sigma_1) \odot \Gamma \vdash t_2 : A}{(\Delta \mid \sigma_4 \mid \sigma_1) \odot \Gamma \vdash t_2 = t_2 : A} \text{ TEQ\_Refl} \end{array}}{(\Delta \mid \sigma_2 + s * \sigma_4 \mid \sigma_3 + r * \sigma_4) \odot \Gamma \vdash t_1\, t_2 = t_1'\, t_2 : [t_2/x]B} \text{ TEQ\_App}$$

Giving us our goal.

Case.

$$(\Delta \mid \sigma_3 \mid \sigma_1 + \sigma_2) \odot \Gamma \vdash (t_1, t_2) : (x :_r A) \otimes B$$
$$(\Delta, (\sigma_1 + \sigma_2) \mid \sigma_5, r' \mid \mathbf{0}) \odot \Gamma, z : (x :_r A) \otimes B \vdash C : \mathsf{Type}_l$$
$$\dfrac{(\Delta, \sigma_1, (\sigma_2, r) \mid \sigma_4, s, s \mid \sigma_5, r', r') \odot \Gamma, x : A, y : B \vdash t_3 : [(x,y)/z]C}{(\Delta \mid \sigma_4 + s * \sigma_3 \mid \sigma_5 + r' * \sigma_3) \odot \Gamma \vdash \mathsf{let}\,(x,y) = (t_1, t_2)\,\mathsf{in}\,t_3 : [(t_1, t_2)/z]C} \text{ T\_TenCut}$$

With:

$$\dfrac{}{\mathsf{let}\,(x,y) = (t_1, t_2)\,\mathsf{in}\,t_3 \rightsquigarrow [t_1/x][t_2/y]t_3} \text{ Sem\_BetaTen}$$

Then we obtain our goal by the following derivation:

$$\dfrac{\begin{array}{c}\mathcal{A}\\ (\Delta, \sigma_1, (\sigma_2, r) \mid \sigma_4, s, s \mid \sigma_5, r', r') \odot \Gamma, x : A, y : B \vdash t_3 : [(x,y)/z]C\end{array}}{(\Delta, \sigma_1 \mid \sigma_2, r \mid \mathbf{0}) \odot \Gamma, x : A \vdash B : \mathsf{Type}_l} \text{ L. 3.5}$$

By Lemma 3.12 with $\mathcal{A}$ (and premise), we have:

- $(\mathcal{B})\ (\Delta \mid \hat{\sigma_3} \mid \sigma_1) \odot \Gamma \vdash t_1 : A$;
- $(\mathcal{C})\ (\Delta \mid \hat{\sigma_3}' \mid \sigma_2 + r * \hat{\sigma_3}) \odot \Gamma \vdash t_2 : [t_1/x]B$;
- $\hat{\sigma_3} + \hat{\sigma_3}' = \sigma_3$

$$\dfrac{\begin{array}{c}\mathcal{B}\quad \mathcal{C}\\ (\Delta, (\sigma_1 + \sigma_2) \mid \sigma_5, r' \mid \mathbf{0}) \odot \Gamma, z : (x :_r A) \otimes B \vdash C : \mathsf{Type}_l\\ (\Delta, \sigma_1, (\sigma_2, r) \mid \sigma_4, s, s \mid \sigma_5, r', r') \odot \Gamma, x : A, y : B \vdash t_3 : [(x,y)/z]C\end{array}}{(\Delta \mid \sigma_4 + s * \sigma_3 \mid \sigma_5 + r' * \sigma_3) \odot \Gamma \vdash \mathsf{let}\,(x,y) = (t_1, t_2)\,\mathsf{in}\,t_3 = [t_1/x][t_2/y]t_3 : [(t_1, t_2)/z]C} \text{ TEQ\_TenComp}$$

---

Case.

$$(\Delta \mid \sigma_3 \mid \sigma_1 + \sigma_2) \odot \Gamma \vdash t_1 : (x :_r A) \otimes B$$
$$(\Delta, (\sigma_1 + \sigma_2) \mid \sigma_5, r' \mid \mathbf{0}) \odot \Gamma, z : (x :_r A) \otimes B \vdash C : \mathsf{Type}_l$$
$$\dfrac{(\Delta, \sigma_1, (\sigma_2, r) \mid \sigma_4, s, s \mid \sigma_5, r', r') \odot \Gamma, x : A, y : B \vdash t_2 : [(x,y)/z]C}{(\Delta \mid \sigma_4 + s * \sigma_3 \mid \sigma_5 + r' * \sigma_3) \odot \Gamma \vdash \mathsf{let}\,(x,y) = t_1\,\mathsf{in}\,t_2 : [t_1/z]C} \text{ T\_TenCut}$$

Remaining cases proceed similarly to those for T_App, using Sem_CongTenCut1, and TEQ_TenCut.

---

Case.

$$\dfrac{(\Delta \mid \sigma \mid \mathbf{0}) \odot \Gamma \vdash A : \mathsf{Type}_l}{(\Delta \mid \sigma \mid \mathbf{0}) \odot \Gamma \vdash \square_s A : \mathsf{Type}_l} \text{ T\_Box}$$

With:
$$\frac{A \rightsquigarrow A'}{\Box_s A \rightsquigarrow \Box_s A'} \text{ Sem\_CongBox4}$$

Then by induction we have $(\Delta \mid \sigma \mid \mathbf{0}) \odot \Gamma \vdash A = A' : \mathsf{Type}_l$, and obtain our goal by application to TEQ_Box.

Case.

$$\frac{(\Delta, \sigma_2 \mid \sigma_4, r \mid \mathbf{0}) \odot \Gamma, z : \Box_s A \vdash B : \mathsf{Type}_l}{(\Delta \mid \sigma_1 \mid \sigma_2) \odot \Gamma \vdash \Box t_1 : \Box_s A \quad (\Delta, \sigma_2 \mid \sigma_3, s \mid \sigma_4, (s*r)) \odot \Gamma, x : A \vdash t_2 : [\Box x/z]B}{(\Delta \mid \sigma_1 + \sigma_3 \mid \sigma_4 + r*\sigma_1) \odot \Gamma \vdash \mathsf{let}\,\Box x = \Box t_1\,\mathsf{in}\,t_2 : [\Box t_1/z]B}} \text{ T\_BoxE}$$

With:

$$\overline{\mathsf{let}\,\Box x = \Box t_1\,\mathsf{in}\,t_2 \rightsquigarrow [t_1/x]t_2} \text{ Sem\_BetaBox}$$

By Lemma 3.13 we have $(\Delta \mid \hat{\sigma_1} \mid \sigma_2) \odot \Gamma \vdash t_1 : A$ with $s * \hat{\sigma_1} = \sigma_1$, therefore we have the following derivation:

$$\frac{(\Delta \mid \hat{\sigma_1} \mid \sigma_2) \odot \Gamma \vdash t_1 : A \quad \begin{pmatrix} \Delta, \sigma_2 \\ \sigma_3, s \\ \sigma_4, (s*r) \end{pmatrix} \odot \Gamma, x : A \vdash t_2 : [\Box x/z]B \quad (\Delta, \sigma_2 \mid \sigma_4, r \mid \mathbf{0}) \odot \Gamma, z : \Box_s A \vdash B : \mathsf{Type}_l}{(\Delta \mid \sigma_3 + \sigma_1 \mid \sigma_4 + r*\sigma_1) \odot \Gamma \vdash (\mathsf{let}\,\Box x = \Box t_1\,\mathsf{in}\,t_2) = [t_1/x]t_2 : [\Box t_1/z]B} \text{ TEQ\_BoxB}$$

Giving us our goal.

Case.

$$\frac{\begin{array}{c}(\Delta \mid \sigma_1 \mid \sigma_2) \odot \Gamma \vdash t_1 : \Box_s A \\ (\Delta, \sigma_2 \mid \sigma_4, r \mid \mathbf{0}) \odot \Gamma, z : \Box_s A \vdash B : \mathsf{Type}_l \\ (\Delta, \sigma_2 \mid \sigma_3, s \mid \sigma_4, (s*r)) \odot \Gamma, x : A \vdash t_2 : [\Box x/z]B\end{array}}{(\Delta \mid \sigma_1 + \sigma_3 \mid \sigma_4 + r*\sigma_1) \odot \Gamma \vdash \mathsf{let}\,\Box x = t_1\,\mathsf{in}\,t_2 : [t_1/z]B} \text{ T\_BoxE}$$

Remaining cases proceed similarly to those for T_App, using Sem_CongBox1, and TEQ_BoxE.

Case.

$$\frac{(\Delta \mid \sigma_1 \mid \sigma_2) \odot \Gamma \vdash t : A \quad (\Delta \mid \sigma_2) \odot \Gamma \vdash A \leq B}{(\Delta \mid \sigma_1 \mid \sigma_2) \odot \Gamma \vdash t : B} \text{ T\_TyConv}$$

With $t \rightsquigarrow t'$. By induction, we have $(\Delta \mid \sigma_1 \mid \sigma_2) \odot \Gamma \vdash t = t' : A$, therefore we have the following derivation:

$$\frac{(\Delta \mid \sigma_1 \mid \sigma_2) \odot \Gamma \vdash t = t' : A \quad (\Delta \mid \sigma_2) \odot \Gamma \vdash A \leq B}{(\Delta \mid \sigma_1 \mid \sigma_2) \odot \Gamma \vdash t = t' : B} \text{ TEQ\_ConvTy}$$

Giving us our goal.

$\square$

**Lemma 3.32** (Equality inversion). *If* $(\Delta \mid \sigma_1 \mid \sigma_2) \odot \Gamma \vdash t_1 = t_2 : A$, *then* $(\Delta \mid \sigma_1 \mid \sigma_2) \odot \Gamma \vdash t_1 : A$ *and* $(\Delta \mid \sigma_1 \mid \sigma_2) \odot \Gamma \vdash t_2 : A$.

*Proof.* By induction on the structure of $(\Delta \mid \sigma_1 \mid \sigma_2) \odot \Gamma \vdash t_1 = t_2 : A$, as follows:

Case.
$$\frac{(\Delta \mid \sigma_1 \mid \sigma_2) \odot \Gamma \vdash t : A}{(\Delta \mid \sigma_1 \mid \sigma_2) \odot \Gamma \vdash t = t : A} \text{ TEQ\_REFL}$$

Then our goal holds by premise.

Case.
$$\frac{(\Delta \mid \sigma_1 \mid \sigma_2) \odot \Gamma \vdash t_1 = t_2 : A \quad (\Delta \mid \sigma_1 \mid \sigma_2) \odot \Gamma \vdash t_2 = t_3 : A}{(\Delta \mid \sigma_1 \mid \sigma_2) \odot \Gamma \vdash t_1 = t_3 : A} \text{ TEQ\_TRANS}$$

Then our goal holds by induction.

Case.
$$\frac{(\Delta \mid \sigma_1 \mid \sigma_2) \odot \Gamma \vdash t_1 = t_2 : A}{(\Delta \mid \sigma_1 \mid \sigma_2) \odot \Gamma \vdash t_2 = t_1 : A} \text{ TEQ\_SYM}$$

Then our goal holds by induction.

Case.
$$\frac{(\Delta \mid \sigma_1 \mid \sigma_2) \odot \Gamma \vdash t_1 = t_2 : A \quad (\Delta \mid \sigma_2) \odot \Gamma \vdash A \leq B}{(\Delta \mid \sigma_1 \mid \sigma_2) \odot \Gamma \vdash t_1 = t_2 : B} \text{ TEQ\_CONVTY}$$

Then our goals hold by the following derivations:

$$\frac{\dfrac{(\Delta \mid \sigma_1 \mid \sigma_2) \odot \Gamma \vdash t_1 = t_2 : A}{(\Delta \mid \sigma_1 \mid \sigma_2) \odot \Gamma \vdash t_1 : A} \text{ I.H.} \quad (\Delta \mid \sigma_2) \odot \Gamma \vdash A \leq B}{(\Delta \mid \sigma_1 \mid \sigma_2) \odot \Gamma \vdash t_1 : B} \text{ T\_TYCONV}$$

$$\frac{\dfrac{(\Delta \mid \sigma_1 \mid \sigma_2) \odot \Gamma \vdash t_1 = t_2 : A}{(\Delta \mid \sigma_1 \mid \sigma_2) \odot \Gamma \vdash t_2 : A} \text{ I.H.} \quad (\Delta \mid \sigma_2) \odot \Gamma \vdash A \leq B}{(\Delta \mid \sigma_1 \mid \sigma_2) \odot \Gamma \vdash t_2 : B} \text{ T\_TYCONV}$$

Case.

$$\frac{(\Delta \mid \sigma_1 \mid \mathbf{0}) \odot \Gamma \vdash A = C : \mathsf{Type}_{l_1} \quad (\Delta, \sigma_1 \mid \sigma_2, r \mid \mathbf{0}) \odot \Gamma, x : A \vdash B = D : \mathsf{Type}_{l_2}}{(\Delta \mid \sigma_1 + \sigma_2 \mid \mathbf{0}) \odot \Gamma \vdash (x :_{(s,r)} A) \to B = (x :_{(s,r)} C) \to D : \mathsf{Type}_{l_1 \sqcup l_2}} \text{ TEQ\_Arrow}$$

Then our goals hold by the following derivations:

$$\frac{\dfrac{(\Delta \mid \sigma_1 \mid \mathbf{0}) \odot \Gamma \vdash A = C : \mathsf{Type}_{l_1}}{(\Delta \mid \sigma_1 \mid \mathbf{0}) \odot \Gamma \vdash A : \mathsf{Type}_{l_1}} \text{ I.H.} \quad \dfrac{(\Delta, \sigma_1 \mid \sigma_2, r \mid \mathbf{0}) \odot \Gamma, x : A \vdash B = D : \mathsf{Type}_{l_2}}{(\Delta, \sigma_1 \mid \sigma_2, r \mid \mathbf{0}) \odot \Gamma, x : A \vdash B : \mathsf{Type}_{l_2}} \text{ I.H.}}{(\Delta \mid \sigma_1 + \sigma_2 \mid \mathbf{0}) \odot \Gamma \vdash (x :_{(s,r)} A) \to B : \mathsf{Type}_{l_1 \sqcup l_2}} \text{ T\_Arrow}$$

$$\frac{\dfrac{(\Delta \mid \sigma_1 \mid \mathbf{0}) \odot \Gamma \vdash A = C : \mathsf{Type}_{l_1}}{(\Delta \mid \sigma_1 \mid \mathbf{0}) \odot \Gamma \vdash C : \mathsf{Type}_{l_1}} \text{ I.H.} \quad \dfrac{\dfrac{(\Delta, \sigma_1 \mid \sigma_2, r \mid \mathbf{0}) \odot \Gamma, x : A \vdash B = D : \mathsf{Type}_{l_2}}{(\Delta, \sigma_1 \mid \sigma_2, r \mid \mathbf{0}) \odot \Gamma, x : A \vdash D : \mathsf{Type}_{l_2}} \text{ I.H.}}{(\Delta, \sigma_1 \mid \sigma_2, r \mid \mathbf{0}) \odot \Gamma, x : C \vdash D : \mathsf{Type}_{l_2}} \text{ L. 3.33}}{(\Delta \mid \sigma_1 + \sigma_2 \mid \mathbf{0}) \odot \Gamma \vdash (x :_{(s,r)} C) \to D : \mathsf{Type}_{l_1 \sqcup l_2}} \text{ T\_Arrow}$$

The case for TEQ\_Ten proceeds similarly.

---

Case.

$$\frac{(\Delta, \sigma_1 \mid \sigma_2, s \mid \sigma_3, r) \odot \Gamma, x : A \vdash t_1 : B \quad (\Delta \mid \sigma_4 \mid \sigma_1) \odot \Gamma \vdash t_2 : A}{(\Delta \mid \sigma_2 + s * \sigma_4 \mid \sigma_3 + r * \sigma_4) \odot \Gamma \vdash (\lambda x.t_1)\, t_2 = [t_2/x]t_1 : [t_2/x]B} \text{ TEQ\_ArrowComp}$$

Then our goals hold by the following derivations:

$$\mathcal{A}$$
$$\frac{(\Delta, \sigma_1 \mid \sigma_2, s \mid \sigma_3, r) \odot \Gamma, x : A \vdash t_1 : B}{(\Delta, \sigma_1 \mid \sigma_3, r \mid \mathbf{0}) \odot \Gamma, x : A \vdash B : \mathsf{Type}_l} \text{ L. 3.6}$$

$$\mathcal{C}$$
$$\frac{\mathcal{A} \quad (\Delta, \sigma_1 \mid \sigma_2, s \mid \sigma_3, r) \odot \Gamma, x : A \vdash t_1 : B}{(\Delta \mid \sigma_2 \mid \sigma_1 + \sigma_3) \odot \Gamma \vdash \lambda x.t_1 : (x :_{(s,r)} A) \to B} \text{ T\_Fun}$$

$$\frac{\mathcal{C} \quad \mathcal{A} \quad (\Delta \mid \sigma_4 \mid \sigma_1) \odot \Gamma \vdash t_2 : A}{(\Delta \mid \sigma_2 + s * \sigma_4 \mid \sigma_3 + r * \sigma_4) \odot \Gamma \vdash (\lambda x.t_1)\, t_2 : [t_2/x]B} \text{ T\_App}$$

$$\frac{(\Delta \mid \sigma_4 \mid \sigma_1) \odot \Gamma \vdash t_2 : A \quad (\Delta, \sigma_1 \mid \sigma_2, s \mid \sigma_3, r) \odot \Gamma, x : A \vdash t_1 : B}{(\Delta \mid \sigma_2 + s * \sigma_4 \mid \sigma_3 + r * \sigma_4) \odot \Gamma \vdash [t_2/x]t_1 : [t_2/x]B} \text{ L. 3.27}$$

The case for TEQ\_TenComp proceeds similarly, using double substitution. The case for TEQ\_BoxB proceeds similarly.

Case.

$$\dfrac{(\Delta \mid \sigma_1 \mid \sigma_2) \odot \Gamma \vdash t : (x :_{(s,r)} A) \to B}{(\Delta \mid \sigma_1 \mid \sigma_2) \odot \Gamma \vdash t = \lambda x.(t\,x) : (x :_{(s,r)} A) \to B} \text{ TEQ\_ArrowUniq}$$

Then our first goal holds by premise, and our second goal holds by the following derivation:

$$\dfrac{\begin{array}{c}\mathcal{A}\\ (\Delta \mid \sigma_1 \mid \sigma_2) \odot \Gamma \vdash t : (x :_{(s,r)} A) \to B\end{array}}{(\Delta \mid \sigma_2 \mid \mathbf{0}) \odot \Gamma \vdash (x :_{(s,r)} A) \to B : \mathsf{Type}_l} \text{ L. 3.6}$$

By Lemma 3.8 with $\mathcal{A}$, we have:

- ($\mathcal{B}$) $(\Delta \mid \hat{\sigma_2} \mid \mathbf{0}) \odot \Gamma \vdash A : \mathsf{Type}_{l'}$;
- ($\mathcal{C}$) $(\Delta, \hat{\sigma_2} \mid \hat{\sigma_2}', r \mid \mathbf{0}) \odot \Gamma, x : A \vdash B : \mathsf{Type}_{l''}$;
- $\hat{\sigma_2} + \hat{\sigma_2}' = \sigma_2$

$$\dfrac{\mathcal{D}}{\dfrac{\dfrac{\mathcal{B}}{\Delta, \hat{\sigma_2} \odot \Gamma, x' : A \vdash} \text{ Wf\_Ext}}{(\Delta, \hat{\sigma_2} \mid \mathbf{0}, 1 \mid \hat{\sigma_2}, 0) \odot \Gamma, x' : A \vdash x' : A} \text{ T\_Var}}$$

$$\dfrac{\mathcal{E}}{\dfrac{\mathcal{C} \quad \mathcal{B}}{(\Delta, \hat{\sigma_2}, (\hat{\sigma_2}, 0) \mid \hat{\sigma_2}', 0, r \mid \mathbf{0}) \odot \Gamma, x' : A, x : A \vdash B : \mathsf{Type}_{l''}} \text{ L. 3.26}}$$

$$\dfrac{\mathcal{F}}{\dfrac{(\Delta \mid \sigma_1 \mid \sigma_2) \odot \Gamma \vdash t : (x :_{(s,r)} A) \to B \quad \mathcal{B}}{(\Delta, \hat{\sigma_2} \mid \sigma_1, 0 \mid \sigma_2, 0) \odot \Gamma, x' : A \vdash t : (x :_{(s,r)} A) \to B} \text{ L. 3.26}}$$

$$\dfrac{\mathcal{G}}{\dfrac{\dfrac{\dfrac{\mathcal{E} \quad \mathcal{F} \quad \mathcal{D}}{(\Delta, \hat{\sigma_2} \mid \sigma_1, s \mid \hat{\sigma_2}', r) \odot \Gamma, x' : A \vdash t\,x' : [x'/x]B} \text{ T\_App} \quad \mathcal{B}}{(\Delta, \hat{\sigma_2}, (\hat{\sigma_2}, 0) \mid \sigma_1, 0, s \mid \hat{\sigma_2}', 0, r) \odot \Gamma, x : A, x' : A \vdash t\,x' : [x'/x]B} \text{ L. 3.26}}{(\Delta, \hat{\sigma_2} \mid \sigma_1, s \mid \hat{\sigma_2}', r) \odot \Gamma, x : A \vdash t\,x : B} \text{ L. 3.22}}$$

$$\dfrac{\mathcal{C} \quad \mathcal{G}}{(\Delta \mid \sigma_1 \mid \sigma_2) \odot \Gamma \vdash \lambda x.(t\,x) : (x :_{(s,r)} A) \to B} \text{ T\_Fun}$$

Case.

$$\frac{(\Delta, \sigma_1 \mid \sigma_2, s \mid \sigma_3, r) \odot \Gamma, x : A \vdash t_1 = t_2 : B}{(\Delta \mid \sigma_2 \mid \sigma_1 + \sigma_3) \odot \Gamma \vdash \lambda x.t_1 = \lambda x.t_2 : (x :_{(s,r)} A) \to B} \text{ TEQ\_FUN}$$

Then our goals hold by the following derivations:

$$\frac{\mathcal{A} \quad \dfrac{(\Delta, \sigma_1 \mid \sigma_2, s \mid \sigma_3, r) \odot \Gamma, x : A \vdash t_1 = t_2 : B}{(\Delta, \sigma_1 \mid \sigma_3, r \mid \mathbf{0}) \odot \Gamma, x : A \vdash B : \mathsf{Type}_l}}{} \text{ L. 3.6}$$

$$\frac{\mathcal{A} \quad \dfrac{(\Delta, \sigma_1 \mid \sigma_2, s \mid \sigma_3, r) \odot \Gamma, x : A \vdash t_1 = t_2 : B}{(\Delta, \sigma_1 \mid \sigma_2, s \mid \sigma_3, r) \odot \Gamma, x : A \vdash t_1 : B} \text{ I.H.}}{(\Delta \mid \sigma_2 \mid \sigma_1 + \sigma_3) \odot \Gamma \vdash \lambda x.t_1 : (x :_{(s,r)} A) \to B} \text{ T\_FUN}$$

$$\frac{\mathcal{A} \quad \dfrac{(\Delta, \sigma_1 \mid \sigma_2, s \mid \sigma_3, r) \odot \Gamma, x : A \vdash t_1 = t_2 : B}{(\Delta, \sigma_1 \mid \sigma_2, s \mid \sigma_3, r) \odot \Gamma, x : A \vdash t_2 : B} \text{ I.H.}}{(\Delta \mid \sigma_2 \mid \sigma_1 + \sigma_3) \odot \Gamma \vdash \lambda x.t_2 : (x :_{(s,r)} A) \to B} \text{ T\_FUN}$$

Case.

$$\frac{\begin{array}{c}(\Delta, \sigma_3 \mid \sigma_4, r \mid \mathbf{0}) \odot \Gamma, x : A \vdash B : \mathsf{Type}_l \\ (\Delta \mid \sigma_1 \mid \sigma_3 + \sigma_4) \odot \Gamma \vdash t_1 = t_2 : (x :_{(s,r)} A) \to B \quad (\Delta \mid \sigma_2 \mid \sigma_3) \odot \Gamma \vdash t_3 = t_4 : A\end{array}}{(\Delta \mid \sigma_1 + s * \sigma_2 \mid \sigma_4 + r * \sigma_2) \odot \Gamma \vdash t_1 \, t_3 = t_2 \, t_4 : [t_3/x]B} \text{ TEQ\_APP}$$

Then our goals hold by the following derivations:

$$\frac{\begin{array}{cc}(\Delta, \sigma_3 \mid \sigma_4, r \mid \mathbf{0}) \odot \Gamma, x : A \vdash B : \mathsf{Type}_l & \\ \dfrac{(\Delta \mid \sigma_1 \mid \sigma_3 + \sigma_4) \odot \Gamma \vdash t_1 = t_2 : (x :_{(s,r)} A) \to B}{(\Delta \mid \sigma_1 \mid \sigma_3 + \sigma_4) \odot \Gamma \vdash t_1 : (x :_{(s,r)} A) \to B} \text{ I.H.} & \dfrac{(\Delta \mid \sigma_2 \mid \sigma_3) \odot \Gamma \vdash t_3 = t_4 : A}{(\Delta \mid \sigma_2 \mid \sigma_3) \odot \Gamma \vdash t_3 : A} \text{ I.H.}\end{array}}{(\Delta \mid \sigma_1 + s * \sigma_2 \mid \sigma_4 + r * \sigma_2) \odot \Gamma \vdash t_1 \, t_3 : [t_3/x]B} \text{ T\_APP}$$

$\mathcal{C}$

$$\frac{\begin{array}{cc}(\Delta, \sigma_3 \mid \sigma_4, r \mid \mathbf{0}) \odot \Gamma, x : A \vdash B : \mathsf{Type}_l & \\ \dfrac{(\Delta \mid \sigma_1 \mid \sigma_3 + \sigma_4) \odot \Gamma \vdash t_1 = t_2 : (x :_{(s,r)} A) \to B}{(\Delta \mid \sigma_1 \mid \sigma_3 + \sigma_4) \odot \Gamma \vdash t_2 : (x :_{(s,r)} A) \to B} \text{ I.H.} & \dfrac{(\Delta \mid \sigma_2 \mid \sigma_3) \odot \Gamma \vdash t_3 = t_4 : A}{(\Delta \mid \sigma_2 \mid \sigma_3) \odot \Gamma \vdash t_4 : A} \text{ I.H.}\end{array}}{(\Delta \mid \sigma_1 + s * \sigma_2 \mid \sigma_4 + r * \sigma_2) \odot \Gamma \vdash t_2 \, t_4 : [t_4/x]B} \text{ T\_APP}$$

$$\frac{\mathcal{C} \quad \dfrac{\dfrac{\dfrac{(\Delta, \sigma_3 \mid \sigma_4, r \mid \mathbf{0}) \odot \Gamma, x : A \vdash B : \mathsf{Type}_l \quad (\Delta \mid \sigma_2 \mid \sigma_3) \odot \Gamma \vdash t_3 = t_4 : A}{(\Delta \mid \sigma_4 + r * \sigma_2 \mid \mathbf{0}) \odot \Gamma \vdash [t_3/x]B = [t_4/x]B : \mathsf{Type}_l} \text{ L. 3.28}}{(\Delta \mid \sigma_4 + r * \sigma_2 \mid \mathbf{0}) \odot \Gamma \vdash [t_4/x]B = [t_3/x]B : \mathsf{Type}_l} \text{ TEQ\_SYM}}{(\Delta \mid \sigma_4 + r * \sigma_2) \odot \Gamma \vdash [t_4/x]B \le [t_3/x]B} \text{ ST\_EQ}}{(\Delta \mid \sigma_1 + s * \sigma_2 \mid \sigma_4 + r * \sigma_2) \odot \Gamma \vdash t_2 \, t_4 : [t_3/x]B} \text{ T\_TYCONV}$$

The cases for TEQ_TENCUT and TEQ_BOXE proceed similarly.

Case.

$$\dfrac{\begin{array}{c}(\Delta, \sigma_1 \mid \sigma_3, r \mid \mathbf{0}) \odot \Gamma, x : A \vdash B : \mathsf{Type}_l \\ (\Delta \mid \sigma_2 \mid \sigma_1) \odot \Gamma \vdash t_1 = t'_1 : A \quad (\Delta \mid \sigma_4 \mid \sigma_3 + r * \sigma_2) \odot \Gamma \vdash t_2 = t'_2 : [t_1/x]B\end{array}}{(\Delta \mid \sigma_2 + \sigma_4 \mid \sigma_1 + \sigma_3) \odot \Gamma \vdash (t_1, t_2) = (t'_1, t'_2) : (x :_r A) \otimes B} \; \text{TEQ\_PAIR}$$

Then our goals hold by the following derivations:

$$\dfrac{\dfrac{(\Delta \mid \sigma_2 \mid \sigma_1) \odot \Gamma \vdash t_1 = t'_1 : A}{(\Delta \mid \sigma_2 \mid \sigma_1) \odot \Gamma \vdash t_1 : A} \text{I.H.} \quad \begin{array}{c}(\Delta, \sigma_1 \mid \sigma_3, r \mid \mathbf{0}) \odot \Gamma, x : A \vdash B : \mathsf{Type}_l \\ \dfrac{(\Delta \mid \sigma_4 \mid \sigma_3 + r * \sigma_2) \odot \Gamma \vdash t_2 = t'_2 : [t_1/x]B}{(\Delta \mid \sigma_4 \mid \sigma_3 + r * \sigma_2) \odot \Gamma \vdash t_2 : [t_1/x]B} \text{I.H.}\end{array}}{(\Delta \mid \sigma_2 + \sigma_4 \mid \sigma_1 + \sigma_3) \odot \Gamma \vdash (t_1, t_2) : (x :_r A) \otimes B} \; \text{T\_PAIR}$$

$$\dfrac{\dfrac{\mathcal{C} \quad (\Delta, \sigma_1 \mid \sigma_3, r \mid \mathbf{0}) \odot \Gamma, x : A \vdash B : \mathsf{Type}_l \quad (\Delta \mid \sigma_2 \mid \sigma_1) \odot \Gamma \vdash t_1 = t'_1 : A}{(\Delta \mid \sigma_3 + r * \sigma_2 \mid \mathbf{0}) \odot \Gamma \vdash [t_1/x]B = [t'_1/x]B : \mathsf{Type}_l} \; \text{L. 3.28}}{(\Delta \mid \sigma_3 + r * \sigma_2) \odot \Gamma \vdash [t_1/x]B \leq [t'_1/x]B} \; \text{ST\_EQ}$$

$$\dfrac{\dfrac{\mathcal{D} \quad (\Delta \mid \sigma_4 \mid \sigma_3 + r * \sigma_2) \odot \Gamma \vdash t_2 = t'_2 : [t_1/x]B}{\dfrac{(\Delta \mid \sigma_4 \mid \sigma_3 + r * \sigma_2) \odot \Gamma \vdash t'_2 : [t_1/x]B}{(\Delta \mid \sigma_4 \mid \sigma_3 + r * \sigma_2) \odot \Gamma \vdash t'_2 : [t'_1/x]B} \text{I.H.} \quad \mathcal{C}}}{} \; \text{T\_TYCONV}$$

$$\dfrac{\dfrac{(\Delta \mid \sigma_2 \mid \sigma_1) \odot \Gamma \vdash t_1 = t'_1 : A}{(\Delta \mid \sigma_2 \mid \sigma_1) \odot \Gamma \vdash t'_1 : A} \text{I.H.} \quad (\Delta, \sigma_1 \mid \sigma_3, r \mid \mathbf{0}) \odot \Gamma, x : A \vdash B : \mathsf{Type}_l \quad \mathcal{D}}{(\Delta \mid \sigma_2 + \sigma_4 \mid \sigma_1 + \sigma_3) \odot \Gamma \vdash (t'_1, t'_2) : (x :_r A) \otimes B} \; \text{T\_PAIR}$$

Case.

$$\dfrac{(\Delta \mid \sigma_1 \mid \sigma_2) \odot \Gamma \vdash t : (x :_r A) \otimes B}{(\Delta \mid \sigma_1 \mid \sigma_2) \odot \Gamma \vdash t = (\mathsf{let}\ (x, y) = t\ \mathsf{in}\ (x, y)) : (x :_r A) \otimes B} \; \text{TEQ\_TENU}$$

Then our first goal holds by premise, and our second goal holds by the following derivation:

$$\dfrac{\begin{array}{c}\mathcal{A} \\ (\Delta \mid \sigma_1 \mid \sigma_2) \odot \Gamma \vdash t : (x :_r A) \otimes B\end{array}}{(\Delta \mid \sigma_2 \mid \mathbf{0}) \odot \Gamma \vdash (x :_r A) \otimes B : \mathsf{Type}_l} \; \text{L. 3.6}$$

By Lemma 3.9 with $\mathcal{A}$, we have:

- ($\mathcal{B}$) $(\Delta \mid \hat{\sigma_2} \mid \mathbf{0}) \odot \Gamma \vdash A : \mathsf{Type}_{l'}$;
- ($\mathcal{C}$) $(\Delta, \hat{\sigma_2} \mid \hat{\sigma_2}', r \mid \mathbf{0}) \odot \Gamma, x : A \vdash B : \mathsf{Type}_{l''}$;
- $\hat{\sigma_2} + \hat{\sigma_2}' = \sigma_2$

77

$$\mathcal{D}$$

$$\dfrac{\dfrac{\mathcal{C} \quad \mathcal{B}}{(\Delta, \hat{\sigma_2}, (\hat{\sigma_2}, 0) \mid \hat{\sigma_2}', r, 0 \mid \mathbf{0}) \odot \Gamma, x' : A, x : A \vdash B : \mathsf{Type}_{l''}} \text{ L. 3.26}}{(\Delta, \hat{\sigma_2} \mid \hat{\sigma_2}', r \mid \mathbf{0}) \odot \Gamma, x' : A \vdash [x'/x]B : \mathsf{Type}_{l''}} \text{ L. 3.22}$$

$$\mathcal{E}$$

$$\dfrac{\dfrac{\mathcal{D}}{\Delta, \hat{\sigma_2}, (\hat{\sigma_2}', r) \odot \Gamma, x' : A, y : [x'/x]B \vdash} \text{ Wf\_Ext}}{(\Delta, \hat{\sigma_2}, (\hat{\sigma_2}', r) \mid \mathbf{0}, 1, 0 \mid \hat{\sigma_2}, 0, 0) \odot \Gamma, x' : A, y : [x'/x]B \vdash x' : A} \text{ T\_Var}$$

$$\mathcal{F}$$

$$\dfrac{\dfrac{\mathcal{C} \quad \mathcal{B}}{(\Delta, \hat{\sigma_2}, (\hat{\sigma_2}, 0) \mid \hat{\sigma_2}', 0, r \mid \mathbf{0}) \odot \Gamma, x' : A, x : A \vdash B : \mathsf{Type}_{l''}} \text{ L. 3.26} \quad \mathcal{D}}{(\Delta, \hat{\sigma_2}, (\hat{\sigma_2}, r), (\hat{\sigma_2}', 0, 0) \mid \hat{\sigma_2}', 0, 0, r \mid \mathbf{0}) \odot \Gamma, x' : A, y : [x'/x]B, x : A \vdash B : \mathsf{Type}_{l''}} \text{ L. 3.26}$$

$$\mathcal{G}$$

$$\dfrac{\mathcal{E} \quad \dfrac{\dfrac{\mathcal{F}}{\Delta, \hat{\sigma_2}, (\hat{\sigma_2}', r), (\hat{\sigma_2}, 0, 0) \odot \Gamma, x' : A, y : [x'/x]B, x : A \vdash} \text{ L. 3.3}}{(\Delta, \hat{\sigma_2}, (\hat{\sigma_2}', r), (\hat{\sigma_2}, 0, 0) \mid \mathbf{0}, 0, 1, 0 \mid \hat{\sigma_2}', r, 0, 0) \odot \Gamma, x' : A, y : [x'/x]B, x : A \vdash y : [x'/x]B} \text{ T\_Var}}{(\Delta, \hat{\sigma_2}, (\hat{\sigma_2}', r) \mid \mathbf{0}, 0, 1 \mid \hat{\sigma_2}', r, 0) \odot \Gamma, x' : A, y : [x'/x]B \vdash y : [x'/x]B} \text{ L. 3.27}$$

$$\mathcal{H}$$

$$\dfrac{\dfrac{\dfrac{\mathcal{F} \quad \mathcal{E} \quad \mathcal{G}}{(\Delta, \hat{\sigma_2}, (\hat{\sigma_2}', r) \mid \mathbf{0}, 1, 1 \mid \sigma_2, 0, 0) \odot \Gamma, x' : A, y : [x'/x]B \vdash (x', y) : (x :_r A) \otimes B} \text{ T\_Pair} \quad \mathcal{B}}{(\Delta, \hat{\sigma_2}, (\hat{\sigma_2}, 0), (\hat{\sigma_2}', 0, r) \mid \mathbf{0}, 0, 1, 1 \mid \sigma_2, 0, 0, 0) \odot \Gamma, x : A, x' : A, y : [x'/x]B \vdash (x', y) : (x :_r A) \otimes B} \text{ L. 3.26}}{(\Delta, \hat{\sigma_2}, (\hat{\sigma_2}', 0, r) \mid \mathbf{0}, 1, 1 \mid \sigma_2, 0, 0) \odot \Gamma, x : A, y : B \vdash (x, y) : (x :_r A) \otimes B} \text{ L. 3.22}$$

$$\dfrac{(\Delta \mid \sigma_1 \mid \sigma_2) \odot \Gamma \vdash t : (x :_r A) \otimes B \quad \dfrac{\mathcal{A} \quad \mathcal{A}}{(\Delta, \sigma_2 \mid \sigma_2, 0 \mid \mathbf{0}) \odot \Gamma, z : (x :_r A) \otimes B \vdash (x :_r A) \otimes B : \mathsf{Type}_l} \text{ L. 3.24} \quad \mathcal{H}}{(\Delta \mid \sigma_1 \mid \sigma_2) \odot \Gamma \vdash \mathsf{let}\,(x,y) = t \,\mathsf{in}\, (x,y) : (x :_r A) \otimes B} \text{ T\_TenCut}$$

---

Case.

$$\dfrac{(\Delta \mid \sigma \mid \mathbf{0}) \odot \Gamma \vdash A = B : \mathsf{Type}_l}{(\Delta \mid \sigma \mid \mathbf{0}) \odot \Gamma \vdash \Box_s A = \Box_s B : \mathsf{Type}_l} \text{ TEQ\_Box}$$

Then our goals hold by the following derivations:

$$\dfrac{\dfrac{(\Delta \mid \sigma \mid \mathbf{0}) \odot \Gamma \vdash A = B : \mathsf{Type}_l}{(\Delta \mid \sigma \mid \mathbf{0}) \odot \Gamma \vdash A : \mathsf{Type}_l} \text{ I.H.}}{(\Delta \mid \sigma \mid \mathbf{0}) \odot \Gamma \vdash \Box_s A : \mathsf{Type}_l} \text{ T\_Box}$$

$$\dfrac{\dfrac{(\Delta \mid \sigma \mid \mathbf{0}) \odot \Gamma \vdash A = B : \mathsf{Type}_l}{(\Delta \mid \sigma \mid \mathbf{0}) \odot \Gamma \vdash B : \mathsf{Type}_l} \text{ I.H.}}{(\Delta \mid \sigma \mid \mathbf{0}) \odot \Gamma \vdash \Box_s B : \mathsf{Type}_l} \text{ T\_Box}$$

Case.

$$\dfrac{(\Delta \mid \sigma_1 \mid \sigma_2) \odot \Gamma \vdash t_1 = t_2 : A}{(\Delta \mid s * \sigma_1 \mid \sigma_2) \odot \Gamma \vdash \Box t_1 = \Box t_2 : \Box_s A} \text{ TEQ\_BoxI}$$

Then our goals hold by the following derivations:

$$\dfrac{\dfrac{(\Delta \mid \sigma_1 \mid \sigma_2) \odot \Gamma \vdash t_1 = t_2 : A}{(\Delta \mid \sigma_1 \mid \sigma_2) \odot \Gamma \vdash t_1 : A} \text{ I.H.}}{(\Delta \mid s * \sigma_1 \mid \sigma_2) \odot \Gamma \vdash \Box t_1 : \Box_s A} \text{ T\_BoxI}$$

$$\dfrac{\dfrac{(\Delta \mid \sigma_1 \mid \sigma_2) \odot \Gamma \vdash t_1 = t_2 : A}{(\Delta \mid \sigma_1 \mid \sigma_2) \odot \Gamma \vdash t_2 : A} \text{ I.H.}}{(\Delta \mid s * \sigma_1 \mid \sigma_2) \odot \Gamma \vdash \Box t_2 : \Box_s A} \text{ T\_BoxI}$$

Case.

$$\dfrac{(\Delta \mid \sigma_1 \mid \sigma_2) \odot \Gamma \vdash t : \Box_s A}{(\Delta \mid \sigma_1 \mid \sigma_2) \odot \Gamma \vdash t = (\mathsf{let}\,\Box x = t\,\mathsf{in}\,\Box x) : \Box_s A} \text{ TEQ\_BoxU}$$

Then our first goal holds by premise, and our second goal holds by the following derivation:

$$\dfrac{\mathcal{C} \quad (\Delta \mid \sigma_1 \mid \sigma_2) \odot \Gamma \vdash t : \Box_s A}{(\Delta \mid \sigma_2 \mid \mathbf{0}) \odot \Gamma \vdash \Box_s A : \mathsf{Type}_l} \text{ L. 3.6}$$

$$\mathcal{D}$$
$$\dfrac{\dfrac{\dfrac{\mathcal{C}}{(\Delta \mid \sigma_2 \mid \mathbf{0}) \odot \Gamma \vdash A : \mathsf{Type}_{l'}} \text{ L. 3.10}}{\dfrac{\Delta, \sigma_2 \odot \Gamma, x : A \vdash}{(\Delta, \sigma_2 \mid \mathbf{0}, 1 \mid \sigma_2, 0) \odot \Gamma, x : A \vdash x : A} \text{ T\_Var}} \text{ Wf\_Ext}}{(\Delta, \sigma_2 \mid \mathbf{0}, s \mid \sigma_2, 0) \odot \Gamma, x : A \vdash \Box x : \Box_s A} \text{ T\_BoxI}$$

$$\dfrac{(\Delta \mid \sigma_1 \mid \sigma_2) \odot \Gamma \vdash t : \Box_s A \quad \dfrac{\mathcal{C} \quad \mathcal{C}}{(\Delta, \sigma_2 \mid \sigma_2, 0 \mid \mathbf{0}) \odot \Gamma, z : \Box_s A \vdash \Box_s A : \mathsf{Type}_l} \text{ L. 3.24} \quad \mathcal{D}}{(\Delta \mid \sigma_1 \mid \sigma_2) \odot \Gamma \vdash \mathsf{let}\,\Box x = t\,\mathsf{in}\,\Box x : \Box_s A} \text{ T\_BoxE}$$

**Lemma 3.33** (Deriving judgments under equal contexts). *If $(\Delta_1, \sigma_1, \Delta_2 \mid \sigma_2, s, \sigma_3 \mid \sigma_4, r, \sigma_5) \odot \Gamma_1, x : A, \Gamma_2 \vdash \mathcal{J}$ and $(\Delta_1 \mid \sigma_1 \mid \mathbf{0}) \odot \Gamma \vdash A = A' : \mathsf{Type}_l$, then $(\Delta_1, \sigma_1, \Delta_2 \mid \sigma_2, s, \sigma_3 \mid \sigma_4, r, \sigma_5) \odot \Gamma_1, x : A', \Gamma_2 \vdash \mathcal{J}$.*

*Proof.* For well-formed contexts:

> Case.
>
> $$\frac{(\Delta \mid \sigma \mid \mathbf{0}) \odot \Gamma \vdash A : \mathsf{Type}_l}{\Delta, \sigma \odot \Gamma, x : A \vdash} \; \text{WF\_EXT}$$
>
> With $(\Delta \mid \sigma \mid \mathbf{0}) \odot \Gamma \vdash A = A' : \mathsf{Type}_l$. Then our goal holds by the following derivation:
>
> $$\frac{\dfrac{(\Delta \mid \sigma \mid \mathbf{0}) \odot \Gamma \vdash A = A' : \mathsf{Type}_l}{(\Delta \mid \sigma \mid \mathbf{0}) \odot \Gamma \vdash A' : \mathsf{Type}_l} \; \text{L. 3.32}}{\Delta, \sigma \odot \Gamma, x : A' \vdash} \; \text{WF\_EXT}$$

> Case.
>
> $$\frac{(\Delta_1, \sigma_1, \Delta_2 \mid \sigma_2 \mid \mathbf{0}) \odot \Gamma_1, x : A, \Gamma_2 \vdash B : \mathsf{Type}_{l'}}{\Delta_1, \sigma_1, \Delta_2, \sigma_2 \odot \Gamma_1, x : A, \Gamma_2, y : B \vdash} \; \text{WF\_EXT}$$
>
> Then our goal holds by the following derivation:
>
> $$\frac{\dfrac{(\Delta_1, \sigma_1, \Delta_2 \mid \sigma_2 \mid \mathbf{0}) \odot \Gamma_1, x : A, \Gamma_2 \vdash B : \mathsf{Type}_{l'} \quad (\Delta_1 \mid \sigma_1 \mid \mathbf{0}) \odot \Gamma_1 \vdash A = A' : \mathsf{Type}_l}{(\Delta_1, \sigma_1, \Delta_2 \mid \sigma_2 \mid \mathbf{0}) \odot \Gamma_1, x : A', \Gamma_2 \vdash B : \mathsf{Type}_{l'}} \; \text{I.H.}}{\Delta_1, \sigma_1, \Delta_2, \sigma_2 \odot \Gamma_1, x : A', \Gamma_2, y : B \vdash} \; \text{WF\_EXT}$$

For typing, all cases proceed by induction then re-application to respective rules. For equality, all cases proceed by induction then re-application to respective rules. For subtyping, all cases proceed by induction then re-application to respective rules. □

## A.6 Proofs for subtyping

**Lemma 3.34** (Subtyping inversion to typing). *If $(\Delta \mid \sigma) \odot \Gamma \vdash A \leq B$, then $(\Delta \mid \sigma \mid \mathbf{0}) \odot \Gamma \vdash A : \mathsf{Type}_l$ and $(\Delta \mid \sigma \mid \mathbf{0}) \odot \Gamma \vdash B : \mathsf{Type}_l$ for some level $l$.*

*Proof.* By induction on the form of $(\Delta \mid \sigma) \odot \Gamma \vdash A \leq B$, as follows:

> Case.
>
> $$\frac{(\Delta \mid \sigma \mid \mathbf{0}) \odot \Gamma \vdash A = B : \mathsf{Type}_l}{(\Delta \mid \sigma) \odot \Gamma \vdash A \leq B} \; \text{ST\_EQ}$$
>
> By Lemma 3.32 we have $(\Delta \mid \sigma \mid \mathbf{0}) \odot \Gamma \vdash A : \mathsf{Type}_l$ and $(\Delta \mid \sigma \mid \mathbf{0}) \odot \Gamma \vdash B : \mathsf{Type}_l$, as required.

80

Case.

$$\frac{(\Delta \mid \sigma) \odot \Gamma \vdash A \leq B \quad (\Delta \mid \sigma) \odot \Gamma \vdash B \leq C}{(\Delta \mid \sigma) \odot \Gamma \vdash A \leq C} \text{ ST\_Trans}$$

Our goals hold by induction.

Case.

$$\frac{\Delta \odot \Gamma \vdash \quad l \leq l'}{(\Delta \mid \mathbf{0}) \odot \Gamma \vdash \mathsf{Type}_l \leq \mathsf{Type}_{l'}} \text{ ST\_Ty}$$

Our goals hold by the following derivations:

$$\frac{\Delta \odot \Gamma \vdash}{(\Delta \mid \mathbf{0} \mid \mathbf{0}) \odot \Gamma \vdash \mathsf{Type}_l : \mathsf{Type}_{\mathsf{suc}\ l}} \text{ T\_Type}$$

$$\frac{\Delta \odot \Gamma \vdash}{(\Delta \mid \mathbf{0} \mid \mathbf{0}) \odot \Gamma \vdash \mathsf{Type}_{l'} : \mathsf{Type}_{\mathsf{suc}\ l'}} \text{ T\_Type}$$

Case.

$$\frac{\begin{array}{c}(\Delta, \sigma_1 \mid \sigma_2, r \mid \mathbf{0}) \odot \Gamma, x : A \vdash B : \mathsf{Type}_l \\ (\Delta \mid \sigma_1) \odot \Gamma \vdash A' \leq A \quad (\Delta, \sigma_1 \mid \sigma_2, r) \odot \Gamma, x : A' \vdash B \leq B'\end{array}}{(\Delta \mid \sigma_1 + \sigma_2) \odot \Gamma \vdash (x :_{(s,r)} A) \to B \leq (x :_{(s,r)} A') \to B'} \text{ ST\_Arrow}$$

Then we have our goals by the following derivations:

$\mathcal{C}$

$$\frac{\dfrac{(\Delta \mid \sigma_1) \odot \Gamma \vdash A' \leq A}{(\Delta \mid \sigma_1 \mid \mathbf{0}) \odot \Gamma \vdash A' : \mathsf{Type}_{l'}} \text{ I.H.} \quad \dfrac{(\Delta, \sigma_1 \mid \sigma_2, r) \odot \Gamma, x : A' \vdash B \leq B'}{(\Delta, \sigma_1 \mid \sigma_2, r \mid \mathbf{0}) \odot \Gamma, x : A' \vdash B' : \mathsf{Type}_{l''}} \text{ I.H.}}{(\Delta \mid \sigma_1 + \sigma_2 \mid \mathbf{0}) \odot \Gamma \vdash (x :_{(s,r)} A') \to B' : \mathsf{Type}_{l' \sqcup l''}} \text{ T\_Arrow}$$

$\mathcal{D}$

$$\frac{\dfrac{(\Delta \mid \sigma_1) \odot \Gamma \vdash A' \leq A}{(\Delta \mid \sigma_1 \mid \mathbf{0}) \odot \Gamma \vdash A : \mathsf{Type}_{l'}} \text{ I.H.} \quad (\Delta, \sigma_1 \mid \sigma_2, r \mid \mathbf{0}) \odot \Gamma, x : A \vdash B : \mathsf{Type}_l}{(\Delta \mid \sigma_1 + \sigma_2 \mid \mathbf{0}) \odot \Gamma \vdash (x :_{(s,r)} A) \to B : \mathsf{Type}_{l' \sqcup l}} \text{ T\_Arrow}$$

$$\frac{\mathcal{C} \quad \dfrac{\dfrac{(\Delta \mid \sigma_1) \odot \Gamma \vdash A' \leq A}{\Delta \odot \Gamma \vdash} \text{ L. 3.3} \quad l' \sqcup l'' \leq l \sqcup l' \sqcup l''}{(\Delta \mid \mathbf{0}) \odot \Gamma \vdash \mathsf{Type}_{l' \sqcup l''} \leq \mathsf{Type}_{l \sqcup l' \sqcup l''}} \text{ ST\_Ty}}{(\Delta \mid \sigma_1 + \sigma_2 \mid \mathbf{0}) \odot \Gamma \vdash (x :_{(s,r)} A') \to B' : \mathsf{Type}_{l \sqcup l' \sqcup l''}} \text{ T\_TyConv}$$

$$\mathcal{D} \dfrac{\dfrac{(\Delta \mid \sigma_1) \odot \Gamma \vdash A' \leq A}{\Delta \odot \Gamma \vdash} \text{ L. 3.3} \quad \dfrac{(\Delta \mid \mathbf{0}) \odot \Gamma \vdash \mathsf{Type}_{l' \sqcup l} \leq \mathsf{Type}_{l \sqcup l' \sqcup l''}}{(\Delta \mid \sigma_1 + \sigma_2 \mid \mathbf{0}) \odot \Gamma \vdash (x :_{(s,r)} A) \to B : \mathsf{Type}_{l \sqcup l' \sqcup l''}}} \text{ T\_TyConv}$$

The case for ST_TEN proceeds similarly, using T_TEN.

Case.

$$\dfrac{(\Delta, \sigma_1 \mid \sigma_2, r) \odot \Gamma, x : A \vdash B \leq B'}{(\Delta \mid \sigma_1 + \sigma_2) \odot \Gamma \vdash (x :_r A) \otimes B \leq (x :_r A) \otimes B'} \text{ ST\_TEN}$$

Then we have our goals by the following derivations:

$$\dfrac{\dfrac{(\Delta, \sigma_1 \mid \sigma_2, r) \odot \Gamma, x : A \vdash B \leq B'}{(\Delta \mid \sigma_1 \mid \mathbf{0}) \odot \Gamma \vdash A : \mathsf{Type}_l} \text{ L. 3.6} \quad \dfrac{(\Delta, \sigma_1 \mid \sigma_2, r) \odot \Gamma, x : A \vdash B \leq B'}{(\Delta, \sigma_1 \mid \sigma_2, r \mid \mathbf{0}) \odot \Gamma, x : A \vdash B : \mathsf{Type}_{l'}} \text{ I.H.}}{(\Delta \mid \sigma_1 + \sigma_2 \mid \mathbf{0}) \odot \Gamma \vdash (x :_r A) \otimes B : \mathsf{Type}_{l \sqcup l'}} \text{ T\_TEN}$$

$$\dfrac{\dfrac{(\Delta, \sigma_1 \mid \sigma_2, r) \odot \Gamma, x : A \vdash B \leq B'}{(\Delta \mid \sigma_1 \mid \mathbf{0}) \odot \Gamma \vdash A : \mathsf{Type}_l} \text{ L. 3.6} \quad \dfrac{(\Delta, \sigma_1 \mid \sigma_2, r) \odot \Gamma, x : A \vdash B \leq B'}{(\Delta, \sigma_1 \mid \sigma_2, r \mid \mathbf{0}) \odot \Gamma, x : A \vdash B' : \mathsf{Type}_{l'}} \text{ I.H.}}{(\Delta \mid \sigma_1 + \sigma_2 \mid \mathbf{0}) \odot \Gamma \vdash (x :_r A) \otimes B : \mathsf{Type}_{l \sqcup l'}} \text{ T\_TEN}$$

Case.

$$\dfrac{(\Delta \mid \sigma) \odot \Gamma \vdash A \leq A'}{(\Delta \mid \sigma) \odot \Gamma \vdash \Box_s A \leq \Box_s A'} \text{ ST\_Box}$$

Then we have our goals by the following derivations:

$$\dfrac{\dfrac{(\Delta \mid \sigma) \odot \Gamma \vdash A' \leq A}{(\Delta \mid \sigma \mid \mathbf{0}) \odot \Gamma \vdash A' : \mathsf{Type}_l} \text{ I.H.}}{(\Delta \mid \sigma \mid \mathbf{0}) \odot \Gamma \vdash \Box_s A' : \mathsf{Type}_l} \text{ T\_Box}$$

$$\dfrac{\dfrac{(\Delta \mid \sigma) \odot \Gamma \vdash A' \leq A}{(\Delta \mid \sigma \mid \mathbf{0}) \odot \Gamma \vdash A : \mathsf{Type}_{l'}} \text{ I.H.}}{(\Delta \mid \sigma \mid \mathbf{0}) \odot \Gamma \vdash \Box_s A : \mathsf{Type}_l} \text{ T\_Box}$$

$\Box$

## A.7 Proof for type preservation

**Lemma 3.35** (Type preservation). *If $(\Delta \mid \sigma_1 \mid \sigma_2) \odot \Gamma \vdash t : A$ and $t \rightsquigarrow t'$, then $(\Delta \mid \sigma_1 \mid \sigma_2) \odot \Gamma \vdash t' : A$.*

*Proof.* By Lemma 3.31 we have $(\Delta \mid \sigma_1 \mid \sigma_2) \odot \Gamma \vdash t = t' : A$, and therefore by Lemma 3.32 we have $(\Delta \mid \sigma_1 \mid \sigma_2) \odot \Gamma \vdash t' : A$, as required. □

## A.8 Standard results

For inductive results on multi-judgment lemmas, refer to the following lists for how to obtain goals:
For subtyping:

- ST_Eq This holds by induction then re-application to the rule.
- ST_Trans This holds by induction then re-application to the rule.
- ST_Ty This holds by induction then re-application to the rule.
- ST_Arrow This holds by induction then re-application to the rule (see the T_Arrow case for how to handle the typing premise for $[[B]]$).
- ST_Ten This holds by induction then re-application to the rule (see the T_Fun case for how to handle the extended context).
- ST_Box This holds by induction then re-application to the rule.

For equality:

- TEQ_Refl This holds by induction then re-application to the rule.
- TEQ_Trans Holds similarly to the case for TEQ_Refl.
- TEQ_Sym Holds similarly to the case for TEQ_Refl.
- TEQ_ConvTy Holds similarly to the case for TEQ_Refl.
- TEQ_Arrow Holds similarly to the case for T_Arrow.
- TEQ_ArrowComp Holds similarly to the case for T_App (induction then re-application).
- TEQ_ArrowUniq Holds by induction then re-application.
- TEQ_Fun Holds similarly to the case for T_Fun.
- TEQ_App Holds similarly to the case for T_App.
- TEQ_Ten Holds similarly to the case for T_Ten.
- TEQ_TenComp Holds by induction then re-application.
- TEQ_Pair Holds similarly to the case for T_Pair.
- TEQ_TenCut Holds similarly to the case for T_TenCut.
- TEQ_TenU Holds by induction then re-application.
- TEQ_Box Holds similarly to the case for T_Box.
- TEQ_BoxI Holds similarly to the case for T_BoxI.
- TEQ_BoxB Holds by induction then re-application.
- TEQ_BoxE Holds similarly to the case for T_BoxE.
- TEQ_BoxU Holds by induction then re-application.

# B Encoding

## B.1 Simply-typed Lambda Calculus

**A subset of GrTT encodes STLC** We define the following subset of GRTT typing judgments by the predicate $\text{STLC}(\vdash)$:

$$\text{STLC}((\emptyset \mid \emptyset \mid \emptyset) \odot \emptyset \vdash t : A)$$
$$\text{STLC}((\Delta \mid \sigma_1 \mid \sigma_2) \odot \Gamma \vdash t : A) \implies \text{STLC}((\Delta, \mathbf{0} \mid \sigma_1, r \mid \sigma_2, 0) \odot \Gamma, x : B \vdash t : A)$$

$$\text{STLC}(\emptyset \odot \emptyset \vdash)$$
$$\text{STLC}(\Delta \odot \Gamma \vdash) \implies \text{STLC}(\Delta, \mathbf{0} \odot \Gamma, x : A \vdash)$$

Then we define an inductive encoding $[\![-]\!]$ from of a subset of GRTT syntax (ignoring Type and tensors) to STLC. We define a partial inductive encoding on terms $[\![-]\!]_t$ and contexts $[\![-]\!]_\Gamma$:

$$[\![x]\!]_t = x$$
$$[\![\lambda x.t]\!]_t = \lambda x.[\![t]\!]_t$$
$$[\![t_1\ t_2]\!]_t = [\![t_1]\!]_t\ [\![t_2]\!]_t$$

$$[\![(x :_{(r,0)} A) \to B]\!]_\tau = [\![A]\!]_\tau \to [\![B]\!]_\tau$$

$$[\![\emptyset]\!] = \cdot$$
$$[\![\Gamma, x : A]\!] = [\![\Gamma]\!], x : [\![A]\!]_\tau$$

### B.1.1 Key lemmas on quantiative use

Some key lemmas for soundness

**Lemma B.1.** *Given a quantitative semiring, if $(\Delta \mid \sigma \mid \mathbf{0}) \odot \Gamma \vdash A : C$, with $(\Delta \mid \mathbf{0}) \odot \Gamma \vdash C \leq \text{Type}_l$ for some level $l$, and $[\![A]\!]_\tau$ is defined, then:*

$$\mathsf{fv}(A) = \{y \mid \Gamma[i] = y : B \wedge \sigma[i] \neq 0\}$$

*Proof.* By induction on the form of $(\Delta \mid \sigma \mid \mathbf{0}) \odot \Gamma \vdash A : \text{Type}_l$, as follows:

> Case.
> 
> $$\frac{(\Delta \mid \sigma_1 \mid \mathbf{0}) \odot \Gamma \vdash A : \text{Type}_{l_1} \quad (\Delta, \sigma_1 \mid \sigma_2, r \mid \mathbf{0}) \odot \Gamma, x : A \vdash B : \text{Type}_{l_2}}{(\Delta \mid \sigma_1 + \sigma_2 \mid \mathbf{0}) \odot \Gamma \vdash (x :_{(s,r)} A) \to B : \text{Type}_{l_1 \sqcup l_2}} \text{T\_ARROW}$$
> 
> By induction we have:
> 
> $$\mathsf{fv}(A) = \{y \mid \Gamma[i] = y : B \wedge \sigma_1[i] \neq 0\} \quad \text{(ih1)}$$
> $$\mathsf{fv}(B) = \{y \mid \Gamma[i] = y : B \wedge \sigma_2, r[i] \neq 0\} \quad \text{(ih2)}$$
> 
> The goal is then that:
> 
> $$\mathsf{fv}((x :_{(s,r)} A) \to B : \text{Type}_{l_1 \sqcup l_2}) = \mathsf{fv}(A) \cup (\mathsf{fv}(B) \setminus \{x\})$$
> $$= \mathsf{fv}((x :_{(s,r)} A) \to B : \text{Type}_{l_1 \sqcup l_2}) = \{y \mid \Gamma[i] = y : B \wedge \sigma_1 + \sigma_2[i] \neq 0\}$$

84

We get that $\mathsf{fv}(B) \setminus \{x\} = \{y \mid \Gamma[i] = y : B \wedge \sigma_2[i] \neq 0\} = (\text{ih1}) \setminus \{x\}$.

Next we then need to prove by $((\sigma_2[i] \neq 0) \vee (\sigma_3[i] \neq 0)) \Leftrightarrow (\sigma_2 + \sigma_3)[i] \neq 0$:

- Left-to-right: Assuming $(\sigma_2[i] \neq 0 \vee (\sigma_3)[i] \neq 0)$
  By contradiction: Assume $(\sigma_2 + \sigma_3)[i] = 0$ then by positivity $\sigma_2[i] = 0$ and $\sigma_3[i] = 0$.
  Thus we can eliminate the assumption with these to get a contradiction.
- Right-to-left: Assuming $(\sigma_2 + \sigma_3)[i] \neq 0$
  The goal is $(\sigma_2[i] \neq 0 \vee (\sigma_3)[i] \neq 0)$ which is equivalent to $\neg(\sigma_2[i] = 0 \wedge \sigma_3[i] = 0)$ thus we go by contradiction:
  Assuming $(\sigma_2[i] = 0 \wedge \sigma_3[i] = 0)$ Therefore by monoid unit we have $(\sigma_2 + \sigma_3)[i] = 0$ which contradicts the assumption hence we by PBC that $(\sigma_2[i] \neq 0 \vee (\sigma_3)[i] \neq 0)$.

And we are done.

Case.

$$\dfrac{(\Delta \mid \sigma \mid \mathbf{0}) \odot \Gamma \vdash A : B \quad (\Delta \mid \mathbf{0}) \odot \Gamma \vdash B \leq C}{(\Delta \mid \sigma \mid \mathbf{0}) \odot \Gamma \vdash A : C} \text{ T\_TyConv}$$

By ST_Trans we have $(\Delta \mid \mathbf{0}) \odot \Gamma \vdash B \leq \mathsf{Type}_l$, and therefore our goal holds by induction.

$\square$

**Lemma B.2.** *For a quantitative semiring, under the interpretation, if $\Delta \odot \Gamma \vdash$ then*

$$\forall i. 0 \leq i < |\Gamma| . \ \mathsf{fv}(\Gamma[i]) = \{y \mid \forall j. \Gamma[i] = x : A \wedge j < i \wedge \Gamma[j] = y : B \wedge \Delta[i][j] \neq 0)\}$$

*Proof.*
- Case Wf_Empty

$$\overline{\emptyset \odot \emptyset \vdash} \text{ Wf\_Empty}$$

Trivial since $\mathsf{fv}(\emptyset = \emptyset)$

- Case Wf_Ext

$$\dfrac{(\Delta \mid \sigma \mid \mathbf{0}) \odot \Gamma \vdash A : \mathsf{Type}_l}{\Delta, \sigma \odot \Gamma, x : A \vdash} \text{ Wf\_Ext}$$

By Lemma B.1 on $(\Delta \mid \sigma \mid \mathbf{0}) \odot \Gamma \vdash A : \mathsf{Type}_l$ then we have that:

$$\mathsf{fv}(\Gamma[i]) = \{y \mid \forall j. \Gamma[i] = x : A \wedge j < i \wedge \Gamma[j] = y : B \wedge \Delta[i][j] \neq 0)\} \tag{1}$$

$$\mathsf{fv}(A) = \{x \mid \Gamma[i] = x : A' \wedge \sigma[i] \neq 0\} \tag{2}$$

Therefore we can use $\mathsf{fv}(\Gamma[i])$ to prove the goal for all $i$ where $0 \leq i < |\Gamma|$ then for $i' = |\Gamma|$ we need to prove that:

$$\mathsf{fv}(\Gamma[i']) = \{y \mid \forall j. (\Gamma, x : A)[i'] = x : A \wedge j < i' \wedge \Gamma[j] = y : B \wedge \Delta, \sigma[i'][j] \neq 0)\} \quad \text{(goal)}$$

which holds since $(\Gamma, x : A)[i'] = x : A$ is true and $\Delta, \sigma[i'][j] = \sigma[j]$ then our goal reduces to:

$$\mathsf{fv}(\Gamma[i']) = \{y \mid \forall j. j < i' \wedge \Gamma[j] = y : B \wedge \sigma[j] \neq 0)\}$$

which follows from $\mathsf{fv}(A)$ (by alpha equivalence).

$\square$

### B.1.2 Type soundness

Given a derivation of $(\Delta \mid \sigma_1 \mid \sigma_2) \odot \Gamma \vdash t : A$ such that $\text{STLC}((\Delta \mid \sigma_1 \mid \sigma_2) \odot \Gamma \vdash t : A)$ then $[\![\Gamma]\!] \vdash [\![t]\!] : [\![A]\!]$ in STLC.

**Proof** By induction of the typing relation.

- Case T_Type. Trivial since it is rejected by $\text{STLC}(-)$

- Case T_Var, thus under $\text{STLC}(-)$ we have:

$$\frac{(\Delta_1, \mathbf{0}, \Delta_2) \odot \Gamma_1, x : A, \Gamma_2 \vdash}{(\Delta_1, \mathbf{0}, \Delta_2 \mid \mathbf{0}, 1, \mathbf{0} \mid \mathbf{0}) \odot \Gamma_1, x : A, \Gamma_2 \vdash x : A} \text{ T\_Var}$$

We can thus form the derivation in STLC as follows since we know that $A$ is closed by the typing (Corollary B.8.1), i.e., there does not occur any variables inside of $A$ by the subject grade $\mathbf{0}$.

$$\overline{[\![\Gamma_1]\!], x : [\![A]\!], [\![\Gamma_2]\!] \vdash x : [\![A]\!]}$$

- Case T_Arrow, T_Ten are trivial since they do not satisfy the predicate

- Case T_Fun

$$\frac{(\Delta, \sigma_1 \mid \sigma_3, r \mid \mathbf{0}) \odot \Gamma, x : A \vdash B : \mathsf{Type}_l \quad (\Delta, \sigma_1 \mid \sigma_2, s \mid \sigma_3, r) \odot \Gamma, x : A \vdash t : B}{(\Delta \mid \sigma_2 \mid \sigma_1 + \sigma_3) \odot \Gamma \vdash \lambda x.t : (x :_{(s,r)} A) \to B} \text{ T\_Fun}$$

By induction on $[\![(\Delta, \sigma_1 \mid \sigma_2, s \mid \sigma_3, r) \odot \Gamma, x : A \vdash t : B]\!]$ we have:

$$[\![\Gamma]\!], x : [\![A]\!] \vdash [\![t]\!] : [\![B]\!]$$

Thus we can form the STLC derivation:

$$\frac{[\![\Gamma]\!], x : [\![A]\!] \vdash [\![t]\!] : [\![B]\!]}{[\![\Gamma]\!] \vdash \lambda x.[\![t]\!] : [\![A]\!] \to [\![B]\!]}$$

- Case T_App

$$\frac{(\Delta \mid \sigma_2 \mid \mathbf{0}) \odot \Gamma \vdash t_1 : (x :_{(s,0)} A) \to B \quad (\Delta \mid \sigma_4 \mid \mathbf{0}) \odot \Gamma \vdash t_2 : A}{(\Delta \mid \sigma_2 + s * \mathbf{0} \mid \mathbf{0} + 0 * \sigma_4) \odot \Gamma \vdash t_1 \, t_2 : [t_2/x]B} \text{ T\_App}$$

By Corollary B.8.1, then $B$ is closed and thus the conclusion here is equal to:

$$(\Delta \mid \sigma_2 \mid \mathbf{0}) \odot \Gamma \vdash t_1 \, t_2 : B$$

Trivially by induction we get the STLC judgments:

$$[\![\Gamma]\!] \vdash [\![t_1]\!] : [\![A]\!] \to [\![B]\!] \tag{ih1}$$
$$[\![\Gamma]\!] \vdash [\![t_2]\!] : [\![A]\!] \tag{ih2}$$

Thus we can form the typing: $[\![\Gamma]\!] \vdash [\![t_1]\!] \, [\![t_2]\!] : [\![B]\!]$.

- Rest of the rules are trivial since they are not part of the interpretation.

### B.1.3 Soundness

**Lemma B.3.** $[\![[t_2/x]t_1]\!] = [[\![t_2]\!]/x][\![t_1]\!]$.

*Proof.* Straightforward by induction on the definition of syntactic substitution for GRTT. □

**Lemma B.4** (Soundness). *Given* $\text{STLC}((\Delta \mid \sigma_1 \mid \sigma_2) \odot \Gamma \vdash t : A)$ *then* $[\![(\Delta \mid \sigma_1 \mid \sigma_2) \odot \Gamma \vdash t : A]\!]$ *implies* $[\![\Gamma]\!] \vdash [\![t]\!] : [\![A]\!]$ *in* GRTT *and whenever* $t \leadsto t'$ *then* $[\![t]\!] \leadsto [\![t']\!]$ *in the CBN STLC (+ application right congruence reduction).*

*Proof.* First we observe that by Type Preservation (Lemma 3.35) that reduction preserves typing thus if $\text{STLC}(-)$ holds for $[\![(\Delta \mid \sigma_1 \mid \sigma_2) \odot \Gamma \vdash t : A]\!]$ then $\text{STLC}((\Delta \mid \sigma_1 \mid \sigma_2) \odot \Gamma \vdash t' : A)$. Then we show that reduction in GRTT is simulated in STLC by induction on the typing:

- Case T_TYPE trivial since it is rejected by $\text{STLC}(-)$.

- Case T_VAR does not have a reduction (normal form);

- Case T_ARROW, T_TEN are trivial since they do not satisfy the predicate

- Case T_FUN does not have a reduction (normal form);

- Case T_APP with

$$\frac{(\Delta \mid \sigma_2 \mid \mathbf{0}) \odot \Gamma \vdash t_1 : (x :_{(s,0)} A) \to B \quad (\Delta \mid \sigma_4 \mid \mathbf{0}) \odot \Gamma \vdash t_2 : A}{(\Delta \mid \sigma_2 + s * \mathbf{0} \mid \mathbf{0} + 0 * \sigma_4) \odot \Gamma \vdash t_1 \, t_2 : [t_2/x]B} \text{ T\_APP}$$

  which translates to the derivation:

$$\frac{[\![\Gamma]\!] \vdash [\![t_1]\!] : [\![A]\!] \to [\![B]\!] \quad [\![\Gamma]\!] \vdash [\![t_2]\!] : [\![A]\!]}{[\![\Gamma]\!] \vdash [\![t_1]\!] \, [\![t_2]\!] : [\![B]\!]}$$

  There are then two possible reductions:

  – $(\lambda x.t_1) \, t_2 \leadsto [t_2/x]t_1$ Goal is that $[\![(\lambda x.t_1) \, t_2]\!] \to [\![[t_2/x]t_1]\!]$.
    By the translation the left-hand side is equivalent to:

$$(\lambda x.[\![t_1]\!]) \, [\![t_2]\!]$$

    By $\beta$-reduction in STLC we get

$$(\lambda x.[\![t_1]\!]) \, [\![t_2]\!] \leadsto [[\![t_2]\!]/x][\![t_1]\!]$$

    which via Lemma B.3 provides the goal

  – The other case is the congruence:

$$\frac{t_1 \leadsto t_1'}{t_1 \, t_2 \leadsto t_1' \, t_2} \text{ SEM\_CONGFUNONE}$$

    By induction we have $[\![t_1]\!] \leadsto [\![t_1']\!]$ which using application congruence for STLC yields the reduction in STLC:

$$\frac{[\![t_1]\!] \leadsto [\![t_1']\!]}{[\![t_1]\!][\![t_2]\!] \leadsto [\![t_1']\!][\![t_2]\!]}$$

    which satisfies the goal since $[\![t_1 \, t_2]\!] = [\![t_1]\!][\![t_2]\!]$.

- Case T_PAIR and case T_TENCUT are trivial since we exclude tensor types.

- Case T_TyConv trivial by induction because it introduced no additional syntax.

- Case T_Box trivial as we don't translate type formation rules

- Case T_BoxI trivial since it has no reduction

- Case T_BoxE

$$\dfrac{(\Delta \mid \sigma_1 \mid \sigma_5) \odot \Gamma \vdash t_1 : \Box_s A \qquad (\Delta, \sigma_5 \mid \sigma_4, 0 \mid \mathbf{0}) \odot \Gamma, z : \Box_s A \vdash B : \mathsf{Type}_l \qquad (\Delta, \sigma_5 \mid \sigma_3, s \mid \sigma_4, 0) \odot \Gamma, x : A \vdash t_2 : B}{(\Delta \mid \sigma_1 + \sigma_3 \mid \sigma_4) \odot \Gamma \vdash \mathsf{let}\,\Box x = t_1 \,\mathsf{in}\, t_2 : B} \text{ T\_BoxE}$$

There are two possible case for reduction:

–
$$\dfrac{}{\mathsf{let}\,\Box x = \Box t_1 \,\mathsf{in}\, t_2 \rightsquigarrow [t_1/x]t_2} \text{ Sem\_BetaBox}$$

Thus the reduction here yields the translated term $[\![[t_1/x]t_2]\!]$.

The translation gives that $[\![\mathsf{let}\,\Box x = \Box t_1 \,\mathsf{in}\, t_2]\!] = (\lambda x.[\![t_2]\!])\,[\![t_1]\!]$. Which in the STLC then reduces by $\beta$-reduction to $[[\![t_1]\!]/x][\![t_2]\!]$ which matches the translated term by Lemma B.3.

–
$$\dfrac{t_1 \rightsquigarrow t_1'}{\mathsf{let}\,\Box x = t_1 \,\mathsf{in}\, t_2 \rightsquigarrow \mathsf{let}\,\Box x = t_1' \,\mathsf{in}\, t_2} \text{ Sem\_CongBox1}$$

The translation gives that $[\![\mathsf{let}\,\Box x = t_1 \,\mathsf{in}\, t_2,]\!] = (\lambda x.[\![t_2]\!])\,[\![t_1]\!]$. Thus the reduction here yields the encoded term $(\lambda x.[\![t_2]\!])\,[\![t_1']\!]$.

By induction we have $[\![t_1]\!] \rightsquigarrow [\![t_1']\!]$ from which we can then apply the STLC reduction (right congruence).

$$\dfrac{[\![t_1]\!] \rightsquigarrow [\![t_1']\!]}{(\lambda x.[\![t_2]\!])\,[\![t_1]\!] \rightsquigarrow (\lambda x.[\![t_2]\!])\,[\![t_1']\!]}$$

□

### B.1.4 Completeness

**Lemma B.5** (Completeness). *Given* $\textsc{Stlc}((\Delta \mid \sigma_1 \mid \sigma_2) \odot \Gamma \vdash t : A)$ *and* $[\![(\Delta \mid \sigma_1 \mid \sigma_2) \odot \Gamma \vdash t : A]\!]$ *and* $[\![\Gamma]\!] \vdash [\![t]\!] : [\![A]\!]$ *and* $[\![t]\!] \rightsquigarrow t_a$ *in the full beta STLC then* $t \rightsquigarrow t'$ *and* $[\![t']\!] \equiv_{\beta\eta} t_a$ *in* $\textsc{Grtt}$.

*Proof.* By induction on typing.

- Case T_Type trivial since it is rejected by $\textsc{Stlc}(-)$.

- Case T_Var does not have a reduction (normal form);

- Case T_Arrow, T_Ten are trivial since they do not satisfy the predicate

- Case T_Fun with

$$\dfrac{(\Delta, \sigma_1 \mid \sigma_3, r \mid \mathbf{0}) \odot \Gamma, x : A \vdash B : \mathsf{Type}_l \qquad (\Delta, \sigma_1 \mid \sigma_2, s \mid \sigma_3, r) \odot \Gamma, x : A \vdash t : B}{(\Delta \mid \sigma_2 \mid \sigma_1 + \sigma_3) \odot \Gamma \vdash \lambda x.t : (x :_{(s,r)} A) \to B} \text{ T\_Fun}$$

which translates to the derivation:

$$\dfrac{[\![\Gamma]\!] \vdash [\![t]\!] : [\![B]\!]}{[\![\Gamma]\!] \vdash \lambda x.[\![t]\!] : [\![A]\!] \to [\![B]\!]}$$

We then reduce under the lambda to get:

$$\frac{[\![t]\!] \rightsquigarrow t_a}{\lambda x.[\![t]\!] \rightsquigarrow \lambda x.t_a}$$

By induction on the typing with the premise we then habe the GRTT reduction $t \rightsquigarrow t'$ where $[\![t']\!] = t_a$. Thue we can construct the GRTT reduction:

$$\frac{t \rightsquigarrow t'}{\lambda x.t \rightsquigarrow \lambda x.t'}$$

for which the intereperatation $[\![\lambda x.t']\!] = \lambda x.t'_a$.

- Case T_App with

$$\frac{(\Delta \mid \sigma_2 \mid \mathbf{0}) \odot \Gamma \vdash t_1 : (x :_{(s,0)} A) \to B \quad (\Delta \mid \sigma_4 \mid \mathbf{0}) \odot \Gamma \vdash t_2 : A}{(\Delta \mid \sigma_2 + s * \mathbf{0} \mid \mathbf{0} + 0 * \sigma_4) \odot \Gamma \vdash t_1\, t_2 : [t_2/x]B} \text{ T\_App}$$

which translates to the derivation:

$$\frac{[\![\Gamma]\!] \vdash [\![t_1]\!] : [\![A]\!] \to [\![B]\!] \quad [\![\Gamma]\!] \vdash [\![t_2]\!] : [\![A]\!]}{[\![\Gamma]\!] \vdash [\![t_1]\!]\, [\![t_2]\!] : [\![B]\!]}$$

Then depending on $t_1$ there are three possible reductions:

  - (application left congruence)

    $$\frac{[\![t_1]\!] \rightsquigarrow t_b}{[\![t_1]\!]\, [\![t_2]\!] \rightsquigarrow t_b\, [\![t_2]\!]}$$

    where $t_a = t_b\, [\![t_2]\!]$
    By induction on typing with the premise here then we have that $t_1 \rightsquigarrow t'_1$ and $[\![t'_1]\!] \equiv t_b$. Thus we can construct the GRTT reduction (congruence on left):

    $$\frac{t_1 \rightsquigarrow t'_1}{t_1\, t_2 \rightsquigarrow t'_1\, t_2} \text{ Sem\_CongFunOne}$$

    For which the interpretation $[\![t'_1\, t_2]\!] = [\![t'_1]\!]\, [\![t_2]\!]$
    By congruence on ($\beta\eta$-equality for STLC) with $[\![t'_1]\!] \equiv t_b$ then $[\![t'_1]\!]\, [\![t_2]\!] \equiv t_b\, [\![t_2]\!]$ thus matching our goal

  - (application right congruence)

    $$\frac{[\![t_2]\!] \rightsquigarrow t_b}{[\![t_1]\!]\, [\![t_2]\!] \rightsquigarrow [\![t_1]\!]\, t_b}$$

    where $t_a = [\![t_1]\!]\, t_b$.
    By induction on typing with the premise here then we have that $t_2 \rightsquigarrow t'_2$ and $[\![t'_2]\!] \equiv t_b$. Thus we can construct the GRTT reduction (congruence on right).

  - (beta) The only way for the translation to yield a $\lambda$ abstraction on the left is by translation a $\lambda$, thus we must have $t_1 = \lambda x.t'_1$ and thus:

    $$\overline{[\![\lambda x.t'_1]\!]\, [\![t_2]\!] \rightsquigarrow [[\![t_2]\!]/x][\![t'_1]\!]}$$

    Therefore in GRTT we have $(\lambda x.t'_1)t_2$ which reduces by $\beta$ to $[t_2/x](t'_1)$ whose interpretation $[\![[t_2/x](t'_1)]\!] = [[\![t_2]\!]/x][\![t'_1]\!]$ by Lemma B.3 matching the goal here.

    $\square$

## B.2 Stratified System F

We define the following subset of GRTT typing judgments by the predicate $\text{SSF}(\vdash)$:

$$\text{SSF}((\emptyset \mid \emptyset \mid \emptyset) \odot \emptyset \vdash t : A)$$
$$\text{SSF}((\Delta \mid \sigma_1 \mid \sigma_2) \odot \Gamma \vdash t : A) \implies \text{SSF}((\Delta, \mathbf{0} \mid \sigma_1, 0 \mid \sigma_2, r) \odot \Gamma, x : \text{Type}_l \vdash t : A)$$
$$\text{SSF}((\Delta \mid \sigma_1 \mid \sigma_2) \odot \Gamma \vdash t : A) \implies \text{Type}_l \not\in^{+ve} B \land \text{SSF}((\Delta, \sigma_3 \mid \sigma_1, s \mid \sigma_2, 0) \odot \Gamma, x : B \vdash t : A)$$

By $\text{Type}_l \not\in^{+ve} B$ we mean that $\text{Type}_l$ is not a positive subterm of $B$, so that we avoid higher-order typing terms (like type constructors) which do not exist in SSF.

We give a type directed encoding mapping from typing derivations of GRTT to SSF. Thus given a GRTT derivation of judgment $[\![(\Delta \mid \sigma_1 \mid \sigma_2) \odot \Gamma \vdash t : A]\!]$ we have that $\exists t'$ (an SSF term) such that there is a derivation of judgment $[\![\Gamma]\!] \vdash t' : [\![A]\!]$ in SSF where we interpret $A$ as:

$$[\![x]\!]_\tau = x$$
$$[\![\text{Type}_l]\!]_\tau = \star_l$$
$$[\![(x :_{(0,r)} \text{Type}_l) \to B_\tau]\!] = \forall x : \star_l. [\![B]\!]_\tau$$
$$[\![(x :_{(s,0)} A) \to B]\!]_\tau = [\![A]\!]_\tau \to [\![B]\!]_\tau \qquad (\text{where } \text{Type}_l \not\in^{+ve} A)$$

A few key lemmas are needed to prove this is sound and complete. Which we develop first.

### B.2.1 Key lemmas for soundness of encoding

**Lemma B.6.** *Given a quantitative semiring, if $(\Delta \mid \sigma \mid \mathbf{0}) \odot \Gamma \vdash A : \text{Type}_l$ and $[\![A]\!]$ is defined, then:*

$$\text{fv}(A) = \{y \mid \Gamma[i] = y : B \land \sigma[i] \neq 0\}$$

*Proof.* Given the restriction of $[\![-]\!]$, the only two possibilities are:

- Case T_VAR

$$\frac{\Delta_1, \sigma, \Delta_2 \odot \Gamma_1, x : A, \Gamma_2 \vdash \quad |\Delta_1| = |\Gamma_1|}{(\Delta_1, \sigma, \Delta_2 \mid \mathbf{0}^{|\Delta_1|}, 1, \mathbf{0} \mid \sigma, 0, \mathbf{0}) \odot \Gamma_1, x : A, \Gamma_2 \vdash x : A} \text{T\_VAR}$$

  (where $A = \text{Type}_l$)

  Goal is to prove $\text{fv}(x) = \{y \mid \Gamma[i] = y : B \land \mathbf{0}^{|\Delta_1|}, 1, \mathbf{0}[i] \neq 0\} = \{x\}$.

  Since $x$ is at position $j = |M1|$ then then $\mathbf{0}^{|\Delta_1|}, 1, \mathbf{0}[j] \neq 0$ is true and otherwise $\mathbf{0}^{|\Delta_1|}, 1, \mathbf{0}[i] = 0$ for all other values therefore,

  $$\text{fv}(x) = \{x\}$$

  matching the lemma goal.

- Case T_TYPE

$$\frac{\Delta \odot \Gamma \vdash}{(\Delta \mid \mathbf{0} \mid \mathbf{0}) \odot \Gamma \vdash \text{Type}_l : \text{Type}_{\text{suc } l}} \text{T\_TYPE}$$

  The goal follows since $\neg(\sigma[i] \neq 0)$ for all $i$ thus:

  $$\text{fv}(\text{Type}_l) = \emptyset = \{x \mid \Gamma[i] = y : B \land \sigma[i] \neq 0\}$$

- Case T_ARROW

$$\frac{(\Delta \mid \sigma_1 \mid \mathbf{0}) \odot \Gamma \vdash A : \text{Type}_{l_1} \quad (\Delta, \sigma_1 \mid \sigma_2, r \mid \mathbf{0}) \odot \Gamma, x : A \vdash B : \text{Type}_{l_2}}{(\Delta \mid \sigma_1 + \sigma_2 \mid \mathbf{0}) \odot \Gamma \vdash (x :_{(s,r)} A) \to B : \text{Type}_{l_1 \sqcup l_2}} \text{T\_ARROW}$$

90

By induction on the interpretation, we then have $\llbracket A \rrbracket$ and $\llbracket B \rrbracket$ must be defined. By induction of this lemma on the premises we get:

$$\mathsf{fv}(A) = \{x \mid \Gamma[i] = y : B \wedge \sigma_1[i] \neq 0\}$$
$$\mathsf{fv}(B) = \{x \mid \Gamma, x : A[i] = y : B \wedge \sigma_2, r[i] \neq 0\}$$

Then the goal is:

$$\mathsf{fv}((x :_{(s,r)} A) \to B) = \mathsf{fv}(A) \cup (\mathsf{fv}(B) \setminus \{x\}) = \{x \mid \Gamma[i] = y : B \wedge \sigma_1 + \sigma_2[i] \neq 0\}$$

We get $\mathsf{fv}(B) \setminus \{x\}) = \{x \mid \Gamma[i] = y : B \wedge \sigma_2[i] \neq 0\}$.

Then we can get the above goal if: $(\sigma_1[i] \neq 0 \vee \sigma_2[i] \neq 0) \Leftrightarrow \sigma_1 + \sigma_2[i] \neq 0$.

–

- Left-to-right: Assuming $(\sigma_1[i] \neq 0 \vee (\sigma_2)[i] \neq 0)$
  By contradiction: Assume $(\sigma_1 + \sigma_2)[i] = 0$ then by positivity $\sigma_1[i] = 0$ and $\sigma_2[i] = 0$.
  Thus we can eliminate the assumption with these to get a contradiction.
- Right-to-left: Assuming $(\sigma_1 + \sigma_2)[i] \neq 0$
  The goal is $(\sigma_1[i] \neq 0 \vee (\sigma_2)[i] \neq 0)$ which is equivalent to $\neg(\sigma_1[i] = 0 \wedge \sigma_2[i] = 0)$ thus we go by contradiction:
  Assuming $(\sigma_1[i] = 0 \wedge \sigma_2[i] = 0)$ Therefore by monoid unit we have $(\sigma_1 + \sigma_2)[i] = 0$ which contradicts the assumption hence we by PBC that $(\sigma_1[i] \neq 0 \vee (\sigma_2)[i] \neq 0)$.

$\square$

**Lemma B.7.** *Given a type* $(\Delta \mid \sigma \mid \mathbf{0}) \odot \Gamma \vdash (x :_{(s,0)} A) \to B : C$ *where* $(\Delta \mid \mathbf{0}) \odot \Gamma \vdash C \leq \mathsf{Type}_l$, *then when* $\llbracket B \rrbracket$ *is defined, then* $x \notin \mathsf{fv}(B)$.

*Proof.* By inversion on T_ARROW we have:

> Case.
>
> $$\dfrac{(\Delta \mid \sigma_1 \mid \mathbf{0}) \odot \Gamma \vdash A : \mathsf{Type}_{l_1} \quad (\Delta, \sigma_1 \mid \sigma_2, 0 \mid \mathbf{0}) \odot \Gamma, x : A \vdash B : \mathsf{Type}_{l_2}}{(\Delta \mid \sigma_1 + \sigma_2 \mid \mathbf{0}) \odot \Gamma \vdash (x :_{(s,0)} A) \to B : \mathsf{Type}_{l_1 \sqcup l_2}} \text{ T\_ARROW}$$
>
> Such that $\sigma = \sigma_1 + \sigma_2$.
>
> By Lemma B.6 on the second premise then we have that that $\mathsf{fv}(B) = \{y \mid (\Gamma, x : A)[i] = y : B \wedge (\sigma_2, 0)[i] \neq 0\}$.
>
> Since at the position where $(\Gamma, x : A)[i] = x : A$ then $(\sigma_2, 0)[i] = 0$ therefore $x \notin \mathsf{fv}(B)$.

> Case.
>
> $$\dfrac{(\Delta \mid \sigma \mid \mathbf{0}) \odot \Gamma \vdash (x :_{(s,0)} A) \to B : D \quad (\Delta \mid \mathbf{0}) \odot \Gamma \vdash D \leq C}{(\Delta \mid \sigma \mid \mathbf{0}) \odot \Gamma \vdash (x :_{(s,0)} A) \to B : C} \text{ T\_TYCONV}$$
>
> By ST_TRANS we have $(\Delta \mid \mathbf{0}) \odot \Gamma \vdash D \leq \mathsf{Type}_l$, and therefore our goal holds by induction.

**Lemma B.8.** *Given a quantitative semiring, if* $(\Delta \mid \sigma_1 \mid \sigma_2) \odot \Gamma \vdash t : B$ *and* $[\![B]\!]$ *is defined and for all* $x : A \in \Gamma$ *such that* $[\![A]\!]$ *is defined, then:*

- $\mathsf{fv}(B) = \{x \mid \Gamma[i] = x : A \wedge \sigma_2[i] \neq 0\}$

- $\forall i. 0 \leq i < |\Gamma|.$
  $\mathsf{fv}(\Gamma[i]) = \{y \mid \forall j. \Gamma[i] = x : A \wedge j < i \wedge \Gamma[j] = y : C \wedge \Delta[i][j] \neq 0)\}$

*Proof.*
- Case T_TYPE
$$\dfrac{\Delta \odot \Gamma \vdash}{(\Delta \mid \mathbf{0} \mid \mathbf{0}) \odot \Gamma \vdash \mathsf{Type}_l : \mathsf{Type}_{\mathsf{suc}\ l}} \text{T\_TYPE}$$

By Lemma B.2 on the premise (under the encodability restriction) then we get (3).

Conclusion (1) follows since $\neg \sigma_1[i] \neq 0$ for all $i$ thus:

$$\mathsf{fv}(\mathsf{Type}_l) = \emptyset = \{x \mid \Gamma[i] = x : A \wedge \sigma_2[i] \neq 0\}$$

- Case T_VAR
$$\dfrac{\Delta_1, \sigma, \Delta_2 \odot \Gamma_1, x : A, \Gamma_2 \vdash \quad |\Delta_1| = |\Gamma_1|}{(\Delta_1, \sigma, \Delta_2 \mid \mathbf{0}^{|\Delta_1|}, 1, \mathbf{0} \mid \sigma, 0, \mathbf{0}) \odot \Gamma_1, x : A, \Gamma_2 \vdash x : A} \text{T\_VAR}$$

By Lemma 3.4 we get that $\Delta_1, \sigma \odot \Gamma_1, x : A \vdash$ then by inversion we have that:

$$(\Delta_1 \mid \sigma \mid \mathbf{0}) \odot \Gamma_1 \vdash A : \mathsf{Type}_l$$

Applying, Lemma B.6 then we know that:

$$(*) \quad \mathsf{fv}(A) = \{y \mid \Gamma[i] = y : B \wedge \sigma[i] \neq 0\}$$

We can thus prove the two goals as follows:

1. Goal is $\mathsf{fv}(A) = \{x \mid (\Gamma_1, x : A, \Gamma_2)[i] = y : A' \wedge \sigma, 0[i] \neq 0\}$ which is provided by (*) since for all of $i > |G1|$ we have that we have that $\sigma, 0 = 0$ thus, the goal collapse to (*).

2. This follows by Lemma B.2 applied to the premise then we have the goal:

   $\forall i. 0 \leq i < |\Gamma_1, x : A, \Gamma_2|.$
   $\mathsf{fv}(\Gamma_1, x : A, \Gamma_2[i]) =$
   $\quad \{y \mid \forall j. (\Gamma_1, x : A, \Gamma_2)[i] = x : A \wedge j < i \wedge (\Gamma_1, x : A, \Gamma_2)[j] = y : C \wedge (\Delta_1, \sigma, \Delta_2)[i][j] \neq 0)\}$

- Case T_FUN
$$\dfrac{(\Delta, \sigma_1 \mid \sigma_3, r \mid \mathbf{0}) \odot \Gamma, x : A \vdash B : \mathsf{Type}_l \quad (\Delta, \sigma_1 \mid \sigma_2, s \mid \sigma_3, r) \odot \Gamma, x : A \vdash t : B}{(\Delta \mid \sigma_2 \mid \sigma_1 + \sigma_3) \odot \Gamma \vdash \lambda x.t : (x :_{(s,r)} A) \to B} \text{T\_FUN}$$

By induction we have:

(1) $\mathsf{fv}(B) = \{y \mid \Gamma, x : A[i] = y : C \wedge (\sigma_3, r)[i] \neq 0\}$
(2) $\forall i. 0 \leq i < |\Gamma, x : A|.$
$\quad \mathsf{fv}(\Gamma, x : A[i]) = \{y \mid \forall j. (\Gamma, x : A)[i] = x' : A' \wedge j < i \wedge (\Gamma, x : A)[j] = y' : B' \wedge (\Delta, \sigma_1)[i][j] \neq 0)\}$

Two goals are then:

92

1.
$$\mathsf{fv}((x :_{(s,r)} A) \to B)$$
$$= \mathsf{fv}(A) \cup (\mathsf{fv}(B) \setminus \{x\})$$
$$= \{y \mid \Gamma[i] = y : C \wedge (\sigma_1 + \sigma_3)[i] \neq 0\}$$

By Lemma 3.3 and inversion of well-formedness we have that $(\Delta \mid \sigma_1 \mid \mathbf{0}) \odot \Gamma \vdash A : \mathsf{Type}_l$ which we apply Lemma B.6 to, yielding:

$$\mathsf{fv}(A) = \{y \mid \Gamma[i] = y : B \wedge \sigma_1[i] \neq 0\}$$

From (1) we can get $\mathsf{fv}(B) = \{y \mid \Gamma, x : A[i] = y : C \wedge (\sigma_3, r)[i] \neq 0\} \setminus \{x\} = \{y \mid \Gamma[i] = y : C \wedge (\sigma_3)[i] \neq 0\}$.

We then prove the goal by showing $((\sigma_1[i] \neq 0) \vee (\sigma_3[i] \neq 0)) \Leftrightarrow (\sigma_1 + \sigma_3)[i] \neq 0$. Which follows by positivity (see above arguments for similar reasoning).

2. And the second goal follows by the inductive hypothesis (2) removing $\{x\}$.

- Case T_App

$$\frac{(\Delta, \sigma_1 \mid \sigma_3, r \mid \mathbf{0}) \odot \Gamma, x : A \vdash B : \mathsf{Type}_l \quad (\Delta \mid \sigma_2 \mid \sigma_1 + \sigma_3) \odot \Gamma \vdash t_1 : (x :_{(s,r)} A) \to B \quad (\Delta \mid \sigma_4 \mid \sigma_1) \odot \Gamma \vdash t_2 : A}{(\Delta \mid \sigma_2 + s * \sigma_4 \mid \sigma_3 + r * \sigma_4) \odot \Gamma \vdash t_1 \, t_2 : [t_2/x]B} \; \text{T\_App}$$

By induction we have:

(A1) $\mathsf{fv}((x :_{(s,r)} A) \to B) = \{x \mid \Gamma[i] = x : A \wedge (\sigma_1 + \sigma_3)[i] \neq 0\}$

(A2) $\forall i. 0 \leq i < |\Gamma|$.
$\mathsf{fv}(\Gamma[i]) = \{y \mid \forall j. \Gamma[i] = x : A \wedge j < i \wedge \Gamma[j] = y : B \wedge \Delta[i][j] \neq 0)\}$

(B1) $\mathsf{fv}(A) = \{x \mid \Gamma[i] = x : A \wedge \sigma_1[i] \neq 0\}$

(B2) (same as (A3))

The second goal follows from A2 or B2 trivially.

To get the first goal, there are two cases depending on the grades:

- $r = 0$
  Goal is $\mathsf{fv}([t_2/x]B) == \{y \mid \Gamma[i] = y : C \wedge (\sigma_3)[i] \neq 0\}$ since $\sigma_3 + 0 * \sigma_4 = \sigma_3$ by absorption.
  By Lemma B.7 then we know that $x \notin \mathsf{fv}(B)$, therefore $\mathsf{fv}([t_2/x]B) = \mathsf{fv}(B)$, which follows by inversion on the typing (and since we cannot do promotion here under the translation) (Lemma 3.6 with no box).

- $s = 0$. And in this case $A = \mathsf{Type}_l$ therefore $t_2$ is a type term $A'$ subject to the interpretation thus we apply Lemma B.6 to it, yielding:

$$\mathsf{fv}(A) = \{y \mid \Gamma[i] = y : B \wedge \sigma_4[i] \neq 0\}$$

By inversion on the first premise (Lemma 3.6) and since we have no promotion then we get:

$$\mathsf{fv}(B) = \{y \mid \Gamma[i] = y : B \wedge \sigma_3[i] \neq 0\}$$

Now the goal is that $\mathsf{fv}([A/x]B) = \{x \mid \Gamma[i] = x : A \wedge \sigma_3 + r * \sigma_4[i] \neq 0\}$.

We consider two cases depending on whether $r$ is 0 or not.

$r = 0$ $\mathsf{fv}([[A/x]B]) = \{x \mid \Gamma[i] = x : A \wedge \sigma_3[i] \neq 0\}$ and by Lemma B.7 and the same reasoning about for $r = 0$ (previous case) then we have the goal.

$r \neq 0$

We can prove the goal by showing that $(\sigma_3 + r * \sigma_4)[i] \neq 0 \leftrightarrow (\sigma_3[i] \neq 0) \vee (\sigma_4[i] \neq 0)$

* Going left-to-right. Assume: $(\sigma_3 + r * \sigma_4)[i] \neq 0$.
  The goal is equivalent to $\neg(\sigma_3[i] = 0 \land \sigma_4[i] = 0)$ (De Morgan's). Then we prove by contradiction by assuming $\sigma_4[i] = 0 \land \sigma_4[i] = 0$ By absorption then $(r * \sigma_3)[i] = 0$. Then by monoid unit we have $(r * \sigma_3 + \sigma_4)[i] = 0$. which contradicts the premise. Thus we have $(\sigma_3[i] \neq 0) \lor (\sigma_4[i] \neq 0)$.
  If $\sigma_3[i] = 0$ then $r * \sigma_4[i] = 0$ by absorb therefore
* Going right-to-right. Assume: $(\sigma_3[i] \neq 0) \lor (\sigma_4[i] \neq 0)$
  Proof by contradiction. Assume $(\sigma_3 + r * \sigma_4)[i] = 0$ then from positivity we have $\sigma_3[i] = 0$ and $(r * \sigma_4)[i] = 0$ and from zero-product we have that $r[i] = 0 \lor \sigma_4[i] = 0$.
  We eliminate the assumption via:
  · $(\sigma_3[i] \neq 0)$ this contradicts $\sigma_3[i] = 0$.
  · $(\sigma_4[i] \neq 0)$ since $r[i] \neq 0$ then we must have $\sigma_4[i] = 0$ which contradicts this assumption.
  Thus we have reached falsum, and so we prove the goal by contradiction.

□

**Corollary B.8.1.** *For a quantitative semiring, with $[\![B]\!]$ defined then $(\Delta \mid \sigma_1 \mid \mathbf{0}) \odot \Gamma \vdash t : B$ implies that $B$ is closed:*

**Corollary B.8.2.** *For a quantitaitve semiring, with $[\![A]\!]$ define then $(\Delta, \mathbf{0}, \Delta' \mid \sigma_1 \mid \sigma_2) \odot \Gamma, x : A, \Gamma' \vdash t : B$ implies that $A$ is closed:*

### B.2.2 Interpretation of well-formedness

* Case WF_EMPTY

$$\frac{}{\emptyset \odot \emptyset \vdash} \text{WF\_EMPTY}$$

We can thus construct the SSF well-formedness derivation:

$$\frac{}{\cdot \ Ok}$$

* Case WF_EXT

$$\frac{(\Delta \mid \sigma \mid \mathbf{0}) \odot \Gamma \vdash A : \mathsf{Type}_l}{\Delta, \sigma \odot \Gamma, x : A \vdash} \text{WF\_EXT}$$

By Lemma 3.3 then we have $\Delta \odot \Gamma \vdash$ and thus by induction we have $[\![\Gamma]\!]\ Ok$.

– Then if $A = \mathsf{Type}_l$ then we can derfine the derivation:

$$\frac{[\![\Gamma]\!]\ Ok}{[\![\Gamma, x : \star_l]\!]\ Ok}$$

– Otherwise, we apply the main encoding to get some term $T$ such that $[\![\Gamma]\!] \vdash T : \star_l$ Thus we can build the derivation:

$$\frac{[\![\Gamma]\!] \vdash T : \star_l \quad [\![\Gamma]\!]\ Ok}{[\![\Gamma]\!], x : T\ Ok}$$

### B.2.3 Interpretation of derivations

* Case T_TYPE

$$\frac{\Delta \odot \Gamma \vdash}{(\Delta \mid \mathbf{0} \mid \mathbf{0}) \odot \Gamma \vdash \mathsf{Type}_l : \mathsf{Type}_{\mathsf{suc}\ l}} \text{T\_TYPE}$$

This has no analogue in SSF since $\star_l$ is not a type itself in SSF.

- Case T_VAR

$$\dfrac{\Delta_1, \sigma, \Delta_2 \odot \Gamma_1, x : A, \Gamma_2 \vdash \quad |\Delta_1| = |\Gamma_1|}{(\Delta_1, \sigma, \Delta_2 \mid \mathbf{0}^{|\Delta_1|}, 1, \mathbf{0} \mid \sigma, 0, \mathbf{0}) \odot \Gamma_1, x : A, \Gamma_2 \vdash x : A} \text{ T\_VAR}$$

By the interepretation of well-formedness we have that $[\![\Gamma_1]\!], x : [\![A]\!], [\![\Gamma_2]\!]$ $Ok$. Then we form the derivation:

$$\dfrac{[\![\Gamma_1]\!], x : [\![A]\!], [\![\Gamma_2]\!] \ Ok}{[\![\Gamma_1]\!], x : [\![A]\!], [\![\Gamma_2]\!] \vdash x : [\![A]\!]}$$

- Case T_ARROW

$$\dfrac{(\Delta \mid \sigma_1 \mid \mathbf{0}) \odot \Gamma \vdash A : \mathsf{Type}_{l_1} \quad (\Delta, \sigma_1 \mid \sigma_2, r \mid \mathbf{0}) \odot \Gamma, x : A \vdash B : \mathsf{Type}_{l_2}}{(\Delta \mid \sigma_1 + \sigma_2 \mid \mathbf{0}) \odot \Gamma \vdash (x :_{(s,r)} A) \to B : \mathsf{Type}_{l_1 \sqcup l_2}} \text{ T\_ARROW}$$

We have two cases depending on $A$.

  - $A = \mathsf{Type}_l$ then by induction we have SSF term $B'$ such that:

    $$[\![\Gamma]\!] \vdash \star_l : \star_{l_1} \tag{ih1}$$
    $$[\![\Gamma]\!], x : \star_l \vdash B' : \star_{l_2} \tag{ih2}$$

    where $l < l_1$ by the GRTT universes. We can thus form the following SSF derivation:

    $$\dfrac{(\text{ih2})}{[\![\Gamma]\!] \vdash \forall (X : \star_l).B' : \star_{(l+1) \sqcup l_2}}$$

    and then by either $(l+1) \sqcup l_2 = l_1 \sqcup l_2$ and we are done, or if $l + 1 \leq l_1$ then by monotonicity we have $(l+1) \sqcup l_2 \leq l_1 \sqcup l_2$ and we use Eades and Stump [20] (Lemma 3) to get:

    $$\dfrac{(\text{ih2})}{[\![\Gamma]\!] \vdash \forall (X : \star_l).B' : l_1 \sqcup l_2}$$

  - $\mathsf{Type}_l \not\in^{+ve} A$
    By the Corlloary B.8.2 and the definition of the $\mathrm{SSF}(-)$, we know that $B$ (and in this case $r = 0$) here must in fact be closed, so we can *strengthen* then premise to:

    $$(\Delta \mid \sigma_2 \mid \mathbf{0}) \odot \Gamma \vdash B : \mathsf{Type}_{l_2}$$

    Then by induction we have SSF terms $A'$ and $B'$ such that:

    $$[\![\Gamma]\!] \vdash A' : \star_{l_1} \tag{ih1}$$
    $$[\![\Gamma]\!] \vdash B' : \star_{l_2} \tag{ih2}$$

    We can thus form the following SSF derivation:

    $$\dfrac{(\text{ih1}) \quad (\text{ih2})}{[\![\Gamma]\!] \vdash A' \to B' : \star_{l_1 \sqcup l_2}}$$

- Case T_TEN trivial since we ignore products

- Case T_FUN

$$\dfrac{(\Delta, \sigma_1 \mid \sigma_3, r \mid \mathbf{0}) \odot \Gamma, x : A \vdash B : \mathsf{Type}_l \quad (\Delta, \sigma_1 \mid \sigma_2, s \mid \sigma_3, r) \odot \Gamma, x : A \vdash t : B}{(\Delta \mid \sigma_2 \mid \sigma_1 + \sigma_3) \odot \Gamma \vdash \lambda x.t : (x :_{(s,r)} A) \to B} \text{ T\_FUN}$$

We define two cases depending on $A$:

- $\text{Type}_l = A$

$$\frac{(\Delta, \sigma_1 \mid \sigma_2, s \mid \sigma_3, r) \odot \Gamma, x : \text{Type}_l \vdash t : B}{(\Delta \mid \sigma_2 \mid \sigma_1 + \sigma_3) \odot \Gamma \vdash \lambda x.t : (x :_{(0,r)} \text{Type}_l) \to B} \text{ T\_FUN}$$

By induction on the premise we have some SSF term $t'$ such that:

$$[\![\Gamma]\!], x : \star_l \vdash t' : [\![B]\!]$$

Then we construct the following derivation in SSF:

$$\frac{[\![\Gamma]\!], x : \star_l \vdash t' : [\![B]\!]}{[\![\Gamma]\!] \vdash \Lambda(x : \star_l).t' : \forall x : \star_l.[\![B]\!]}$$

- $\text{Type}_l \not\in^{+ve} A$ (no positive occurences of $\text{Type}_l$)

$$\frac{(\Delta, \sigma_1 \mid \sigma_2, s \mid \sigma_3, r) \odot \Gamma, x : A \vdash t : B}{(\Delta \mid \sigma_2 \mid \sigma_1 + \sigma_3) \odot \Gamma \vdash \lambda x.t : (x :_{(0,r)} A) \to B} \text{ T\_FUN}$$

By induction on the premise we have that exists some SSF term $t'$ such that:

$$[\![\Gamma]\!], x : [\![A]\!] \vdash t' : [\![B]\!]$$

Then we construct the following derivation in SSF:

$$\frac{[\![\Gamma]\!], x : [\![A]\!] \vdash t' : [\![B]\!]}{[\![\Gamma]\!] \vdash \lambda(x : [\![A]\!]).t' : [\![A]\!] \to [\![B]\!]}$$

- Case T\_APP

$$\frac{\begin{array}{c}(\Delta, \sigma_1 \mid \sigma_3, r \mid \mathbf{0}) \odot \Gamma, x : A \vdash B : \text{Type}_l \\ (\Delta \mid \sigma_2 \mid \sigma_1 + \sigma_3) \odot \Gamma \vdash t_1 : (x :_{(s,r)} A) \to B \quad (\Delta \mid \sigma_4 \mid \sigma_1) \odot \Gamma \vdash t_2 : A\end{array}}{(\Delta \mid \sigma_2 + s * \sigma_4 \mid \sigma_3 + r * \sigma_4) \odot \Gamma \vdash t_1 \, t_2 : [t_2/x]B} \text{ T\_APP}$$

We define two cases depending on $A$:

- $\text{Type}_l = A$ By induction we have terms $t'$ and $T'$ such that:

$$[\![\Gamma]\!] \vdash t' : \forall(x : \star_l).[\![B]\!] \tag{3}$$
$$[\![\Gamma]\!] \vdash T : \star_l \tag{4}$$

Then we can form the SSF derivation:

$$\frac{ih1 \quad ih2}{[\![\Gamma]\!] \vdash t'[T] : [T/x][\![B]\!]}$$

where $[T/x][\![B]\!]$ matches $[\![[t_2/x]B]\!]$ by Lemma B.9.

- $\text{Type}_l \not\in^{+ve} A$ (no positive occurences of $\text{Type}_l$) thus $r = 0$.
By induction we have two terms $t'_1$ and $t'_2$

$$[\![\Gamma]\!] \vdash t'_1 : [\![A]\!] \to [\![B]\!] \tag{5}$$
$$[\![\Gamma]\!] \vdash t'_2 : [\![A]\!] \tag{6}$$

Then we can form the SSF derivation:

$$\frac{(\text{ih1}) \quad (\text{ih2})}{[\![\Gamma]\!] \vdash t'_1 \, t'_2 : [\![B]\!]}$$

which is equal to the goal here by since by Lemma B.7 since $[t_2/x]B = B$.

- Case T\_PAIR and case T\_TENCUT are trivial since we exclude tensor types.

- Case T\_BOX, T\_BOXI, T\_BOXE trivial as we don't translate graded modalities

### B.2.4 Soundness

We first need some further auxiliary lemmas

**Lemma B.9** (Preservation for substitution on terms).

$$[\![[A/x]B]\!] = [[\![A]\!]/x][\![B]\!]$$

*Proof.* By induction on the syntax:

- $[\![[A/x]\mathsf{Type}_l]\!] = [\![\mathsf{Type}_l]\!] = \star_l = [[\![A]\!]/x]\star_l$

- $[\![[A/x]((y :_{(0,r)} \mathsf{Type}_l) \to B)]\!] = [\![((y :_{(0,r)} \mathsf{Type}_l) \to [A/x]B)]\!] = \forall x : \star_l.[\![[A/x]B]\!] = \forall x : \star_l.[[\![A]\!]/x][\![B]\!] = [[\![A]\!]/x](\forall x : \star_l.[\![B]\!])$ (penulimate step by induction)

- $[\![[A/x]((y :_{(s,0)} A') \to B)]\!] = [\![((y :_{(s,0)} [A/x]A') \to [A/x]B)]\!] = [\![[A/x]A']\!] \to [\![[A/x]B]\!] = [[\![A]\!]/x][\![A']\!] \to [[\![A]\!]/x][\![B]\!] = [[\![A]\!]/x]([\![A']\!] \to [\![B]\!])$ (penultimate step is by induction)
  (*where* $\mathsf{Type}_l \not\in^{+ve} A$)

$\square$

**Lemma B.10** (Substitution soundness). *For all* $\Delta, \sigma_1, \sigma_2, \sigma_3, \sigma_4, t_1, t_2, l, B, \Gamma, T, y, t_0$, *where* $A \neq \mathsf{Type}_l$, *if:*

$$[\![(\Delta \mid \sigma_4 \mid \sigma_1) \odot \Gamma \vdash t_2 : A]\!] = [\![\Gamma]\!] \vdash t_b : [\![A]\!] \quad (1)$$

$$[\![(\Delta, \sigma_1 \mid \sigma_2, s \mid \sigma_3, r) \odot \Gamma, y : A \vdash t_1 : B]\!] = [\![\Gamma]\!], y : [\![A]\!] \vdash t_a : B \quad (2)$$

*then*

$$[\![(\Delta \mid \sigma_2 + s * \sigma_4 \mid \sigma_3 + r * \sigma_4) \odot \Gamma \vdash [t_2/y]t_1 : [t_2/y]B]\!] = [\![\Gamma]\!] \vdash [t_b/y]t_a : [\![B]\!]$$

*Proof.* By induction on the typing of $t_1$ under the restriction of the translation

- Case T_VAR. Depends on whether the variable $x$ introduced by the T_VAR is equal to $y$ or not:

  - $y = x$ then $t_1 = x$ and $[t_2/y]t_1 = t_2$ the interpretation of which is $t_b$ by premise (1).
    In this case then $B = A$ (from that fact that $y = x$). By the encoding then $t_a = y$ and thus $[t_b/y]t_a = t_b$ satisfying the goal here.
  - $y \neq x$ then $t_1 = y$ and $[t_2/y]t_1 = y = x$.

- Case T_ARROW thus $t_1 = (x :_{(s,r)} A') \to B'$ with:

$$\frac{(\Delta, \sigma_1 \mid \sigma'_1 \mid \mathbf{0}) \odot \Gamma, y : A \vdash A' : \mathsf{Type}_{l_1} \quad (\Delta, \sigma_1, \sigma'_1 \mid \sigma'_2, r' \mid \mathbf{0}) \odot \Gamma, y : A, x : A' \vdash B' : \mathsf{Type}_{l_2}}{(\Delta, \sigma_1 \mid \sigma'_1 + \sigma'_2 \mid \mathbf{0}) \odot \Gamma, y : A \vdash (x :_{(s',r')} A') \to B' : \mathsf{Type}_{l_1 \sqcup l_2}} \text{ T\_ARROW}$$

  where $\sigma_2, s = \sigma'_1 + \sigma'_2$ and $\mathbf{0} = \sigma_3, \sigma$.

  - Case $A' = \mathsf{Type}_{l'}$
    Thus, $t_a = \forall(X : \star_{l'}).B''$ by the interpretation on function types.
    By induction, then we know:

$$[\![[t_2/y]\mathsf{Type}_l]\!] = [t_b/y]\star_{l'}$$
$$[\![[t_2/y]B']\!] = [t_b/y]B''$$

  Therefore, $[\![[t_2/y]((x :_{(s',r')} A') \to B')]\!] = [\![(x :_{(s',r')} [t_2/y]A') \to [t_2/y]B']\!] = \forall(X : \star_{l'}).[t_b/y]B'' = [t_b/y](\forall(X : \star_{l'}).B'')$ Satisfying the goal here.

97

- Case $\mathsf{Type}_l \not\in^{+ve} B$
  
  Thus, $t_a = A'' \to B''$ by the interprettion.
  
  By induction we know:
  $$[\![[t_2/y]A']\!] = [t_b/y]A''$$
  $$[\![[t_2/y]B']\!] = [t_b/y]B''$$
  
  Therefore, $[\![[t_2/y]((x :_{(s',r')} A') \to B')]\!] = [\![(x :_{(s',r')} [t_2/y]\mathsf{Type}_l) \to [t_2/y]B']\!] = [t_b/y]A'' \to [t_b/y]B''$
  $= [t_b/y](A'' \to B'')$ Satisfying the goal here.

- Case T_FUN We define two cases depending on $A'$:
  
  - $A' = \mathsf{Type}_{l'}$
  $$\dfrac{(\Delta, \sigma_1' \mid \sigma_2', s' \mid \sigma_3', r') \odot \Gamma, y : A, x : \mathsf{Type}_{l'} \vdash t : B'}{(\Delta \mid \sigma_2' \mid \sigma_1' + \sigma_3') \odot \Gamma, y : A \vdash \lambda x.t : (x :_{(0,r)} \mathsf{Type}_{l'}) \to B'} \text{ T\_FUN}$$
  
  By induction on the premise we have
  $$[\![[t_2/y]t]\!] = [t_b/y]t_a'$$
  
  Therefore $[\![[t_2/y](\lambda x.t)]\!] = \Lambda(x : \star_{l'}).[t_2/y]t_a' = [t_2/y](\Lambda(x : \star_{l'}).t_a')$ satisfying the goal.
  
  - $\mathsf{Type}_{l'} \not\in^{+ve} A$ (no positive occurences of $\mathsf{Type}_{l'}$)
  $$\dfrac{(\Delta, \sigma_1, \sigma_1' \mid \sigma_2', s, s' \mid \sigma_3', r, r') \odot \Gamma, y : A, x : A' \vdash t : B'}{(\Delta, \sigma_1 \mid \sigma_2', s \mid \sigma_1' + \sigma_3', r') \odot \Gamma, y : A \vdash \lambda x.t : (x :_{(0,r)} A') \to B'} \text{ T\_FUN}$$
  
  Let $\sigma_5, r, \sigma_6 = \sigma_1' + (\sigma_3'', r)$
  
  By induction on the premise we have that exists some SSF term $t'$ such that:
  $$[\![\begin{pmatrix} \Delta, (\emptyset, \sigma_1' \backslash \pi + (\emptyset, \sigma_1'/\pi) * \sigma_4) \\ (\sigma_2' + s * \sigma_4), s' \\ (\sigma_3' + r * \sigma_4), \sigma_6 \end{pmatrix}]\!] \odot \Gamma, y : \mathsf{Type}_l, x : [t_2/y]A' \vdash [t_2/y]t : [t_2/y]B']\!] = [t_a/y]t_0'$$
  
  Therefore
  $$[\![[t_2/y](\lambda x.t)]\!] = \lambda(x : [\![[t_2/y]A']\!]).[t_a/y]t_0'$$
  
  Then by Lemma B.9 this is equal to:
  $$(\lambda(x : [t_a/y][\![A']\!]).[t_a/y]t_0')$$
  $$=[t_a/y](\lambda(x : [\![A']\!]).t_0')$$
  
  satisfying the goal.

- Case T_APP
  $$\dfrac{(\Delta \mid \sigma_2' \mid \sigma_1' + \sigma_3') \odot \Gamma, y : A \vdash t_1' : (x :_{(s',r')} A') \to B' \quad (\Delta \mid \sigma_4' \mid \sigma_1') \odot \Gamma, y : A \vdash t_2' : A'}{(\Delta \mid \sigma_2' + s' * \sigma_4' \mid \sigma_3' + r' * \sigma_4') \odot \Gamma, y : A \vdash t_1' \, t_2' : [t_2/x]B'} \text{ T\_APP}$$
  
  We define two cases depending on $A'$:
  
  - $A' = \mathsf{Type}_l$
  
  Thus $t_a = t_a'[T']$ (where the intermrediate of $t_2'$ is interpreted as to $T'$)
  
  By induction, then we know:
  $$[\![[t2/y]t1']\!] = [t_b/y]T'$$
  $$[\![[t2/y]t2']\!] = [t_b/y]t_a'$$
  
  Therefore
  $$[\![[t_2/y](t_1' \, t_2')]\!] = [t_b/y]t_a'[[t_b/y]T'] = [t_b/y](t_a'[T'])$$
  
  satisfying the goal.

- $\mathsf{Type}_{l'} \not\in^{+ve} A'$ (no positive occurences of $\mathsf{Type}_{l'}$)
  Thus $t_a = t'_a \, t''_a$.
  By induction we know that:
  $$[\![[t_2/y]t1']\!] = [t_b/y]t'_a$$
  $$[\![[t_2/y]t2']\!] = [t_b/y]t''_a$$
  Therefore
  $$[\![[t_2/y](t'_1 \, t'_2)]\!] = [t_b/y]t'_a \, [t_b/y]t''_a = [t_b/y](t'_a \, t''_a)$$
  satisfying the goal.

$\square$

**Lemma B.11** (Type substitution soundness)**.** *For all* $\Delta, \sigma_1, \sigma_2, \sigma_3, \sigma_4, t_1, t_2, l, B, \Gamma, T, y, t_0$, *if:*

$$[\![(\Delta \mid \sigma_4 \mid \sigma_1) \odot \Gamma \vdash t_2 : \mathsf{Type}_l]\!] = [\![\Gamma]\!] \vdash T : \star_l \tag{1}$$
$$[\![(\Delta, \sigma_1 \mid \sigma_2, s \mid \sigma_3, r) \odot \Gamma, y : \mathsf{Type}_l \vdash t_1 : B]\!] = [\![\Gamma]\!], y : \star_l \vdash t_0 : [\![B]\!] \tag{2}$$

*then:*

$$[\![(\Delta \mid \sigma_2 + s \ast \sigma_4 \mid \sigma_3 + r \ast \sigma_4) \odot \Gamma \vdash [t_2/y]t_1 : [t_2/y]B]\!] = [\![\Gamma]\!] \vdash [T/y]t_0 : [T/y][\![B]\!]$$

*Proof.* By induction on the typing of $t_1$ under the restriction of the translation

- Case T_VAR. Depends on whether the variable $x$ introduced by the T_VAR is equal to $y$ or not:
  - $y = x$ then $t_1 = x$ and $[t_2/y]t_1 = t_2$ the interpretation of which is $T$ by premise (1).
    In this case then $B = \mathsf{Type}_l$ (from that fact that $y = x$) and thus (2) is a kinding derivation of SSF.
    By the encoding then $t_0 = y$ and thus $[T/y]t_0 = T$ satisfying the goal here.
  - $y \neq x$ then $t_1 = y$ and $[t_2/y]t_1 = y = x$.
    By the encoding then $t_0 = x$ and thus $[T/y]t_0 = x$ satisfying the goal here.

- Case T_ARROW thus $t_1 = (x :_{(s,r)} A') \to B'$ with:

  $$\frac{\begin{array}{c}(\Delta, \sigma_1 \mid \sigma'_1 \mid \mathbf{0}) \odot \Gamma, y : \mathsf{Type}_l \vdash A' : \mathsf{Type}_{l_1} \\ (\Delta, \sigma_1, \sigma'_1 \mid \sigma'_2, r' \mid \mathbf{0}) \odot \Gamma, y : \mathsf{Type}_l, x : A' \vdash B' : \mathsf{Type}_{l_2}\end{array}}{(\Delta, \sigma_1 \mid \sigma'_1 + \sigma'_2 \mid \mathbf{0}) \odot \Gamma, y : \mathsf{Type}_l \vdash (x :_{(s',r')} A') \to B' : \mathsf{Type}_{l_1 \sqcup l_2}} \text{ T\_ARROW}$$

  where $\sigma_2, s = \sigma'_1 + \sigma'_2$ and $\mathbf{0} = \sigma_3, \sigma$.

  - Case $A' = \mathsf{Type}_{l'}$
    Thus, $t_0 = \forall (X : \star_{l'}).B''$ by the interpretation on function types.
    By induction, then we know:
    $$[\![[t_2/y]\mathsf{Type}_l]\!] = [T/y]\star_{l'}$$
    $$[\![[t_2/y]B']\!] = [T/y]B''$$
    
    Therefore, $[\![[t_2/y]((x :_{(s',r')} A') \to B')]\!] = [\![(x :_{(s',r')} [t_2/y]\mathsf{Type}_l) \to [t_2/y]B']\!] = \forall (X : \star_{l'}).[T/y]B''$
    $= [T/y](\forall (X : \star_{l'}).B'')$ Satisfying the goal here.

- Case $\mathsf{Type}_l \not\in^{+ve} B$

  Thus, $t_0 = A'' \to B''$ by the interpretation.

  By induction we know:
  $$[\![[t_2/y]A']\!] = [T/y]A''$$
  $$[\![[t_2/y]B']\!] = [T/y]B''$$

  Therefore, $[\![[t_2/y]((x :_{(s',r')} A') \to B')]\!] = [\![(x :_{(s',r')} [t_2/y]\mathsf{Type}_l) \to [t_2/y]B']\!] = [T/y]A'' \to [T/y]B''$
  $= [T/y](A'' \to B'')$ Satisfying the goal here.

- Case T_FUN We define two cases depending on $A'$:

  - $A' = \mathsf{Type}_{l'}$

  $$\frac{(\Delta, \sigma_1' \mid \sigma_2', s' \mid \sigma_3', r') \odot \Gamma, y : \mathsf{Type}_l, x : \mathsf{Type}_{l'} \vdash t : B'}{(\Delta \mid \sigma_2' \mid \sigma_1' + \sigma_3') \odot \Gamma, y : \mathsf{Type}_l \vdash \lambda x.t : (x :_{(0,r)} \mathsf{Type}_{l'}) \to B'} \text{ T\_FUN}$$

  By induction on the premise we have
  $$[\![[t_2/y]t]\!] = [T/y]t_0'$$

  Therefore $[\![[t_2/y](\lambda x.t)]\!] = \Lambda(x : \star_{l'}).[T/y]t0' = [T/y](\Lambda(x : \star_l).t0')$ satisfying the goal.

  - $\mathsf{Type}_l \not\in^{+ve} A$ (no positive occurences of $\mathsf{Type}_l$)

  $$\frac{(\Delta, \sigma_1, \sigma_1' \mid \sigma_2', s, s' \mid \sigma_3', r, r') \odot \Gamma, y : \mathsf{Type}_l, x : A' \vdash t : B'}{(\Delta, \sigma_1 \mid \sigma_2', s \mid \sigma_1' + \sigma_3', r') \odot \Gamma, y : \mathsf{Type}_l \vdash \lambda x.t : (x :_{(0,r)} A') \to B'} \text{ T\_FUN}$$

  Let $\sigma_5, r, \sigma_6 = \sigma_1' + (\sigma_3'', r)$

  By induction on the premise we have that exists some SSF term $t'$ such that:

  $$\left[\!\!\left[\begin{array}{c}\Delta, (\emptyset, \sigma_1' \backslash \pi + (\emptyset, \sigma_1'/\pi) * \sigma_4) \\ (\sigma_2' + s * \sigma_4), s' \\ (\sigma_3' + r * \sigma_4), \sigma_6\end{array}\right) \odot \Gamma, y : \mathsf{Type}_l, x : [t_2/y]A' \vdash [t_2/y]t : [t_2/y]B'\right]\!\!\right] = [T/y]t_0'$$

  Therefore
  $$[\![[t_2/y](\lambda x.t)]\!] = \lambda(x : [\![[t_2/y]A']\!]).[T/y]t_0'$$

  Then by Lemma B.9 this is equal to:
  $$(\lambda(x : [T/y][\![A']\!]).[T/y]t_0')$$
  $$= [T/y](\lambda(x : [\![A']\!]).t_0')$$

  satisfying the goal.

  - Case T_APP

  $$\frac{(\Delta \mid \sigma_2 \mid \sigma_1 + \sigma_3) \odot \Gamma \vdash t_1' : (x :_{(s,r)} A') \to B \quad (\Delta \mid \sigma_4 \mid \sigma_1) \odot \Gamma \vdash t_2' : A'}{(\Delta \mid \sigma_2 + s * \sigma_4 \mid \sigma_3 + r * \sigma_4) \odot \Gamma \vdash t_1' t_2' : [t_2/x]B} \text{ T\_APP}$$

  We define two cases depending on $A$:

  * $A' = \mathsf{Type}_l$

    Thus $t_0 = t'[T']$ by induction.

    By induction, then we know:
    $$[\![[t2/y]t1']\!] = [T/y]t'$$
    $$[\![[t2/y]t2']\!] = [T/y]T'$$

    Therefore
    $$[\![[t_2/y](t_1' t_2')]\!] = [T/y]t'[[T/y]T'] = [T/y](t'[T'])$$

    satisfying the goal.

* $\mathsf{Type}_{l'} \not\in^{+ve} A'$ (no positive occurences of $\mathsf{Type}_{l'}$)
    Thus $t_0 = t'_1 \ t'_2$.
    By induction we know that:
    $$[\![[t_2/y]t1']\!] = [T/y]t'_1$$
    $$[\![[t_2/y]t2']\!] = [T/y]t'_2$$

Therefore
$$[\![[t_2/y](t'_1 \ t'_2)]\!] = [T/y]t'_1 \ [T/y]t'_2 = [T/y](t'_1 \ t'_2)$$

satisfying the goal.

$\square$

**Lemma B.12** (Soundness). *Given* $\mathrm{SSF}((\Delta \mid \sigma_1 \mid \sigma_2) \odot \Gamma \vdash t_1 : A)$ *and* $[\![(\Delta \mid \sigma_1 \mid \sigma_2) \odot \Gamma \vdash t_1 : A]\!] = [\![\Gamma]\!] \vdash t' : [\![A]\!]$ *in* GRTT *and* $t_1 \rightsquigarrow t'_1$ *and* $[\![(\Delta \mid \sigma_1 \mid \sigma_2) \odot \Gamma \vdash t'_1 : A]\!] = [\![\Gamma]\!] \vdash t'' : [\![A]\!]$ *then* $t' \rightsquigarrow t''$ *in the CBN SSF*

*Proof.* First we observe that by Type Preservation (Lemma 3.35) that reduction preserves typing thus if $\mathrm{SSF}(-)$ holds for $[\![(\Delta \mid \sigma_1 \mid \sigma_2) \odot \Gamma \vdash t_1 : A]\!]$ then $\mathrm{SSF}((\Delta \mid \sigma_1 \mid \sigma_2) \odot \Gamma \vdash t'_1 : A)$. Then we show that reduction in GRTT is simulated in SSF by induction on the typing:

- Case T_Type, T_Var, T_Arrow, are all trivial as they already in normal forms (no reduction)

- Case T_Fun.

$$\frac{(\Delta, \sigma_1 \mid \sigma_3, r \mid \mathbf{0}) \odot \Gamma, x : A \vdash B : \mathsf{Type}_l \quad (\Delta, \sigma_1 \mid \sigma_2, s \mid \sigma_3, r) \odot \Gamma, x : A \vdash t : B}{(\Delta \mid \sigma_2 \mid \sigma_1 + \sigma_3) \odot \Gamma \vdash \lambda x.t : (x :_{(s,r)} A) \rightarrow B} \ \text{T\_Fun}$$

If $A = \mathsf{Type}_l$ then there are no such reductions possible in GRTT under the interpretation (just arrows, vars, and Type) Otherwise we have the encoding: $\dfrac{[\![\Gamma]\!], x : [\![A]\!] \vdash t' : [\![B]\!]}{[\![\Gamma]\!] \vdash \lambda(x : [\![A]\!]).t' : [\![A]\!] \rightarrow [\![B]\!]}$ for which there is no possible reduction.

- Case T_App

$$\frac{(\Delta \mid \sigma_2 \mid \sigma_1 + \sigma_3) \odot \Gamma \vdash t_1 : (x :_{(s,r)} A) \rightarrow B \quad (\Delta \mid \sigma_4 \mid \sigma_1) \odot \Gamma \vdash t_2 : A}{(\Delta \mid \sigma_2 + s * \sigma_4 \mid \sigma_3 + r * \sigma_4) \odot \Gamma \vdash t_1 \ t_2 : [t_2/x]B} \ \text{T\_App}$$

We define two cases depending on $A$:

- $\mathsf{Type}_l = A$ Then the encoding yields

$$\frac{[\![\Gamma]\!] \vdash t' : \forall (x : \star_l).[\![B]\!] \quad [\![\Gamma]\!] \vdash T : \star_l}{[\![\Gamma]\!] \vdash t'[T] : [T/x][\![B]\!]}$$

There are then three possible GRTT reductions:

* $(\lambda x.t'_1) \ t_2 \rightsquigarrow [t_2/x]t'_1$ Thus:

$$[\![(\Delta \mid \sigma_2 + s * \sigma_4 \mid \sigma_3 + r * \sigma_4) \odot \Gamma \vdash [t_2/x]t_1 : [t_2/x]B]\!]$$
$$= [\![\Gamma]\!] \vdash t'' : [\![[t_2/x]B]\!]$$

$$\frac{[\![(\Delta \mid \sigma_2, s \mid \sigma_3, r) \odot \Gamma, x : \mathsf{Type}_l \vdash t'_1 : B]\!] = [\![\Gamma]\!], x : \star_l \vdash t'_0 : [\![B]\!]}{[\![(\Delta \mid \sigma_2 \mid \sigma_1 + \sigma_3) \odot \Gamma \vdash (\lambda x.t'_1) : (x :_{(s,r)} \mathsf{Type}_l) \rightarrow B]\!] = [\![\Gamma]\!] \vdash \Lambda(x : \star_l).t'_0 : \forall(x : \star_l).[\![B]\!]} \ \text{T\_Fun}$$
(7)

By type-application $\beta$ in SSF we then get:

$$(\Lambda(x : \star_l).t'_0)[T] \rightsquigarrow [T/x]t'_0$$

Applying Lemma B.10 with the translation here with $(\Delta \mid \sigma_4 \mid \sigma_1) \odot \Gamma \vdash t_2 : \mathsf{Type}_l = [\![\Gamma]\!] \vdash T : \star_l$ and the premise of (7) above gives us:

$$[\![(\Delta \mid \sigma_2 + s * \sigma_4 \mid \sigma_3 + r * \sigma_4) \odot \Gamma \vdash [t_2/x]t'_1 : [t_2/x]B]\!] = [\![\Gamma]\!] \vdash [T/x]t'_0 : [T/x][\![B]\!]$$

Thus satisfying the goal here.

* The other case is the congruence:

$$\frac{t_1 \rightsquigarrow t'_1}{t_1\, t_2 \rightsquigarrow t'_1\, t_2} \text{ Sem\_CongFunOne}$$

By type preservation (Lemma 3.35) on the premise we know that

$$(\Delta \mid \sigma_2 \mid \sigma_1 + \sigma_3) \odot \Gamma \vdash t'_1 : (x :_{(s,r)} A) \rightarrow B$$

and thus we can apply the encoding again to get some term $t''$ in SSF:

$$(\Delta \mid \sigma_2 \mid \sigma_1 + \sigma_3) \odot \Gamma \vdash t'_1 : (x :_{(s,r)} A) \rightarrow B = [\![\Gamma]\!] \vdash t'' : [\![(x :_{(s,r)} A) \rightarrow B]\!]$$

Thus the interpretation of the result is:

$$[\![(\Delta \mid \sigma_2 + s * \sigma_4 \mid \sigma_3 + r * \sigma_4) \odot \Gamma \vdash t'_1\, t_2 : [t_2/x]B]\!] = [\![\Gamma]\!] \vdash t''[T] : [T/x][\![B]\!]$$

By induction on the premise, with the encoding of $t_1$ and $t_1 \rightsquigarrow t'_1$ and (8) then we have that the SSF reduction:

$$t' \rightsquigarrow t''$$

Then we can construct the SSF reduction:

$$\frac{t' \rightsquigarrow t''}{t'[T] \rightsquigarrow t''[T]}$$

which satisfies the goal here.

– $\mathsf{Type}_l \notin^{+ve} B$ then the encoding yields:

$$\frac{[\![\Gamma]\!] \vdash t_a : [\![A]\!] \rightarrow [\![B]\!] \quad [\![\Gamma]\!] \vdash t_b : [\![A]\!]}{[\![\Gamma]\!] \vdash t_a\, t_b : [\![B]\!]} \tag{8}$$

There are then two possible GRTT reductions:

* $(\lambda x.t'_1)\, t_2 \rightsquigarrow [t_2/x]t'_1$ Thus we have the encoding:

$$\frac{[\![(\Delta \mid \sigma_2, s \mid \sigma_3, r) \odot \Gamma, x : A \vdash t'_1 : B]\!] = [\![\Gamma]\!], x : [\![A]\!] \vdash t'_a : [\![B]\!]}{[\![(\Delta \mid \sigma_2 \mid \sigma_1 + \sigma_3) \odot \Gamma \vdash (\lambda x.t'_1) : (x :_{(s,r)} A) \rightarrow B]\!] = [\![\Gamma]\!] \vdash \lambda(x : [\![A]\!]).t'_a : [\![A]\!] \rightarrow [\![B]\!]} \text{ T\_Fun} \tag{9}$$

Applying Lemma B.10 with the translation here of $t_2$ to $t_b$ and the premise of (9) then gives:

$$[\![(\Delta \mid \sigma_2 + s * \sigma_4 \mid \sigma_3 + r * \sigma_4) \odot \Gamma \vdash [t_2/x]t'_1 : [t_2/x]B]\!] = [\![\Gamma]\!] \vdash [t_b/x]t'_a : [\![B]\!]$$

By $\beta$ in SSF we then get:

$$(\lambda(x : [\![A]\!]).t'_a)t_b \rightsquigarrow [t_b/x]t'_a$$

Thus showing soundness of this reduction.

* The other case is the congruence:

$$\frac{t_1 \leadsto t_1'}{t_1\, t_2 \leadsto t_1'\, t_2}\ \text{Sem\_CongFunOne}$$

By type preservation (Lemma 3.35) on the premise we know that

$$(\Delta \mid \sigma_2 \mid \sigma_1 + \sigma_3) \odot \Gamma \vdash t_1' : (x :_{(s,r)} A) \to B$$

and thus we can apply the encoding again to get some term $t_a'$ in SSF:

$$(\Delta \mid \sigma_2 \mid \sigma_1 + \sigma_3) \odot \Gamma \vdash t_1' : (x :_{(s,r)} A) \to B = [\![\Gamma]\!] \vdash t_a' : [\![(x :_{(s,r)} A) \to B]\!]$$

Thus the interpretation of the result is:

$$[\![(\Delta \mid \sigma_2 + s * \sigma_4 \mid \sigma_3 + r * \sigma_4) \odot \Gamma \vdash t_1'\, t_2 : [t_2/x]B]\!] = [\![\Gamma]\!] \vdash t_a'\, t_b : [\![B]\!]$$

By induction on the premise, with the encoding of $t_1$ and $t_1 \leadsto t_1'$ and (8) then we have that the SSF reduction:

$$t_a \leadsto t_a'$$

Then we can construct the SSF reduction:

$$\frac{t_a \leadsto t_a'}{t_a\, t_b \leadsto t_a'\, t_b}$$

which satisfies the goal here.

$\square$

### B.2.5 Completeness

**Lemma B.13** (Completeness). *Given* $\text{SSF}((\Delta \mid \sigma_1 \mid \sigma_2) \odot \Gamma \vdash t : A)$ *and* $[\![(\Delta \mid \sigma_1 \mid \sigma_2) \odot \Gamma \vdash t : A]\!] = [\![\Gamma]\!] \vdash t_s : [\![A]\!]$ *in* $\text{GRTT}$ *and* $t_s \leadsto^{\text{ssf}} t_s'$ *then* $t \leadsto t'$ *and* $[\![(\Delta \mid \sigma_1 \mid \sigma_2) \odot \Gamma \vdash t' : A]\!] = [\![\Gamma]\!] \vdash t_s'' : [\![A]\!]$ *with* $t_s'' = t_s'$

*Proof.* By induction on typing derivations in the image of the encoding:

- Case T_Var

$$\frac{\Delta_1, \sigma, \Delta_2 \odot \Gamma_1, x : A, \Gamma_2 \vdash \quad |\Delta_1| = |\Gamma_1|}{(\Delta_1, \sigma, \Delta_2 \mid \mathbf{0}^{|\Delta_1|}, 1, \mathbf{0} \mid \sigma, 0, \mathbf{0}) \odot \Gamma_1, x : A, \Gamma_2 \vdash x : A}\ \text{T\_Var}$$

Give us the derivation:

$$\frac{[\![\Gamma_1]\!], x : [\![A]\!], [\![\Gamma_2]\!]\ Ok}{[\![\Gamma_1]\!], x : [\![A]\!], [\![\Gamma_2]\!] \vdash x : [\![A]\!]}$$

Since variables cannot reduce the premise of the lemma here is false, thus trivially holds.

- Case T_Arrow maps to the type language of SSF and so has no reduction.

- Case T_Fun is interpreted as an SSF function which has no reduction (under the CBN semantics) so this case trivially holds similar to the T_Var case.

- Case T_App with $t = t_1 t_2$:

$$\frac{(\Delta, \sigma_1 \mid \sigma_3, r \mid \mathbf{0}) \odot \Gamma, x : A \vdash B : \text{Type}_l \quad (\Delta \mid \sigma_2 \mid \sigma_1 + \sigma_3) \odot \Gamma \vdash t_1 : (x :_{(s,r)} A) \to B \quad (\Delta \mid \sigma_4 \mid \sigma_1) \odot \Gamma \vdash t_2 : A}{(\Delta \mid \sigma_2 + s * \sigma_4 \mid \sigma_3 + r * \sigma_4) \odot \Gamma \vdash t_1\, t_2 : [t_2/x]B}\ \text{T\_App}$$

There are two cases depending on $A$:

103

- $\mathsf{Type}_l = A$ By induction we have terms $t_0$ and $T'$ such that:

$$[\![\Gamma]\!] \vdash t_0 : \forall (x : \star_l).[\![B]\!] \tag{10}$$

$$[\![\Gamma]\!] \vdash T : \star_l \tag{11}$$

  Then with the SSF derivation:

$$\frac{ih1 \quad ih2}{[\![\Gamma]\!] \vdash t_0[T] : [T/x][\![B]\!]}$$

  where $[T/x][\![B]\!]$ matches $[\![t_2/x]B]\!]$ by Lemma B.9.

  There are two possible reductions depending on the form of $t_0$:

  *
  
$$\frac{t_0 \leadsto^{\mathsf{ssf}} t_0'}{t_0[T] \leadsto^{\mathsf{ssf}} t_0'[T]}$$

    Thus by induction on $t_1$ (ih1) with the premise of this reduction here we have:

$$t_1 \leadsto t_1' \quad \wedge \quad [\![(\Delta \mid \sigma_2 \mid \sigma_1 + \sigma_3) \odot \Gamma \vdash t_1' : (x :_{(s,r)} A) \to B]\!] = [\![\Gamma]\!] \vdash t_1'' : [\![A]\!] \to [\![B]\!] \quad \wedge \quad t_1'' = t_0'$$

    Thus we can derive the reduction in GRTT of:

$$\frac{t_1 \leadsto t_1'}{t_1\, t_2 \leadsto t_1'\, t_2}$$

    and we have that $t_1''[T] = t_0'[T]$ satisfying the goal.

  * The remaining case is of a $\beta$-reduction and thus we must have that $t_1 = \lambda x.t_1'$. By the interpretation we then get $(\Lambda(x : [\![A]\!]).t_s') = t_0$ in SSF (thus $[\![t1']\!] = t_s'$) thus here we have reduction:

$$(\Lambda(x : [\![A]\!]).t_s')[T] \leadsto^{\mathrm{SSF}} [T/x]t_s'$$

    GrTT can then make a $\beta$ reduction itself as:

$$(\lambda x.t_1')t_2 \leadsto [t_2/x]t_1'$$

    The goal is then that $[\![(\Delta \mid \sigma_1 \mid \sigma_2) \odot \Gamma \vdash [t_2/x]t_1' : A]\!] = [\![\Gamma]\!] \vdash t_s'' : [\![A]\!]$ and $t_s'' = [T/x]t_s'$ which we get from Lemma B.11.

- $\mathsf{Type}_l \notin^{+ve} A$ (no positive occurences of $\mathsf{Type}_l$) thus $r = 0$.

  Which inductively has two interpeted terms:

$$[\![\Gamma]\!] \vdash t_{s1} : [\![A]\!] \to [\![B]\!] \tag{12}$$

$$[\![\Gamma]\!] \vdash t_{s2} : [\![A]\!] \tag{13}$$

  and the SSF derivation:

$$\frac{(\text{ih1}) \quad (\text{ih2})}{[\![\Gamma]\!] \vdash t_{s1}\ t_{s2} : [\![B]\!]}$$

  There are then two possible redctions depending on the form of $t_1/t_{s1}$:

  *

$$\frac{t_{s1} \leadsto^{\mathsf{ssf}} t_{s1}'}{t_{s1}\ t_{s2} \leadsto^{\mathsf{ssf}} t_{s1}'\ t_{s2}}$$

    Thus, by induction on $t_1$ with the premise of this reduction we have:

$$t_1 \leadsto t_1' \quad \wedge \quad [\![(\Delta \mid \sigma_2 \mid \sigma_1 + \sigma_3) \odot \Gamma \vdash t_1' : (x :_{(s,r)} A) \to B]\!] = [\![G]\!] \vdash t_{s1}'' : [\![A]\!] \to [\![B]\!] \quad \wedge \quad t_{s1}'' = t_{s1}'$$

Thus we can derive the GrTT reduction of:

$$\frac{t_1 \rightsquigarrow t'_1}{t_1\ t_2 \rightsquigarrow t'_1\ t_2}$$

with $t''_{s1}\ t_{s2} = t'_{s1}\ t_{s2}$

* $\beta$-reduction with $t_1 = t_1 = \lambda x.t'_1$. By the interpretation we then get $(\lambda(x : [\![A]\!]).t_s) = t_{s1}$ in SSF (thus $[\![t1']\!] = t_s$) thus here we have reduction:

$$(\lambda(x : [\![A]\!]).t_s)\ t_{s2} \rightsquigarrow^{\text{SSF}} [t_{s2}/x]t_s$$

GrTT can then make a $\beta$ reduction itself as:

$$(\lambda x.t'_1)t_2 \rightsquigarrow [t_2/x]t'_1$$

The goal is then that $[\![(\Delta \mid \sigma_1 \mid \sigma_2) \odot \Gamma \vdash [t_2/x]t'_1 : A]\!] = [\![\Gamma]\!] \vdash t''_s : [\![A]\!]$ and $t''_s = [t_{s2}/x]t_s$ which we get from Lemma B.10.

$\square$

## C  Strong Normalization

**Definition C.1.** *Typing can be broken up into the following stages:*

$$\begin{aligned}
\text{Kind} &:= \{A \mid \exists \Delta, \sigma_1, \Gamma.(\Delta \mid \sigma_1 \mid \mathbf{0}) \odot \Gamma \vdash A : \text{Type}_1\} \\
\text{Type} &:= \{A \mid \exists \Delta, \sigma_1, \Gamma.(\Delta \mid \sigma_1 \mid \mathbf{0}) \odot \Gamma \vdash A : \text{Type}_0\} \\
\text{Const} &:= \{t \mid \exists \Delta, \sigma_1, \sigma_2, \Gamma, A.(\Delta \mid \sigma_1 \mid \sigma_2) \odot \Gamma \vdash t : A \wedge \\
&\qquad (\Delta \mid \sigma_2 \mid \mathbf{0}) \odot \Gamma \vdash A : \text{Type}_1\} \\
\text{Term} &:= \{t \mid \exists \Delta, \sigma_1, \sigma_2, \Gamma, A.(\Delta \mid \sigma_1 \mid \sigma_2) \odot \Gamma \vdash t : A \wedge \\
&\qquad (\Delta \mid \sigma_2 \mid \mathbf{0}) \odot \Gamma \vdash A : \text{Type}_0\}
\end{aligned}$$

**Lemma C.2** (Classification). *It is the case that $\text{Kind} \cap \text{Type} = \emptyset$ and $\text{Const} \cap \text{Term} = \emptyset$.*

**Definition C.3.** *The set of base terms $\mathcal{B}$ is defined by:*

1. *For $x$ a variable, $x \in \mathcal{B}$,*

2. $\text{Type}_0, \text{Type}_1 \in \mathcal{B}$,

3. *If $t_1 \in \mathcal{B}$ and $t_2 \in \text{SN}$, then $(t_1\ t_2) \in \mathcal{B}$,*

4. *If $t_2 \in \mathcal{B}$ and $t_1 \in \text{SN}$, then $(\text{let}\ (x, y) = t_1\ \text{in}\ t_2) \in \mathcal{B}$,*

5. *If $t_2 \in \mathcal{B}$ and $t_1 \in \text{SN}$, then $(\text{let}\ \Box x = t_1\ \text{in}\ t_2) \in \mathcal{B}$,*

6. *If $A, B \in \text{SN}$, then $((x :_{(r,s)} A) \to B) \in \mathcal{B}$ for any $r, s \in R$,*

7. *If $A, B \in \text{SN}$, then $((x :_r A) \otimes B) \in \mathcal{B}$ for any $r \in R$,*

8. *If $A \in \text{SN}$, then $(\Box_r A) \in \mathcal{B}$ for any $r \in R$.*

**Definition C.4.** *The* key redex *of a term is defined by:*

1. *If $t$ is a redex, then $t$ is its own key redex,*

2. *If $t_1$ has key redex $t$, then $(t_1\ t_2)$ has key redex $t$,*

3. If $t_1$ has key redex $t$, then $(\text{let }(x,y) = t_1 \text{ in } t_2)$ has key redex $t$,

4. If $t_1$ has key redex $t$, then $(\text{let } \Box x = t_1 \text{ in } t_2)$ has key redex $t$.

The term obtained from $t$ by contracting its key redex is denoted by $\text{red}_k t$.

**Lemma C.5.** *The following are both true:*

1. $\mathcal{B} \subseteq \mathsf{SN}$

2. *The key redex of a term is unique and a head redex.*

**Definition C.6.** *A set of terms $X$ is* saturated *if:*

1. $X \subset \mathsf{SN}$,

2. $\mathcal{B} \subset X$,

3. *If $\text{red}_k t \in X$ and $t \in \mathsf{SN}$, then $t \in X$.*

*The collection of saturated sets is denoted by* $\mathsf{SAT}$.

**Lemma C.7** ($\mathsf{SN}$ is saturated)**.** *Every saturated set is non-empty and $\mathsf{SN}$ is saturated.*

*Proof.* By definition. $\square$

**Definition C.8.** *For $T \in \mathsf{Kind}$, the kind interpretation, $\mathcal{K}[\![T]\!]$, is defined inductively as follows:*

$$\mathcal{K}[\![\mathsf{Type}_0]\!] = \mathsf{SAT}$$
$$\mathcal{K}[\![(x :_{(r,s)} A) \to B]\!] = \{f \mid f : \mathcal{K}[\![A]\!] \to \mathcal{K}[\![B]\!]\}, \text{ if } A, B \in \mathsf{Kind}$$
$$\mathcal{K}[\![(x :_{(r,s)} A) \to B]\!] = \mathcal{K}[\![A]\!], \text{ if } A \in \mathsf{Kind}, B \in \mathsf{Type}$$
$$\mathcal{K}[\![(x :_{(r,s)} A) \to B]\!] = \mathcal{K}[\![B]\!], \text{ if } A \in \mathsf{Type}, B \in \mathsf{Kind}$$
$$\mathcal{K}[\![(x :_s A) \otimes B]\!] = \mathcal{K}[\![A]\!] \times \mathcal{K}[\![B]\!], \text{ if } A, B \in \mathsf{Kind}$$
$$\mathcal{K}[\![(x :_s A) \otimes B]\!] = \mathcal{K}[\![A]\!], \text{ if } A \in \mathsf{Kind}, B \in \mathsf{Type}$$
$$\mathcal{K}[\![(x :_s A) \otimes B]\!] = \mathcal{K}[\![B]\!], \text{ if } A \in \mathsf{Type}, B \in \mathsf{Kind}$$
$$\mathcal{K}[\![\Box_s A]\!] = \mathcal{K}[\![A]\!]$$

**Definition C.9.** *Type valuations, $\Delta \odot \Gamma \models \varepsilon$, are defined as follows:*

$$\frac{}{\emptyset \odot \emptyset \models \emptyset} \text{ Ep\_Empty} \qquad \frac{X \in \mathcal{K}[\![A]\!] \quad \Delta \odot \Gamma \models \varepsilon \quad (\Delta \mid \sigma_2 \mid \mathbf{0}) \odot \Gamma \vdash A : \mathsf{Type}_1}{(\Delta, \sigma_2) \odot (\Gamma, x : A) \models \varepsilon[x \mapsto X]} \text{ Ep\_ExtTy}$$

$$\frac{\Delta \odot \Gamma \models \varepsilon \quad (\Delta \mid \sigma_2 \mid \mathbf{0}) \odot \Gamma \vdash A : \mathsf{Type}_0}{(\Delta, \sigma_2) \odot (\Gamma, x : A) \models \varepsilon} \text{ Ep\_ExtTm}$$

**Definition C.10.** *Given a type valuation $\Delta \odot \Gamma \models \varepsilon$ and a type $T \in (\mathsf{Kind} \cup \mathsf{Type} \cup \mathsf{Con})$ with $T$ typable in*

$\Delta \odot \Gamma$, we define the interpretation of types, $[\![T]\!]_\varepsilon$ inductively as follows:

$$
\begin{aligned}
[\![\mathsf{Type}_1]\!]_\varepsilon &= \mathsf{SN} \\
[\![\mathsf{Type}_0]\!]_\varepsilon &= \lambda X \in \mathsf{SAT}.\mathsf{SN} \\
[\![x]\!]_\varepsilon &= \varepsilon\, x,\ \text{if } x \in \mathsf{Con} \\
[\![(x :_{(r,s)} A) \to B]\!]_\varepsilon &= \lambda X \in \mathcal{K}[\![A]\!] \to \mathcal{K}[\![B]\!].\bigcap_{Y \in \mathcal{K}[\![A]\!]}(([\![A]\!]_\varepsilon)(Y) \to ([\![B]\!]_{\varepsilon[x \mapsto Y]})(X(Y))), \\
&\quad \text{if } A, B \in \mathsf{Kind} \\
[\![(x :_{(r,s)} A) \to B]\!]_\varepsilon &= \bigcap_{Y \in \mathcal{K}[\![A]\!]}(([\![A]\!]_\varepsilon)(Y) \to ([\![B]\!]_{\varepsilon[x \mapsto Y]})), \\
&\quad \text{if } A \in \mathsf{Kind}, B \in \mathsf{Type} \\
[\![(x :_{(r,s)} A) \to B]\!]_\varepsilon &= \lambda X \in \mathcal{K}[\![B]\!].[\![A]\!]_\varepsilon \to ([\![B]\!]_\varepsilon)(X),\ \text{if } A \in \mathsf{Type}, B \in \mathsf{Kind} \\
[\![(x :_{(r,s)} A) \to B]\!]_\varepsilon &= [\![A]\!]_\varepsilon \to [\![B]\!]_\varepsilon,\ \text{if } A, B \in \mathsf{Type} \\
[\![\lambda x : A.B]\!]_\varepsilon &= \lambda X \in \mathcal{K}[\![A]\!].[\![B]\!]_{\varepsilon[x \mapsto X]},\ \text{if } A \in \mathsf{Kind}, B \in \mathsf{Con} \\
[\![\lambda x : A.B]\!]_\varepsilon &= [\![A]\!]_\varepsilon,\ \text{for any } A \in \mathsf{Kind}, B \in \mathsf{Term} \\
[\![\lambda x : A.B]\!]_\varepsilon &= [\![B]\!]_\varepsilon,\ \text{for any } A \in \mathsf{Type}, B \in \mathsf{Con} \\
[\![A\, B]\!]_\varepsilon &= [\![A]\!]_\varepsilon([\![B]\!]_\varepsilon),\ \text{if } B \in \mathsf{Con} \\
[\![A\, t]\!]_\varepsilon &= [\![A]\!]_\varepsilon,\ \text{if } t \in \mathsf{Term} \\
[\![(x :_s A) \otimes B]\!]_\varepsilon &= \lambda X \in \mathcal{K}[\![A]\!] \times \mathcal{K}[\![B]\!].([\![A]\!]_\varepsilon)(\pi_1\, X) \times ([\![B]\!]_{\varepsilon[x \mapsto \pi_1\, X]})(\pi_2\, X), \\
&\quad \text{if } A, B \in \mathsf{Kind} \\
[\![(x :_s A) \otimes B]\!]_\varepsilon &= \lambda X \in \mathcal{K}[\![B]\!].[\![A]\!]_\varepsilon \times ([\![B]\!]_\varepsilon)(X),\ \text{if } A \in \mathsf{Type}, B \in \mathsf{Kind} \\
[\![(x :_s A) \otimes B]\!]_\varepsilon &= \lambda X \in \mathcal{K}[\![A]\!].([\![A]\!]_\varepsilon)(X) \times [\![B]\!]_{\varepsilon[x \mapsto X]},\ \text{if } A \in \mathsf{Kind}, B \in \mathsf{Type} \\
[\![(x :_s A) \otimes B]\!]_\varepsilon &= [\![A]\!]_\varepsilon \times [\![B]\!]_\varepsilon,\ \text{if } A, B \in \mathsf{Type} \\
[\![(A, B)]\!]_\varepsilon &= ([\![A]\!]_\varepsilon, [\![B]\!]_\varepsilon) \\
[\![(A, t)]\!]_\varepsilon &= [\![A]\!]_\varepsilon,\ \text{if } t \in \mathsf{Term} \\
[\![(t, B)]\!]_\varepsilon &= [\![B]\!]_\varepsilon,\ \text{if } t \in \mathsf{Term} \\
[\![\mathsf{let}\,(x : A, y : B) = t\,\mathsf{in}\,C]\!]_\varepsilon &= [\![C]\!]_{\varepsilon[x \mapsto \pi_1\,[\![t]\!]_\varepsilon, y \mapsto \pi_2\,[\![t]\!]_\varepsilon]},\ \text{if } A, B \in \mathsf{Kind} \\
[\![\mathsf{let}\,(x : A, y : B) = t\,\mathsf{in}\,C]\!]_\varepsilon &= [\![C]\!]_{\varepsilon[x \mapsto [\![t]\!]_\varepsilon]},\ \text{if } A \in \mathsf{Kind}, B \in \mathsf{Type} \\
[\![\mathsf{let}\,(x : A, y : B) = t\,\mathsf{in}\,C]\!]_\varepsilon &= [\![C]\!]_{\varepsilon[y \mapsto [\![t]\!]_\varepsilon]},\ \text{if } A \in \mathsf{Type}, B \in \mathsf{Kind} \\
[\![\mathsf{let}\,(x : A, y : B) = t\,\mathsf{in}\,C]\!]_\varepsilon &= [\![C]\!]_\varepsilon,\ \text{if } A, B \in \mathsf{Type} \\
[\![\Box_{(r,s)} A]\!]_\varepsilon &= [\![A]\!]_\varepsilon \\
[\![\Box A]\!]_\varepsilon &= [\![A]\!]_\varepsilon \\
[\![\mathsf{let}\,\Box(x : A) = t\,\mathsf{in}\,B]\!]_\varepsilon &= [\![B]\!]_{\varepsilon[x \mapsto [\![t]\!]_\varepsilon]},\ \text{if } A \in \mathsf{Kind} \\
[\![\mathsf{let}\,\Box(x : A) = t\,\mathsf{in}\,B]\!]_\varepsilon &= [\![B]\!]_\varepsilon,\ \text{if } A \in \mathsf{Type}
\end{aligned}
$$

**Definition C.11.** *Suppose $\Delta \odot \Gamma \models \varepsilon$. Then* Valid term valuations, $\Delta \odot \Gamma \models_\varepsilon \rho$, *are defined as follows:*

$$
\frac{}{\emptyset \odot \emptyset \models_\emptyset \emptyset}\ \textsc{Rho\_Empty} \qquad \frac{t \in ([\![A]\!]_\varepsilon)(\varepsilon\,x) \quad \Delta \odot \Gamma \models_\varepsilon \rho \quad (\Delta \mid \sigma_2 \mid \mathbf{0}) \odot \Gamma \vdash A : \mathsf{Type}_1}{(\Delta, \sigma_2) \odot (\Gamma, x : A) \models_\varepsilon \rho[x \mapsto t]}\ \textsc{Rho\_ExtTy}
$$

$$
\frac{t \in [\![A]\!]_\varepsilon \quad \Delta \odot \Gamma \models_\varepsilon \rho \quad (\Delta \mid \sigma_2 \mid \mathbf{0}) \odot \Gamma \vdash A : \mathsf{Type}_0}{(\Delta, \sigma_2) \odot (\Gamma, x : A) \models_\varepsilon \rho[x \mapsto t]}\ \textsc{Rho\_ExtTm}
$$

**Definition C.12.** *Suppose $\Delta \odot \Gamma \models_\varepsilon \rho$. Then the* interpretation of a term $t$ typable in $\Delta \odot \Gamma$ *is $(\!|t|\!)_\rho = \rho\,t$, but where all let-expressions are translated into substitutions, and all graded modalities are erased.*

**Definition C.13.** *Suppose $(\Delta \mid \sigma_1 \mid \sigma_2) \odot \Gamma \vdash t : A$. Then* semantic typing, $(\Delta \mid \sigma_1 \mid \sigma_2) \odot \Gamma \models t : A$, *is defined as follows:*

1. *If $(\Delta \mid \sigma \mid \mathbf{0}) \odot \Gamma \vdash A : \mathsf{Type}_1$, then for every $\Delta \odot \Gamma \models_\varepsilon \rho$, $(\!|t|\!)_\rho \in [\![A]\!]_\varepsilon([\![t]\!]_\varepsilon)$.*

2. *If $(\Delta \mid \sigma \mid \mathbf{0}) \odot \Gamma \vdash A : \mathsf{Type}_0$, then for every $\Delta \odot \Gamma \models_\varepsilon \rho$, $(\!|t|\!)_\rho \in [\![A]\!]_\varepsilon$.*

**Lemma C.14** (Substitution for Typing Interpretation)**.** *Suppose $\Delta \odot \Gamma \models \varepsilon$ and we have types $T_2 \in (\mathsf{Kind} \cup \mathsf{Con})$ and $T_1 \in \mathsf{Con}$ with $T_1$ and $T_2$ typable in $\Delta \odot \Gamma$, and a term $t \in \mathsf{Term}$ typable in $\Delta \odot \Gamma$. Then:*

1. $[\![T_2]\!]_{\varepsilon[x \mapsto [\![T_1]\!]_\varepsilon]} = [\![[T_1/x]T_2]\!]_\varepsilon$.

2. $[\![T_2]\!]_\varepsilon = [\![[t/x]T_2]\!]_\varepsilon$.

*Proof.* By straightforward induction on $T_2$ with the fact that substitutions disappear in the kind interpretation. □

**Lemma C.15** (Equality of Interpretations). *Suppose $\Delta \odot \Gamma \models \varepsilon$. Then:*

1. *if $(\Delta \mid \sigma_1 \mid \sigma_2) \odot \Gamma \vdash T_1 = T_2 : \mathsf{Type}_1$, then $\mathcal{K}[\![T_1]\!] = \mathcal{K}[\![T_2]\!]$.*

2. *if $(\Delta \mid \sigma_1 \mid \sigma_2) \odot \Gamma \vdash T_1 = T_2 : \mathsf{Type}_l$, then $[\![T_1]\!]_\varepsilon = [\![T_2]\!]_\varepsilon$.*

*Proof.* Part one follows by induction on $(\Delta \mid \sigma_1 \mid \sigma_2) \odot \Gamma \vdash T_1 = T_2 : \mathsf{Type}_l$, and so does part two, but it also depends on part one and the previous lemma. □

**Lemma C.16** (Interpretation Soundness). *Suppose $\Delta \odot \Gamma \models \varepsilon$ and $(\Delta \mid \sigma \mid \mathbf{0}) \odot \Gamma \vdash A : \mathsf{Type}_1$. Then:*

1. *If $(\Delta \mid \sigma_1 \mid \sigma_2) \odot \Gamma \vdash t : A$, then $[\![t]\!]_\varepsilon \in \mathcal{K}[\![A]\!]$*

2. *$[\![A]\!]_\varepsilon \in \mathcal{K}[\![A]\!] \to \mathsf{SAT}$*

*Proof.* This is a proof by simultaneous induction over $(\Delta \mid \sigma_1 \mid \sigma_2) \odot \Gamma \vdash t : A$ and $(\Delta \mid \sigma \mid \mathbf{0}) \odot \Gamma \vdash A : \mathsf{Type}_1$. We consider part 1 assuming 2, and vice versa.

Proof of part 1:

Case 1:
$$\frac{\Delta \odot \Gamma \vdash}{(\Delta \mid \mathbf{0} \mid \mathbf{0}) \odot \Gamma \vdash \mathsf{Type}_0 : \mathsf{Type}_1} \text{ T\_Type}$$

This case holds trivially, because $\mathsf{Type}_1$ cannot be of type $\mathsf{Type}_1$.

Case 2:
$$\frac{\Delta_1, \sigma, \Delta_2 \odot \Gamma_1, x : A, \Gamma_2 \vdash \quad |\Delta_1| = |\Gamma_1|}{(\Delta_1, \sigma, \Delta_2 \mid \mathbf{0}^{|\Delta_1|}, 1, \mathbf{0} \mid \sigma, 0, \mathbf{0}) \odot \Gamma_1, x : A, \Gamma_2 \vdash x : A} \text{ T\_Var}$$

In this case we have:
$$\Delta = (\Delta_1, \sigma, \Delta_2)$$
$$\sigma_1 = (\mathbf{0}^{|\Delta_1|}, 1, \mathbf{0})$$
$$\sigma_2 = (\sigma, 0, \mathbf{0})$$
$$t = x$$

Thus, we must show that:
$$[\![x]\!]_\varepsilon \in \mathcal{K}[\![A]\!]$$

We know by Definition C.10 and Definition C.9 that $[\![x]\!]_\varepsilon = \varepsilon\, x \in \mathcal{K}[\![A]\!]$; thus, we obtain our result.

Case 3:
$$\frac{(\Delta \mid \sigma_3 \mid \mathbf{0}) \odot \Gamma \vdash B_1 : \mathsf{Type}_{l_1} \quad (\Delta, \sigma_3 \mid \sigma_4, r \mid \mathbf{0}) \odot \Gamma, x : B_1 \vdash B_2 : \mathsf{Type}_{l_2}}{(\Delta \mid \sigma_3 + \sigma_4 \mid \mathbf{0}) \odot \Gamma \vdash (x :_{(s,r)} B_1) \to B_2 : \mathsf{Type}_{l_1 \sqcup l_2}} \text{ T\_Arrow}$$

In this case we know:
$$\sigma_1 = (\sigma_3 + \sigma_4)$$
$$\sigma_2 = \mathbf{0}$$
$$t = (x :_{(s,r)} B_1) \to B_2$$
$$A = \mathsf{Type}_{l_1 \sqcup l_2}$$

However, by assumption we know that $(\Delta \mid \mathbf{0} \mid \mathbf{0}) \odot \Gamma \vdash \mathsf{Type}_{l_1 \sqcup l_2} : \mathsf{Type}_1$, and hence, $\mathsf{Type}_{l_1 \sqcup l_2} = \mathsf{Type}_0$ which implies that $l_1 = l_2 = 0$. This all implies that $B_1, B_2 \in \mathsf{Type}$. Furthermore, we know that $(\Delta, \sigma_3) \odot (\Gamma, x : B_1) \models \varepsilon$, because $B_1 \in \mathsf{Type}$.

Thus, by Definition C.8 we must show that:
$$[\![(x :_{(s,r)} B_1) \to B_2]\!]_\varepsilon \in \mathcal{K}[\![\mathsf{Type}_0]\!] = \mathsf{SAT}$$

By Definition C.10 we know that:
$$[\![(x :_{(s,r)} B_1) \to B_2]\!]_\varepsilon = [\![B_1]\!]_\varepsilon \to [\![B_2]\!]_\varepsilon$$

By the IH:

108

IH(1): $[\![B_1]\!]_\varepsilon \in \mathcal{K}[\![\mathsf{Type}_0]\!] = \mathsf{SAT}$

IH(2): $[\![B_2]\!]_\varepsilon \in \mathcal{K}[\![\mathsf{Type}_0]\!] = \mathsf{SAT}$

Therefore, $[\![(x :_{(s,r)} B_1) \to B_2]\!]_\varepsilon = ([\![B_1]\!]_\varepsilon \to [\![B_2]\!]_\varepsilon) = (\mathsf{SAT} \to \mathsf{SAT}) \in \mathsf{SAT}$ which holds by the closer of $\mathsf{SAT}$ under function spaces.

Case 4:
$$\frac{(\Delta \mid \sigma_3 \mid \mathbf{0}) \odot \Gamma \vdash B_1 : \mathsf{Type}_{l_1} \quad (\Delta, \sigma_3 \mid \sigma_4, r \mid \mathbf{0}) \odot \Gamma, x : B_1 \vdash B_2 : \mathsf{Type}_{l_2}}{(\Delta \mid \sigma_3 + \sigma_4 \mid \mathbf{0}) \odot \Gamma \vdash (x :_r B_1) \otimes B_2 : \mathsf{Type}_{l_1 \sqcup l_2}} \; \text{T\_Ten}$$

This case follows nearly exactly as the previous case, but ending with a cartesian product, $\mathsf{SAT} \times \mathsf{SAT}$, rather than the function space.

Case 5:
$$\frac{(\Delta, \sigma_3 \mid \sigma_4, s \mid \sigma_5, r) \odot \Gamma, x : B_1 \vdash t' : B_2}{(\Delta \mid \sigma_4 \mid \sigma_3 + \sigma_5) \odot \Gamma \vdash \lambda x : B_1.t' : (x :_{(s,r)} B_1) \to B_2} \; \text{T\_Fun}$$

In this case we know:
$$\sigma_1 = \sigma_4$$
$$\sigma_2 = \sigma_3 + \sigma_5$$
$$t = \lambda x : B_1.t'$$
$$A = (x :_{(s,r)} B_1) \to B_2$$

We also know that $(\Delta \mid \sigma_3 + \sigma_5 \mid \mathbf{0}) \odot \Gamma \vdash (x :_{(s,r)} B_1) \to B_2 : \mathsf{Type}_1$ by assumption. This implies by inversion that $B_1, B_2 \in \mathsf{Kind}$, $B_1 \in \mathsf{Kind}$ and $B_2 \in \mathsf{Type}$, or $B_1 \in \mathsf{Type}$ and $B_2 \in \mathsf{Kind}$. We consider each case in turn:

Subcase 1: Suppose $B_1, B_2 \in \mathsf{Kind}$. Then we must show that:
$$[\![\lambda x : B_1.t']\!]_\varepsilon \in \mathcal{K}[\![(x :_{(s,r)} B_1) \to B_2]\!] = \{f \mid f : \mathcal{K}[\![B_1]\!] \to \mathcal{K}[\![B_2]\!]\}$$

By Definition C.10 we know that:
$$[\![\lambda x : B_1.t']\!]_\varepsilon = \lambda X \in \mathcal{K}[\![B_1]\!].[\![t']\!]_{\varepsilon[x \mapsto X]}$$

Suppose $X \in \mathcal{K}[\![B_1]\!]$. Then we will show that $[\![t']\!]_{\varepsilon[x \mapsto X]} \in \mathcal{K}[\![B_2]\!]$. We know by assumption that $\Delta \odot \Gamma \models \varepsilon$ and $(\Delta, \sigma_3 \mid \sigma_4, s \mid \sigma_5, r) \odot \Gamma, x : B_1 \vdash t' : B_2$. Hence, by Lemma 3.5, we know that $(\Delta \mid \sigma_3 \mid \mathbf{0}) \odot \Gamma \vdash B_1 : \mathsf{Type}_1$. Thus, we know that $(\Delta, \sigma_3) \odot (\Gamma, x : B_1) \models \varepsilon[x \mapsto X]$.

Therefore, by the IH:

IH(1): $[\![t']\!]_{\varepsilon[x \mapsto X]} \in \mathcal{K}[\![B_2]\!]$

which is what was to be shown.

Subcase 2: Suppose $B_1 \in \mathsf{Kind}$ and $B_2 \in \mathsf{Type}$. Then we must show that:
$$[\![\lambda x : B_1.t']\!]_\varepsilon \in \mathcal{K}[\![(x :_{(s,r)} B_1) \to B_2]\!] = \mathcal{K}[\![B_1]\!]$$

By Definition C.10 we know that:
$$[\![\lambda x : B_1.t']\!]_\varepsilon = \mathcal{K}[\![B_1]\!]$$

which is what was to be shown.

Subcase 3: Suppose $B_1 \in \mathsf{Type}$ and $B_2 \in \mathsf{Kind}$. Then we must show that:
$$[\![\lambda x : B_1.t']\!]_\varepsilon \in \mathcal{K}[\![(x :_{(s,r)} B_1) \to B_2]\!] = \mathcal{K}[\![B_2]\!]$$

By Definition C.10 we know that:
$$[\![\lambda x : B_1.t']\!]_\varepsilon = \mathcal{K}[\![B_2]\!]$$

which is what was to be shown.

Case 6:
$$\frac{(\Delta \mid \sigma_4 \mid \sigma_3 + \sigma_5) \odot \Gamma \vdash t_1 : (x :_{(s,r)} B_1) \to B_2 \quad (\Delta \mid \sigma_6 \mid \sigma_3) \odot \Gamma \vdash t_2 : B_1}{(\Delta \mid \sigma_4 + s * \sigma_6 \mid \sigma_5 + r * \sigma_6) \odot \Gamma \vdash t_1\, t_2 : [t_2/x]B_2} \text{ T\_App}$$

In this case we know that:
$$\sigma_1 = \sigma_4 + s * \sigma_6$$
$$\sigma_2 = \sigma_5 + r * \sigma_6$$
$$t = t_1\, t_2$$
$$A = [t_2/x]B_2$$

It suffices to show that:
$$[\![t_1\, t_2]\!]_\varepsilon \in \mathcal{K}[\![[t_2/x]B_2]\!] = \mathcal{K}[\![B_2]\!]$$

In this case we know by assumption that $[t_2/x]B_2 \in \mathsf{Kind}$ which implies that $B_2 \in \mathsf{Kind}$, and this further implies that $((x :_{(s,r)} B_1) \to B_2) \in \mathsf{Kind}$. However, there are two cases for $B_1$, either $B_1 \in \mathsf{Kind}$ or $B_1 \in \mathsf{Type}$. We consider each case in turn.

Subcase 1: Suppose $B_1 \in \mathsf{Kind}$. Then we know by assumption that $\Delta \odot \Gamma \models \varepsilon$. Thus, we apply the IH to obtain:

IH(1): $[\![t_1]\!]_\varepsilon \in \mathcal{K}[\![(x :_{(s,r)} B_1) \to B_2]\!] = \{f \mid f : \mathcal{K}[\![B_1]\!] \to \mathcal{K}[\![B_2]\!]\}$
IH(2): $[\![t_2]\!]_\varepsilon \in \mathcal{K}[\![B_1]\!]$

Using the IH's and Definition C.10 we know:
$$[\![t_1\, t_2]\!]_\varepsilon = [\![t_1]\!]_\varepsilon ([\![t_2]\!]_\varepsilon) \in \mathcal{K}[\![B_2]\!]$$

which was what was to be shown.

Subcase 2: Suppose $B_1 \in \mathsf{Type}$. Then we know by assumption that $\Delta \odot \Gamma \models \varepsilon$. Thus, we apply the IH to obtain:

IH(1): $[\![t_1]\!]_\varepsilon \in \mathcal{K}[\![(x :_{(s,r)} B_1) \to B_2]\!] = \mathcal{K}[\![B_2]\!]$

Using the IH's and Definition C.10 we know:
$$[\![t_1\, t_2]\!]_\varepsilon = [\![t_1]\!]_\varepsilon \in \mathcal{K}[\![B_2]\!]$$

which was what was to be shown.

Case 7:
$$\frac{\begin{array}{c}(\Delta, \sigma_3 \mid \sigma_5, r \mid \mathbf{0}) \odot \Gamma, x : B_1 \vdash B_2 : \mathsf{Type}_l \\ (\Delta \mid \sigma_4 \mid \sigma_3) \odot \Gamma \vdash t_1 : B_1 \\ (\Delta \mid \sigma_6 \mid \sigma_5 + r * \sigma_4) \odot \Gamma \vdash t_2 : [t_1/x]B_2\end{array}}{(\Delta \mid \sigma_4 + \sigma_6 \mid \sigma_3 + \sigma_5) \odot \Gamma \vdash (t_1, t_2) : (x :_r B_1) \otimes B_2} \text{ T\_Pair}$$

In this case we know that:
$$\sigma_1 = \sigma_4 + \sigma_6$$
$$\sigma_2 = \sigma_3 + \sigma_5$$
$$t = (t_1, t_2)$$
$$A = (x :_r B_1) \otimes B_2$$

It suffices to show that:
$$[\![(t_1, t_2)]\!]_\varepsilon \in \mathcal{K}[\![(x :_r B_1) \otimes B_2]\!]$$

By assumption we know that $((x :_r B_1) \otimes B_2) \in \mathsf{Kind}$. Thus, it must be the case that either $B_1, B_2 \in \mathsf{Kind}$, $B_1 \in \mathsf{Kind}$ and $B_2 \in \mathsf{Type}$, or $B_1 \in \mathsf{Type}$ and $B_2 \in \mathsf{Kind}$. We consider each of these cases in turn.

Subcase 1: Suppose $B_1, B_2 \in \mathsf{Kind}$ so $l = 1$. Then by Definition C.8 and Definition C.10 it suffices to show that:
$$\begin{aligned}[\![(t_1, t_2)]\!]_\varepsilon &= ([\![t_1]\!]_\varepsilon, [\![t_2]\!]_\varepsilon) \\ &\in \mathcal{K}[\![(x :_r B_1) \otimes B_2]\!] \\ &= \mathcal{K}[\![B_1]\!] \times \mathcal{K}[\![B_2]\!]\end{aligned}$$

We know by assumption that $\Delta \odot \Gamma \models \varepsilon$ and $(\Delta \mid \sigma_4 \mid \sigma_3) \odot \Gamma \vdash t_1 : B_1$. By kinding for typing (Lemma 3.6) we know that $(\Delta \mid \sigma_3 \mid \mathbf{0}) \odot \Gamma \vdash B_1 : \mathsf{Type}_1$, and by assumption we know that $(\Delta, \sigma_3 \mid \sigma_5, r \mid \mathbf{0}) \odot (\Gamma, x : B_1) \vdash B_2 : \mathsf{Type}_1$. By substitution for typing (Lemma 3.27), we know that $(\Delta \mid \sigma_5 + r * \sigma_4 \mid \mathbf{0}) \odot \Gamma \vdash [t_1/x]B_2 : \mathsf{Type}_1$. We can now apply the IH to the premises for $t_1$ and $t_2$ using the fact that we know $(\Delta \mid \sigma_3 \mid \mathbf{0}) \odot \Gamma \vdash B_1 : \mathsf{Type}_1$ and $(\Delta \mid \sigma_5 + r * \sigma_4 \mid \mathbf{0}) \odot \Gamma \vdash [t_1/x]B_2 : \mathsf{Type}_1$.

Thus, by the IH:

IH(1): $[\![t_1]\!]_\varepsilon \in \mathcal{K}[\![B_1]\!]$

IH(2): $[\![t_2]\!]_\varepsilon \in \mathcal{K}[\![[t_1/x]B_2]\!] = \mathcal{K}[\![B_2]\!]$

Therefore, $([\![t_1]\!]_\varepsilon, [\![t_2]\!]_\varepsilon) \in \mathcal{K}[\![B_1]\!] \times \mathcal{K}[\![B_2]\!]$.

Subcase 2: Suppose $B_1 \in \mathsf{Kind}$ and $B_2 \in \mathsf{Type}$. Then by Definition C.8 and Definition C.10 it suffices to show that:
$$\begin{aligned}[\![(t_1, t_2)]\!]_\varepsilon &= [\![t_1]\!]_\varepsilon \\ &\in \mathcal{K}[\![(x :_r B_1) \otimes B_2]\!] \\ &= \mathcal{K}[\![B_1]\!]\end{aligned}$$

We now by assumption that $\Delta \odot \Gamma \models \varepsilon$ and $(\Delta \mid \sigma_4 \mid \sigma_3) \odot \Gamma \vdash t_1 : B_1$. By kinding for typing (Lemma 3.6) we know that $(\Delta \mid \sigma_3 \mid \mathbf{0}) \odot \Gamma \vdash B_1 : \mathsf{Type}_1$. Thus, by the IH:

IH(1): $[\![t_1]\!]_\varepsilon \in \mathcal{K}[\![B_1]\!]$

which is what was to be shown.

Subcase 3: Suppose $B_1 \in \mathsf{Type}$ and $B_2 \in \mathsf{Kind}$. Then by Definition C.8 and Definition C.10 it suffices to show that:
$$\begin{aligned}[\![(t_1, t_2)]\!]_\varepsilon &= [\![t_2]\!]_\varepsilon \\ &\in \mathcal{K}[\![(x :_r B_1) \otimes B_2]\!] \\ &= \mathcal{K}[\![B_2]\!]\end{aligned}$$

We know by assumption that $\Delta \odot \Gamma \models \varepsilon \ (\Delta \mid \sigma_4 \mid \sigma_3) \odot \Gamma \vdash t_1 : B_1$, and $(\Delta, \sigma_3 \mid \sigma_5, r \mid \mathbf{0}) \odot (\Gamma, x : B_1) \vdash B_2 : \mathsf{Type}_1$. By substitution for typing (Lemma 3.27), we know that $(\Delta \mid \sigma_5 + r * \sigma_4 \mid \mathbf{0}) \odot \Gamma \vdash [t_1/x]B_2 : \mathsf{Type}_1$. We can now apply the IH to the premises for $t_2$ using the fact that we know $(\Delta \mid \sigma_5 + r * \sigma_4 \mid \mathbf{0}) \odot \Gamma \vdash [t_1/x]B_2 : \mathsf{Type}_1$.

Thus, by the IH:

IH(1): $[\![t_2]\!]_\varepsilon \in \mathcal{K}[\![[t_1/x]B_2]\!] = \mathcal{K}[\![B_2]\!]$

which is what was to be shown.

Case 8:
$$\frac{\begin{array}{c}(\Delta \mid \sigma_5 \mid \sigma_3 + \sigma_4) \odot \Gamma \vdash t_1 : (x :_r B_1) \otimes B_2 \\ (\Delta, (\sigma_3 + \sigma_4) \mid \sigma_7, r' \mid \mathbf{0}) \odot \Gamma, z : (x :_r B_1) \otimes B_2 \vdash C : \mathsf{Type}_1 \\ (\Delta, \sigma_3, (\sigma_4, r) \mid \sigma_6, s, s \mid \sigma_7, r', r') \odot \Gamma, x : B_1, y : B_2 \vdash t_2 : [(x, y)/z]C\end{array}}{(\Delta \mid \sigma_6 + s * \sigma_5 \mid \sigma_7 + r' * \sigma_5) \odot \Gamma \vdash \mathsf{let}\,(x : B_1, y : B_2) = t_1\,\mathsf{in}\,t_2 : [t_1/z]C} \text{ T\_TenCut}$$

In this case we know that:
$$\sigma_1 = \sigma_6 + s * \sigma_5$$
$$\sigma_2 = \sigma_7 + r' * \sigma_5$$
$$t = (\text{let } (x : B_1, y : B_2) = t_1 \text{ in } t_2)$$
$$A = [t_1/z]C$$

It suffices to show that:
$$[\![\text{let } (x : B_1, y : B_2) = t_1 \text{ in } t_2]\!]_\varepsilon \in \mathcal{K}[\![[t_1/z]C]\!] = \mathcal{K}[\![C]\!]$$

Based on these assumptions, it must be the case that $C \in \text{Kind}$ and $t_2 \in \text{Type}$. First, we can conclude by kinding for typing that:
$$(\Delta, \sigma_3, (\sigma_4, r) \mid \sigma_{11} \mid \mathbf{0}) \odot (\Gamma, x : B_1, y : B_2) \vdash [(x,y)/z]C : \text{Type}_1$$

for some vector $\sigma_1$. Then by well-formed contexts for typing we know that:
$$(\Delta, \sigma_3, (\sigma_4, r)) \odot (\Gamma, x : B_1, y : B_2) \vdash$$

This then implies that:
$$(\Delta \mid \sigma_3 \mid \mathbf{0}) \odot \Gamma \vdash B_1 : \text{Type}_{l_1} (\Delta, \sigma_3 \mid \sigma_4, r \mid \mathbf{0}) \odot (\Gamma, x : B_1) \vdash B_2 : \text{Type}_{l_2}$$

We will use these to apply the IH and define the typing evaluation below.

We must now consider cases for when $B_1, B_2 \in \text{Kind}$, $B_1 \in \text{Kind}$ and $B_2 \in \text{Type}$, $B_1 \in \text{Type}$ and $B_2 \in \text{Kind}$, and $B_1, B_2 \in \text{Type}$. We consider each case in turn.

Subcase 1: Suppose $B_1, B_2 \in \text{Kind}$. Then we know:
$$(\Delta \mid \sigma_3 \mid \mathbf{0}) \odot \Gamma \vdash B_1 : \text{Type}_1$$
$$(\Delta, \sigma_3 \mid \sigma_4, r \mid \mathbf{0}) \odot (\Gamma, x : B_1) \vdash B_2 : \text{Type}_1$$
$$(\Delta \mid \sigma_3 + \sigma_4 \mid \mathbf{0}) \odot \Gamma \vdash (x :_r B_1) \otimes B_2 : \text{Type}_1$$

By Definition C.8 and Definition C.10 it suffices to show:
$$[\![\text{let } (x : B_1, y : B_2) = t_1 \text{ in } t_2]\!]_\varepsilon$$
$$= [\![t_2]\!]_{\varepsilon[x \mapsto \pi_1 [\![t_1]\!]_\varepsilon, y \mapsto \pi_2 [\![t_1]\!]_\varepsilon]}$$
$$\in \mathcal{K}[\![[t_1/z]C]\!]$$
$$= \mathcal{K}[\![C]\!]$$

By the applying the IH to the premise for $t_1$ and Definition C.8:

IH(1): $[\![t_1]\!]_\varepsilon \in \mathcal{K}[\![(x :_r B_1) \otimes B_2]\!] = \mathcal{K}[\![B_1]\!] \times \mathcal{K}[\![B_2]\!]$

Hence, $\pi_1 [\![t_1]\!]_\varepsilon \in \mathcal{K}[\![B_1]\!]$ and $\pi_2 [\![t_1]\!]_\varepsilon \in \mathcal{K}[\![B_2]\!]$. This along with the kinding judgments given above imply that
$$(\Delta, \sigma_3, (\sigma_4, r)) \odot (\Gamma, x : B_1, y : B_2) \models \varepsilon[x \mapsto \pi_1 [\![t_1]\!]_\varepsilon, y \mapsto \pi_2 [\![t_1]\!]_\varepsilon]$$

given the assumption $\Delta \odot \Gamma \models \varepsilon$. We also know:
$$(\Delta, \sigma_3, (\sigma_4, r) \mid \sigma_{12} \mid \mathbf{0}) \odot (\Gamma, x : B_1, y : B_2) \vdash [(x,y)/z]C : \text{Type}_1$$

by applying kinding for typing to the premise for $t_2$. We now have everything we need to apply the IH to the premise for $t_2$:

IH(2): $[\![t_2]\!]_{\varepsilon[x \mapsto \pi_1 [\![t_1]\!]_\varepsilon, y \mapsto \pi_2 [\![t_1]\!]_\varepsilon]} \in \mathcal{K}[\![[(x,y)/z]C]\!] = \mathcal{K}[\![C]\!]$

Now by Definition C.10 and the previous results:

$$\llbracket \mathsf{let}\,(x:B_1, y:B_2) = t_1\,\mathsf{in}\,t_2 \rrbracket_\varepsilon$$
$$= \llbracket t_2 \rrbracket_{\varepsilon[x \mapsto \pi_1\,\llbracket t_1 \rrbracket_\varepsilon, y \mapsto \pi_2\,\llbracket t_1 \rrbracket_\varepsilon]}$$
$$\in K[|C|]$$

which was what was to be shown.

Subcase 2: Suppose $B_1 \in \mathsf{Kind}$ and $B_2 \in \mathsf{Type}$. Then we know:

$$(\Delta \mid \sigma_3 \mid \mathbf{0}) \odot \Gamma \vdash B_1 : \mathsf{Type}_1$$
$$(\Delta, \sigma_3 \mid \sigma_4, r \mid \mathbf{0}) \odot (\Gamma, x : B_1) \vdash B_2 : \mathsf{Type}_0$$
$$(\Delta \mid \sigma_3 + \sigma_4 \mid \mathbf{0}) \odot \Gamma \vdash (x :_r B_1) \otimes B_2 : \mathsf{Type}_1$$

By Definition C.8 and Definition C.10 it suffices to show:

$$\llbracket \mathsf{let}\,(x:B_1, y:B_2) = t_1\,\mathsf{in}\,t_2 \rrbracket_\varepsilon$$
$$= \llbracket t_2 \rrbracket_{\varepsilon[x \mapsto \llbracket t_1 \rrbracket_\varepsilon]}$$
$$\in \mathcal{K}\llbracket [t_1/z]C \rrbracket$$
$$= \mathcal{K}\llbracket C \rrbracket$$

By the applying the IH to the premise for $t_1$ and Definition C.8:

IH(1): $\llbracket t_1 \rrbracket_\varepsilon \in \mathcal{K}\llbracket (x :_r B_1) \otimes B_2 \rrbracket = \mathcal{K}\llbracket B_1 \rrbracket$

This along with the kinding judgments given above imply that

$$(\Delta, \sigma_3, (\sigma_4, r)) \odot (\Gamma, x : B_1, y : B_2) \models \varepsilon[x \mapsto \llbracket t_1 \rrbracket_\varepsilon]$$

given the assumption $\Delta \odot \Gamma \models \varepsilon$. We also know:

$$(\Delta, \sigma_3, (\sigma_4, r) \mid \sigma_{12} \mid \mathbf{0}) \odot (\Gamma, x : B_1, y : B_2) \vdash [(x,y)/z]C : \mathsf{Type}_1$$

by applying kinding for typing to the premise for $t_2$. We now have everything we need to apply the IH to the premise for $t_2$:

IH(2): $\llbracket t_2 \rrbracket_{\varepsilon[x \mapsto \llbracket t_1 \rrbracket_\varepsilon]} \in \mathcal{K}\llbracket [(x,y)/z]C \rrbracket = \mathcal{K}\llbracket C \rrbracket$

Now by Definition C.10 and the previous results:

$$\llbracket \mathsf{let}\,(x:B_1, y:B_2) = t_1\,\mathsf{in}\,t_2 \rrbracket_\varepsilon$$
$$= \llbracket t_2 \rrbracket_{\varepsilon[x \mapsto \llbracket t_1 \rrbracket_\varepsilon]}$$
$$\in K[|C|]$$

which was what was to be shown.

Subcase 3: Suppose $B_1 \in \mathsf{Type}$ and $B_2 \in \mathsf{Kind}$. This case is similar to the previous case using:

$$(\Delta, \sigma_3, (\sigma_4, r)) \odot (\Gamma, x : B_1, y : B_2) \models \varepsilon[x \mapsto \llbracket t_1 \rrbracket_\varepsilon]$$

Subcase 4: Suppose $B_1, B_2 \in \mathsf{Type}$. This case is similar to the previous case using:

$$(\Delta, \sigma_3, (\sigma_4, r)) \odot (\Gamma, x : B_1, y : B_2) \models \varepsilon$$

Case 9:

$$\frac{(\Delta \mid \sigma \mid \mathbf{0}) \odot \Gamma \vdash B : \mathsf{Type}_0}{(\Delta \mid \sigma \mid \mathbf{0}) \odot \Gamma \vdash \Box_s B : \mathsf{Type}_0}\ \text{T\_Box}$$

This case follows directly from the induction hypothesis.

Case 10:
$$\frac{(\Delta \mid \sigma_3 \mid \sigma_4) \odot \Gamma \vdash t' : B}{(\Delta \mid s * \sigma_3 \mid \sigma_4) \odot \Gamma \vdash \Box t' : \Box_s B} \text{ T\_BoxI}$$

This case follows directly from the induction hypothesis.

Case 11:
$$\frac{\begin{array}{c}(\Delta \mid \sigma_3 \mid \sigma_7) \odot \Gamma \vdash t_1 : \Box_s B_1 \\ (\Delta, \sigma_7 \mid \sigma_5, s \mid \sigma_6, (s * r)) \odot \Gamma, x : B_1 \vdash t_2 : [\Box x/z]B_2 \\ (\Delta, \sigma_7 \mid \sigma_6, r \mid \mathbf{0}) \odot \Gamma, z : \Box_s B_1 \vdash B_2 : \text{Type}_l\end{array}}{(\Delta \mid \sigma_3 + \sigma_5 \mid \sigma_6 + r * \sigma_3) \odot \Gamma \vdash \text{let}\,\Box(x : B_1) = t_1 \,\text{in}\, t_2 : [t_1/z]B_2} \text{ T\_BoxE}$$

In this case we know that:
$$\sigma_1 = \sigma_3 + \sigma_5$$
$$\sigma_2 = \sigma_6 + r * \sigma_3$$
$$t = (\text{let}\,\Box(x : B_1) = t_1 \,\text{in}\, t_2)$$
$$A = [t_1/z]B_2$$

It suffices to show that:
$$[\![\text{let}\,\Box(x : B_1) = t_1 \,\text{in}\, t_2]\!]_\varepsilon \in \mathcal{K}[\![[t_1/z]B_2]\!] = \mathcal{K}[\![B_2]\!]$$

In this case we know that $([t_1/z]B_2) \in \text{Kind}$, and thus, $B_2 \in \text{Kind}$ and either $B_1 \in \text{Kind}$ or $B_1 \in \text{Type}$. We cover both of these cases in turn.

Subcase 1: Suppose $B_1 \in \text{Kind}$. It suffices to show that:
$$[\![\text{let}\,\Box(x : B_1) = t_1 \,\text{in}\, t_2]\!]_\varepsilon = [\![t_2]\!]_{\varepsilon[x \mapsto [\![t_1]\!]_\varepsilon]}$$
$$\in \mathcal{K}[\![[t_1/z]B_2]\!]$$
$$= \mathcal{K}[\![B_2]\!]$$

As we have seen in the previous cases we can apply well-formed contexts for typing to obtain that:
$$(\Delta \mid \sigma_7 \mid \mathbf{0}) \odot \Gamma \vdash \Box_s B_1 : \text{Type}_1$$

We can now apply the IH to the premise for $t_1$ to obtain:

IH(1): $[\![t_1]\!]_\varepsilon \in \mathcal{K}[\![\Box_s B_1]\!] = \mathcal{K}[\![B_1]\!]$

Using the previous two facts along with the assumption that $\Delta \odot \Gamma \models \varepsilon$ we may obtain
$$(\Delta, \sigma_7) \odot (\Gamma, x : B_1) \models \varepsilon[x \mapsto [\![t_1]\!]_\varepsilon]$$

In addition, we know that $(\Delta, \sigma_7 \mid \sigma_6, r \mid \mathbf{0}) \odot \Gamma, z : \Box_s B_1 \vdash B_2 : \text{Type}_1$. Thus, we can now apply the induction hypothesis a second time.

IH(2): $[\![t_2]\!]_{\varepsilon[x \mapsto [\![t_1]\!]_\varepsilon]} \in \mathcal{K}[\![B_2]\!]$

which was what was to be shown.

Subcase 2: Suppose $B_1 \in \text{Type}$. It suffices to show that:
$$[\![\text{let}\,\Box(x : B_1) = t_1 \,\text{in}\, t_2]\!]_\varepsilon = [\![t_2]\!]_\varepsilon$$
$$\in \mathcal{K}[\![[t_1/z]B_2]\!]$$
$$= \mathcal{K}[\![B_2]\!]$$

As we have seen in the previous cases we can apply well-formed contexts for typing to obtain that:
$$(\Delta \mid \sigma_7 \mid \mathbf{0}) \odot \Gamma \vdash \Box_s B_1 : \text{Type}_0$$

Using this along with the assumption that $\Delta \odot \Gamma \models \varepsilon$ we may obtain
$$(\Delta, \sigma_7) \odot (\Gamma, x : B_1) \models \varepsilon$$

In addition, we know that $(\Delta, \sigma_7 \mid \sigma_6, r \mid \mathbf{0}) \odot \Gamma, z : \Box_s B_1 \vdash B_2 : \mathsf{Type}_1$. Thus, we can now apply the induction hypothesis a second time.

IH(2): $[\![t_2]\!]_\varepsilon \in \mathcal{K}[\![B_2]\!]$

which was what was to be shown.

**Case 12:**
$$\dfrac{(\Delta \mid \sigma_1 \mid \sigma_2) \odot \Gamma \vdash t : A' \quad (\Delta \mid \sigma_2 \mid \mathbf{0}) \odot \Gamma \vdash A' = A : \mathsf{Type}_1}{(\Delta \mid \sigma_1 \mid \sigma_2) \odot \Gamma \vdash t : A} \; \text{T\_TyConv}$$

This case follows by first applying the induction hypothesis to the typing premise, and then applying Lemma C.15 to obtain that $\mathcal{K}[\![A]\!] = \mathcal{K}[\![A']\!]$ obtaining our result.

We now move onto the second part of this result assuming the first. In this part we will show:

If $\Delta \odot \Gamma \models \varepsilon$ and $(\Delta \mid \sigma \mid \mathbf{0}) \odot \Gamma \vdash A : \mathsf{Type}_1$, then $[\![A]\!]_\varepsilon \in \mathcal{K}[\![A]\!] \to \mathsf{SAT}$.

Recall that this is a proof by mutual induction on $(\Delta \mid \sigma_1 \mid \sigma_2) \odot \Gamma \vdash t : A$ and $(\Delta \mid \sigma \mid \mathbf{0}) \odot \Gamma \vdash A : \mathsf{Type}_1$.

**Case 1:**
$$\dfrac{\Delta \odot \Gamma \vdash}{(\Delta \mid \mathbf{0} \mid \mathbf{0}) \odot \Gamma \vdash \mathsf{Type}_0 : \mathsf{Type}_1} \; \text{T\_Type}$$

In this case we know that:
$$\sigma = \mathbf{0}$$
$$A = \mathsf{Type}_0$$

It suffices to show that:
$$[\![\mathsf{Type}_0]\!]_\varepsilon \in (\mathcal{K}[\![\mathsf{Type}_0]\!] \to \mathsf{SAT}) = (\mathsf{SAT} \to \mathsf{SAT})$$

But, $[\![\mathsf{Type}_0]\!]_\varepsilon = \lambda x \in \mathsf{SAT}.\mathsf{SN}$, and by Lemma C.7 $\mathsf{SN} \in \mathsf{SAT}$; hence, $(\lambda x \in \mathsf{SAT}.\mathsf{SN}) \in (\mathsf{SAT} \to \mathsf{SAT})$.

**Case 2:**
$$\dfrac{\Delta_1, \mathbf{0}, \Delta_2 \odot \Gamma_1, x : \mathsf{Type}_1, \Gamma_2 \vdash \quad |\Delta_1| = |\Gamma_1|}{(\Delta_1, \mathbf{0}, \Delta_2 \mid \mathbf{0}^{|\Delta_1|}, 1, \mathbf{0} \mid \mathbf{0}) \odot (\Gamma_1, x : \mathsf{Type}_1, \Gamma_2) \vdash x : \mathsf{Type}_1} \; \text{T\_Var}$$

This case is impossible, because the well-formed context premise fails, because $\mathsf{Type}_1$ has no type.

**Case 3:**
$$\dfrac{(\Delta \mid \sigma_3 \mid \mathbf{0}) \odot \Gamma \vdash B_1 : \mathsf{Type}_{l_1} \quad (\Delta, \sigma_3 \mid \sigma_4, r \mid \mathbf{0}) \odot \Gamma, x : B_1 \vdash B_2 : \mathsf{Type}_{l_2}}{(\Delta \mid \sigma_3 + \sigma_4 \mid \mathbf{0}) \odot \Gamma \vdash (x :_{(s,r)} B_1) \to B_2 : \mathsf{Type}_1} \; \text{T\_Arrow}$$

In this case we know that:
$$\sigma = \sigma_3 + \sigma_4$$
$$A = (x :_{(s,r)} B_1) \to B_2$$

It suffices to show that:
$$[\![(x :_{(s,r)} B_1) \to B_2]\!]_\varepsilon \in (\mathcal{K}[\![(x :_{(s,r)} B_1) \to B_2]\!] \to \mathsf{SAT})$$

In this case either $B_1, B_2 \in \mathsf{Kind}$, $B_1 \in \mathsf{Kind}$ and $B_2 \in \mathsf{Type}$, or $B_1 \in \mathsf{Type}$ and $B_2 \in \mathsf{Kind}$. We consider each case in turn.

Subcase 1: Suppose $B_1, B_2 \in \mathsf{Kind}$. It suffices to show:

$$[\![(x :_{(s,r)} B_1) \to B_2]\!]_\varepsilon$$
$$= \lambda X \in \mathcal{K}[\![B_1]\!] \to \mathcal{K}[\![B_2]\!]. \bigcap_{Y \in \mathcal{K}[\![B_1]\!]} (([\![B_1]\!]_\varepsilon)(Y) \to ([\![B_2]\!]_{\varepsilon[x \mapsto Y]})(X(Y)))$$
$$\in (\mathcal{K}[\![(x :_{(s,r)} B_1) \to B_2]\!] \to \mathsf{SAT})$$
$$= \{f \mid f : \mathcal{K}[\![B_1]\!] \to \mathcal{K}[\![B_2]\!]\} \to \mathsf{SAT}$$

Now suppose $X \in (\mathcal{K}[\![B_1]\!] \to \mathcal{K}[\![B_2]\!])$ and $Y \in \mathcal{K}[\![B_1]\!]$. We know by assumption that $(\Delta \mid \sigma_3 \mid \mathbf{0}) \odot \Gamma \vdash B_1 : \mathsf{Type}_{l_1}$ and $\Delta \odot \Gamma \models \varepsilon$, and so we can apply the induction hypothesis to the premise for $B_1$ to obtain:

IH(1): $[\![B_1]\!]_\varepsilon \in (\mathcal{K}[\![B_1]\!] \to \mathsf{SAT})$

The previous facts now allow us, by Definition C.9, to obtain:

$$(\Delta, \sigma_3) \odot (\Gamma, x : B_1) \models \varepsilon[x \mapsto Y]$$

Thus, we can now apply the induction hypothesis to the premise for $B_2$ to obtain:

IH(2): $[\![B_2]\!]_{\varepsilon[x \mapsto Y]} \in (\mathcal{K}[\![B_2]\!] \to \mathsf{SAT})$

Then we know by IH(1) that $[\![B_1]\!]_\varepsilon(Y) \in \mathsf{SAT}$ and by IH(2) $[\![B_2]\!]_{\varepsilon[x \mapsto Y]}(X(Y)) \in \mathsf{SAT}$, thus:

$$([\![B_1]\!]_\varepsilon(Y) \to [\![B_2]\!]_{\varepsilon[x \mapsto Y]}(X(Y))) \in (\mathsf{SAT} \to \mathsf{SAT}) \in \mathsf{SAT}$$

Then by Lemma C.8:

$$\bigcap_{Y \in \mathcal{K}[\![B_1]\!]} ([\![B_1]\!]_\varepsilon(Y) \to [\![B_2]\!]_{\varepsilon[x \mapsto Y]}(X(Y))) \in \mathsf{SAT}$$

Therefore, we obtain our result.

Subcase 2: Suppose $B_1 \in \mathsf{Kind}$ and $B_2 \in \mathsf{Type}$. Similar to the previous case.

Subcase 3: Suppose $B_1 \in \mathsf{Type}$ and $B_2 \in \mathsf{Kind}$. It suffices to show:

$$[\![(x :_{(s,r)} B_1) \to B_2]\!]_\varepsilon$$
$$= \lambda X \in \mathcal{K}[\![B_2]\!].([\![B_1]\!]_\varepsilon \to ([\![B_2]\!]_\varepsilon)(X))$$
$$\in (\mathcal{K}[\![(x :_{(s,r)} B_1) \to B_2]\!] \to \mathsf{SAT})$$
$$= (\mathcal{K}[\![B_2]\!] \to \mathsf{SAT})$$

Now suppose $X \in \mathcal{K}[\![B_2]\!]$. We know by assumption that $(\Delta \mid \sigma_3 \mid \mathbf{0}) \odot \Gamma \vdash B_1 : \mathsf{Type}_0$ and $\Delta \odot \Gamma \models \varepsilon$, and so we can apply the first part of the induction hypothesis to the premise for $B_1$ to obtain:

IH(1): $[\![B_1]\!]_\varepsilon \in \mathcal{K}[\![\mathsf{Type}_0]\!] = \mathsf{SAT}$

Now we know by assumption that $(\Delta, \sigma_3 \mid \sigma_4, r \mid \mathbf{0}) \odot \Gamma, x : B_1 \vdash B_2 : \mathsf{Type}_0$, and we can now show by Definition C.9 that $(\Delta, \sigma_3) \odot (\Gamma, x : B_1) \models \varepsilon$ holds. So we can apply the IH to the former judgment to obtain:

IH(2): $[\![B_2]\!]_\varepsilon \in (\mathcal{K}[\![B_2]\!] \to \mathsf{SAT})$

At this point we can see that $(\lambda X \in \mathcal{K}[\![B_2]\!].([\![B_1]\!]_\varepsilon \to [\![B_2]\!]_\varepsilon(X))) \in \mathsf{SAT}$ by the previous facts, and the fact that $\mathsf{SAT}$ is closed under function spaces.

Case 4:

$$\frac{(\Delta \mid \sigma_3 \mid \mathbf{0}) \odot \Gamma \vdash B_1 : \mathsf{Type}_{l_1} \quad (\Delta, \sigma_3 \mid \sigma_4, r \mid \mathbf{0}) \odot \Gamma, x : B_1 \vdash B_2 : \mathsf{Type}_{l_2}}{(\Delta \mid \sigma_3 + \sigma_4 \mid \mathbf{0}) \odot \Gamma \vdash (x :_r B_1) \otimes B_2 : \mathsf{Type}_{l_1 \sqcup l_2}} \; \text{T\_Ten}$$

This case is similar to the previous case.

Case 5:
$$\frac{(\Delta \mid \sigma \mid \mathbf{0}) \odot \Gamma \vdash B : \mathsf{Type}_1}{(\Delta \mid \sigma \mid \mathbf{0}) \odot \Gamma \vdash \Box_s B : \mathsf{Type}_1} \text{ T\_Box}$$

This case follows from the IH.

Case 6:
$$\frac{(\Delta \mid \sigma_1 \mid \sigma_2) \odot \Gamma \vdash A : A' \quad (\Delta \mid \sigma_2 \mid \mathbf{0}) \odot \Gamma \vdash A' = \mathsf{Type}_1 : B}{(\Delta \mid \sigma_1 \mid \sigma_2) \odot \Gamma \vdash A : \mathsf{Type}_1} \text{ T\_TyConv}$$

This case is impossible, because $\mathsf{Type}_1$ has no type $B$.

$\Box$

**Theorem C.17** (Soundness for Semantic Typing). *If $(\Delta \mid \sigma_1 \mid \sigma_2) \odot \Gamma \vdash t : A$, then $(\Delta \mid \sigma_1 \mid \sigma_2) \odot \Gamma \models t : A$.*

*Proof.* This is a proof by induction on the assumed typing derivation.

Case 1:
$$\frac{\Delta \odot \Gamma \vdash}{(\Delta \mid \mathbf{0} \mid \mathbf{0}) \odot \Gamma \vdash \mathsf{Type}_0 : \mathsf{Type}_1} \text{ T\_Type}$$

In this case we have:
$$\sigma_1 = \mathbf{0}$$
$$\sigma_2 = \mathbf{0}$$
$$t = \mathsf{Type}_0$$
$$A = \mathsf{Type}_1$$

We can now see that this case holds trivially, because $\mathsf{Type}_1$ has no type.

Case 2:
$$\frac{\Delta_1, \sigma, \Delta_2 \odot \Gamma_1, x : A, \Gamma_2 \vdash \quad |\Delta_1| = |\Gamma_1|}{(\Delta_1, \sigma, \Delta_2 \mid \mathbf{0}^{|\Delta_1|}, 1, \mathbf{0} \mid \sigma, 0, \mathbf{0}) \odot \Gamma_1, x : A, \Gamma_2 \vdash x : A} \text{ T\_Var}$$

In this case we have:
$$\Delta = (\Delta_1, \sigma, \Delta_2)$$
$$\sigma_1 = \mathbf{0}^{|\Delta_1|}, 1, \mathbf{0}$$
$$\sigma_2 = \sigma, 0, \mathbf{0}$$
$$t = x$$

Now either $A \in \mathsf{Kind}$ or $A \in \mathsf{Type}$. We consider both cases in turn.

Subcase 1: Suppose $A \in \mathsf{Kind}$ and $\Delta \odot \Gamma \models_\varepsilon \rho$. It suffices to show:
$$(\!(x)\!)_\rho = \rho\, x \in [\![A]\!]_\varepsilon\ ([\![x]\!]_\varepsilon) = [\![A]\!]_\varepsilon\ (\varepsilon\, x)$$

But, this holds by Definition C.12, because the well-formed context premise above implies the proper kinding of $A$.

Subcase 2: Suppose $A \in \mathsf{Type}$ and $\Delta \odot \Gamma \models_\varepsilon \rho$. It suffices to show:
$$(\!(x)\!)_\rho = \rho\, x \in [\![A]\!]_\varepsilon$$

But, this holds by Definition C.12, because the well-formed context premise above implies the proper kinding of $A$.

Case 3:
$$\frac{(\Delta \mid \sigma_3 \mid \mathbf{0}) \odot \Gamma \vdash B_1 : \mathsf{Type}_0 \quad (\Delta, \sigma_3 \mid \sigma_4, r \mid \mathbf{0}) \odot \Gamma, x : B_1 \vdash B_2 : \mathsf{Type}_0}{(\Delta \mid \sigma_3 + \sigma_4 \mid \mathbf{0}) \odot \Gamma \vdash (x :_{(s,r)} B_1) \to B_2 : \mathsf{Type}_0} \text{ T\_ARROW}$$

In this case we have:
$$\sigma_1 = (\sigma_3 + \sigma_4)$$
$$\sigma_2 = \mathbf{0}$$
$$t = (x :_{(s,r)} B_1) \to B_2$$
$$A = \mathsf{Type}_0$$

We only need to consider the first case of this theorem, because $\mathsf{Type}_1$ has no type. So suppose $\Delta \odot \Gamma \models_\varepsilon \rho$. It suffices to show:

$$\begin{aligned}
&(\!|(x :_{(s,r)} B_1) \to B_2|\!)_\rho \\
&= ((x :_{(s,r)} (\!|B_1|\!)_\rho) \to (\!|B_2|\!)_\rho) \\
&\in [\![\mathsf{Type}_0]\!]_\varepsilon ([\![(x :_{(s,r)} B_1) \to B_2]\!]_\varepsilon) \\
&= \mathsf{SN}
\end{aligned}$$

We know by assumption that $\Delta \odot \Gamma \models_\varepsilon \rho$ so we can apply the IH to conclude:

IH(1): $(\!|B_1|\!)_\rho \in [\![\mathsf{Type}_0]\!]_\varepsilon ([\![B_1]\!]_\varepsilon) = \mathsf{SN}$

Now suppose $t \in [\![B_1]\!]_\varepsilon$. Then we know by Definition C.12 that $(\Delta, \sigma_3) \odot (\Gamma, x : B_1) \models_\varepsilon \rho[x \mapsto t]$. Thus, by applying the IH to the premise for $B_2$ we may conclude that

IH(2): $(\!|B_2|\!)_{\rho[x \mapsto t]} \in [\![\mathsf{Type}_0]\!]_\varepsilon ([\![B_2]\!]_\varepsilon) = \mathsf{SN}$

holds for every $t$. Therefore, we may conclude out result.

Case 4:
$$\frac{(\Delta \mid \sigma_3 \mid \mathbf{0}) \odot \Gamma \vdash B_1 : \mathsf{Type}_0 \quad (\Delta, \sigma_3 \mid \sigma_4, r \mid \mathbf{0}) \odot \Gamma, x : B_1 \vdash B_2 : \mathsf{Type}_0}{(\Delta \mid \sigma_3 + \sigma_4 \mid \mathbf{0}) \odot \Gamma \vdash (x :_r B_1) \otimes B_2 : \mathsf{Type}_0} \text{ T\_TEN}$$

Similar to the previous case.

Case 5:
$$\frac{(\Delta, \sigma_3 \mid \sigma_4, s \mid \sigma_5, r) \odot \Gamma, x : B_1 \vdash t' : B_2}{(\Delta \mid \sigma_4 \mid \sigma_3 + \sigma_5) \odot \Gamma \vdash \lambda x : B_1.t' : (x :_{(s,r)} B_1) \to B_2} \text{ T\_FUN}$$

In this case we have:
$$\sigma_1 = \sigma_4$$
$$\sigma_2 = (\sigma_3 + \sigma_5)$$
$$t = \lambda x : B_1.t'$$
$$A = (x :_{(s,r)} B_1) \to B_2$$

We have two cases to consider, either $((x :_{(s,r)} B_1) \to B_2) \in \mathsf{Kind}$ or $((x :_{(s,r)} B_1) \to B_2) \in \mathsf{Type}$. We cover both cases in turn.

Subcase 1: Suppose $((x :_{(s,r)} B_1) \to B_2) \in \mathsf{Kind}$. We now have three subcases depending on the typing for both $B_1$ and $B_2$.

Subcase 1: Suppose $B_1, B_2 \in \mathsf{Kind}$. It suffices to show:

$$\begin{aligned}
&(\!|\lambda x : B_1.t'|\!)_\rho \\
&= \lambda x : (\!|B_1|\!)_\rho.(\!|t'|\!)_\rho \\
&\in [\![(x :_{(s,r)} B_1) \to B_2]\!]_\varepsilon ([\![\lambda x : B_1.t']\!]_\varepsilon) \\
&= \bigcap_{Y \in \mathcal{K}[\![B_1]\!]} (([\![B_1]\!]_\varepsilon)(Y) \to ([\![B_2]\!]_{\varepsilon[x \mapsto Y]}) ([\![\lambda x : B_1.t']\!]_\varepsilon (Y))) \\
&= \bigcap_{Y \in \mathcal{K}[\![B_1]\!]} (([\![B_1]\!]_\varepsilon)(Y) \to ([\![B_2]\!]_{\varepsilon[x \mapsto Y]}) ([\![t']\!]_{\varepsilon[x \mapsto Y]}))
\end{aligned}$$

So suppose we have a $Y \in \mathcal{K}[\![B_1]\!]$ and a $t \in [\![B_1]\!]_\varepsilon (Y) = [\![B_1]\!]_{(\varepsilon[x \mapsto Y])} (\varepsilon\, x)$ (since $x \notin \mathsf{FV}(B_2)$). Then we know that $(\Delta, \sigma_3) \odot (\Gamma, x : B_1) \models_{\varepsilon[x \mapsto Y]} \rho[x \mapsto t]$. Then by the IH:

IH(1): $(\!|t'|\!)_{\rho[x \mapsto t]} \in [\![B_2]\!]_{\varepsilon[x \mapsto Y]} ([\![t']\!]_{\varepsilon[x \mapsto Y]})$

Thus, we obtain our result.

Subcase 2: Suppose $B_1 \in \mathsf{Kind}$ and $B_2 \in \mathsf{Type}$. This case is similar to the previous case, except we will use part two of the IH.

Subcase 3: Suppose $B_1 \in \mathsf{Type}$ and $B_2 \in \mathsf{Kind}$. Similar to the previous case.

Subcase 2: Suppose $((x :_{(s,r)} B_1) \to B_2) \in \mathsf{Type}$. This case is similar to the above, but we will use the second part of the IH.

Case 6:

$$\dfrac{(\Delta \mid \sigma_4 \mid \sigma_3 + \sigma_5) \odot \Gamma \vdash t_1 : (x :_{(s,r)} B_1) \to B_2 \qquad (\Delta \mid \sigma_6 \mid \sigma_3) \odot \Gamma \vdash t_2 : B_1}{(\Delta \mid \sigma_4 + s * \sigma_6 \mid \sigma_5 + r * \sigma_6) \odot \Gamma \vdash t_1\, t_2 : [t_2/x]B_2} \text{ T\_App}$$

In this case we have:
$$\sigma_1 = (\sigma_4 + s * \sigma_6)$$
$$\sigma_2 = (\sigma_5 + r * \sigma_6)$$
$$t = (t_1\, t_2)$$
$$A = [t_2/x]B_2$$

We have several cases to consider.

Subcase 1: Suppose $B_1, [t_2/x]B_2 \in \mathsf{Kind}$. It suffices to show:

$$(\!|t_1\, t_2|\!)_\rho$$
$$= (\!|t_1|\!)_\rho (\!|t_2|\!)_\rho$$
$$\in [\![[t_2/x]B_2]\!]_\varepsilon ([\![t_1\, t_2]\!]_\varepsilon)$$
$$= [\![B_2]\!]_{\varepsilon[x \mapsto [\![t_2]\!]_\varepsilon]} ([\![t_1]\!]_\varepsilon ([\![t_2]\!]_\varepsilon))$$

By the IH we have:

IH(1): $(\!|t_1|\!)_\rho \in [\![(x :_{(s,r)} B_1) \to B_2]\!]_\varepsilon ([\![t_1]\!]_\varepsilon)$

IH(2): $(\!|t_2|\!)_\rho \in [\![B_1]\!]_\varepsilon ([\![t_2]\!]_\varepsilon)$

Notice that by Lemma C.16 we know that $[\![t_2]\!]_\varepsilon \in \mathcal{K}[\![B_1]\!]$. So using C.10 we know:

$$(\!|t_1|\!)_\rho \in [\![(x :_{(s,r)} B_1) \to B_2]\!]_\varepsilon ([\![t_1]\!]_\varepsilon)$$
$$= \bigcap_{Y \in \mathcal{K}[\![A]\!]} (([\![A]\!]_\varepsilon)(Y) \to ([\![B]\!]_{\varepsilon[x \mapsto Y]})([\![t_1]\!]_\varepsilon (Y)))$$

Thus, $(\!|t_1|\!)_\rho (\!|t_2|\!)_\rho \in ([\![B]\!]_{\varepsilon[x \mapsto [\![t_2]\!]_\varepsilon]}) ([\![t_1]\!]_\varepsilon ([\![t_2]\!]_\varepsilon))$.

Subcase 2: Suppose $B_1 \in \mathsf{Kind}$ and $[t_2/x]B_2 \in \mathsf{Type}$. It suffices to show:

$$(\!|t_1\, t_2|\!)_\rho$$
$$= (\!|t_1|\!)_\rho (\!|t_2|\!)_\rho$$
$$\in [\![[t_2/x]B_2]\!]_\varepsilon$$
$$= [\![B_2]\!]_{\varepsilon[x \mapsto [\![t_2]\!]_\varepsilon]}$$

By the IH we have:

IH(1): $(\!|t_1|\!)_\rho \in [\![(x :_{(s,r)} B_1) \to B_2]\!]_\varepsilon ([\![t_1]\!]_\varepsilon)$

IH(2): $(\!|t_2|\!)_\rho \in [\![B_1]\!]_\varepsilon ([\![t_2]\!]_\varepsilon)$

Notice that by Lemma C.16 we know that $[\![t_2]\!]_\varepsilon \in \mathcal{K}[\![B_1]\!]$. So using C.10 we know:

$$\begin{aligned}
(\!|t_1|\!)_\rho &\in [\![(x :_{(s,r)} B_1) \to B_2]\!]_\varepsilon ([\![t_1]\!]_\varepsilon) \\
&= \bigcap_{Y \in \mathcal{K}[\![B_1]\!]} (([\![B_1]\!]_\varepsilon)(Y) \to ([\![B_2]\!]_{\varepsilon[x \mapsto Y]}))
\end{aligned}$$

Thus, $(\!|t_1|\!)_\rho (\!|t_2|\!)_\rho \in ([\![B_2]\!]_{\varepsilon[x \mapsto [\![t_2]\!]_\varepsilon]})$.

Subcase 3: Suppose $B_1 \in \mathsf{Type}$ and $[t_2/x]B_2 \in \mathsf{Type}$. It suffices to show:

$$\begin{aligned}
(\!|t_1\ t_2|\!)_\rho \\
= (\!|t_1|\!)_\rho (\!|t_2|\!)_\rho \\
\in [\![[t_2/x]B_2]\!]_\varepsilon \\
= [\![B_2]\!]_{\varepsilon[x \mapsto [\![t_2]\!]_\varepsilon]}
\end{aligned}$$

By the IH we have:

IH(1): $(\!|t_1|\!)_\rho \in [\![(x :_{(s,r)} B_1) \to B_2]\!]_\varepsilon$

IH(2): $(\!|t_2|\!)_\rho \in [\![B_1]\!]_\varepsilon$

Using C.10 we know:

$$\begin{aligned}
(\!|t_1|\!)_\rho &\in [\![(x :_{(s,r)} B_1) \to B_2]\!]_\varepsilon ([\![t_1]\!]_\varepsilon) \\
&= [\![B_1]\!]_\varepsilon \to [\![B_2]\!]_\varepsilon
\end{aligned}$$

Thus, $(\!|t_1|\!)_\rho (\!|t_2|\!)_\rho \in [\![B_2]\!]_\varepsilon$.

Case 7:

$$\frac{\begin{array}{c}(\Delta, \sigma_3 \mid \sigma_5, r \mid \mathbf{0}) \odot \Gamma, x : B_1 \vdash B_2 : \mathsf{Type}_l \\ (\Delta \mid \sigma_4 \mid \sigma_3) \odot \Gamma \vdash t_1 : B_1 \\ (\Delta \mid \sigma_6 \mid \sigma_5 + r * \sigma_4) \odot \Gamma \vdash t_2 : [t_1/x]B_2\end{array}}{(\Delta \mid \sigma_4 + \sigma_6 \mid \sigma_3 + \sigma_5) \odot \Gamma \vdash (t_1, t_2) : (x :_r B_1) \otimes B_2} \text{ T\_PAIR}$$

This case is similar to the case for $\lambda$-abstraction above.

Case 8:

$$\frac{\begin{array}{c}(\Delta \mid \sigma_5 \mid \sigma_3 + \sigma_4) \odot \Gamma \vdash t_1 : (x :_r B_1) \otimes B_2 \\ (\Delta, (\sigma_3 + \sigma_4) \mid \sigma_7, r' \mid \mathbf{0}) \odot \Gamma, z : (x :_r B_1) \otimes B_2 \vdash C : \mathsf{Type}_1 \\ (\Delta, \sigma_3, (\sigma_4, r) \mid \sigma_6, s, s \mid \sigma_7, r', r') \odot \Gamma, x : B_1, y : B_2 \vdash t_2 : [(x,y)/z]C\end{array}}{(\Delta \mid \sigma_6 + s * \sigma_5 \mid \sigma_7 + r' * \sigma_5) \odot \Gamma \vdash \mathsf{let}\ (x : B_1, y : B_2) = t_1\ \mathsf{in}\ t_2 : [t_1/z]C} \text{ T\_TENCUT}$$

Similar to the application case above.

Case 9:

$$\frac{(\Delta \mid \sigma \mid \mathbf{0}) \odot \Gamma \vdash B : \mathsf{Type}_0}{(\Delta \mid \sigma \mid \mathbf{0}) \odot \Gamma \vdash \Box_s B : \mathsf{Type}_0} \text{ T\_BOX}$$

This case follows from the IH.

Case 10:

$$\frac{(\Delta \mid \sigma_3 \mid \sigma_2) \odot \Gamma \vdash t' : B}{(\Delta \mid s * \sigma_3 \mid \sigma_2) \odot \Gamma \vdash \Box t' : \Box_s B} \text{ T\_BOXI}$$

In this case we have:

$$\begin{aligned}
\sigma_1 &= (s * \sigma_3) \\
t &= (\Box t') \\
A &= \Box_s B
\end{aligned}$$

We have two cases to consider.

120

Subcase 1: Suppose $B \in \mathsf{Kind}$. It suffices to show:

$$\begin{aligned}
& (\!|\Box t'|\!)_\rho \\
=\ & (\!|t'|\!)_\rho \\
\in\ & [\![\Box_s B]\!]_\varepsilon\,([\![\Box t']\!]_\varepsilon) \\
=\ & [\![B]\!]_\varepsilon\,([\![t']\!]_\varepsilon)
\end{aligned}$$

At this point, this case holds by the IH.

Subcase 2: Suppose $B \in \mathsf{Type}$. Similar to the previous case.

Case 11:

$$\dfrac{\begin{array}{c}(\Delta \mid \sigma_3 \mid \sigma_7) \odot \Gamma \vdash t_1 : \Box_s B_1 \\ (\Delta, \sigma_7 \mid \sigma_6, r \mid \mathbf{0}) \odot \Gamma, z : \Box_s B_1 \vdash B_2 : \mathsf{Type}_0 \\ (\Delta, \sigma_7 \mid \sigma_5, s \mid \sigma_6, (s * r)) \odot \Gamma, x : B_1 \vdash t_2 : [\Box x/z] B_2\end{array}}{(\Delta \mid \sigma_3 + \sigma_5 \mid \sigma_6 + r * \sigma_3) \odot \Gamma \vdash \mathsf{let}\,\Box(x : B_1) = t_1\,\mathsf{in}\,t_2 : [t_1/z] B_2}\ \text{T\_BoxE}$$

In this case we have:

$$\begin{aligned}
\sigma_1 &= (\sigma_3 + \sigma_5) \\
\sigma_2 &= (\sigma_6 + r * \sigma_3) \\
t &= (\mathsf{let}\,\Box(x : B_1) = t_1\,\mathsf{in}\,t_2) \\
A &= [t_1/z] B_2
\end{aligned}$$

We have several cases to consider:

Subcase 1: Suppose $B_1, B_2 \in \mathsf{Kind}$. It suffices to show:

$$\begin{aligned}
& (\!|\mathsf{let}\,\Box(x : B_1) = t_1\,\mathsf{in}\,t_2|\!)_\rho \\
=\ & (\!|t_2|\!)_{\rho[x \mapsto (\!|t_1|\!)_\rho]} \\
\in\ & [\![[t_1/z]B_2]\!]_\varepsilon\,([\![\mathsf{let}\,\Box(x : B_1) = t_1\,\mathsf{in}\,t_2]\!]_\varepsilon) \\
=\ & [\![B_2]\!]_{\varepsilon[x \mapsto [\![t_1]\!]_\varepsilon]}\,([\![t_2]\!]_{\varepsilon[x \mapsto [\![t_1]\!]_\varepsilon]})
\end{aligned}$$

By the IH:

IH(1): $(\!|t_1|\!)_\rho \in [\![\Box_s B_1]\!]_\varepsilon\,([\![t_1]\!]_\varepsilon) = [\![B_1]\!]_\varepsilon\,([\![t_1]\!]_\varepsilon)$

At this point we need:

$$(\Delta, \sigma_7) \odot (\Gamma, x : B_1) \models_{\varepsilon[x \mapsto [\![t_1]\!]_\varepsilon]} \rho[x \mapsto (\!|t_1|\!)_\rho]$$

But, this follows by definition and Lemma C.16.

By the IH:

IH(2): $(\!|t_2|\!)_{\rho[x \mapsto (\!|t_1|\!)_\rho]} \in [\![B_2]\!]_{\varepsilon[x \mapsto [\![t_1]\!]_\varepsilon]}\,([\![t_2]\!]_{\varepsilon[x \mapsto [\![t_1]\!]_\varepsilon]})$

Subcase 2: Suppose $B_1 \in \mathsf{Kind}$ and $B_2 \in \mathsf{Type}$. Similar to the previous case.

Subcase 3: Suppose $B_1 \in \mathsf{Type}$ and $B_2 \in \mathsf{Kind}$. It suffices to show:

$$\begin{aligned}
& (\!|\mathsf{let}\,\Box(x : B_1) = t_1\,\mathsf{in}\,t_2|\!)_\rho \\
=\ & (\!|t_2|\!)_{\rho[x \mapsto (\!|t_1|\!)_\rho]} \\
\in\ & [\![[t_1/z]B_2]\!]_\varepsilon\,([\![\mathsf{let}\,\Box(x : B_1) = t_1\,\mathsf{in}\,t_2]\!]_\varepsilon) \\
=\ & [\![B_2]\!]_\varepsilon\,([\![t_2]\!]_\varepsilon)
\end{aligned}$$

By the IH:

IH(1): $(\!|t_1|\!)_\rho \in [\![\Box_s B_1]\!]_\varepsilon = [\![B_1]\!]_\varepsilon$

At this point we need:

$$(\Delta, \sigma_7) \odot (\Gamma, x : B_1) \models_\varepsilon \rho[x \mapsto (\!|t_1|\!)_\rho]$$

But, this follows by definition.

By the IH:

IH(2): $(\!| t_2 |\!)_{\rho[x \mapsto (\!| t_1 |\!)_\rho]} \in [\![B_2]\!]_\varepsilon ([\![t_2]\!]_\varepsilon)$

Subcase 4: Suppose $B_1, B_2 \in \mathsf{Type}$. Similar to the previous case.

Case 12:
$$\frac{(\Delta \mid \sigma_1 \mid \sigma_2) \odot \Gamma \vdash t : A' \quad (\Delta \mid \sigma_2 \mid \mathbf{0}) \odot \Gamma \vdash A' = A : \mathsf{Type}_l}{(\Delta \mid \sigma_1 \mid \sigma_2) \odot \Gamma \vdash t : A} \text{ T\_TyConv}$$

By the induction hypothesis we know $(\Delta \mid \sigma_1 \mid \sigma_2) \odot \Gamma \models t : A'$ holds, but this implies $(\Delta \mid \sigma_1 \mid \sigma_2) \odot \Gamma \models t : A'$ holds by Lemma C.15.

□

**Corollary C.17.1** (Strong Normalization). *For every* $(\Delta \mid \sigma_1 \mid \sigma_2) \odot \Gamma \vdash t : A$, $t \in \mathsf{SN}$.

*Proof.* Similarly to CC, we can define a notion of canonical element in $\mathcal{K}[\![A]\!]$, and define a term valuation $\Delta \odot \Gamma \models_\varepsilon \rho$, and then conclude $\mathsf{SN}$ by the previous theorem. □

# D  Graded Modal Dependent Type Theory complete specification

$\boxed{\Delta \odot \Gamma \models_\varepsilon \rho}$   Valid Term Substitutions

$$\frac{}{\emptyset \odot \emptyset \models_\emptyset \emptyset} \text{ Rho\_Empty}$$

$$\frac{t \in [\![A]\!]_\varepsilon \quad \Delta \odot \Gamma \models_\varepsilon \rho \quad (\Delta \mid \sigma_2 \mid \mathbf{0}) \odot \Gamma \vdash A : \mathsf{Type}_0}{(\Delta, \sigma_2) \odot (\Gamma, x : A) \models_\varepsilon \rho[x \mapsto t]} \text{ Rho\_ExtTm}$$

$$\frac{t \in ([\![A]\!]_\varepsilon)(\varepsilon\, x) \quad \Delta \odot \Gamma \models_\varepsilon \rho \quad (\Delta \mid \sigma_2 \mid \mathbf{0}) \odot \Gamma \vdash A : \mathsf{Type}_1}{(\Delta, \sigma_2) \odot (\Gamma, x : A) \models_\varepsilon \rho[x \mapsto t]} \text{ Rho\_ExtTy}$$

$\boxed{\Delta \odot \Gamma \models \varepsilon}$   Valid Type Substitutions

$$\frac{}{\emptyset \odot \emptyset \models \emptyset} \text{ Ep\_Empty}$$

$$\frac{\Delta \odot \Gamma \models \varepsilon \quad (\Delta \mid \sigma_2 \mid \mathbf{0}) \odot \Gamma \vdash A : \mathsf{Type}_0}{(\Delta, \sigma_2) \odot (\Gamma, x : A) \models \varepsilon} \text{ Ep\_ExtTm}$$

$$\frac{X \in \mathcal{K}[\![A]\!] \quad \Delta \odot \Gamma \models \varepsilon \quad (\Delta \mid \sigma_2 \mid \mathbf{0}) \odot \Gamma \vdash A : \mathsf{Type}_1}{(\Delta, \sigma_2) \odot (\Gamma, x : A) \models \varepsilon[x \mapsto X]} \text{ Ep\_ExtTy}$$

$\boxed{\Delta \odot \Gamma \vdash}$   Well-formed contexts

$$\frac{}{\emptyset \odot \emptyset \vdash} \text{ Wf\_Empty}$$

$$\frac{(\Delta \mid \sigma \mid \mathbf{0}) \odot \Gamma \vdash A : \mathsf{Type}_l}{\Delta, \sigma \odot \Gamma, x : A \vdash} \text{ Wf\_Ext}$$

$\boxed{\Gamma \vdash t_1 \leq t_2}$   Subtyping

$$\frac{(\Delta \mid \sigma \mid \mathbf{0}) \odot \Gamma \vdash A = B : \mathsf{Type}_l}{(\Delta \mid \sigma) \odot \Gamma \vdash A \leq B} \text{ ST\_Eq}$$

$$\frac{(\Delta \mid \sigma) \odot \Gamma \vdash A \leq B \quad (\Delta \mid \sigma) \odot \Gamma \vdash B \leq C}{(\Delta \mid \sigma) \odot \Gamma \vdash A \leq C} \text{ST\_TRANS}$$

$$\frac{\Delta \odot \Gamma \vdash \quad l \leq l'}{(\Delta \mid \mathbf{0}) \odot \Gamma \vdash \mathsf{Type}_l \leq \mathsf{Type}_{l'}} \text{ST\_TY}$$

$$\frac{\begin{array}{c}(\Delta, \sigma_1 \mid \sigma_2, r \mid \mathbf{0}) \odot \Gamma, x : A \vdash B : \mathsf{Type}_l \\ (\Delta \mid \sigma_1) \odot \Gamma \vdash A' \leq A \quad (\Delta, \sigma_1 \mid \sigma_2, r) \odot \Gamma, x : A' \vdash B \leq B'\end{array}}{(\Delta \mid \sigma_1 + \sigma_2) \odot \Gamma \vdash (x :_{(s,r)} A) \to B \leq (x :_{(s,r)} A') \to B'} \text{ST\_ARROW}$$

$$\frac{(\Delta, \sigma_1 \mid \sigma_2, r) \odot \Gamma, x : A \vdash B \leq B'}{(\Delta \mid \sigma_1 + \sigma_2) \odot \Gamma \vdash (x :_r A) \otimes B \leq (x :_r A) \otimes B'} \text{ST\_TEN}$$

$$\frac{(\Delta \mid \sigma) \odot \Gamma \vdash A \leq A'}{(\Delta \mid \sigma) \odot \Gamma \vdash \square_s A \leq \square_s A'} \text{ST\_BOX}$$

$\boxed{\Gamma \vdash t_1 : t_2}$  Typing

$$\frac{\Delta \odot \Gamma \vdash}{(\Delta \mid \mathbf{0} \mid \mathbf{0}) \odot \Gamma \vdash \mathsf{Type}_l : \mathsf{Type}_{\mathsf{suc}\ l}} \text{T\_TYPE}$$

$$\frac{\Delta_1, \sigma, \Delta_2 \odot \Gamma_1, x : A, \Gamma_2 \vdash \quad |\Delta_1| = |\Gamma_1|}{(\Delta_1, \sigma, \Delta_2 \mid \mathbf{0}^{|\Delta_1|}, 1, \mathbf{0} \mid \sigma, 0, \mathbf{0}) \odot \Gamma_1, x : A, \Gamma_2 \vdash x : A} \text{T\_VAR}$$

$$\frac{(\Delta \mid \sigma_1 \mid \mathbf{0}) \odot \Gamma \vdash A : \mathsf{Type}_{l_1} \quad (\Delta, \sigma_1 \mid \sigma_2, r \mid \mathbf{0}) \odot \Gamma, x : A \vdash B : \mathsf{Type}_{l_2}}{(\Delta \mid \sigma_1 + \sigma_2 \mid \mathbf{0}) \odot \Gamma \vdash (x :_{(s,r)} A) \to B : \mathsf{Type}_{l_1 \sqcup l_2}} \text{T\_ARROW}$$

$$\frac{(\Delta \mid \sigma_1 \mid \mathbf{0}) \odot \Gamma \vdash A : \mathsf{Type}_{l_1} \quad (\Delta, \sigma_1 \mid \sigma_2, r \mid \mathbf{0}) \odot \Gamma, x : A \vdash B : \mathsf{Type}_{l_2}}{(\Delta \mid \sigma_1 + \sigma_2 \mid \mathbf{0}) \odot \Gamma \vdash (x :_r A) \otimes B : \mathsf{Type}_{l_1 \sqcup l_2}} \text{T\_TEN}$$

$$\frac{(\Delta, \sigma_1 \mid \sigma_3, r \mid \mathbf{0}) \odot \Gamma, x : A \vdash B : \mathsf{Type}_l \quad (\Delta, \sigma_1 \mid \sigma_2, s \mid \sigma_3, r) \odot \Gamma, x : A \vdash t : B}{(\Delta \mid \sigma_2 \mid \sigma_1 + \sigma_3) \odot \Gamma \vdash \lambda x.t : (x :_{(s,r)} A) \to B} \text{T\_FUN}$$

$$\frac{\begin{array}{c}(\Delta, \sigma_1 \mid \sigma_3, r \mid \mathbf{0}) \odot \Gamma, x : A \vdash B : \mathsf{Type}_l \\ (\Delta \mid \sigma_2 \mid \sigma_1 + \sigma_3) \odot \Gamma \vdash t_1 : (x :_{(s,r)} A) \to B \quad (\Delta \mid \sigma_4 \mid \sigma_1) \odot \Gamma \vdash t_2 : A\end{array}}{(\Delta \mid \sigma_2 + s * \sigma_4 \mid \sigma_3 + r * \sigma_4) \odot \Gamma \vdash t_1\, t_2 : [t_2/x]B} \text{T\_APP}$$

$$\frac{\begin{array}{c}(\Delta, \sigma_1 \mid \sigma_3, r \mid \mathbf{0}) \odot \Gamma, x : A \vdash B : \mathsf{Type}_l \\ (\Delta \mid \sigma_2 \mid \sigma_1) \odot \Gamma \vdash t_1 : A \quad (\Delta \mid \sigma_4 \mid \sigma_3 + r * \sigma_2) \odot \Gamma \vdash t_2 : [t_1/x]B\end{array}}{(\Delta \mid \sigma_2 + \sigma_4 \mid \sigma_1 + \sigma_3) \odot \Gamma \vdash (t_1, t_2) : (x :_r A) \otimes B} \text{T\_PAIR}$$

$$\frac{\begin{array}{c}(\Delta \mid \sigma_3 \mid \sigma_1 + \sigma_2) \odot \Gamma \vdash t_1 : (x :_r A) \otimes B \\ (\Delta, (\sigma_1 + \sigma_2) \mid \sigma_5, r' \mid \mathbf{0}) \odot \Gamma, z : (x :_r A) \otimes B \vdash C : \mathsf{Type}_l \\ (\Delta, \sigma_1, (\sigma_2, r) \mid \sigma_4, s, s \mid \sigma_5, r', r') \odot \Gamma, x : A, y : B \vdash t_2 : [(x,y)/z]C\end{array}}{(\Delta \mid \sigma_4 + s * \sigma_3 \mid \sigma_5 + r' * \sigma_3) \odot \Gamma \vdash \mathsf{let}\,(x,y) = t_1\,\mathsf{in}\,t_2 : [t_1/z]C} \text{T\_TENCUT}$$

$$\frac{(\Delta \mid \sigma \mid \mathbf{0}) \odot \Gamma \vdash A : \mathsf{Type}_l}{(\Delta \mid \sigma \mid \mathbf{0}) \odot \Gamma \vdash \square_s A : \mathsf{Type}_l} \text{T\_BOX}$$

$$\frac{(\Delta \mid \sigma_1 \mid \sigma_2) \odot \Gamma \vdash t : A}{(\Delta \mid s * \sigma_1 \mid \sigma_2) \odot \Gamma \vdash \square t : \square_s A} \text{T\_BOXI}$$

$$\frac{\begin{array}{c}(\Delta \mid \sigma_1 \mid \sigma_2) \odot \Gamma \vdash t_1 : \square_s A \\ (\Delta, \sigma_2 \mid \sigma_4, r \mid \mathbf{0}) \odot \Gamma, z : \square_s A \vdash B : \mathsf{Type}_l \\ (\Delta, \sigma_2 \mid \sigma_3, s \mid \sigma_4, (s * r)) \odot \Gamma, x : A \vdash t_2 : [\square x/z]B\end{array}}{(\Delta \mid \sigma_1 + \sigma_3 \mid \sigma_4 + r * \sigma_1) \odot \Gamma \vdash \mathsf{let}\,\square x = t_1\,\mathsf{in}\,t_2 : [t_1/z]B} \text{T\_BOXE}$$

$$\frac{(\Delta \mid \sigma_1 \mid \sigma_2) \odot \Gamma \vdash t : A \quad (\Delta \mid \sigma_2) \odot \Gamma \vdash A \leq B}{(\Delta \mid \sigma_1 \mid \sigma_2) \odot \Gamma \vdash t : B} \text{ T\_TyConv}$$

$\boxed{t_1 \rightsquigarrow t_2}$   Reduction rules

$$\overline{(\lambda x.t_1)\, t_2 \rightsquigarrow [t_2/x]t_1} \text{ Sem\_BetaFun}$$

$$\overline{\text{let } (x,y) = (t_1, t_2) \text{ in } t_3 \rightsquigarrow [t_1/x][t_2/y]t_3} \text{ Sem\_BetaTen}$$

$$\overline{\text{let } \Box x = \Box t_1 \text{ in } t_2 \rightsquigarrow [t_1/x]t_2} \text{ Sem\_BetaBox}$$

$$\frac{t_1 \rightsquigarrow t_1'}{(t_1, t_2) \rightsquigarrow (t_1', t_2)} \text{ Sem\_CongPair}$$

$$\frac{t_2 \rightsquigarrow t_2'}{(t_1, t_2) \rightsquigarrow (t_1, t_2')} \text{ Sem\_CongPairTwo}$$

$$\frac{t_1 \rightsquigarrow t_1'}{t_1\, t_2 \rightsquigarrow t_1'\, t_2} \text{ Sem\_CongFunOne}$$

$$\frac{t_1 \rightsquigarrow t_1'}{\text{let } (x,y) = t_1 \text{ in } t_2 \rightsquigarrow \text{let } (x,y) = t_1' \text{ in } t_2} \text{ Sem\_CongTenCut1}$$

$$\frac{t_1 \rightsquigarrow t_1'}{\text{let } \Box x = t_1 \text{ in } t_2 \rightsquigarrow \text{let } \Box x = t_1' \text{ in } t_2} \text{ Sem\_CongBox1}$$

$$\frac{A \rightsquigarrow A'}{\Box_s A \rightsquigarrow \Box_s A'} \text{ Sem\_CongBox4}$$

$$\frac{A \rightsquigarrow A'}{(x :_{(s,r)} A) \rightarrow B \rightsquigarrow (x :_{(s,r)} A') \rightarrow B} \text{ Sem\_CongArrow1}$$

$$\frac{B \rightsquigarrow B'}{(x :_{(s,r)} A) \rightarrow B \rightsquigarrow (x :_{(s,r)} A) \rightarrow B'} \text{ Sem\_CongArrow2}$$

$$\frac{A \rightsquigarrow A'}{(x :_r A) \otimes B \rightsquigarrow (x :_r A') \otimes B} \text{ Sem\_CongTen1}$$

$$\frac{B \rightsquigarrow B'}{(x :_r A) \otimes B \rightsquigarrow (x :_r A) \otimes B'} \text{ Sem\_CongTen2}$$

$\boxed{\Delta; \Gamma \vdash t_1 \Leftarrow t_2; \sigma_1; \sigma_2}$   Checking

$$\frac{\begin{array}{c}\Delta; \Gamma \vdash A \Rightarrow \text{Type}_l; \sigma_1; \mathbf{0} \\ \Delta, \sigma_1; \Gamma, x : A \vdash t \Leftarrow B; \sigma_2, s; \sigma_3, r\end{array}}{\Delta; \Gamma \vdash \lambda x.t \Leftarrow (x :_{(s,r)} A) \rightarrow B; \sigma_2; \sigma_1 + \sigma_3} \text{ ChkAlg\_Fun}$$

$\boxed{\Delta; \Gamma \vdash t_1 \Rightarrow t_2; \sigma_1; \sigma_2}$   Inference

$$\frac{\begin{array}{c}\Delta; \Gamma \vdash t_1 \Rightarrow (x :_{(s,r)} A) \rightarrow B; \sigma_2; \sigma_{13} \\ \Delta; \Gamma \vdash t_2 \Leftarrow A; \sigma_4; \sigma_1 \\ \Delta, \sigma_1; \Gamma, x : A \vdash B \Rightarrow \text{Type}_l; \sigma_3, r; \mathbf{0} \\ \sigma_{13} = \sigma_1 + \sigma_3\end{array}}{\Delta; \Gamma \vdash t_1\, t_2 \Rightarrow [t_2/x]B; \sigma_2 + s * \sigma_4; \sigma_3 + r * \sigma_4} \text{ InfAlg\_App}$$

$$\frac{\begin{array}{c}\Delta = \Delta_1, \sigma, \Delta_2 \\ \Gamma = \Gamma_1, x : A, \Gamma_2 \\ |\Gamma_1| = |\Delta_1| \\ \Delta_1; \Gamma_1 \vdash A \Rightarrow \mathsf{Type}_l; \sigma; \mathbf{0}\end{array}}{\Delta; \Gamma \vdash x \Rightarrow A; (\mathbf{0}^{|\Delta_1|}, 1, \mathbf{0}^{|\Delta_2|}); (\sigma, 0, \mathbf{0}^{|\Delta_2|})} \text{INFALG\_VAR}$$

$\boxed{\Gamma \vdash t_1 = t_2 : t_3}$  Term equality

$$\frac{(\Delta \mid \sigma_1 \mid \sigma_2) \odot \Gamma \vdash t : A}{(\Delta \mid \sigma_1 \mid \sigma_2) \odot \Gamma \vdash t = t : A} \text{TEQ\_REFL}$$

$$\frac{(\Delta \mid \sigma_1 \mid \sigma_2) \odot \Gamma \vdash t_1 = t_2 : A \quad (\Delta \mid \sigma_1 \mid \sigma_2) \odot \Gamma \vdash t_2 = t_3 : A}{(\Delta \mid \sigma_1 \mid \sigma_2) \odot \Gamma \vdash t_1 = t_3 : A} \text{TEQ\_TRANS}$$

$$\frac{(\Delta \mid \sigma_1 \mid \sigma_2) \odot \Gamma \vdash t_1 = t_2 : A}{(\Delta \mid \sigma_1 \mid \sigma_2) \odot \Gamma \vdash t_2 = t_1 : A} \text{TEQ\_SYM}$$

$$\frac{(\Delta \mid \sigma_1 \mid \sigma_2) \odot \Gamma \vdash t_1 = t_2 : A \quad (\Delta \mid \sigma_2) \odot \Gamma \vdash A \leq B}{(\Delta \mid \sigma_1 \mid \sigma_2) \odot \Gamma \vdash t_1 = t_2 : B} \text{TEQ\_CONVTY}$$

$$\frac{(\Delta \mid \sigma_1 \mid \mathbf{0}) \odot \Gamma \vdash A = C : \mathsf{Type}_{l_1} \quad (\Delta, \sigma_1 \mid \sigma_2, r \mid \mathbf{0}) \odot \Gamma, x : A \vdash B = D : \mathsf{Type}_{l_2}}{(\Delta \mid \sigma_1 + \sigma_2 \mid \mathbf{0}) \odot \Gamma \vdash (x :_{(s,r)} A) \to B = (x :_{(s,r)} C) \to D : \mathsf{Type}_{l_1 \sqcup l_2}} \text{TEQ\_ARROW}$$

$$\frac{(\Delta, \sigma_1 \mid \sigma_2, s \mid \sigma_3, r) \odot \Gamma, x : A \vdash t_1 : B \quad (\Delta \mid \sigma_4 \mid \sigma_1) \odot \Gamma \vdash t_2 : A}{(\Delta \mid \sigma_2 + s * \sigma_4 \mid \sigma_3 + r * \sigma_4) \odot \Gamma \vdash (\lambda x.t_1)\, t_2 = [t_2/x]t_1 : [t_2/x]B} \text{TEQ\_ARROWCOMP}$$

$$\frac{(\Delta \mid \sigma_1 \mid \sigma_2) \odot \Gamma \vdash t : (x :_{(s,r)} A) \to B}{(\Delta \mid \sigma_1 \mid \sigma_2) \odot \Gamma \vdash t = \lambda x.(t\, x) : (x :_{(s,r)} A) \to B} \text{TEQ\_ARROWUNIQ}$$

$$\frac{(\Delta, \sigma_1 \mid \sigma_2, s \mid \sigma_3, r) \odot \Gamma, x : A \vdash t_1 = t_2 : B}{(\Delta \mid \sigma_2 \mid \sigma_1 + \sigma_3) \odot \Gamma \vdash \lambda x.t_1 = \lambda x.t_2 : (x :_{(s,r)} A) \to B} \text{TEQ\_FUN}$$

$$\frac{\begin{array}{c}(\Delta, \sigma_3 \mid \sigma_4, r \mid \mathbf{0}) \odot \Gamma, x : A \vdash B : \mathsf{Type}_l \\ (\Delta \mid \sigma_1 \mid \sigma_3 + \sigma_4) \odot \Gamma \vdash t_1 = t_2 : (x :_{(s,r)} A) \to B \quad (\Delta \mid \sigma_2 \mid \sigma_3) \odot \Gamma \vdash t_3 = t_4 : A\end{array}}{(\Delta \mid \sigma_1 + s * \sigma_2 \mid \sigma_4 + r * \sigma_2) \odot \Gamma \vdash t_1\, t_3 = t_2\, t_4 : [t_3/x]B} \text{TEQ\_APP}$$

$$\frac{(\Delta \mid \sigma_1 \mid \mathbf{0}) \odot \Gamma \vdash A = C : \mathsf{Type}_{l_1} \quad (\Delta, \sigma_1 \mid \sigma_2, r \mid \mathbf{0}) \odot \Gamma, x : A \vdash B = D : \mathsf{Type}_{l_2}}{(\Delta \mid \sigma_1 + \sigma_2 \mid \mathbf{0}) \odot \Gamma \vdash (x :_s A) \otimes B = (x :_s C) \otimes D : \mathsf{Type}_{l_1 \sqcup l_2}} \text{TEQ\_TEN}$$

$$\frac{\begin{array}{c}(\Delta, (\sigma_1 + \sigma_2) \mid \sigma_5, r' \mid \mathbf{0}) \odot \Gamma, z : (x :_r A) \otimes B \vdash C : \mathsf{Type}_l \\ (\Delta \mid \sigma_3 \mid \sigma_1) \odot \Gamma \vdash t_1 : A \quad (\Delta \mid \sigma_6 \mid \sigma_2 + r * \sigma_3) \odot \Gamma \vdash t_2 : [t_1/x]B \\ (\Delta, \sigma_1, (\sigma_2, r) \mid \sigma_4, s, s \mid \sigma_5, r', r') \odot \Gamma, x : A, y : B \vdash t_3 : [(x,y)/z]C\end{array}}{(\Delta \mid \sigma_4 + s * (\sigma_3 + \sigma_6) \mid \sigma_5 + r' * (\sigma_3 + \sigma_6)) \odot \Gamma \vdash \mathsf{let}\,(x,y) = (t_1, t_2)\,\mathsf{in}\,t_3 = [t_1, t_2/x, y]t_3 : [(t_1,t_2)/z]C} \text{TEQ\_TENCOMP}$$

$$\frac{\begin{array}{c}(\Delta, \sigma_1 \mid \sigma_3, r \mid \mathbf{0}) \odot \Gamma, x : A \vdash B : \mathsf{Type}_l \\ (\Delta \mid \sigma_2 \mid \sigma_1) \odot \Gamma \vdash t_1 = t_1' : A \quad (\Delta \mid \sigma_4 \mid \sigma_3 + r * \sigma_2) \odot \Gamma \vdash t_2 = t_2' : [t_1/x]B\end{array}}{(\Delta \mid \sigma_2 + \sigma_4 \mid \sigma_1 + \sigma_3) \odot \Gamma \vdash (t_1, t_2) = (t_1', t_2') : (x :_r A) \otimes B} \text{TEQ\_PAIR}$$

$$\frac{\begin{array}{c}(\Delta \mid \sigma_3 \mid \sigma_1 + \sigma_2) \odot \Gamma \vdash t_1 = t_2 : (x :_r A) \otimes B \\ (\Delta, (\sigma_1 + \sigma_2) \mid \sigma_5, r' \mid \mathbf{0}) \odot \Gamma, z : (x :_r A) \otimes B \vdash C : \mathsf{Type}_l \\ (\Delta, \sigma_1, (\sigma_2, r) \mid \sigma_4, s, s \mid \sigma_5, r', r') \odot \Gamma, x : A, y : B \vdash t_3 = t_4 : [(x,y)/z]C\end{array}}{(\Delta \mid \sigma_4 + s * \sigma_3 \mid \sigma_5 + r' * \sigma_3) \odot \Gamma \vdash (\mathsf{let}\,(x,y) = t_1\,\mathsf{in}\,t_3) = (\mathsf{let}\,(x,y) = t_2\,\mathsf{in}\,t_4) : [t_1/z]C} \text{TEQ\_TENCUT}$$

$$\frac{(\Delta \mid \sigma_1 \mid \sigma_2) \odot \Gamma \vdash t : (x :_r A) \otimes B}{(\Delta \mid \sigma_1 \mid \sigma_2) \odot \Gamma \vdash t = (\mathsf{let}\,(x,y) = t\,\mathsf{in}\,(x,y)) : (x :_r A) \otimes B} \text{TEQ\_TENU}$$

$$\frac{(\Delta \mid \sigma \mid \mathbf{0}) \odot \Gamma \vdash A = B : \mathsf{Type}_l}{(\Delta \mid \sigma \mid \mathbf{0}) \odot \Gamma \vdash \Box_s A = \Box_s B : \mathsf{Type}_l} \text{ TEQ\_Box}$$

$$\frac{(\Delta \mid \sigma_1 \mid \sigma_2) \odot \Gamma \vdash t_1 = t_2 : A}{(\Delta \mid s * \sigma_1 \mid \sigma_2) \odot \Gamma \vdash \Box t_1 = \Box t_2 : \Box_s A} \text{ TEQ\_BoxI}$$

$$\frac{\begin{array}{c}(\Delta \mid \sigma_1 \mid \sigma_2) \odot \Gamma \vdash t_1 : A \\ (\Delta, \sigma_2 \mid \sigma_4, r \mid \mathbf{0}) \odot \Gamma, z : \Box_s A \vdash B : \mathsf{Type}_l \\ (\Delta, \sigma_2 \mid \sigma_3, s \mid \sigma_4, (s * r)) \odot \Gamma, x : A \vdash t_2 : [\Box x/z]B\end{array}}{(\Delta \mid \sigma_3 + s * \sigma_1 \mid \sigma_4 + s * r * \sigma_1) \odot \Gamma \vdash (\mathsf{let}\,\Box x = \Box t_1 \,\mathsf{in}\, t_2) = [t_1/x]t_2 : [\Box t_1/z]B} \text{ TEQ\_BoxB}$$

$$\frac{\begin{array}{c}(\Delta, \sigma_2 \mid \sigma_4, r \mid \mathbf{0}) \odot \Gamma, z : \Box_s A \vdash B : \mathsf{Type}_l \\ (\Delta \mid \sigma_1 \mid \sigma_2) \odot \Gamma \vdash t_1 = t_2 : \Box_s A \quad (\Delta, \sigma_2 \mid \sigma_3, s \mid \sigma_4, (s * r)) \odot \Gamma, x : A \vdash t_3 = t_4 : [\Box x/z]B\end{array}}{(\Delta \mid \sigma_1 + \sigma_3 \mid \sigma_4 + r * \sigma_1) \odot \Gamma \vdash (\mathsf{let}\,\Box x = t_1 \,\mathsf{in}\, t_3) = (\mathsf{let}\,\Box x = t_2 \,\mathsf{in}\, t_4) : [t_1/z]B} \text{ TEQ\_BoxE}$$

$$\frac{(\Delta \mid \sigma_1 \mid \sigma_2) \odot \Gamma \vdash t : \Box_s A}{(\Delta \mid \sigma_1 \mid \sigma_2) \odot \Gamma \vdash t = (\mathsf{let}\,\Box x = t \,\mathsf{in}\, \Box x) : \Box_s A} \text{ TEQ\_BoxU}$$




\end{document}